\documentclass[a4paper,11pt]{book}
\usepackage{macros}

\begin{document}

\title{\bfseries {\sc{The graph alignment problem: fundamental limits and efficient algorithms}}}

\author{Luca Ganassali \\[6pt]
		\emph{Inria, Département d'Informatique de l'ENS}\\
		\emph{PSL Research University, Paris, France}\\[2.5cm]
		PhD thesis under the supervision of Laurent Massoulié and Marc Lelarge \\[2.5cm]}
		\date{September 20, 2022}
\maketitle

\thispagestyle{empty}
\paragraph{\sc Abstract}
This thesis focuses on statistical inference in graphs (or matrices) in high dimension and studies the graph alignment problem which aims to recover a hidden underlying matching between the nodes of two correlated random graphs. 

Similarly to many other inference problems in planted models, we are interested in understanding the fundamental information-theoretical limits as well as the computational  hardness of graph alignment.

First, we study the Gaussian setting, when the graphs are complete and the signal lies on correlated Gaussian edges weights. We prove that the exact recovery task exhibits a sharp information-theoretic threshold (and characterize it), and study a simple and natural spectral method for recovery, \alg{EIG1}, which consists in aligning the leading eigenvectors of the adjacency matrices of the two graphs. 

While most of the recent work on the subject was dedicated to recovering the hidden signal in dense graphs, we next explore graph alignment in the sparse regime, where the mean degree of the nodes are constant, not scaling with the graph size. In this particularly challenging setting, for sparse Erd\H{o}s-R\'enyi graphs, only a fraction of the nodes can be correctly matched by any algorithm. Our second contribution is an information-theoretical result which characterizes a regime where even this partial alignment is impossible, and gives upper bounds on the reachable overlap between any estimator and the true planted matching.

We next propose an algorithm that performs partial alignment, \alg{NTMA}, which is based on a measure of similarity -- called the tree matching weight -- between tree-like neighborhoods of the nodes in the graphs.

Under this local approach in the sparse regime, we are brought to study a related problem: correlation detection in random unlabeled trees. This hypothesis testing problem consists in testing whether two trees are correlated or independent. The tree matching weight yields a first method for this question as well; another contribution is to study an optimal test based on the likelihood ratio. In a correlated Galton-Watson model, which is well-known to be the local approximation of the sparse Erd\H{o}s-R\'enyi model, we characterize the regimes of performance of this test.

Finally, we come back to graph alignment and propose a message-passing algorithm, \alg{MPAlign}, naturally inspired by the study of the related problem on trees. This message-passing algorithm is analyzed and provably recovers a fraction of the planted signal in some regimes of parameters. 

\vspace{1.5cm}


\newpage
\thispagestyle{empty}
\section*{Acknowledgments}
This work was partially supported by the French government under management of Agence Nationale de la Recherche as part of the “Investissements d’avenir” program, reference ANR19-P3IA-0001 (PRAIRIE 3IA Institute).

\newpage
%

\thispagestyle{plain}
\section*{Contributions and outline}
\subsubsection{Chapter 1.} This open chapter is a general introduction to the dissertation. We start with stating general frameworks for inference on graphs, and we then go over basic concepts of random graph theory, with general results that will be useful throughout the thesis. We give a general overview on inference problems in random graphs, with several iconic examples, and introduce the phase transition phenomena arising in the high-dimensional regime. We next introduce and motivate the graph alignment problem, discuss general elementary results and give a short survey of prior techniques, methods and theoretical work on the subject. We also introduce the problem of detecting correlation in trees.

\subsubsection{Chapter 2.} This chapter, based on the paper \cite{Ganassali21MSML} published at \emph{MSML 2021}, investigates information-theoretic limits for exact alignment in the Gaussian setting, when the graphs are complete and the signal lies on correlated Gaussian edges weights. This model is often viewed as an interesting playground for graph alignment. We prove that the exact recovery task exhibits a sharp fundamental threshold (and characterize it).

\subsubsection{Chapter 3.} We continue the exploration of the Gaussian setting in this chapter, studying a simple and natural spectral method for recovery which consists in aligning the leading eigenvectors of the adjacency matrices of the two graphs. We give theoretical guarantees for this algorithm, showing a zero-one law property for this method to work, in terms of the signal-to-noise ratio. This chapter, based on the paper \cite{GLM19} published in \emph{Advances in Probability}.

\subsubsection{Chapter 4.} We focus in this chapter on the study of \ER graph alignment in the sparse regime, where the mean degree of the graphs are constant, not scaling with the number of nodes. Based on the paper \cite{ganassali2021impossibility} published at \emph{COLT 2021}, we prove an information-theoretical result characterizing a regime where even partial alignment is impossible, and giving upper bounds on the reachable overlap between any estimator and the planted matching.

\subsubsection{Chapter 5.} This chapter investigates an algorithm for sparse graph alignment, which relies on a measure of similarity -- called the tree matching weight -- between tree-like neighborhoods of the nodes in the graphs. We give theoretical guarantees for this method to work in the \ER model, and propose along the way a test to decide whether two trees are correlated or independent. This chapter is based on the paper \cite{Ganassali20a}, published at \emph{COLT 2020}.

\subsubsection{Chapter 6.} We follow the previous local approach in the sparse regime, and get interested in a related problem: correlation detection in random unlabeled trees. For this hypothesis testing problem, we study an optimal test based on the likelihood ratio. In a correlated Galton-Watson model, which is well-known to be the local approximation of the sparse Erd\H{o}s-R\'enyi model, we characterize regimes of performance of this test. Then, we come back to graph alignment and propose a message-passing algorithm naturally inspired by the study of the related problem on trees. This message-passing algorithm is analyzed and provably recovers a fraction of the planted signal in some regimes of parameters.  The chapter is based on the paper \cite{GMLTrees2021journal}, which short version is published at \emph{ITCS 2021}.

\subsubsection{Chapter 7 (Addendum).} We added a last chapter at the end of the dissertation, presenting recent results for correlation detection in trees. These results are significantly improving on previous work and give a general understanding of the fundamental limits of the problem, as well as some interesting perspectives discussed afterwards in the conclusion.
 
\newpage

\thispagestyle{empty}
\tableofcontents

\newpage

\thispagestyle{plain}
\section*{Notations}
\addcontentsline{toc}{chapter}{Notations}
{\small 
\begin{flushleft}
	
\begin{tabular}{r l} 
& \textit{Basics} \vspace*{0.2cm}\\
$i,j,k,\ell,m...$ & non negative integers, most of the time\\
$[m]$ & set $\left\{1, \ldots, m \right\}$ of integers from $1$ to $m$\\
$\card{\cX}$ & cardinal of a finite set $\cX$ \\
$\cS_{m}$ & set of permutations on $[m]$ (we often identify $\cS_{k}$ to $\cS_{\cX}$ whenever $\card{\cX}=k$)\\ \\
$\cS(A,B)$ & set of injective mappings between finite sets $A$ and $B$ \\
$\cS(k,\ell)$ & set of injective mappings from $[k]$ to $[\ell]$ \\
$\pi,\sigma$ & permutations, most of the time\\
$\Pi,\Sigma$ & permutation matrices, most of the time\\
& \\
$O,o,\Omega, \omega, \Theta, \sim$ & standard Landau notations\\
$\one_C$ & indicator function at $C$; $\one_C=1$ if $C$ is satisfied, $0$ otherwise\\

& \vspace*{0.1cm}\\
& \textit{Graphs} \vspace*{0.2cm}\\
$G=(V,E)$ & a graph $G$ with vertex set $V$ and edge set $E$ \\
$\conn$ or $\connin{G}$ & connectivity in undirected graph $G$ (eluded if no ambiguity)\\
$A(G)$ & adjacency matrix of graph $G$, sometimes $A$ if no amibiguity\\
$n$ & number of nodes of a graph, most of the time \\
$u,v$ & nodes of a graph, most of the time \\
& \\
$T$ & a tree $T$ \\
$d$ & depth of a tree, most of the time \\
$\cX_d$ & set of unlabeled finite trees of depth at most $d$ \\
$\dT_k$ & set of unlabeled trees of size $k$\\

& \vspace*{0.1cm}\\
& \textit{Probability} \vspace*{0.2cm}\\
& \emph{General convention: sometimes lowercase characters are used to distinguish} \\
& \emph{deterministic objects from random variables (uppercase).} \\
& \\
$\sim$ & is distributed according to (sometimes also denotes asymptotic equivalence) \\
$\eqd$ & equality in distribution \\
$\Ber(p)$ & Bernoulli distribution with parameter $p \in [0,1]$\\
$\Bin(n,p)$ & Binomial distribution with parameters $n \geq 0$, $p \in [0,1]$\\
$\Poi(\lambda)$ & Poisson distribution with parameter $\lambda >0$\\
$\Exp(\mu)$ & exponential distribution with parameter $\mu >0$\\
$\cN(\mu,v)$ & Gaussian distribution with mean $\mu$ and variance $v$\\
& \\
$\GOE$ & Gaussian Orthogonal Ensemble \\
$\Wig(n,\xi)$, $\Wig'(n,\rho)$ & Correlated Gaussian Wigner model (or a variant thereof) with size $n \times n$, \\
& noise parameter $\xi>0$ or correlation $\rho \in [0,1]$.\\
$\G(n,p)$ & \ER model with $n$ nodes and edge probability $p$\\
$\G(n,q,s)$ & correlated \ER model with $n$ nodes, edge probability $q$ and correlation $s$\\
$\SBM(n,\alpha,P)$ & stochastic block model with $n$ nodes, community distribution $\alpha$\\
&  and edge probabilities $P$\\
$\GWmu_d$ & Galton-Watson model with offspring distribution $\Poi(\mu)$ up to depth $d$\\
$\dPls_d$ & model of correlated Galton-Watson trees \\
$\dPl_d$ &  model of independent Galton-Watson trees 
\end{tabular}

\end{flushleft}
}

\newpage

\chapter{Introduction}\label{chapter:intro}
\section{Context}
A myriad of datasets found in real life can be represented as graphs, which are nowadays becoming more and more useful to model complex systems. The network of Facebook users can be visualized in a graph where each edge encodes a friendship relationship, Netflix can be viewed as a graph between users and movies, where each edge carries a rating or browsing statistics. Many examples of the overwhelming presence of graphs can be found across applications in a variety of different fields: modelling interaction between proteins in an organism, representing cities or destinations for route optimization, extracting a mesh from a 3D object, analyzing the spread of epidemics or fake news on Twitter, finding similar patterns in data, etc. 

\begin{figure}[h]
	\centering
	\medskip
	\begin{subfigure}[t]{0.49\linewidth}
		\centering
		\includegraphics[scale=0.3]{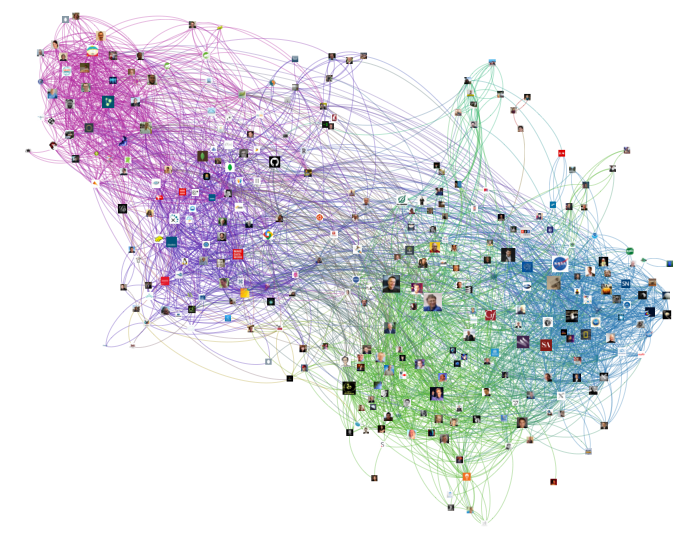}
		\caption{{\footnotesize a Twitter network with four communities}}
	\end{subfigure}
	\hfill
	\begin{subfigure}[t]{0.49\linewidth}
		\centering
		\includegraphics[scale=0.2]{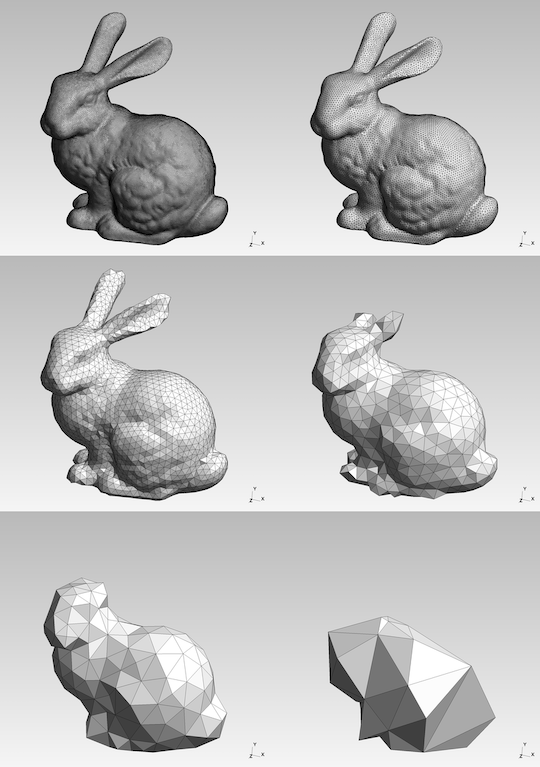}
		\caption{{\footnotesize some 3D finite element meshes}}
	\end{subfigure}
	
	\begin{subfigure}[t]{0.49\linewidth}
		\centering
		\vspace{0.5cm}
		\includegraphics[scale=0.24]{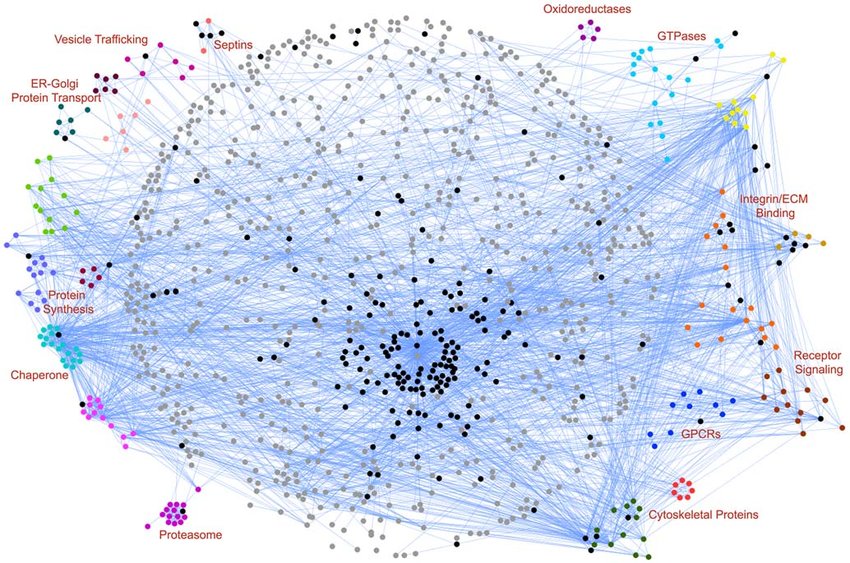}
		\caption{{\footnotesize a graph showing the platelet Protein-Protein Interaction network}}
	\end{subfigure}
	\caption{Some graphs\protect\footnotemark\, in real life.}
	\label{fig:intro:real_life_graphs} 
\end{figure} 
\footnotetext{sources: http://allthingsgraphed.com/2014/11/02/twitter-friends-network/ for $(a)$, https://gmsh.info/ for $(b)$, \cite{OureshiPPI09} for $(c)$.}

This thesis focuses on statistical inference in graphs, which aims to extract relevant information from the observation of graph-shaped data. We are interested in understanding the fundamental aspects of these problems, as well as designing and analyzing algorithms for the considered tasks, seeking to characterize the regimes in which they may suceed. 

\subsubsection{Statistical inference}
In its broadest sense, statistical inference aims to draw meaningful conclusions based upon the observation of data. Suppose that we are given samples (in the form of measurements in $\dR^d$, graphs, matrices, etc.) assumed to be drawn according to some probability distribution. The statistician designs methods to recover some information about this probability distribution, e.g. testing hypotheses or deriving estimates for some parameter $\theta \in \Theta$ of the distribution. The usual framework is as follows:

$$ \mbox{parameter } \theta \in \Theta \longrightarrow \mbox{data } Y \longrightarrow \mbox{estimator } \hat{\theta}$$ 

The main (informal) questions that arise in this setting are: ‘how well can we discriminate between different models/hypotheses?', ‘can we efficiently estimate the parameter $\theta$?'

This thesis focuses on statistical inference in random graphs -- or random matrices -- in \emph{high dimension}, when both the intrinsic dimension of the data and that of the parameter are large. This common assumption is particularly relevant for two main reasons: first, it fits well with real data, datasets being nowadays larger and larger; second, results in this asymptotic regime exhibit interesting and unexpected phenomena, see Section \ref{intro:section:phases}. 

\subsubsection{Planted models}
Other inference problems fall into the intuitive, conventional, and slightly different \emph{planted framework}, where the observation is the result of a perturbation of some underlying \emph{signal} of interest.
 
In this planted framework, the random model -- sometimes also referred to as  \emph{teacher-student model} in the statistical physics community -- is as follows: some signal $X$ is drawn according to some prior (we hence work in a \emph{bayesian setting}), and 'planted' in the data. Given the signal $X$, the observation $Y$ is drawn according to some conditional distribution $p(\cdot \cond X)$. The framework is as follows:

$$ \mbox{signal } X \sim p_X \longrightarrow \mbox{observation } Y \sim p(\cdot | X)  \longrightarrow \mbox{estimator } \hat{X}$$ 

We refer to Section \ref{intro:subsection:zoo} for a closer look at planted models in the context of inference in random graphs. 

\subsubsection{Detection and reconstruction tasks}
The two main (informal) questions stated earlier (testing hypotheses and estimating some parameter) can now be reformulated, and their counterparts specifically belong to the planted framework.
\begin{itemize}
	\setlength\itemsep{0cm}
 	\item[$(i)$] Can we detect the presence of a planted signal in the data?
 	\item [$(ii)$] If yes, are we also able to recover the signal?
\end{itemize} 
Let us give a succinct mathematical formulation of questions $(i)$ and $(ii)$ more formally. Let us define $n$ to be a generic dimension parameter; working in the high-dimensional regime here corresponds to making $n$ tend to infinity.

Question $(i)$ defines the \emph{detection task}. Detecting the presence of signal exactly consists in discriminating a model with no planted signal -- the \emph{null model} -- from the planted model, based on the observation of $Y$. In other words, detection task corresponds to the following hypothesis testing 
\begin{equation*}
	\cH_0 := `` \mbox{$Y$ is drawn from the null model} " \; \mbox{versus} \; \cH_1 := `` \mbox{$Y$ is drawn from the planted model} " 
\end{equation*}
Given a detection task, we are thus interested in designing a test $\cT$ (i.e., a measurable function of $Y$ taking values in $\{0,1\}$) for which we are able to give guarantees with probability tending to $1$ asymptotically in $n$.

Question $(ii)$ defines the \emph{reconstruction} -- or \emph{recovery} -- \emph{task}. Recovering the signal now corresponds to designing an estimator $\hat{X}$ (i.e., a measurable function of $Y$) for which we are able to give guarantees (e.g. prove that $\hat{X}$ is somehow close to $X$) with probability tending to $1$ asymptotically in $n$.\\

Let us pause for a moment after these broad definitions. Intuitively, the detection task is in general easier than the reconstruction task -- even though some counter-examples can be found in \cite{Banks17} -- and impossibility of detection almost always implies impossibility of reconstruction. Conversely, non-equivalence between detection and reconstruction may also seem rather counter-intuitive, when considering that the usual strategy for detecting some signal consists precisely in exhibiting the latter. Let us give hereafter an easy example in which these two problems are indeed definitely different.

\subsubsection{A toy example (1/2): finding hay in a haystack}


Consider a bit sequence $(\xi_1, \ldots, \xi_n)$ of length $n$ made of i.i.d. entries taking the value $0$ or $1$ with equal probability. Let $k>0$ that may depend on $n$. 

Under the planted model, the observation $Y = (Y_1, \ldots, Y_n)$ is generated as follows: we choose $k$ positions $1 \leq i_1 < \ldots < i_k \leq n$ uniformly at random, and set for all $i \in \{1,\ldots,n\}$, $Y_i = 1$ if $i \in \{i_1, \ldots, i_k\}$, and $Y_i = \xi_i$ otherwise. We denote $Y \sim \dP_{1,k}$.

Under the null model, we simply set  $Y_i = \xi_i$ for all $i \in \{1,\ldots,n\}$, and denote $Y \sim \dP_{0}$.

\begin{figure}[H]\label{fig:intro:haystack}
\vspace*{0.4em}
\centering
\begin{tabular}{c||cccccccccccccccccc}
%
		\hline
		$\xi$ & 1 & 0 & 0 & 0 & 1 & 1 & 0 & 1 & 0 & 1 & 1 & 1 & 1 & 0 & 1 & 0 & 0 & 0  \\ \hline
		$k$ positions & \textcolor{red}{$\times$} &  & \textcolor{red}{$\times$} & \textcolor{red}{$\times$} &  &  &  &  &  & \textcolor{red}{$\times$} & &  & \textcolor{red}{$\times$} & \textcolor{red}{$\times$} &  &  & \textcolor{red}{$\times$} &    \\ \hline
		$Y$ & {$1$} & 0 & {$1$} & {$1$} & 1 & 1 & 0 & 1 & 0 & {$1$} & 1 & 1 & {$1$} & {$1$} & 1 & 0 & {$1$} & 0   \\ 
		\hline
\end{tabular}

\caption{A realization with $n=18$, $k=7$. After transformation, in the sequence $Y$, the presence of a planted signal is highly probable; but where are the $k$ extra ones?}
\end{figure}

\textit{Detection.} In this simple example the planted signal consists in extra ones somewhere in the data. It then easy to see check that an optimal test for detection is simply based on counting occurrences of ones. Define
$N_1(Y) := \card{\set{i: Y_i=1}}.$ Standard concentration inequalities (e.g. Hoeffding's inequality) straightaway give that with high probability -- that is, with probability tending to $1$ when $n \to \infty$:

$$
N_1(Y) = \left\{
\begin{array}{ll}
n/2 + \Theta(\sqrt{n}) & \mbox{under the null model } \dP_{0}, \\
(n+k)/2 + \Theta(\sqrt{n-k}) & \mbox{under the planted model } \dP_{1,k}.
\end{array}
\right.
$$

Comparing these typical values shows that as soon as $k = \omega (\sqrt{n})$, extra ones can be detected, e.g. with a test $\cT_n$ outputting $1$ if and only if $N_1(Y)$ is greater than $n/2 + k/4$. Indeed, if $k = \omega (\sqrt{n})$, we will have $\dP_0(\cT_n=0) \to 1$, $\dP_{1,k}(\cT_n=1) \to 1$. Such a test is said to achieve \emph{strong detection} (see Section \ref{intro:subsection:hypothesis_testing_trees}). On the contrary, unreachability of strong detection when $k = O(\sqrt{n})$ is established by applying the central limit theorem -- details are left to the reader.\\

\textit{Reconstruction.} We are now in a position to understand why the two tasks are very different here. Though it is rather simple to detect extra ones when $k = \omega (\sqrt{n})$, the reconstruction task would consist in recovering the exact positions of the extra ones. If one had no idea about the data, a naive -- and somehow the worst -- method would consist in choosing these $k$ positions uniformly at random among the $N_1(Y)$ possibilities. It is easy to check that the number of positions that are correctly recovered -- or, the \emph{overlap} -- with this method is of order $k^2/n$ which is almost always very small compared to $k$ even when $k = \omega (\sqrt{n})$.

A moment of thought shows that this naive method can never be outperformed. Indeed, the posterior distribution of the positions of extra ones is given by 
\begin{flalign*}
	\dP_{1,k}(i_1, \ldots, i_k \cond Y) & = \frac{1}{\dP_{1,k}(Y)} \one_{i_1 < \ldots < i_k} \one_{Y_{i_1} = \ldots = Y_{i_k} = 1} \binom{n}{k}^{-1} \left( \frac{1}{2}\right)^{n-k}.
\end{flalign*} The dependence on $i_1, \ldots, i_k$ lying only in the terms $ \one_{i_1 < \ldots < i_k}$ and $\one_{Y_{i_1} = \ldots = Y_{i_k} = 1}$, it is therefore the uniform distribution on the set of ordered lists of length $k$ among the $N_1(Y)$ positions of ones. 

In particular, if $ k = \omega(\sqrt{n})$ and $k = o(n)$, then detection is easy but reconstruction is impossible, in the sense that no method can recover more than $o(k)$ of the planted extra ones. See Section \ref{intro:section:phases} for the definition of a formalized context for these observations.

\section*{Organization of rest of the introduction}

We start in Section \ref{intro:section:inference_on_rg} with some basics of random graph theory as well as general results and famous techniques that will be useful throughout this work. We then give a general overview on inference problems in random graphs through several widely studied examples, as well as the definition of the phase transition phenomena that crop up in these problems.

We introduce in Section \ref{intro:section:ga} the graph alignment problem, which lies at the very heart of this thesis and which various aspects will be discussed in the next chapters. We give insights on the motivations, discuss general related topics and give an overview of prior techniques and methods for this problem as well as theoretical guarantees, aside from our work.

We finally describe in Section \ref{intro:section:cdt} a related problem which will be the focus of Chapters \ref{chapter:NTMA} and \ref{chapter:MPAlign}, correlation detection in random trees, which is interesting for its own sake, and has a strong connection between tree correlation detection and graph alignment.

\section{Inference on random graphs: a short tour}\label{intro:section:inference_on_rg}

In this section, we will present the general framework of inference on randoms graphs, introducing some basic concepts and notations, and describing several -- fundamental -- examples for problems of this sort.

\subsection{Basics of random graph theory}\label{intro:subsection:basics_rg}

\subsubsection{Graphs} 
A (simple) \emph{graph} $G=(V,E)$ is a discrete structure consisting in a vertex set $V$ and an edge set $E$. Elements of $V$ are called \emph{vertices}, sometimes \emph{nodes}.

In an \emph{undirected} graph, $E$ is a subset of $\binom{V}{2}$, the set of unordered pairs (of $2-$sets) of distinct elements of $V$, and an edge $e$ between nodes $u$ and $v$ is denoted by $\{u,v\}$. If the graph is \emph{oriented}, then $E$ contains ordered pairs (or $2-$tuples) of elements of $V$, and they are denoted by $(u,v)$.

All graphs considered throughout along this manuscript are finite (namely $V$ and $E$ are finite sets), and undirected, unless stated otherwise. If $u,v \in V$ are such that $\{u,v\} \in E$, we denote $u \connin{G} v$ and $u \conn v$ when there is no ambiguity on the graph, and the vertices $u$ and $v$ are said to be \emph{connected}, or \emph{neighbors} in $G$.

\subsubsection{Adjacency matrix, weighted graphs}
A graph $G=(V,E)$ with node set $V=[n]$ is often represented through its \emph{adjacency matrix} $A = A(G) \in \dR^{n \times n}$ defined as follows:

\begin{equation*}
\forall \, u,v \in [n], \; A_{u,v} = \one_{\{u,v\} \in E} \,.
\end{equation*}

An undirected \emph{weighted graph} $G=(V,E)$ is a graph with additional information on edges, namely
$$ \forall \, u,v \in [n], \; A_{u,v} = \one_{\{u,v\} \in E} W_{u,v} \, ,$$ where  variables $W_{u,v} \in \dR$ are \emph{edge weights}.

\subsubsection{The \ER model} A simple, greatly celebrated, and widely used model of random graphs is the \ER model, introduced by Paul Erd\H{o}s and Alfr\'ed R\'enyi in 1959 \cite{Erdos59}. In this model, denoted by $\G(n,p)$, the graph $G$ has node set $V=[n]$ and each pair $\{u,v\}$ for $u \neq v \in [n]$ is present in $E$ independently with probability $p$.

\begin{figure}[h]
	\centering
	\medskip
	\begin{subfigure}[t]{0.49\linewidth}
		\centering
		\includegraphics[scale=0.23]{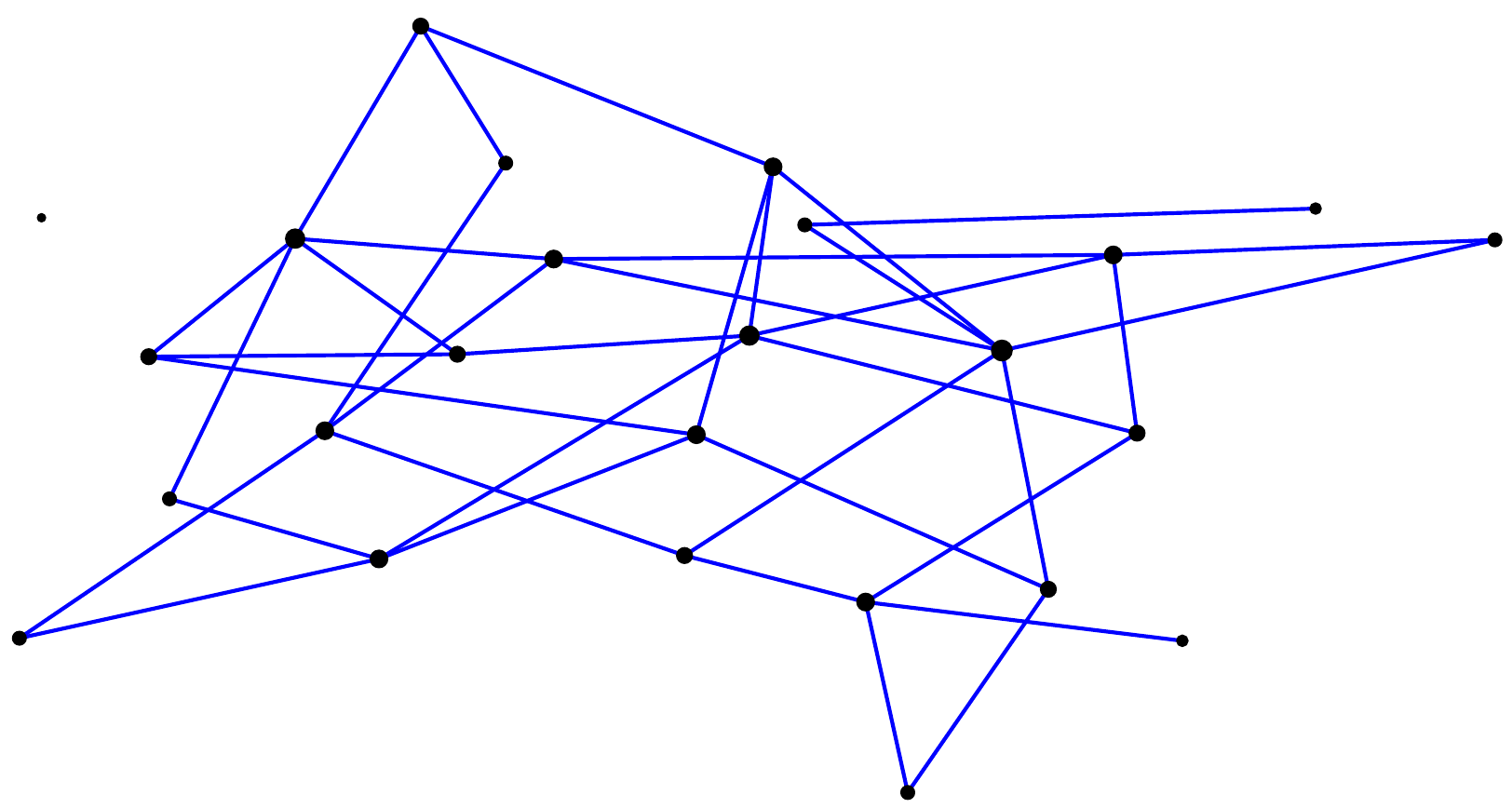}
		\caption{{\footnotesize $n=25, p = 0.14$}}
	\end{subfigure}
	\hfill
	\begin{subfigure}[t]{0.49\linewidth}
		\centering
		\includegraphics[scale=0.3]{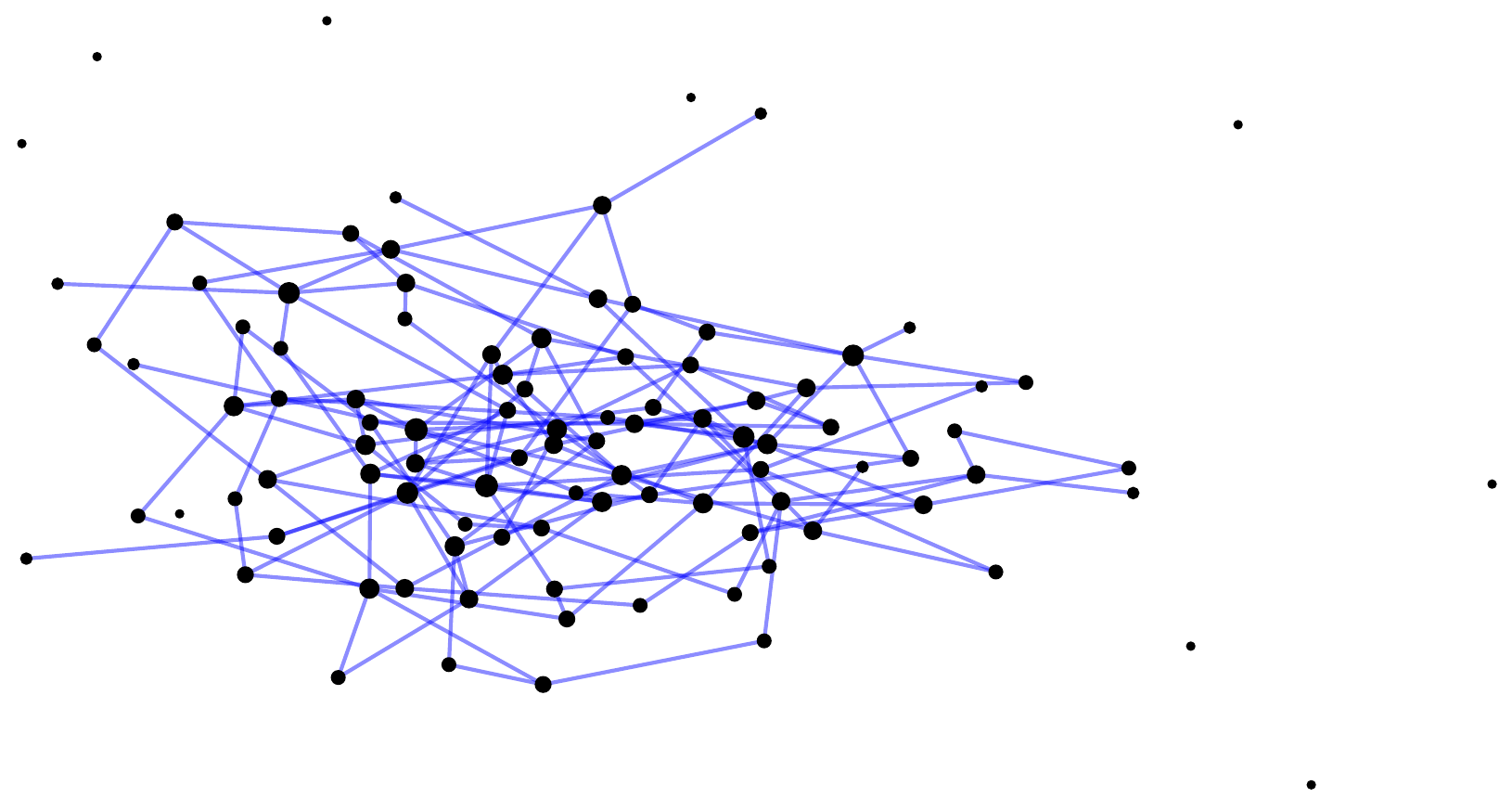}
		\caption{{\footnotesize $n = 100, p = 0.03$}}
	\end{subfigure}
	
	\caption{Some realizations of $\G(n,p)$.}
	\label{fig:intro:Gnp_samples} 
\end{figure} 

Note that the \ER model is in some sense the simplest model of random graphs one can ever think of: edges are drawn independently with the same probability, there is no particular geometry in the graph. An immediate result illustrating this absence of geometry is the following

\begin{lemma}
	Fix $0 \leq m \leq \binom{n}{2}$. Let $G \sim \G(n,p)$, conditioned to have $m$ edges. Then $G$ is uniform among all graphs with node set $[n]$ and $m$ edges.
\end{lemma} 

Many interesting results with high probability are known for \ER graphs, and literature investigating this model is abundant. For a general and thorough view on this very rich topic, the reader can refer to the books of Bollob\'as \cite{Bollobas2001}, Janson, Luczak and Rucinski \cite{Janson00} and Hofstadt \cite{hofstad2016}.

\subsubsection{High probability properties, first and second moment methods}
Some event $A$ depending on a size (or dimension) parameter $n$ is said to be verified \emph{with high probability (w.h.p.)} if the probability of $A$ tends to $1$ when $n \to \infty$.

We will start with merely giving one of the most elementary -- and famous -- results for the \ER model, which proof will be the occasion to introduce a basic technique, instrumental for solving many probabilistic questions in random graphs: the so-called \emph{first} and \emph{second-moment} methods (see e.g. \cite{ProbaMethod}).

\begin{lemma}[First moment method]\label{intro:lemma:first_moment_method}
Let $X$ be a non-negative, integer-valued random variable. Then
$$ \dP(X>0) \leq \dE[X]. $$
\end{lemma}

\begin{proof}
This is a consequence of Markov's inequality : for all $b>0$, $\dP(X \geq b) \leq \frac{\dE[X]}{b}$. Taking $b=1$ gives the desired result.
\end{proof}
In particular, in the case where $X$ depends on $n$, and $\dE\left[X\right] \to 0$ when $n \to \infty$, then Lemma \ref{intro:lemma:first_moment_method} implies that $X = 0$ with high probability.

\begin{lemma}[Second moment method\footnote{In this form, this result is known as the \emph{Paley–Zygmund inequality.}}]\label{intro:lemma:second_moment_method}
	Let $X$ be a real random variable with positive mean and finite variance. Then for all $0 \leq c \leq 1$,
	\begin{equation*}
	\dP\left(X \geq c \,{\dE\left[X\right]}\right) \geq \left(1 - c\right)^2 \frac{\dE\left[X\right]^2}{\dE\left[X^2\right]}.
	\end{equation*}
\end{lemma}

\begin{proof}
	Using Cauchy-Schwarz inequality,
	\begin{flalign*}
	\dE[X]  = \dE\left[X\one_{X < c \, {\dE\left[X\right]}} \right] + \dE\left[X\one_{X \geq c \, {\dE\left[X\right]}} \right]
	\leq c \, {\dE\left[X\right]} + \dE[X^2]^{1/2} \, \dP(X \geq c \, \dE\left[X\right])^{1/2},
	\end{flalign*} which gives $\dE[X]^2 \left(1 - c \right)^2 \leq \dE[X^2] \, \dP(X \geq c \, {\dE\left[X\right]})$.
\end{proof} 

In particular, in the case where $X$ depends on $n$, $\dE\left[X\right] \to \infty$ and $\dE\left[X^2\right] \sim \dE\left[X\right]^2$ when $n \to \infty$, taking $c\to0$ in Lemma \ref{intro:lemma:second_moment_method} implies that $X\geq o(\dE\left[X\right])$ with high probability and hence that $X \to \infty$ w.h.p.\\

Let us now state an elementary result that will be proven by appealing to these standard methods. A graph $G$ is \emph{connected} if for any $u \neq v \in G$, there exists a path from $u$ to $v$ made of edges of $G$. A node $u \in V$ is \emph{isolated} if it has no neighbors in $G$. 

\begin{theorem}[Connectivity of \ER graphs]\label{intro:theorem:connectivity_ER}
Let $G \sim \G(n,p)$ with $p$ depending on $n$. Then, with high probability,
\begin{itemize}
	\item[$(i)$] if $np \leq (1-\eps) \log n$ for some $\eps>0$, then $G$ contains isolated vertices and hence is not connected.
	\item[$(ii)$] if $np \geq (1+\eps) \log n$ for some $\eps>0$, then $G$ is connected.
\end{itemize}
\end{theorem} 

\begin{figure}[h]
	\centering
	\medskip
	\begin{subfigure}[t]{0.49\linewidth}
		\centering
		\includegraphics[scale=0.25]{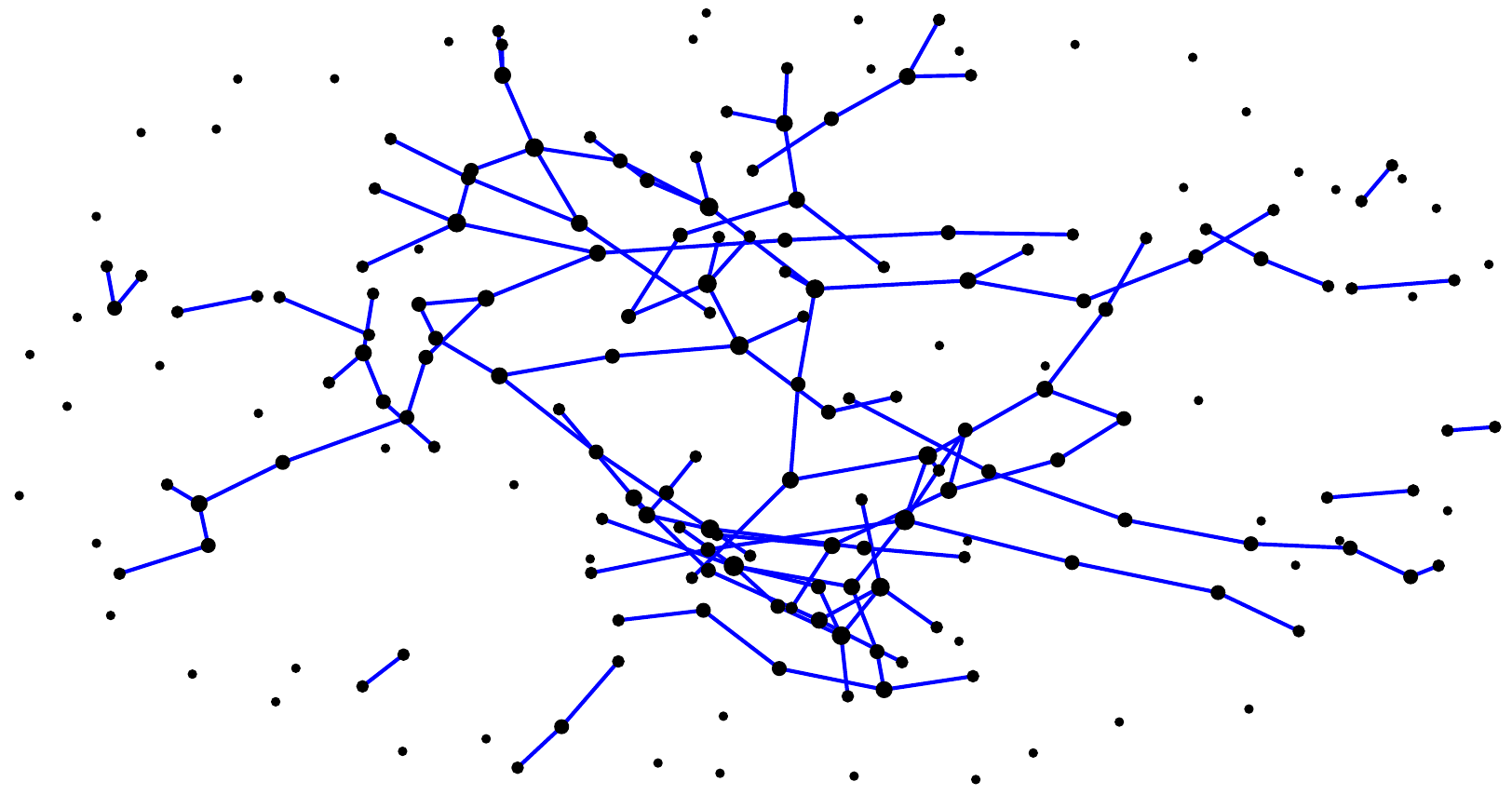}
		\caption{{\footnotesize $p = 1.3/n$}}
	\end{subfigure}
	\hfill
	\begin{subfigure}[t]{0.49\linewidth}
		\centering
		\includegraphics[scale=0.25]{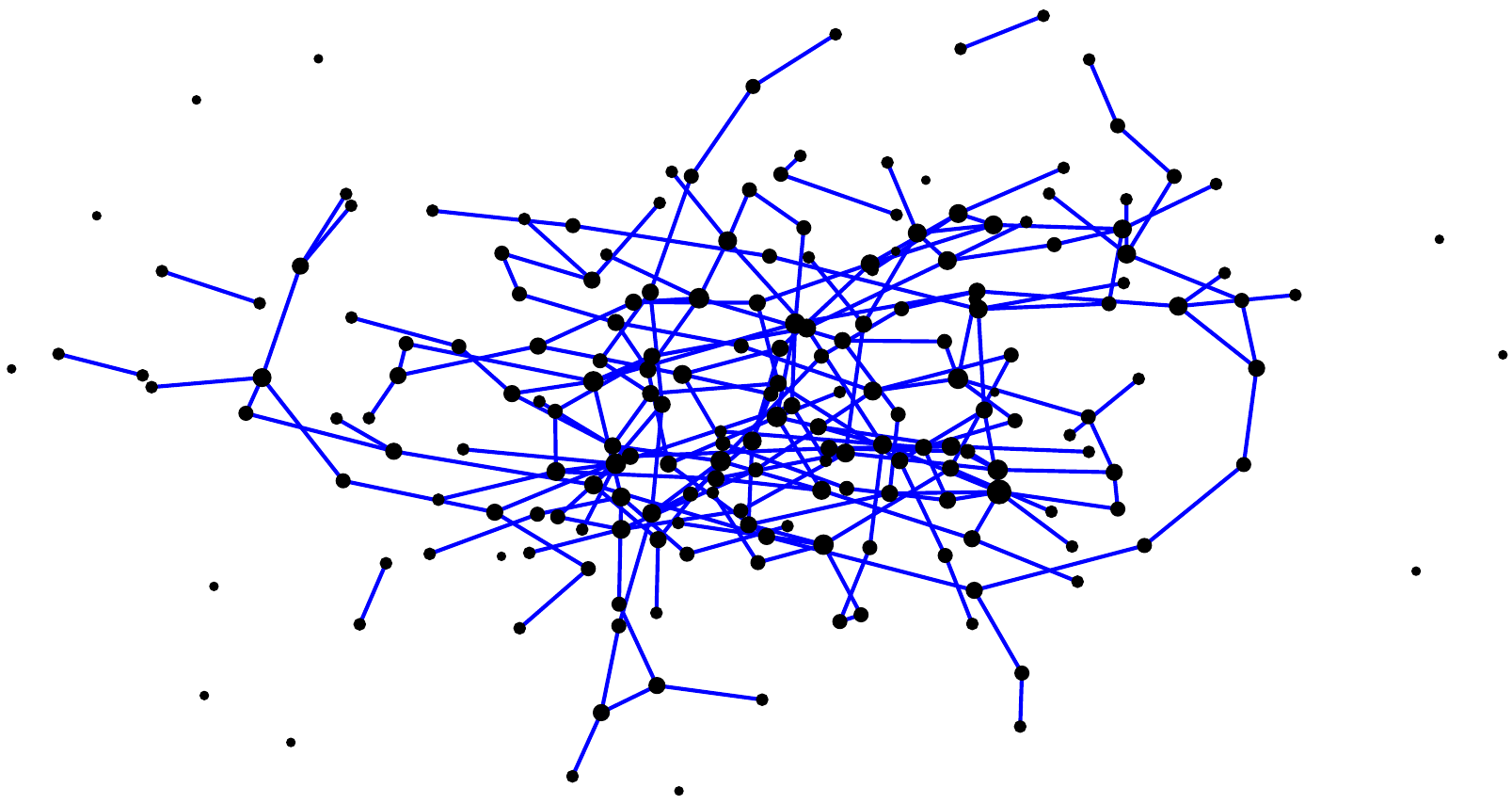}
		\caption{{\footnotesize $p = 2/n$}}
	\end{subfigure}

	\begin{subfigure}[t]{0.49\linewidth}
	\centering
	\includegraphics[scale=0.25]{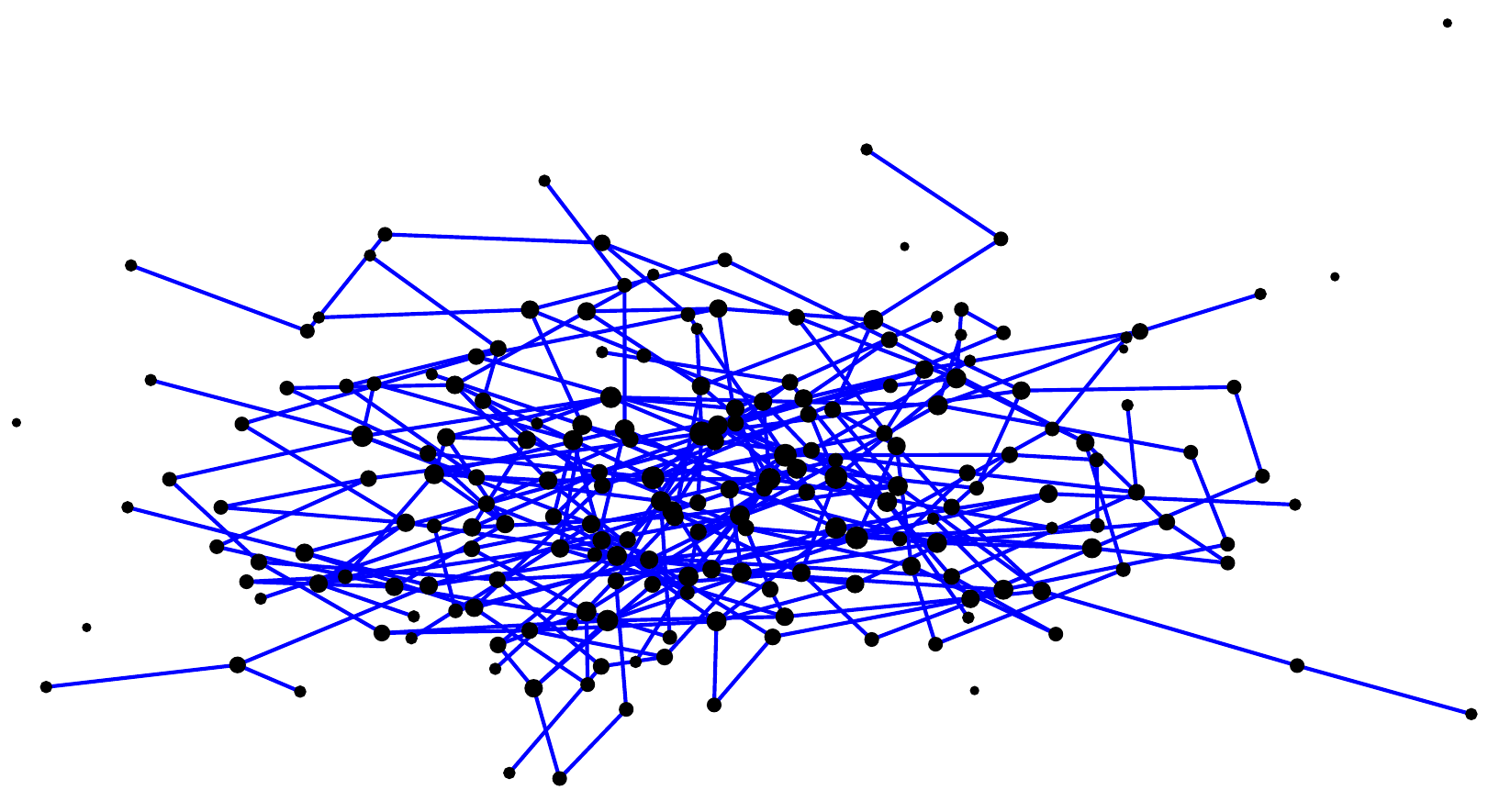}
	\caption{{\footnotesize $p = 3.3/n$}}
	\end{subfigure}
	\hfill
	\begin{subfigure}[t]{0.49\linewidth}
	\centering
	\includegraphics[scale=0.25]{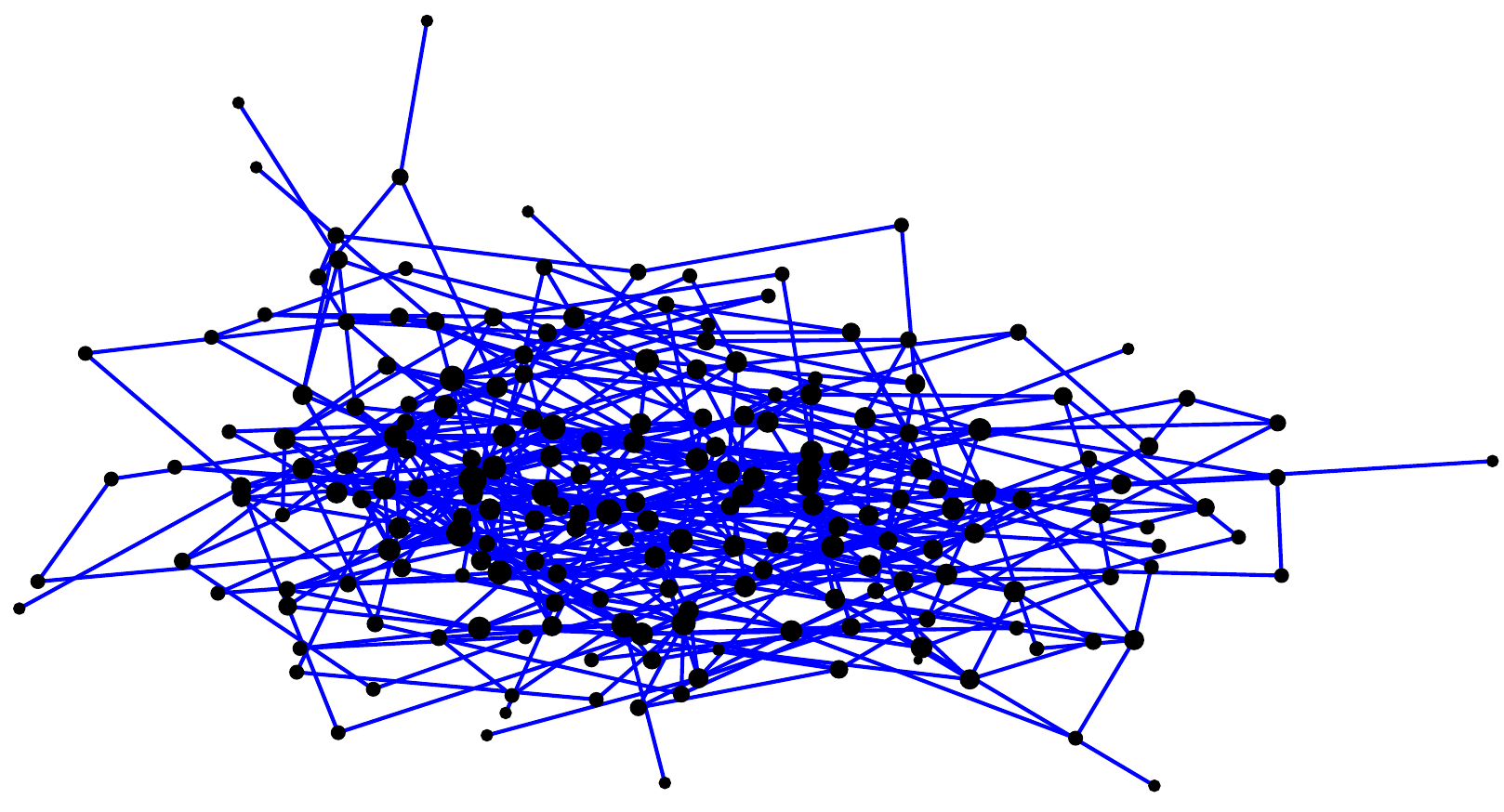}
	\caption{{\footnotesize $p = 3.9/n$}}
	\end{subfigure}
	
	\caption{Some realizations of $\G(n,p)$ with $n=200$, connected and disconnected.}
	\label{fig:intro:Gnp_connectivity} 
\end{figure} 

\begin{proof}
\proofstep{Proof of $(i)$.} We will use the second moment method for the proof of point $(i)$. Let us denote 
$$ X:= \card{\set{u \in V, u \mbox{ is isolated in } G}} = \sum_{u \in V} \one_{u \mbox{ {\footnotesize is isolated in} } G} \,. $$
For any $u \in V$, $\dP(u \mbox{ is isolated in } G) = (1-p)^{n-1}$, hence $\dE[X] = n(1-p)^{n-1}$ which equals $\exp((1+o(1))[\log n - np]) \geq \exp((1+o(1))\eps\log n )\to \infty$ under the assumption $np \leq (1-~\eps) \log n$. 

Let us now check that we indeed have $\dE\left[X^2\right] \sim \dE\left[X\right]^2$ when $n \to \infty$. 
\begin{flalign*}
	\dE[X^2] & = \sum_{u,v \in V} \dE\left[\one_{u, v \mbox{ {\footnotesize are isolated in} } G}\right]\\
	& = \dE[X] + \sum_{u,v \in V, u \neq v} \dP\left( u, v \mbox{ are isolated in } G \right)\\
	& = o(\dE[X]^2) + n(n-1)(1-p)^{n-1 + n-2} = (1+o(1))\dE\left[X\right]^2.
\end{flalign*}

\proofstep{Proof of $(ii)$.} Point $(ii)$ proved with the first moment method. Let us assume that $np \geq (1+\eps) \log n$. Define a \emph{zero-cut} of $G$ to be a partition of $V$ into two sets which are crossed by no edges. It is clear that $G$ is disconnected if and only if $G$ admits a non trivial zero-cut, the trivial one being the partition $\{V,\varnothing\}$. Let $Y$ de defined as the number of non trivial zero-cuts of $G$.
For a given partition $\{S,V \setminus S\}$ of $V$ into two sets of size $k$ and $n-k$, $\{S,V \setminus S\}$ is a zero-cut with probability $(1-p)^{k(n-k)}$, hence 
$$ \dE[Y] = \sum_{k=1}^{\lfloor n/2 \rfloor} \binom{n}{k} (1-p)^{k(n-k)}.$$
The right hand term being decreasing with $p$, we can assume without loss of generality that $p = (1+\eps) \frac{\log n}{n}$. Using $\binom{n}{k} \leq (en/k)^k$, and splitting the last sum at $k = \alpha n $, where $\alpha$ is to be specified later, we get the following upper bound:
\begin{flalign*}
\dE[Y] & \leq \sum_{k=1}^{\alpha n} \exp\left( k \left[\log(en/k) - (1-\alpha)(1+\eps)\log n \right] \right) + \sum_{k=\alpha n + 1}^{\lfloor n/2 \rfloor} \binom{n}{k} (1-p)^{\alpha n^2/2}\\
& \overset{(a)}{\leq} \sum_{k=1}^{\alpha n} \exp\left( - (c+o(1)) k  \log n \right) + 2^n e^{-\alpha(1+\eps) n \log n} = o(1),
 \end{flalign*} where $(a)$ is true as soon as $\alpha \in (0,1)$ and $c$ are such that $(1-\alpha)(1+\eps)-1>c >0$, which is true whenever $\alpha \in (0,\eps/(1+\eps))$. 
\end{proof}

\begin{remark}\label{intro:remark:disconnected_automorphisms}
Note that result of Theorem \ref{intro:theorem:connectivity_ER} shows that the asymptotic probability of connectivity in an \ER graph abruptly jumps from $0$ to $1$ when $np/(\log n)$ begins to exceed $1$: in this case it is agreed to say that the connectivity property exhibits a \emph{(sharp) threshold}. This remarkable fact is rather usual and occurs for a large range of properties in random graphs (see e.g. \cite{Bollobas2001, Janson00}). In inference problems, such underlying threshold phenomena involving parameters of the random models are often the cause of the emergence of so-called \emph{phases} (impossible, hard or easy), see Section \ref{intro:section:phases}.
\end{remark}

\subsection{The zoo of inference problems on graphs}\label{intro:subsection:zoo}

\subsubsection{Why planting signal?}
At first sight, the planted framework described earlier may leave the reader somewhat bemused; the question of finding some interesting information in data -- in the broadest meaning -- is very different from assuming that the data is \emph{literally} constructed out of some underlying signal, and aiming to recover it. We would like to start by emphasizing and elaborating on this important point.

Here are few words to clarify the above statement and release this apparent tension: for the overwhelming majority of inference problems in random graphs, the planted formulation is in fact a probabilistic rephrasing of an original deterministic combinatorial optimization problem, which we often refer to as the \emph{worst-case} version. The planted approach differs from the initial problem, but has the advantage of carrying it its very essence a notion of \emph{ground truth}, offering a comfortable framework for the evaluation of the performance of algorithms as well as a direct control on the signal-to-noise ratio. Also, theoretical guarantees can be obtained with high probability in planted models under less stringent constraints, taking into account the typical properties of data sampled from the model. Cris Moore echoes these statements in \cite{Moore17}, justifying this approach in the context of community detection in the following words:

\begin{flushright}
	\begin{quotation}
	\emph{``For the most part we are used to thinking about worst-case instances rather than random ones, since we want algorithms that are guaranteed to work on any instance. But why should we expect a community detection algorithm to work, or care about its results, unless there really are communities in the first place? And when Nature adds noise to a data set, isn’t it fair to assume that this noise is random, rather than diabolically designed by an adversary?''}
	\end{quotation}
\end{flushright}

This duality is believed to be fundamental and should be kept in mind when facing inference problems (on graphs). We will endeavour to shed light on the two flavours of the problems given as examples hereafter: a worst-case -- deterministic -- formulation, as well as a planted -- probabilistic -- counterpart, leading to different objectives and results.  

\subsubsection{A non-exhaustive bestiary} 
Without further ado, we will now give three representative examples of inference problems on graphs.

\textit{$(a)$ Max-clique, planted clique.} A \emph{clique} of a graph is a subset of vertices all adjacent to each other, i.e. a complete subgraph. The \emph{max-clique} problem consists in finding the maximum clique in a graph $G=(V,E)$, that is solving the following
\begin{equation}\label{eq:max_clique}
	\argmax_{\substack{S \subset V \\ S \mboxs{ is a clique}}} \; \card{S} \, .
\end{equation}

The max-clique problem, as well as the problem of deciding whether the graph contains a clique of given size is NP-hard \cite{Karp72}, as well as some of its approximations \cite{Hastad99}, unless $\mathsf{P = NP}$.

The planted version of max-clique, namely the \emph{planted clique} problem, is defined as follows. Consider two integers $n$ and $k \leq n$ possibly scaling with $n$. Let us generate a graph $G=(V,E)$ with vertex set $V=[n]$ as follows. First, a subset $K^\star \subset V$ of size $k$ is chosen uniformly among the $k-$subsets of $V$. $K^\star$ will form a clique: all possible edges between vertices of $K^\star$ are added to $E$. Then, all possible remaining edges are drawn independently with probability $1/2$. We denote this model by $\G_k(n,1/2)$. 

\begin{figure}[h]
	\centering
	\medskip
	\begin{subfigure}[t]{0.49\linewidth}
		\centering
		\includegraphics[scale=0.25]{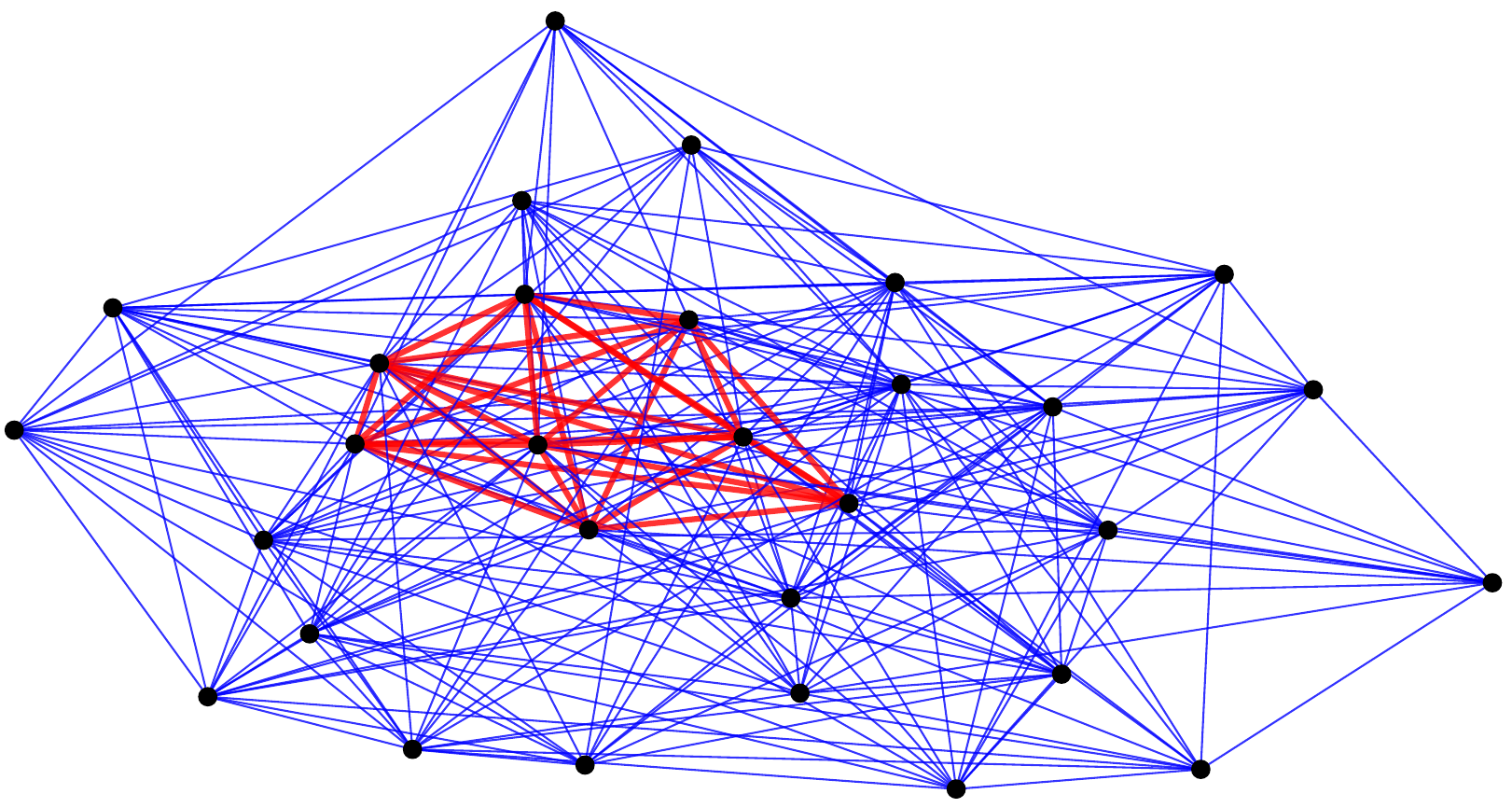}
		\caption{{\footnotesize $n = 30, k = 8$}}
	\end{subfigure}
	\hfill
	\begin{subfigure}[t]{0.49\linewidth}
		\centering
		\includegraphics[scale=0.25]{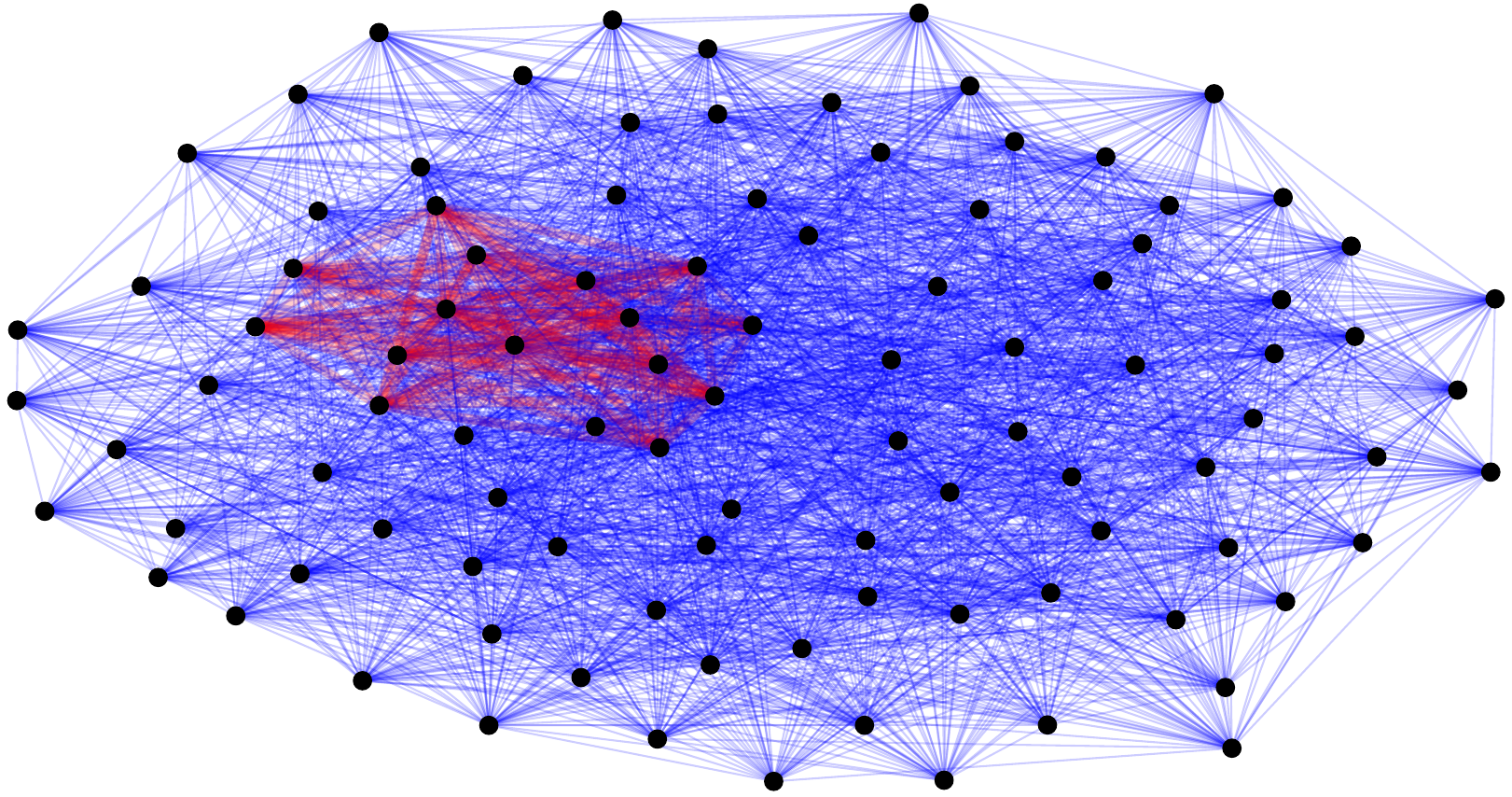}
		\caption{{\footnotesize $n = 100, k = 15$}}
	\end{subfigure}
	
	\caption{Some realizations of $\G_k(n,1/2)$, where the planted clique is highlighted in red}
	\label{fig:intro:planted_clique} 
\end{figure} 

Note that when $k=0$, that is in the absence of a planted clique, $G$ is simply distributed according to $\G(n,1/2)$. Given $k$, the statistical test for the detection task is thus
\begin{equation*}
\cH_0 := `` G \sim \G(n,1/2) " \; \mbox{versus} \; `` \cH_1 := G \sim \G_k(n,1/2)"  \,.
\end{equation*}

For reconstruction, the goal is to find an estimator $\hat{K}= \hat{K}(G)$ that recovers the whole true clique $K^\star$ with high probability, that is such that $\dP(\hat{K}(G)=K^\star) \to 1$ when $n \to \infty$ and $G \sim \G_k(n,1/2)$.

Note that under the planted model, the posterior distribution of $K^\star$ (conditionally on $G$) is given by 
$$ \dP \left(K^\star = K \cond G \right) = \frac{1}{\dP(G)} \left(\frac{1}{2}\right)^{n-k} \one_{K \mboxs{ is a clique in $G$}} \, ,$$ which shows that, unsurprisingly, the posterior distribution of $K^\star$ is the uniform distribution among all cliques of size $k$ in $G$. In the case where $k \geq (2+\eps)\log_2(n)$ it can be shown \cite{matula1972} that the only clique of size $k$ in $G$ is the planted one, and that it is the maximal clique, with high probability. Hence w.h.p. in this regime, the maximum a posteriori estimator of $K^\star$ is \emph{precisely} the solution of the max-clique problem \eqref{eq:max_clique} in $G$. This last statement makes the worst-case/planted duality even more explicit. For more insights on the performance of simple algorithms for this problem, we refer to the lecture notes by Wu and Xu \cite{WuXuLecturenotes}. \\ 

\textit{$(b)$ Min-bisection, community detection.} A \emph{bisection} of a graph $G$ is a partition of the vertex set $V$ into two sets of equal size (we assume that $\card{V}$ is even). In a weighted graph $G=(V,E)$ with adjacency matrix $A$, the \emph{min-bisection} problem  amounts to find a bisection with minimal crossing edge weights, that is solving
\begin{equation}\label{eq:min_bisection}
\argmin_{(V_1,V_2) \mboxs{ bisection of } G} \; \sum_{u \in V_1, \, v \in V_2} A_{u,v} \, .
\end{equation}

Finding the min-bisection of a graph is known to be NP-hard \cite{Garey74}. In the planted version of min-bisection, the random graph $G=(V,E)$ has to satisfy the following property: there is an underlying optimal partition of $V$ consisting in two subsets that are referred to as \emph{communities}. Therefore, the problem amounts to recovering these communities, which can very well be more than two in a general setting. In order for the graph to satisfy this property, it is sampled according to the celebrated \emph{stochastic block model}, originally introduced in \cite{Holland1983}, widely studied in recent threads of research \cite{Decelle11,Massoulie14,Mossel15,Mossel2018,Abbe18}. 

For a number of nodes $n \geq 0$, a number of blocks $r \geq 1$, a distribution $\alpha = (\alpha_i)_{i \in [r]}$ on $[r]$ and a $r \times r$ symmetric matrix $P$ with non-negative entries, the stochastic block model $\SBM(n,\alpha,P)$ is defined as follows.
First, draw independently for every node $u \in V=[n]$ a community (or type) $\chi^\star(u) \sim \alpha$. Then, every edge $\{u,v\}$ for $u \neq v \in V$ is present independently with probability $P_{\chi^\star(u), \chi^\star(v)}$.

\begin{figure}[h]
	\centering
	\medskip
	\begin{subfigure}[t]{\linewidth}
		\centering
		\includegraphics[scale=0.42]{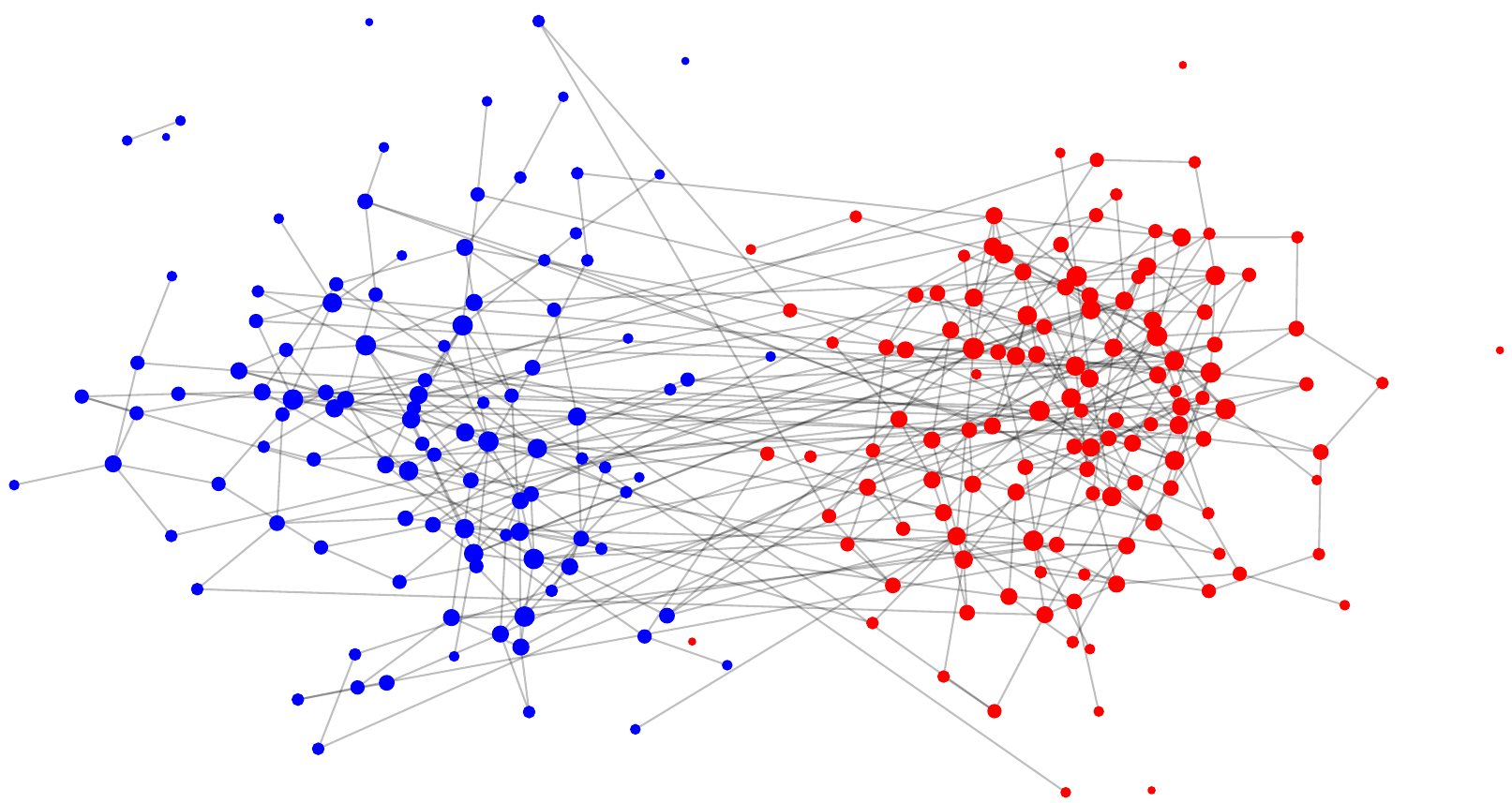}
		\caption{{\footnotesize $n = 230, r = 2, P_{c,c'} = 3.2 \cdot \one_{c=c'} + 0.7 \cdot \one_{c \neq c'}$}}
	\end{subfigure}
	\begin{subfigure}[t]{\linewidth}
		\centering
		\vspace{0.3cm}
		\includegraphics[scale=0.42]{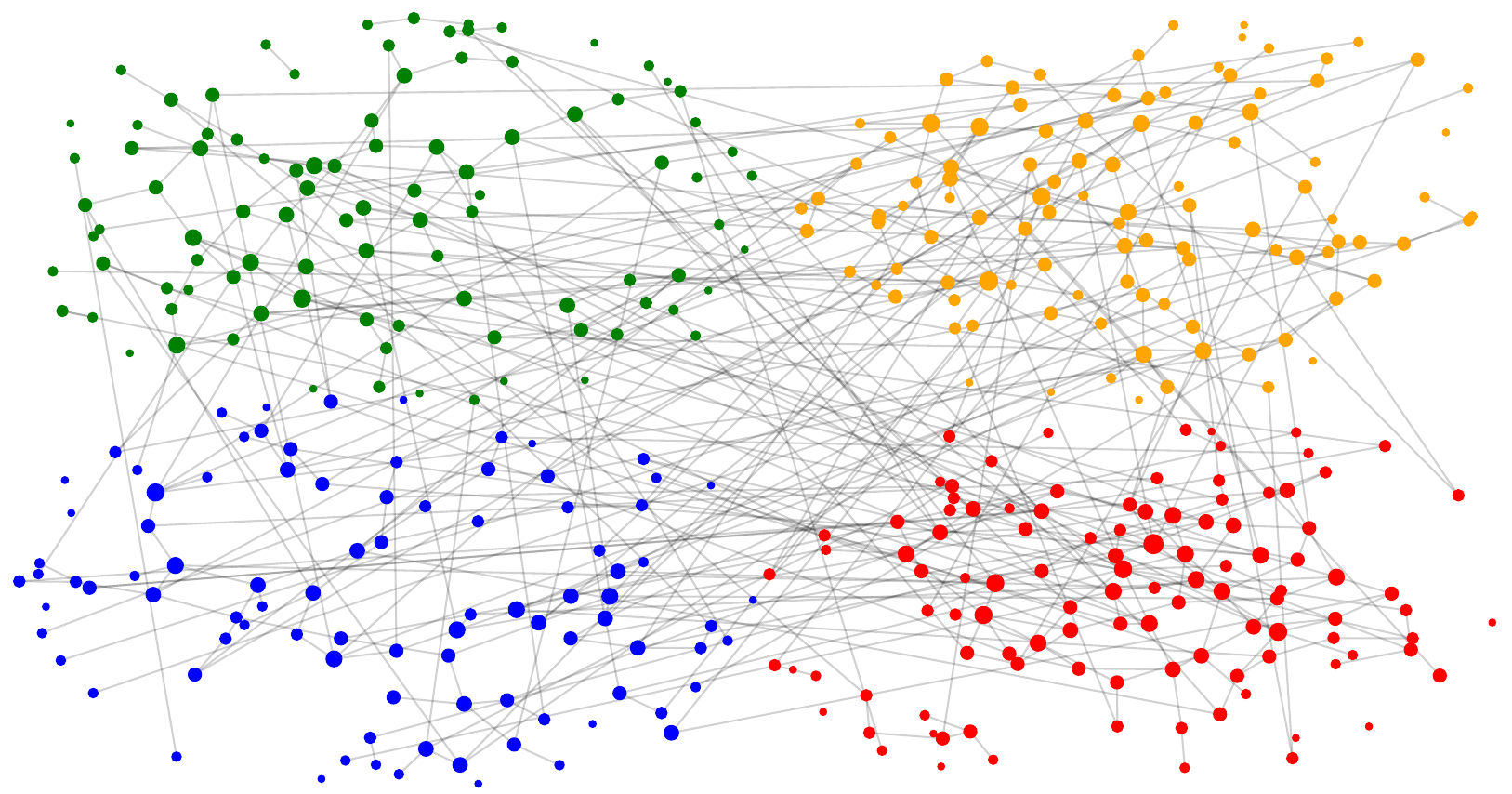}
		\caption{{\footnotesize $n = 400, r = 4, P_{c,c'} = 3.3 \cdot \one_{c=c'} + 0.5 \cdot \one_{c \neq c'}$}}
	\end{subfigure}
	
	\caption{Some realizations of $\SBM(n,\alpha,P)$. For both cases, $\alpha$ is the uniform distribution among the $r$ communities, and nodes are colored and placed accordingly.}
	\label{fig:intro:SBM} 
\end{figure} 

Note that when $r=1$ (single community), $G \sim \G(n,p)$ with $p := P_{11}$. Given $n$ and $P$, detection of planted communities consists in testing
\begin{equation*}
\cH_0 := `` G \sim \G(n,p) \mbox{ for some } p \in (0,1) " \; \mbox{versus} \; \cH_1 := `` G \sim \SBM(n,\alpha,P) " \,.
\end{equation*}

For reconstruction, we assert the performance of an estimator $\hat{\chi}= \hat{\chi}(G) : V \to [r]$ of the communities through its rescaled \emph{overlap} with the ground truth $\chi^\star$, defined by

\begin{equation*}
\ov(\hat{\chi},\chi^\star) := \frac{1}{n} \max_{\sigma \in \cS_r} \sum_{u \in V} \one_{\sigma \circ \hat{\chi}(u)=\chi^\star(u)} - \sum_{i \in [r]} \alpha_i^2 \,.
\end{equation*}

The second term in the right-hand side in the above ensures that  $\ov(\hat{\chi},\chi^\star)>0$ implies that the estimator $\hat{\chi}$ strictly outperforms random guess. Indeed, $\sum_{i \in [r]} \alpha_i^2$ is the expected fraction of good predictions achieved by the random guess estimator outputting communities drawn under the prior distribution $\alpha$.

The connection between the maximum a posteriori (MAP) estimator and the min-bisection problem can be illustrated in the standard case of two symmetric communities, in the sparse regime  with $\alpha = (1/2, 1/2)$ and $P = \begin{pmatrix} a/n & b/n \\ b/n & a/n \end{pmatrix}$ in the \emph{assortative} setting where $0 < b< a <1$. In this case, denoting by $S^\star := \set{u \in V, \chi^\star(u)=1}$, the posterior distribution of $S^\star$ under the stochastic block model writes
\begin{flalign*}
\dP \left(S^\star = S \cond G \right) & \propto \exp\left( \log\left(\frac{b/n}{a/n}\right) \sum_{u \in S, v \in V \setminus S} A_{u,v} + \log\left(\frac{1-b/n}{1-a/n}\right) \sum_{u \in S, v \in V \setminus S} (1-A_{u,v}) \right).
\end{flalign*} where $\propto$ stands for proportionality up to terms that do not depend on $S$. Neglecting the effect of non-edges, which is fair in the sparse regime (see \cite{Moore17}), since $0<b/a<1$ by assumption, and since $(S^\star, V \setminus S^\star)$ is w.h.p. close to a bisection of $G$, heuristically the MAP estimator of the two communities $(S^\star, V \setminus S^\star)$ is well approximated by the solution to the min-bisection problem \eqref{eq:min_bisection} in graph $G$. For an excellent survey on the subject with a statistical physics approach, we refer to \cite{Moore17}.\\

\textit{$(c)$ Min-weight perfect matching, planted matching.} 
Let $G=(V,E)$ be a graph with $\card{V} =2 n$ and adjacency matrix $A$. Assume that there is a partition $\set{V_0, V_1}$ of $V$ with $\card{V_0}=\card{V_1}=n$ such that every edge $\{u,v\} \in E$ satisfies $u \in V_0$ and $v \in V_1$ (we say that $G$ is \emph{bipartite}). A \emph{perfect matching} (p.m. hereafter) of $G$ is a set $M :=\set{e_1, \ldots, e_n}$ of $n$ edges of $E$ such that each node $u \in V$ appears exactly once in $M$. The \emph{weight} of a matching $M$ is defined by 
\begin{equation*}
\weight(M) := \sum_{e = \set{u,v} \in M} A_{u,v}
\end{equation*}
The \emph{min-weight perfect matching} problem writes
\begin{equation}\label{eq:min_weight_matching}
\argmin_{M \mboxs{ p.m. of } G} \; \weight(M) \, .
\end{equation}

Unlike the first two examples $(a)$ and $(b)$,  the min-weight perfect matching problem can be solved in polynomial-time, e.g. but the Hungarian algorithm \cite{Kuhn55} which runs in $O(n^3)$ time.
In the \emph{planted matching} problem   \cite{Semerjian20,Ding21_matching,Moharrami21}, the graph $G$ is taken to be a subgraph of a complete bipartite graph $K_{n,n}$, namely $V=[2n]$ and $E \subseteq \set{\set{u,v}, \, 1 \leq u \leq n, \, n+1 \leq v \leq 2n}$. A planted matching $M^\star$ is first picked uniformly at random from the set of perfect matchings of $K_{n,n}$. The remaining possible $n^2-n = n(n-1)$ edges are then sampled independently with probability $p$. Then, for all edge $e \in E$, edge weights are drawn independently from some distribution $\cP$ if $e \in M^\star$ and from another distribution $\cQ$ otherwise.

\begin{figure}[h]
	\centering
	\medskip
	\begin{subfigure}[t]{0.49\linewidth}
		\centering
		\includegraphics[scale=0.28]{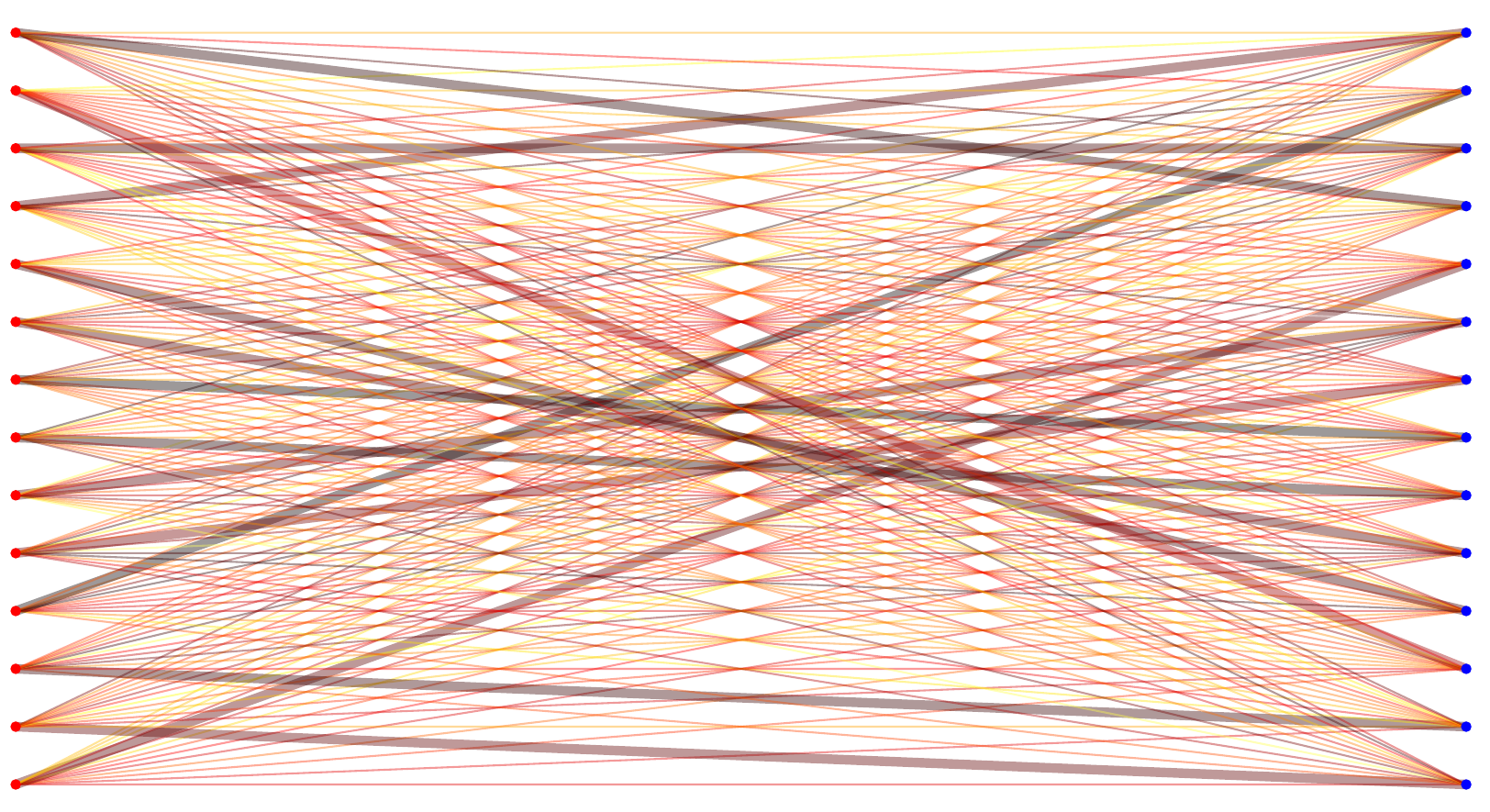}
		\caption{{\footnotesize $n = 14$, highlighted $M^*$}}
	\end{subfigure}
	\begin{subfigure}[t]{0.49\linewidth}
		\centering
		\includegraphics[scale=0.28]{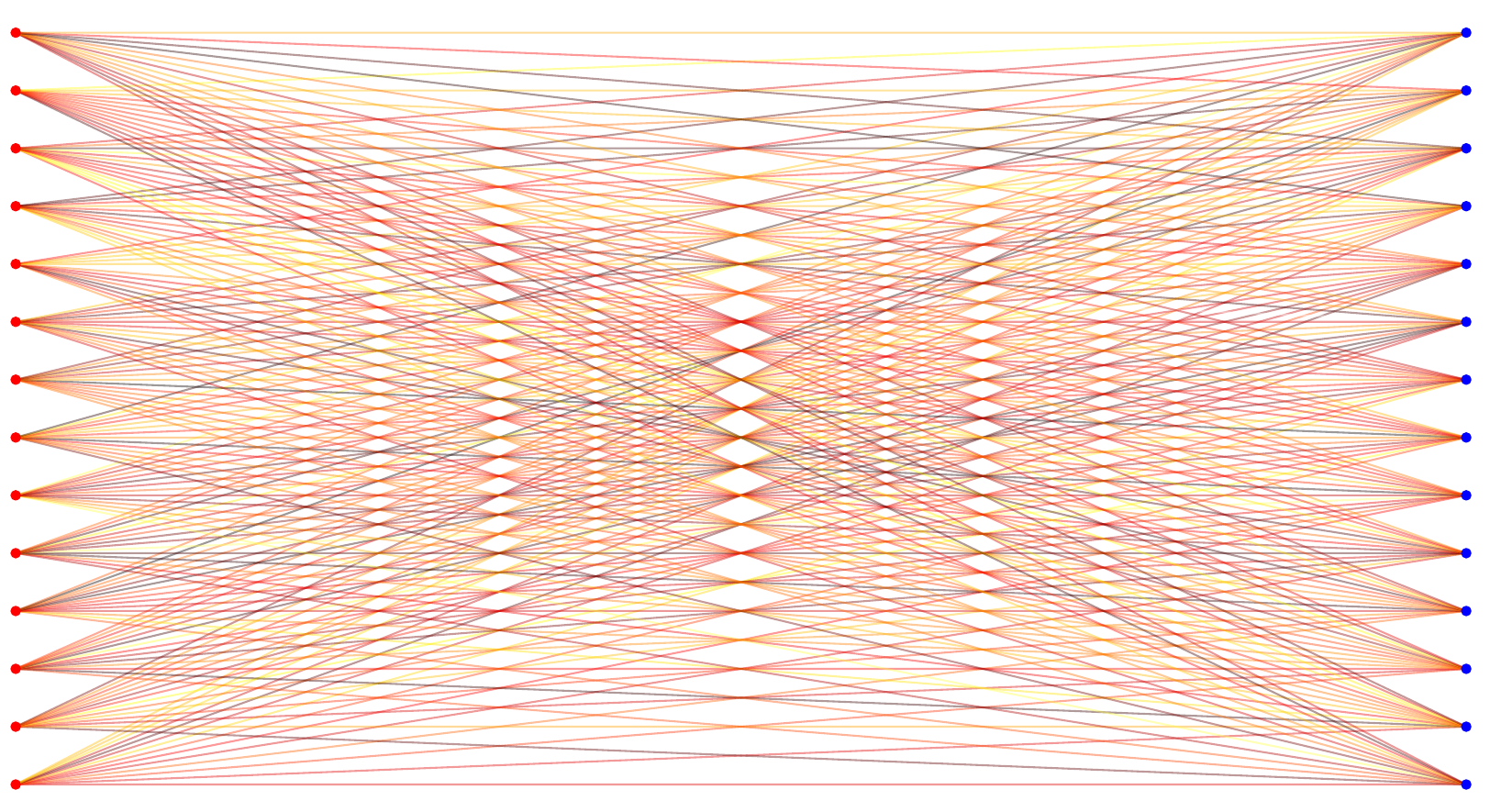}
		\caption{{\footnotesize $n = 14$, hidden $M^*$}}
	\end{subfigure}

	\begin{subfigure}[t]{0.49\linewidth}
	\centering
	\vspace{0.3cm}
	\includegraphics[scale=0.28]{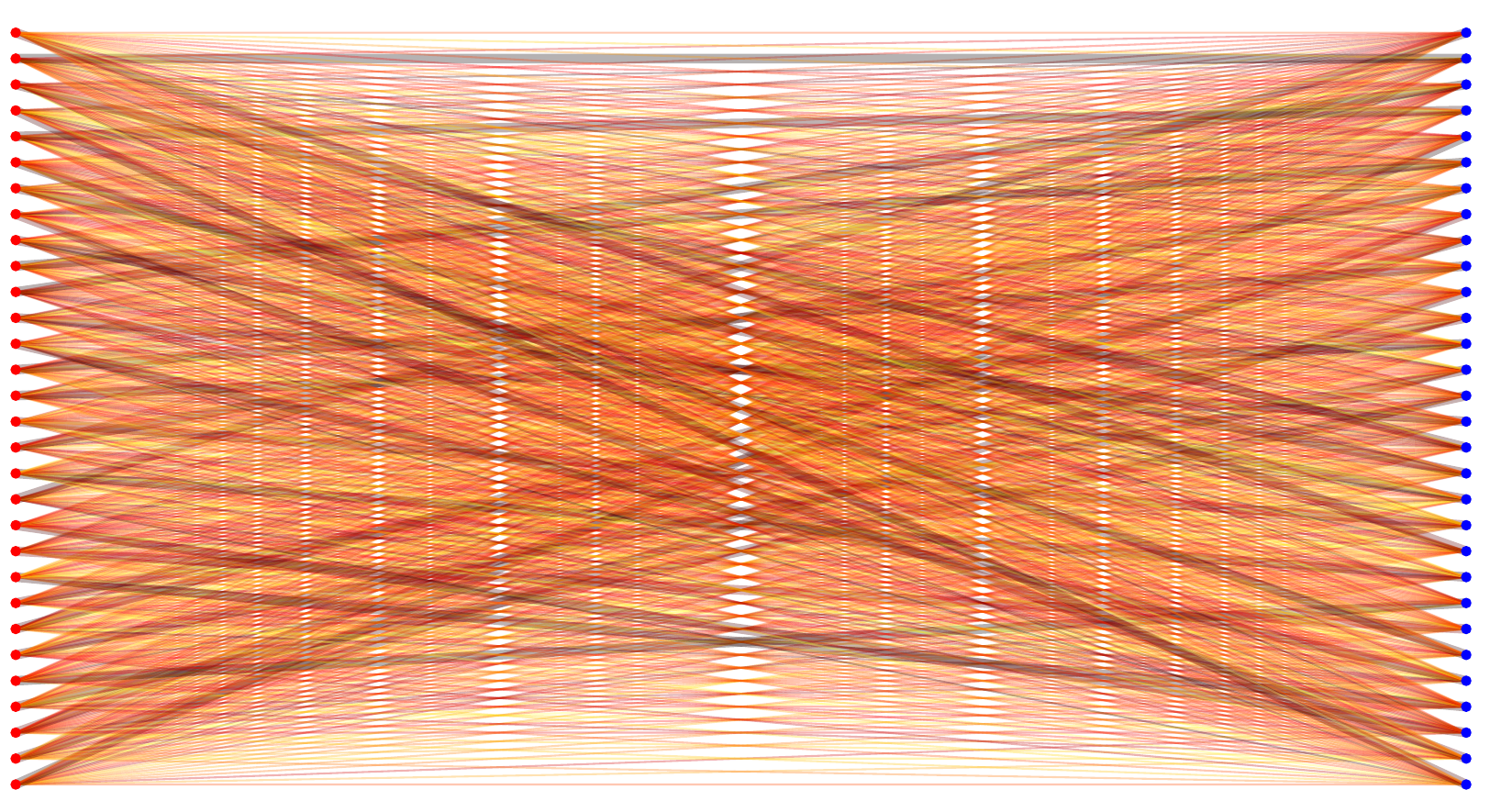}
	\caption{{\footnotesize $n = 30$, highlighted $M^*$}}
	\end{subfigure}
	\begin{subfigure}[t]{0.49\linewidth}
	\centering
	\vspace{0.3cm}
	\includegraphics[scale=0.28]{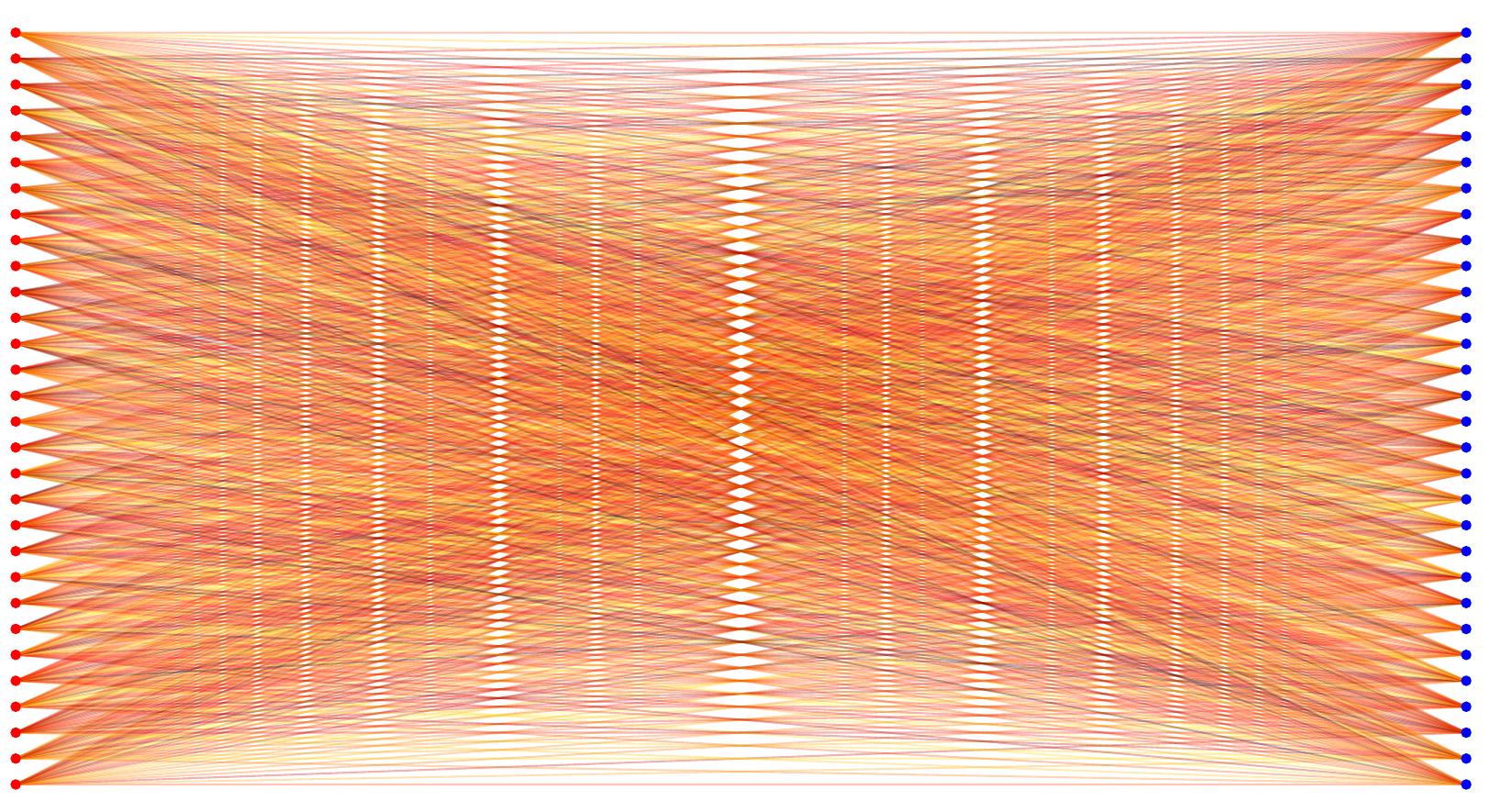}
	\caption{{\footnotesize $n = 30$, hidden $M^*$}}
	\end{subfigure}
	
	\caption{Some realizations of planted matchings on the complete bipartite graph $K_{n,n}$. For both cases, $\cP = \Exp(\mu_P)$, $\cQ = \Exp(\mu_Q)$ with $\mu_P = 4.2$ and $\mu_Q = 1/n$, and edges are colored according to their weights.}
	\label{fig:intro:planted_matching} 
\end{figure} 

Reconstructing the planted matching refers to finding an estimator $\widehat{M}= \widehat{M}(G)$ such that the overlap $\frac{1}{n} \card{\widehat{M} \cap M^\star}$ is as large as possible. Let us here again derive the posterior distribution of the signal in the planted matching model. Denoting by $\propto$ proportionality up to terms that do not depend on $M$, we have 
\begin{flalign*}
 \dP \left(M^\star = M \cond G \right) & \propto \exp\left( \sum_{e = \set{u,v} \in M} \log \cP(A_{u,v}) + \sum_{e = \set{u,v} \in E \setminus M} \log \cQ(A_{u,v}) \right)  \one_{M \mboxs{ is a p.m. of $G$}} \\
 & \propto \exp\left( - \sum_{e = \set{u,v} \in M} \log \frac{\cQ}{\cP}(A_{u,v})\right)  \one_{M \mboxs{ is a p.m. of $G$}}.
\end{flalign*} This gives immediately that once again, the MAP estimator of $M^\star$ is precisely the solution to the worst-case instance \eqref{eq:min_weight_matching} on the reweighted graph $\widetilde{G}$ such that for each edge $e = \set{u,v}$, $$ A(\widetilde{G})_{u,v} := \log \frac{\cQ}{\cP}(A(G)_{u,v}). $$ Note that in particular, if $\cP = \Exp(\mu_P)$ and $\cQ = \Exp(\mu_Q)$ with $\mu_Q < \mu_P$ -- which is the model considered in \cite{Semerjian20,Ding21_matching,Moharrami21}, see Figure \ref{fig:intro:planted_matching} -- then we exactly have $\log \frac{\cQ}{\cP}(A_{u,v}) = c + (\mu_P-\mu_Q) A_{u,v}$ with some constant $c$, and hence the MAP estimator is exactly the solution to the min-weight perfect matching problem \eqref{eq:min_weight_matching} directly on $G$.\\

We close this short glimpse on the bestiary by mentioning other planted structures in graphs which have been recently studied, such as trees \cite{stephan19}, colorings \cite{Roee16}, or hamiltonian paths \cite{Bagaria20}.

After having given these classical examples, we are now ready to elaborate about some asymptotic (high-dimensional) phenomena that arise in these inference problems, namely the emergence of some regimes of model parameters scaling with the dimension $n$, where the task -- reconstruction or detection -- turns out to be impossible, hard or easy. 

\subsection{Impossible, hard and easy phases}\label{intro:section:phases}

\subsubsection{Definitions} 

Let us consider an inference task (e.g. detection, reconstruction) in a planted model where the data -- not necessarily graphs -- is sampled from a parametric distribution with parameters $\theta \in \Theta$.

\begin{itemize}
	\item The \emph{impossible phase} (or \emph{impossible regime}) is defined as a subset $\Theta_{\mathrm{impossible}}$ of the set of parameters $\Theta$ such that for all $\theta \in \Theta_{\mathrm{impossible}}$, provably no algorithm can perform the task with high probability.
	
	\item The \emph{easy phase} (or \emph{easy regime}) is the regime of parameters $\Theta_{\mathrm{easy}}$ where the task can provably be solved \emph{with a polynomial-time algorithm}, with high probability.
	
	\item The \emph{hard phase} (or \emph{hard regime}) is the regime $\Theta_{\mathrm{hard}}$ where some exhaustive,  non-polynomial search probably works -- namely $\Theta_{\mathrm{hard}} \cap \Theta_{\mathrm{impossible}} = \varnothing$ --  but where no polynomial-time is known to succeed with high probability.
\end{itemize} 
The understanding -- and the pinning down -- of these three regimes, gathered in a so-called \emph{phase diagram}, is of course of paramount importance for the understanding of inference problems, their related algorithms, and has thus been the subject of many recent threads of research (see e.g. \cite{Bubeck18} for a unified view).

\subsubsection{A toy example (2/2): finding hay in a haystack}

The remarks made earlier in our simple toy example (see 'A toy example (1/2)') can be specified to fit this context. For the (strong) detection task, the impossible phase covers the regime where $k = O(\sqrt{n})$. When $k = \omega(\sqrt{n})$, counting the occurrences of ones takes $O(n)$ time and enables to detect the presence of signal with high probability, and the task is easy.

The (partial) reconstruction tasks has however a much larger impossible phase ($k = o(n)$), and when $k = \Theta(n)$, the optimal method still consists in choosing $k$ positions at random among the positions of ones, which is done in $O(n)$ time to achieve partial recovery, which is never interesting in this case since $k$ is already of order $n$.


%
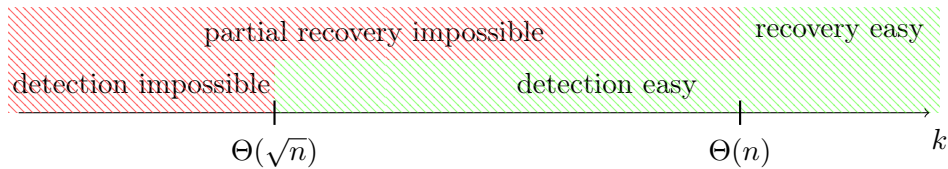
\begin{figure}[H]\label{fig:intro:haystack_phase_diagram}
\vspace*{0.4em}
\centering
\begin{tikzpicture}[scale=1.75]
\node[fill=white] (L) at (0,0) {};
\node[fill=white] (R) at (7,0) {};
\draw[black,->] (L) -- (R)  ;
\node[fill=white] at (7,-0.2) {$k$};

\node[fill=white] at (1,0.2) {detection impossible};
\fill[pattern=north west lines,pattern color=red!80!white,opacity=0.85] (0,0)--(0,0.4)--(2,0.4)--(2,0);

\node[fill=white] at (4.5,0.2) {detection easy};
\fill[pattern=north west lines,pattern color=green!80!yellow,opacity=0.7] (2,0)--(2,0.4)--(7,0.4)--(7,0);

\node[fill=white] at (2.75,0.6) {partial recovery impossible};
\fill[pattern=north west lines,pattern color=red!80!white,opacity=0.85] (0,0.4)--(0,0.8)--(5.5,0.8)--(5.5,0.4);

\node[fill=white] at (6.25,0.6) {recovery easy};
\fill[pattern=north west lines,pattern color=green!80!yellow,opacity=0.7] (5.5,0.4)--(5.5,0.8)--(7,0.8)--(7,0.4);

\node[scale=0.1,fill=white] (DU) at (2,0.1) {};
\node[scale=0.1,fill=white] (DD) at (2,-0.1) {};
\draw[thick,black] (DU) -- (DD)  ;
\node[fill=white] at (2,-0.3) {$\Theta(\sqrt{n})$};

\node[scale=0.1,fill=white] (RU) at (5.5,0.1) {};
\node[scale=0.1,fill=white] (RD) at (5.5,-0.1) {};
\draw[thick,black] (RU) -- (RD)  ;
\node[fill=white] at (5.5,-0.3) {$\Theta(n)$};

\end{tikzpicture}
\caption{Phase diagram for detection and reconstruction in the 'find hay in a haystack' problem}
\end{figure}

Note that in this (very simple) case, there is no hard phase neither for detection nor reconstruction. However, a considerable variety of inference problems are conjectured to exhibit a hard phase in their phase diagram. Planted clique may be the most appealing example. A significant amount of recent contributions \cite{Gamarnik2019TheLO,Deshpande13,Feige10,Jerrum92,Barak16} has agreed on the fact that no polynomial-time algorithm is known to recover a planted clique smaller than $\Theta(\sqrt{n})$, even though as discussed in $(a)$, an exhaustive -- non polynomial -- search recovers a planted clique of size $k$ as soon as $k \geq (2+\eps) \log_2(n)$. The phase diagram for reconstruction in the planted clique problem is hence as follows:


%
\begin{figure}[H]\label{fig:intro:clique_phase_diagram}
\vspace*{0.4em}
\centering
\begin{tikzpicture}[scale=1.75]
\node[fill=white] (L) at (0,0) {};
\node[fill=white] (R) at (7,0) {};
\draw[black,->] (L) -- (R)  ;
\node[fill=white] at (7,-0.2) {$k$};

\node[scale=0.1,fill=white] (DU) at (1.5,0.1) {};
\node[scale=0.1,fill=white] (DD) at (1.5,-0.1) {};
\draw[thick,black] (DU) -- (DD)  ;
\node[fill=white] at (1.5,-0.3) {$2 \log_2(n)$};

\node[scale=0.1,fill=white] (RU) at (4.5,0.1) {};
\node[scale=0.1,fill=white] (RD) at (4.5,-0.1) {};
\draw[thick,black] (RU) -- (RD)  ;
\node[fill=white] at (4.5,-0.3) {$\Theta(\sqrt{n})$};

\node[fill=white] at (0.75,0.2) {impossible};
\fill[pattern=north west lines,pattern color=red!80!white,opacity=0.85] (0,0)--(0,0.4)--(1.5,0.4)--(1.5,0);

\node[fill=white] at (3,0.2) {hard};
\fill[pattern=north west lines,pattern color=orange,opacity=0.85] (1.5,0)--(1.5,0.4)--(4.5,0.4)--(4.5,0);

\node[fill=white] at (5.75,0.2) {easy};
\fill[pattern=north west lines,pattern color=green!80!yellow,opacity=0.85] (4.5,0)--(4.5,0.4)--(7,0.4)--(7,0);
\end{tikzpicture}
\caption{Phase diagram for reconstruction in planted clique -- see $(a)$}
\end{figure}
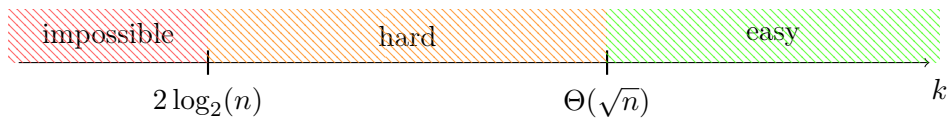

For another example where this phase transition of this type also appears, we can mention community detection for $3$ or more communities (see e.g. Abbé's survey \cite{Abbe18} on the subject). \\

Existence of a hard phase is also a fundamental and central question for the graph alignment problem, which we are now ready to introduce.

\section{Graph alignment}\label{intro:section:ga}

After this short introduction to the general topic of inference in random graphs, let us now dive into the core of this thesis, namely the graph alignment problem, which will be the subject of interest in the upcoming chapters.

\subsection{Motivations}\label{intro:subsection:motivations}
Graph alignment\footnote{The same problem is sometimes found under the name \emph{graph matching}. However, for the sake of clarity, we will only refer to graph alignment throughout the manuscript, in order not to confuse the reader with the different problem of planted matching evoked earlier (see Section \ref{intro:subsection:zoo}, example $(b)$).} (or network alignment) aims to answer the following (informal) question: \emph{‘what is the best way to match the nodes of two graphs?'}
\begin{figure}[h]\label{fig:intro:two_graphs}
	\centering
	\medskip
	\begin{subfigure}[t]{0.49\linewidth}
		\centering
		\hspace{-0.6cm}
		\includegraphics[scale=0.29]{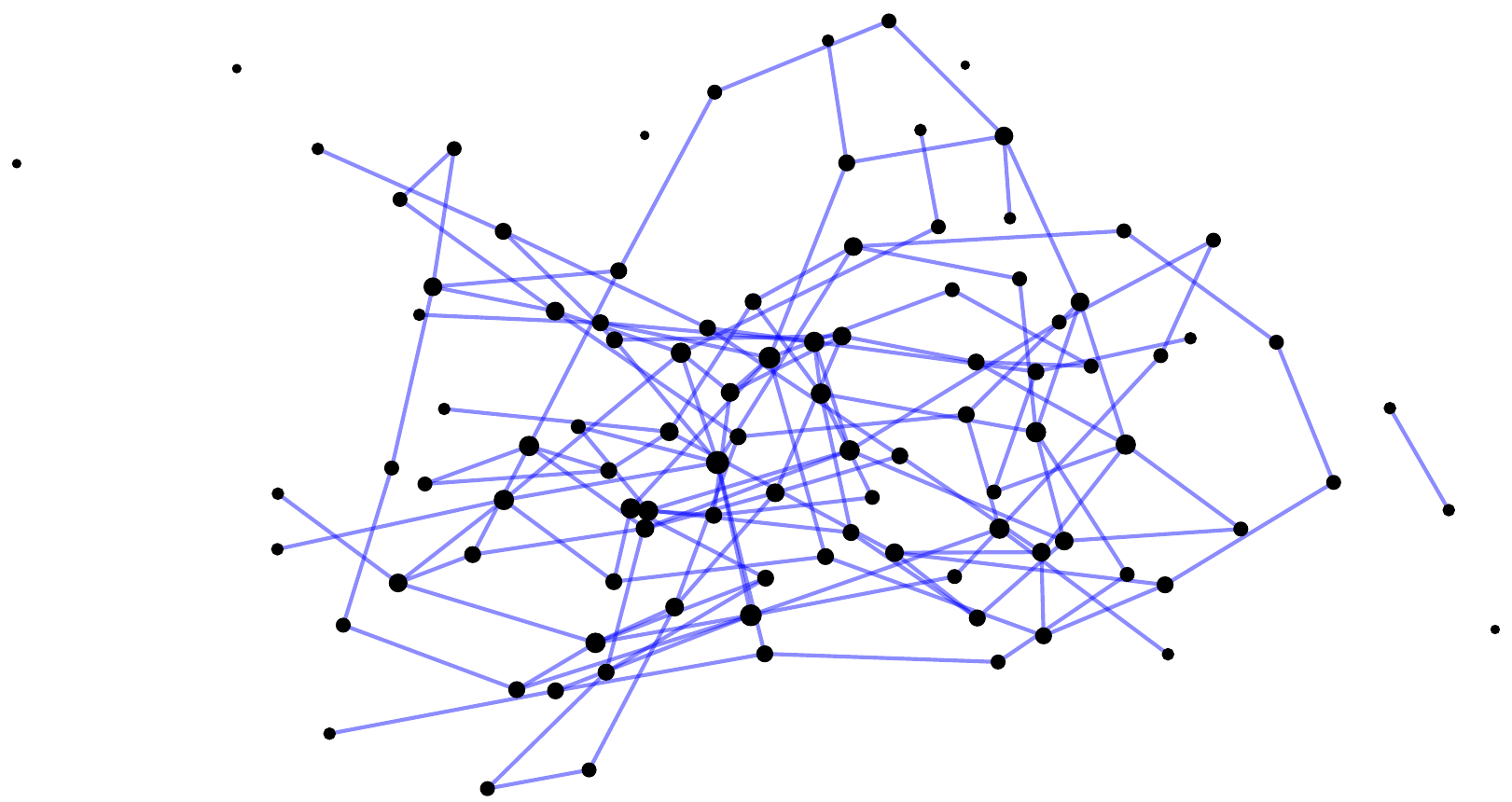}
	\end{subfigure}
	\begin{subfigure}[t]{0.49\linewidth}
		\centering
		\includegraphics[scale=0.29]{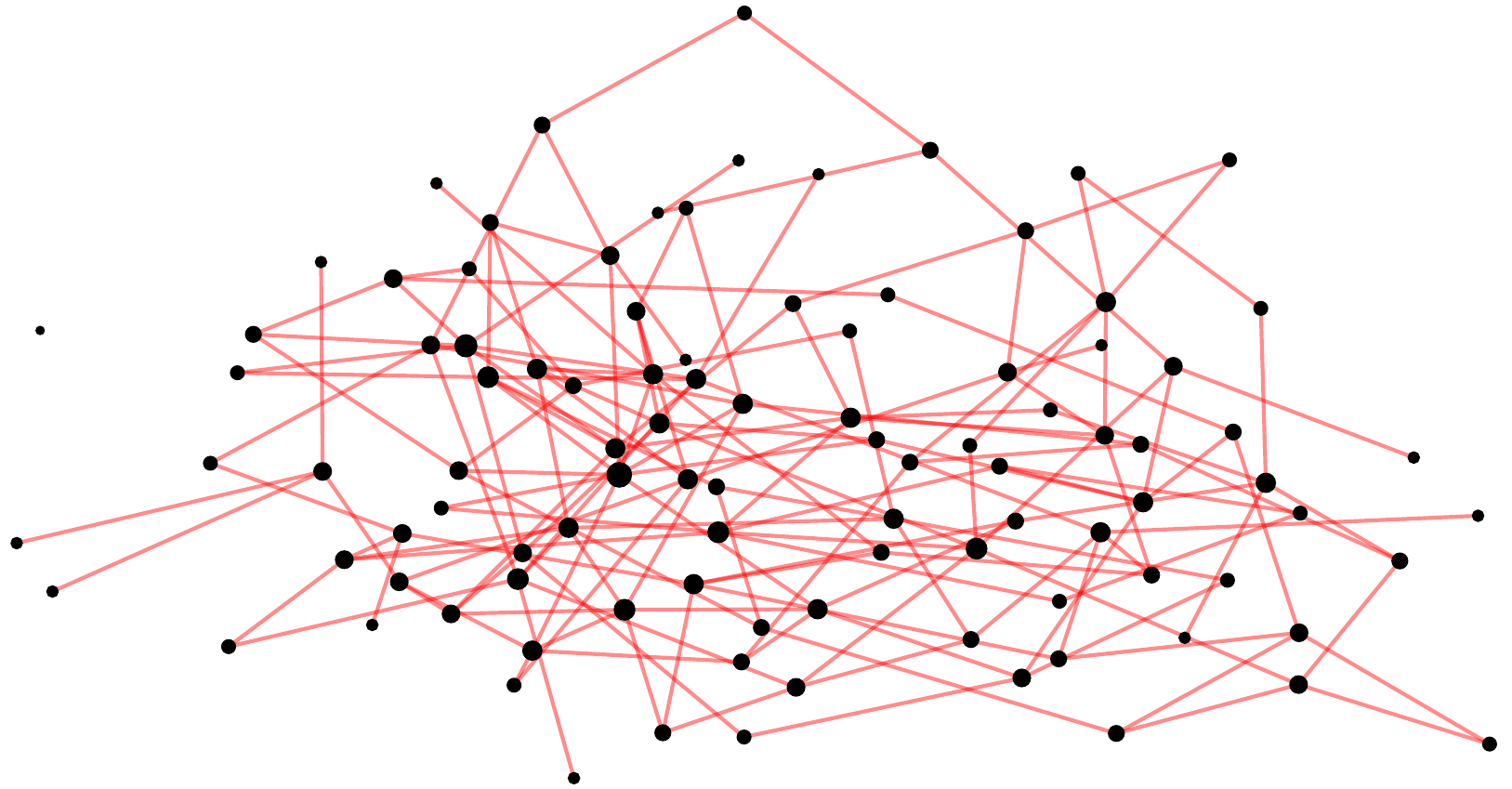}
	\end{subfigure}
	\caption{Graph alignment consists in the following informal question:  ``what is the best way to match the nodes of the two graphs?''}
\end{figure}
Providing an answer to this very general query amounts to exhibiting a vertex correspondence, or \emph{alignment}, between two -- labeled or unlabeled -- graphs so that the aligned -- labeled -- versions of the graphs are resembling or close enough, for some well-defined distance.

Motivations for the investigation of this problem are numerous, since many questions from various fields can be phrased as graph alignment problems. Let us now give a short overview of several applications for the graph alignment framework.

\begin{itemize}
	\item \emph{(De-)anonymization.} De-anonymization problems in networks aroused great interest when Narayanan and Shmatikov \cite{Narayanan08} were
	able to de-anonymize an unlabeled dataset of film ratings (subsampled from the Netflix dataset \cite{NET06}) with the help of auxiliary information given by the observation of a publicly available database (namely IMDb, the Internet Movie Database \cite{IMBD07}). The authors proposed a simple method relying on the correspondence (or, the correlation) between movies and ratings across the databases, that were able to match some pairs of records and thus to recover the entire movie viewing history of a given subscriber, which may be in turn used as input to uncover their political preferences or other sensitive information. 
	
	Since then, de-anonymization problems have been studied in recent literature in several versions and reformulations: related topics such as quantifying privacy
	issues related to databases \cite{Dwork08} or social networks \cite{Narayanan09} have been investigated.
	
	\item \emph{Image processing and pattern recognition.} Some recognition tasks in image processing such as shape matching and object recognition \cite{Berg05} can be achieved by finding correspondences between feature points across two (or several) images. The similarities are based on correspondence between vertex features as well as the cost of geometric transformation between pairs of nodes. 
	 
	Some popular algorithms for graph matching have been widely proposed for pattern recognition (see \cite{CFVS04} for a broad survey) in many areas since the late seventies: 2D/3D image analysis \cite{madi2016}, document processing, video analysis, biometric identification as well as biomedical/biological applications. All these fields have in common that some structured information is represented by graphs, and the goal is to find a correspondence that somehow ensures that substructures in the first graph are mapped to similar substructures in the other.
	
	\item \emph{Protein interaction networks in computational biology.} Authors of \cite{Singh07, Singh08} study protein-protein interaction represented as labeled networks, namely PPIs networks, where the nodes are proteins and edges represent interaction. These networks are observed across different species. They provide an algorithm, IsoRank, encoded as an eigenvalue problem, which performs a global alignment of two or more PPIs networks, using both the network structure of the data and sequence similarity. 
	
	Aligning these networks proves to be a very valuable tool: first, it provides a phylogenetic function-oriented comparison of proteins across different species, identifying those that may play the same role, thus transferring knowledge and insights across species; second, it can be used to perform \emph{ortholog prediction}, that is being able to spot genes that derive from the same ancestor.   
	
	Some following works elaborated on approximations \cite{Kazemi16} or refined versions \cite{Liao2009} of IsoRank, and also developed further competitive methods for this problem \cite{Kollias2013,elkebir2015}.
	
	\item \emph{Natural language processing and semantic entailment.} A fundamental task in natural language processing (NLP) is the recognition of \emph{semantic entailment}, that is, given a piece of text, whether an hypothesis can be concluded by logical implication, or simply by general world-knowledge. 
	
	In \cite{Haghighi05}, the authors use a representation of sentences as directed, labeled graphs between words and phrases, originally introduced in \cite{lin2001dirt} (for a recent general survey on graph representations in NLP, see \cite{osman20}). In these small networks, edges encode underlying dependency relationships in the sentence. Given a sentence and an hypothesis, the proposed strategy in order to identify entailment is to represent both the sentence and the hypothesis as graphs, and then to measure similarity between them, that is to find a mapping in the two graphs minimizing a score built from both the semantic resemblance of the matched vertices and how well the edges (namely, the relationships) are preserved by the mapping.
\end{itemize} 
Graph alignment recently grew some new interest in other applied fields, including computational neurosciences \cite{Frigo2021}, analysis of data from diffusion magnetic resonance imaging \cite{Olivetti16}, and cross-lingual knowledge alignment \cite{Muhao17}.

We refer to Section \ref{intro:subsection:short_survey} for a brief history of theoretical aspects and results in graph alignment.

\subsection{The quadratic assignment problem} 

Given two graphs $G=(V,E)$, $G'=(V',E')$ with same number of vertices $n=|V|=|V'|$, the problem of graph alignment consists in identifying a bijective mapping, or alignment $\pi: V\to V'$ that minimizes 
\begin{equation}\label{eq:QAP_edge_disagreements}
\sum_{i,j \in V} \left(\one_{\set{i,j} \in E} -  \one_{\set{\pi(i),\pi(j)} \in E'}\right)^2 \, , 
\end{equation} that is the number of disagreements between adjacencies in the two graphs under the alignment $\pi$. In the case where the two graphs are isomorphic, the two node sets $V$ and $V'$ can be matched perfectly: an isomorphism between $G$ and $G'$ achieves zero cost in \eqref{eq:QAP_edge_disagreements}.

However, we are interested in graph alignment for general, non necessarily isomorphic graphs: the problem can hence be viewed as the noisy version of the isomorphism problem. 

Given the adjacency matrices $A$ and $B$ of the two graphs $G$ and $G'$, the graph matching problem can be phrased as an instance of the quadratic assignment problem (QAP) \cite{Pardalos94} which is the following
\begin{equation}\label{eq:QAP}
\argmin_{\Pi \in \cS_n} \| A - \Pi B \Pi^\top \|^2  = \argmax_{\Pi \in \cS_n} \langle A, \Pi B \Pi^\top \rangle \, ,
\end{equation}
where $\Pi$ ranges over all $n\times n$ permutation matrices, and $\langle \cdot, \cdot \rangle$ denotes the matrix Frobenius inner product, i.e. $\langle C,D \rangle := \Tr(C^\top D)$. 

In a more general setting, including that of applications discussed in Section \ref{intro:subsection:motivations}, the loss function can also take into account a matching cost for pairs of vertices, namely the problem becomes 

\begin{equation}\label{eq:QAP_with_cost}
\argmax_{\Pi \in \cS_n} \langle A, \Pi B \Pi^\top \rangle + \langle C, \Pi  \rangle \, ,
\end{equation} where $C$ is a $n \times n$ matrix such that the cost for matching vertex $u \in V$ and $u' \in V'$ is given by $-C_{u,u'}$.

Under its general formulation, QAP is known to be a NP-hard problem, as well as some of its approximations \cite{Pardalos94,Makarychev14}. These hardness results are applicable in the worst case, where the observed graphs are designed by an adversary. In line with the worst-case/planted duality detailed earlier in Section \ref{intro:subsection:zoo}, a natural idea is then to study the average-case formulation, when $A$ and $B$ are random instances. 

\subsection{Planted graph alignment}\label{intro:subsection:planted_GA}
We now study the planted version of graph alignment, namely the \emph{planted graph alignment problem}, where the pair of graphs $(G,H)$ is sampled with the following general procedure. We generate a pair $(G,G')$ of graphs, (or adjacency matrices $(A,A')$) with same node set such that $G$ and $G'$ are edge correlated, and relabel the nodes of $G'$ with some uniform random permutation $\pi^\star \in \cS_n$, independent from everything else, to form $H$.

Henceforward, we will always refer to graph alignment for planted graph alignment.

Let us describe several models of random correlated graphs that will be studied in the sequel: the Gaussian model, where the graph is complete and the signal lies on the edge weights, and the correlated \ER model, where the correlated graphs are both \ER marginally distributed.

\subsubsection{Correlated Gaussian Wigner model} 
The \emph{correlated Gaussian Wigner model} was first introduced by Ding et al. \cite{Ding18} as a simple playground for graph alignment, and 
has been further investigated for its own sake in some recent works (see Section \ref{intro:subsection:short_survey}). 

Under this model, the graphs are complete and the signal lies in the weights of edges between all pairs of nodes. The correlated weighted adjacency matrices $A$ and $A'$ are simply sampled as follows: first, $A$ is drawn from the Gaussian Orthogonal Ensemble ($\GOE$), namely, independently for all $1 \leq u \leq v \leq n$, 
\begin{equation}\label{eq:GOE_model}
A_{u,v} = A_{v,u} \sim 
\begin{cases}
\mathcal{N}(0,1/n) & \text{if $u \neq v$} \, , \\
\mathcal{N}(0,2/n) & \text{if $u = v$}  \, .
\end{cases}
\end{equation}
Given $H$ an independent copy of $A$, we define 
\begin{equation}\label{eq:GOE_model2}
A'=A+\xi H \, ,
\end{equation} where $\xi >0$ is the noise parameter. We denote $(A,A') \sim \Wig(n,\xi)$. Under this model, coefficients of $A$ and $A'$ are pairwise correlated with correlation parameter
$\frac{1}{\sqrt{1+\xi^2}}$. This model is the subject of Chapter \ref{chapter:EIG1}.
 
A natural variant of the model is to ensure that the two marginals are the same, and to remove self-loops in the graphs (i.e. diagonal coefficients). All pairs of edge weights $(A_{u,v}, A'_{u,v})_{1 \leq u<v \leq n}$ can be taken to be i.i.d. couples of normal variables with zero mean, unit variance and correlation parameter $\rho \in [0,1]$. An equivalent sampling procedure is to generate matrix $A'$ from $A$ as follows: 
\begin{equation}\label{eq:CGW_model}
A' = \rho \cdot  A  + \sqrt{1-\rho^2}  \cdot H,
\end{equation} where $H$ is an independent copy of $A$. We denote 
$(A,A') \sim \Wig'(n,\rho)$. This model, very close to $\Wig(n,\xi)$ is the subject of Chapter \ref{chapter:gaussian_alignment_IT}.

\subsubsection{Correlated \ER model} 
The \emph{correlated \ER model} is the simplest, most natural model of correlated random graphs, which explains why it has been the focus of recent threads of research (see Section \ref{intro:subsection:short_survey}) for the study of graph alignment. We refer to Section \ref{intro:subsection:basics_rg} for the definition of the (non-correlated) \ER model. This model will be studied in detail in Chapters \ref{chapter:impossibility}, \ref{chapter:NTMA} and \ref{chapter:MPAlign}.

For a number of nodes $n$, edge probability $q \in [0,1]$ and correlation parameter $s \in [0,1]$ such that $s \geq q$, the correlated \ER model, denoted $\G(n,q,s)$, consists of two random graphs $G,G'$ with symmetric adjacency matrices $A,A'$, with same node set $V=[n]$, such that $\set{( A_{u,v},A'_{u,v})}_{u<v\in [n]}$ are i.i.d. pairs of correlated Bernoulli random variables such that
\begin{equation}\label{eq:CER_model}
(A_{u,v},A'_{u,v}) =
\begin{cases}
(1,1) & \text{with probability $qs$}\\
(1,0) & \text{with probability $q(1-s)$}\\
(0,1) & \text{with probability $q(1-s)$}\\
(0,0) & \text{with probability $1-q(2-s)$}.
\end{cases}       
\end{equation}

Note that in this setting, $G$ and $G'$ both have $\G(n,q)$ marginal distributions.

\begin{figure}[h]
	\centering
	\medskip
	\begin{subfigure}[t]{0.49\linewidth}
		\centering
		\hspace{-0.6cm}
		\includegraphics[scale=0.29]{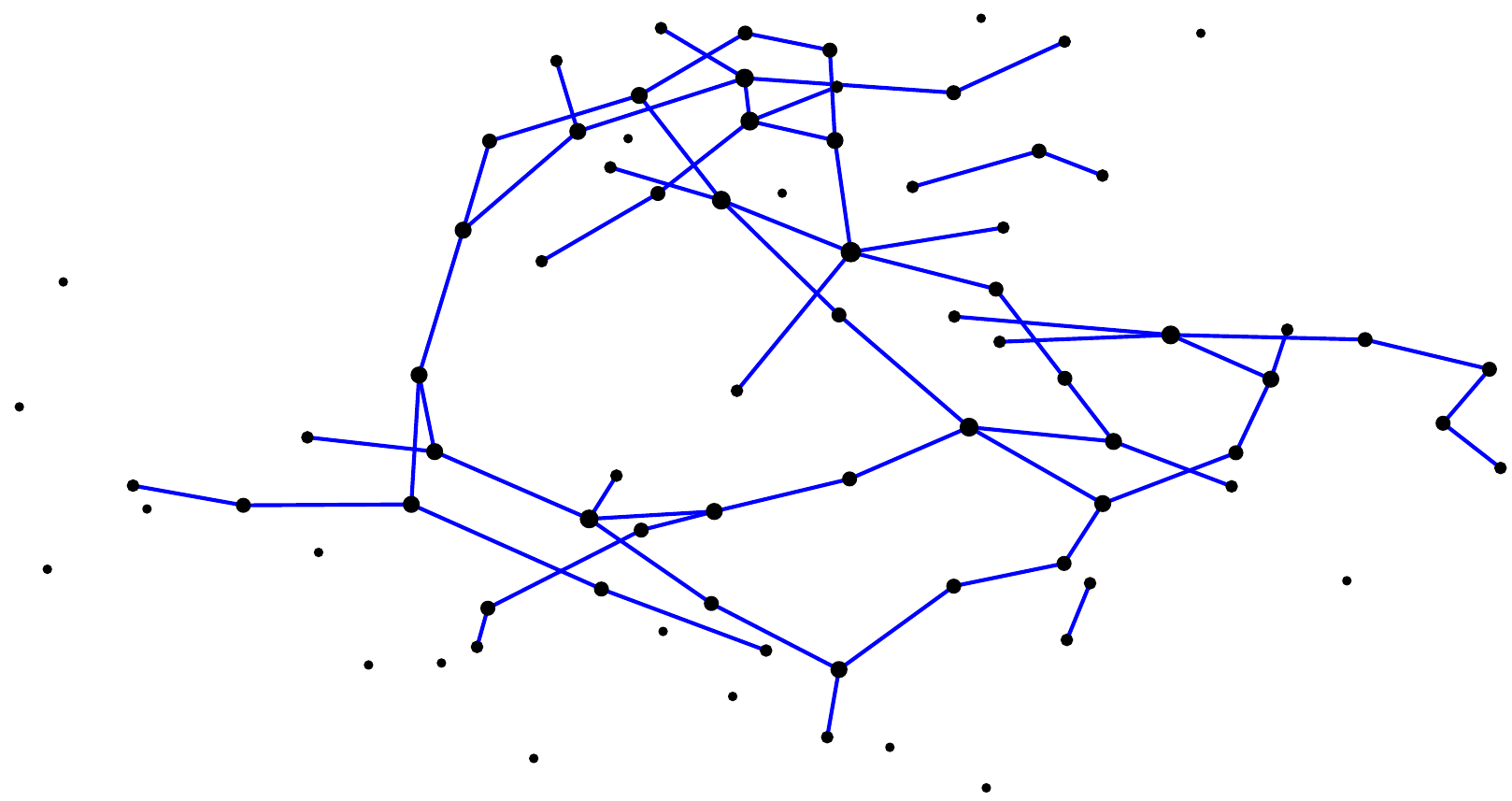}
	\end{subfigure}
	\begin{subfigure}[t]{0.49\linewidth}
		\centering
		\includegraphics[scale=0.29]{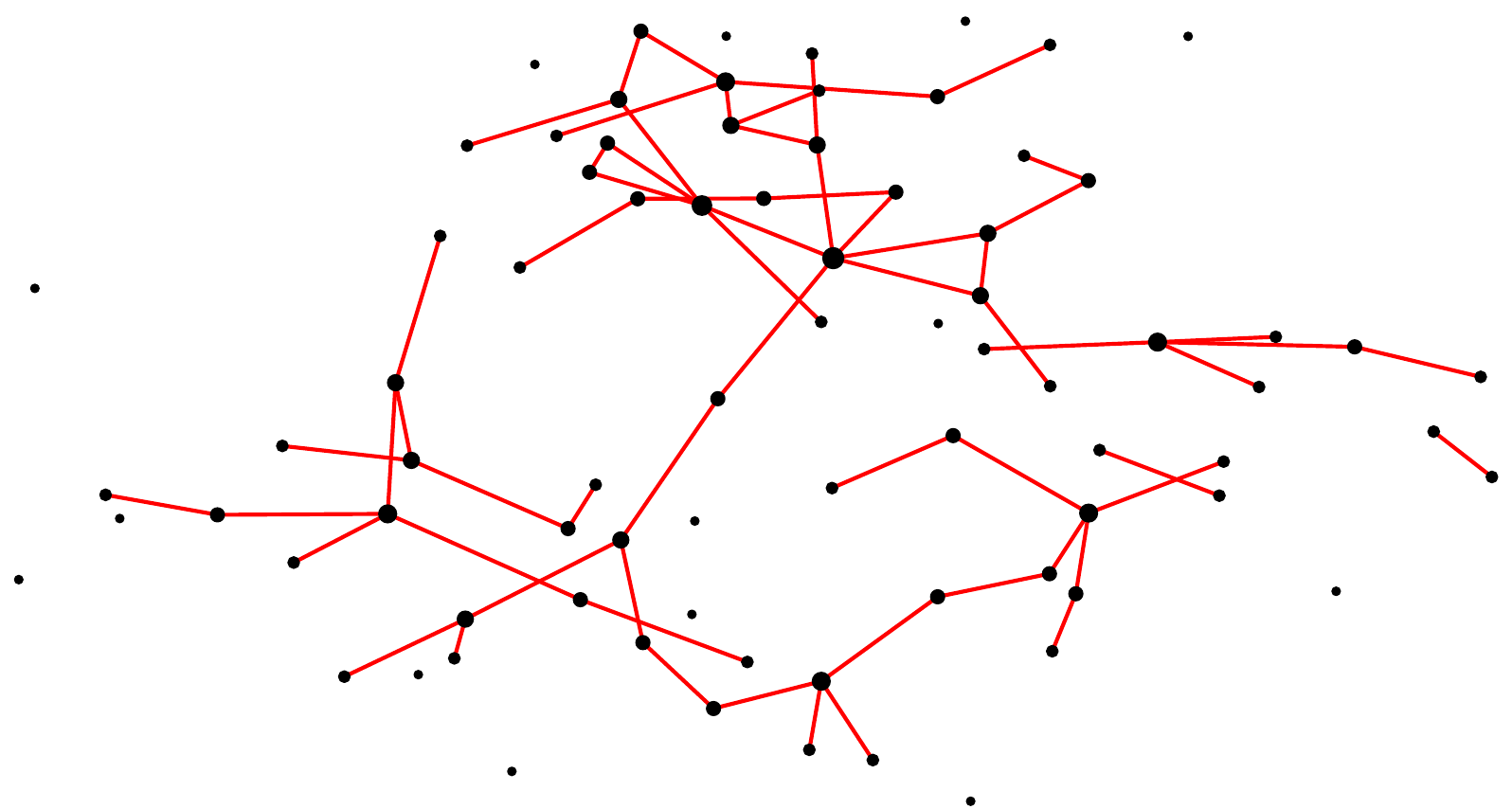}
	\end{subfigure}
	\caption{Two samples from $\G(n,q,s)$ with $n=80$, $p=1.9/n$ and $s=0.8$.}
	\label{fig:intro:Gnqs}
\end{figure}

\begin{remark}
Note that if $s=1$, the graphs $G$ and $G'$ are identical, and in the case $s=q$, the two graphs are independent, that is $\G(n,q,q) \eqd \G(n,q) \otimes \G(n,q) $.
\end{remark}

\begin{remark}
Another equivalent sampling procedure for the correlated \ER model is as follows. Starting from a parent graph $F \sim \G(n,q/s)$, $G$ and $G'$ are obtained by two independent $s-$subsamplings of $F$ (a $s-$subsampling of $F$ simply consists in keeping each edge of $F$ independently with probability $s$).
\end{remark}

\subsubsection{Planting the alignment}

As mentioned earlier, in the planted model for graph alignment, after having generated two labeled correlated graphs $G$ ans $G'$, the last step is to draw the planted permutation $\pi^\star$ uniformly at random in $\cS_n$. and $\Pi^\star$ is the $n \times n$ matrix representation of permutation $\pi^\star$, that is $\Pi^\star_{u,v} = \one_{v = \pi^\star(u)}$.

We then relabel the second graph $G'$ according to this permutation $\pi^\star$, forming the graph $H$ with adjacency matrix $B$ such that for all $1 \leq u,v \leq n$, 
\begin{equation}\label{eq:relabeling_G'_H}
	B_{\pi^\star(u),\pi^\star(v)} = A'_{u,v} \, ,
\end{equation} or, equivalently, $B = (\Pi^{\star})^\top A' \Pi^\star$. 

\begin{figure}[h]
	\centering
	\medskip
	\begin{subfigure}[t]{0.49\linewidth}
		\centering
		\hspace{-0.5cm}
		\includegraphics[scale=0.29]{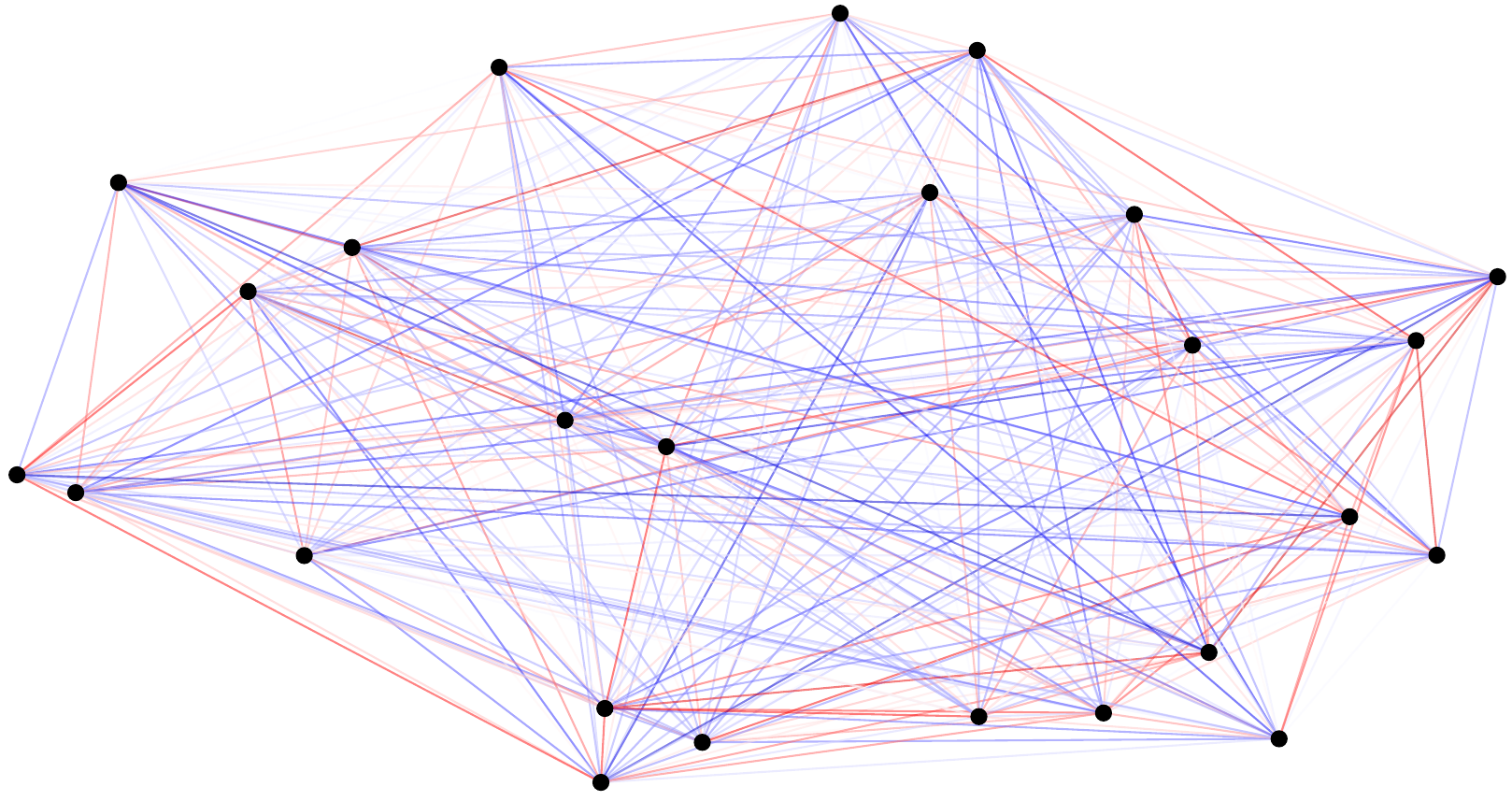}
	\end{subfigure}
	\begin{subfigure}[t]{0.49\linewidth}
		\centering
		\includegraphics[scale=0.29]{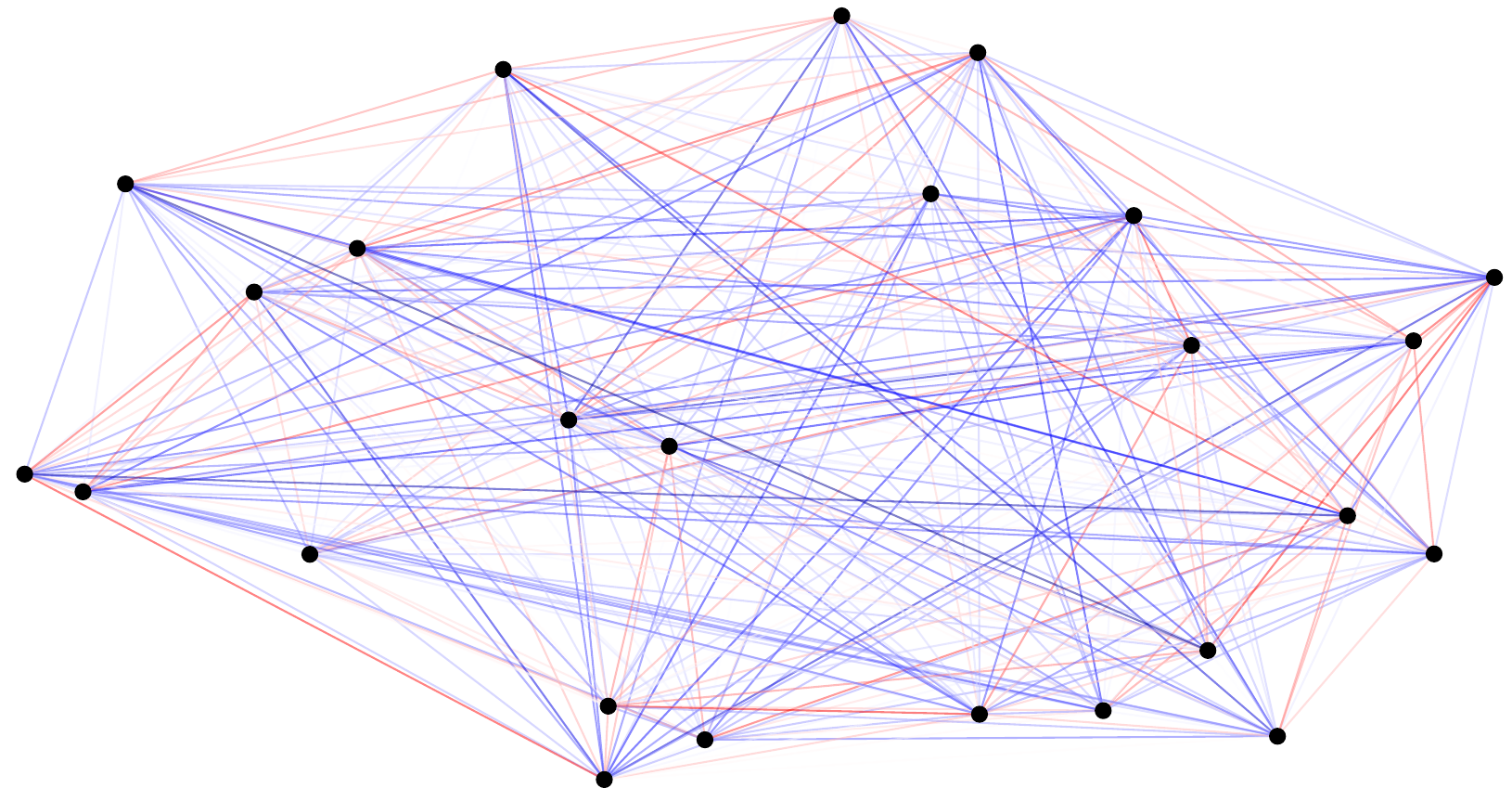}
	\end{subfigure}
	\caption{A sample from model $\Wig'(n,\rho)$ \eqref{eq:CGW_model} with $n=25$ and $\rho = 0.8$. Edges are colored according to their weights. }
	\label{fig:intro:corr_gaussian_wigner}
\end{figure}

\subsubsection{Reconstruction} 
Given the two graphs $(G,H)$ generated from the planted model described above, the reconstruction task consists in finding an estimator $\hat{\pi}$ of the planted solution $\pi^\star$ upon observing $G$ and $H$ (or equivalently $A$ and $B$). 

The performance of any estimator $\hat{\pi} = \hat{\pi}(G,H): [n] \to [n]$ will be assessed through its \emph{overlap} with the unknown planted permutation $\pi^\star$, defined as
\begin{equation}\label{eq:intro:def_overlap}
\ov(\hat{\pi},\pi^\star):=\frac{1}{n}\sum_{u \in [n]}\one_{\hat{\pi}(u) = \pi^\star(u)} \, .
\end{equation} 
The overlap \eqref{eq:intro:def_overlap} is now the measure of performance which we seek to optimize when performing graph alignment in this planted setting, and differs from that of the non-planted case \eqref{eq:QAP}. 

We can however straightaway note that in the \ER setting, the posterior distribution of the planted alignment $\Pi^\star$ is given by
\begin{flalign}\label{eq:posterior_erdos_renyi}
	 \dP(\Pi^* & =  \Pi | A,B)  \nonumber \\
	 & \propto \dP\left(\Pi^* = \Pi , A,B\right) \nonumber \\
	 &= \frac{1}{n!} (qs)^{\frac{1}{2}\langle A, \Pi B \Pi^\top \rangle}(q(1-s))^{\frac{1}{2}\langle A, \one \rangle + \frac{1}{2}\langle \one, B \rangle - \langle A, \Pi B \Pi^\top \rangle} (1-q(2-s))^{\frac{1}{2} \langle \one - A, \one - \Pi B \Pi^\top \rangle} \nonumber \\
	 & \propto \left( \frac{s(1-q(2-s))}{q(1-s)^2}\right)^{\frac{1}{2} \langle A, \Pi B \Pi^\top \rangle} \, ,
\end{flalign} where $\one$ denotes the all-ones matrix. Since $s \geq q$, $\frac{s(1-q(2-s))}{q(1-s)^2}\geq 1$, and the maximum-a-posteriori estimator of $\pi^\star$ given $G,H$ is thus exactly the solution of the QAP \eqref{eq:QAP}. This is again unsurprisingly in accordance with the worst-case/planted duality (see Section \ref{intro:subsection:zoo}). The same result holds in the Gaussian Wigner models $\Wig(n,\xi)$ \eqref{eq:GOE_model} and $\Wig'(n,\rho)$ \eqref{eq:GOE_model2}.

We now specify different types of reconstruction tasks that will be referred to in the rest of the thesis. A sequence of estimators $\{\hat{\pi}_n\}_n$ (i.e. measurable functions of $G,H$) -- omitting the dependence in $n$ -- is said to achieve
\begin{itemize}
	\item \emph{Exact recovery} if $\, \dP(\hat\pi=\pi^\star) \underset{n \to \infty}{\longrightarrow} 1$,
	\item \emph{Almost exact recovery} if $\; \dP(\ov(\hat\pi,\pi^\star)= 1-o(1)) \underset{n \to \infty}{\longrightarrow} 1$,
	\item \emph{Partial recovery} if there exists some $\eps>0$ such that $\; \dP(\ov(\hat\pi,\pi^\star) > \eps) \underset{n \to \infty}{\longrightarrow} 1$ \,.
\end{itemize}

\begin{remark}\label{remark:limits_of_partial_recov}
	Partial recovery consists in ensuring that the estimator $\hat\pi$ matches a non-vanishing fraction of nodes -- we hope this fraction to be as large as possible. 
	Though simple to formulate, from an application standpoint it may however be of little use to know that one has a permutation with 30\% of correctly matched nodes if one does not have a clue about which pairs are correctly matched. 
	Motivated in part by the following remark, we will introduce in Section \ref{intro:subsection:heuristics_tree_graph}, Chapters \ref{chapter:impossibility} and \ref{chapter:MPAlign} another slightly different recovery task, namely one-sided partial recovery, which is believed to be more relevant both for theory and practice.
\end{remark} 

\subsection{A summary of related work}\label{intro:subsection:short_survey}

We give in the following a brief overview of related algorithmic and theoretical contributions, aside from our work.

\subsubsection{Seeded graph alignment} 
An interesting and widely studied setting is graph alignment with presence of side information, namely \emph{seeds}, that are correct pre-mapped vertex pairs. The idea of seeded alignment methods is that a vertex pair $(u,u')$ will have more witnesses -- that is, correct pairs $(w,w')$ such that $u \conn w$ and $u' \conn w'$ -- if they are matched than if they are not. 

It is proved in \cite{lyzinski14_seeded} that when the graphs are dense enough, a logarithmic number of correct seeds is sufficient to recover the whole alignment. To perform this task, several methods are proposed \cite{Pedarsani11,Fishkind19, Shirani2017SeededGM,Yartseva13,Mossel19,Araya22}, some of them relying on percolation techniques \cite{Janson12,Yartseva13}, large neighborhoods statistics \cite{Mossel19}, or projected power method \cite{Araya22}.

An interesting line of work extends the problem in the case where some of them are likely to be incorrect, e.g. since they may be provided by seedless methods. The NoisySeeds algorithm, also built on a percolation procedure, is proposed in \cite{kasemi15}, and \cite{lubars18} uses 1-hop witnesses to recover the full alignment. 
The recent paper \cite{Yu21} establishes information-theoretic results for seeded alignment in the noisy case, and proposes a method which considers both 1-hop and 2-hop witnesses, improving previous theoretical guarantees. 

\subsubsection{Information-theoretic results}

First fundamental results for \ER graph alignment are due to Pedarsani and Grossglauser \cite{Pedarsani11}, followed by Cullina and Kiyavash \cite{Cullina2017} who prove that under some mild sparsity constraints, feasibility of exact alignment exhibits a sharp threshold at $nqs \simeq \log n$. Their approach is based on the analysis of the maximum a posteriori estimator for the positive side, and the impossibility result is the consequence of the large number\footnote{Note that the illustrative result in Theorem \ref{intro:theorem:connectivity_ER} proved earlier on gives a short proof of this fact.} of automorphisms of an \ER graph with mean degree less than $\log n$ \cite{Bollobas2001}.

Results for almost-exact recovery proved in \cite{Cullina18} establish that almost-exact recovery is feasible if and only if $nqs \to +\infty$, under some mild sparsity assumptions.

These reconstruction thresholds are sharpened by the recent work \cite{Wu2021SettlingTS}. This paper shows, among other results, a sharp all-or-nothing phenomenon in the Gaussian setting at $n \rho^2= 4\log n$. If $n \rho^2 > (4+\eps)\log n$, exact alignment is feasible, whereas below the threshold even partial recovery is infeasible. 

For dense \ER graphs with $q/s = n^{-o(1)}$, another phase transition arise between infeasibility of partial alignment and possibility of almost-exact alignment, at 

\begin{equation}\label{eq:sharp_IT_partial_almost_exact}
	\frac{nqs(\log(s/q)-1+q/s)}{\log n} = 2 \, .
\end{equation}


Partial recovery was first studied by Hall and Massoulié \cite{Hall20}, who showed that $nqs \to 0$ is an impossibility condition, whereas $nqs > C$ (with a large, non-explicit constant $C$), together with some additional sparsity constraints, ensures feasibility. These results are improved in \cite{Wu2021SettlingTS}, where the authors show that if $q = \lambda/n$, and $s$ is an constant, then partial recovery is shown to be feasible when $nqs > 4+\eps$. The impossibility result requires $nq/s = \omega(\log^2 n)$ and $nqs < 1- \eps$. 

At the time this manuscript is being completed, a very recent contribution \cite{Ding22} sharpens this last result in the case where $q/s = n^{-\alpha+o(1)}$, showing a sharp threshold for partial recovery at $nqs = \lambda^\star(\alpha)$, where $\lambda^\star(\alpha)$ is given as the asymptotic maximal edge-vertex ratio over all nonempty subgraphs of an \ER graph $\G(n,\frac{1}{\alpha n})$.

\subsubsection{Algorithms for exact recovery in the dense case}
As mentioned earlier, the vast majority of prior work focused on exact recovery with no side information, seeking to provide polynomial-time (or quasi-polynomial time) algorithms that (sometimes provably) recover the entire permutation $\pi^\star$ under some conditions on the parameters $n,q,s$ (or $n,\rho,\xi$ is the Gaussian setting).\\

\textit{Spectral methods.}
A first spectral method for recovery is due to \cite{Umeyama88}, and uses spectral decompositions and relaxation of the QAP on the orthogonal group. Another spectral, rank-reduction method is proposed by Feizi et al. \cite{Feizi16}, and an another algorithm, \alg{GRAMPA}, is proposed and analyzed  in \cite{Fan2019Wigner, fan2019ERC}, both for the Wigner and the \ER model. \alg{GRAMPA} builds a similarity matrix based on outer products between pairs of eigenvectors of the two graphs, and outputs a matching via a rounding procedure. \\

\textit{QAP relaxations.}
A class of algorithms designed for recovery follows a Frank-Wolfe approach (see \cite{Anstreicher2002,Vogelstein2011,Bach09}) , which consists in relaxing the integer programming formulation of the QAP \eqref{eq:QAP} to a continuous optimization problem on which iterative linearized procedures are used, and then projecting the final iterate on the space of solutions. Every linear optimization step at each iteration is a \emph{linear assignment problem (LAP)} of the form
\begin{equation}\label{eq:LAP}
\argmax_{\Pi \in \cS_n} \langle \Pi u, v \rangle  \, ,
\end{equation} with $u,v \in \dR^n$, which can be solved using the Hungarian algorithm \cite{Kuhn55} in $O(n^3)$ time complexity. In particular, authors of \cite{Bach09} study a path following algorithm on a concave relaxation of the QAP. 

The question of giving theoretical guarantees on the performance of such relaxations of the QAP is interestingly discussed in \cite{Lyzinski14}. In this paper, the common convex relaxation of the QAP which consists in minimizing $\| AD - DB\|^2$ over all doubly-stochastic matrices $D$ is proved to almost always fail, whereas the indefinite relaxed
graph alignment problem\footnote{Note that since $\|DA\|^2 \neq \|A\|^2$ in general, these two relaxations have now different solutions. Moreover, the objective in this second relaxation is no more convex in $D$, the Hessian being indefinite.} which minimizes $-\langle AD , DB\rangle$ over all doubly-stochastic matrices $D$ almost always discovers the true permutation, if solved exactly. Though non-convex quadratic programming is NP-hard in general, this indefinite relaxation can still be efficiently approximately solved with the Frank-Wolfe methodology evoked here above. 

Relevant to QAP relaxation methods is the recent contribution \cite{dym2017ds++} which proposes a convex quadratic programming relaxation, proved to be more accurate than the classical double-stochastic and spectral relaxations, although with same time complexity as the former.\\

\textit{Methods using network topology.} Another class of methods are based on the exploration of the network topology, in order to design vertex signatures that can efficiently recover the matched pairs. Ding et al. \cite{Ding18} introduced a matching procedure based on degree profiles, that is the empirical distribution of the degrees of neighbors. A method proposed in \cite{Barak2019} relies on counting copies of subgraphs adjacent to a given node, for a well-chosen family of graphs. In recent contributions, Mao, Rudelson and Thikhomirov design an method involving a two-generation partitioning procedure \cite{Mao21} and another algorithm \cite{Mao21_constant_corr} based on comparison of partition trees associated with the graph vertices. These methods are shown to improve the previously state-of-the art performances in terms of noise robustness (see below). \\

\textit{Theoretical guarantees.} From the methods for exact recovery cited here above, those giving rigorous theoretical guarantees all require a mean degree at least $nq \geq \polylog n$, and $1-s \leq 1/(\polylog n)$ in the \ER setting or $1-\rho^2 \leq 1/(\polylog n)$ in the Gaussian setting, that is a correlation close enough to $1$. The only exceptions for exact recovery are the recent works of Mao, Rudelson and Thikhomirov: in the \ER model, \cite{Mao21} tolerates a noise $1-s$ up to $(\log\log n)^{-c}$, and \cite{Mao21_constant_corr} can tolerate up to constant noise -- the constant being unspecified. The recent algorithm proposed in \cite{Araya22} can also tolerate constant noise for exact recovery in the seeded setting.

Note that these methods do not work in the sparse setting where both the correlation and the mean degree are constant, setting on which we will focus in Chapters \ref{chapter:impossibility}, \ref{chapter:NTMA} and \ref{chapter:MPAlign}.

\subsubsection{Detection problem}
Aside from the reconstruction tasks, the detection has less been studied until very recently. Given $n,q,s$, the associated hypothesis problem is as follows: $$\cH_0 := `` (G,H) \mbox{ are two independent $\G(n,q)$ graphs}"$$ versus $$\cH_1 := `` (G,H) \mbox{ are drawn under the \ER planted model}" \,.$$ 

Wu, Xu and Yu \cite{Wu20} give fundamental results for the detection problem. They establish a sharp threshold for detection in the Gaussian model at $n\rho^2 / \log n = 4$, and show that in the dense case $q/s = n^{-o(1)}$, the sharp threshold for partial/almost-exact alignment given in \eqref{eq:sharp_IT_partial_almost_exact} also holds for detection. The picture in the sparser case $q/s=n^{-\Omega(1)}$ is however less clear, but in the case where $q=\lambda/n$ and $s$ is a constant, their result implies that strong detection is feasible if $\lambda s >2$ and infeasible if $\lambda s<1$ and $s<0.01$.

Improving on a previous work \cite{Barak2019}, the state-of-the art algorithm for this detection task is proposed in \cite{Mao21_counting} which consider a test based on counting trees in the two graphs. This algorithm runs in $O(n^{2+o(1)})$ time and is proved to succeed with high probability if $n\min(q,1-q) \geq n^{-o(1)}$ (this assumption is very mild) and the correlation coefficient (which is asymptotically $s$ if $q \to 0$) is greater that $\sqrt{\alpha} \sim 0.58$, where $\alpha$ is Otter's constant \cite{Otter48}, defined as the inverse of the exponential growing rate of the number of unlabeled trees with $K$ edges. This algorithm also improves on previous IT results and will be the object of further discussion in Chapter \ref{chapter:MPAlign}.

\section{Correlation detection in random trees}\label{intro:section:cdt}
In this Section, we introduce a problem which will be at the heart of Chapters \ref{chapter:NTMA} and \ref{chapter:MPAlign}: correlation detection in random trees. 

\subsection{Problem statement}

The problem of detecting correlation in random trees is a fundamental statistical task, consisting in deciding whether two rooted trees are correlated up to a relabeling of the nodes, that is if they contain a common planted subtree, or if they are independent. 

This problem could very well be defined per se and studied as such; we nevertheless explain briefly how this problem arises from the study of sparse graph alignment. Let us imagine that we are given correlated graphs $G,H$ from the \ER planted model, and that one would like to know whether node $u \in V(G)$ is matched to $u' \in V(H)$, namely if $u' = \pi^\star(u)$. An answer to this question can be to build an estimator $\hat\pi$ such that $\hat\pi(u)=u'$ if and only if the local structure of graph $G$ in the neighborhood of node $u$ is somehow 'close' to the local structure of graph $H$ in the neighborhood of node $u'$. 

In the sparse regime, it is well known that the neighborhoods up to distance $d$ of node $u$ (resp. $u'$) in $G$ (resp. $G'$), are both asymptotically distributed as Galton-Watson branching trees\footnote{This convergence in in fact even stronger, since it happens in in the sense of \emph{Benjamini-Schramm}, see \cite{Benjamini2011}.}. More specifically, if $u' = \pi^\star(u)$, then the pair of neighborhoods are asymptotically jointly distributed as correlated Galton-Watson branching trees (distribution denoted $\dPls_{d}$). On the other hand, for pairs of nodes $(u,u')$ taken at random in $[n]$, the neighborhoods are asymptotically independent Galton-Watson branching trees (distribution denoted $\dPl_{d}$). Hence, we are now left with the problem of detecting correlation in random trees.

\subsubsection{Rooted labeled trees} 
A \emph{labeled rooted tree} $t=(V,E)$ is an undirected graph with node set $V$ and edge set $E$ which is connected and contains no cycle. The \emph{root} of $t$ is a given distinguished node $\rho \in V$, and the \emph{depth} of a node $u$ is defined as its graph distance to the root $\rho$. The depth of tree $t$ is given as the maximum depth of all nodes in $t$. 

In a rooted tree $t$, each node $u$ at depth $d \geq 1$ has a unique \emph{parent} in $t$, which can be defined as the unique node at depth $d-1$ on the path from $u$ to the root $\rho$. Similarly, the \emph{children} of a node $u$ of depth $d$ are all the neighbors of $u$ at depth $d+1$. For any node $u$ of the tree $t$, we denote by $t_u$ the subtree of $t$ rooted at node $u$, that is the tree obtained by deleting the edge between $u$ and its parent and keeping the connected component of $u$.

\subsubsection{Models of random trees, hypothesis testing}
We describe hereafter models of random trees that will be useful in the sequel. For more detailed definitions we refer to Chapter \ref{chapter:NTMA}, Section \ref{NTMA:subsection:models_trees}.\\

\textit{Galton-Watson trees with Poisson offspring.} The \emph{Galton-Watson tree with offspring $\Poi(\mu)$ up to depth $d$}, denoted by $\GWmu_d$, is defined recursively as follows. First, the distribution $\GWmu_0$ is a Dirac at the trivial tree only consisting in the root. Then, for $d \geq 1$, sample a number $Z \sim \Poi(\mu)$ of independent $\GWl_{d-1}$ trees, and attach each of them as children of the root, to form a tree of depth at most $d$. \\

\textit{Tree augmentation.} For $\lambda >0$ and $s \in [0,1]$, a (random) \emph{$(\lambda,s)-$augmentation} of a given tree $\tau =(V,E)$, denoted by $\Augls_d(\tau)$, is defined as follows. First, attach to each node $u$ in $V$ at depth $<d$ a number $Z^{+}_u$ of additional children, where the $Z^{+}_u$ are i.i.d. of distribution $\Poi(\lambda (1-s))$. Let $V^+$ be the set of these additional children. To each $v \in V^+$ at depth $d_v \in [d]$, we attach another random tree of distribution $\GWl_{d-d_v}$, independently of everything else.\\

We are now ready to describe the two models $\dPl_{d}$ and $\dPls_{d}$ in simple words. Under the independent model $\dPl_{d}$, $T$ and $T'$ are two independent $\GWl_{d}$ trees. 
The correlated model $\dPls_{d}$ is built as follows. Starting from an \emph{intersection tree} $\tau^\star \sim \GWls_d$, and $T$ and $T'$ are obtained as two independent $(\lambda,s)-$augmentations of $\tau^\star$.

In both modes, the labels of the trees are always forgotten, or randomly uniformly re-sampled. We however still distinguish the roots af the two trees. It can easily be verified that the marginals of $T$ and $T'$ are the same under $\dPl_{d}$ and $\dPls_{d}$, namely $\GWl_d$. The parameters are $\lambda$, the mean number of children of a node, and the correlation $s$.

\begin{figure}[h]
	\centering
	\begin{subfigure}[b]{\textwidth}
		\centering
		\includegraphics[scale=0.4]{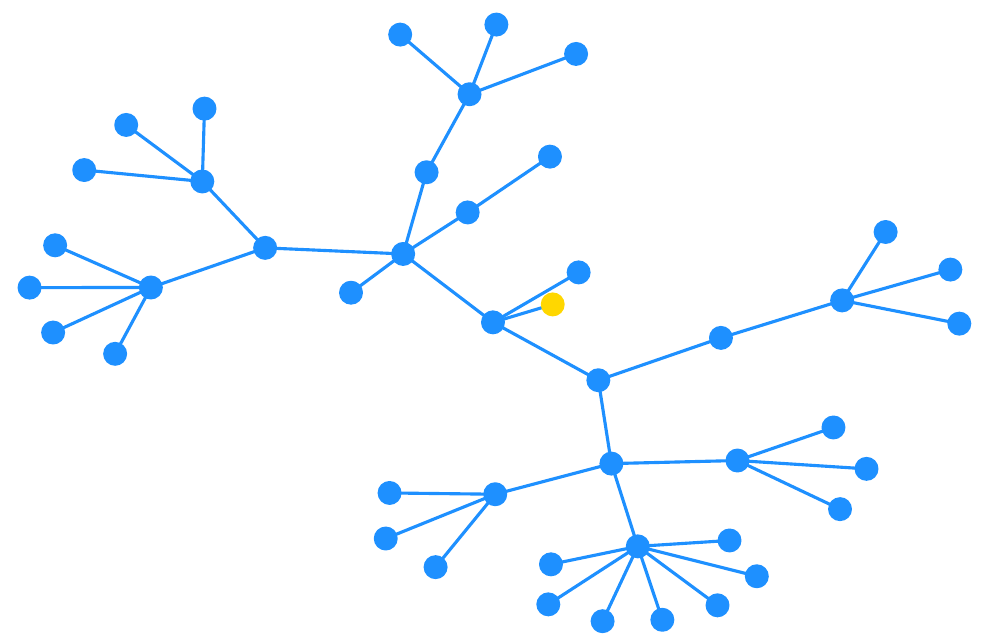}
		\hspace{0.25cm}
		\includegraphics[scale=0.4]{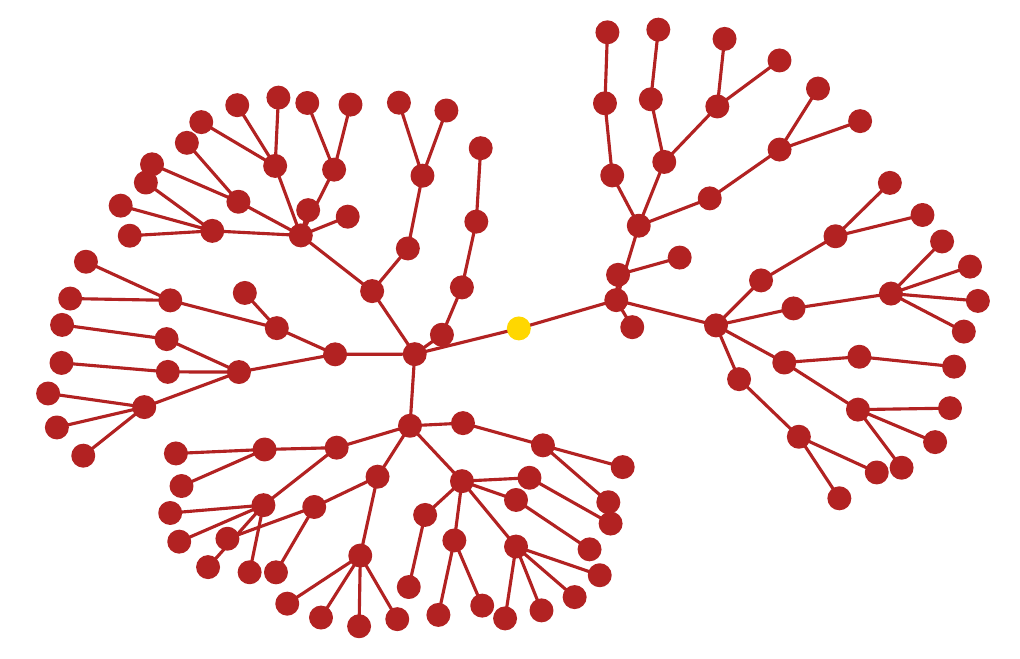}
		\caption{Samples $T,T'$ from $\dPl_{d}$.}
		\label{fig:intro:P0}
	\end{subfigure}
	\begin{subfigure}[b]{\textwidth}
		\vspace{0.3cm}
		\centering
		\includegraphics[scale=0.4]{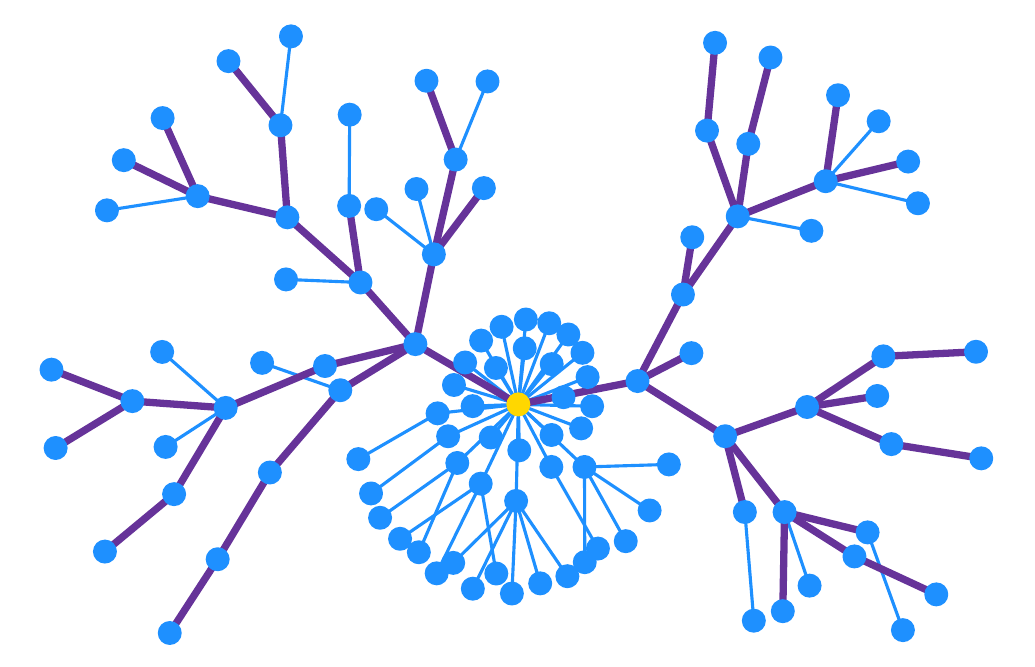}
		\hspace{0.25cm}
		\includegraphics[scale=0.4]{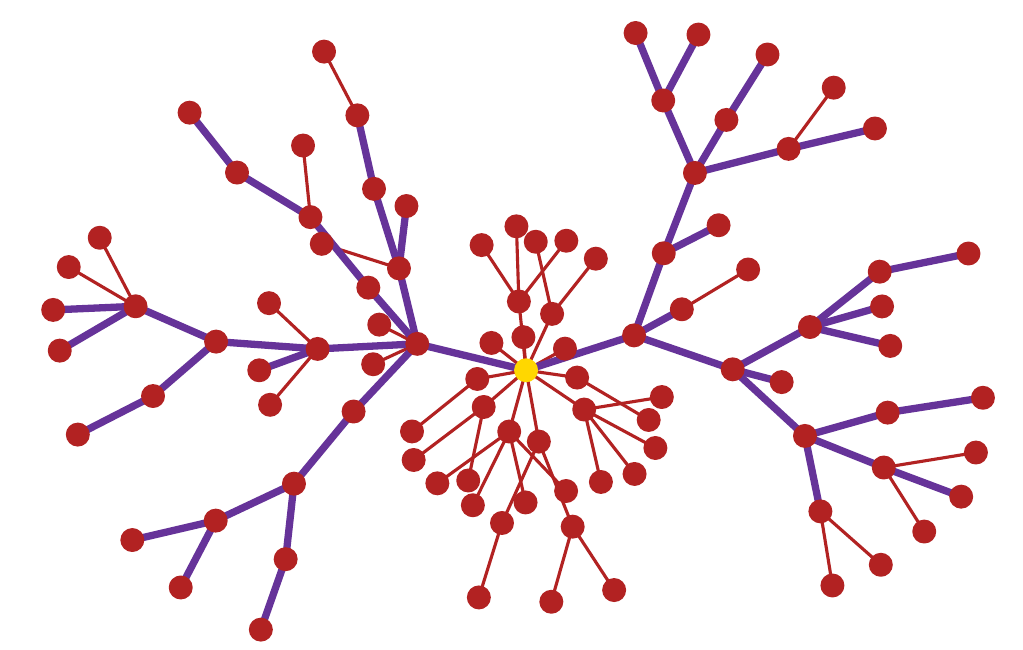}
		\caption{Samples $T,T'$ from $\dPls_{d}$. The common subtree $\tau^\star$ is drawn thick and purple.}
		\label{fig:intro:P1}
	\end{subfigure}
	
	\caption{Samples from models $\dPl_{d}$ and $\dPls_{d}$, with $\lambda = 1.8$, $s=0.8$, and $d=5$. The root node is highlighted in yellow. Labels are forgotten.}
	\label{fig:samples_P01}
\end{figure}

\subsection{Hypothesis testing, one-sided test}\label{intro:subsection:hypothesis_testing_trees}
The corresponding hypothesis test can be formalized as follows: given the observation of a pair of trees $(T,T')$ of depth at most $d$, we want to test
\begin{equation}
\label{eq:test_hypotheses_p}
\cH_0 = \mbox{"$T,T'$ are drawn under $\dPl_d$"} \quad \mbox{versus} \quad \cH_1 = \mbox{"$T,T'$ are drawn under $\dPls_d$"}.
\end{equation} 

In statistical detection problems (see Section \ref{intro:section:inference_on_rg}), the commonly considered tasks are that of  
\begin{itemize}
	\item \emph{strong detection}, i.e. designing tests $\cT_d$ that verify
	\begin{equation*}
	\underset{d \to \infty}{\lim} \left[\dPl_d\left( \cT_n(T,T') = 1 \right) + \dPls_d\left( \cT_n(T,T') = 0 \right)\right] = 0,
	\end{equation*}
	\item \emph{weak detection}, i.e. tests $\cT_n$ that verify
	\begin{equation*}
	\underset{d \to \infty}{\limsup} \left[\dPl_d\left( \cT_n(T,T') = 1 \right) + \dPls_d\left( \cT_n(T,T') = 0 \right)\right] <1,
	\end{equation*}
\end{itemize} In other words, strong detection corresponds to exactly discriminate w.h.p. between $\dPl_d$ and $\dPls_d$, whereas weak detection corresponds to strictly outperforming random guessing. We here argue that neither strong detection nor weak detection are relevant for our problem. 

First, because of the event  that the intersection tree does not survive, which is of positive probability  under $\dPls_d$: we always have $\dPls_d(t,t') \geq C_{\lambda,s} \cdot \dPl_d(t,t')$ for some $C_{\lambda,s}>0$. This implies that $\dPl_d$ is always absolutely continuous w.r.t. $\dPls_d$, hence strong detection can never be achieved.

Second, weak detection is always achievable as soon as $s>0$: with the same notations as here above, the difference of the degree of the root in $T$ and that of the root in $T'$ is always centered but has different variance under $\dPl_d$ and under $\dPls_d$, hence these two distributions can be weakly distinguished, without any further assumption than $s>0$. 

Moreover, if we want our test to be relevant for partial alignment -- for which we know that only a fraction of the nodes can be recovered -- it is natural to require a positive power (i.e., being able to detect matched nodes with some positive probability), but also to ensure that the output of the algorithm contain almost no wrong pair (i.e. imposing a vanishing type I error).\\

We are thus interested in being able to ensure the existence of an (asymptotic) \textit{one-sided test}, that is a test $\cT_d: \cX_d \times \cX_d \to \left\lbrace 0,1 \right\rbrace$ such that $\cT_d$ chooses hypothesis $\cH_0$ under $\dPl_d$ with probability $1-o(1)$, and chooses $\cH_1$ with some positive probability uniformly bounded away from 0 under $\dPls_d$. 

\subsection{Two methods} 

We now give the outline of two methods for detection of correlation in random trees, that will be the object of Chapters \ref{chapter:NTMA} and \ref{chapter:MPAlign}.

\subsubsection{Tree matching weight}
In Chapter \ref{chapter:NTMA}, we build a test based on a measure of similarity between two trees: the \emph{tree matching weight}.\\

\textit{Matching weight of two rooted trees.} For any $d \geq 0$, let $\cA_d$ denote the collection of rooted trees whose leaves are all of depth $d$. Given two rooted trees $t$ and $t'$ of depth at most $d$, let $M(t,t')$ denote the collection of trees $\tau \in\cA_d$ such that there exist injective embeddings of $\tau$ in $t$ and $t'$ that preserve the rooted tree structure, that is the depth of the nodes and the child-parent relationship. The \emph{matching weight of trees $t$ and $t'$ at depth $d$} is then defined as:
\begin{equation}\label{eq:def:matching_weight}
	\cW_d(t,t'):=\sup_{\tau \in M(t,t')}\card{\cL_d(\tau)},
\end{equation} where $\cL_d$ is the number of leaves at depth $d$ of tree $\tau$. In other words, the tree matching weight fo a pair of trees is defined as the maximal size of a common subtree, measured in terms of number of leaves. \\

\begin{figure}[h]
	\centering
	\includegraphics[scale=1]{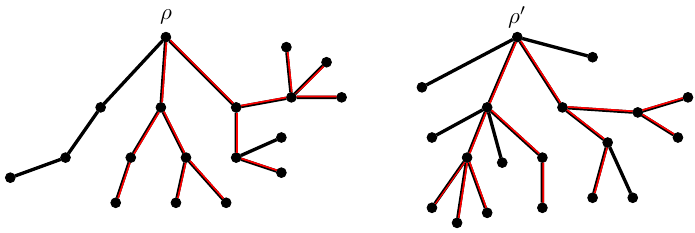}
	\caption{\label{image_tree_matching_weight} Example of two trees $t$, $t'$ with $\cW_3(t,t')=7$, where an optimal $t \in \cA_3$ is drawn in red.}
\end{figure}
\textit{Recursive computation of $\cW_d$.}
From the previous definition, a first step conditioning yields a recursion formula for the matching weight $\cW_d$, of the following form:
\begin{equation}
\label{eq:intro:rec_formula_W(i,u)}
\cW_d(t,t') = \sup_{\mathfrak{m} \; : \; \cM(C_t, C_{t'})} \sum_{(u,u') \in \mathfrak{m}} \cW_{d-1}(t_u,t'_{u'}) \, ,
\end{equation}where the supremum is taken over all matchings, that is one-to-one mappings $\mathfrak{m}$ from a subset of the root's children set $C_t$ in $t$ to the root's children set $C_{t'}$ in $t'$. This recursion formula \eqref{eq:intro:rec_formula_W(i,u)} is at the heart of the analysis of this statistic as well as the design of related algorithms. The general idea is that if the trees are correlated, with positive probability they will tend to have a significantly higher matching weight than if they are independent.  We refer to Section \ref{NTMA:subsection:recursionW} for a proof of \eqref{eq:intro:rec_formula_W(i,u)} and more generally to Chapter \ref{chapter:NTMA} for the study of this method.

\subsubsection{The likelihood ratio}
In Chapter \ref{chapter:MPAlign}, we are interested in studying the existence of one-sided tests for the task, which we recall are tests guarantying an asymptotic vanishing type I error and non vanishing power. According to the Neyman-Pearson Lemma, optimal one-sided tests are based on the \emph{likelihood ratio} $L_d$ of the distributions under the distinct hypotheses $\dPls_{d}$ and $\dPl_{d}$. For a pair of trees $(t,t')$, this likelihood ratio is given by
\begin{equation}
\label{eq:intro:def_LR}
L_d(t,t'):= \frac{\dPls_{d}(t,t')}{\dPl_d(t,t')} \, .
\end{equation}

\textit{Recursive computation of $L_d$.}
This likelihood ratio also satisfies the following nice recursive property:

\begin{equation}\label{eq:intro:LR_rec}
L_d(t,t')=\sum_{k=0}^{c \wedge c'}\psi(k,c,c')\sum_{\substack{\sigma : [k] \to [c] \\ \sigma' : [k] \to [c']}}\prod_{i=1}^k L_{n-1}(t_{\sigma(i)},t'_{\sigma'(i)}) \, ,
\end{equation} where $c$ (resp. $c'$) is the degree of the root in $t$ (resp. in $t'$), and the second sum of the RHS is taken over injective mappings $\sigma$ and $\sigma'$. The coefficients $\psi(k,c,c')$ are given by
\begin{equation*}
\psi(k,c,c') = e^{\lambda s} \times \frac{s^k \bar{s}^{c+c'-2k}}{\lambda^{k} k!} \, .
\end{equation*}

The idea here again is that with positive probability the likelihood ratio is going to be significantly larger for correlated trees than for independent pairs, hence a test $\cT_d$ of the form $\cT_d(t,t') = \one_{L_d(t,t')>\beta_d}$ for a well chosen threshold $\beta_d$ should solve one-sided detection. We refer to Section \ref{MPAlign:subsection:recursion_L} for the details and proof of \eqref{eq:intro:LR_rec}, and more generally to Chapter \ref{chapter:MPAlign} for a thorough study of this method.

\subsection{Heuristics for partial graph alignment}\label{intro:subsection:heuristics_tree_graph}

We briefly state the results that establish a link between tree correlation detection and graph alignment. Given that the tests considered above are one-sided tests, we are going to perform \emph{one-sided partial recovery}.
 
\subsubsection{One-sided partial recovery}
In order to define this notion, we are left to consider estimators of $\pi^\star$ that are no longer necessarily permutations, but only one-to-one functions from a subset $\cC \subset [n]$ of the node set of $G$ to the node set of $H$ -- which we recall is also $[n]$. 

For any subset $\cC \subset [n]$, the performance of any one-to-one estimator $\hat{\pi}: \cC \to [n]$ is still assessed through its overlap $\ov(\hat{\pi},\pi^\star)$, defined as in \eqref{eq:intro:def_overlap} by:
\begin{equation*}
\ov(\hat{\pi},\pi^\star)=\frac{1}{n}\sum_{u \in \cC}\one_{\hat{\pi}(u) = \pi^\star(u)} \, .
\end{equation*} Note that the estimator may not be in $\cS_n$, and only consist in a partial matching. We also need to define the \emph{error fraction} of $\hat{\pi}$ with the unknown permutation $\pi^\star$: 
\begin{equation}\label{eq:intro:def_error}
\err(\hat{\pi},\pi^\star) := \frac{1}{n}\sum_{u\in \cC}\one_{\hat{\pi}(u) \neq \pi^\star(u)}  = \frac{\card{\cC}}{n} - \ov(\hat{\pi},\pi^\star).
\end{equation} 
A sequence of injective estimators $\{\hat{\pi}_n\}_n$ -- omitting the dependence in $n$ -- is said to achieve one-sided partial recovery if there exists some $\eps>0$ such that w.h.p. $ \ov(\hat{\pi},\pi^\star) > \eps$ and also $\err(\hat{\pi},\pi^\star) = o(1)$.

We end this Section and the introduction with an informal statement that will be discussed further in Chapter \ref{chapter:MPAlign}.

\begin{informal*}
	For given $(\lambda,s)$, if there exists a one-sided test for tree correlation detection, then one-sided partial alignment in the correlated \ER model $\cG(n,\lambda/n,s)$ is achieved in polynomial time by the \alg{MPAlign} algorithm defined in Chapter \ref{chapter:MPAlign}.
\end{informal*}

\chapter{Alignment of graph databases with Gaussian weights: fundamental limits}\label{chapter:gaussian_alignment_IT}
In this chapter, we study the fundamental limits for reconstruction in weighted graph (or matrix) database alignment. We consider the Wigner model $\Wig'(n,\rho)$ \eqref{eq:CGW_model}, and we prove that there is a sharp threshold for exact recovery of $\pi^\star$: if $n \rho^2 \geq (4+\eps) \log n + \omega(1)$ for some $\eps>0$, there is an estimator  $\hat{\pi}$ -- namely the MAP estimator -- based on the observation of databases $A,B$ that achieves exact reconstruction with high probability. Conversely, if $n \rho^2 \leq 4 \log n - \log \log n - \omega(1)$, then any estimator $\hat{\pi}$ verifies $\hat{\pi}=\pi^\star$ with probability $o(1)$. 
		
This result shows that the information-theoretic threshold for exact recovery is the same as the one obtained for detection in \cite{Wu20}: in other words, for Gaussian weighted graph alignment, the problem of reconstruction is fundamentally not more difficult than that of detection. 
		
The proofs build upon the analysis of the MAP estimator and the second moment method -- introduced earlier in Section \ref{intro:subsection:basics_rg} -- together with the study of the correlation structure of energies of permutations.\\

This chapter is based on the paper \textit{Sharp threshold for alignment of graph databases with gaussian weights} \cite{Ganassali21MSML}, published at \emph{MSML 2021}.
	
\section{Introduction}
\subsection{Aligning databases}
We address the following problem: suppose that we have two databases consisting in weighted graphs represented by their adjacency matrices $A$ and $B$. For simplicity, assume that the two graphs have same size and that each individual appears in both graphs. For a given individual, its attached signal consists in weighted edges with all other users. Across databases, edges that correspond to pairs of matched individuals are correlated. We consider the following question: \emph{if the graphs are shown unlabeled (that is, if users are anonymized), is it possible to recover the corresponding matching between databases by aligning them at the sight of their correlation structure?}

Intuitively, when the matrices are correlated enough, one can learn the true matching between individuals present in the databases. We investigate the precise conditions on correlation under which exact reconstruction (or perfect de-anonymization) is feasible with high probability.

As mentioned in Section \ref{intro:subsection:motivations}, \emph{de-anonymization problems} aroused great interest when \cite{Narayanan08} were able to de-anonymize an unlabeled dataset of film ratings with the observation of a publicly available database, using correlations between the ratings. Since then, some authors have sought to quantify privacy issues related to databases \cite{Dwork08} or social networks \cite{Narayanan09}, one of the starting points of the
widespread attention given on the more general \emph{graph alignment problem}. We refer to Section \ref{intro:subsection:motivations} for further applications and to Section \ref{intro:subsection:short_survey} for a survey of theoretical results.
	
\subsubsection{Vector-shaped and graph-shaped databases} From the theoretical point of view, fundamental limits for the deanonymisation problem are now well understood when data only consists in vectors  $u,v$ of size $n$ (or more generally, rectangular databases of size $n \times k$) \cite{Cullina18data, Dai19}, that is when each user has its own signal, regardless of its connections with others. In this setting, the problem can be phrased in terms of a \emph{Linear Assigment Problem (LAP)}:
\begin{equation}\label{eq:LAP_2}
\argmax_{\Pi} \langle \Pi u,v \rangle,
\end{equation} where the maximum runs over all permutation matrices of size $n$. LAP can be solved efficiently in $O(n^3)$ steps using the classical Hungarian algorithm (\cite{Kuhn55}).

Another related problem is that of linear regression with an unknown permutation, studied in \cite{Pananjady16}: this time, one observes $y = \Pi^\star A x^\star + w$, where $x^\star \in \dR^d$ is  an unknown vector, $\Pi^\star$ is an unknown $n \times n$ permutation matrix, and $w \in \dR^n$ is additive Gaussian noise. Here again, the permutation $\Pi^\star$ applies only on the left side of $A$, which corresponds to row permutation.

On the other hand, we recall that when the databases are graphs, the problem is different and can be phrased this time in terms of a \emph{Quadratic Assigment Problem (QAP)}:
\begin{equation}\label{eq:QAP_2}
\argmax_{\Pi} \langle A ,\Pi B \Pi^T \rangle.
\end{equation} A significant difference with the previous vector-shaped setting is that this problem is known to be NP-hard in the worst case, as well as some of its approximations \cite{Makarychev14,Pardalos94}. In the case where the signal lies in the graph structure itself -- that is, when the pairs $(A_{u,v}, B_{\pi^\star(i),\pi^\star(j)})$ are correlated pairs of Bernoulli variables -- \cite{Cullina2017} shows that there exists a sharp threshold for exact recovery, where the signal-to-noise ratio can be expressed in the correlated \ER model in terms of the size $n$ of both graphs, the marginal edge probability $q$ and the correlation parameter $s$ between edges of the two graphs. Namely, this sharp threshold is at $nqs \sim \log n$. 

This chapter focuses on the case where signal lies in weights on edges between all pairs of nodes. We recall hereafter the correlated Gaussian Wigner model $\Wig'(n,\rho)$ defined in \eqref{eq:CGW_model}.

\subsubsection{Model of Gaussian Wigner matrices} 
In the correlated Gaussian Wigner model $\Wig'(n,\rho)$ \eqref{eq:CGW_model}, the weighted adjacency matrices $A$ and $B$ of the two graphs $G$ and $H$ are symmetric, and sampled as follows: first draw the planted permutation $\pi^\star$ uniformly at random in $\cS_n$. Then all pairs of edge weights $(A_{u,v}, B_{\pi^\star(u),\pi^\star(v)})_{1 \leq i<j \leq n}$ are i.i.d. couples of normal variables with zero mean, unit variance and correlation parameter $\rho \in [0,1]$.
Since all Gaussian variables are independent from $\pi^\star$, matrix $B$ can also be drawn from $A$ as follows: 
\begin{equation}\label{eq:model}
B = \rho \cdot \Pi^{\star \top} A\Pi^\star + \sqrt{1-\rho^2}  \cdot H,
\end{equation} where $H$ is an independent copy of $A$, and $\Pi^\star$ is the $n \times n$ matrix representation of permutation $\pi^\star$, that is $\Pi^\star_{u,v} = \one_{v = \pi^\star(u)}$.

\begin{figure}[h]
	\centering
	\includegraphics[scale=1.3]{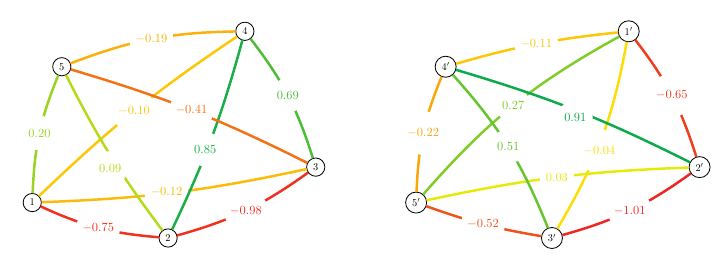}
	\caption{A sample from model \eqref{eq:model} with $n=5$. For representation, edges are colored according to their weights, and the underlying alignment is $u \mapsto u'$ for $u \in \left\lbrace 1,2,3,4,5 \right\rbrace $.}
	\label{fig:model} 
\end{figure}

\subsubsection{Detection problem} A most recent paper (\cite{Wu20}) studies fundamental limits for detection, both in correlated Gaussian weighted and correlated \ER graphs. This time, the problem is as follows: \emph{given $A,B$, are we able to distinguish between model \eqref{eq:model} and a null model, where the two graphs are just independent Gaussian weighted graphs?} Intuitively, this problem is less demanding than that of exact alignment, since the task is to detect -- wherever in the graph -- the presence of a hidden planted alignment. Under the same model \eqref{eq:model}, Y. Wu, J. Xu and S. Yu showed that strong detection is feasible with high probability if $n \rho^2 \geq 4 \log n$, whereas it is impossible if $n \rho^2 \leq (4-\eps) \log n$ for some $\eps>0$. Their study builds on an analysis of the likelihood ratio, as often done in detection problems. The contribution of this chapter is to show that this sharp detection threshold is also that of exact reconstruction. Interestingly, for Gaussian weighted graph alignment, the problem of reconstruction is in fact fundamentally not more difficult than that of detection.

After this study was completed, the author was made aware of recent and independent work conducted by \cite{Wu2021SettlingTS}, which also obtains -- among other things -- the results of this study, albeit with different proof techniques.

\subsection{Main results}
In the sequel, we work with the correlated Gaussian Wigner model described in \eqref{eq:model}, and establish the precise (sharp) threshold for exact recovery of $\pi^\star$ in this model.
\begin{theorem}[Achievability part]
	\label{GA_IT:theorem:GA_IT_positive}
	If for $n$ large enough
	\begin{equation}\label{eq:cond_ach_glo}
	\rho^2 \geq \frac{(4+\eps) \log n}{n}
	\end{equation}for some $\eps>0$, then there is an estimator (namely, the MAP estimator) $\hat{\pi}$ of $\pi^\star$ given $A,B$ such that $\hat{\pi}=\pi^\star$ with probability $1-o(1)$.
\end{theorem}
\begin{theorem}[Converse part]
	\label{GA_IT:theorem:GA_IT_negative}
	Conversely, if
	\begin{equation}\label{eq:cond_imp_glo}
	\rho^2 \leq \frac{4 \log n - \log \log n - \omega(1)}{n}
	\end{equation}then any estimator $\hat{\pi}$ of $\pi$ given $A,B$ verifies $\hat{\pi}=\pi^\star$ with probability $o(1)$.
\end{theorem}

\subsubsection{Computational limits of exact recovery}
For the correlated Gaussian Wigner model \eqref{eq:model}, several algorithms have been studied, usually as a first step in order to analyze further graph alignment algorithms. The state-of-the-art polynomial-time algorithms are either based on \emph{degree profiles} \cite{Ding18}, or on a spectral method \cite{Fan2019Wigner}. In both cases, these methods require the noise parameter $\sqrt{1-\rho^2}$ to be $O\left(\log^{-1} n\right)$. In Chapter \ref{chapter:EIG1}, we will study a simpler algorithm with lower computational complexity, requiring $\sqrt{1-\rho^2}$ to be $O(n^{-7/6})$ \cite{GLM19}. In any case, $\rho$ needs to tend to $1$, and the regimes in which these methods work well are far from the fundamental limits established in Theorems \ref{GA_IT:theorem:GA_IT_positive} and \ref{GA_IT:theorem:GA_IT_negative}. The main result of this chapter thus corroborates the idea that matrix alignment may be computationally hard even  in the feasibility regime. In other words, the hard phase can be conjectured to be wide for this reconstruction problem. Proving a result of that form however remains a very thorny question.

\subsubsection{Organization of the chapter}
We first some notations at the beginning of Section \ref{section:preliminaries}, and then establish a control on correlations between energies of permutations, using Hanson-Wright inequality. The achievability result is proved in Section \ref{section:ach}: after showing that the classical first moment method fails, we take advantage of the correlation structure established before to handle the sharp bound. Then, second moment method is applied in Section \ref{section:conv} to show that lots of small perturbations of the true underlying permutation have lower energies, establishing the converse bound. Finally, some additional proofs are deferred to Appendix \ref{appendix}. 
The proof techniques are not far from those used by \cite{Dai19}, the main novelty being the use of correlation of energies, which is essential to both achievability and impossibility result.

\section{Preliminaries}\label{section:preliminaries}
\subsection{Definitions and notations}\label{subsection:def}
Recall that for any positive integer $n$, $[n] := \left\lbrace 1,2,\ldots,n \right\rbrace $. For two positive sequences $\left\lbrace u_n\right\rbrace $ and $\left\lbrace v_n\right\rbrace $, denote $u_n = O(v_n)$ if there exists $C>0$ such that $u_n \leq Cv_n$ for all $n$. We will also write $u_n = o(v_n)$ (resp. $u_n = \omega(v_n)$) if $u_n/v_n \to 0$ (resp. $v_n/u_n \to 0$). All limits considered are taken when $n \to \infty$.\\

\textit{Linear algebra.} We work with the canonical euclidean norm $\|\cdot \|$ on $\dR^n$, and $\langle \cdot, \cdot \rangle$ the canonical inner product on $\dR^n$ or $\dR^{n \times n}$. For any $n \times n$ matrix $M$ with real entries, its \emph{Frobenius norm} $\|M\|_{F}$ and its \emph{operator norm} $\|M\|_{\mathrm{op}}$ are defined as follows:
\begin{equation*}
\|M\|_{F}  := \left(\sum_{1\leq u,v \leq n} A_{u,v}^2\right)^{1/2} \quad \mbox{and} \quad
\|M\|_{\mathrm{op}}  := \sup_{X \in \dR^n \setminus \left\lbrace 0\right\rbrace } \frac{\| M X \|}{\|X\|}. 
\end{equation*} Note that for any normal matrix (that is, if $M^T M = M M^T$), then $\|M\|_{\mathrm{op}}$ equals $\rho(M)$, the spectral radius of $M$.\\

\textit{Probability.} When working with model \eqref{eq:model}, we will denote by $\dP_A$ (resp. $\dE_A$) the conditional probability (resp. the conditional expectation) with respect to the random matrix $A$. We recall that $\cN(\mu,v)$ denotes a Gaussian variable (resp. vector) with mean $\mu$ and variance (resp. covariance matrix) $v$. Such a Gaussian variable (resp. vector) is called \emph{standard} if $\mu=0$ and $v=1$ (resp. $v$ is the identity matrix). We say that an event $\cA_n$ happens \emph{with high probability (w.h.p)} if $\dP(\cA_n) \to 1$ when $n \to \infty$.\\

\textit{Permutations.} 
We denote by $\cS_m$ the set of permutations of $[m]$. To any permutation $\sigma \in \cS_m$, we can associate its $m \times m$ matrix representation $\Sigma$ defined by $\Sigma_{u,v} = \one_{v = \sigma(u)}$. Define $\mathcal{F}_{\sigma}$ the set of \emph{fixed points} of $\sigma$:
\begin{equation}
\mathcal{F}_{\sigma} := \left\lbrace u \in [m], \sigma(u) = u\right\rbrace,
\end{equation} and denote $f_{\sigma} := \card{\mathcal{F}_{\sigma}}$. Similarly, we define the set of \emph{unfixed points} of $\sigma$:
\begin{equation}
\mathcal{D}_{\sigma} := [m] \setminus \mathcal{F}_{\sigma} = \left\lbrace i \in [m], \sigma(i) \neq i\right\rbrace,
\end{equation} and we denote $d_{\sigma} :=  \card{ \mathcal{D}_{\sigma}}$. For any $d \in \left\lbrace 0,\ldots,m \right\rbrace $ we define $\cS_{m,d}$ the set of permutations of $\cS_{m}$ with exactly $d$ unfixed points. Note that $ \card{ \cS_{m,1}}=0$ and that we have the inequality\begin{equation}\label{eq:ineq_S_n,d}
 \card{\cS_{m,d}} = \binom{m}{m-d}  \card{\set{\sigma \in \cS_{d}, F_{\sigma} = 0 }} \leq \binom{m}{m-d} d! \leq m^d.
\end{equation} We recall that similarity between two permutations $\sigma, \sigma' \in \cS_n$ is measured by their \emph{overlap}:
\begin{equation*}\label{eq:overlap}
\ov(\sigma,\sigma') := \frac{1}{n} \sum_{u=1}^{n} \mathbf{1}_{\sigma(u)=\sigma'(u)} = \frac{1}{n} f_{\sigma^{-1} \circ \sigma'}. 
\end{equation*}

Observe that on a graph of size $n$, each permutation $\sigma$ of the vertices $[n]$ has a natural extension to a canonical permutation on edges $\sigma^{\mathrm{E}} : \binom{[n]}{2} \to \binom{[n]}{2}$ defined as follows:
\begin{equation*}\label{eq:sigma_E}
\sigma^{\mathrm{E}} : e = \left\{u,v\right\} \mapsto \sigma^{\mathrm{E}}(e) = \left \{ \sigma(u),\sigma(jv)\right\}.
\end{equation*} Note that the mapping $\sigma \mapsto \sigma^{\mathrm{E}}$ is one-to-one as soon as $n\geq 3$, since for all $u \in [n]$ and $v\neq v' \in [n] \setminus \left\lbrace u\right\rbrace $, edges $\sigma^{\mathrm{E}}(\left\lbrace u,v\right\rbrace )$ and $\sigma^{\mathrm{E}}(\left\lbrace u,v'\right\rbrace )$ have only one node in common, which is $\sigma(u)$. We will use the notation $\mathcal{F}^{\mathrm{E}}_{\sigma} := \mathcal{F}_{\sigma^{\mathrm{E}}}$ (resp. $\mathcal{D}^{\mathrm{E}}_{\sigma} := \mathcal{D}_{\sigma^{\mathrm{E}}}$) the set of \emph{fixed edges} (resp. \emph{unfixed edges}) of $\sigma$. Similarly we denote ${f}^{\mathrm{E}}_{\sigma} := {f}_{\sigma^{\mathrm{E}}}$ and $f^{\mathrm{E}}_{\sigma} := d_{\sigma^{\mathrm{E}}}$, for brievity.

Note that $d^{\mathrm{E}}_{\sigma}$ and are $d_{\sigma}$ are closely tied, since for all $\sigma \in \cS_n$, we have the inequality
\begin{equation}\label{eq:ineq_d}
d_{\sigma} \left(n - \frac{d_{\sigma}}{2}\right) \leq d^{\mathrm{E}}_{\sigma} \leq d_{\sigma} \left(n - \frac{d_{\sigma}-1}{2}\right).
\end{equation}
Indeed, observe that 
\begin{itemize}
\item[$(i)$] the number of fixed edges is at least the number of pairs of fixed points, and
\item[$(ii)$] the number of fixed edges is exactly the number of pairs of fixed points plus the number of pairs $(u,v), u<v$ that are exchanged by $\sigma$ (that is, the number of transpositions), this number being at most $d_{\sigma}/2$.
\end{itemize}
These remarks give that
\begin{equation*}
\binom{n-d_{\sigma}}{2} \leq \binom{n}{2} - d^{\mathrm{E}}_{\sigma} \leq \binom{n-d_{\sigma}}{2} + \frac{d_{\sigma}}{2},
\end{equation*} which directly implies \eqref{eq:ineq_d}.
\begin{remark}\label{remark:equiv_d}
Note that inequality \eqref{eq:ineq_d} gives the almost sure equivalents $d^{\mathrm{E}}_{\sigma} \sim d_{\sigma}n$ when $d_{\sigma}=o(n)$, and $d^{\mathrm{E}}_{\sigma} \sim \frac{1}{2}\alpha(2-\alpha)n^2$ when $d_{\sigma}=\alpha n$. In any case, $d^{\mathrm{E}}_{\sigma} \in \left[\frac{1}{2} d_{\sigma} n , d_{\sigma} n \right]$.
\end{remark}

\subsection{MAP estimation, relative energy of permutations}  
Since $\pi^\star$ is uniformly chosen, we work in a Bayesian setting: let us evaluate the posterior probability density of $\pi^\star$ given $A,B$:
\begin{flalign*}
p_{\pi^\star|A,B}\left(\pi | a,b\right) & \propto p_{\pi^\star,A,B}\left(\pi,a,b\right)\\
& \propto \exp\left(-  \frac{1}{2 (1-\rho^2)}\sum_{1\leq i<j\leq n} \left(B_{\pi(u),\pi(v)} - \rho A_{u,v} \right)^2 \right),
\end{flalign*} where $\propto$ indicates equality up to some factors that do not depend on $\sigma$. Define the \emph{loss function}
\begin{equation}\label{eq:L_pi}
\cE(\pi,A,B) := \sum_{1\leq i<j\leq n} \left(B_{\pi(u),\pi(v)} - \rho A_{u,v} \right)^2.
\end{equation} 
This loss function can also be viewed as the \emph{energy} associated with permutation $\pi$. Note that the posterior distribution is a Gibbs measure corresponding to this energy $\cE$, with inverse temperature $\beta =  \frac{1}{2(1-\rho^2)}$. The MAP (maximum a posteriori) estimator is thus
\begin{equation}\label{eq:MAP}
\hat{\pi}_{\MAP} := \argmax_{\pi} p_{\pi^\star|A,B}\left(\pi | A,B\right) = \argmin_{\pi} \cE(\pi,A,B),
\end{equation} where the minimum is taken over all permutations $\pi \in \cS_n$. As previously stated in Section \ref{intro:subsection:planted_GA}, the above formulation \eqref{eq:MAP} is standard in the literature of graph alignment and meets the classical QAP formulation \eqref{eq:QAP_2}, since 
$$\argmin_{\pi} \cE(\pi,A,B) = \argmax_{\Pi} \langle A, \Pi B \Pi^T\rangle. $$

Theory from Bayesian optimal estimation guarantees that the best possible estimator for our exact reconstruction problem, in the Bayes risk sense, is $\hat{\pi}_{\MAP}$. Thus, if MAP estimator fails with high probability, then no estimator can succeed. This is why this estimator is often studied in exact reconstruction problems, as already done in previous works (\cite{Cullina2017,Cullina18data,Dai19}).

From now on we work conditionally on $\pi^\star$ which can always be assumed to be $\mathrm{id}$ without loss of generality. More precisely, we will make the variable change $\sigma = \pi^\star \circ \pi^{-1}$ ; writing $B$ as a function of $\sigma, A$ and $H$, \eqref{eq:L_pi} becomes
\begin{flalign*}
\cE(\sigma,A,H) & = \rho^2 \sum_{1\leq i<j\leq n} \left(A_{u,v} -  A_{\sigma(u),\sigma(v)} \right)^2 - 2\rho \sqrt{1-\rho^2} \sum_{1\leq i<j\leq n}  H_{u,v} \left(A_{u,v} -  A_{\sigma(u),\sigma(v)} \right)\\ & + (1-\rho^2) \sum_{1\leq i<j\leq n} H_{u,v} ^2 .
\end{flalign*} 
The loss function $\cE$ applied to the ground truth $\pi=\pi^\star$ -- that is $\sigma = \id$ -- gives the energy reference $(1-\rho^2) \sum_{1\leq i<j\leq n} H_{u,v} ^2$. In order to compare any $\pi$ with $\pi^\star$ -- or any $\sigma$ with $\id$ -- we further define the \emph{relative energy} of a permutation $\sigma \in \cS_n$:
\begin{flalign}
\delta(\sigma) &:= \cE(\sigma,A,H)-\cE(\mathrm{id},A,H) \nonumber\\
& = \rho^2 \sum_{1\leq i<j\leq n} \left(A_{u,v} -  A_{\sigma(u),\sigma(v)} \right)^2 - 2\rho \sqrt{1-\rho^2} \sum_{1\leq i<j\leq n}  H_{u,v} \left(A_{u,v} -  A_{\sigma(u),\sigma(v)} \right).\label{eq:delta_1} 
\end{flalign}
We next omit in our notations the dependency on $A$ and $H$ of $\delta(\sigma)$.
\begin{remark}
This relative energy $\delta$, also introduced by \cite{Cullina2017} for \ER graph alignment, is a measurement of the quality of a proposed alignment: $\delta(\sigma) \leq 0$ means that $\sigma^{-1} \circ \pi^\star$ is a better alignment than $\pi^\star$ for $A$ and $B$ in the posterior sense. A crucial set is then
\begin{equation*}\label{eq_Q}
\cQ := \left\lbrace \sigma \in \cS_n, \, \delta(\sigma) \leq 0 \right\rbrace.
\end{equation*}
Points of $\cQ$ are alignments on which the posterior distribution puts important weights -- at least greater weights than that of the ground truth -- or equivalently points of low energy. Note that $\id \in \cQ$. 
\end{remark}

In view of \eqref{eq:delta_1}, conditionally on $A$, $\delta(\sigma)$ is as follows:
\begin{equation}\label{eq:delta_gaussien_general}
\delta(\sigma) = \rho^2 v_\sigma - 2 \rho \sqrt{1-\rho^2} X_{\sigma},
\end{equation} where 
\begin{equation*}\label{eq:def_v_sigma}
v_\sigma := \sum_{1\leq i<j\leq n} \left(A_{u,v} -  A_{\sigma(u),\sigma(v)} \right)^2,
\end{equation*} and $X=(X_\sigma)_{\sigma \in \cS_{n}}$ is a Gaussian vector, centered, with covariance given by
\begin{equation*}\label{eq:def_c_sigma}
\Cov(X_\sigma,X_{\sigma'}) = \sum_{1\leq i<j\leq n} \left(A_{u,v} -  A_{\sigma(u),\sigma(v)} \right)\left(A_{u,v} -  A_{\sigma'(i),\sigma'(j)} \right) := c_{\sigma,\sigma'}.
\end{equation*}Note that for all $\sigma \in \cS_n$, $c_{\sigma,\sigma} = v_\sigma$. Elaborating on the correlation structure of these relative energies is the object of the end of this section.

\subsection{Control of covariance structure of relative energies}
For all $\sigma, \sigma' \in \cS_{n}$, $c_{\sigma,\sigma'}$ can be written as follows
\begin{equation*}\label{eq:c_sigma_E}
c_{\sigma,\sigma'} = \sum_{e \in \binom{[n]}{2}} \left(A_e -  A_{\sigma^E(e)} \right) \left(A_e -  A_{\sigma'^E(e)} \right)
\end{equation*} and satisfies
\begin{equation*}\label{eq:c_exp}
\dE \left[c_{\sigma,\sigma'} \right] =
 \card{\cD^{\mathrm{E}}_\sigma \cap \cD^{\mathrm{E}}_{\sigma'}} +  \card{\cD^{\mathrm{E}}_\sigma \cap \cD^{\mathrm{E}}_{\sigma'} \cap \cF^{\mathrm{E}}_{\sigma^{-1} \circ \sigma'}}.
\end{equation*} In particular,
\begin{equation*}\label{eq:v_exp}
\dE \left[v_{\sigma} \right] = d^{\mathrm{E}}_{\sigma} +d^{\mathrm{E}}_{\sigma} = 2 d^{\mathrm{E}}_{\sigma}.
\end{equation*} 

Random variables $c_{\sigma,\sigma'}$ only depend on the entries of $A$, which are Gaussian. Moreover, $c_{\sigma,\sigma'}$ being a quadratic form evaluated on a Gaussian vector, it can be controlled using Hanson-Wright inequality:
\begin{lemma}[Hanson-Wright inequality (\cite{HansonWright1971})] \label{lemma:HW_ineq}
	Let $X$ be a standard Gaussian vector, and $M$ a deterministic matrix. Then there exists a universal constant $c>0$ such that with probability at least $1-2\delta$:
	\begin{equation}\label{eq:HW_ineq}
	\big| X^T M X - \mathrm{Tr} M \big| \leq c \left(\|M\|_F \sqrt{\log(1/\delta)} + \|M\|_{\mathrm{op}} \log(1/\delta)\right).
	\end{equation}
\end{lemma}
We refer to \cite{HansonWright1971} for a proof. Inequality \eqref{eq:HW_ineq} used in our context leads to the following

\begin{corollary} \label{corollary:control_C}
	There exists a universal constant $C>0$ such that with high probability, for every $d \in \left\lbrace 2,\ldots, n \right\rbrace$, for all $\sigma,\sigma' \in \mathcal{S}_{n,d}$, 
	\begin{equation*}\label{eq:control_C}
	\big|c_{\sigma,\sigma'} -  \card{\cD^{\mathrm{E}}_\sigma \cap \cD^{\mathrm{E}}_{\sigma'}} -  \card{\cD^{\mathrm{E}}_\sigma \cap \cD^{\mathrm{E}}_{\sigma'} \cap \cF^{\mathrm{E}}_{\sigma^{-1} \circ \sigma'}} \big|  \leq   C d\sqrt{n \log n}.
	\end{equation*}
\end{corollary}

\begin{proof}
	We first make the following observation: for any $\sigma,\sigma' \in \mathcal{S}_n$, 
	\begin{flalign*}
	c_{\sigma,\sigma'} & = \sum_{e} \left(A_{e} - A_{\sigma(e)}\right) \left(A_{e} - A_{\sigma'(e)}\right)\\
	& = A^T (I_N - \Sigma)^T(I_N - \Sigma') A,
	\end{flalign*} where $A = (A_e)_{e}$ is viewed as a standard Gaussian vector of size $N=\binom{n}{2}$, and $\Sigma$ (resp. $\Sigma'$) is the $N \times N$ permutation matrix associated with $\sigma^E$ (resp. $\sigma'^E$). Note that
	\begin{flalign*}
	\Tr ( (I_N - \Sigma)^T(I_N - \Sigma')) &= N - f^{\mathrm{E}}_{\sigma} - f^{\mathrm{E}}_{ \sigma'} +f^{\mathrm{E}}_{\sigma^{-1} \circ \sigma'} \\ 
	&\overset{(a)}{=}  \card{\cD^{\mathrm{E}}_\sigma \cap \cD^{\mathrm{E}}_{\sigma'}} +  \card{\cD^{\mathrm{E}}_\sigma \cap \cD^{\mathrm{E}}_{\sigma'} \cap \cF^{\mathrm{E}}_{\sigma^{-1} \circ \sigma'}},
	\end{flalign*} where $(a)$ is obtained by noticing that 
	\begin{flalign*}
	 \card{\cD^{\mathrm{E}}_\sigma \cap \cD^{\mathrm{E}}_{\sigma'}} +  \card{\cD^{\mathrm{E}}_\sigma \cap \cD^{\mathrm{E}}_{\sigma'} \cap \cF^{\mathrm{E}}_{\sigma^{-1} \circ \sigma'}} & = d^{\mathrm{E}}_\sigma + d^{\mathrm{E}}_{\sigma'} -  \card{\cD^{\mathrm{E}}_\sigma \cup \cD^{\mathrm{E}}_{\sigma'}} + f^{\mathrm{E}}_{\sigma^{-1} \circ \sigma'} -  \card{\cF^{\mathrm{E}}_{\sigma} \cup \cF^{\mathrm{E}}_{\sigma'}}
	\end{flalign*} and that 
	$
	 \card{\cD^{\mathrm{E}}_\sigma \cup \cD^{\mathrm{E}}_{\sigma'}} +  \card{\cF^{\mathrm{E}}_{\sigma} \cup \cF^{\mathrm{E}}_{\sigma'}} = N.
	$
	For a fixed $d$ and $\sigma,\sigma' \in \mathcal{S}_{n,d}$, one has
	\begin{flalign*}
	\|(I_N - \Sigma)^T(I_N - \Sigma')\|_F \leq \|(I_N - \Sigma')\|_F + \| \Sigma^T(I_N - \Sigma')\|_F
	& = 2 \|(I_N - \Sigma') \|_F \\
	& \leq 2 \sqrt{2 d^{\mathrm{E}}_{\sigma'}} \\
	&  \leq 2 \sqrt{2 d n},
	\end{flalign*} where we used \eqref{eq:ineq_d} in the last step. One also has
	\begin{flalign*}
	\|(I_N - \Sigma)^T(I_N - \Sigma')\|_{\mathrm{op}} &\leq \rho(I_N-\Sigma) \times \rho(I_N - \Sigma') \\
	& \leq 2 \times 2 = 4. 
	\end{flalign*} Taking $\delta = n^{-(2d+2)}$, Lemma \ref{lemma:HW_ineq} gives that with probability at least $1-2\delta$, 
	\begin{flalign}
	\label{eq:concentration_c}
	\big|c_{\sigma,\sigma'} -  \card{\cD^{\mathrm{E}}_\sigma \cap \cD^{\mathrm{E}}_{\sigma'}} -  \card{\cD^{\mathrm{E}}_\sigma \cap \cD^{\mathrm{E}}_{\sigma'} \cap \cF^{\mathrm{E}}_{\sigma^{-1} \circ \sigma'}} \big|  & \leq c \left(2\sqrt{2}\sqrt{d(2d+2)} \sqrt{n \log n} + 4 (2d+2) \log n\right) \nonumber \\
	& \leq C d\sqrt{n \log n}, 
	\end{flalign}for some universal constant $C>0$. The proof is concluded by checking that this inequality holds w.h.p. for all $d$ and $\sigma,\sigma' \in \cS_{n,d}$ : the probability that at least one pair $(\sigma,\sigma')$ contradicts \eqref{eq:concentration_c} is upper bounded by
	\begin{equation*}
	n \times \card{\cS_{n,d}} ^2 \times 2 \delta \leq 2n^{1+2d-2d-2} =o(1).
	\end{equation*}
\end{proof}

In the rest of the chapter we define the event
\begin{equation}\label{eq:event_A}
\cA := \left\lbrace \forall d \in [n], \forall \sigma,\sigma' \in \cS_{n,d}, \big|c_{\sigma,\sigma'} -  \card{\cD^{\mathrm{E}}_\sigma \cap \cD^{\mathrm{E}}_{\sigma'}} -  \card{\cD^{\mathrm{E}}_\sigma \cap \cD^{\mathrm{E}}_{\sigma'} \cap \cF^{\mathrm{E}}_{\sigma^{-1} \circ \sigma'}} \big|   \leq   C d\sqrt{n \log n} \right\rbrace,
\end{equation} which happens with probability $1-o(1)$ by Corollary \ref{corollary:control_C}.

\section{Achievability result}\label{section:ach}
In this section, we establish the result of Theorem \ref{GA_IT:theorem:GA_IT_positive}.
\subsection{Failure of first moment method}\label{subsection:first_moment}
For the achievability result, the first strategy is to use the union bound (or first moment method) to show that under condition \eqref{eq:cond_ach_glo} of Theorem \ref{GA_IT:theorem:GA_IT_positive},
\begin{equation*}\label{eq:union_bound_ach}
\dP\left(\mbox{MAP fails} \right) = \dP\left( \hat{\pi}_{\MAP} \neq \pi \right) =o(1).
\end{equation*} As described hereafter, this naive method does not give the correct bound. Indeed, let us evaluate $\dP\left(\delta(\sigma) \leq 0 \right)$ for a given $\sigma \neq \id$. In view of the conditional distribution \eqref{eq:delta_gaussien_general} of $\delta(\sigma)$ we have
\begin{flalign*}
\dP\left(\delta(\sigma) \leq 0 \right) & = \dE \left[ \dE_A \left[\one_{\delta(\sigma) \leq 0 } \right] \right] = \dE \left[ \dP_A \left( \rho^2 v_\sigma - 2 \rho \sqrt{1-\rho^2} X_{\sigma} \leq 0\right)\right] \\
& = \dE \left[ \dP_A \left( \rho^2 v_\sigma - 2 \rho \sqrt{1-\rho^2} \sqrt{v_\sigma} \cdot \cN(0,1) \leq 0\right)\right] \\
& = \dE \left[ \dP_A \left(\cN(0,1) \geq \frac{\rho \sqrt{v_\sigma}}{2 \sqrt{1-\rho^2}}  \right)\right] \leq \dE \left[ \exp \left( - \frac{\rho^2}{8(1-\rho^2)} v_\sigma\right)\right],
\end{flalign*} where we used standard Gaussian concentration in the last inequality: $\dP\left(\cN(0,1) \geq t  \right) \leq \exp(-t^2/2)$. Note that on event $\cA$ defined in \eqref{eq:event_A} and inequality \eqref{eq:ineq_d},
\begin{equation*}\label{eq:controle_v_first_mom}
\forall d \in [n], \forall \sigma \in \cS_{n,d}, \, v_\sigma \geq 2 d^{\mathrm{E}}_{\sigma} - C d_{\sigma} \sqrt{n \log n} \geq  d^{\mathrm{E}}_{\sigma} \left(2-2\eps_n\right),
\end{equation*}
setting $\eps_n=2C \sqrt{\log n/n}$. Union bound then gives
\begin{flalign*}
\dP\left(\mbox{MAP fails} \right) & \leq \dP\left(\exists \sigma \in \cS_n \setminus \left\{\id\right\}, \; \delta(\sigma) \leq 0 \right) \\
& \leq o(1) + \sum_{\sigma \in \cS_n \setminus \left\{\id\right\}} \dE \left[ \exp \left( - \frac{\rho^2}{8(1-\rho^2)} v_\sigma\right) \one_{\cA}\right] \\
& \leq o(1) + \sum_{\sigma \in \cS_n \setminus \left\{\id\right\}} \exp \left( - \frac{\rho^2}{8(1-\rho^2)} (2-2\eps_n) d^{\mathrm{E}}_{\sigma}\right) \\
& \leq o(1) + \sum_{\sigma \in \cS_n \setminus \left\{\id\right\}} \exp \left( - \frac{\rho^2}{4} (1-\eps_n) d^{\mathrm{E}}_{\sigma}\right),
\end{flalign*}where we used $1/(1-\rho^2)>1$ in the last step. Let us now study the last sum, distinguishing the terms according to $d:=d_\sigma$: 
\begin{itemize}
	\item As long as $d=o(n)$, by Remark \ref{remark:equiv_d}, the terms behave like $\exp \left(- \frac{\rho^2}{4} (1-\eps_n) d n \right)$. By \eqref{eq:ineq_S_n,d}, $\log  \card{\cS_{n,d}} \leq d \log n$ so the partial sum is small if $\frac{\rho^2}{4} (1-\eps_n) n - \log n >0$, which gives the necessary condition $\rho^2 \geq 4 \frac{\log n}{n}$. 
	
	\item However, the situation is different when it comes to large values of $d$. For instance, let us study the contribution of \emph{derangements} to the sum (that is, $\sigma$ such that $d_\sigma=n$). Note that these derangements are very numerous (their number is $\sim e^{-1} n!$). Again by Remark \ref{remark:equiv_d}, their contribution is thus of order 
	\begin{equation*}
		e^{-1} n! \exp \left(\rho^2(1-\eps_n) n^2/8(1-o(1))\right) = \exp \left(\left(n \log n - \rho^2 n^2/8\right)(1-o(1))\right),
	\end{equation*} which gives a more restrictive condition: $\rho^2 \geq 8 \frac{\log n}{n}$. 
\end{itemize}

As seen here-above, this naive first moment method enables to ensure feasibility of exact reconstruction only in the regime where $\rho^2 \geq 8 \frac{\log n}{n}$, which is not the optimal one. This bound is actually quite rough here, because the variables are substantially correlated when $d$ gets large and their contributions make the first moment explode. We take advantage of these correlations in the next section in order to get access to the sharp bound.

\subsection{Improving the first moment method with correlations.}
For all $d \in \left\lbrace 2, \ldots, n \right\rbrace $, define $\cE_d$ the event:
\begin{equation*}\label{eq:def_Ed}
\cE_d := \left\lbrace \exists \sigma \in \cS_{n,d}, \; \delta(\sigma) \leq 0 \right\rbrace.
\end{equation*} In this Section we will assume that
\begin{equation*}
\rho \geq (2+\eps)\sqrt{\frac{\log n}{n}},
\end{equation*} for some $\eps>0$. Recall that we work on the event $\cA$ defined in \eqref{eq:event_A}, and that conditionally on entries of matrix $A$, we can write
\begin{equation}\label{eq:delta_gaussien}
\delta(\sigma) = \rho^2 v_\sigma - 2 \rho \sqrt{1-\rho^2} X_{\sigma},
\end{equation} where $X=(X_\sigma)_{\sigma \in \cS_{n,d}}$ is a Gaussian vector, centered, with covariance given by $\Cov(X_\sigma,X_{\sigma'}) = c_{\sigma,\sigma'}$. Also note that on event $\cA$, for all $d\leq \alpha n$ and $\sigma \in \cS_{n,d}$, inequality \eqref{eq:ineq_d} gives 
\begin{equation}\label{eq:equiv_v_alpha_n}
v_\sigma = (1-o(1))2dn (1-\alpha/2).
\end{equation} In view of \eqref{eq:equiv_v_alpha_n}, as previously done in Section \ref{subsection:first_moment}, naive first moment method may suffice for $d \leq \alpha n$:
\begin{flalign*}
\dP\left(\bigcup_{2 \leq d \leq \alpha n} \cE_d \right) & \leq o(1) + \sum_{d=2}^{\alpha n}  \card{\cS_{n,d}} \times \dP\left(\cN(0,1) \geq \frac{\rho \sqrt{v_{\sigma}}}{2 \sqrt{1-\rho^2}} \cap \cA\right)\\
& \leq o(1) + \sum_{d=2}^{\alpha n}  \card{\cS_{n,d}} \times \dP\left(\cN(0,1) \geq (1+\eps/2) \sqrt{2d \log n (1-\alpha/2) }(1-o(1)) \right)\\
& \leq o(1) + \sum_{d=2}^{\alpha n} \exp \left(d \log n -  d \log n (1+\eps) (1-\alpha/2) + o(d \log n)\right),
\end{flalign*} which is $o(1)$ as soon as $\alpha < \alpha_0 := \frac{2 \eps}{1-\eps/2}$. It then remains to control the probabilities $\dP(\cE_d)$ for $d \geq \alpha_0 n$. 
As mentioned earlier, we take advantage of the correlation structure in \eqref{eq:delta_gaussien}. More precisely, we show that all variables $X_\sigma$ at a given level $d=\alpha n$ have substantial positive covariance when compared to their variance -- of order $\alpha (2-\alpha) n^2$ on $\cA$ by \eqref{eq:equiv_v_alpha_n} -- as shown in Figure \ref{fig:GA_IT:plot_f}. To do so, we derive an appropriate lower bound for $c_{\sigma,\sigma'}$ for $\sigma,\sigma' \in \cS_{n,\alpha n}$. This is the scope of the following Lemma:
\begin{lemma}\label{lemma:min_corr_delta}
	With high probability, there exists a universal constant $C_1>0$ such that for any $d = \alpha n$ with fixed $\alpha>0$ and $\sigma,\sigma' \in \cS_{n,\alpha n}$:
	\begin{equation*}\label{eq:min_corr_delta}
	\Cov(X_{\sigma},X_{\sigma'}) = c_{\sigma,\sigma'} \geq f(\alpha) n^2 - C_1 n^{3/2} \log^{1/2} n ,
	\end{equation*}with
	\begin{equation}\label{eq:f_alpha}
	f(\alpha) := \left\{
	\begin{array}{ll}
	\alpha^2 & \mbox{if } \alpha<1/2 \\
	\alpha^2- \frac{1}{2}(2\alpha-1)^2 & \mbox{if } \alpha \geq 1/2
	\end{array}
	\right.
	\end{equation}
	Thus for any $\eps'>0$, with high probability, for any $d=\alpha n$ with fixed $\alpha>0$,
	\begin{equation*}\label{eq:max_control}
	\max_{\sigma \in \cS_{n,\alpha n}} X_\sigma \leq \sqrt{2 \alpha \left(\alpha(2-\alpha)-f(\alpha)\right)} n^{3/2}\log^{1/2} n + (2+\eps') n \log^{1/2} n.
	\end{equation*}
\end{lemma}


The proof of this Lemma is obtained by working on event $\cA$ defined in \eqref{eq:event_A}, and establishing a lower bound on $\card{\cD^{\mathrm{E}}_\sigma \cap \cD^{\mathrm{E}}_{\sigma'}}$, which is simply the number of edges that are deranged both by $\sigma^E$ and $\sigma'^E$. It can be found in Appendix \ref{appendix:lower_bound_corr}.

\begin{figure}[h]
	\centering
	\includegraphics[width=0.8\textwidth]{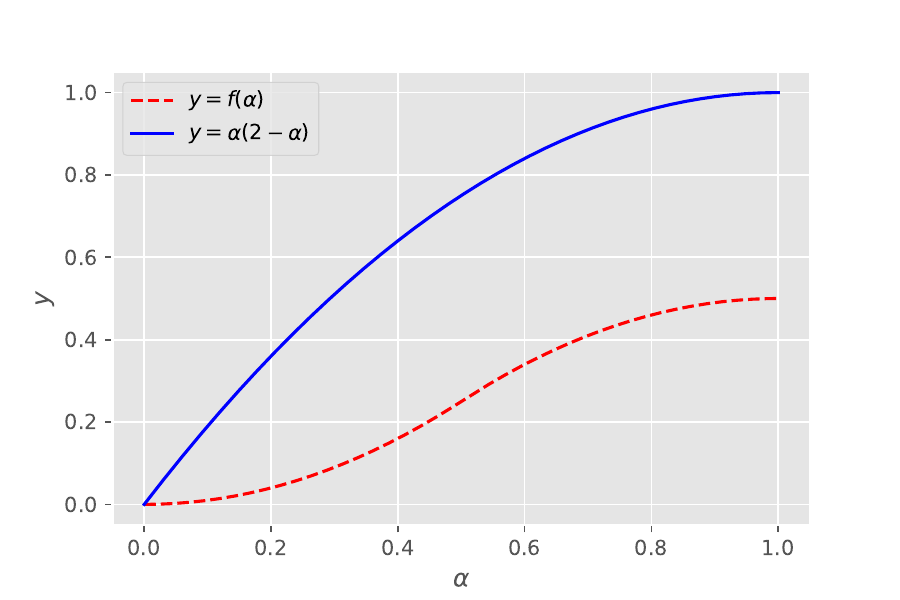}
	\caption{Plot on $[0,1]$ of normalized variance $\alpha (2-\alpha)$, together with the lower bound on the normalized covariance (function $f$) defined by \eqref{eq:f_alpha}.}
	\label{fig:GA_IT:plot_f} 
\end{figure}

Then, since $f(\alpha) \leq \alpha(2-\alpha)$ with elementary computations, according to Lemma \ref{lemma:min_corr_delta}, there is an event $\cB$ of probability $1-o(1)$  such that $$\max_{\sigma \in \cS_{n,d}} X_\sigma \leq (1+o(1)) \sqrt{2 \alpha \left(\alpha(2-\alpha)-f(\alpha)\right)} n^{3/2} \log^{1/2} n$$ holds for all $d=\alpha n$ with $\alpha>\alpha_0$. Note that on event $\cA \cap \cB$, for all $d  = \alpha n$ and $\sigma \in \cS_{n,d}$,
\begin{flalign*}
\rho^{-1}  \delta(\sigma) & \geq \rho v_{\sigma} - 2 \sqrt{1-\rho^2} \max_{\sigma \in \cS_{n,d}} X_\sigma \\
& \geq (1+o(1)) n^{3/2} \log^{1/2} n \left[(2+\eps) \alpha(2-\alpha) - 2 \sqrt{2 \alpha \left(\alpha(2-\alpha)-f(\alpha)\right)}\right]\\
& \geq (1+o(1)) \times 2 \times \left[\alpha(2-\alpha) - \sqrt{2 \alpha \left(\alpha(2-\alpha)-f(\alpha)\right)}\right] n^{3/2} \log^{1/2} n  \geq 0,
\end{flalign*} for $n$ large enough, since it can be easily checked (see Appendix \ref{appendix:final_function}) that 

\begin{lemma}\label{lemma:final_function_study}
	For every $\alpha \in [0,1]$,
	\begin{equation} \label{eq:final_function_study}
	\alpha(2-\alpha) - \sqrt{2 \alpha \left(\alpha(2-\alpha)-f(\alpha)\right)} \geq 0.
	\end{equation}
\end{lemma}

Previous computations hence give that $\dP\left(\bigcup_{d \geq \alpha n} \cE_d\right) \leq 1 - \dP(\cA \cap \cB) =o(1)$, and ends the proof of Theorem \ref{GA_IT:theorem:GA_IT_positive}.

\section{Converse bound: second moment method for transpositions}\label{section:conv}
In this section, we prove Theorem \ref{GA_IT:theorem:GA_IT_negative}. As claimed in the introduction, theory from Bayesian optimal estimation guarantees that the best possible estimator for our exact reconstruction problem, in the Bayes risk sense, is $\hat{\pi}_{\MAP}$. We will show that under assumption \eqref{eq:cond_imp_glo} of Theorem \ref{GA_IT:theorem:GA_IT_negative}, this MAP estimator fails with high probability, which implies that no estimator can succeed.

This converse bound is obtained by a second moment argument, showing that with high probability, there are  lots of permutation $\tau \neq \id$ -- in fact, transpositions -- such that $\delta(\tau)$ is negative, that is, $\tau^{-1} \circ \pi^\star$ is a substantially better alignment than $\pi^\star$, with lowest energy. Let us denote $\cT_n \subset \cS_n$ the set of all permutations of $[n]$ that are transpositions. 
For all $\tau \in \cT_n$, we have $d^{\mathrm{E}}_\tau = 2(n-2)$. Corollary \ref{corollary:control_C} gives that the event
\begin{equation*}\label{eq:event_D}
\cC := \left\lbrace \forall \tau, \tau' \in \cT_n, 
\big|c_{\tau,\tau'} -  \card{\cD^{\mathrm{E}}_\tau \cap \cD^{\mathrm{E}}_{\tau'}} -  \card{\cD^{\mathrm{E}}_\tau \cap \cD^{\mathrm{E}}_{\tau'} \cap \cF^{\mathrm{E}}_{\tau \circ \tau'}} \big|  \leq   C\sqrt{n \log n} \right\rbrace 
\end{equation*} happens with probability $1-o(1)$ for $C>0$ large enough. In particular, on $\cC$, for $C>0$ large enough,
\begin{equation*}\label{eq:control_v_tau}
\forall \tau \in \cT_n, \; \left|v_\tau - 4n\right| \leq C \sqrt{n \log n}.
\end{equation*}

In this section we are working under the assumption \eqref{eq:cond_imp_glo} that we recall here:
\begin{equation*}
\rho^2 \leq \frac{4 \log n - \log \log n - \omega(1)}{n}
\end{equation*} 
We are about to show the following: under condition \eqref{eq:cond_imp_glo}, with high probability,
\begin{equation}\label{eq_converse_transp}
\card{\set{\tau \in \cT_n, \, \delta(\tau) < 0}}  =\omega(1).
\end{equation} To do so, we use the classical Paley-Zygmund inequality (Lemma \ref{intro:lemma:first_moment_method} of Section \ref{intro:subsection:basics_rg}) that implies (taking $c\to0$ in Lemma \ref{intro:lemma:first_moment_method}) that if $Y$ is a positve random variable such that $\dE\left[Y^2\right] \sim \dE\left[Y\right]^2$, then $Y\geq o(\dE\left[Y\right])$ with high probability. Define
\begin{equation}\label{eq_def_X}
X := \sum_{\tau \in \cT_n} \mathbf{1}_{\delta(\tau)<0}.
\end{equation} Using a standard coupling argument in \eqref{eq_def_X}, one can see that $X$ is decreasing with $\rho$, thus we can assume without loss of generality that
\begin{equation}
\label{eq:sec_cond_conv}
\rho^2 = \frac{4 \log n - \log \log n - a_n}{n},
\end{equation} with a sequence $(a_n)_n$ such that $a_n = \omega(1)$ and $a_n = o(\log \log n)$, e.g. $a_n = \log \log \log n$. We compute the first moment of $X$, in view of the conditional distribution of $\delta(\tau)$ given in \eqref{eq:delta_gaussien_general}:
\begin{flalign*}
\dE\left[X\right] & \geq  \dE\left[X \mathbf{1}_\cC \right] = \frac{n(n-1)}{2} \dE\left[\dP_A \left(\cN(0,1) \geq \frac{\rho \sqrt{v_\tau}}{2 \sqrt{1-\rho^2}} \cap  \cC \right) \right]\\
& \geq \frac{n(n-1)}{2} \dE\left[(1-o(1)) \dP_A \left(\cN(0,1) \geq \frac{1}{2} \sqrt{4 \log n - \log \log n - a_n} \sqrt{4 - C n^{-1/2} \log^{1/2}n}  \right)  \right]\\
& = \frac{n(n-1)}{2} \dE\left[(1-o(1)) \dP_A \left(\cN(0,1) \geq  \sqrt{4 \log n - \log \log n - a_n} - o(1) \right)  \right]\\
& \sim \frac{n^2}{4 \sqrt{2 \pi} \sqrt{ \log n }}\exp \left(- 2\log n + \frac{\log \log n }{2} + \frac{a_n }{2}\right)\\
& = \frac{1}{4 \sqrt{2 \pi}} \exp \left( \frac{a_n}{2}  \right) \to \infty.
\end{flalign*} Note that \eqref{eq:sec_cond_conv} is thus precisely the condition ensuring that $\dE\left[X \mathbf{1}_\cC \right] \to \infty$. The second moment argument computation being a little more technical, we encapsulate it into the following Lemma:
\begin{lemma}[Second moment computation of $X \mathbf{1}_\cC$]\label{lemma_second_moment}
Let $Y:= X \mathbf{1}_\cC$. Under assumption \eqref{eq:sec_cond_conv},
\begin{equation*}
\dE\left[Y^2\right]
\leq (1+o(1)) \dE\left[Y\right]^2.
\end{equation*}
\end{lemma}

\begin{proof}[Proof of Lemma \ref{lemma_second_moment}]

We represent a transposition $\tau$ by its only $2-$cycle $(i \; j)$ with $i<j$. We then distinguish two cases in couples $\tau = (i \; j) \neq \tau' = (k \; \ell) \in \cT_n$:
\begin{itemize}
\item We write $\tau \cap \tau' = \varnothing$ when $\tau$ and $\tau'$ have no common point in their $2-$cycle: $i \neq k$ and $j \neq l$. When $\tau \in \cT_n$ is fixed, note that
\begin{equation*}\label{count_1}
 \card{\set{\tau' \in \cT_n, \,  \tau \cap \tau' = \varnothing}} = \frac{(n-2)(n-3)}{2}.
\end{equation*}

\item We write $\tau \cap \tau' \neq \varnothing$ when $\tau$ and $\tau'$ are different but share one common point: for instance $\tau = (3 \; 5) $ and $\tau = (5 \; 11)$ verify $\tau \cap \tau' \neq \varnothing$. When $\tau \in \cT_n$ is fixed, note that
\begin{equation*}\label{count_2}
\card{\set{\tau' \in \cT_n, \,  \tau \cap \tau' \neq \varnothing}} = 2 (n-2).
\end{equation*}
\end{itemize} 

Note that
\begin{flalign*}
\dE\left[Y^2\right] = \dE\left[Y\right] + \sum_{\tau \in \cT_n} \sum_{\tau', \tau \cap \tau' = \varnothing} \dP(\delta(\tau)<0, \delta(\tau')<0, \cC) + \sum_{\tau \in \cT_n} \sum_{\tau', \tau \cap \tau' \neq \varnothing} \dP(\delta(\tau)<0, \delta(\tau')<0, \cC).
	\end{flalign*} We now evaluate these two sums. For this, we will need the following Lemma, which proof is deferred to Appendix \ref{appendix:corr_gaussians}.

\begin{lemma}[Control of deviation probabilities for correlated Gaussians]\label{lemma:control_dev_cor_gaussian}
Let $Z_1, Z_2$ be two Gaussian variables with mean $0$, variance $1$ and correlation $\alpha_n \in [0,1]$. For any $t_n$ such that $t_n \to \infty$, 
\begin{itemize}
	\item[$(i)$] If $\alpha_n t_n \to 0$, then for $n$ large enough
	\begin{equation}\label{eq:control_dev_gaussian_1}
	\dP \left(Z_1 > t_n, Z_2 > t_n \right) \leq e^{-2t_n^2} + (1+o(1))\dP \left(Z_1 > t_n\right) \dP \left(Z_2 > t_n\right).
	\end{equation} 
	\item[$(ii)$] More generally,
	\begin{equation}\label{eq:control_dev_gaussian_2}
	\dP \left(Z_1 > t_n, Z_2 > t_n \right) \leq (1+o(1))\frac{1+\alpha_n}{\sqrt{2\pi} \, t_n}\exp\left(- \frac{t_n^2}{1+ \alpha_n} \right).
	\end{equation} 
\end{itemize}
\end{lemma}

\proofstep{First case: $\tau \cap \tau' = \varnothing$.} Without loss of generality we can assume that $\tau = (1 \; 2)$ and $\tau' = (3 \; 4)$. The following diagram shows the simple action of $\tau$ and $\tau'$ on an interesting (overlapping) subset of edges.
\begin{center}
	\begin{tabular}{ c c c }
		$\left\lbrace 1,3\right\rbrace $ & $\overset{\tau}{\longleftrightarrow} $& $\left\lbrace 2,3\right\rbrace $\\ 
		{\scriptsize $\tau'$} $\updownarrow$&  & $\updownarrow$ {\scriptsize $\tau'$}  \\  
		$\left\lbrace 1,4\right\rbrace $ & $\overset{\tau}{\longleftrightarrow}$ & $\left\lbrace 2,4\right\rbrace $   
	\end{tabular}
\end{center} 

We then see that $ \card{\cD^{\mathrm{E}}_\tau \cap \cD^{\mathrm{E}}_{\tau'}} +  \card{\cD^{\mathrm{E}}_\tau \cap \cD^{\mathrm{E}}_{\tau'} \cap \cF^{\mathrm{E}}_{\tau \circ \tau'}} = 4 + 0 = 4$. So, denoting $\alpha_{\tau,\tau'} := \frac{c_{\tau,\tau'}}{\sqrt{v_\tau v_{\tau'}}}$, on $\cC$,
\begin{equation*}\label{eq:case_1_alpha}
\left|\alpha_{\tau,\tau'}\right| \leq \frac{C\sqrt{n \log n}+4}{4n-C\sqrt{n \log n}} = O \left(\sqrt{\frac{\log n}{n}}\right).
\end{equation*}
In view of the conditional distribution of $\delta(\tau)$ given in \eqref{eq:delta_gaussien_general}:

\begin{equation}\label{eq:sum_case1}
\sum_{\tau \in \cT_n} \sum_{\tau', \tau \cap \tau' = \varnothing} \dP(\delta(\tau)<0, \delta(\tau')<0, \cC) = (1-o(1)) \sum_{\tau \in \cT_n} \sum_{\tau', \tau \cap \tau' = \varnothing} \dP \left( Z_\tau > t_n, \, Z_{\tau'} >t_n\right),
\end{equation} with $t_n = \sqrt{4\log n - \log \log n - a_n}$, where $Z_{\tau}, Z_{\tau'}$ are two Gaussian variables of mean $0$, with correlation coefficient $\alpha_n$ of order $O(\log^{1/2}n^{-1/2})$. Since $\alpha_n t_n \to 1$, by lemma \ref{lemma:control_dev_cor_gaussian} case $(i)$, the sum in \eqref{eq:sum_case1} is upper bounded by
\begin{flalign*}
 & (1-o(1))\frac{n(n-1)}{2} \times \frac{(n-2)(n-3)}{2} \times \left[C e^{-2t_n^2} + (1-o(1))\dP \left(Z_1 > t_n\right) \dP \left(Z_2 > t_n\right)\right] \\
 & \leq (1+o(1))\dE\left[Y\right]^2.
\end{flalign*}

\proofstep{Second case: $\tau \cap \tau' \neq \varnothing$.} Without loss of generality we can assume that $\tau = (1 \; 2)$ and $\tau' = (2 \; 3)$. We can immediately deduce that $ \card{\cD^{\mathrm{E}}_\tau \cap \cD^{\mathrm{E}}_{\tau'}} +  \card{\cD^{\mathrm{E}}_\tau \cap \cD^{\mathrm{E}}_{\tau'} \cap \cF^{\mathrm{E}}_{\tau \circ \tau'}} = (n-2) + 0 = n-2$. So, denoting $\alpha_{\tau,\tau'} := \frac{c_{\tau,\tau'}}{\sqrt{v_\tau v_{\tau'}}}$, on $\cC$,
\begin{equation*}\label{eq:case_2_alpha}
\left|\alpha_{\tau,\tau'}\right| \leq \frac{C\sqrt{n \log n}+n-2}{4n-C\sqrt{n \log n}} \sim \frac{1}{4}.
\end{equation*}
Again, in view of the conditional distribution of $\delta(\tau)$ given in \eqref{eq:delta_gaussien_general}:

\begin{equation}\label{eq:sum_case2}
\sum_{\tau \in \cT_n} \sum_{\tau', \tau \cap \tau' \neq \varnothing} \dP(\delta(\tau)<0, \delta(\tau')<0, \cC) = (1-o(1)) \sum_{\tau \in \cT_n} \sum_{\tau', \tau \cap \tau' \neq \varnothing} \dP \left( Z_\tau > t_n, \, Z_{\tau'} >t_n\right),
\end{equation} with $t_n = \sqrt{4\log n - \log \log n - a_n}$, where $Z_{\tau}, Z_{\tau'}$ are two Gaussian variables of mean $0$, with correlation coefficient $\alpha_n\sim 1/4$. By Lemma \ref{lemma:control_dev_cor_gaussian} case $(ii)$, the sum in \eqref{eq:sum_case2} is upper bounded by
\begin{flalign*}
(1-o(1))\frac{n(n-1)}{2} \times 2(n-2) \times \left[ (1+o(1))\frac{1+\alpha_n}{\sqrt{2\pi} \, t_n}\exp\left(- \frac{t_n^2}{1+ \alpha_n} \right)\right]\\
\leq C'' n^3 \log^{-1/2}(n)  \exp\left(- \frac{16}{5} \log n + o(\log n) \right)=o(1)=o(\dE\left[Y\right]^2).
\end{flalign*}
\end{proof}	
	
Lemma \ref{lemma_second_moment} together with Payley-Zigmund inequality (Lemma \ref{intro:lemma:first_moment_method} of Section \ref{intro:subsection:basics_rg}) implies that $Y \geq o\left(\mathbb{E}[Y]\right)$ with high probability and thus proves \eqref{eq_converse_transp} and the converse result of Theorem \ref{GA_IT:theorem:GA_IT_negative}.

\begin{remark}
We have shown here that under condition \eqref{eq:cond_imp_glo}, there is with high probability a great number of negative relative energy points near the ground truth, none of them being of significant interest to recover \emph{exactly} our permutation. 
We may also study this relative energy far from the planted permutation, which would be interesting to address the problem of almost exact (resp. partial) alignment, which consists in finding an estimator $\hat{\pi}$ that coincides with $\pi$ on at least $n - o(n)$ (resp. some positive fraction of $n$) points. In the light of our result which shows that exact recovery is not more difficult than detection, we can also conjecture that the same threshold $n \rho^2 / \log n = 4$ is sharp for the tasks of almost exact and partial recovery.
\end{remark}

\newpage

\begin{subappendices}
\addtocontents{toc}{\protect\setcounter{tocdepth}{0}}
\section{Additional proofs}\label{appendix}

\subsection{Proof of Lemma \ref{lemma:min_corr_delta}: lower bound on correlations of relative energies}\label{appendix:lower_bound_corr}
\begin{proof}
	Recall that we work under event $\cA$. Fix $\alpha \in (0,1]$ and take $d=\alpha n$ and $\sigma,\sigma' \in \cS_{n,d}$. The proof is obtained by establishing a fine lower bound on $ \card{\cD^{\mathrm{E}}_\sigma \cap \cD^{\mathrm{E}}_{\sigma'}}$, which is simply the number of edges that are deranged both by $\sigma^E$ and $\sigma'^E$. In order to establish this lower bound, let us assume that $\sigma$ and $\sigma'$ have $ \card{\cD_\sigma \cap \cD_{\sigma'}} = \beta n$ common unfixed points, with $\beta \in [0,\alpha]$. We then form edges in $\cD^{\mathrm{E}}_\sigma \cap \cD^{\mathrm{E}}_{\sigma'}$ in the following way:
	\begin{itemize}
		\item First, by taking all pairs but the pairs made of points in the complement of $\cD_\sigma \cap \cD_{\sigma'}$ and those made of pairs $(i,j)$ that are transpositions of $\sigma$ or $\sigma'$, we obtain at least $\frac{1}{2}\beta(2-\beta) n^2 - \alpha n$ edges.
		\item Then, add new edges made of one extremity in $\cD_{\sigma} \setminus \cD_{\sigma'}$ and one in $\cD_{\sigma'} \setminus \cD_{\sigma}$. Since $\cD_{\sigma}$ (resp $\cD_{\sigma}$) is stable by $\sigma$ (resp. by $\sigma'$), all these $(\alpha-\beta)^2 n^2$ edges are in $\cD^{\mathrm{E}}_\sigma \cap \cD^{\mathrm{E}}_{\sigma'}$.
	\end{itemize} Finally we formed $g(\alpha,\beta)n^2 - \alpha n$ edges, with 
	\begin{equation}
	g(\alpha,\beta) := \frac{1}{2} \beta^2 + (1-2\alpha) \beta + \alpha^2,
	\end{equation} which is minimal on $[0,\alpha]$ at $\beta = 2\alpha-1$ if $\alpha \geq 1/2$, or at $\beta=0$ if $\alpha<1/2$. In any case, this minimum is $f(\alpha)$. The first inequality is established by applying inequality \eqref{eq:event_A} of event $\cA$.\\
	
	For the second part, consider a centered vector $Z= (Z_{\sigma})_{\sigma \in \cS_{n,\alpha n}}$ such that all $Z_{\sigma}$ have same variance $v_\alpha$ and $\Cov(Z_\sigma,Z_{\sigma'}) = c_\alpha $ for $\sigma \neq \sigma'$, with $v_\alpha, c_\alpha $ defined as follows:
	\begin{flalign*}
	v_\alpha &:= \alpha (2-\alpha)n^2 - C_1 n^{3/2}\log^{1/2}n,\\
	c_\alpha &:=  f(\alpha)n^2 - C_1 n^{3/2}\log^{1/2}n.
	\end{flalign*}for some $C_1>0$ large enough. Note that on event $\cA$, for all $\alpha \in (0,1]$, all $\sigma,\sigma' \in \cS_{n,\alpha n}$, $$\Cov(Z_\sigma,Z_{\sigma'}) \leq \Cov(X_\sigma,X_{\sigma'}),$$ so one has that for all $t>0$,
	\begin{equation}\label{eq:comparison_max_cov}
	\dP\left(\max_{\sigma \in \cS_{n,\alpha n}} X_\sigma > t \; \cap \cA\right) \leq \dP\left(\max_{\sigma \in \cS_{n,\alpha n}} Z_\sigma > t \right).
	\end{equation} We now control the right-hand side of \eqref{eq:comparison_max_cov} with this classical Lemma, which proof is find hereafter in Appendix \ref{appendix:max_corr_Gaussians}:
	\begin{lemma}[Maximum of totally correlated Gaussian variables]\label{lemma:max_TC_gaussian}
		Let $Z$ be a centered Gaussian vector of size $N$, such that all $Z_i$ have same variance $v$ and $\Cov(Z_i,Z_j) = c $ for $i \neq j$. Then
		\begin{equation}\label{eq:max_TC_gaussian}
		\dP\left(\max_{1 \leq i \leq N} Z_i > \sqrt{2(v-c) \log N} + 2\sqrt{v\log \log N} \right) \leq \frac{2}{\log N}.
		\end{equation}
	\end{lemma}
	Note that for $v_\alpha, c_\alpha$ previously defined, one has
	\begin{equation}\label{eq:ineq_Z_1}
	\sqrt{2(v_\alpha-c_\alpha)\log  \card{\cS_{n,\alpha n}}} \leq \sqrt{2 \alpha (\alpha(2-\alpha)-f(\alpha))} n^{3/2}\log^{1/2} n,
	\end{equation} and for $n$ large enough,
	\begin{equation}\label{eq:ineq_Z_2}
	2\sqrt{v_\alpha \log\log  \card{\cS_{n,\alpha n}}} \leq 2 \sqrt{ \alpha(2-\alpha)} n \sqrt{\log n + \log \log n} \leq (2+\eps') n \log^{1/2} n.
	\end{equation} Finally, we use equations \eqref{eq:comparison_max_cov}--\eqref{eq:ineq_Z_2} to conclude that for $n$ large enough:
	\begin{flalign*}
	&\dP\left(\exists d = \alpha n , \alpha>\alpha_0, \; \max_{\sigma \in \cS_{n,d}} X_\sigma > \sqrt{2 \alpha \left(\alpha(2-\alpha)-f(\alpha)\right)} n^{3/2}\log^{1/2} n + (2+\eps') n \log^{1/2} n\right)\\
	&\leq 1-\dP(\cA) + \sum_{d=\alpha n, \, \alpha>\alpha_0} \dP\left(\max_{\sigma \in \cS_{n,\alpha n}} Z_i > \sqrt{2(v_\alpha-c_\alpha)\log  \card{\cS_{n,\alpha n}}} + 2\sqrt{v_\alpha \log\log  \card{\cS_{n,\alpha n}}} \right)\\
	&\leq o(1) + \sum_{d=\alpha n, \, \alpha>\alpha_0} \frac{2}{\log  \card{\cS_{n,\alpha n}}} \leq o(1) + \frac{2n}{\log  \card{ \cS_{n,\alpha_0 n}}} =  o(1) + \frac{2}{\alpha_0 \log n} = o(1),
	\end{flalign*} and Lemma \ref{lemma:min_corr_delta} is proved.
\end{proof}

\subsection{Proof of Lemma \ref{lemma:max_TC_gaussian}: maximum of totally correlated Gaussian variables}\label{appendix:max_corr_Gaussians}

\begin{proof}
	Let us make a change of variables which preserves the joint distribution:
	\begin{equation*}
	\left(Z_1, Z_2, \ldots, Z_N\right) = \left(\sqrt{c} \,\xi_0+\sqrt{v-c}\, \xi_1, \ldots,\sqrt{c}\, \xi_0+\sqrt{v-c}\, \xi_N\right),
	\end{equation*} where $\xi_0, \ldots, \xi_N$ are independent standard Gaussian random variables. The maximum thus writes
	\begin{equation*}
	\max_{1 \leq i \leq N} Z_i = \sqrt{c} \, \xi_0 + \sqrt{v-c} \max_{1 \leq i \leq N} \xi_i
	\end{equation*}
	Then, with the classical inequality $\dP \left(\cN(0,1) \geq t\right) \leq e^{-t^2/2}$, then with probability at least $1-1/(\log N)$, one has:
	\begin{equation*}\label{eq:control_gaussian}
	\sqrt{c} \, \xi_0 \leq \sqrt{2c \log \log N}, \quad \mbox{and} \quad \sqrt{v-c} \max_{1 \leq i \leq N} \xi_i \leq \sqrt{2(v-c) \log N \left(1+\frac{\log \log N}{\log N}\right)},
	\end{equation*}so with probability at least $1-2/(\log N)$:
	\begin{flalign*}
	\max_{1 \leq i \leq N} Z_i & \leq \sqrt{2(v-c) \log N} + \sqrt{2\log \log N} \left(\sqrt{c}+ \sqrt{v-c}\right) \\
	& \leq  \sqrt{2(v-c) \log N} + 2 \sqrt{v \log \log N},
	\end{flalign*} where we used $\sqrt{c}+\sqrt{v-c} \leq \sqrt{2v}$ in the last step.
\end{proof}

\subsection{Proof of Lemma \ref{lemma:final_function_study}}\label{appendix:final_function}
\begin{proof}
	For $\alpha \in (0,1]$,
	\begin{flalign*}
	\eqref{eq:final_function_study} & \iff \alpha^2(2-\alpha)^2 \geq 2 \alpha \left(\alpha(2-\alpha)-f(\alpha)\right)\\
	& \iff f(\alpha) \geq \alpha^2-\alpha^3/2. 
	\end{flalign*} The inequality is verified for $\alpha < 1/2$. To conclude the proof of \eqref{eq:final_function_study}, it remains to check that for $1\geq \alpha \geq 1/2$, $f(\alpha) \geq \alpha^2-\alpha^3/2$, which is equivalent to 
	\begin{flalign*}
	\alpha^2 - \frac{1}{2}(2\alpha -1)^2 \geq \alpha^2-\alpha^3/2 & \iff \alpha^3 - 4\alpha^2 + 4\alpha -1 \geq 0\\
	& \iff (\alpha-1)(\alpha^2-3\alpha+1) \geq 0\\
	& \iff  \alpha^2-3\alpha+1 \leq 0 \iff \alpha \geq \frac{3 - \sqrt{5}}{2} \sim 0.382...
	\end{flalign*} 
\end{proof}

\subsection{Proof of Lemma \ref{lemma:control_dev_cor_gaussian}: control of deviation probabilities for correlated Gaussians}\label{appendix:corr_gaussians}

\begin{proof} Let us first make a change of variable which preserves the joint distribution:
	\begin{flalign*}
	(Z_1,Z_2) = (Z , \alpha_n Z + \sqrt{1-\alpha_n^2} Z'),
	\end{flalign*} with $Z,Z'$ two independent standard Gaussian variables. 
		
	\proofstep{Proof of $(i)$.} Note that standard Gaussian concentration gives $\dP\left(Z>2t_n \big| Z>t_n \right) \sim \frac{1}{2}e^{-3t_n^2/2}$. Thus, for $n$ large enough
	\begin{flalign*}
	\dP \left(Z_1 > t_n, Z_2 > t_n \right) & \leq \dP \left(Z > t_n \right) e^{-3t_n^2/2} + \dP \left(Z > t_n \right)\dP \left(\alpha_n Z + \sqrt{1-\alpha_n^2} Z' > t_n, Z\leq 2t_n \big| Z>t_n \right) \\
	& \leq e^{-2t_n^2} + \dP \left(Z > t_n \right)\dP \left(Z' > t_n -2\alpha_n t_n + O(t_n \alpha_n^2)\right)\\
	& \leq e^{-2t_n^2} + \dP \left(Z > t_n \right)\dP \left(Z' > t_n -o(1)\right)\\
	& \leq e^{-2t_n^2} +  (1+o(1))\dP \left(Z > t_n\right) \dP \left(Z' > t_n\right)\\
	& = e^{-2t_n^2} +  (1+o(1))\dP \left(Z_1 > t_n\right) \dP \left(Z_2 > t_n\right).
	\end{flalign*}
	
	\proofstep{Proof of $(ii)$.} For any $(s_n)$ such that $s_n \leq  t_n$ for all $n$, one has
	\begin{flalign*}
	\dE\left[e^{s_n Z} \, \big| \, Z > t_n \right] & = \frac{1}{\sqrt{2 \pi}} \int_{t_n}^{+ \infty} e^{s_n z-z^2/2} dz \left({\frac{1}{\sqrt{2 \pi}} \int_{t_n}^{+ \infty} e^{-z^2/2} dz} \right)^{-1}\\
	& = {e^{s_n^2/2} \int_{t_n - s_n}^{+ \infty} e^{-z^2/2} dz} \left({\int_{t_n}^{+ \infty} e^{-z^2/2} dz}\right)\\
	& \sim \frac{t_n}{t_n - s_n} \exp \left(s_n^2/2 - (t_n-s_n)^2/2 + t_n^2/2\right) = \frac{t_n}{t_n - s_n} e^{s_n t_n}.
	\end{flalign*} Using independence of $Z,Z'$ and Chernoff bound, we get, taking $s_n$ such that $\alpha s_n = u t_n$ with $u<1$, for $n$ large enough,
	\begin{flalign*}
	\dP\left(\alpha Z + \sqrt{1-\alpha^2} Z' > t_n	\big | Z >t_n \right) & 
	\leq (1+o(1))\frac{t_n}{t_n-\alpha s_n} \exp\left( \alpha s_n t_n + \frac{1-\alpha^2}{2} s_n^2 - s_n t_n \right)\\
	& \leq (1+o(1))\frac{1}{1-u} \exp\left(\left(u   + \frac{u^2 (1-\alpha^2)}{2\alpha^2} - \frac{u}{\alpha}\right) t_n^2 \right)\\
	& \overset{(a)}{\leq} (1+o(1)) (1+\alpha) \exp\left(- \frac{1-\alpha}{1+\alpha} \cdot \frac{t_n^2}{2} \right)
	\end{flalign*} where we took $u = \frac{\alpha}{1+\alpha} <1$ in $(a)$.  The proof follows from this last inequality, together with the bound $\dP\left( Z >t_n \right) \leq \frac{1}{\sqrt{2\pi}t_n} \exp\left(-\frac{t_n^2}{2}\right).$
\end{proof}
\addtocontents{toc}{\protect\setcounter{tocdepth}{2}}

\end{subappendices}

\chapter{Alignment of graph databases with Gaussian weights: analysis of a spectral method}\label{chapter:EIG1}
In this chapter, we analyze a simple spectral method (\alg{EIG1}) for the problem of matrix alignment, consisting in aligning their leading eigenvectors: given two adjacency matrices $A$ and $B$, \alg{EIG1} aligns $v_1$ and $v'_1$, their two corresponding leading eigenvectors (up to the sign of $v'_1$).
	
We will consider the Gaussian model $\Wig(n,\xi)$ defined earlier in \eqref{eq:GOE_model}: $A$ belongs to the Gaussian Orthogonal Ensemble ($\GOE$) of size $n \times n$, and $B$ is a noisy version of $A$ where all nodes have been relabeled according to some planted permutation $\pi^\star$. We show the following zero-one law: with high probability, under the condition $\xi n^{7/6+\epsilon} \to 0$ for some $\eps>0$, \alg{EIG1} recovers all but a vanishing part of the underlying permutation $\pi^\star$, whereas if $\xi n^{7/6-\epsilon} \to \infty$, this method cannot recover more than $o(n)$ correct matches.
	
This result gives an understanding of the simplest and fastest spectral method for matrix alignment (or complete weighted graph alignment), and involves proof methods and techniques which could be of independent interest.\\

This chapter is based on the paper \textit{Spectral alignment of correlated gaussian matrices} \cite{GLM19}, published in \emph{Advances in Applied Probability}, a joint work with M. Lelarge and L. Massoulié.

\section{Introduction}

\subsection{The \alg{EIG1} algorithm}

As in Chapter \ref{chapter:gaussian_alignment_IT}, we are interested in alignment of Gaussian databases, which is one of the instances of the graph alignment problem. For a general overview, we refer here again to the introduction of this manuscript, to Section \ref{intro:subsection:motivations} for applications and to Section \ref{intro:subsection:short_survey} for theoretical results.

\subsubsection{Related work: spectral methods for graph alignment} Some general spectral methods for random graph alignment are introduced in \cite{Feizi16}, based on representation matrices and low-rank approximations. These methods are tested over synthetic graphs and real data; however no precise theoretical guarantee -- e.g. an error control of the inferred mapping depending on the signal-to-noise ratio -- can be found for such techniques. 

Most recently, a spectral method for matrix and graph alignment (\alg{GRAMPA}) was proposed in \cite{Fan2019Wigner,fan2019ERC} and computes a similarity matrix which takes into account all pairs of eigenvalues $(\lambda_i, \mu_j)$ and eigenvectors $(u_i,v_j)$ of matrices $A$ and $B$. The authors study the regime in which the method exactly recovers the underlying vertex correspondence: this method can tolerate a noise $\xi$ up to $O\left(1/\log n\right)$ to recover the entire underlying vertex correspondence. Since the computations of all eigenvectors is required, the time complexity of \alg{GRAMPA} is at least $O(n^3)$. 

It is important to note that the signs of eigenvectors are ambiguous: in practice, it is necessary to test over all possible signs of eigenvectors. This additional complexity has no consequence when reducing $A$ and $B$ to rank-one matrices, but becomes costly when the reduction made is of rank $k \gg 1$. This combinatorial observation makes implementation and analysis of general rank-reduction methods (as the ones proposed in \cite{Feizi16}) more difficult. We therefore focus on the analysis of the rank-one reduction (\alg{EIG1} hereafter) which is the simplest and most natural spectral alignment method, where only the leading eigenvectors of $A$ and $B$ are computed, with time complexity $O(n^2)$, which is significantly less than \alg{GRAMPA}.

\subsubsection{Model and method} Let us recall the model $\Wig(n,\xi)$ defined in \eqref{eq:GOE_model2}. In this model, $A$ is a matrix from the normalized Gaussian Orthogonal Ensemble ($\GOE$), i.e. for all $1 \leq u \leq v \leq n$, 
\begin{equation}
\label{eq:pre_GOEmodel}
A_{u,v} = A_{v,u} \sim \begin{cases}
\cN(0,1/n) & \text{if $u \neq v$}, \\
\cN(0,2/n) & \text{if $u = v$},
\end{cases}
\end{equation} and $H$ is an independent copy of $A$. We define \begin{equation}
\label{eq:GOEmodel_bis}
B=\Pi^{\star \top} \left(A+\xi H\right) \Pi^\star
\end{equation} where $\Pi^\star$ is a random uniform matrix of a permutation $\pi^\star$ -- e.g. random uniform -- of $[n]$ and $\xi = \xi(n)$ is the \textit{noise parameter}.

Given two vectors $x=\left(x_1,\ldots,x_n\right)$ and $y=\left(y_1,\ldots,y_n\right)$ having all distinct coordinates, the permutation $\rho$ which \textit{aligns} $x$ and $y$ is the permutation such that for all $1 \leq i \leq n$, the rank (for the usual order) of $y_{\rho(i)}$ in $y$ is the rank of $x_{i}$ in $x$.

\begin{remark}
\label{densitylebesgue}
Note that in our model, all the probability distributions are absolutely continuous with respect to Lebesgue measure, thus the eigenvectors of $A$ and $B$ all have almost surely pairwise distinct coordinates.
\end{remark}

We recall that the aim is to infer the underlying permutation $\Pi^{\star}$ given the observation of $A$ and $B$. We now introduce our simple spectral algorithm derived from \cite{Feizi16}, which we call \alg{EIG1}, that consists in computing and aligning the leading eigenvectors $v_1$ and $v'_1$ of $A$ and $B$. This very natural method can be thought of as the relaxation of the QAP formulation \eqref{eq:QAP} when reducing $A$ and $B$ to rank-one matrices $\lambda_1 v_1 v_1^{\top}$ and $\lambda'_1 v'_1 v_1^{'T}$. Indeed, as soon as $v_1$ and $v'_1$ have pairwise distinct coordinates, is it easy to see that
\begin{equation*}
\argmax_{\Pi \in \cS_n} \langle \lambda_1 v_1 v_1^{\top}, \Pi \lambda'_1 v'_1 v_1^{'T} \Pi^{\top} \rangle = \argmax_{\Pi \in \cS_n} \pm v_1^{\top} \Pi v'_1  = \rho,
\end{equation*} where $\rho$ is the aligning permutation of $v_1$ and $\pm v'_1$. Computing the two normalized leading eigenvectors (i.e. corresponding to the highest eigenvalues) $v_1$ and $v'_1$ of $A$ and $B$, the \alg{EIG1} algorithm returns the aligning permutation of $v_1$ and $\pm  v'_1$. The method then decides which permutation to output according to the scores.

\begin{algorithm}[h]
	\caption{\alg{EIG1} Algorithm for matrix alignment}
	\label{algo:EIG1}
	\SetAlgoLined
	Compute $v_1$ a normalized leading eigenvector of $A$\;
	Compute $v'_1$ a normalized leading eigenvector of $B$\;
	Compute $\Pi_{+}$ the permutation aligning $v_1$ and $v'_1$\;
	Compute $\Pi_{-}$ the permutation aligning $v_1$ and $-v'_1$\;
	\eIf{$\langle A, \Pi_{+} B \Pi_{+}^{\top} \rangle \geq \langle A, \Pi_{-} B \Pi_{-}^{\top} \rangle$}{
		return $\Pi_{+}$}
	{
		return $\Pi_{-}$
	}
\end{algorithm}

The aim of this chapter is to find the regime in which \alg{EIG1} achieves almost exact recovery, i.e. recovers all but a vanishing fraction of nodes of the planted ground truth $\Pi^{\star}$.

\subsection{Main results and proof scheme} \label{EIG1:subsection:notations}

We start by introducing specific notations and recall some useful basic definitions. Throughout the chapter, all limits are taken when $n \to \infty$, and the dependency in $n$ will most of the time be eluded, as an abuse of notation.\\

\textit{Eigenvalues, eigenvectors.} In the following, $\left(v_1, v_2, \ldots, v_n\right)$ (resp. $\left(v'_1, v'_2, \ldots, v'_n\right)$) denote two orthonormal bases of eigenvectors of $A$ (resp. of $B$) with respect to the (real) eigenvalues $\lambda_1~\geq~\lambda_2~\geq~\ldots~\geq~\lambda_n$ of $A$ (resp. $\lambda'_1~\geq~\lambda'_2~\geq~\ldots~\geq~\lambda'_n$ of $B$). Through all the study, the sign of $v'_1$ is fixed such that $\langle \Pi^{\star} v_1,v'_1 \rangle >0$.\\

\textit{Overlap.} For any (matrix) estimator $\hat{\Pi}$ of $\Pi^{\star}$ its overlap is defined as follows
\begin{equation}
\label{overlap}
\ov(\hat{\Pi},\Pi^{\star}) := \ov(\hat{\pi},\pi^{\star}) = \frac{1}{n} \sum_{u=1}^{n} \one_{\hat{\pi}(u)=\pi^{\star}(u)} \, , 
\end{equation} where $\hat{\pi}$ (resp. $\pi^{\star}$) is the permutation corresponding to matrix $\hat{\Pi}$ (resp. $\Pi^{\star}$).

\textit{Probability.} The equality $\overset{(d)}{=}$ will refer to equality in distribution. Some event $A_n$ is said to hold \textit{with high probability} (we will use the abbreviation "w.h.p."), if $\dP(A_n)$ converges to $1$ when $n \to \infty$.

For two random variables $u=u_n$ and $v=v_n$, we will use the notation $u=o_{\mathbb{P}}\left(v\right)$ if $u_n / v_n \overset{\mathbb{P}}{\longrightarrow} 0$ when $n \to \infty$. We also use this notation when $X=X_n$ and $Y=Y_n$ are $n-$dimensional random vectors: $X=o_{\mathbb{P}}\left(Y\right)$ if $\|X_n\| / \|Y_n\| \overset{\mathbb{P}}{\longrightarrow} 0$ when $n \to \infty$.

Define 
\begin{equation}
\label{nologfunctions}
\mathcal{F} := \left\lbrace f : \dN \to \dR \; | \; \forall t>0, n^t f(n) \to \infty, n^{-t} f(n) \to 0 \right\rbrace. 
\end{equation} For two random variables $u=u(n)$ and $v=v(n)$, $u \asymp v$ refers to equivalence with high probability up to some sub-polynomial factor, meaning that there exists a function $f \in \mathcal{F}$ such that
\begin{equation}
\label{defasymp}
\mathbb{P}\left(\frac{v(n)}{f(n)} \leq u(n) \leq f(n) v(n)\right) \to 1.
\end{equation}

\subsubsection{Main results, proof scheme} 
The main result of this chapter can be stated as follows: there exists a condition -- a threshold -- on $\xi$ and $n$ under which the \alg{EIG1} method enables us to recover $\Pi^{\star}$ almost exactly, in terms of the overlap defined in \eqref{overlap}. Above this threshold, we show that \alg{EIG1} Algorithm cannot recover more than a vanishing part of $\Pi$. 

\begin{theorem}[Zero-one law for \alg{EIG1} method]
\label{EIG1:theorem:01law}
For all $n$, $\Pi_n$ denotes an arbitrary permutation of size $n$, $\hat{\Pi}_n$ is the estimator obtained with Algorithm \alg{EIG1}, for $A$ and $B$ of model \eqref{eq:GOEmodel_bis}, with permutation $\Pi^{\star}_n$ and noise parameter $\xi$. We have the following zero-one law:
	\begin{itemize}
		\item[$(i)$] If there exists $\epsilon>0$ such that $\xi = o(n^{-7/6-\epsilon})$ then $$\ov(\hat{\Pi}_n,\Pi^{\star}_n) \overset{L^1}{\longrightarrow} 1.$$
		\item[$(ii)$] If there exists $\epsilon>0$ such that $\xi = \omega(n^{-7/6+\epsilon})$ then $$\ov(\hat{\Pi}_n,\Pi^{\star}_n) \overset{L^1}{\longrightarrow} 0.$$
	\end{itemize}
\end{theorem}

Results of Theorem \ref{EIG1:theorem:01law} are illustrated on Figure \ref{fig:EIG1:image_overlap} showing the zero-one law at $\xi \asymp n^{-7/6}$. Note that the convergence to the step function appears to be slow.

\begin{figure}[h]
	\centering
	\includegraphics[width=14cm,height=8cm]{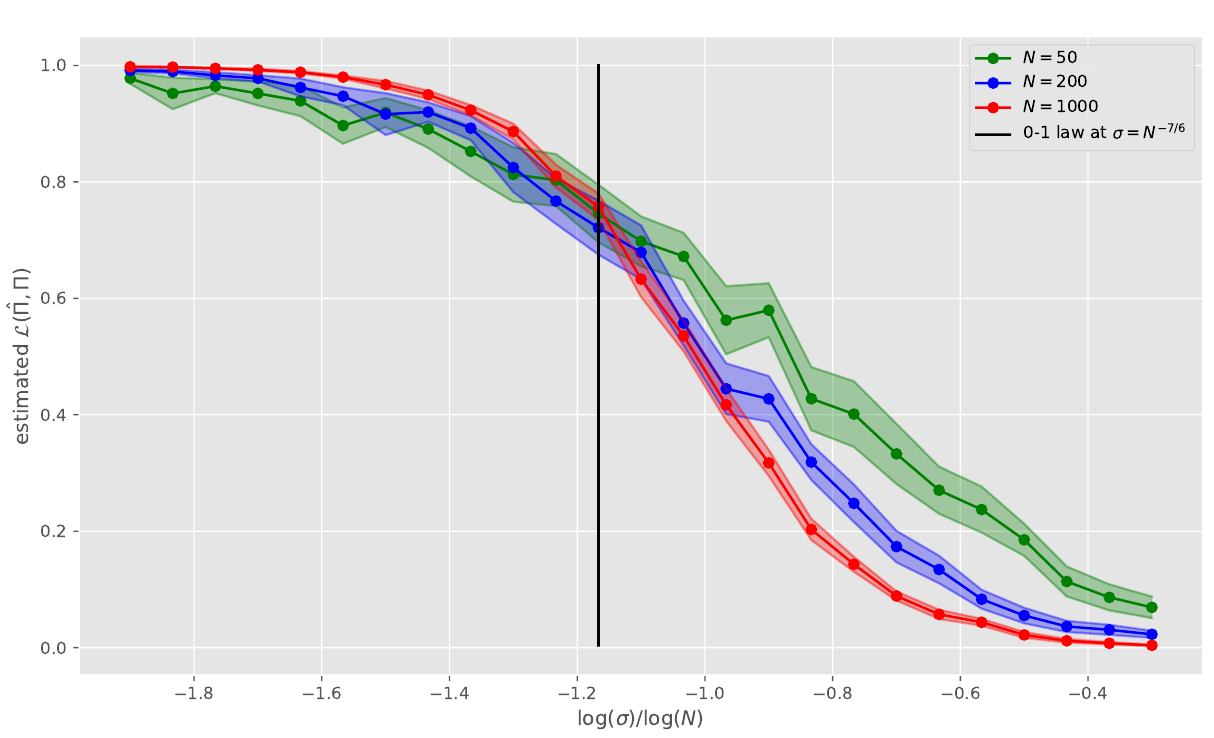}
	\caption{Estimated overlap $\ov(\hat{\Pi},\Pi^{\star})$ reached by \alg{EIG1} in model \eqref{eq:GOEmodel_bis}, for varying $n$ and $\xi$. \footnotesize{With $95\%$ confidence intervals}.}
	\label{fig:EIG1:image_overlap} 
\end{figure}

\begin{remark}
	We can now underline that without loss of generality, we can assume that $\Pi^{\star} = I_n$, the identity matrix. Indeed, one can return to the general case applying transformations $A \rightarrow \Pi^{\star} A \Pi^{\star,\top}$ and $H \rightarrow \Pi^{\star} H \Pi^{\star,\top}$. From now on we will assume in the rest of the chapter that $\Pi^{\star} = I_n$.
\end{remark}

In order to prove this theorem, it is necessary to establish two intermediate results along the way, which could also be of independent interest. First, we study the behavior of $v'_1$ with respect to $v_1$, showing that under some conditions on $\xi$ and $n$, the difference $v_1 - v'_1$ can be approximated by a renormalized Gaussian standard vector, multiplied by a  variance term $\mathbf{S}$, where $\mathbf{S}$ is a random variable which behavior is well understood in terms of $n$ and $\xi$ when $n \to \infty$. For this we work under the following assumption: 
\begin{equation}
\label{microscopicregime}
\exists \, \alpha >0, \;\xi = o \left(n^{-1/2-\alpha}\right),
\end{equation} 

\begin{proposition} \label{EIG1:prop:gaussian_decomp}
Under assumption \eqref{microscopicregime}, there exists a standard Gaussian vector $Z \sim \cN\left(0,I_n\right)$ independent from $v_1$ and a random variable $\mathbf{S} \asymp \xi n^{1/6}$, such that
	\begin{equation*}
	v'_1 = \left(1 + o_{\mathbb{P}}(1)\right) \left(v_1+ \mathbf{S} \frac{Z}{\|Z\|}\right).
	\end{equation*}
\end{proposition}

\begin{remark} \label{remark_assumption6} This assumption (\ref{microscopicregime}) (or a tighter formulation) arises when studying the diffusion trajectories of eigenvalues and eigenvectors in random matrices, and corresponds to the \textit{microscopic regime} in \cite{Allez14}. This assumption ensures that all eigenvalues of $B$ are close enough to the eigenvalues of $A$. This comparison term is justified from the random matrix theory ($n^{-1/2}$ is the typical amplitude of the spectral gaps $\sqrt{n}(\lambda_i - \lambda_{i+1})$ in the bulk, which are the smaller ones). 

Eigenvectors diffusions in similar models (diffusion processes dawn with the scaling $\xi = \sqrt{t}$) are studied in \cite{Allez14}, where the main tool is the Dyson Brownian motion (see e.g. \cite{Anderson09}) and its formulation for eigenvectors trajectories, giving stochastic differential equations for the evolutions of $v'_j(t)$ with respect to vectors $v_i=v'_i(0)$. These equations lead to a system of stochastic differential equations for the overlaps $\langle v_i, v'_j(t) \rangle$, which is quite difficult to analyze rigorously. In this work a more elementary method to get a expansion of $v'_1$ around $v_1$, for which this very condition (\ref{microscopicregime}) also appears.

Note that here, spectral gaps at the edge are of order $n^{-1/6}$ so assumption \eqref{microscopicregime} may not optimal for our study, and we expect Proposition \ref{EIG1:prop:gaussian_decomp} to hold up to $\xi = o \left(n^{-1/6-\alpha}\right)$. However, since the positive result of Theorem \ref{EIG1:theorem:01law} holds in a way more restrictive regime -- see condition $(i)$, condition \eqref{microscopicregime} is enough for our purpose and allows a short and simple proof.
\end{remark}

Proposition \ref{EIG1:prop:gaussian_decomp} suggests the study of $v'_1$ as a Gaussian perturbation of $v_1$. The main question is now formulated as follows: \textit{what is the probability that the perturbation on $v_1$ has an impact on the overlap of the estimator $\hat{\Pi}$ from the \alg{EIG1} method?} To answer this question, we introduce a correlated Gaussian vectors model (or \emph{toy model} hereafter) of parameters $n$ and $s>0$. In this model, we draw a standard Gaussian vector $X$ of size $n$ and $Y= X + s Z$ where $Z$ is an independent copy of $X$. We will use the notation $(X,Y) \sim \J(n,s)$.

Define $r_1$ the function that associates to any vector $T=(t_1,\ldots,t_p)$ the rank of $t_1$ in $T$ (for the usual decreasing order). For $(X,Y) \sim \J(n,s)$ we evaluate
\begin{equation*}
p(n,s) := \mathbb{P}\left(r_1(X)=r_1(Y)\right).
\end{equation*} Our second result shows that there is a zero-one law for the property of rank preservation in the toy model $\J(n,s)$.

\begin{proposition}[Zero-one law for $p(n,s)$]
\label{EIG1:prop:zero_one_toy}
In the correlated Gaussian vectors model we have the following:
	\begin{itemize}
		\item[$(i)$] If $s = o(1/n)$ then $$p(n,s) \underset{n \to \infty}{\longrightarrow} 1.$$
		\item[$(ii)$] If $s = \omega(1/n)$ then $$p(n,s) \underset{n \to \infty}{\longrightarrow} 0.$$
	\end{itemize}
\end{proposition} These results are illustrated on Figure \ref{fig:EIG1:image_jouet}, showing the zero-one law at $s \asymp n^{-1}$.

\begin{figure}[h]
	\centering
	\includegraphics[width=14cm,height=8cm]{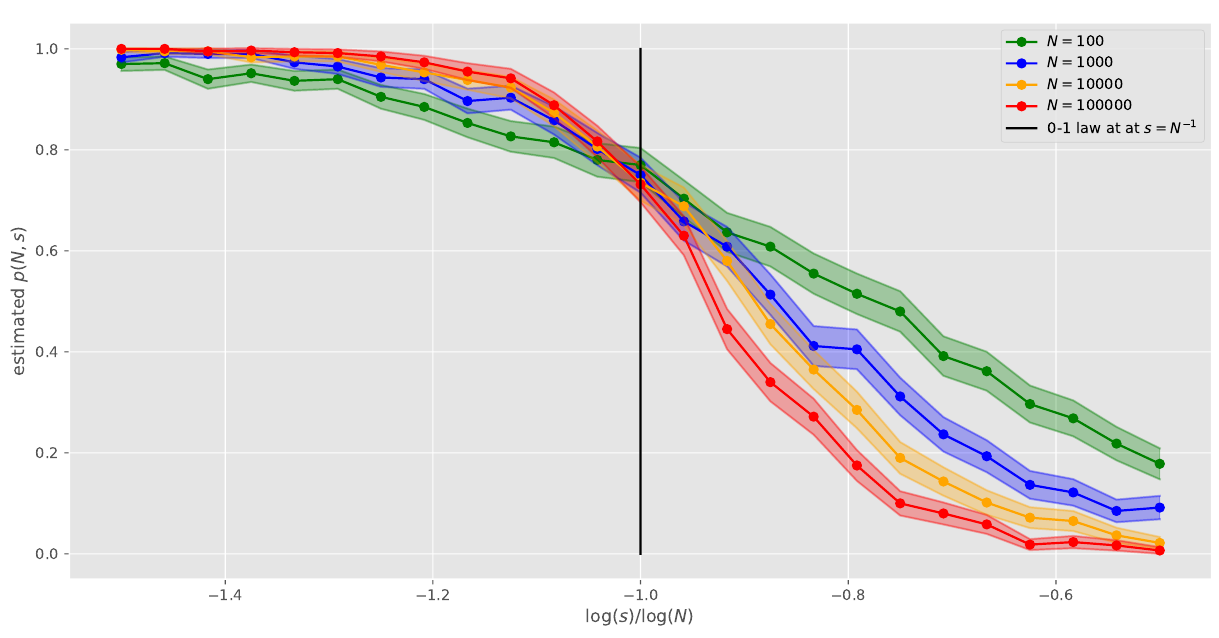}
	\caption{Estimated $p(n,s)$ in the toy model $\J(n,s)$. \footnotesize{With $95\%$ confidence intervals}.}
	\label{fig:EIG1:image_jouet} 
\end{figure}

\paragraph*{Organization of the chapter} The gaussian approximation of $v_1 - v'_1$ is established in Section \ref{linkGOEtoy} with the proof of Proposition \ref{EIG1:prop:gaussian_decomp}. The toy model defined here above is studied in Section \ref{toymodel} where Proposition \ref{EIG1:prop:zero_one_toy} is established. Finally, we gather results of Propositions \ref{EIG1:prop:gaussian_decomp} and \ref{EIG1:prop:zero_one_toy} in Section \ref{EIG1threshold} to show Theorem \ref{EIG1:theorem:01law}. Some additional proofs are deferred to Appendices \ref{section3_add_proofs} and \ref{section45_add_proofs}.

\section{Behavior of the leading eigenvectors of correlated matrices}
\label{linkGOEtoy}
The main idea of this section is to find a first order expansion of $v'_1$ around $v_1$. Recall that we use the notations $\left(v_1, v_2, \ldots, v_n\right)$ for normalized eigenvectors of $A$, corresponding to the eigenvalues $\lambda_1~\geq~\lambda_2~\geq~\ldots~\geq~\lambda_n$. Similarly, $\left(v'_1, v'_2, \ldots, v'_n\right)$ and $\lambda'_1~\geq~\lambda'_2~\geq~\ldots~\geq~\lambda'_n$ will refer to eigenvectors and eigenvalues of $B = A + \xi H$. Since $A$ and $B$ are symmetric, all these eigenvalues are real and the vectors $\left\{v_i\right\}_i$ (resp. $\left\{v'_i\right\}_i$) are pairwise orthogonal. We also recall that $v'_1$ is taken such that $\langle v_1,v'_1 \rangle >0$.

\subsection{Computation of a leading eigenvector of $B$}
Recall now that we are working under assumption \eqref{microscopicregime}:
\begin{equation*}
\exists \, \alpha >0, \;\xi = o \left(n^{-1/2-\alpha}\right).
\end{equation*}  Let $w'$ be an (non normalized) eigenvector of $B$ for the eigenvalue $\lambda'_1$ of the form
\begin{equation*}
w' := \sum_{i=1}^{n} \theta_i v_i,
\end{equation*}where we assume that $\theta_1 =1$. Such an assumption can be made a.s. since any hyperplane of $\dR^n$ has a null Lebesgue measure in $\dR^n$ (see Remark \ref{densitylebesgue}). 

The defining eigenvector equations projected on vectors $v_i$ give
\begin{equation} \label{eq:system_eigenequations}
\left \{
\begin{array}{c c c}
\theta_1 & = & 1,\\
\forall i>1, \; \theta_i &  = &  \dfrac{\xi}{\lambda'_1 - \lambda_i} \sum_{j=1}^{n} \theta_j \langle H v_j,v_i \rangle, \\
\lambda'_1 - \lambda_1 &  = & \xi  \sum_{j=1}^{n} \theta_j \langle H v_j,v_1 \rangle. \\
\end{array}
\right.
\end{equation} The strategy is then to approximately solve \eqref{eq:system_eigenequations} with an iterative scheme, leading to the following expansion:

\begin{proposition}
	\label{EIG1:prop:w'expansion}
	Under the assumption \eqref{microscopicregime} one has the following:
	\begin{equation}
	\label{w'}
	w' = v_1 +  \xi \sum_{i=2}^{n} \frac{\langle Hv_i,v_1 \rangle }{\lambda_1 - \lambda_i} v_i + o_{\mathbb{P}}\left(\xi \sum_{i=2}^{n} \frac{\langle Hv_i,v_1 \rangle }{\lambda_1 - \lambda_i} v_i  \right).
	\end{equation}
\end{proposition}
We refer to Appendix \ref{appendix_proof_prop_w'} for the details regarding the definition of the mentioned iterative scheme, as well as a proof of Proposition \ref{EIG1:prop:w'expansion}. The proof uses assumption \eqref{microscopicregime} an builds upon some standard results on the distribution of eigenvalues in the $\GOE$.

\begin{remark}
\label{microscopicregimenotoptimal}
The above proposition could easily be extended for all eigenvectors of $B$, under assumption \eqref{microscopicregime}.
Based on the studies of the trajectories of the eigenvalues and eigenvectors in the $\GUE$ \cite{Allez14} and the $\GOE$ \cite{Allez14bis}, since we are only interested here in the leading eigenvectors, we expect the result of Proposition \ref{EIG1:prop:w'expansion} to hold under the weaker assumption $\xi n^{1/6+\alpha} \to 0$,  for $n^{-1/6}$ is the typical spectral gap $\sqrt{n}(\lambda_1-\lambda_2)$ on the edge. However, as explained before (see Remark \ref{remark_assumption6}), our analysis doesn't require this more optimal assumption. We also know that the expansion \eqref{w'} doesn't hold as soon as $\xi = \omega(n^{-1/6})$. A result proved by Chatterjee (\cite{Chatterjee14}, Theorem 3.8) shows that the eigenvectors corresponding to the highest eigenvalues $v_1$ of $A$ and $v'_1$ of $B=A+\xi H$, when $A$ and $H$ are two independent matrices from the $\GUE$, are delocalized (in the sense that $\langle v_1,v'_1 \rangle$ converges in probability to $0$ as $n \to \infty$), when $\xi = \omega(n^{-1/6})$.
\end{remark}

\subsection{Gaussian representation of $v'_1 - v_1$}
We still work under assumption (\ref{microscopicregime}). After renormalization, we have $v'_1 = \frac{w'}{\| w' \|}$. We are now able to study the behavior of the overlap $\langle v'_1, v_1 \rangle$:
\begin{equation*}
\langle v'_1, v_1 \rangle = \left(1+\xi^2 (1+o_{\mathbb{P}}(1)) \sum_{i=2}^{n} \dfrac{\langle H v_i,v_1 \rangle ^2 }{\left(\lambda_1 - \lambda_i\right)^2}\right)^{-1/2}
\end{equation*}
Hence
\begin{equation}
\label{eq:diffusion_vp}
\langle v'_1, v_1 \rangle = 1-\frac{\xi^2}{2} \sum_{i=2}^{n} \frac{\langle H v_i,v_1 \rangle ^2 }{\left(\lambda_1 - \lambda_i\right)^2} + o_{\mathbb{P}}\left(\xi^2 \sum_{i=2}^{n} \frac{\langle H v_i,v_1 \rangle ^2 }{\left(\lambda_1 - \lambda_i\right)^2} \right).
\end{equation} 

Let us give the heuristic to evaluate the first sum in the right-hand side of \eqref{eq:diffusion_vp}: since the $\GOE$ distribution is invariant by rotation (see e.g. \cite{Anderson09}), the random variables $\langle H v_i,v_1 \rangle$ are zero-mean Gaussian, with variance $1/n$. Moreover, it is well known \cite{Anderson09} that the eigenvalue gaps $\lambda_1 - \lambda_i$ are of order $n^{-1/6}$ when $i$ is small, and $n^{-1/2}$ in the bulk (when $i$ is typically of order $n$). These considerations lead to the following:
\begin{lemma}
	\label{lemma_concentrationksi}
	We have the following concentration
	\begin{equation}
	\label{lemma_concentrationksieq}
	\sum_{i=2}^{n} \frac{\langle H v_i,v_1 \rangle ^2 }{\left(\lambda_1 - \lambda_i\right)^2} \asymp n^{1/3}.
	\end{equation}
\end{lemma} We refer to Appendix \ref{proof_lemma_concentrationksi} for a rigorous proof of this result. With this Lemma, we are now able to give the first order expansion of $\langle v'_1, v_1 \rangle$ with respect to $\xi$:
\begin{equation}
\label{eq:diffusion_vp_2}
\langle v'_1, v_1 \rangle = 1-\frac{\xi^2}{2} n^{1/3} + o_{\mathbb{P}}\left(\xi^2 n^{1/3} \right).
\end{equation} 
\begin{remark}
The comparison between $\xi $ and $n^{1/6}$ made in \cite{Chatterjee14} naturally reappears here, as $\xi^2 n^{1/3}$ is the typical shift of $v'_1$ with respect to $v_1$.
\end{remark}

The intuition is that the scalar product $\langle v'_1, v_1 \rangle$ is sufficient to derive a Gaussian representation of $v'_1$ w.r.t. $v_1$. We formalize this in the following
\begin{lemma}
	\label{lemma_invrotation}
	Given $v_1$, when writing the decomposition $w' = v_1 + w$, with
	\begin{equation*}
	w :=\sum_{i=2}^{n} \theta_i v_i,
	\end{equation*}
	the distribution of $w$ is invariant by rotation in the orthogonal complement of $v_1$. This implies in particular that given $v_1$, $\|w\|$ and $\frac{w}{\|w\|}$ are independent, and that $\frac{w}{\|w\|}$ is uniformly distributed on $\mathbb{S}^{n-2}$, the unit sphere of $v_1^\perp$.
\end{lemma}
\begin{proof}[Proof of Lemma \ref{lemma_invrotation}]
	We work conditionnally on $v_1$. Let $O$ be an orthogonal transformation of the hyperplane $v_1^{\perp}$ (such that $Ov_1=v_1$). Since the $\GOE$ distribution is invariant by rotation and $A$ and $H$ are independent, $\widetilde{B} := O^{\top}AO + \xi O^{\top}HO$ has he same distribution as $B = A+\xi H$.
	
	Note that $Ow' = v_1 + Ow$ is an eigenvector of $\widetilde{B}$ for the eigenvalue $\lambda_{1}$. Since the distribution of the matrix of eigenvectors $(v_2, \ldots, v_n)$ is the Haar measure on the orthogonal group $\mathcal{O}_{n-1}\left(v_1^{\perp}\right)$, denoted by $d\mathcal{H}$, the distribution of $w$ is also invariant by rotation in the orthogonal complement of $v_1$. Furthermore, for any $f,g$ bounded continuous functions and $O \in \mathcal{O}_{n-1}\left(v_1^{\perp}\right)$,
	\begin{flalign*}
	\mathbb{E}\left[f(\|w\|)g\left(\frac{w}{\|w\|}\right)\right] &= \mathbb{E}\left[f(\|w\|)g\left(\frac{Ow}{\|Ow\|}\right)\right] = \mathbb{E}\left[f(\|w\|) \int_{\mathcal{O}_{n-1}\left(v_1^{\perp}\right)} d\mathcal{
		H}(O) g\left(\frac{Ow}{\|Ow\|}\right)\right] \\
	&= \mathbb{E}\left[f(\|w\|) \int_{\mathbb{S}^{n-2}} \frac{g(u) du}{\mathrm{Vol}\left(\mathbb{S}^{n-2}\right)}\right] = \mathbb{E}\left[f(\|w\|)\right]\mathbb{E}\left[g\left(\frac{w}{\|w\|}\right)\right].
	\end{flalign*} 
	This completes the proof of Lemma  \ref{lemma_invrotation}.
\end{proof} 

We can now show the main result of this section, Proposition \ref{EIG1:prop:gaussian_decomp}.
\begin{proof}[Proof of Proposition \ref{EIG1:prop:gaussian_decomp}]
	Recall the decomposition $w'= v_1 + w$ with $w =\sum_{i=2}^{n} \theta_i v_i$. According to Lemma \ref{lemma_invrotation}, conditioned to $v_1$, $\frac{w}{\|w\|}$ is uniformly distributed on $\mathbb{S}^{n-2}$, the unit sphere of $v_1^\perp$. We now state a classical result about sampling uniform vectors on a sphere:
	
\begin{lemma}
		Let $E$ be $p-$dimensional Euclidean space, endowed with an orthogonal basis $\cB = (e_1, \ldots, e_p)$. Let $u$ be a random vector uniformly distributed on the unit sphere $\mathbb{S}^{p-1}$ of $E$. Then, in basis $\cB$, $u$ has the same distribution as
		$$ \left( \frac{Z_1}{\sqrt{\sum_{i=1}^{p} Z_i^2}}, \ldots, \frac{Z_p}{\sqrt{\sum_{i=1}^{p} Z_i^2}} \right), $$
		where $Z_1, \ldots, Z_p$ are i.i.d. standard normal random variables.
\end{lemma} 

We refer e.g. to \cite{ORourke16}, Lemma 10.1, for the proof of this result. In our context, this proves that the joint distribution of the coordinates $w_2, \ldots, w_n$ of $w$ along $v_2, \ldots, v_n$ is always that of a normalized standard Gaussian vector (on $\dR^{n-1}$). This joint probability does not dependent on $v_1$. Hence, there exist $Z_2, \ldots, Z_n$ standard Gaussian independent variables, independent from $v_1$ (and from $\| w\|$ by Lemma \ref{lemma_invrotation}), such that:
	\begin{equation*}
	w' = v_1 + \frac{\|w\|}{\left(\sum_{i=2}^{n} Z_i^2\right)^{1/2}} \sum_{i=2}^{n} Z_i v_i.
	\end{equation*}
	Let $Z_1$ be another standard Gaussian variable, independent from everything else. Then
	\begin{equation*}
	w' = \left(1 - \frac{\|w\| Z_1}{\left(\sum_{i=2}^{n} Z_i^2\right)^{1/2}}\right) v_1 + \frac{\|w\|}{\left(\sum_{i=2}^{n} Z_i^2\right)^{1/2}} \sum_{i=1}^{n} Z_i v_i.
	\end{equation*} Let $Z =\sum_{i=1}^{n} Z_i v_i$, which is a standard Gaussian vector. Since the distribution of $Z$ is invariant by permutation of the $\left(Z_i\right)_{1 \leq i \leq n}$, $Z$ and $v_1$ are independent. We have
	\begin{flalign*}
	v'_1 &= \frac{w'}{\| w'\|} = \frac{w'}{\sqrt{1+\|w\|^2}} \\&= \frac{1}{\sqrt{1+\|w\|^2}} \left(1 - \frac{\|w\| Z_1}{\left(\sum_{i=2}^{n} Z_i^2\right)^{1/2}}\right) v_1+ \frac{\|w\| \|Z\|}{\sqrt{1+\|w\|^2} \left(\sum_{i=2}^{n} Z_i^2\right)^{1/2}} \frac{Z}{\|Z\|}.
	\end{flalign*} 
	Taking 
	\begin{equation*}
	\mathbf{S}=\frac{\|w\| \|Z\|}{\left(\sum_{i=2}^{n} Z_i^2\right)^{1/2} - \|w\| Z_1},
	\end{equation*} we get 
	\begin{equation}\label{eq:v'1_wZ}
	v'_1 = \frac{1}{\sqrt{1+\|w\|^2}} \left(1 - \frac{\|w\| Z_1}{\left(\sum_{i=2}^{n} Z_i^2\right)^{1/2}}\right) \left(v_1+ \mathbf{S} \frac{Z}{\|Z\|}\right).
	\end{equation}
	Proposition \ref{EIG1:prop:w'expansion} together with Lemma \ref{lemma_concentrationksi} yield
	\begin{equation*}
		\|w\|^2=\|w'-v_1\|^2 = (1+o_{\mathbb{P}}(1)) \cdot \xi^2 \sum_{i=2}^{n} \frac{\langle H v_i,v_1 \rangle ^2 }{\left(\lambda_1 - \lambda_i\right)^2} \asymp \xi^2 n^{1/3},
	\end{equation*} the last quantity being $o(1)$ under assumption \eqref{microscopicregime}. With the previous computation, equation \eqref{eq:v'1_wZ} becomes
	\begin{align*}
	v'_1 = \left(1 + o_{\mathbb{P}}(1)\right) \left(v_1+ \mathbf{S} \frac{Z}{\|Z\|}\right),
	\end{align*}
	with $\mathbf{S} = (1+o_{\mathbb{P}}(1)) \|w\|\; {\asymp} \; \xi n^{1/6}$.
\end{proof}

\section{Definition and analysis of a toy model}
\label{toymodel}
Now that we have established a expansion of $v'_1$ with respect to $v_1$, our main question boils down to the study of the effect of a random Gaussian perturbation of a Gaussian vector in terms of rank of its coordinates: if these ranks are preserved, the permutation that aligns these two vectors will be $\hat{\Pi}=\Pi^{\star} = I_n$. Otherwise we want to understand the error made between $\hat{\Pi}$ and $\Pi^{\star}=I_n$.

\subsection{Definitions and notations}
\label{relevantlink}
We refer to Section \ref{EIG1:subsection:notations} for the definition of the toy model $\J(n,s)$. Recall that we want to compute, when $(X,Y) \sim \J(n,s)$, the probability
\begin{equation*}
p(n,s) := \mathbb{P}\left(r_1(X)=r_1(Y)\right).
\end{equation*}
In this section, we denote by $E$ the probability density function of a standard Gaussian variable, and $F$ its cumulative distribution function. Namely
\begin{equation*}
E(u) := \frac{1}{\sqrt{2 \pi}} e^{-u^2 /2} \quad \mbox{and} \quad F(u) := \frac{1}{\sqrt{2 \pi}} \int_{- \infty}^{u} e^{-z^2 /2} dz.
\end{equation*}

We hereafter elaborate on the link between this toy model and our first matrix model \eqref{eq:GOEmodel_bis} in Section \ref{linkGOEtoy}. Since $v_1$ is uniformly distributed on the unit sphere, we have the equality in distribution $v_1 = \frac{X}{\|X\|}$ where $X$ is a standard Gaussian vector of size $n$, independent of $Z$ by Proposition \ref{EIG1:prop:gaussian_decomp}. We write
\begin{align*}
v_1 &= \frac{X}{\|X\|}, \\
v'_1 &=  \left( 1 + o_{\mathbb{P}}(1) \right)\left(  \frac{X}{\|X\|} + \mathbf{S} \frac{Z}{\|Z\|}\right).
\end{align*}
Note that for all $\lambda>0$, $r_1(\lambda T) = r_1(T)$, hence
\begin{equation}
\label{GOEtoylink}
r_1(v_1) = r_1(X), \quad r_1(v'_1) = r_1 \left(X + \mathbf{s} Z \right),\\
\end{equation} where 
\begin{equation*}
\mathbf{s}  = \frac{\mathbf{S}  \|X\|}{\|Z\|} \asymp \xi n^{1/6},
\end{equation*} where we used the law of large numbers ($\|X\|/\|Z\| \to 1$ p.s.) as well as Proposition \ref{EIG1:prop:gaussian_decomp} in the last expansion. Equation \eqref{GOEtoylink} shows that this toy model is relevant for our initial problem, up to the fact that the noise term $\mathbf{s}$ is random in the matrix model (though we know its order of magnitude to be $\asymp \xi n^{1/6}$).
\begin{remark}
The intuition for the zero-one law for $p(n,s)$ is as follows. If we sort the $n$ coordinates of $X$ on the real axis, all coordinates being typically perturbed by a factor $s$, it seems natural to compare $s$ with the typical gap between two coordinates of order $1/n$ to decide whether the rank of the first coordinate of $X$ is preserved in $Y$. 
\end{remark}
Let us show that this intuition is rigorously verified. For every couple $(x,y)$ of real numbers, define
\begin{equation*}
\cN_{n,s}^{+}(x,y) := \card{\set{1 \leq i \leq n, \; X_i > x, Y_i < y}},
\end{equation*}
\begin{equation*}
\cN^{-}_{n,s}(x,y) :=  \card{\set{1 \leq i \leq n, \; X_i < x, Y_i > y}} .
\end{equation*} In the following, we omit all dependencies in $n$ and $s$, using the notations $\cN^{+}$ and $\cN^{-}$. The corresponding regions are shown on Figure \ref{fig:EIG1:imagen+-}.
\begin{figure}[h]
	\centering
	\begin{picture}(240,240)
\put(140,90){\textcolor{green}{\line(1,1){20}}}
\put(120,90){\textcolor{green}{\line(1,1){40}}}
\put(100,90){\textcolor{green}{\line(1,1){60}}}
\put(80,90){\textcolor{green}{\line(1,1){80}}}
\put(60,90){\textcolor{green}{\line(1,1){100}}}
\put(40,90){\textcolor{green}{\line(1,1){120}}}
\put(20,90){\textcolor{green}{\line(1,1){140}}}
\put(0,90){\textcolor{green}{\line(1,1){160}}}
\put(0,110){\textcolor{green}{\line(1,1){140}}}
\put(0,130){\textcolor{green}{\line(1,1){120}}}
\put(0,150){\textcolor{green}{\line(1,1){100}}}
\put(0,170){\textcolor{green}{\line(1,1){80}}}
\put(0,190){\textcolor{green}{\line(1,1){60}}}
\put(65,165){\mbox{$\mathcal{N}^{-}(x,y)$}}

\put(160,70){\textcolor{red}{\line(1,1){20}}}
\put(160,50){\textcolor{red}{\line(1,1){40}}}
\put(160,30){\textcolor{red}{\line(1,1){60}}}
\put(160,10){\textcolor{red}{\line(1,1){80}}}
\put(170,0){\textcolor{red}{\line(1,1){75}}}
\put(190,0){\textcolor{red}{\line(1,1){55}}}
\put(210,0){\textcolor{red}{\line(1,1){35}}}
\put(230,0){\textcolor{red}{\line(1,1){15}}}
\put(185,45){\mbox{$\mathcal{N}^{+}(x,y)$}}

\put(0,120){\vector(1,0){240}}

\put(120,0){\vector(0,1){240}}

\put(160,90){\circle*{4}}
\put(174,98){\makebox(0,0){$(x,y)$}}

\multiput(0,90)(10,0){25}{\line(1,0){5}}

\multiput(160,0)(0,10){25}{\line(0,1){5}}

\end{picture}
	\caption{Areas corresponding to $\cN^{+} (x,y)$ and $\cN^{-} (x,y)$.}
	\label{fig:EIG1:imagen+-}
\end{figure}
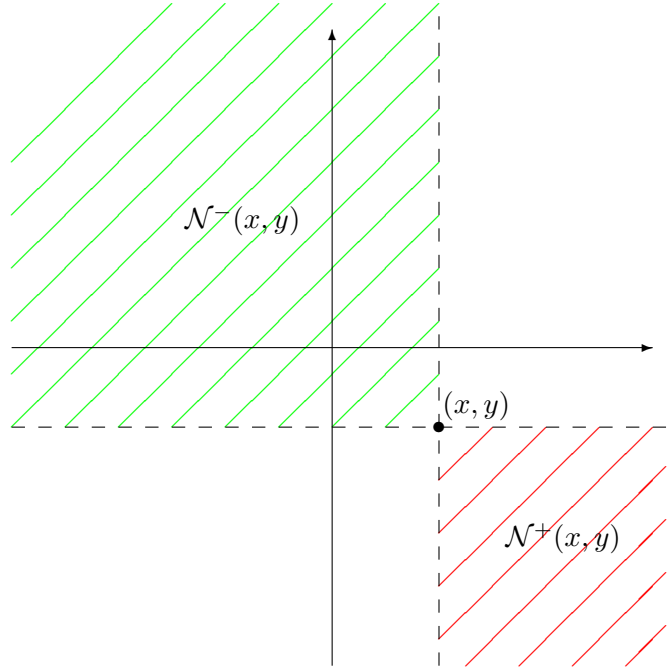 We will also need the following probabilities

\begin{align*}
S^{+}(x,y) &:= \mathbb{P}\left(X_1 > x, Y_1 < y\right), \mbox{ and}\\
S^{-}(x,y) &:= \mathbb{P}\left(X_1 < x, Y_1 > y\right) = S^{+}(-x,-y).
\end{align*}
In terms of distribution, the random vector $$\left(\cN^{+}(x,y), \cN^{-}(x,y), n-1-\cN^{+}(x,y)-\cN^{+}(x,y)\right)$$ follows a multinomial distribution of parameters $$\left(n-1,S^{+}(x,y),S^{-}(x,y),1-S^{+}(x,y)-S^{-}(x,y)\right).$$ 
In order to have $r_1(X)=r_1(Y)$, there must be the same number of points on the two domains on Figure \ref{fig:EIG1:imagen+-}, for $x=X_1$ and $y=Y_1$. We then have the following expression of $p(n,s)$:
\begin{align*} 
p(n,s) &= \mathbb{E}\left[\mathbb{P}\left(\cN^{+}(X_1,Y_1)=\cN^{-}(X_1,Y_1)\right)\right] \\
&= \int_{\dR} \int_{\dR} \mathbb{P}(dx,dy) \mathbb{P}(\cN^{+}(x,y)=\cN^{-}(x,y)) \\
&= \int_{\dR} \int_{\dR} E(x) E(z) \phi_{x,z}(n,s) \, dx \, dz,
\end{align*} with
\begin{equation}
\label{phiexpression}
\phi_{x,z}(n,s) := \sum_{k=0}^{\lfloor(n-1)/2\rfloor} \binom{n-1}{k}  \binom{n-1-k}{k} \left(S^{+}_{x,z}\right)^k \left(S^{-}_{x,z}\right)^k \left(1-S^{+}_{x,z}-S^{-}_{x,z}\right)^{n-1-2k},
\end{equation} using the notations $S^{+}_{x,z} = S^{+}(x,x+sz)$ and $S^{-}_{x,z} = S^{-}(x,x+sz)$. A simple computation shows that
\begin{flalign}
S^{+}(x,x+sz) &= \int_{x}^{+ \infty} \frac{1}{\sqrt{2\pi}} e^{-u^2 /2} \left( \int_{- \infty}^{z+\frac{x-u}{s}} \frac{1}{ \sqrt{2\pi}} e^{-v^2 / 2} \, dv \right) du \nonumber \\
&= \int_{x}^{+ \infty} E(u) \, F \left(z-\frac{u-x}{s}\right) du,  \label{S+-inf} \\
&= s\int_{0}^{+ \infty} E(x+vs) \, F \left(z-v\right) dv. \label{S+-0} 
\end{flalign}
We have the classical integration result
\begin{equation}
\label{integrateS_+-}
\int_{- \infty}^{z} F(u) du = z F(z) + E(z).
\end{equation}
From (\ref{S+-inf}), (\ref{S+-0}) and (\ref{integrateS_+-}) we derive the following easy lemma:
\begin{lemma}
	\label{s+-}
	For all $x$ and $z$,
	\begin{flalign*}
	S^{+}(x,x+s z) & \underset{s \to 0}{=} s \left[E(x) \left(z F(z) +E(z)\right)\right] + o(s),\\
	S^{+}(x,x+s z) & \underset{s \to \infty}{\longrightarrow} F(z) \left(1-F(x)\right),\\
	S^{-}(x,x+s z) & \underset{s \to 0}{=} s \left[E(x) \left(-z+ z F(z) +E(z)\right)\right] + o(s),\\
	S^{-}(x,x+s z) & \underset{s \to \infty}{\longrightarrow} F(x) \left(1-F(z)\right).
	\end{flalign*}
	Moreover, both $s \mapsto S^{+}(x,x+s z)$ and $s \mapsto S^{-}(x,x+s z)$ are increasing.
\end{lemma}

\subsection{Zero-one law for $p(n,s)$} 
In this Section we give a proof of Proposition \ref{EIG1:prop:zero_one_toy}.

\begin{proof}[Proof of Proposition \ref{EIG1:prop:zero_one_toy}]
	\proofstep{First case $(i)$.} If $s = o(1/n)$, we have the following inequality
	\begin{equation}
	\label{pi}
	p(n,s) \geq \int_{\dR} \int_{\dR} dx dz E(x) E(z) \mathbb{P} \left(\cN^{+}(x,x+sz)=\cN^{-}(x,x+sz)=0\right).
	\end{equation}
	According to Lemma \ref{s+-}, for all $x,z \in \dR$
	\begin{flalign*}
	\mathbb{P} \left(\cN^{+}(x,x+sz)=\cN^{-}(x,x+sz)=0\right)&= \left(1-S^{+}(x,x+sz)-S^{-}(x,x+sz)\right)^{n-1}\\ & \sim \exp\left(- nsE(x)\left[z(2F(z)-1)+2E(z)\right]\right) \\
	& \underset{n \to \infty}{\longrightarrow} 1,&&
	\end{flalign*} By applying the dominated convergence theorem in (\ref{pi}), we conclude that $p(n,s) \to 1$.\\
	
	\proofstep{Second case $(ii)$.} If $s n \to \infty$, recall that
	\begin{equation}
	\label{pii}
	p(n,s)= \int_{\dR} \int_{\dR} dx dz E(x) E(z) \phi_{x,z}(n,s),
	\end{equation} with $\phi_{x,z}$ defined in equation $(\ref{phiexpression})$. In the rest of the proof, we fix $x$ and $z$ two real numbers. Letting
\begin{equation*}
b(n,s,k) := \binom{n-1}{k} \left(S^{+}_{x,z}\right)^k \left(1- S^{+}_{x,z}\right)^{n-1-k}
\end{equation*} and 
\begin{equation*}
M(n,s) := \underset{0 \leq k \leq n-1}{\max} b(n,s,k).
\end{equation*} Note that by Lemma \ref{s+-}, there exists $C = C(x,z) <1$ such that for $n$ large enough, $S^{+}_{x,z}<C<1$. Moreover, combining this Lemma with assumption $(ii)$ gives that $ n S^{+}_{x,z} \to \infty$. It is also known that $M(n,s) = b(n,s,\lfloor n S^{+}_{x,z} \rfloor)$ and a classical computation shows that in this case (see e.g. \cite{Bollobas2001}, formula 1.5):
\begin{flalign*}
M(n,s) & = \binom{n-1}{\lfloor n S^{+}_{x,z} \rfloor} \left(S^{+}_{x,z}\right)^{\lfloor n S^{+}_{x,z} \rfloor} \left(1- S^{+}_{x,z}\right)^{n-1-{\lfloor n S^{+}_{x,z} \rfloor}}\\
& \sim \frac{1}{\sqrt{2 \pi n t(1-t)}} t^{-(n-1)t} (1-t)^{-(n-1)(1-t)} \left(S^{+}_{x,z}\right)^{(n-1)t} \left(1- S^{+}_{x,z}\right)^{(n-1)(1-t)}\\
& = \left(n S^{+}_{x,z}\right)^{-1/2} (1+O(1)) \to 0.
\end{flalign*} where $t := \frac{\lfloor n S^{+}_{x,z} \rfloor}{n-1} \sim  S^{+}_{x,z}$. Working with equation \eqref{phiexpression}, we obtain the following control
\begin{flalign*}
\phi_{x,z}(n,s) & \leq M(n,s) \times \sum_{k=0}^{\lfloor(n-1)/2\rfloor}  \binom{n-1-k}{k} \left(S^{-}_{x,z}\right)^k \frac{\left(1-S^{+}_{x,z}-S^{-}_{x,z}\right)^{n-1-2k}}{\left(1-S^{+}_{x,z}\right)^{n-1-k}}\\
& \overset{(a)}{=} M(n,s) \times \frac{(1-S^{+}_{x,z})\left(1-\left(\frac{-S^{+}_{x,z}}{1-S^{-}_{x,z}} \right)^n\right)}{1+S^{-}_{x,z}-S^{+}_{x,z}}\\
& \overset{(b)}{=} M(n,s) \times O(1) \underset{n \to \infty}{\longrightarrow} 0. &&
\end{flalign*}
We used in $(b)$ the fact that $S^{+}_{x,z}+S^{-}_{x,z}$ is increasing in $s$, and that given $x$ and $z$, for all $s>0$, by Lemma \ref{s+-}, $$S^{+}_{x,z}+S^{-}_{x,z}< F(x)\left(1-F(z)\right)+F(z)\left(1-F(x)\right)<1.$$ We used in $(a)$ the following combinatorial result:

\begin{lemma}
\label{fibo}
For all $\alpha>0$, 
\begin{equation}
\label{sum_fibo}
\sum_{k=0}^{\lfloor(n-1)/2\rfloor}  \binom{n-1-k}{k} \alpha^k = \frac{1}{\sqrt{1+4 \alpha}} \left[\left(\frac{1+\sqrt{1+4 \alpha}}{2}\right)^n - \left(\frac{1-\sqrt{1+4 \alpha}}{2}\right)^n\right].
\end{equation}
\end{lemma} 
We refer to Appendix \ref{proof_lemma_fibo} for a proof of this result. To obtain $(a)$ from Lemma \ref{fibo}, we apply $(\ref{sum_fibo})$ to $\alpha = \frac{S^{-}_{x,z} \left(1-S^{+}_{x,z}\right)}{\left(1-S^{+}_{x,z}-S^{-}_{x,z}\right)^2}$, with $\sqrt{1+4 \alpha} = \frac{1-S^{+}_{x,z}+S^{-}_{x,z}}{1-S^{+}_{x,z}-S^{-}_{x,z}}$. Some simple simplifications then give the claimed result. The dominated convergence theorem in $(\ref{pii})$ shows that $p(n,s) \to 0$ and ends the proof.
\end{proof}

\begin{remark}
The above computations also imply the existence of a non-degenerate limit of $p(n,s)$ in the critical case where $s n \to c>0$: in this case, previous discussions as well as Lemma \ref{s+-} show that the joint distribution of $(\cN^+(x,x+sz), \cN^-(x,x+sz))$ is asymptotically
$$\Poi(c \left[E(x) \left(z F(z) +E(z)\right)\right]) \otimes \Poi(c \left[E(x) \left(-z+z F(z) +E(z)\right)\right]).  $$ 
Therefore, $p(n,s)$ has a non-degenerate limit given by 
\begin{equation}
\int_{\dR} \int_{\dR} E(x) E(z) \cdot \mathbf{G}\left(c \left[E(x) \left(z F(z) +E(z)\right)\right],c \left[E(x) \left(-z+z F(z) +E(z)\right)\right] \right) \, dx \, dz,
\end{equation} where
\begin{equation}
\mathbf{G}(a,b) := \mathbb{P}(\Poi(a)=\Poi(b)) = e^{-(a+b)} \sum_{k \geq 0} \frac{a^k b^k}{(k!)^2}.
\end{equation} 
\end{remark}

\section{Analysis of the \alg{EIG1} method for matrix alignment}
\label{EIG1threshold}
By now, we come back to our initial problem, which is the analysis of \alg{EIG1} method. Recall that for any estimator $\hat{\Pi}$ of $\Pi^{\star}$, its overlap is defined as follows
\begin{equation*}
\ov(\hat{\Pi},\Pi^{\star}) := \frac{1}{n} \sum_{u=1}^{n} \one_{\hat{\Pi}(u)=\Pi^{\star}(u)}.
\end{equation*}
The aim of this section is to show how Propositions \ref{EIG1:prop:gaussian_decomp} and \ref{EIG1:prop:zero_one_toy} can be assembled to show the main result of our study, namely Theorem \ref{EIG1:theorem:01law}.

\begin{proof}[Proof of Theorem \ref{EIG1:theorem:01law}]
\proofstep{First case $(i)$.} Assuming $\xi = o(n^{-7/6-\epsilon})$ for some $\epsilon>0$, then in particular condition (\ref{microscopicregime}) holds. Proposition \ref{EIG1:prop:gaussian_decomp} as well as equation \eqref{GOEtoylink} in Section \ref{toymodel} enable to identify $v_1$ and $v'_1$ with the following vectors:
\begin{equation}
\label{identificationrang}
v_1 \sim X, \quad v'_1 \sim X + \mathbf{s} Z,
\end{equation} where $X$ and $Z$ are two independent Gaussian vectors from the toy model, and where $\mathbf{s} \asymp \xi n^{1/6} $  w.h.p. Recall that we work under the assumptions $\Pi^{\star}=I_n$ and $\langle v_1, v'_1 \rangle >0$. In this case, we expect $\Pi_{+}$ to be very close to $I_n$.

We will use the notations of Section \ref{toymodel} hereafter. Let's take $f \in \mathcal{F}$ such that w.h.p., $\xi n^{1/6}f(n)^{-1} \leq \mathbf{s} \leq \xi n^{1/6 }f(n)$. We have for all $1 \leq i \leq n$,
\begin{flalign*}
\mathbb{P}\left({\Pi_{+}}(i)=\Pi^{\star}(u)\right) &=  \mathbb{P}\left({\Pi_{+}}(1)=\Pi^{\star}(1)\right) \\
& = \mathbb{E}\left[\iint dx dz E(x) E(z) \phi_{x,z}\left(n, \mathbf{s} \right)\one_{\xi n^{1/6}f(n)^{-1} \leq \mathbf{s} \leq \xi n^{1/6 }f(n)}\right] + o(1)\\
& = \iint dx dz E(x) E(z) \mathbb{E}\left[\phi_{x,z}\left(n, \mathbf{s} \right)\one_{\xi n^{1/6}f(n)^{-1} \leq \mathbf{s} \leq \xi n^{1/6 }f(n)}\right] + o(1).&&
\end{flalign*} When conditioning on the event $\cA$ where $\xi n^{1/6}f(n)^{-1} \leq \mathbf{s} \leq \xi n^{1/6 }f(n)$, we know that $\mathbf{s}n \to 0$ by condition $(i)$ and for all $x,z$, $\dE\left[\phi_{x,z}\left(n, \mathbf{s} \right) \, | \, \cA \right] \to 1$ as shown in Section \ref{toymodel}. Since $\cA$ occurs w.h.p. we have
\begin{flalign*}
\dE\left[\phi_{x,z}\left(n, \mathbf{s} \right)\one_{\cA}\right] {\longrightarrow} 1,
\end{flalign*} which implies with the dominated convergence theorem that 
\begin{equation}\label{eq:last_eq(i)}
\mathbb{E}\left[\ov({\Pi_{+}},\Pi^{\star})\right] \underset{n \to \infty}{\longrightarrow} 1
\end{equation}
and thus
\begin{equation*}
\ov({\Pi_{+}},\Pi^{\star}) \overset{L^1}{\rightarrow} 1.
\end{equation*} 
We now check that w.h.p., $\Pi_{+}$ is preferred to $\Pi_{-}$ in the \alg{EIG1} method:
\begin{lemma}
		\label{Pi+casei}
		In the case $(i)$ of Theorem \ref{EIG1:theorem:01law}, if $\langle v_1,v'_1 \rangle >0$, we have w.h.p.
		\begin{equation*}
		\langle A, \Pi_{+} B \Pi_{+}^{\top} \rangle > \langle A, \Pi_{-} B \Pi_{-}^{\top} \rangle,
		\end{equation*} in other words Algorithm \alg{EIG1} returns w.h.p. $\hat{\Pi}=\Pi_{+}$.
\end{lemma} This Lemma is proved in Appendix \ref{proof_lemma_Pi+(i)} and implies, together with \eqref{eq:last_eq(i)}, that
\begin{flalign*}
\mathbb{E}\left[\ov(\hat{\Pi},\Pi^{\star})\right] & \geq \mathbb{E}\left[\ov(\hat{\Pi},\Pi^{\star})\one_{\hat{\Pi}=\Pi_{+}}\right] = \mathbb{E}\left[\ov(\Pi_{+},\Pi^{\star})\one_{\hat{\Pi}=\Pi_{+}}\right] \\
& = \mathbb{E}\left[\ov(\Pi_{+},\Pi^{\star})\right]  - \mathbb{E}\left[\ov(\Pi_{+},\Pi^{\star})\one_{\hat{\Pi}=\Pi_{-}}\right] \\
& = 1 -o(1).
\end{flalign*} and thus 
\begin{equation}\label{eq:conv_i}
\ov(\hat{\Pi},\Pi^{\star}) \underset{n \to \infty}{\overset{L^1}{\longrightarrow}} 1.
\end{equation} 

\proofstep{Second case $(ii)$.} If condition (\ref{microscopicregime}) is verified then the identification (\ref{identificationrang}) still holds and the proof of case $(i)$ adapts well. However, if (\ref{microscopicregime}) is not verified, we can still make a link with the toy model studied in Section \ref{toymodel}. Let's use a simple coupling argument: if $\xi = \omega(n^{-1/2-\alpha})$ for some $\alpha \geq 0$, let's take $\xi_1, \xi_2 >0$ such that $$ \xi^2 = \xi_1^2 + \xi_2^2 $$ and $$ n^{-7/6+\epsilon} \ll \xi_1 \ll n^{-1/2-\alpha}, $$ fixing for instance $\xi_1 = n^{-1}$. We will use the notation $\widetilde{v}_1$, now viewed as the leading eigenvector of the matrix 
\begin{equation*}
\widetilde{B} = A + \xi_1 H + \xi_2 \widetilde{H},
\end{equation*} where $\widetilde{H}$ is an independent copy of $H$. This has no consequence in terms of distribution : $(A,\widetilde{B})$ is still drawn under model \eqref{eq:GOEmodel_bis}. Let's denote $v'_1$ the leading eigenvector of $B_1=A+\xi_1 H$, chosen so that $\langle v_1,v'_1\rangle >0$. It is clear that condition \eqref{microscopicregime} holds for $\xi_1$. We have the following result, based on the invariance by rotation of the $\GOE$ distribution:

\begin{lemma}
	\label{stillthelink}
	We still have the following equality in distribution:
	\begin{equation*}
	\left(r_1(v_1),r_1(\widetilde{v}_1)\right) \overset{(d)}{=} \left(r_1(X),r_1(X+\mathbf{s}Z) \right),
	\end{equation*} 
	where $X$, $Z$ are two standard Gaussian vectors from the toy model, with w.h.p. 
	\begin{equation*}
	\mathbf{s} \geq \mathbf{s^1} \asymp \xi_1 n^{1/6}.
	\end{equation*} 
\end{lemma}
We refer to Appendix \ref{proof_stillthelink} for a proof. Since w.h.p. $\mathbf{s} \geq \mathbf{s^1}$ and $\mathbf{s^1} n \asymp \xi_1 n^{7/6} \to \infty$, we have for all $1 \leq i \leq n$,
\begin{flalign*}
\mathbb{P}\left(\Pi_{+}(i)=\Pi^{\star}(u)\right) &=  \mathbb{P}\left(\Pi_{+}(1)=\Pi^{\star}(1)\right) \\
& = \mathbb{E} \left[\iint dx dz E(x) E(z) \phi_{x,z}(n,\mathbf{s}) \one_{\mathbf{s} n  \to \infty} \right] + o(1)\\
& = \iint dx dz E(x) E(z) \mathbb{E} \left[\phi_{x,z}(n,\mathbf{s}) \one_{\mathbf{s} n  \to \infty} \right] + o(1).
\end{flalign*}
With the same arguments as in the case $(i)$, we show that 
$\phi_{x,z}(n,\mathbf{s}) \one_{\mathbf{s} n  \to \infty} \overset{L^1}{\longrightarrow} 0,$ which implies
\begin{equation*}
\mathbb{E}\left[\ov(\Pi_{+},\Pi^{\star})\right] \underset{n \to \infty}{\longrightarrow} 0,
\end{equation*}
hence $\ov(\Pi_{+},\Pi^{\star}) \underset{n \to \infty}{\overset{L^1}{\longrightarrow}} 0.$ The last step is to verify that the overlap achieved by $\Pi_{-}$ does not outperform that of $\Pi_{+}$. We prove the following Lemma in Appendix \ref{proof_lemma_Pi+caseii}:
\begin{lemma}
	\label{Pi+caseii}
	In the case $(ii)$, if $\langle v_1,v'_1 \rangle >0$, we also have
	\begin{equation*}
	\ov(\Pi_{-},\Pi^{\star}) \underset{n \to \infty}{\overset{L^1}{\longrightarrow}} 0.
	\end{equation*} 
\end{lemma} Lemma \ref{Pi+caseii} then gives
\begin{equation*}
\mathbb{E}\left[\ov(\hat{\Pi},\Pi^{\star})\right] \leq \mathbb{E}\left[\ov(\Pi_{+},\Pi^{\star})\right] + \mathbb{E}\left[\ov(\Pi_{-},\Pi^{\star})\right] \underset{n \to \infty}{\longrightarrow} 0,
\end{equation*}and thus 
\begin{equation}\label{eq:conv_ii}
\ov(\hat{\Pi},\Pi^{\star}) \underset{n \to \infty}{\overset{L^1}{\longrightarrow}} 0.
\end{equation} 

Of course, the convergences in \eqref{eq:conv_i} and \eqref{eq:conv_ii} also hold in probability, by Markov's inequality.
\end{proof}

\begin{subappendices}
\section{Additional proofs for Section \ref{linkGOEtoy}}
\addtocontents{toc}{\protect\setcounter{tocdepth}{0}}
\label{section3_add_proofs}
Throughout the proofs, all variables denoted by $C_i$ with $i=1,2,\ldots$ are unspecified, independent, positive constants.

\subsection{Proof of Proposition \ref{EIG1:prop:w'expansion}}
\label{appendix_proof_prop_w'}
\begin{proof}[Proof of Proposition \ref{EIG1:prop:w'expansion}]
	Let us establish a first inequality: since the $\GOE$ distribution is invariant by rotation (see e.g. \cite{Anderson09}), the random variables $\langle H v_j,v_i \rangle$ are zero-mean Gaussian, with variance $1/n$ of $i \neq j$ and $2/n$ if $i=j$. Hence, w.h.p.
	\begin{equation}
	\label{ineq:encadrementgaussiennes}
	\underset{1 \leq i,j \leq n}{\sup} \left|\langle H v_j,v_i \rangle \right|\leq C_1 \sqrt{\frac{\log n}{n}}.
	\end{equation} We will use the following short-hand notation for $1 \leq i,j \leq n$:
	\begin{equation*}
	m_{i,j} := \langle H v_j,v_i \rangle,
	\end{equation*} 
	The defining eigenvector equations projected on vectors $v_i$ write
	\begin{equation}
	\left \{
	\begin{array}{c @{=} c}
	\theta_i & \dfrac{\xi}{\lambda'_1 - \lambda_i} \sum_{j=1}^{n} \theta_j m_{i,j}, \\
	\lambda'_1 - \lambda_1 & \xi  \sum_{j=1}^{n} \theta_j m_{1,j}. \\
	\end{array}
	\right.
	\end{equation}
	In order to approximate the $\theta_i$ variables, we define the following iterative scheme:
	\begin{equation}
	\label{schemadepicard}
	\left \{
	\begin{array}{c @{\; = } c}
	\theta_i^k & \dfrac{\xi}{\lambda_1^{k-1} - \lambda_i} \sum_{j=1}^{n} \theta_j^{k-1} m_{i,j} ,\\
	\lambda^{k}_1 - \lambda_1 & \xi  \sum_{j=1}^{n} \theta_j^{k-1} m_{1,j} ,
	\end{array}
	\right.
	\end{equation} with initial conditions $\left(\theta_i^{0}\right)_{2 \leq i \leq n}=0$ and $\lambda_{1}^0 = \lambda_1$, and setting $\theta_1^{k}=1$ for all $k$. For $k \geq 1$, define
	\begin{equation*}
	\Delta_k := \sum_{i \geq 2} \left|\theta_i^k - \theta_i^{k-1}\right|,
	\end{equation*} 
	and for $k \geq 0$,
	\begin{equation*}
	S_k := \sum_{i \geq 1} \left|\theta_i^k \right|.
	\end{equation*} Recall that under assumption (\ref{microscopicregime}), there exists $\alpha>0$ such that $\xi = o \left(n^{-1/2-\alpha}\right)$. We define $\epsilon$ as follows: 
	\begin{equation*}
	\epsilon = \epsilon(n) = \sqrt{\xi n^{1/2+\alpha}}.
	\end{equation*} The idea is to show that the sequence $\left\{\Delta_k \right\}_{k \geq 1}$ decreases geometrically with $k$ at rate $\epsilon$. More specifically, we show the following result:
	\begin{lemma}
		\label{propagation_picard}
		With the same notations and under the assumption (\ref{microscopicregime}) of Proposition \ref{EIG1:prop:w'expansion}, one has w.h.p.
	\begin{itemize}
		\item[$(i)$] $\forall k \geq 1, \; \Delta_k \leq \Delta_1 \epsilon^{k-1}$, \\
		\item[$(ii)$] $\forall k \geq 0, \forall \, 2 \leq i \leq n, \; \left|\lambda_{1}^k - \lambda_i\right| \geq \frac{1}{2} \left|\lambda_{1} - \lambda_i\right| \left(1- \epsilon - \ldots - \epsilon^{k-1}\right)$, \\
		\item[$(iii)$] $\forall k \geq 0, \; S_k \leq 1 +(1+\ldots+\epsilon^{k-1}) \Delta_1$,\\
		\item[$(iv)$]  $\sum_{i=2}^{n} \left| \theta_i - \theta_i^{1} \right|^2 = o\left(\sum_{i=2}^{n} \left| \theta_i^{1} \right|^2 \right)$.
	\end{itemize}
	\end{lemma}
	This Lemma is proved in the next section. Equation $(iv)$ of Lemma \ref{propagation_picard} yields
	\begin{flalign*}
	w' & = v_1 + \sum_{i=2}^{n} \theta_i^1 v_i + \sum_{i=2}^{n} \left(\theta_i - \theta_i^1\right) v_i\\
	& = v_1 + \xi \sum_{i=2}^{n} \frac{\langle Hv_i,v_1 \rangle }{\lambda_1 - \lambda_i} v_i + o_{\mathbb{P}}\left(\xi \sum_{i=2}^{n} \frac{\langle Hv_i,v_1 \rangle }{\lambda_1 - \lambda_i} v_i  \right). 
	\end{flalign*}
\end{proof}

\subsection{Proof of Lemma \ref{propagation_picard}}
\begin{proof}[Proof of Lemma \ref{propagation_picard}]
In this proof we will use the same notations as defined in the proof of Proposition \ref{EIG1:prop:w'expansion}, and we make the assumption $(\ref{microscopicregime})$. We now state three technical lemmas controlling some statistics of eigenvalues in the $\GOE$ which are useful hereafter.

\begin{lemma}
	\label{lemma_sommevp1}
	W.h.p., for all $\delta>0$,
	\begin{equation}
	\label{sommevp1eq}
	\sum_{j = 2}^{n} \frac{1}{\lambda_1 - \lambda_j} \leq O\left( n ^{1+\delta}\right).
	\end{equation}
\end{lemma}

\begin{lemma}
	\label{lemma_sommevp2}
	We have
	\begin{equation}
	\label{sommevp2eq}
	\; \sum_{j = 2}^{n} \frac{1}{\left(\lambda_1 - \lambda_j \right)^2} \asymp n^{4/3}.
	\end{equation} 
\end{lemma}

\begin{lemma}
	\label{lemma_controletrousp}
	For any $C>0$, w.h.p.
	\begin{equation}
	\label{controletrouspeq}
	\lambda_1 - \lambda_2 \geq n^{-2/3} \left(\log n\right)^{-C \log \log n}.
	\end{equation}
\end{lemma}

Proofs of these three Lemmas can be found in the next sections. We will work under the event (that occurs w.h.p.) on which the equations (\ref{sommevp1eq}), (\ref{sommevp2eq}), (\ref{controletrouspeq}), (\ref{lemma_concentrationksieq}) and (\ref{ineq:encadrementgaussiennes}) are satisfied. We show the following inequalities:
	
	\begin{itemize}
		\item[$(i)$] $\forall k \geq 1, \; \Delta_k \leq \Delta_1 \epsilon^{k-1}$, \\
		\item[$(ii)$] $\forall k \geq 0, \forall \, 2 \leq i \leq n, \; \left|\lambda_{1}^k - \lambda_i\right| \geq \frac{1}{2} \left|\lambda_{1} - \lambda_i\right| \left(1- \epsilon - \ldots - \epsilon^{k-1}\right)$, \\
		\item[$(iii)$] $\forall k \geq 0, \; S_k \leq 1 + (1+\ldots+\epsilon^{k-1}) \Delta_1$.\\
	\end{itemize}
	
	Recall that $\epsilon$ is given by
	\begin{equation*}
	\epsilon = \epsilon(n) = \sqrt{\xi n^{1/2+\alpha}}.
	\end{equation*}
	
	We will denote by $f_{i}(n)$, with $i$ an integer, functions as defined in Lemma \ref{lemma_sommevp2}. All the following inequality will be valid for $n$ large enough (uniformly in $i$ and in $k$).\\
	
	\proofstep{Step 1: propagation of the first equation.}
	Let $k \geq 3$. We work by induction, assuming that $(i)$, $(ii)$ and $(iii)$ are verified until $k-1$. 
	\begin{flalign*}
	\left| \theta_i^k - \theta_i^{ k-1} \right| & \leq \left| \frac{\xi}{\lambda_1^{k-1}-\lambda_i} \sum_{j=2}^{n} \left(\theta_j^{k-1} - \theta_j^{k-2}\right)m_{i,j} \right|  + \left| \frac{\xi \left(\lambda_1^{k-2}-\lambda_1^{k-1}\right)}{\left(\lambda_1^{k-1}-\lambda_i\right)\left(\lambda_1^{k-2}-\lambda_i\right)} \sum_{j=1}^{n} \theta_j^{k-2} m_{i,j} \right|\\
	& \leq \frac{ \xi}{\left|\lambda_1^{k-1}-\lambda_i\right|} C_1 \sqrt{\frac{\log n}{n}} \Delta_{k-1} + \xi C_1 \sqrt{\frac{\log n}{n}} S_{k-2} \frac{\left|\lambda_1^{k-2}-\lambda_1^{k-1}\right|}{\left|\lambda_1^{k-1}-\lambda_i\right|\left|\lambda_1^{k-2}-\lambda_i\right|} \\ 
	& \overset{(a)}{\leq}  \xi \frac{3}{\left|\lambda_1-\lambda_i\right|} C_1 \sqrt{\frac{\log n}{n}} \Delta_{k-1} + \xi C_1 \sqrt{\frac{\log n}{n}} S_{k-2} \frac{9 \left|\lambda_1^{k-2}-\lambda_1^{k-1}\right|}{\left|\lambda_1-\lambda_i\right|^2} \\ 
	& \overset{(b)}{\leq}  \xi \frac{3}{\left|\lambda_1-\lambda_i\right|} C_1 \sqrt{\frac{\log n}{n}} \Delta_{k-1} + \xi C_1 \sqrt{\frac{\log n}{n}} 2 \frac{9 \left|\lambda_1^{k-2}-\lambda_1^{k-1}\right|}{\left|\lambda_1-\lambda_i\right|^2} .
	\end{flalign*} We applied $(ii)$ to $k-1, k-2$ in $(a)$ and $(iii)$ to $k-2$ in $(b)$.
	Note that
	\begin{align*}
	\left|\lambda_1^{k-2}-\lambda_1^{k-1}\right| & = \left|\xi \sum_{j=1}^{n} \left(\theta_j^{k-2}-\theta_j^{k-3}\right) m_{i,j}\right| \leq \xi C_1 \sqrt{\frac{\log n}{n}} \Delta_{k-2},
	\end{align*} which yields the inequality:
	\begin{align*}
	\left| \theta_i^k - \theta_i^{ k-1} \right| & \leq  \frac{\xi}{\left|\lambda_1-\lambda_i\right|} f_1(n) n^{-1/2} \Delta_{k-1} + \frac{\xi^2 }{\left|\lambda_1-\lambda_i\right|^2} f_2(n) n^{-1} \Delta_{k-2}.
	\end{align*}
 	We choose $\delta $ such that $0 < \delta < \alpha$ (where $\alpha$ is fixed by $(\ref{microscopicregime})$), and we sum from $i=2$ to $n$:
	\begin{align*}
	\Delta_k & \leq \xi f_1(n) n^{1/2 + \delta} \Delta_{k-1} + \xi^2 f_3(n) n^{1/3} \Delta_{k-2} \\
	& \overset{(a)}{\leq} o(\epsilon) \epsilon^{k-2} \Delta_1 + o(\epsilon^2) \epsilon^{k-3} \Delta_1 \\
	& \leq  \epsilon^{k-1} \Delta_1.
	\end{align*} We used $\xi f_1(n) n^{1/2 + \delta} =o(\epsilon)$, $\xi^2 f_3(n) n^{1/3} = o(\epsilon^2)$ and we applied $(i)$ to $k-1$ and $k-2$ in $(a)$.\\
	
	\proofstep{Step 2: propagation of the second equation.}	
	Let $k \geq 2$, and $0 < \delta < \alpha$. We work by induction, assuming that $(i)$, $(ii)$ and $(iii)$ are verified until $k-1$. 
	\begin{align*}
	\left| \lambda_1^k - \lambda_1^{ k-1} \right| & \leq \xi f_1(n) n^{-1/2} \Delta_{k-1} \\ 
	& \overset{(a)}{\leq} \xi f_1(n) n^{-1/2} {\epsilon^{k-2} \Delta_1}\\
	& \leq n^{-2/3} (\log n)^{-C \log \log n} \epsilon^{k-2} \Delta_1\\
	& \leq \frac{\lambda_1-\lambda_2}{2} \epsilon^{k-2} \Delta_1\\
	& \leq \frac{\lambda_1-\lambda_i}{2} \epsilon^{k-2} \Delta_1.
	\end{align*} We applied $(i)$ to $k-1$ in $(a)$. Note that
	\begin{align*}
	\Delta_1 & = \sum_{j=2}^{n} \frac{\xi}{\lambda_1 - \lambda_i} \left| m_{i,1} \right| \leq \xi f_1(n) n^{1/2+\delta} \leq o(\epsilon).
	\end{align*}
	Applying $(ii)$ to $k-1$, we get
	\begin{align*}
	\left| \lambda_1^k - \lambda_i \right| & \geq \left| \lambda_1 - \lambda_1^{k-1} \right| - \left| \lambda_1^k - \lambda_1^{k-1} \right| \\ 
	& \geq \frac{\lambda_1-\lambda_i}{2} \left(1-\epsilon-\ldots-\epsilon^{k-2}\right) - \frac{\lambda_1-\lambda_i}{2} \epsilon^{k-1}\\
	& \geq \frac{\lambda_1-\lambda_i}{2} \left(1-\epsilon-\ldots-\epsilon^{k-1}\right).
	\end{align*}
	
	\proofstep{Step 3: propagation of the third equation.} Let $k \geq 1$. Here again, we work by induction, assuming that $(i)$, $(ii)$ and $(iii)$ are verified until $k-1$. 
	\begin{align*}
	S_k &= 1 + \sum_{j=2}^{n} \left|\theta^k_j\right|\\
	& \leq 1 + \Delta_k + S_{k-1} - 1\\
	& \overset{(a)}{\leq} \epsilon^{k-1} \Delta_1 + 1 + \left(1+\ldots+\epsilon^{k-2}\right)\Delta_1\\
	& \leq 1 + \left(1+\epsilon+\ldots+\epsilon^{k-1}\right) \Delta_1.
	\end{align*} We applied $(i)$ to $k$ and $(iii)$ to $k-1$ in $(a)$.
	
	\proofstep{Step 4: Proof of $(i)$ for $k=1,2$, $(ii)$ for $k=0,1$ and $(iii)$ for $k=0,1$.}
	The equation $(i)$ for $k=1$ is obvious. For $k=2$ :
	\begin{align*}
	\left| \theta_i^2 - \theta_i^1 \right| & \leq \left| \frac{\xi}{\lambda_1^1-\lambda_i} \sum_{j=2}^{n} \left(\theta_j^{1} - \theta_j^{0}\right)m_{i,j} \right|  + \left| \frac{\xi \left(\lambda_1^{0}-\lambda_1^{1}\right)}{\left(\lambda_1^{1}-\lambda_i\right)\left(\lambda_1^{0}-\lambda_i\right)} \sum_{j=1}^{n} \theta_j^{0} m_{i,j} \right|.
	\end{align*}
	We have
	\begin{align*}
	\left| \lambda_1^1 - \lambda_i \right| & \geq \left| \lambda_1 - \lambda_i \right| - \left| \lambda_1 - \lambda_1^1 \right| 
	 \geq \left| \lambda_1 - \lambda_i \right| - \xi \left|m_{1,1}\right|\\
	& \geq \left| \lambda_1 - \lambda_i \right| - \frac{1}{2} \left| \lambda_1 - \lambda_2 \right|
	\geq \frac{1}{2} \left| \lambda_1 - \lambda_i \right|,
	\end{align*}
	which shows $(ii)$ for $k=0,1$. Thus, for $0<\delta < \alpha$:
	\begin{align*}
	\left| \theta_i^2 - \theta_i^1 \right| &\leq \frac{2 \xi}{\lambda_1-\lambda_i} C_1 \sqrt{\frac{\log n}{n}} \Delta_1 + \frac{4 \xi}{\left(\lambda_{1}-\lambda_i\right)^2 } \xi \left|m_{1,1}\right| \left|m_{i,1}\right|,
	\end{align*} and
	\begin{align*}
	\Delta_2 & \leq  \xi f_1(n) n^{1/2 + \delta} \Delta_1 + 4 \xi \left|m_{1,1}\right| \sum_{i=2}^{n} \frac{\xi \left|m_{i,1}\right|}{\left(\lambda_1 - \lambda_i\right)^2}\\
	& \leq \xi f_1(n) n^{1/2 + \delta} \Delta_1 + 4 \xi f_4(n) n^{-1/2} n^{2/3}  \sum_{i=2}^{n} \frac{\xi \left|m_{i,1}\right|}{\left(\lambda_1 - \lambda_i\right)}\\
	& \leq \xi f_1(n) n^{1/2 + \delta} \Delta_1 + 4 \xi f_4(n) n^{1/6}  \Delta_1\\
	& \leq \epsilon \Delta_1.
	\end{align*}
	The proof of $(iii)$ for $k=0,1$ is obvious.\\
	
\proofstep{Step 5: Proof of equation $(iv)$.}	
Let $k \geq 2$ and $2 \leq i \leq n$. In the same way as in Step 1, we have
\begin{flalign*}
\left| \theta_i^k - \theta_i^{ k-1} \right| & \leq \frac{2 \xi C_1}{\lambda_1-\lambda_i} \sqrt{\frac{\log n}{n}} \epsilon^{k-2} \Delta_1 + \frac{8 \xi^2 C_1^2}{\left(\lambda_1-\lambda_i\right)^2} \frac{\log n}{n} \epsilon^{(k-3)_+}\Delta_1.
\end{flalign*}
In the right-hand term, the ratio of the second term on the first one is smaller that
\begin{equation*}
\frac{4 \xi C_1}{\lambda_1 - \lambda_i} \sqrt{\frac{\log n}{n}} \epsilon^{-1} \leq \xi n^{1/6} f(n) \epsilon^{-1} \leq \epsilon \to 0,
\end{equation*}
using Lemma \ref{lemma_controletrousp}, with $f \in \mathcal{F}$. It follows that for $n$ big enough (uniformly in $k$ and $i$) one has
\begin{equation}
\label{ecartik}
\left| \theta_i^k - \theta_i^{ k-1} \right| \leq \frac{\xi f(n) }{\lambda_1-\lambda_i} n^{-1/2} \epsilon^{k-2} \Delta_1.
\end{equation}
Equation \eqref{ecartik} shows that the scheme \eqref{schemadepicard} converges, and that the limits are indeed the solutions $\theta_1=1, \theta_2, \ldots, \theta_n$ of the fixed-point equations. By a simple summation of (\ref{ecartik}) over $k \geq 2$, applying Lemma \ref{lemma_sommevp1} and inequality \eqref{ineq:encadrementgaussiennes} we have
\begin{flalign*}
\left| \theta_i - \theta_i^{1} \right| & \leq \frac{2 \xi f(n) }{\lambda_1-\lambda_i} n^{-1/2} \Delta_1 \leq \frac{2 \xi^2 f(n) }{\lambda_1-\lambda_i} n^{\delta},
\end{flalign*}
where $\delta>0$ is a positive quantity of Lemma \ref{lemma_sommevp1} specified later. Using Lemma \ref{lemma_sommevp2} one has the following control
\begin{equation*}
\sum_{i=2}^{n} \left| \theta_i - \theta_i^{1} \right|^2 \leq 4 \xi^4 n^{2 \delta} f(n) n^{4/3}.
\end{equation*}
Moreover, Lemma \ref{lemma_concentrationksi} shows that
\begin{equation*}
\sum_{i=2}^{n} \left| \theta_i^{1} \right|^2 \asymp \xi^2 n^{1/3} \geq g(n)^{-1} \xi^2 n^{1/3},
\end{equation*}
where $g$ is another function in $\mathcal{F}$. This yields
\begin{equation*}
\sum_{i=2}^{n} \left| \theta_i - \theta_i^{1} \right|^2 \leq \sum_{i=2}^{n} \left| \theta_i^{1} \right|^2 4 \xi^2 n^{2 \delta+1} {f(n)}{g(n)}.
\end{equation*}
The proof is completed by taking $\delta = \alpha/2$ and applying (\ref{microscopicregime}).
\end{proof}

\subsection{Proof of Lemma \ref{lemma_controletrousp}}\label{proof_lemma_controletrousp}
\begin{proof}[Proof of Lemma \ref{lemma_controletrousp}]
	This lemma provides a control of the spectral gap $\lambda_1 - \lambda_2$. Given a good rescaling (in $n^{2/3}$), the asymptotic joint law of the eigenvalues in the edge has been investigated in a great amount of research work, for Gaussian ensembles, and for more general Wigner matrices. The $\GOE$ case has been mostly studied by Tracy, Widom, and Forrester among many others; in \cite{Forrester93} and \cite{Tracy98}, the convergence of the joint distribution of the first $k$ eigenvalues towards a density distribution is established:
	\begin{proposition}[\cite{Forrester93}, \cite{Tracy98}]
		For a given $k\geq 1$, and all $s_1, \ldots, s_k$ real numbers,
		\begin{equation}
		\label{loilimitetrousp}
		\mathbb{P}\left(n^{2/3} \left(\lambda_1 - 2\right) \leq s_1, \ldots, n^{2/3} \left(\lambda_k - 2\right) \leq s_k \right) \underset{n \to \infty}{\longrightarrow} \mathcal{F}_{1,k}(s_1, \ldots, s_k),
		\end{equation}
	\end{proposition} where the $\mathcal{F}_{1,k}$ are continuous and can be expressed as solutions of non linear PDEs. Thus the re-scaled spectral gap $n^{2/3}\left(\lambda_{1}-\lambda_2\right)$ has a limit probability density law supported by $\mathbb{R_+}$, which implies that
	\begin{equation*}
	\mathbb{P}\left(n^{2/3} \left(\lambda_1 - \lambda_2\right) \geq \left(\log n\right)^{-C \log \log n} \right) \underset{n \to \infty}{\longrightarrow} 1.
	\end{equation*}
	Of course, the choice of the function $n \mapsto \left(\log n\right)^{-C \log \log n}$ is here arbitrary and the result is also true for any function tending to 0.
\end{proof}

\subsection{Proof of Lemma \ref{lemma_sommevp2}}\label{proof_lemma_sommevp2}
\begin{proof}[Proof of Lemma \ref{lemma_sommevp2}]
	This result needs an understanding of the behavior of the spectral gaps of matrix $A$, in the bulk and in the edges (left and right). The eigenvalues in the \textit{edge} correspond to indices $i$ such that $i=o(n)$ (left) or $i = n-o(n)$ (right). Eigenvalues in the \textit{bulk} are the remaining eigenvalues. For this, we use a result of rigidity of eigenvalues, due to L. Erdös et al. \cite{Erdos10}, which consists in a control of the probability of the gap between the eigenvales of  $A$ and the typical eigenvalues $\gamma_j$ of the semi-circle law, defined as follows
	\begin{equation}
	\forall i \in \left\lbrace 1, \dots, n\right\rbrace , \; \frac{1}{2 \pi} \int_{-2}^{\gamma_j} \sqrt{4-x^2} dx =1- \frac{j}{n}.
	\end{equation}
	
	\begin{proposition}[\cite{Erdos10}]
		For some positive constants $C_5>0$ and $C_6>0$, for $n$ large enough,
		\begin{multline}
		\label{erdos}
		\mathbb{P}\left( \exists j \in \left\lbrace 1, \dots, n\right\rbrace \, | \, \left|\lambda_j - \gamma_j \right| \geq \left(\log n\right)^{C_5 \log \log n} \left(\min \left(j, n+1-j\right)\right)^{-1/3} n^{-2/3} \right) \\ \leq C_5 \exp \left(- \left(\log n\right)^{C_6 \log \log n} \right).
		\end{multline}
	\end{proposition}	
	\begin{remark}
		Another similar result that goes in the same direction for the $\GOE$ is already known: it has been shown by O'Rourke in \cite{ORourke10} that the variables $\lambda_i - \gamma_i$ behave as Gaussian variables when $n \to \infty$. However, the rigidity result in \eqref{erdos} obtained in \cite{Erdos10} can apply in more general models. This quantitative probabilistic statement was not previously known even for the $\GOE$ case.
	\end{remark} 
	\begin{remark}
		Let us note that one of the assumptions made in \cite{Erdos10} is that variances of each column sum to 1, which is not directly the case in our model \eqref{eq:GOEmodel_bis}. Nevertheless, one may use (\ref{erdos}) for the re-scaled matrix $\tilde{A} := A \left(1+\frac{1}{n}\right)^{-1/2}$, then easily check that there is a possible step back to $A$: $|\lambda_j - \gamma_j|\leq \left|\lambda_j \left(1+\frac{1}{n}\right)^{-1/2}- \gamma_j \right|+ n^{-1} + o(n^{-1})$, and $ n^{-1} + o(n^{-1}) \leq 2\left(\min \left(j, n+1-j\right)\right)^{-1/3} n^{-2/3}$ for $n$ big enough. Tolerating a slight increase of the constant $C_5$, the result (\ref{erdos}) is thus valid in the $\GOE$.
	\end{remark}
	Let us now compute an asymptotic expansion of $\gamma_j$ in the right edge, which is for  $j = o(n)$. Define
	\begin{equation}
	\label{G}
	G(x) := \frac{1}{2 \pi} \int_{-2}^{x} \sqrt{4-t^2} dt = \frac{x\sqrt{4-x^2}+4 \arcsin(x/2)}{4 \pi}+\frac{1}{2} , 
	\end{equation}
	for all $x \in [-2,2]$. We have $\gamma_j = G^{-1}(1-j/n) = - G^{-1}(j/n)$, observing that the integrand in (\ref{G}) is an even function. We get the following expansion when $x \to -2$,
	\begin{equation*}
	G(x) \underset{x \to -2}{=} \frac{2(x+2)^{3/2}}{3 \pi} + o\left((x+2)^{3/2}\right)
	\end{equation*} which implies that
	\begin{equation*}
	G^{-1}(y) \underset{y \to 0}{=} -2 + \left(\frac{3 \pi y}{2}\right)^{2/3} + o\left(y^{2/3}\right), 
	\end{equation*} hence
	\begin{equation}
	\label{gammaedge}
	\gamma_j \underset{j/n \to 0}{=} 2 - \left(\frac{3 \pi j}{2 n}\right)^{2/3} + o\left((j/n)^{2/3}\right).
	\end{equation} 
	\begin{remark}
		One can observe the coherence of this result that arises naturally in \cite{ORourke10} as the expectation of the eigenvalues in the edge.
	\end{remark}
	
	Let $\epsilon>0$, to be specified later. To establish our result we will split the variables $j$ in three sets:
	\begin{align*}
	A_1 &:= \left\lbrace 2 \leq j \leq \left(\log n\right)^{(C_5+1)\log \log n} \right\rbrace \; \mbox{(a small part of the right edge)},\\ A_2 &:= \left\lbrace \left(\log n\right)^{(C_5+1)\log \log n} < j \leq n^{1-\epsilon} \right\rbrace \;  \mbox{(a larger part of the right edge)}, \\ A_3 &:= \left\lbrace n^{1-\epsilon} < j \leq n \right\rbrace  \; \mbox{(everything else)}.
	\end{align*} We show that the sum over $A_1$ is the major contribution in (\ref{sommevp2eq}). The split in the right edge in $A_1$ and $A_2$ is driven by the error term of (\ref{erdos}): this term is small compared to $\gamma_j$ if and only if $\left(\log n\right)^{C_5 \log \log n} = o(j)$.\\
	
	\proofstep{Step 1: estimation of the sum over $A_1$.}
	According to (\ref{erdos}) and Lemma \ref{lemma_controletrousp}, w.h.p.
	\begin{equation*}
	n^{-4/3} \left(\log n\right)^{-C_6 \log \log n}  \leq \left(\lambda_1 - \lambda_2 \right)^2 \leq C_7 n^{-4/3} \left( \log n \right)^{C_6 \log \log n},
	\end{equation*} where $C_6, C_7$ are positive constants. Hence, w.h.p.
	\begin{flalign*}
	\label{a1}
	\frac{n^{4/3}}{C_7 \left( \log n \right)^{C_6 \log \log n}} &\leq \sum_{j \in A_1} \frac{1}{\left(\lambda_1 - \lambda_j\right)^2} \\ & \leq \sum_{j \in A_1} \frac{1}{\left(\lambda_1- \lambda_2\right)^2}\\ & \leq n^{4/3} \left( \log n \right)^{(C_5+C_6+1) \log \log n}. 
	\end{flalign*}
	
	\proofstep{Step 2: estimation of the sum over $A_2$.} Let us show that the sum over $A_2$ is asymptotically small compared to the sum over $A_1$: using (\ref{erdos}) and (\ref{gammaedge}), we know that there exists $C_8>0$ such that for all $j \in A_2$, w.h.p.
	\begin{equation*}
	\lambda_j = 2 - C_8 \left(\frac{j}{n}\right)^{2/3} + o\left((j/n)^{2/3}\right),
	\end{equation*} and we know furthermore (se e.g. \cite{Anderson09}) that w.h.p.
	\begin{equation}
	\label{lambda1}
	\lambda_1 = 2 + o\left((j/n)^{2/3}\right), \forall j \in A_2
	\end{equation} hence w.h.p.
	\begin{flalign*}
	\label{a2}
	\sum_{j \in A_2} \frac{1}{\left(\lambda_1 - \lambda_j\right)^2} &= n^{4/3} \sum_{j \in A_2} \frac{1}{C_9 j^{4/3}(1+o(1))}\\& =  n^{4/3} (1+o(1)) \sum_{j \in A_2} \frac{1}{C_9 j^{4/3}} = o\left(n^{4/3}\right),
	\end{flalign*}
	using in the last line the fact that the Riemann's series $\sum j^{-4/3}$ converges.\\
	
	\proofstep{Step 3: estimation of the sum under $A_3$.} With the previous results (\ref{erdos}), (\ref{gammaedge}) and (\ref{lambda1}), assuming that $\epsilon<1$, we get w.h.p.
	\begin{equation*}
	\lambda_1 - \lambda_{n^{1-\epsilon}} = C_8 n^{-2\epsilon/3} + O\left(n^{-2\epsilon/3}\right),
	\end{equation*} which gives w.h.p. the following control
	\begin{flalign*}
	\label{a3}
	\sum_{j \in A_3} \frac{1}{\left(\lambda_1- \lambda_j\right)^2} &\leq \left(n-n^{1-\epsilon}\right) \frac{1}{\left(\lambda_1 - \lambda_{n^{1-\epsilon}}\right)^2}\\& = \left(n-n^{1-\epsilon}\right) \frac{n^{4\epsilon/3}}{C_9 (1+o(1))} = O\left(n^{1+4\epsilon/3}\right) = o\left(n^{4/3}\right),&&
	\end{flalign*}
	as long as $\epsilon < 1/4$. Taking such a $\epsilon$, these three controls end the proof.
\end{proof}

\subsection{Proof of Lemma \ref{lemma_sommevp1}} \label{proof_lemma_sommevp1}
\begin{proof}[Proof of Lemma \ref{lemma_sommevp1}]
	We follow the same steps as in the proof of Lemma \ref{lemma_sommevp2}. Let's take $\delta>0$. We split the $j$ variables in three sets:
	\begin{align*}
	A_1 &:= \left\lbrace 2 \leq j \leq n^{1/3} \right\rbrace,\\ A_2 &:=  \left\lbrace n^{1/3} < j \leq n^{1-\delta} \right\rbrace,\\ A_3 &:= \left\lbrace n^{1-\delta} < j \leq n \right\rbrace.
	\end{align*}
	We use Lemma \ref{lemma_controletrousp} to obtain the following control w.h.p.
	\begin{flalign*}
	\sum_{j \in A_1} \frac{1}{\lambda_1 - \lambda_j} \leq n^{1/3} n^{2/3} \left(\log n\right)^{C_5 \log \log n} = O(n^{1+\delta}).
	\end{flalign*}
	Similarly, for $A_2$
	\begin{flalign*}
	\sum_{j \in A_2} \frac{1}{\lambda_1 - \lambda_j} &\leq \sum_{j \in A_2} \frac{1}{o(n^{-2/3}) + C_8(j/n)^{2/3} + O\left(\left(\log n\right)^{C_5 \log \log n}n^{-2/3}j^{-1/3}\right)} \\ &= n^{2/3} \sum_{j \in A_2} \frac{1}{o(1)+C_8j^{2/3}} \leq C_{10} n^{2/3} n^{(1-\delta)/3} \leq O(n^{1+\delta}).
	\end{flalign*}
	Finally, using Cauchy–Schwarz inequality
	\begin{flalign*}
	\sum_{j \in A_3} \frac{1}{\lambda_1 - \lambda_j} \leq  \sqrt{n} \left(\sum_{j \in A_3} \frac{1}{\left(\lambda_1 - \lambda_j\right)^2}\right)^{1/2} \leq  \sqrt{n} O(n^{1/2+2\delta/3}) = O(n^{1+\delta}).
	\end{flalign*}
\end{proof}

\subsection{Proof of Lemma \ref{lemma_concentrationksi}} \label{proof_lemma_concentrationksi}
\begin{proof}[Proof of Lemma \ref{lemma_concentrationksi}]
We show that w.h.p.
\begin{equation}
\label{concentration}
\sum_{i=2}^{n} \frac{\langle Hv_i,v_1 \rangle ^2 }{\left(\lambda_1 - \lambda_i\right)^2} - \frac{1}{n} \sum_{i=2}^{n} \frac{1}{\left(\lambda_1 - \lambda_i\right)^2}
= o\left(\frac{1}{n} \sum_{i=2}^{n} \frac{1}{\left(\lambda_1 - \lambda_i\right)^2}\right)
\end{equation}
Let us recall that $H$ is drawn according to the $\GOE$, hence its law is invariant by rotation. This implies that the $\langle Hv_i,v_1 \rangle$ are independent variables with variance $1/n$, independent of $\lambda_1, \ldots, \lambda_n$. Define
\begin{equation*}
M_n := \sum_{i=2}^{n} \frac{\langle Hv_i,v_1 \rangle ^2 - 1/n }{\left(\lambda_1 - \lambda_i\right)^2}.
\end{equation*}
Computing the second moment of $M_n$, we get
\begin{flalign*}
\mathbb{E}\left[M_n ^2 | \lambda_1, \ldots, \lambda_n \right]& = \mathrm{Var}(M_n | \lambda_1, \ldots, \lambda_n) = \frac{1}{n^4} \sum_{i=2}^{n} \frac{2 }{\left(\lambda_1 - \lambda_i\right)^4}. 
\end{flalign*}
Adapting the proof of Lemma \ref{lemma_sommevp2}, following the same steps, one can also show that w.h.p.
\begin{equation}
\label{controlevp4}
\sum_{i=2}^{n} \frac{1}{\left(\lambda_1 - \lambda_i\right)^4} \asymp n^{8/3}.
\end{equation}
Let $\epsilon = \epsilon(n)>0$ to be specified later. By Markov's inequality
\begin{flalign*}
\mathbb{P}\left(\left|M_n\right| \geq \frac{\epsilon}{n} \sum_{i=2}^{n} \frac{1}{\left(\lambda_1 - \lambda_i\right)^2} | \lambda_1, \ldots, \lambda_n \right)& \leq \frac{n^2}{\epsilon^2} \frac{\mathbb{E}\left[M_n ^2 | \lambda_1, \ldots, \lambda_n \right]}{\left(\sum_{i=2}^{n} \frac{1}{\left(\lambda_1 - \lambda_i\right)^2}\right)^2} \\
&  \asymp \frac{1}{\epsilon^2 n^2},
\end{flalign*}
by Lemma \ref{lemma_sommevp2} and equation (\ref{controlevp4}). Taking e.g. $\epsilon(n)=n^{-1/2}$ concludes the proof.
\end{proof}
\addtocontents{toc}{\protect\setcounter{tocdepth}{2}}

\section{Additional proofs for Sections \ref{toymodel} \& \ref{EIG1threshold}}
\addtocontents{toc}{\protect\setcounter{tocdepth}{0}}
\label{section45_add_proofs}
\subsection{Proof of Lemma \ref{fibo}}\label{proof_lemma_fibo}
\begin{proof}[Proof of Lemma \ref{fibo}]
We fix $\alpha>0$ and we want to prove
\begin{equation}
\label{sum_fibop}
\sum_{k=0}^{\lfloor(n-1)/2\rfloor}  \binom{n-1-k}{k} \alpha^k = \frac{1}{\sqrt{1+4 \alpha}} \left[\left(\frac{1+\sqrt{1+4 \alpha}}{2}\right)^n - \left(\frac{1-\sqrt{1+4 \alpha}}{2}\right)^n\right].
\end{equation}We denote in the following $\phi_+ := \frac{1+\sqrt{1+4 \alpha}}{2}$ and $\phi_- := \frac{1-\sqrt{1+4 \alpha}}{2}$, and for all $n \geq 1$:
\begin{equation*}
u_n = u_n(\alpha) := \sum_{k=0}^{\lfloor(n-1)/2\rfloor}  \binom{n-1-k}{k} \alpha^k.
\end{equation*}
We clearly have $u_n(\alpha) \leq \left(1+\alpha\right)^n$. For all $t>0$ small enough (e.g. $t < \frac{1}{1+\alpha}$), define
\begin{equation*}
f(t) := \sum_{n =1}^{\infty}  u_n t^n.
\end{equation*}
On one hand,
\begin{flalign*}
\frac{t}{1-t-\alpha t^2} & = t \sum_{m=0}^{\infty} (t+\alpha t^2)^m = \sum_{m=0}^{\infty} \sum_{l=0}^{m} \binom{m}{l} \alpha ^l t^{l + m + 1}\\
& = \sum_{n=1}^{\infty} \left(\sum_{\substack{0 \leq l \leq m \\ l+m = n-1}} \binom{m}{l} \alpha^l \right) t^n = \sum_{n=1}^{\infty} u_n t^n = f(t). 
\end{flalign*}
On the other hand,
\begin{flalign*}
\frac{t}{1-t-\alpha t^2} & = \frac{t}{\left(1-\phi_{-}t\right)\left(1-\phi_{+}t\right)} =  \frac{1}{\phi_{+}-\phi_{-}}\left(\frac{1}{1-\phi_{+}t}-\frac{1}{1-\phi_{-}t}\right)\\
& =  \frac{1}{\sqrt{1+4\alpha}} \sum_{n=1}^{\infty}\left(\phi_{+}^n - \phi_{-}^n\right) t^n. 
\end{flalign*}
This proves $(\ref{sum_fibop})$.
\end{proof}

\subsection{Proof of Lemma \ref{stillthelink}}
\begin{proof}[Proof of Lemma \ref{stillthelink}]\label{proof_stillthelink}
Let us represent the situation in the plane spanned by $v_1$ and $v'_1$, as shown on Figure \ref{fig:EIG1:imagev}.
\begin{figure}[h]
	\centering
	\begin{picture}(150,150)

\put(75,75){\circle*{4}}
\put(0,75){\vector(1,0){150}}
\put(157,75){\makebox(0,0){$b$}}
\put(75,0){\vector(0,1){150}}
\put(75,157){\makebox(0,0){$v'_1$}}

\put(75,75){\vector(1,3){23}}
\put(100,149){\makebox(0,0){$v_1$}}

\put(75,75){\vector(-1,2){23}}
\put(35,128){\makebox(0,0){$\mathrm{proj}_{\mathcal{P}}(\widetilde{v'_1})$}}

\multiput(53,75)(0,10){5}{\line(0,1){5}}
\multiput(75,120)(-9,0){3}{\line(-1,0){4}}

\put(42,97){\makebox(0,0){$\sqrt{p}$}}

\end{picture}
	\caption{Orthogonal projection of $\widetilde{v}_1$ on $\mathcal{P} := \mathrm{span}(v'_1, v_1)$.}
	\label{fig:EIG1:imagev}
\end{figure}
Since $\widetilde{v}_1$ is taken such that $\langle v_1,\widetilde{v}_1 \rangle >0$ and $\xi_1$ satisfies (\ref{microscopicregime}), we have  $\langle \widetilde{v}_1,v'_1 \rangle >0$ for $n$ large enough by Proposition \ref{EIG1:prop:gaussian_decomp}. Let $p := \langle \widetilde{v}_1,v'_1 \rangle^2$ and $\widetilde{w} := \widetilde{v}_1 - \sqrt{p} v'_1 \in \left(v'_1\right)^{\perp}$. By invariance by rotation we can obtain that $\frac{\widetilde{w}}{\| \widetilde{w}\|}=\frac{\widetilde{w}}{\sqrt{1-p}}$ is uniformly distributed on the unit sphere $\mathbb{S}^{n-2}$ of $\left(v'_1\right)^{\perp}$, and independent of $p, v_1$ and $v'_1$. Hence 
\begin{equation*}
\langle b,\widetilde{v}_1 \rangle = \langle b, \widetilde{w} \rangle \overset{(d)}{=} \sqrt{1-p} \cdot \frac{\widetilde{Z}_1}{\sqrt{\sum_{i=1}^{n-1} \left(\widetilde{Z}_i\right)^2}},
\end{equation*} where the $\widetilde{Z}_i$ are independent Gaussian standard variables, independent from everything else.
According to Section \ref{linkGOEtoy} we know that $1 - \langle v_1,v'_1 \rangle \asymp \xi_1^2 n^{1/3} $ and thus $\langle v_1,b \rangle \asymp \xi_1 n^{1/6}$. This yields, for $n$ large enough, w.h.p,
\begin{flalign*}
0 < \langle \widetilde{v}_1, v_1 \rangle & \leq \sqrt{p} \langle v_1, v'_1 \rangle + \sqrt{\frac{1-p}{n}} \widetilde{Z}_1 \xi_1 n^{1/6} f(n)\\
& \leq \sqrt{p} \langle v_1, v'_1 \rangle + \sqrt{1-p} n^{-4/3} g(n) \\
& \leq \max \left(\sqrt{p},\sqrt{1-p}\right) \langle v_1, v'_1 \rangle \\
& \leq \langle v_1, v'_1 \rangle,
\end{flalign*} where $f$ and $g$ are two functions as defined in Lemma \ref{lemma_sommevp2}. From this point one can still make the link with the toy model, as done in the beginning of Section \ref{toymodel}. By invariance by rotation, letting $t := \widetilde{v}_1 - \langle \widetilde{v}_1,v_1 \rangle v_1$, we know that $\|t\|$ and $\frac{t}{\| t\|}$ are independent, and that $\frac{t}{\| t\|}$ is uniformly distributed on the unit sphere in $v_1^{\perp}$. We have the following equality in distribution:
\begin{equation*}
\left(r_1(v_1),r_1(\widetilde{v}_1)\right) \overset{(d)}{=} \left(r_1(X),r_1(X+\mathbf{s}Z) \right),
\end{equation*} 
with w.h.p. $$\mathbf{s} \geq \mathbf{s^1} = \frac{\|w\| \|X\|}{\left(\sum_{i=2}^{n} Z_i^2\right)^{1/2} \left(1 - \frac{\|w\| Z_1}{\left(\sum_{i=2}^{n} Z_i^2\right)^{1/2}}\right)} \asymp \xi_1 n^{1/6},$$ where the $X_i$, $Z_i$ and $w$ are defined in Section \ref{toymodel}, for $\xi=\xi_1$.
\end{proof}

\subsection{Proof of Lemma \ref{Pi+casei}} \label{proof_lemma_Pi+(i)}
\begin{proof}[Proof of Lemma \ref{Pi+casei}]
Recall that we work in the case $(i)$ ($\xi = o(n^{-7/6-\epsilon})$ for some $\epsilon>0$), with $\langle v_1,v'_1 \rangle >0$ and $\Pi^{\star} = I_n$. We want to show that w.h.p.
\begin{equation}
\label{Pi+caseieq}
	\langle A, \Pi_{+} B \Pi_{+}^{\top} \rangle > \langle A, \Pi_{-} B \Pi_{-}^{\top} \rangle.
\end{equation} 
Define
\begin{equation*}
\mathcal{G} := \left\lbrace u, \Pi_{+}(u)=\Pi^{\star}(u)=i \right\rbrace.
\end{equation*} and
\begin{equation*}
\mathcal{A} := \left\lbrace \xi n^{1/6}f(n)^{-1} \leq \mathbf{s} \leq \xi n^{1/6 }f(n) \right\rbrace,
\end{equation*} with $f \in \mathcal{F}$ such that $\mathbb{P} \left(\mathcal{A}\right) \to 1$. For $n$ large enough, on the event $\mathcal{A}$, we have $ 0 \leq \mathbf{s}n \leq  n^{-\epsilon }f(n)$. Hence, retaking the proof of Proposition \ref{EIG1:prop:zero_one_toy}, we have 
\begin{flalign*}
\phi_{x,z}\left(n, \mathbf{s} \right) & \geq \mathbb{P} \left(\cN^{+}(x,x+\mathbf{s}z)=\cN^{-}(x,x+\mathbf{s}z)=0\right) \\
& \sim \exp\left(- \mathbf{s}nE(x)\left[z(2F(z)-1)+2E(z)\right]\right) = 1 - O(n^{-\epsilon}f(n)).&&
\end{flalign*}

Thus, with dominated convergence, for $n$ large enough,
\begin{flalign}
\label{probassurA}
\mathbb{P}\left({\Pi_{+}}(u)=\Pi^{\star}(u) | \mathcal{A}  \right)  = \iint dx dz E(x) E(z) \mathbb{E}\left[\phi_{x,z}\left(n, \mathbf{s} \right)| \mathcal{A}\right] \geq 1 - O(n^{-\epsilon}f(n)).
\end{flalign} We use Markov's inequality with (\ref{probassurA}) to show that $\mathbb{P}\left( \card{\mathcal{G}} \leq n- n^{1-\epsilon/2} \; | \; \mathcal{A} \right) \leq O\left(n^{-\epsilon/2}f(n)\right) $, hence w.h.p.
\begin{equation}
\label{controlgoods}
\card{\mathcal{G}} \geq n- n^{1-\epsilon/2}.
\end{equation}

Splitting the sum
\begin{flalign*}
\langle A, \Pi_{+} B \Pi_{+}^{\top} \rangle & = \sum_{u,v} A_{u,v} B_{\Pi_{+}(u),\Pi_{+}(v)} = \sum_{(u,v) \in \mathcal{G}^2 } A_{u,v} B_{u,v} + \sum_{(i,j) \notin \mathcal{G}^2 } A_{u,v} B_{\Pi_{+}(u),\Pi_{+}(v)},
\end{flalign*} one has, w.h.p.,
\begin{multline*}
\langle A, \Pi_{+} B \Pi_{+}^{\top} \rangle = \sum_{(i,j) \in \mathcal{G}^2 } A_{u,v}^2 + \sum_{\substack{(i,j) \notin \mathcal{G}^2\\ (\Pi_{+}(u),\Pi_{+}(v)) \neq (v,u) } } A_{u,v} A_{\Pi_{+}(u),\Pi_{+}(v)} \\+ \sum_{\substack{(i,j) \notin \mathcal{G}^2\\ (\Pi_{+}(u),\Pi_{+}(v)) = (v,u) } } A_{u,v}^2 + \xi \sum_{1 \leq i,j \leq n } A_{u,v} H_{\Pi_{+}(u),\Pi_{+}(v)} \\ 
\geq C_1 \frac{\card{\mathcal{G}}^2}{n} - C_2 \left(n^2-\card{\mathcal{G}}^2\right)\frac{\log n}{n} -  C_2 \xi n^2 \frac{\log n}{n}.
\end{multline*} We applied the law of large numbers for the first sum, lower-bounded the third sum by zero, and the classical inequality $\max_{u,v}\left\lbrace A_{u,v}, H_{u,v} \right\rbrace \leq C_2 \frac{\log n}{n}$ (which holds w.h.p.) for the two others. \\
Inequality (\ref{controlgoods}) and condition $(i)$ lead to, w.h.p.
\begin{equation*}
\langle A, \Pi_{+} B \Pi_{+}^{\top} \rangle \geq C_1 n - 2 C_1 n^{1- \epsilon/2} - 2 C_2 n^{1- \epsilon/2} \log n - C_2 n^{-1/6-\epsilon}\log n \geq C_3 n.
\end{equation*}
On the other hand, since by definition $\Pi_{-}(i)=\Pi_{+}(n+1-i)$, w.h.p.,
\begin{multline*}
\langle A, \Pi_{-} B \Pi_{-}^{\top} \rangle = \sum_{(u,v) \in \mathcal{G}^2 } A_{u,v}B_{n+1-u,n+1-v} + \sum_{\substack{(u,v) \notin \mathcal{G}^2 } } A_{u,v} B_{\Pi_{-}(u),\Pi_{-}(v)} \\
\leq O(\log n) + \frac{\card{\mathcal{G}}^2 }{n}o(1) + C_2 \left(n^2-\card{\mathcal{G}}^2\right)\frac{\log n}{n}.
\end{multline*} For the first sum, we used the law of large numbers: the variables $A_{u,v}$ and $B_{n+1-u,n+1-v}$ are independent in all cases but at most $n+1$, and this part of the sum is bounded by $O(\log n) $. We used the same control on Gaussian variables as above. \\
This gives
\begin{equation*}
\left(\langle A, \Pi_{-} B \Pi_{-}^{\top} \rangle \right)_{+} = o_{\mathbb{P}}(n),
\end{equation*} where $(x)_{+} := \max (0,x)$, which proves (\ref{Pi+caseieq}).
\end{proof}

\subsection{Proof of Lemma \ref{Pi+caseii}}
\begin{proof}[Proof of Lemma \ref{Pi+caseii}]\label{proof_lemma_Pi+caseii}
Recall that we work in the case $(ii)$ ($\xi = \omega(n^{-7/6+\epsilon})$ for some $\epsilon>0$), with $\langle v_1,v'_1 \rangle >0$ and $\Pi^{\star} = I_n$. We want to show that the aligning permutation between $v_1$ and $-v'_1$ has a very bad overlap. Considering the pair $(X,-Y)$ where $(X,Y) \sim \J(n,s)$, one can adapt the proof of Proposition \ref{EIG1:prop:zero_one_toy}, with the new definitions
\begin{align*}
	\widetilde{S}^{+}(x,y) &:= \mathbb{P}\left(X_1 > x, -Y_1 < -y\right), \mbox{ and}\\
	\widetilde{S}^{-}(x,y) &:= \mathbb{P}\left(X_1 < x, -Y_1 > -y\right).
\end{align*} The analysis is even easier since for all $x,z$, there exist two constants $c,C$ such that $$0 < c \leq \widetilde{S}^{+}(x,x+s z) , \, \widetilde{S}^{-}(x,x+ sz) \leq C <1.$$
It is then easy to check that the proof of Proposition \ref{EIG1:prop:zero_one_toy}, case $(ii)$ adapts well.
\end{proof}
\addtocontents{toc}{\protect\setcounter{tocdepth}{2}}
 
\end{subappendices}

\chapter{Alignment of sparse \ER graphs: information-theoretic results}\label{chapter:impossibility}
In this chapter, we study fundamental limits of graph alignment: in the correlated \ER model, we prove an impossibility result for partial recovery in the sparse regime, with constant average degree and correlation, as well as a general bound on the maximal reachable overlap. This bound is tight in the noiseless case (the graph isomorphism problem) and we conjecture that it is still tight with noise. The proof of this negative result relies on a careful application of the probabilistic method to build automorphisms between tree components of a subcritical \ER graph.\\

This chapter is based on the paper \textit{Impossibility of partial recovery in the graph alignment problem} \cite{ganassali2021impossibility}, published at \emph{COLT 2021}, which is a joint work with M. Lelarge and L. Massoulié.

\section{Introduction}
As we have seen in Section \ref{intro:subsection:short_survey} of the introduction, a vast majority of previous works focus on the exact (resp. quasi-exact) alignment, which is known to be feasible in the dense case, when $nqs \geq \log n$ (resp. $nqs \to \infty$). On the computational side, many algorithms are proposed for (quasi-)exact alignment; however, none of these succeed in the sparse setting with constant correlation and average degree $\lambda >0$, i.e. with $q=\lambda/n$. It is thus natural and interesting to tackle the challenging question of partial alignment in the sparse setting.

\subsection{A colored view on the correlated \ER model}\label{impossibility:section:introduction}
	Let us recall the definition of the \emph{correlated \ER model} (already introduced in \eqref{eq:CER_model}) in the sparse case: in this chapter, we represent the graphs $(G,G') \sim \G(n,\lambda/n,s)$ with respectively blue and red edges, and with the same set of nodes $[n]$. For each edge, the colors are samples independently:
	\begin{itemize}
		\item with probability $\lambda s/n$ to get two-colored edges;
		\item with probability $\lambda(1-s)/n$ to get a blue (monochromatic) edge;
		\item with probability $\lambda(1-s)/n$ to get a red (monochromatic) edge;
		\item with probability $1-\lambda(2-s)/n$ to get a non-edge,
	\end{itemize}
	where $\lambda>0$ and $s\in [0,1]$ are fixed parameters and $n$ is large. In this model, $G$ and $G'$ are both sparse $\G(n,\lambda/n)$ graphs. For large values of $n$, the fraction of edges of one graph that are shared with the other is of order $s$ (see Figure \ref{fig:img_GGp}).
	
	\begin{figure}[h]
		\centering
		\includegraphics[scale=0.8]{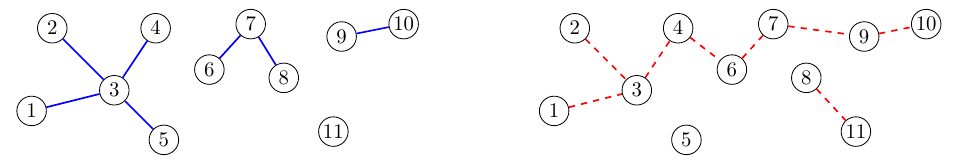}
		\caption{A realization of $(G,G')$ from the correlated \ER model, with $n=11$, $\lambda = 1.9$, and $s=0.7$. For the sake of readability, red edges are always dashed.}
		\label{fig:img_GGp}
	\end{figure}
	
	We then relabel the vertices of the red graph $G'$ with a uniform independent permutation $\pi^{\star} \in \cS_n$, and we observe $G$ and $H := G'^{\pi^{\star}}$, see Figure \ref{fig:img_GH}.
	Upon observing $G$ and $H$, the goal is to recover (or, reconstruct) partially the latent vertex correspondence $\pi^{\star}$ with probability converging to $1$ as $n\to \infty$.
	
	\begin{figure}[h]
		\centering
		\includegraphics[scale=0.8]{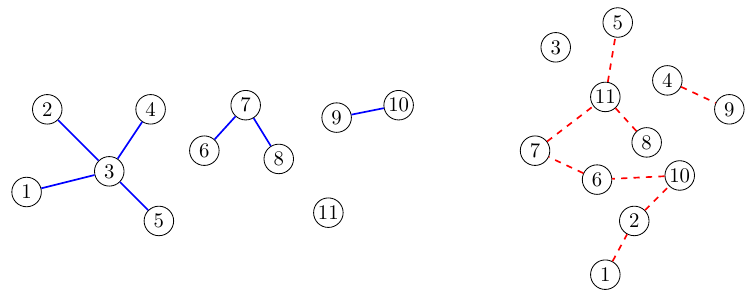}
		\caption{The pair $(G,H)$ corresponding to $(G,G')$ of Figure \ref{fig:img_GGp}, after relabeling $G'$ with the permutation $\pi^{\star} = (6)(1 \; 5 \; 3 \; 11\; 9 \;2 \;8 \;4 \;7 \;10)$.}
		\label{fig:img_GH}
	\end{figure}
	
	\subsection{Partial alignment in the sparse regime}
	 First note that since we are in the sparse regime, then even without noise, i.e. with $s=1$, there is no way to be able to map the $\Theta(n)$ isolated vertices\footnote{We refer to Theorem \ref{intro:theorem:connectivity_ER} of the introduction for a proof of this result.} in $G$ and $H$ better than chance. Hence, we rather focus on the partial alignment problem where we ask for the best possible fraction of matched vertices between $G$ and $H$. More formally, an \emph{estimator} $\hat{\pi}$ (of $\pi^{\star}$) is a $\cS_n$-valued measurable function of $(G, H)$. 
	 
	 Note that in this problem, the graphs could very well be unlabeled in the first place. We could assign the labels uniformly, the only interesting information being the correspondence between vertices. Hence, any estimator $\hat{\pi}$ must satisfy the \emph{equivariance} property, in the sense that for all $\sigma \in \cS_n$, 
	 
	 \begin{equation}\label{eq:equivariance}
	 \hat{\pi}(G^{\sigma},H) = \hat{\pi}(G,H) \circ \sigma^{-1} \,.
	 \end{equation}
	 
	\begin{remark}
	Note that unsurprisingly, the maximum a posteriori estimator $\hat{\pi}_{\MAP}$, which is the permutation solving the maximization problem \eqref{eq:QAP}, satisfies \eqref{eq:equivariance}.
	
	Another -- though more cumbersome -- approach to enforce some notion of equivariance (and put aside some trivial estimators such as $\hat{\pi}=\id$) would be to redefine the overlap as follows:
	\begin{equation*}\label{eq:overlap3}
	\ov(\hat{\pi}(G,H),\pi^{\star}) := \frac{1}{n \cdot n!} \sum_{\sigma \in \cS_n} \sum_{u=1}^{n} \one_{\hat{\pi}(G^{\sigma},H)(u)= \pi^{\star} \circ  \sigma^{-1} (u)} \,.
	\end{equation*} With this definition, it is ensured that for any $\sigma \in \cS_n$,
	\begin{equation*}
	\ov(\hat{\pi}(G^{\sigma},H),\pi^{\star}) = \ov(\hat{\pi}(G,H),\pi^{\star} \circ \sigma) \,.
	\end{equation*}
	\end{remark}
	
	We recall that partial alignment consists in finding a estimator $\hat{\pi}$ of $\pi^{\star}$ satisfying $\ov(\hat{\pi},\pi^{\star}) > \alpha n$ with high probability, for some $\alpha>0$.  Let us start by stating a conjecture\footnote{At the time this manuscript is being completed, this conjecture and a more general form are proved in \cite{Ding22}. The result of Theorem \ref{impossibility:theorem:sparse_threshold} however does not lose of its appeal, since it gives an upper bound on the best reachable overlap.}:
	\begin{conj*}
		\item[$(i)$] If $\lambda s \leq 1$, partial reconstruction is impossible, i.e. for any $\alpha >0$, for all estimator $\hat{\pi}$, $$\dP\left(\ov(\hat{\pi},\pi^{\star}) > \alpha n  \right) \underset{n \to \infty}{\longrightarrow} 0.$$
		\item[$(ii)$] If $\lambda s > 1$, partial reconstruction is possible (feasible), i.e. there exists $\alpha>0$ and an estimator $\hat{\pi}$ such that $$\dP\left(\ov(\hat{\pi},\pi^{\star}) > \alpha n  \right) \underset{n \to \infty}{\longrightarrow} 1.$$
	\end{conj*}
		
	\begin{figure}[h]
		\centering
		\includegraphics[scale=0.92]{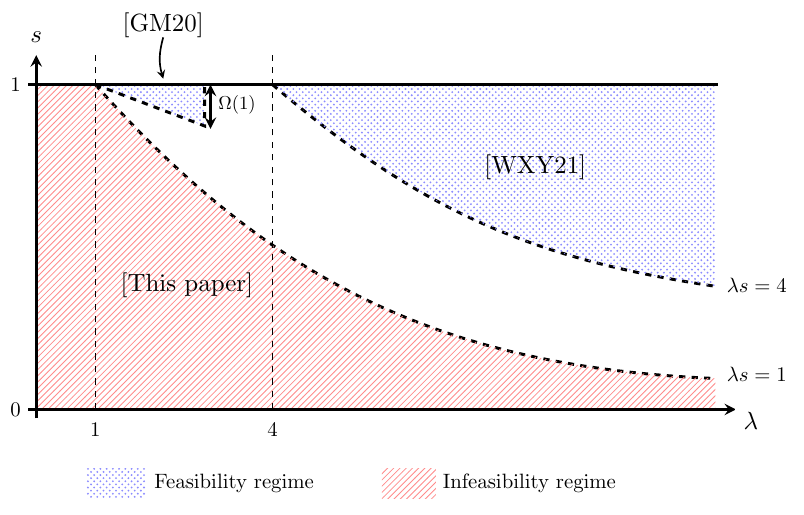}
		\caption[Caption for LOF]{Diagram of the $(\lambda,s)$ regions where partial reconstruction is known\protect\footnotemark\,to be impossible (resp. possible), in the sparse regime where $\lambda, s$ are fixed constants.}
		\label{fig:diagram}
	\end{figure}

\footnotetext{at the time of this contribution.}

	\subsubsection{Main result} The main result of the chapter is as follows:
	\begin{theorem}\label{impossibility:theorem:sparse_threshold} 
		For $\lambda>0$ and $s\in[0,1]$, we have for any $\alpha>0$, for any estimator $\hat{\pi}$:
		\begin{equation}
		\dP\left(\ov(\hat{\pi},\pi^{\star}) > (c(\lambda s)+ \alpha) n  \right) \underset{n \to \infty}{\longrightarrow} 0,
		\end{equation}
		where $c(\mu)$ is the greatest non-negative solution to the equation $e^{-\mu x}=1-x$.
	\end{theorem}
	Note that a well-known result (see e.g. \cite{Bollobas2001}) is that $c(\mu)$ is the typical fraction of nodes in the largest component of an \ER graph with average degree $\mu$, and that $c(\mu)=0$ if $\mu \leq 1$, and $c(\mu) \in (0,1)$ whenever $\mu >1$. Hence, Theorem \ref{impossibility:theorem:sparse_threshold} implies that partial reconstruction is impossible for $\lambda s \leq 1$. Moreover, if $\lambda s > 1$, any estimator can reach an overlap of at most $c(\lambda s) n+o(n)$. Note that $c(\lambda s)$ is the typical fraction of nodes in the largest component of the intersection graph. 
	
	\subsubsection{Related work in the sparse regime}
	In this chapter, we work in the regime where $\lambda >0$ and $s \in [0,1]$ are fixed constants. Our results prove part $(i)$ of the conjecture, which had not been previously studied, and give an upper bound on the maximal reachable overlap in case $(ii)$. 
	
	Most relevant related results\footnote{at the time of this contribution.} for our conjecture \cite{Ganassali20a} which proves\footnote{see Chapter \ref{chapter:NTMA}.} that partial recovery is possible in polynomial time in a region $\cR := \{(\lambda,s); \: \lambda\in [1,\lambda_0) \text{ and } s \in (s^*(\lambda),1] \}$ for some function $s^*(\lambda)<1$, so that interestingly the case $\lambda > \lambda_0$ is left open, nevertheless much in step with $(ii)$. Previous results from \cite{Hall20} showed that partial reconstruction was feasible for $\lambda s>C$, with an unspecified constant $C>20$. The work \cite{Wu2021SettlingTS} significantly improves these results, narrowing down the gap for $(ii)$: when translated with our notations, it is shown that partial alignment is possible (theoretically) if $\lambda s \geq 4+\eps$. In addition, an impossibility condition of the form $n q s \leq 1-\eps$ is also established, but in a denser case, where $nq/s = \omega(\log^2 n)$. Note that this last impossibility result does not cover our regime, where both the mean degree $nq$ and the correlation parameter $s$ are of order $1$.
	
	These results are summed up in a diagram in Figure \ref{fig:diagram}. In particular, our bound is tight and our conjecture is almost solved for the case $s=1$, with a remaining gap $[\lambda_0,4]$ being still open. 
	
	For the impossibility part, \cite{Wu2021SettlingTS} works with the mutual information $I(\pi^{\star}; G, G')$, closely related to the minimum mean squared error. They are able to derive an upper bound on the expectation of $\ov(\hat{\pi},\pi^{\star})$, for any estimator, which happens to be $o(1)$ when the mean degree in the parent graph of $G$ and $G'$ is at least of order $\log^2 n$, but not when $\lambda, s$ are of order $1$. In our result, we do not work directly with the mutual information, but we are considering the posterior distribution of $\pi^{\star}$: in simple words, we show that under the assumption $\lambda s \leq 1$ the posterior distribution puts equal weights on permutations that overlap only on a vanishing fraction of points. This is done by building ad hoc permutations with the probabilistic method.\\
	
	In this work, we derive information-theoretic results: our proof is not related to a particular algorithm. The search for efficient algorithms in this field is a very active field of research: we refer once again to Section \ref{intro:subsection:short_survey}. Unfortunately, all proposed algorithms are not known to give a positive fraction of overlap in the regime $\lambda s \geq 1$, hence leaving the question of the tightness of our bound open. New light will be shed on this question in Chapter \ref{chapter:MPAlign}.

	\section{Main results and global intuition} 
	\subsection{Some definitions} 
	Let us first recall that for two permutations $\sigma,\sigma' \in \cS_n$ we denote by $\ov(\sigma,\sigma')$ the number of points on which $\sigma=\sigma'$, namely 
	\begin{equation*}
		\ov(\sigma,\sigma') := \sum_{u=1}^{n} \one_{\sigma(u)=\sigma'(u)} \, .
	\end{equation*}
	
	Through all the chapter, we will implicitly consider that every graph $G$ of size $n$ has the canonical vertex set $[n]$. We will denote by $E(G)$ its edge set and $e(G)$ its number of edges.
	
	For any pair of graphs $(G,G')$, both labeled on $[n]$, we denote by $G \lor G'$ (resp. by $G \land G'$) the union graph (resp. intersection graph) of $G$ and $G'$, that is the graph with same node set and edge set $E(G) \cup  E(G')$ (resp. $E(G) \cap  E(G')$). The symmetric difference of $G$ and $G'$, denoted by $G \triangle G'$, is the subgraph made of edges of $G \lor G'$ that are not in $G \land G'$.
	
	In the case where edges are colored, say edges of $G$ (resp. $G'$) are blue (resp. red), these definitions extend to ensure colour preservation: note e.g. that in this case $G \land G'$ is simply the subgraph of $G \lor G'$ consisting of two-colored edges (see Figure \ref{fig:img_GunionGp}). 
	
	\begin{figure}[h]
		\centering
		\includegraphics[scale=0.8]{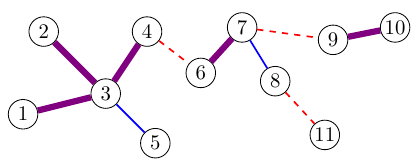}
		\caption{The graph $G \lor G'$ with $(G,G')$ of Figure \ref{fig:img_GGp}. For the sake of readability, the two-colored edges of $G \land G'$ are always drawn thick and purple.}
		\label{fig:img_GunionGp}
	\end{figure} 
	
	When the pair $(G,G')$ is drawn under the correlated \ER model, for all $u,v \in [n]$, we write $u \noir v$ (resp. $u \rouge v$) if $u$ and $v$ are connected in $G$, that is the edge is either blue or two-colored (resp. in $G'$, either red or two-colored).
	
	For $G$ a graph with vertex set $[n]$ and $\sigma \in \cS_n$, we denote by $G^{\sigma}$ the \emph{relabeling of $G$ with $\sigma$}, which is the graph with same vertex set $[n]$ and edges $\set{\sigma(u), \sigma(v)}$ for all $\set{u,v} \in E(G)$.
	
	Finally we recall the definition of $c(\mu)$: for all $\mu > 0$, $c(\mu)$ is the greatest non-negative solution to the equation $e^{-\mu x}=1-x$. We also recall the fact that for $\mu \leq 1$, $c(\mu)=0$.

	\subsection{General intuition on the main result} 
	Let us describe the general intuition for our result : recall that we are given $(G,H)$ drawn under the correlated \ER model with planted relabeling $\pi^{\star}$. The idea of the argument for impossibility is to show that, there are w.h.p. lots of permutations that have the same weight for the posterior distribution of $\pi^{\star}$ given $G, H$, and that are far apart. In other words, an informal statement is as follows :
	
	\begin{informal*}
		We want to show that there exists lots of relabelings $G^{\sigma_i}$ of $G$ such that:
		\begin{itemize}
			\item[$(i)$] There is no way of deciding (statistically) whether the two graphs we observe are $(G, G')$ or some $(G^{\sigma_i}, G')$.
			\item[$(ii)$] These relabelings are far apart from each other on small components of $G \land G'$.
		\end{itemize}
	\end{informal*}
	
	Let us give a formal version of the previous intuition. First note that for any labeled graphs $g,g'$ on $[n]$:
	\begin{flalign*}
		\dP(G=g,G'=g') 
		& = \left(\frac{\lambda s}{n}\right)^{e(g \land g')} \left(\frac{\lambda (1-s)}{n}\right)^{e(g \triangle g')} \left(1 - \frac{\lambda (2-s)}{n}\right)^{\binom{n}{2} - e(g \lor g')}.
	\end{flalign*}
	Since $$e(g \lor g') = e(g) + e(g') - e(g \land g') \quad \mbox{and} \quad e(g \triangle g') = e(g \lor g') - e(g \land g'),$$  
	$\dP(G=g,G'=g')$ is uniquely determined by $e(g), e(g')$ and $e(g \land g')$. In particular, the dependence of the joint distribution in $e(g \land g')$ is given by:
	
	\begin{equation}\label{eq:joint_distribution}
	\dP(G=g,G'=g')  \propto \left(\frac{s(n-\lambda(2-s))}{\lambda (1-s)^2}\right)^{e(g \land g')}.
	\end{equation}
	
	In view of \eqref{eq:joint_distribution}, preserving the posterior distribution by relabeling a graph $G$ is simply preserving the number of edges of their intersection graph. We now have a formal rephrasing for our conditions $(i)$ and $(ii)$ above: we encapsulate them in a theorem, which will constitute the bulk of this chapter.
	
	\begin{theorem}\label{theorem:autos} 
		Fix an integer $p >0$. Consider $(G,G')$ drawn under the correlated \ER model $\G(n,\lambda/n,s)$. Then, with high probability,  there exists $\left\lbrace \sigma_i \right\rbrace_{i \in [p]}$ -- that depend on the intersection graph $G \land G'$ -- such that 
		\begin{itemize}
			\item[$(i)$] $\forall i \in [p], \; e\left(G^{\sigma_i} \land G'\right) = e\left(G \land G'\right)$,
			
			\item[$(ii)$] $\forall i,j \in [p], \; i \neq j \implies \ov(\sigma_i,\sigma_j) \leq c(\lambda s) n + o(n)$, where the $o(n)$ is independent of $i,j \in [p]$.
		\end{itemize}
	\end{theorem} Let us now explain how Theorem \ref{theorem:autos} implies our impossibility result via a simple pigeonhole principle. 
	\begin{proof}[Proof of Theorem \ref{impossibility:theorem:sparse_threshold}]
		Let us take $\alpha>0$. We want to control the probability that the overlap between an estimator $\hat{\pi}$ and $\pi^{\star}$ is greater than $\alpha n + c(\lambda s) n$. Fix $\eps>0$, and take $p$ large enough so that $$\alpha\eps  p >2.$$ 
		First note that point $(i)$ together with \eqref{eq:joint_distribution} gives that the joint probability of $(G,G',\pi^{\star})$ is is equal to that of $(G^{\sigma_i},G',\pi^{\star})$, for all $i \in [p]$. Thus, for all estimator $\hat{\pi}$ depending on $G, H = G'^{\pi^{\star}}$, one has
		\begin{equation}\label{eq:egalite_distrib}
		\forall i \in [p], \; \ov\left(\hat{\pi}(G^{\sigma_i}, H),\pi^{\star}\right) \overset{(d)}{=} \ov\left(\hat{\pi}(G , H),\pi^{\star}\right),
		\end{equation} and by \eqref{eq:equivariance}, we also have
		\begin{equation}
		\forall i \in [p], \; \ov\left(\hat{\pi}(G^{\sigma_i}, H),\pi^{\star}\right) = \ov\left(\hat{\pi}(G, H),\pi^{\star} \circ \sigma_i\right).
		\end{equation}
		Let
		\begin{equation*}
			X :=  \sum_{i \in [p]} \one_{\ov(\hat{\pi},\pi^{\star} \circ {\sigma_i})>(c(\lambda s) + \alpha)n}
		\end{equation*} Note that because of point $(ii)$, all $\ov(\pi^{\star} \circ {\sigma_i},\pi^{\star} \circ {\sigma_j})$ are at most $c(\lambda s) n + o(n)$ for $i \neq j \in [p]$. Thus, there are at least $X \times (\alpha -o(1)) n$ distinct points among the node set $[n]$. This gives that one necessarily has 
		\begin{equation}
		X \leq \frac{1}{\alpha -o(1)}.
		\end{equation}
		
		Then, taking the expectation and considering the event on which the set $\left\lbrace \sigma_i \right\rbrace_{i \in [p]}$ of Theorem \ref{theorem:autos} exists -- which happens with probability $1-o(1)$ -- gives
		\begin{flalign*}
			\dE\left[X \right] & \geq \sum_{i=1}^{p} \dP\left(\ov(\hat{\pi},\pi^{\star} \circ \sigma_i ) > (c(\lambda s) + \alpha)n  \right) - p \times o(1) \\
			& = p \times \dP\left(\ov(\hat{\pi},\pi^{\star}) > (c(\lambda s) + \alpha)n \right) - o(1).
		\end{flalign*} Hence, 
		\begin{equation}
		\dP\left(\ov(\hat{\pi},\pi^{\star}) > (c(\lambda s) + \alpha)n \right) \leq \frac{1}{p (\alpha -o(1))} + o(1).
		\end{equation} 
		For $n$ large enough, the right-hand side of the last term is less that $\frac{1}{p(\alpha/2)}$, which is less than $\eps$. This proves as desired that for all $\alpha >0$
		\begin{equation}
		\dP\left(\ov(\hat{\pi},\pi^{\star}) > (c(\lambda s) + \alpha)n \right) \underset{n \to \infty}{\longrightarrow} 0.
		\end{equation}
	\end{proof}
	
	We are now left to understand how to build ad hoc permutations verifying points $(i)$ and $(ii)$ of Theorem \ref{theorem:autos}. In order to build these permutations, we are going to relabel the vertices on small tree components of the intersection graph $G \land G'$. As a first step, we hereafter check that they indeed nearly cover the whole graph, when letting aside the giant component.
	
	\subsection{Vertices on small tree components}
	We briefly recall the definition of the simple \ER model $\G(n,p)$: it consist in drawing a (single) graph with node set $[n]$ in which every edge is independently present with probability $p$. Let us begin with a classical result:
	
	\begin{lemma}[\cite{Bollobas2001}, Corollary 5.8, Theorem 6.11]
		\label{lemma:bollobas_trees}
		Let $G\sim \G(n,\mu/n)$ with $\mu>0$, and $a_n \to \infty$. Then, with high probability, $G$ has a giant component of order $c(\mu) n +o(n)$ and outside the giant component, at least $(1-c(\mu)) n-a_n$ vertices are on tree components.
	\end{lemma} We need here a slight adaptation of this result, showing that $(1-c(\mu))n-o(n)$ vertices are in fact on \emph{small} tree components. 
	\begin{lemma}
		\label{lemma:small_trees}
		Let $G\sim \G(n,\mu/n)$ with $\mu >0$, and $K(n) \to \infty$. Then with high probability, $(1-c(\mu))n-o(n)$ vertices are on tree components of size at most $K(n)$.
	\end{lemma}
	\begin{proof}
		Assume without loss of generality that $K(n) = o(\log n)$. Let $T_>$ be the number of vertices that are on tree components of size $\geq K(n)$. Taking $a_n = o(n)$ in Lemma \ref{lemma:bollobas_trees}, it remains to show that w.h.p., $T_> = o(n)$. This is done easily by bounding very roughly the first moment. Another classical result (see e.g. \cite{Janson00}, Theorem 5.4) is that with probability $1-o(1)$, all tree components are of size $O(\log n)$, which gives
		\begin{flalign*}
			\frac{\dE\left[T_>\right]}{n} &\leq o(1)+\sum_{k = K(n)}^{O(\log n)} \frac{1}{n} \cdot k \cdot \binom{n}{k} k^{k-2} \left(\frac{\mu}{n}\right)^{k-1} \left(1-\frac{\mu}{n}\right)^{k(n-k)+\binom{k}{2}-k+1}\\
			& \leq o(1)+ (1+o(1)) \sum_{k = K(n)}^{O(\log n)} \frac{e^k}{k} \mu^{k-1} e^{-k\mu},
		\end{flalign*} using $\binom{n}{k} \leq \left(\frac{en}{k}\right)^k$ together with Cayley's formula\footnote{Cayley's formula states that the number of trees on $k$ labeled vertices is $k^{k-2}$.} and the fact that for all indices $K(n) \leq k \leq O(\log n)$ in the sum, $k^2 \leq o(n)$ (uniformly). Now, the series in the right hand term has general terms which is $O\left(e^{-k(\mu - \log \mu +1)}\right)$, and since $\mu - \log \mu +1>0$ the series converges, which implies that $\dE\left[T_>\right/n]=o(1)$. The proof is concluded by Markov's inequality.
	\end{proof}
	
	Since in our model $G \land G'$ is an \ER graph of parameters $(n,\lambda s/n)$, the previous results ensures that all but a vanishing part of the $(1-c(\mu)) n$ vertices outside the giant component are on small (i.e. $\leq K(n)$) tree components of the intersection graph. For the rest of the chapter, we will take $$K(n) = \lfloor \sqrt{\log n} \rfloor.$$
	
	This first step suggests to build the permutations (relabelings) only by looking at $G \land G'$. Hence, we will first consider the random generation of the intersection graph, then create some permutations $\sigma_i$, and finally reveal the monochromatic edges.
	
	The generating process is as follows: since almost all  $(1-c(\mu)) n$ vertices are on small trees in $G \land G'$, we can prove that each small tree up to isomorphism will have a number of occurrences in the intersection graph of order $n$ (this is claimed more precisely in Lemma \ref{lemma:controle_X}). Permuting iteratively these isomorphic trees, we may derange them quite a lot, and each time differently. 
	
	In order to prove Theorem \ref{theorem:autos}, we use the \emph{probabilistic method}\footnote{The main interest of this widely used method (see \cite{ProbaMethod}) is to be non-constructive. Indeed, as detailed in the next Sections, explicitly giving the $p$ permutations considered in Theorem \ref{theorem:autos} is very cumbersome, because of the extra double edges that may appear (see Section \ref{section_extra_double_edges}).}: we give in the next section a simple detailed stochastic method to build $p$ permutation candidates, and we will next prove that these permutations satisfy conditions $(i)$ and $(ii)$ with positive probability, hence proving the desired existence.

	\section{Building automorphisms of $G \land G'$ tree-wise}
	Through all this section, we work conditionally on the intersection graph $G \land G'$ (that is the two-colored edges). 
	
	\subsection{Mathematical formalization}
	Recall that we fix $K:=K(n)= \lfloor \sqrt{\log n} \rfloor$. For all $k \in [K]$, we will denote by $\dT_k$ the set of \emph{unlabeled} trees of size $k$. $\dT_k$ can also be viewed as the set of equivalence classes of labeled trees of size $k$ for the isomorphism relation. Note that $\dT_k$ is finite and that we can roughly upper bound its size by the number of \emph{labeled} trees of size $k$ which equals $k^{k-2}$, by Cayley's formula\footnote{This upper bound is far from being optimal, but is enough for our use.}.
	
	\begin{figure}[h]
		\centering
		\includegraphics[scale=1]{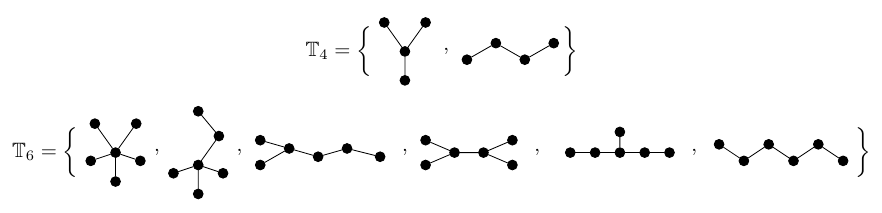}
		\caption{Explicit composition of $\dT_4$ (of size $2$) and $\dT_6$ (of size $6$).}
		\label{fig:img_ex_Tk}
	\end{figure} 
	
	For a given tree $\bT \in \dT_k$, we will denote by $X_\bT$ the number of distinct connected components of $G \land G'$ that are isomorphic to $\bT$, $H_\bT := \left\lbrace T_1, T_2, \ldots, T_{X_\bT} \right\rbrace$ the set of the corresponding labeled subgraphs of $G \land G'$, and $V(H_\bT)$ the set of vertices of $[n]$ that belong to one of the trees in $H_\bT$.
	
	Our global finite recursion will be done on the finite set 
	\begin{equation}\label{eq:def_T}
	\dT := \bigcup_{k=1}^{K} \dT_k = \left\{\bT_1, \bT_2, \ldots, \bT_M\right\},
	\end{equation} which we assume to have been ordered increasingly according to tree sizes, for convenience. The global permutation $\sigma$ is built block-wise by composing permutations $\sigma_\bT$ for $\bT \in \dT$ such that each $\sigma_\bT$ only acts on vertices of $H_\bT$.
	
	More precisely, for a fixed $\bT \in \dT$, $\sigma_\bT$ will consists in permuting the vertices tree by tree, so $\sigma_\bT$ will be determined by a tree permutation $\Sigma_\bT$ of size $X_\bT$. Assume that for all trees $T_1, \ldots, T_{X_\bT}$ isomorphic to $\bT$ in $G \land G'$, we fix some isomorphisms $\psi_1, \ldots, \psi_{X_\bT}$ such that $T_i\underset{\psi_i}{ \; \widehat{=} \; } \bT$ for all $i \in [X_\bT]$. More generally we will denote $\itree(u)$ the index of the tree that $u \in V(H_\bT)$ belongs to (when there is no ambiguity on $\bT$), and $u \simt u'$ when two vertices of $G \land G'$ are sent onto the same point of $\bT$ by these isomorphisms. Then, the natural definition of the node permutation $\sigma_{\bT}$ according to $\Sigma_\bT$ and these isomorphisms is given by
	
	\begin{equation}\label{eq:Sigma_sigma}
	\sigma_\bT : u \mapsto 
	\left\{
	\begin{array}{ll}
	\psi_{\Sigma_\bT (\itree(u))}^{-1} \circ \psi_{\itree(u)} (u) \;\; (\in T_{\Sigma_\bT (\itree(u))}) & \mbox{if } u \in V(H_{\bT}), \\
	u & \mbox{if } u \notin V(H_{\bT}).
	\end{array}
	\right.
	\end{equation} Note that by definition, $V(H_{\bT})$ is stable by $\sigma_\bT$, and $\sigma_\bT$ fixes all nodes in $[n] \setminus V(H_{\bT})$. Recall that $M$ denotes the total size of $\dT$ as defined in \eqref{eq:def_T}. The recursive construction is as follows :
	
	\begin{algorithm}[H]
		\caption{Recursive construction of $\sigma$}
		\label{algo:rec_construction}
		\SetAlgoLined
		Initialize $\sigma_0 \gets \id$\;
		
		\For{$i=1$ to $M$}{
			Consider $\bT = \bT_i$ and draw uniformly at random the tree permutation $\Sigma_{\bT} \in \cS_{X_{\bT}}$, independently from the past\;
			
			Consider $\sigma_{\bT}$ the node permutation associated with $\Sigma_{\bT}$ by \eqref{eq:Sigma_sigma}\;
			
			$\sigma_{i} \gets \sigma_{\bT} \circ \sigma_{i-1}$\;
		}
		\textbf{return} $\sigma = \sigma_M$
	\end{algorithm} Note that at the end of the procedure, $\sigma$ fixes all points that are either on the giant component of the intersection graph, or on a component that is not a tree a size $\leq K(n)$. Figure \ref{fig:img_algo_ok} gives an example of this random recursive construction (for convenience, $\lambda s <1$, the true labels are in red, whereas blue labels enables to keep track of the relabeling recursively built on the blue graph). 
	
	\begin{figure}[h]
		\centering
		\includegraphics[scale=0.83]{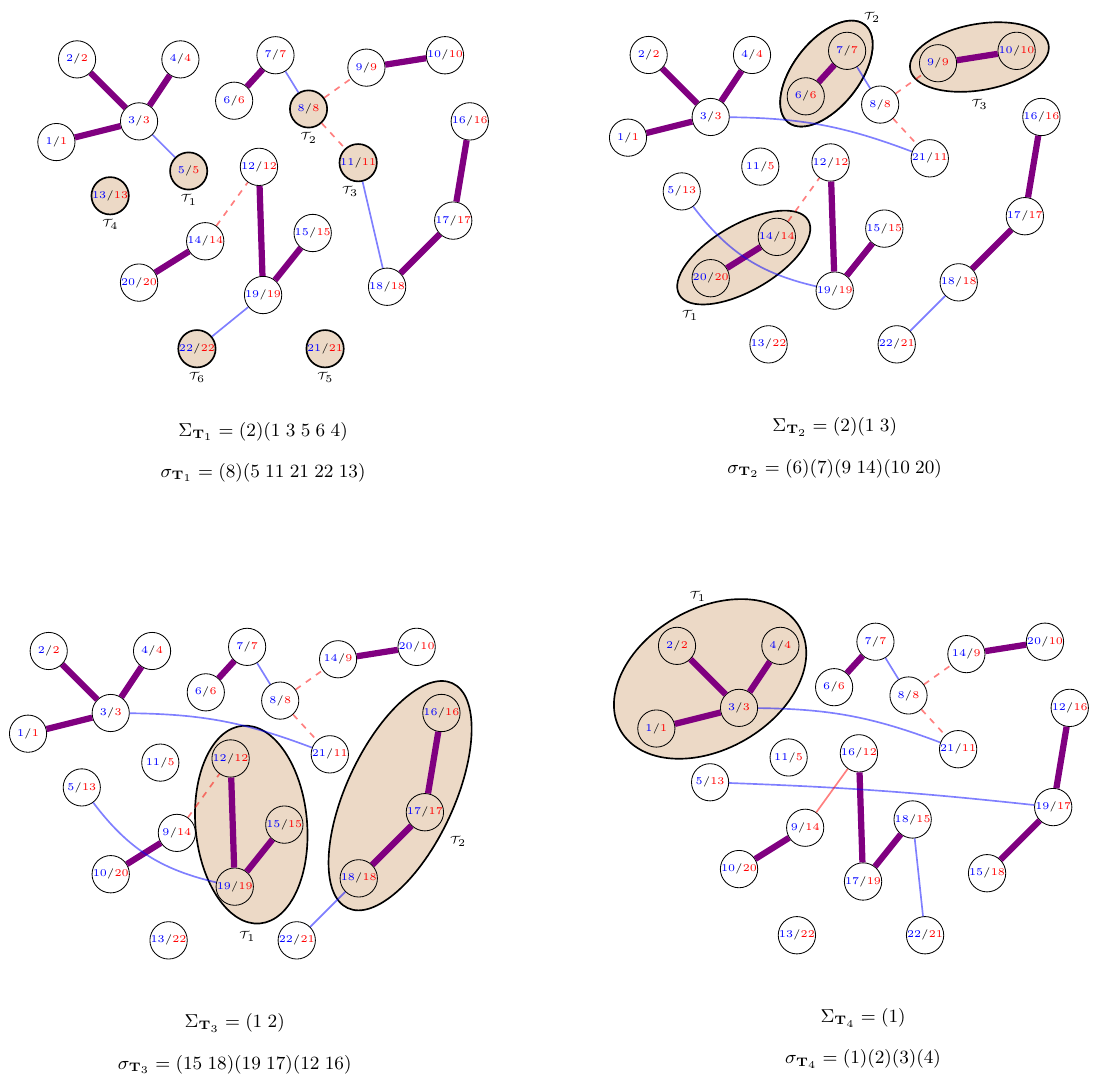}
		\caption{Example of recursive (tree-wise) generation of a permutation with Algorithm \ref{algo:rec_construction}.}
		\label{fig:img_algo_ok}
	\end{figure} 
	
	Through the analysis we will need the following control on $X_\bT$ for $\bT \in \dT$:
	
	\begin{lemma}\label{lemma:controle_X}
		Recall that $K(n)= \lfloor \sqrt{\log n} \rfloor$. For all $k \in [K(n)]$, define $
		f(k) := \frac{(\lambda s)^{k-1}e^{-\lambda s k}}{k!}$. Then, with high probability (on the intersection graph),
		\begin{equation}\label{eq:controle_X}
		\forall k \in [K(n)], \, \forall \bT \in \dT_k, \, X_\bT \geq n (1-o(1)) f(k).\end{equation}
	\end{lemma} The proof of this result is deferred to Appendix \ref{app:proof_control_X}.
	
	\begin{remark}\label{remark:f(K)}
		Note that since $\lambda s e^{-\lambda s}<1$, $k \mapsto f(k)$ is decreasing with $k$. Moreover, for $K(n) \leq \sqrt{\log n}$, we have that for any $t>0$, $$f(K(n)) \geq \exp\left(- C\sqrt{\log n} \log \log n \right) \gg n^{-t}.$$  
	\end{remark}
	
	\subsection{Ensuring that the permutations are 'far apart'} 
	We check in this section that Algorithm \ref{algo:rec_construction} generates permutations that will verify condition $(ii)$ of Theorem \ref{theorem:autos}, w.h.p. Let $\sigma_{1}, \ldots, \sigma_{p}$ be generated independently with Algorithm \ref{algo:rec_construction}. We then have the following results:
	
	\begin{lemma}\label{lemma:controle_overlap_ij} With high probability, for all $i \neq j \in [p]$,
		$$ \ov(\sigma_i,\sigma_j) = c(\lambda s) n + o(n).$$
	\end{lemma}
	
	This lemma is proved in Appendix \ref{app:proof_overlap}. In the sequel we will denote by $V_\infty$ the set of vertices that are on the giant component of $G \land G'$ (if there is one), and by $V_>$ the vertices of $[n] \setminus V_\infty$ that are \emph{not} on tree components of size $\leq K(n)$. Finally we set $V_{\infty,>} := V_\infty \cup V_>$. Define
	\begin{equation}\label{eq:def_Sin_Sout}
	\cS_{in} := \binom{[n]\setminus V_{\infty,>} }{2}, \quad \cS_{out} := \binom{[n]}{2} \setminus \left(\binom{V_{\infty,>}}{2} \cap \binom{[n]\setminus V_{\infty,>}}{2} \right), \quad \cS := \cS_{in} \cup \cS_{out}.
	\end{equation}$\cS_{in}$ is the set of edges that have both endpoints outside $V_{\infty,>}$, whereas edges of $\cS_{out}$ have exactly one endpoint in $V_{\infty,>}$. We say that an edge $(u,v) \in \binom{[n]}{2}$ is a \emph{common fixed edge} of permutations $\sigma_{1}, \ldots, \sigma_r$ if 
	$$ \left\lbrace \sigma_1(u), \sigma_1(v) \right\rbrace = \ldots = \left\lbrace \sigma_r(u), \sigma_r(v) \right\rbrace. $$ For all subset of edges $\cW \subseteq \binom{[n]}{2}$, we define
	\begin{equation}
	F(\cW,\sigma_{1}, \ldots, \sigma_r) := \sum_{e \in \cW} \one_{e \mbox{\footnotesize{ is a common fixed edge of} } \sigma_{1}, \ldots, \sigma_r}.
	\end{equation} 
	
	We now state a result -- which proof is deferred to \ref{app:proof_control_F} -- that will be useful in next section.
	
	\begin{lemma}\label{lemma:controle_F}
		With high probability, we have, for any $t>0$,
		\begin{itemize}
			\item for any $i_1 \neq i_2 \in [p]$,
			\begin{equation}\label{eq:lemma:controle_F2}
			F(\cS,\sigma_{i_1},\sigma_{i_2}) \leq n^{1+t},
			\end{equation}
			
			\item for any $i_1, i_2, i_3 \in [p]$ pairwise distinct,
			\begin{equation}\label{eq:lemma:controle_F3}
			F(\cS,\sigma_{i_1},\sigma_{i_2},\sigma_{i_3}) \leq n^t,
			\end{equation}
			
			\item for any $r \geq 4$, $i_1, \ldots, i_r \in [p]$ pairwise distinct,
			\begin{equation}\label{eq:lemma:controle_F4}
			F(\cS,\sigma_{i_1},\ldots,\sigma_{i_r}) =0.
			\end{equation}
		\end{itemize}
	\end{lemma}
	
	\subsection{Emergence of extra double edges}\label{section_extra_double_edges} 
	In the example of Figure \ref{fig:img_algo_ok}, we can see that the number of two-colored edges in the relabeled union graph $G^{\sigma_i} \lor G'$ is constant through time. This property is fundamental for point $(i)$ of Theorem \ref{theorem:autos}. However, depending on the random $\sigma_{\bT_i}$ drawn through the process -- we recall that they are drawn independently from the monochromatic edges, that are not revealed yet -- we may see extra two-colored edges appear (extra double edges hereafter). Figure \ref{fig:img_algo_pb} shows a case in which there is an emergence of an extra double edge in the process.
	
	\begin{figure}[h]
		\centering
		\includegraphics[scale=0.85]{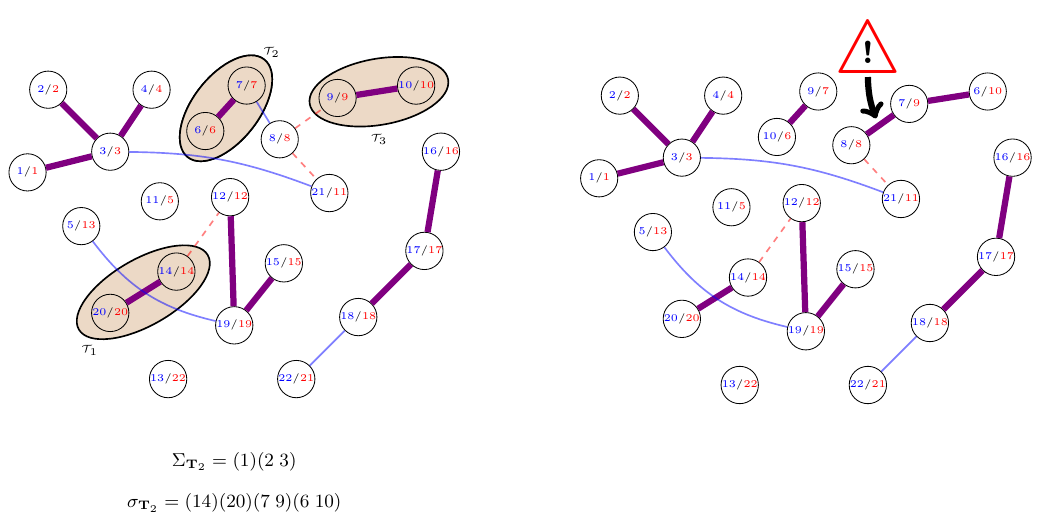}
		\caption{Example of the emergence of an extra double edge in Algorithm \ref{algo:rec_construction}.}
		\label{fig:img_algo_pb}
	\end{figure}
	Note that the number of two-coloured edges can only be greater or equal to $e(G \land G')$ through this process, since by definition we are preserving edges of the intersection graph.
	
	The last part of our work is to prove that there is a positive probability that applying independently Algorithm \ref{algo:rec_construction} $p$ times gives $p$ permutations that do not present extra double edges, before using the probabilistic method. This step will require a Poisson approximation, described hereafter.
	
	\section{Poisson approximation, proof of Theorem \ref{theorem:autos}}
	
	In this section we introduce $n'$ to be the number of vertices that the permutations actually act on:
	\begin{equation}
	n' := \card{[n] \setminus V_{\infty,>}} \sim (1-c(\lambda s)) n \mbox{ w.h.p.}
	\end{equation} 
	
	\subsection{Poisson approximation for extra double edges}
	In the sequel, we will assume that we fix a set $\left\{\sigma_{i}\right\}_{i \in [p]}$ of $p$ permutations of $[n']$, verifying : 
	\begin{equation}\tag{H1}\label{H1}
	\mbox{for all } t>0, \mbox{for all } m\neq m' \in [p], F(\cS,\sigma_{m},\sigma_{m'}) \leq n^{1+t}.
	\end{equation} 
	\begin{equation}\tag{H2}\label{H2}
	\mbox{for all } t>0, \mbox{for all } m_1, m_2, m_3 \in [p] \mbox{ pairwise distinct }, F(\cS,\sigma_{m_1},\sigma_{m_2},\sigma_{m_3}) \leq n^{t}.
	\end{equation} 
	\begin{equation}\tag{H3}\label{H3}
	\mbox{There are no common fixed edge of any $r$-tuple in $\left\{\sigma_{i}\right\}_{i \in [p]}$.}
	\end{equation}
	
	We will work under the event $\cE_\cS$ on which $n' \sim (1-c(\lambda s)) n$ and $\card{\cS} \sim \binom{n'}{2} \sim n'^2/2 = (1-c(\lambda s))^2 n^2/2 $. It is easy (see e.g. \cite{Bollobas2001}) to show that $\cE_\cS$ is satisfied w.h.p. As explained before, some extra double edges (e.d.e. hereafter) may appear when revealing the non double edges of $\cS$ (that is, blue and red edges that are not between vertices of $V_{\infty,>}$). Note that for every edge we have
	\begin{flalign*}
		\dP\left(u \noir v \, | \, (u,v) \notin E(G \land G') \right) &= \dP\left(u \rouge v \, | \, (u,v) \notin E(G \land G') \right)\\
		&= \frac{\dP\left(u \rouge v, (u,v) \notin E(G \land G') \right)}{\dP\left((u,v) \notin E(G \land G') \right)} \\ &= \frac{\lambda(1-s)/n}{1 - \lambda s/n} \sim \frac{\lambda(1-s)}{n}.
	\end{flalign*}
	For any permutation $\sigma$, define the number of created e.d.e. by the relabeling of $G$ by $\sigma$ as follows:
	\begin{equation}\label{def:delta_sigma}
	\Delta(\sigma) := \sum_{\left\{u,v\right\} \in \cS} \one_{u \nnoir v} \one_{\sigma(u) \rrouge \sigma(v)}.
	\end{equation}
	We now present the key result for our analysis, with the notation $n^{\underline{k}}$ for the \emph{falling factorial} $$n^{\underline{k}}:= n(n-1)\cdots(n-k+1).$$
	
	\begin{theorem}[Asymptotic Poisson behavior of $\left\{\Delta(\sigma_{i})\right\}_{i \in [p]}$]\label{theorem:poisson}
		Assume that $\left\{\sigma_{i}\right\}_{i \in [p]}$ verify \eqref{H1}, \eqref{H2} and \eqref{H3}. Then, for all $\ell_1, \ell_2, \ldots, \ell_p \geq 0$,
		\begin{equation}\label{eq:theorem:poisson}
		\dE \left[\Delta(\sigma_1)^{\underline{\ell_1}}\Delta(\sigma_2)^{\underline{\ell_2}} \cdots \Delta(\sigma_p)^{\underline{\ell_p}} \, \big| \, G \land G', \cE_\cS\right] \underset{n \to \infty}{\longrightarrow} \left(\frac{\lambda^2(1-s)^2 (1-c(\lambda s))^2}{2}\right)^{\ell_1 + \ell_2 + \ldots + \ell_p}.
		\end{equation} In other words, conditionally to graph $G \land G'$ and event $\cE_\cS$, the random variables $\left\{\Delta(\sigma_{i})\right\}_{i \in [p]}$ are asymptotically distributed as independent Poisson variables of parameter $\frac{\lambda^2(1-s)^2 (1-c(\lambda s))^2}{2}$.
	\end{theorem}
	
	The proof of Theorem \ref{theorem:poisson}, based on a fine control of terms of unusually high contribution, is deferred to Appendix \ref{app:proof_theorem_poisson}.
	\subsection{Proof of Theorem \ref{theorem:autos}}
	\begin{proof}
		The proof is quite straightforward now. Fixing $p>0$, Lemma \ref{lemma:controle_F} gives that \eqref{H1}, \eqref{H2} and \eqref{H3} are verified w.h.p. by some $\sigma_{1}, \ldots, \sigma_{p}$ generated independently with Algorithm \ref{algo:rec_construction}. Then, the probability (on the remaining monochrome edges) that the $p$ permutations given satisfy conditions $(i)$ and $(ii)$ of Theorem \ref{theorem:autos} is equivalent to
		\begin{multline*}
		(1-o(1)) \times \dP\left(\Poi\left(\frac{\lambda^2(1-s)^2 (1-c(\lambda s))^2}{2}\right)=0 \right)^p \\ = (1-o(1)) \exp\left(- p \frac{\lambda^2(1-s)^2(1-c(\lambda s))^2}{2}\right) >0,
		\end{multline*} which gives the existence with high probability of a set a permutations of size $p$ satisfying conditions $(i)$ and $(ii)$ of Theorem \ref{theorem:autos}.
	\end{proof}

\begin{subappendices}
	\section{Proof of Theorem \ref{theorem:poisson}}\label{app:proof_theorem_poisson}
	\begin{proof}[Proof of Theorem \ref{theorem:poisson}]
		Let $\ell_1, \ell_2, \ldots, \ell_p $ be non negative integers. Recall that conditioned to $G \land G'$, each edge of $\cS$ is independently blue (resp. red) with probability 
		\begin{equation*}
			q = q(\lambda,s,n) := \frac{\lambda(1-s)}{n - \lambda s}.
		\end{equation*}
		Now, let us explain why convergence \eqref{eq:theorem:poisson} holds. First recall that for a given $\ell \geq 0$, $\dE\left[\Delta(\sigma)^{\underline{\ell}}\right]$ is nothing else but the expected number of (ordered) $p-$tuples of edges $\left\{u,v\right\} \in \cS$ such that $\one_{u \nnoir v} \one_{\sigma(u) \rrouge \sigma(v)} = 1$. Using the notation ${\sum}^*$ for summation of ordered tuples of edges in $\cS$ as well as linearity of expectation, we get:
		\begin{multline}\label{eq:expr_multi_mom}
			\dE \left[\Delta(\sigma_1)^{\underline{\ell_1}}\Delta(\sigma_2)^{\underline{\ell_2}} \cdots \Delta(\sigma_p)^{\underline{\ell_p}}\right] = \\
			\sideset{}{^*}\sum_{\substack{\lbrace u^{(1)}_{1},v^{(1)}_{1}\rbrace, \\ \lbrace u^{(1)}_{2},v^{(1)}_{2}\rbrace, \\\ldots, \\ \lbrace u^{(1)}_{\ell_1},v^{(1)}_{\ell_1}\rbrace} } \; \;
			\sideset{}{^*}\sum_{\substack{\lbrace u^{(2)}_{1},v^{(2)}_{1}\rbrace, \\ \lbrace u^{(2)}_{2},v^{(2)}_{2}\rbrace, \\\ldots, \\ \lbrace u^{(2)}_{\ell_2},v^{(2)}_{\ell_2}\rbrace} }
			\ldots
			\sideset{}{^*}\sum_{\substack{\lbrace u^{(p)}_{1},v^{(p)}_{1}\rbrace, \\ \lbrace u^{(p)}_{2},v^{(p)}_{2}\rbrace, \\\ldots, \\ \lbrace u^{(p)}_{\ell_p},v^{(p)}_{\ell_p}\rbrace} }
			\dE \left[\prod_{m=1}^{p} \prod_{j=1}^{\ell_m} \one_{u^{(m)}_{j} \nnoir v^{(m)}_{j}} \one_{\sigma_m(u^{(m)}_{j}) \rrouge \sigma_m(v^{(m)}_{j})}\right] 
		\end{multline}
		
		First observe that the total number of terms $N$ in the previous sum is 
		$$N := \card{\cS}^{\underline{\ell_1}} \times \card{\cS}^{\underline{\ell_2}} \times \cdots \card{\cS}^{\underline{\ell_p}} \sim \left(\frac{(1-c(\lambda s))^2 n^2 }{2}\right)^{\ell_1 + \ldots + \ell_p},$$ since $\card{\cS} \sim \frac{ (1-c(\lambda s))^2 n^2}{2}$ on event $\cE_\cS$.\\
		
		\proofstep{Lower bound.} Observe that the $N$ terms in the sum of eq. \eqref{eq:expr_multi_mom} are made in general of $2(\ell_1 + \ldots + \ell_p)$ indicator variables, not necessarily distinct. For most of the terms however, all involved edges are distinct, thus independent, and their contribution to the sum is $q^{2(\ell_1 + \ldots + \ell_p)}$.
		
		Whenever a pair of blue (resp. red) indicators are equal, at least one term may be canceled, so the contribution to the expectation is higher than $q^{2(\ell_1 + \ldots + \ell_p)}$.
		
		Whenever a pair of edges that appear in a blue/red pair of indicators are equal, the product of the indicators is necessarily $0$ (indeed, an edge in $\cS$ cannot be two-colored). These terms, where at least one equality of the form $\{u_j^{(m)}, v_j^{(m)}\}  = \{\sigma_{m'}(u^{(m')}_{j'}) , \sigma_{m'}(v^{(m')}_{j'})\} $ occurs, cover the case where the contribution is strictly less that $q^{2(\ell_1 + \ldots + \ell_p)}$ (it is $0$). There are at most 
		$$\binom{\ell_1+\ldots+\ell_p}{2} \left(\frac{n^2}{2}\right)^{\ell_1+\ldots+\ell_p-1}$$ such terms. Thus
		\begin{flalign*}
			\dE \left[\Delta(\sigma_1)^{\underline{\ell_1}}\Delta(\sigma_2)^{\underline{\ell_2}} \cdots \Delta(\sigma_p)^{\underline{\ell_p}}\right] & \geq \left(N - \binom{\ell_1+\ldots+\ell_p}{2} \left(\frac{n^2}{2}\right)^{\ell_1+\ldots+\ell_p-1}\right) \times q^{2(\ell_1 + \ldots + \ell_p)}\\
			& \sim \left(\frac{(1-c(\lambda s))^2 n^2}{2}\right)^{\ell_1 + \ldots \ell_p} \times \left(\frac{\lambda (1-s)}{n}\right)^{2(\ell_1 + \ldots + \ell_p)} \\
			&\underset{n \to \infty}{\longrightarrow} \left(\frac{\lambda^2(1-s)^2 (1-c(\lambda s))^2}{2}\right)^{\ell_1 + \ell_2 + \ldots + \ell_p}.
		\end{flalign*}
		
		\proofstep{Upper bound.} The terms that we now want to study are the terms for which the contribution is greater than $q^{2(\ell_1 + \ldots + \ell_p)}$. Looking closely at the general product in \eqref{eq:expr_multi_mom}, an unusual high contribution is the consequence of three possible type of constraints:
		\begin{itemize}
			\item[$(i)$] constraints of the form  $\{u_j^{(m)}, v_j^{(m)}\}  = \{u^{(m')}_{j'} , v^{(m')}_{j'}\} $: note that since the sums are made of ordered tuples, this equality may happen only for pairs such that $m \neq m'$. Moreover, transitivity of equality implies that a constraint implying some $\{u_j^{(m)}, v_j^{(m)}\}$ may happen at most once for each $m' \in [p], m' \neq m$ (otherwise we would have a relationship of the form $\{u_{j'}^{(m')}, v_{j'}^{(m')}\} = \{u_{k'}^{(m')}, v_{k'}^{(m')}\}$, which is impossible).
			
			\item[$(ii)$] constraints of the form  $\{\sigma_{m}(u^{(m)}_{j}) , \sigma_{m}(v^{(m)}_{j})\}  = \{\sigma_{m'}(u^{(m')}_{j'}) , \sigma_{m'}(v^{(m')}_{j'})\} $. For the same reasons as in case $(i)$, a constraint implying some $\{\sigma_{m}(u^{(m)}_{j}) , \sigma_{m}(v^{(m)}_{j})\}$ may happen at most once for each $m' \in [p], m' \neq m$.
			
			\item[$(iii)$] the last case is made of intersection of cases $(i)$ and $(ii)$, i.e. edges satisfying both constraints $\{u_j^{(m)}, v_j^{(m)}\}=\{u^{(m')}_{j'}, v^{(m')}_{j'}\}$ and $\{\sigma_{m}(u^{(m)}_{j}) , \sigma_{m}(v^{(m)}_{j})\}=\{\sigma_{m'}(u^{(m')}_{j'}) , \sigma_{m'}(v^{(m')}_{j'})\}$. This implies in particular that $\{u_j^{(m)}, v_j^{(m)}\}$ is an common fixed edge for $\sigma_{m}$ and $\sigma_{m'}$. By assumption \eqref{H3}, note that there cannot be a connected path of constraints of the form $(iii)$ of length greater or equal to $3$.
		\end{itemize}
		
		Let us now represent these constraints with a dependency graph. Each vertex a the graph represent one edge $\{u_j^{(m)}, v_j^{(m)}\}$ of the sum, that we will align column-wise according to $m \in [p]$. We put a plain (resp. dashed) edge between two nodes if they are enforced by constraint $(i)$ but not $(iii)$ (resp. $(ii)$ but not $(iii)$). Finally we draw a thick plain edge between two nodes if they are enforced by constraint $(iii)$.
		
		In view of discussion in points $(i)-(ii)-(iii)$, this dependency graph must be $p$-partite. Moreover, the subgraph made of plain thick or plain edges (resp. plain thick of dashed edges) only consists in a union of disjoint paths. The thick plain subgraph is only made of isolated edges and paths fo size $3$. Finally, transitivity of the equality relationship enables to draw any path in any order: we shall take the left to right order by convention (no backtracking). 
		
		We denote by $k_1$ (resp. $k_2$) the number of plain (resp. dashed) edges. We also denote track $k_3$ the number of thick plain isolated edges, and $k_4$ the number of thick plain isolated paths of length $2$. Figure \ref{fig:img_dependency} gives an example of such a dependency graph.
		
		\begin{figure}[h]
			\centering
			\includegraphics[scale=0.8]{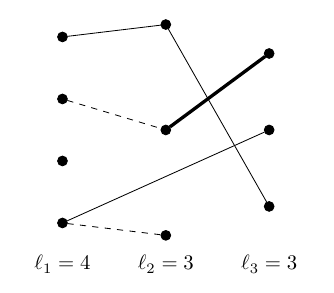}
			\caption{Example of a dependency graph, with $(k_1,k_2,k_3,k_4) = (3,2,1,0)$.}
			\label{fig:img_dependency}
		\end{figure} 
		
		In order to upper bound the contribution due to large terms, we must understand both the expectation of the product of indicators in \eqref{eq:expr_multi_mom} (this only depends on $(k_1,k_2,k_3,k_4)$), as well as the number of possible (labeled) dependency graphs with a given $(k_1,k_2,k_3,k_4)$.
		
		First, all plain (resp. dashed) dependency edge makes $1$ (resp. $1$) indicators disappear in the expectation (for any event $\cA, \one_\cA^2 = \one_\cA$). In the same way, all thick plain isolated edge (resp. thick plain isolated path of length $2$) makes $2$ (resp. $4$) indicators disappear
		the expectation for a given case with given $(k_1,k_2,k_3,k_4)$ is 
		\begin{equation}\label{eq:contrib_cas_patho}
		q^{2(\ell_1 + \ldots + \ell_p) - (k_1+k_2+2k_3+4k_4)} \leq C_1 n^{-2(\ell_1 + \ldots + \ell_p) + (k_1+k_2+2k_3+4k_4)}
		\end{equation} where $C_1$ is a constant depending on $\ell_1, \ldots, \ell_p$,
		
		Second, an upper bound for the number of possible (labeled) dependency graphs with a given $(k_1,k_2,k_3,k_4)$ can be established as follows. First, we have $k_1+k_2+k_3+2k_4$ equalities, leaving at most $\ell_1+\ldots+\ell_p - (k_1+k_2+k_3+2k_4)$ degrees of freedom in the choices of the edges. Moreover, we force $k_3$ of these edges to be common fixed edges between two (distinct) permutations, and $k_4$ of them to be common fixed edges between three (pairwise distinct) permutations. In view of hypotheses \eqref{H1} and \eqref{H2}, the number of possible (labeled) dependency graphs with a given $(k_1,k_2,k_3,k_4)$ is at most 
		\begin{flalign}\label{eq:nombre_cas_patho}
			\binom{k_1+k_2+k_3+k_4}{k_3+k_4}\card{\cS}^{\ell_1+\ldots+\ell_p - (k_1+k_2+k_3+2k_4) - k_3 -k_4} \times (n^{1+t})^{k_3} \times n^{t k_4} \nonumber \\ \leq C_2 n^{2(\ell_1+\ldots+\ell_p) - 2(k_1+k_2)-(3-t)k_3 -(6-t)k_4},
		\end{flalign} where $C_2$ is a constant depending on $\ell_1, \ldots, \ell_p$.
		
		Hence, in view of \eqref{eq:contrib_cas_patho} and \eqref{eq:nombre_cas_patho}, the total contribution of higher terms is upper bounded by
		\begin{multline*}
			 \sum_{s=1}^{\ell_1+\ldots+\ell_p} \sum_{k_1+k_2+k_3+2k_4=s} C_1 C_2 n^{-2(\ell_1 + \ldots + \ell_p) + (k_1+k_2+2k_3+4k_4)} n^{2(\ell_1+\ldots+\ell_p) - 2(k_1+k_2)-(3-t)k_3 -(6-t)k_4}\\
			\leq C_1 C_2 \sum_{s=1}^{\ell_1+\ldots+\ell_p} \sum_{k_1+k_2+k_3+2k_4=s} n^{-k_1} n^{-k_2} n^{-(1-t)k_3} n^{-(2-t)k_4} \\
			\leq C_1 C_2 \times (\ell_1+\ldots+\ell_p)\times (\ell_1+\ldots+\ell_p)^{4 (\ell_1+\ldots+\ell_p)} \times n^{-(1-t)} \underset{n \to \infty}{\longrightarrow} 0.
		\end{multline*} This last convergence concludes the proof.
	\end{proof}
	
	\section{Proofs of Lemmas}
	\addtocontents{toc}{\protect\setcounter{tocdepth}{0}}
	\subsection{Proof of Lemma \ref{lemma:controle_X}}\label{app:proof_control_X}
	\begin{proof}
		For the control of $X_\bT$ we follow classical computations made in \cite{Bollobas2001} to establish asymptotic behavior of $X_\bT$. For our purpose, we only need the two first moments. Assume that $\bT$ is of size $k = k(\bT) \leq K$, and that its automorphism group has $a =a(\bT)$ elements. Then, letting $\mu = \lambda s$,
		\begin{flalign*}
			\dE\left[X_\bT\right] &= \binom{n}{k} \times \frac{k!}{a} \times \left(\frac{\mu}{n}\right)^{k-1} \left(1-\frac{\mu}{n}\right)^{k(n-k)+\binom{k}{2}-k+1}.
		\end{flalign*} Indeed, we have $\binom{n}{k}$ choices for the nodes, then $\frac{k!}{a}$ ways of putting the edges. 
		Using $ \binom{n}{k} \sim \frac{n^k}{k!}$ and $\left(1-\frac{\mu}{n}\right)^{-k^2+\binom{k}{2}-k+1} \sim 1$ as soon as $k = o(\sqrt{n})$, we get
		\begin{flalign*}
			\dE\left[X_\bT\right] &\sim n \mu^{k-1}e^{-\mu k}/a.
		\end{flalign*} We now compute $\dE\left[X_\bT (X_\bT-1)\right] $ by classically counting the number of ordered pairs of distinct isolated tree components of $G \land G'$ isomorphic to $\bT$. This number is then multiplied by the probability of observing these two distinct isolated components. This gives
		\begin{flalign*}
			\dE\left[X_\bT (X_\bT-1)\right] &= \binom{n}{k} \binom{n-k}{k} \times \left(\frac{k!}{a}\right)^2 \times \left(\frac{\mu}{n}\right)^{2(k-1)} \left(1-\frac{\mu}{n}\right)^{2\left(k(n-2k) + \binom{k}{2}-k+1\right)} \left(1-\frac{\mu}{n}\right)^{k^2}.
		\end{flalign*}
		Here again, $k=o(\sqrt{n})$ gives that 
		\begin{flalign*}
			\dE\left[X_\bT (X_\bT-1)\right] &\sim n^2 \mu^{2(k-1)} e^{-2\mu k}/a^2.
		\end{flalign*}Denoting $\alpha = \alpha(\bT) := n \mu^{k-1}e^{-\mu k}/a(\bT)$, these computations give that $\dE\left[X_\bT\right] \sim \Var\left(X_\bT\right) \sim \alpha(\bT)$ when $n \to \infty$, uniformly in $k \leq K(n)$ as soon as $K(n) = o(\sqrt{n})$. Let us fix $\eps = \eps(n)>0$ small enough. Applying Chebyshev's inequality together with the union bound gives
		\begin{flalign*}
			\dP\left(\exists (k,\bT) \in [K(n)] \times \dT, X_{\bT} \leq (1-\eps)\alpha(\bT)\right) & \leq \sum_{k=1}^{K(n)} \sum_{\bT \in \dT_k} \dP\left(X_\bT - \dE\left[X_\bT\right] \leq (1- \eps) \alpha(\bT) - \dE\left[X_\bT\right] \right)\\
			& \overset{(a)}{\leq} \sum_{k=1}^{K(n)} \sum_{\bT \in \dT_k} \frac{\Var\left(X_\bT\right)}{\left((1- \eps) \alpha(\bT) - \dE\left[X_\bT\right]\right)^2} \\
			& \overset{(b)}{\leq} (1+o(1)) \sum_{k=1}^{K(n)} \sum_{\bT \in \dT_k} \frac{1}{\eps^2 \alpha(\bT)} \\
			&  \overset{(c)}{\leq} (1+o(1)) \sum_{k=1}^{K(n)} \sum_{\bT \in \dT_k} \frac{1}{\eps^2 n f(k)} \\
			& \overset{(d)}{\leq}(1+o(1)) K(n)^{K(n)} \frac{1}{\eps^2 n f(K(n))}, && \\
		\end{flalign*} where 
		\begin{equation}\label{eq:def_fk}
		f(k) := \frac{\mu ^{k-1}e^{-\mu k}}{k!}.
		\end{equation}
		We used in $(a)$ that all $ (1- \eps) \alpha(\bT) - \dE\left[X_\bT\right]$ are negative for $n$ large enough, in $(b)$ uniformity in $k \leq K(n)$, in $(c)$ the lower bound $nf(k)$ for $\alpha(T)$, and finally in $(d)$ that $k \mapsto f(k)$ is decreasing since $\mu e^{-\mu}<1$.
		
		Taking now e.g. $\eps = n^{-1/4}$, the last fact to check to establish the Lemma is that $K^{K}/f(K) = o(n^{1/2})$ when $K=K(n) = \log^{1/2}(n)$:
		\begin{flalign*}
			K^{K}/f(K) &= K^K K! (1/\mu)^{K-1} e^{\mu K}\\
			& \leq \exp\left(2 K \log K + (\log(1/\mu) + \mu) K\right)\\
			&=  \exp\left(\log^{1/2}(n) \log \log n + (\log(1/\mu) + \mu) \log^{1/2}(n)\right) = o(n^{1/2}).
		\end{flalign*} 
		
	\end{proof}

	\subsection{Proof of Lemma \ref{lemma:controle_overlap_ij}}\label{app:proof_overlap}
	\begin{proof}
		Denote $T_\infty := \card{V_\infty}$ and $T_> := \card{V_>}$. First notice that for any permutations $\sigma_i, \sigma_j$ with $i \neq j$ generated with Algorithm \ref{algo:rec_construction}, we have the following equality:
		\begin{equation}\label{eq:decompo_overlap}
		\ov(\sigma_i, \sigma_j) = T_\infty + T_> + \sum_{k=1}^{K(n)} \sum_{\bT \in \dT_k} k \cdot \ov(\Sigma^{(i)}_\bT, \Sigma^{(j)}_\bT),
		\end{equation} where $\Sigma^{(i)}_\bT$ (resp. $\Sigma^{(j)}_\bT$) is the tree permutation associated with $\bT$ in $\sigma_i$ (resp. in $\sigma_j$). We know that $T_\infty = c(\lambda s) n +o(n)$ w.h.p. and by Lemma \ref{lemma:small_trees}, $T_> =o(n)$ w.h.p.
		
		Define
		\begin{equation}
		\ov'(\sigma_i, \sigma_j) := \sum_{k=1}^{K(n)} \sum_{\bT \in \dT_k} k \cdot \ov(\Sigma^{(i)}_\bT, \Sigma^{(j)}_\bT),
		\end{equation} the second term in \eqref{eq:decompo_overlap}. We dominate $\ov'(\sigma_i, \sigma_j)$ as follows: \begin{lemma}\label{lemma:laplace_overlap}
			If $X=\ov(\Sigma^{(i)}_\bT, \Sigma^{(j)}_\bT),$, then for all $t \in \dR$,
			\begin{equation}\label{eq:lemma:laplace_overlap}
			\dE\left[e^{tX}\right] \leq \exp(e^t).
			\end{equation}
		\end{lemma}
		\begin{proof}
			\begin{flalign*}
				\dE\left[e^{tX}\right] & = \sum_{m \geq 0} e^{tm} \dP(X \geq m).
			\end{flalign*} Noting that $\dP(X \geq m) \leq \dE\left[\binom{X}{m}\right]$ and that
			\begin{flalign*}
				\dE\left[\binom{X}{m}\right] &= \frac{1}{m!} \dE\left[X(X-1)\ldots (X-m+1)\right]\\
				& = \frac{1}{m!} k(k-1)\ldots (k-m+1) \frac{(k-m)!}{k!} = \frac{1}{m!}
			\end{flalign*} gives
			\begin{flalign*}
				\dE\left[e^{tX}\right] & \leq \sum_{m \geq 0} \frac{e^{tm}}{m!}\leq \exp(e^t).
			\end{flalign*} 
		\end{proof}
		
		Using independence of the $X$ variables, Equation \eqref{eq:lemma:laplace_overlap} of Lemma \ref{lemma:laplace_overlap} give that for all $t \in \dR$,
		\begin{flalign}\label{eq:laplace_overlap_total}
			\dE\left[e^{t \cdot \ov'(\sigma_{i},\sigma_{j})}\right]& \leq \prod_{k=1}^{K(n)} \prod_{\bT \in \dT_k} \exp(e^{tk})  \leq \exp\left(e^{t K(n) } K(n)^{K(n)+1}\right).
		\end{flalign} Now, we use the classical Chernoff bound, for positive $t$,
		\begin{flalign*}
			\dP\left(\ov'(\sigma_i,\sigma_j) \geq n^{\alpha}\right) & \leq \exp\left(- tn^{\alpha} + e^{t K(n)} K(n)^{K(n)+1}\right) \\ 
			& \leq \exp\left(- \frac{n^{\alpha}}{K(n)} \left[\log \left(\frac{n^{1-\alpha}}{K(n)^{K(n)+2}}\right) - 1\right]\right),
		\end{flalign*} taking $t = \frac{1}{K(n)} \log \left(\frac{n^{\alpha}}{K(n)^{K(n)+2}}\right)$.
		The right hand side tend to $0$ for any $\alpha \in (0,1)$, and a simple use of the union bound ends the proof. 
	\end{proof}
	
	\subsection{Proof of Lemma \ref{lemma:controle_F}}\label{app:proof_control_F}
	\begin{proof}
		Fix $t>0$. We use a standard first moment method. We will use the results of Lemmas \ref{lemma:small_trees} and \ref{lemma:controle_X}, conditioning on the event $\cA$ where the corresponding results hold. Since $\dP(\cA) = 1-o(1)$, this conditioning is legitimate for our purpose. \\
		
		\proofstep{Step 1.} Let us first control the term $F(\cS_{out}, \sigma_{i_1},\ldots,\sigma_{i_r})$: edges of $\cS_{out}$ are made of exactly one vertex in $V_{\infty,>}$. There are at most $n^2$ such edges, and the probability for a given edge of $\cS_{out}$ being a common fixed edge of $\sigma_{i_1},\ldots,\sigma_{i_r}$ is $\frac{1}{X_\bT^{r-1}}$, which can be upper-bounded on $\cA$ by $(nf(K(n)))^{1-r} \leq n^{1-r+t/2}$ by Remark \ref{remark:f(K)}. 
		
		Edges of $\cS_{out}$ thus have a contribution in $\dE\left[F(\sigma_{i_1},\ldots,\sigma_{i_r}) | \cA \right]$ of at most $n^{3-r+t/2}$.
		
		\proofstep{Step 2.} In the edges appearing in $F(\sigma_{i_1},\ldots,\sigma_{i_r})$, we consider three cases:
		\begin{itemize}
			\item[$(i)$] edges of $\intra$: these are edges made with two vertices in the same tree $T \; \widehat{=} \; \bT \in \dT$. On event $\cA$, there are at most $$\sum_{k=1}^{K(n)} \sum_{\bT \in \dT_k} X_\bT k^2 \leq n K(n)$$ such edges. The probability for a given edge of $\intra$ made of vertices of  $\bT \in \dT$ being a common fixed edge of $\sigma_{i_1},\ldots,\sigma_{i_r}$ is $\frac{1}{X_\bT^{r-1}}$, which can be upper-bounded by $(nf(K(n)))^{1-r} \leq n^{1-r+t/2}$. Edges of $\intra$ thus have a contribution in $\dE\left[F(\sigma_{i_1},\ldots,\sigma_{i_r}) | \cA \right]$  of at most $n^{2-r+t/2}$.
			
			\item[$(ii)$] edges of $\interu$: these are edges made with two vertices $u,v$ in different trees $T \neq T'$ (but that may be $\sim$ to the same $\bT \in \dT$), and verifying $u \notsimt v$. There are at most $n^2$
			such edges. Since $u \notsimt v$, there are only one possibility to map two edges of $\interu$. The probability for a given edge of $\interu$ made of vertices of  $T \; \widehat{=} \; \bT, T' \; \widehat{=} \; \bT'$ being a common fixed edge is $\frac{1}{(X_\bT(X_\bT -1))^{r-1}}$, and edges of $\interu$ thus have a contribution in the expectation of at most $n^{4-2r+t/2}$.
			
			\item[$(iii)$] edges of $\interd$: these are edges similar to case $(ii)$, except that their endpoints belong necessarily to isomorphic trees, and verifying $u \simt v$. There are at most $n^2$ 
			such edges. Since $u \simt v$, there are two ways to map two edges of $\interd$. The probability for a given edge of $\interd$ made of vertices of  $T, T' \; \widehat{=} \; \bT$ being a common fixed edge is time $\left(\frac{2}{X_\bT(X_\bT -1)}\right)^{r-1}$, and edges of $\interd$ thus have a contribution in the expectation of at most $n^{4-2r+t/2}$.
		\end{itemize}
		
		\proofstep{Step 3.} The first two steps show that $\dE\left[F(\sigma_{i_1},\ldots,\sigma_{i_r}) | \cA \right] \leq C n^{3-r+t/2}$ for all $t>0$. Summing over all possible $r$-tuples of permutations, Markov inequality yields
		\begin{flalign*}
			\dP\left( \exists r \geq 4, \, \exists \sigma_{i_1}, \ldots, \sigma_{i_r} \mbox{ pairwise distinct}, \, F(\cS,\sigma_{i_1},\ldots,\sigma_{i_r}) \geq 1 \right) & \leq o(1) + \sum_{r=4}^{\infty} p^r C n^{3-r+t/2}\\
			& \leq C p^4 n^{t/2-1} \to 0,
		\end{flalign*} for $t$ small enough, and 
		\begin{flalign*}
			\dP\left( \exists \sigma_{i_1}, \sigma_{i_2}, \sigma_{i_3} \mbox{ pairwise distinct}, \, F(\cS,\sigma_{i_1}, \sigma_{i_2}, \sigma_{i_3}) \geq n^t \right) & \leq o(1) + p^3 \times C n^{-t/2} \to 0, 
		\end{flalign*}and
		\begin{flalign*}
			\dP\left( \exists \sigma_{i_1} \neq \sigma_{i_2}, F(\cS,\sigma_{i_1}, \sigma_{i_2}) \geq n^{1+t} \right) & \leq o(1) + p^2 \times C n^{-t/2} \to 0.
		\end{flalign*}
	\end{proof}
\addtocontents{toc}{\protect\setcounter{tocdepth}{2}}

\end{subappendices}

\chapter{From tree matching to sparse graph alignment}\label{chapter:NTMA}
In this chapter, we consider alignment of sparse graphs, for which we introduce the \emph{Neighborhood Tree Matching Algorithm} (\alg{NTMA}), based on a measure of similarity between trees. For correlated \ER random graphs, we prove that the algorithm returns -- in polynomial time -- a positive fraction of correctly matched vertices, and a vanishing fraction of mismatches. This result holds with average degree of the graphs in $O(1)$ and correlation parameter $s$ that can be bounded away from 1, conditions under which random graph alignment is particularly challenging. As a byproduct of the analysis we introduce a matching metric between trees and characterize it for several models of correlated random trees. These results may be of independent interest, yielding for instance efficient tests for determining whether two random trees are correlated or independent\footnote{This related problem first appeared in this contribution, and will be the focus of Chapter \ref{chapter:MPAlign}.}.\\

This chapter is based on the paper \textit{From tree matching to sparse graph alignment} \cite{Ganassali20a}, published at \emph{COLT 2020}, a joint work with L. Massoulié.

\section{Introduction}

As seen in the introduction (Section \ref{intro:subsection:short_survey}), previously existing methods for \ER graph alignment only succeed in a dense regime where the mean degree of the graphs is $\Omega(\log n)$. When the mean degree is constant, several phenomena occur -- degrees do not concentrate any more and the graph looses its connectivity (see Theorem \ref{intro:theorem:connectivity_ER}), among other things -- and make the performance of standard dense methods collapse. 

We recall in particular that results from \cite{Cullina2017,Cullina18} show that in the sparse regime, there is no hope of recovering $\pi^{\star}$ exactly or almost exactly, or in other words, of perfectly re-aligning $G$ and $H$. 
Nevertheless, their result does not rule out the possibility of partially recovering the unknown permutation $\pi^{\star}$. For the applications mentioned earlier in Section \ref{intro:subsection:motivations}, it is at the same time natural to assume that the graphs involved are sparse, and potentially useful to recover only a fraction of the unknown matches $(u,\pi^{\star}(u))$. 

This motivates the present work, whose goal is to  show that \emph{partial alignment of sparse correlated graphs is feasible}, and to introduce a polynomial-time algorithm for producing such partial alignments. 

We do not recall here the definition of the correlated \ER model, already introduced in the introduction (see \eqref{eq:CER_model}), and specified in Section \ref{impossibility:section:introduction} of Chapter \ref{chapter:impossibility} in the sparse case. We only recall that the parameters of $\G(n,\lambda/n,s)$ are the number of nodes $n$, the mean degree $\lambda >0$ and the correlation parameter $s \in [0,1]$. The vertices of the second graph $G'$ are relabeled with a uniform independent permutation $\pi^{\star} \in \cS_n$, and we observe $G$ and $H := G'^{\pi^{\star}}$.

\subsubsection{Notations}
Let us recall a few notations. For an undirected graph $G$, denote by $V(G)$ its set of vertices, $E(G)$ (resp. $\overrightarrow{E}(G):=\lbrace (u,v), \lbrace u,v \rbrace \in E(G) \rbrace$) its set of non-oriented (resp. oriented) edges. We use the notations $u \conn v$ if $\lbrace u,v \rbrace \in E(G)$ and $u \to v$ if $(u,v)\in \overrightarrow{E}(G)$. The usual graph distance in $G$ will be denoted $\delta_G$. For $u \in V(G)$, let $\mathcal{N}_G(u)$ denote the neighborhood -- the set of neighbors -- of $v$ in $G$, and $\deg_{G}(v)$ its degree. 

For $d \geq 1$ we also define $\cB_{G}(v,d)$ the set of vertices at (graph) distance at most $d$ from $v$, and $\cS_{G}(v,d):=\cB_{G}(v,d) \setminus \cB_{G}(v,d-1)$ the set of vertices at distance exactly $d$ from $v$. 

For a rooted tree $t$, we let $\rho(t)$ denote its root node. For any $u \in V(T)\setminus \{\rho(t)\}$, we let $\pi_{t}(u)$ denote the parent of node $u$ in $T$. For $d \geq 1$, we note $\cB_d(t)=\cB_{t}(\rho(t),d)$ and $\cL_d(t)=\cS_{t}(\rho(t),d)$.

We omit the dependencies in $G$ or $t$ of these notations when there is no ambiguity.

\subsubsection{Objectives and main result}
Our main result is the proposal of the so-called \emph{Neighborhood Tree Matching Algorithm} (\alg{NTMA} hereafter) together with the following
\begin{theorem}
	\label{NTMA:thm:theoreme_principal}
	For some $\lambda_0>1$, for all $\lambda \in (1,\lambda_0]$, there exists $s^*(\lambda)<1$ such that, provided $s \in (s^*(\lambda),1]$, for $(G,H) \sim \G(n,\lambda/n,s)$. \alg{NTMA} returns a matching $\cS = \cS(G,H)$ verifying the following properties with high probability:
	\begin{equation}
	|\cS\cap \{(u,\pi^{\star}(u)),\;u\in [n]\}|=\Omega(n) \quad \mbox{and} \quad |\cS\setminus\{(u,\pi^{\star}(u)),\; u \in [n]\}|=o(n) \, .
	\end{equation}
\end{theorem}
In words, our algorithm returns a set of node alignments which contains a negligible fraction of mismatches, and $\Omega(n)$ good matches, that is performs \emph{one-sided partial alignment} (see \ref{intro:subsection:heuristics_tree_graph}). This result covers values of $\lambda$ arbitrarily close to 1, and thus applies to very sparse graphs. For $\lambda s<1$, \ER graphs in our correlated model have connected components of size at most logarithmic in $n$, and we saw earlier on in Chapter \ref{chapter:impossibility} that there is no hope to recover a positive fraction of correct matches. This result can be interpreted as follows. For partial graph alignment of sparse \ER correlated random graphs, there is an ``easy phase'' that includes the parameter range $\{(\lambda,s):\lambda\in(1,\lambda_0],\; s\in(s^*(\lambda),1]\}$.

\subsubsection{Organization of the chapter}
The description of the Neighborhood Tree Matching Algorithm and the proof strategy for establishing Theorem \ref{NTMA:thm:theoreme_principal} are given in Section \ref{NTMA:section:sparse_graph_alignment}. 
Our algorithm relies essentially on a tree matching operation. To pave the way for Section \ref{NTMA:section:sparse_graph_alignment}, we introduce in Section \ref{NTMA:section:tree_matching} a notion of matching weight between trees that is key for our algorithm, and can be computed efficiently in a recursive manner. We further obtain probabilistic guarantees on the matching weights between random trees drawn according to some (correlated) Galton-Watson branching processes. These are instrumental in the proof of Theorem \ref{NTMA:thm:theoreme_principal}, but also of independent interest. Indeed we introduce in Section \ref{NTMA:section:tree_matching} a natural hypothesis testing problem on pairs of random trees, for which we obtain a successful test based on computation of tree matching weights. This last problem will next be the main focus of Chapter \ref{chapter:MPAlign}.

\section{Tree matching} \label{NTMA:section:tree_matching} 
In this section, we introduce the matching weight between rooted trees and the related matching rate. We then establish high probability bounds on the latter for (correlated) Galton-Watson random trees. We also give an application to a hypothesis testing problem of correlation detection in trees. 

\subsection{Matching weight of two rooted trees}
For any pair of rooted trees $(\tau,t)$, we say that a mapping $g: V(\tau)\to V(t)$ is \emph{tree-preserving} if 
\begin{itemize}
	\item $f(\rho(\tau))=\rho(t)$ (the root of $\tau$ is sent onto the root of $t$), and
	\item $\forall \, u \in V(\tau)\setminus\{\rho(\tau)\}, \; f(\pi_\tau(u))=\pi_{t}(f(u))$ (the parent of $u$ is matched on the match of its parent).
\end{itemize}

For any $d\geq 0$, let $\cA_d$ denote the collection of rooted trees whose leaves are all of depth $d$. Given two rooted trees $t$ and $t'$ of depth at most $d$, let $\set{t \cap t'}$ denote the collection of trees $\tau \in \cA_d$ such that there exist tree-preserving injective embeddings $f: V(\tau)\to V(t)$, $f': V(\tau)\to V(t')$. The \textit{matching weight of $t$ and $t'$ at depth $d$}, as introduced in Section \ref{intro:subsection:heuristics_tree_graph}, is defined as follows:
\begin{equation}\label{eq:NTMA:def:Wtt}
\cW_d(t,t'):=\sup_{\tau \in \set{t \; \cap \; t'}}\left| \cL_d(\tau) \right| \, ,
\end{equation} i.e. the size of the largest common subtree of $t$ and $t'$, measured in terms of the number of leaves at depth $d$.

\begin{remark}
	Note that by definition \eqref{eq:NTMA:def:Wtt}, 
	$$\cW_0(t,t')=1 \quad \mbox{and} \quad \cW_1(t,t')=\max \left( \deg_t(\rho(t)), \deg_{t'}(\rho(t'))\right) \, .$$ 
\end{remark}
	
\begin{figure}[H]
	\centering
	\vspace{-0.5cm}
	\includegraphics[scale=1]{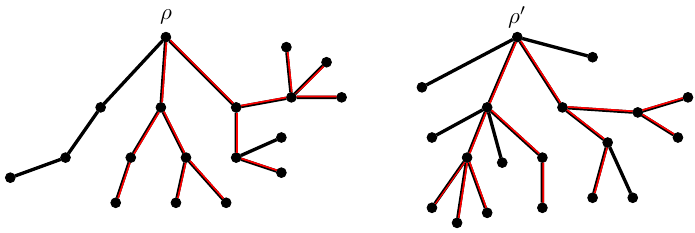}
	\caption{Example of two trees $t$, $t'$ with $\cW_3(t,t')=7$, where an optimal $t \in \cA_3$ is drawn in red.}
	\label{fig:NTMA:image_tree_matching_weight}
\end{figure}
Before going any further, we need to recall and introduce a few new notations. For a rooted tree $t$ of depth at most $d$, $u \in V(t)$, $t_u$ is the downstream subtree of $t$ re-rooted at $u$. More generally\footnote{Note that in tree $t$, $t_u = t_{u \leftarrow \rho(t)}$.} $u,v \in V(t)$ such that $v \to u$, $t_{u \leftarrow v}$ denotes the subtree of $t$ re-rooted at $u$ where edge $\set{u,v}$ has been removed, that is the subtree pointed by the oriented edge $v \to u$.

\begin{figure}[H]
	\centering
	\includegraphics[scale=1]{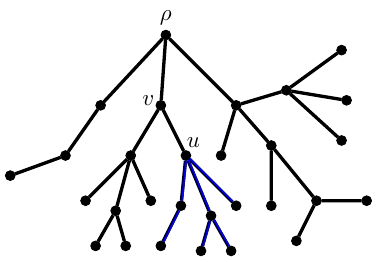}
	\caption{An example of a tree $t$ and its corresponding $t_{u \leftarrow v}$ highlighted in blue.}
	\label{fig:NTMA:image_rerooted_trees}
\end{figure}

	For a given pair of trees $t$ and $t'$ of depth at most $D$, 
	for pairs of vertices $(u,u'),(v,v') \in V(t) \times V(t')$ such that $v \to u$, $v' \to u'$, $t_{u \leftarrow v}$ and $t'_{u' \leftarrow v'}$ are of depth at most $d$, the \textit{matching weight of edges $v \to u$ and $v' \to u'$} is then defined as:
	\begin{equation}\label{eq:NTMA:def:Wuuvv}
	\cW_d(u \leftarrow v, u' \leftarrow v'):=\sup_{\tau \in \{t_{u \leftarrow v} \; \cap \; t'_{u' \leftarrow v'}\}}\left| \cL_d(\tau) \right| \, .
	\end{equation}

\begin{remark}
	Note that by definition \eqref{eq:NTMA:def:Wuuvv}, 
	$$\cW_0(u \leftarrow v, u' \leftarrow v')=1 \quad \mbox{and} \quad \cW_1(u \leftarrow v, u' \leftarrow v')=\max \left( \deg_t(u), \deg_{t'}(u')\right) -1 \, .$$ 
\end{remark}

\subsection{Recursive computation of $\cW_d$}\label{NTMA:subsection:recursionW}
From definition \eqref{eq:NTMA:def:Wuuvv}, doing a first step conditioning, i.e. distinguishing on the matching of pairs of nodes at depth $1$ in both trees, gives the following:
\begin{equation}
\label{eq:NTMA:rec_formula_Wuuvv}
\cW_d(u \leftarrow v, u' \leftarrow v') = \sup_{\mathfrak{m} \in \cM\left(\cN_t(u) \setminus \lbrace v \rbrace \, , \, \cN_{t'}(u') \setminus \lbrace v' \rbrace \right)} \sum_{(w,w') \in \mathfrak{m}} \cW_{d-1}(w \leftarrow u, w' \leftarrow u') \, ,
\end{equation} where for all sets $U,V$, $\cM\left(U,V\right)$ is the set of all partial matchings between $U$ and $V$, that is one-to-one mappings $\mathfrak{m}: U_0 \subseteq U \to V$. In the same way, for two trees $t$, $t'$ of depth at most $d$, we have 
\begin{equation}
\label{eq:NTMA:rec_formula_Wtt}
\cW_d(t,t') = \sup_{\mathfrak{m} \in \cM\left(\cN_t(\rho(t))\, , \, \cN_{t'}(\rho(t')) \right)} \sum_{(u,u') \in \mathfrak{m}} \cW_{d-1}(u \leftarrow \rho(t), u' \leftarrow \rho(t')) \, .
\end{equation}
These recursive formulae \eqref{eq:NTMA:rec_formula_Wtt} and \eqref{eq:NTMA:rec_formula_Wuuvv} show that matching weights at depth $d$ can be obtained by computing weights at depth $d-1$ and solving a linear assignment problem (LAP) \cite{Kuhn55}, and yield the following simple recursive algorithm to compute matching weights at depth $d$.

\begin{algorithm}[H]
	\caption{$\cW_d(u \leftarrow v, u' \leftarrow v') $}
	\label{NTMA:algo:Wuuvv}
	\SetAlgoLined
	\uIf{$d=0$}{
		\textbf{return} 1\;}
	\Else{$U \gets \cN_t(u) \setminus \lbrace v \rbrace$ \; 
		
		$V \gets \cN_{t'}(u') \setminus \lbrace v' \rbrace$ \; 
		\For{$(w,w') \in \cE \times \cF$}{
			Compute $\cW_{d-1}(w \leftarrow u, w' \leftarrow u' )$\;}
		Solve the LAP problem $W^* := \sup_{\mathfrak{m} \in \cM\left(U,V\right)} \sum_{(w,w') \in \mathfrak{m}} \cW_{d-1}(w \leftarrow u, w' \leftarrow u' )$\; 
		
		\textbf{return} $W^*$\;}
\end{algorithm}

\begin{remark}\label{rem:complex1}
It is easy to show that computing the matching weight $\cW_d(t,t')$ with the recursive algorithm \ref{NTMA:algo:Wuuvv} takes $O \left( d_\mathrm{max}^{2d}\right)$ time, where $d_\mathrm{max}$ is the maximal degree in $t$ and $t'$, which is not polynomial in $d$.

However, we can do better using dynamic programming, namely storing for all $k \in [d]$ the weights $\cW_{k}(e,e')$ in a array of size the number of pairs $(e,e’)$ where $e$ and $e’$ are two oriented edges in $t$,$t’$ (there are $4 \times |t| \times |t'|$ such pairs). Each time we increase $k$ and update the array, we solve one LAP for each pair $(e,e’)$, e.g. with the Hungarian algorithm that running in cubic time complexity \cite{Kuhn55}. The size of the -- small -- matrix on which the LAP is done does not exceed $d_\mathrm{max} \times d_\mathrm{max}$, hence updating the array from $k$ to $k+1$ is done in $O(|t| \times |t'|\times d_\mathrm{max}^3)$ steps. This gives a time complexity of $O\left(d \times |t| \times |t'|\times d_\mathrm{max}^3\right)$, which is better in general\footnote{Note that however, if $\card{t}, \card{t'} = \Theta(n^{\alpha})$ and $d_\mathrm{max}=O(\log n)$, which will be the case later in Section \ref{NTMA:section:sparse_graph_alignment}, then for small values of $d$ and large values of $n$, the recursive algorithm \ref{NTMA:algo:Wuuvv} is faster.}. 
\end{remark}
\subsection{Matching rate of random trees}
	For each $d \geq 0$, let us consider a pair of random trees $(T_d,T'_d)$ sampled according to some distribution $\mu_d$ (we will further introduce models in the sequel). The \textit{matching rate of the family of distributions $\set{\mu_d}_{d \geq 0}$} is defined as follows
	\begin{multline}
	\label{eq:NTMA:matching_rate}
	\gamma(\set{\mu_d}_{d \geq 0}) := \\ \inf \left\lbrace \gamma : \exists \, m,c,d_0>0, \; \forall x \geq 0, \; \forall d \geq d_0, \; \mu_d(\{(t,t') : \cW_{d}(t,t') \geq mx \gamma^d\}) \leq e^{-(x-c)_+}\right\rbrace \, .
	\end{multline} This important quantity \eqref{eq:NTMA:matching_rate} captures the asymptotic geometric growth rate of matching weights of random trees drawn under $\mu_d$. A simpler alternative definition could have been $$\widetilde{\gamma}(\set{\mu_d}_{d \geq 0}) := \inf \left\lbrace \gamma : \mu_d(\{(t,t') : \cW_{d}(t,t') \geq \gamma^d\}) \underset{d \to \infty}{\longrightarrow} 0\right\rbrace \, .$$  However, definition \eqref{eq:NTMA:matching_rate} better suits our purpose.

\begin{remark}\label{rem:NTMA:1.5}
	By definition,  note that for any $\gamma>\gamma(\set{\mu_d}_{d \geq 0})$, $\mu_d\left(\cW_{d}(t,t') \geq \gamma^d\right)$ converges to $0$ very fast, like $O\left( \exp \left(-c(\gamma)^d\right) \right)$ with $c(\gamma)>1$, so that $\widetilde{\gamma}(\set{\mu_d}_{d \geq 0}) \leq \gamma(\set{\mu_d}_{d \geq 0})$.
\end{remark}

\subsection{Models of random trees}\label{NTMA:subsection:models_trees}
We now introduce\footnote{Some of them are already mentioned in the introduction, see Section \ref{intro:section:cdt}.} three models of random trees that are relevant to sparse graph alignment.

\subsubsection{Galton-Watson trees with Poisson offspring}  The \emph{Galton-Watson tree with offspring $\Poi(\mu)$ up to depth $d$}, denoted by $\GWmu_d$, is defined recursively as follows. First, the distribution $\GWmu_0$ is a Dirac at the trivial tree, containing only the root. Then, for $d \geq 1$, sample a number $Z \sim \Poi(\mu)$ of independent $\GWl_{d-1}$ trees, and attach each of them as $c$ children of the root, to form a tree of depth at most $d$. 

\subsubsection{Independent model $\dPl_d$} Under the independent model $\dPl_{d}$, $t$ and $t'$ are two independent $\GWl_{d}$, where $\lambda>0$ is the mean number of children in the graph.

\subsubsection{Tree augmentation} For $\lambda >0$ and $s \in [0,1]$, a (random) \emph{$(\lambda,s)-$augmentation} of a given tree $\tau =(V,E)$, denoted by $\Augls_d(\tau)$, is defined as follows. First, to each node $u$ in $V$ of depth $<d$, we attach a number $Z^{+}_u$ of additional children, where the $Z^{+}_u$ are i.i.d. of distribution $\Poi(\lambda (1-s))$. Let $V^+$ be the set of these additional children. To each $v \in V^+$ at depth $d_v$, we attach another random tree of distribution $\GWl_{d-d_v}$, independently of everything else.

\subsubsection{Correlated shifted model $\dPlsd_d$} In the correlated shifted model $\dPlsd_{d}$ , the tree $T$ is rooted at $\rho$ and $T'$ is rooted at $\rho'$, and $\rho'$ is also a node of $T$, at distance $\delta$ from its root $\rho$. The two trees are generated as follows. First, all nodes $u$ in $T$ on the path from $\rho$ to the parent of $\rho'$ in $T$ have, besides their child leading to $\rho'$, extra $Z^{+}_u \sim \Poi(\lambda)$ children in $T$, and all extra child $v$ at depth $d_v$ has an additional offspring in $T$ sampled from $\GWl_{d-d_v}$. Then, sample an \emph{intersection tree} $\tau^\star \sim \GWls_{d-\delta}$ starting from $\rho'$. Independently, we finish the construction of $T$ (resp. of $T'$) with a  $(\lambda,s)-$augmentation of $\tau^\star$ of depth $d-\delta$ (resp. of depth $d$). See Figure \ref{fig:NTMA:image_GW} for an illustration. We denote $(T,T') \sim \dPlsd_d$.

\subsubsection{Correlated model $\dPls_d$} It is the previous model with $\delta=0$, so that the two correlated trees $T$ and $T'$ have same root $\rho$. We denote $(T,T') \sim \dPls_d$.
In other words, the correlated model $\dPls_{d}$ is built as follows: starting from an \emph{intersection tree} $\tau^\star \sim \GWls_d$, and $T$ and $T'$ are obtained as two independent $(\lambda,s)-$augmentations of $\tau^\star$. We denote $(T,T') \sim \dPls_d$.

In all these models, the labels of the trees $T$ and $T'$ are always forgotten, or randomly uniformly re-sampled. We however still distinguish the roots af the two trees. It can easily be verified that the marginals of $T$ and $T'$ are the same under $\dPl_{d}$ and $\dPls_{d}$, namely $\GWl_d$. The parameters are $\lambda$, the mean number of children of a node, and the correlation $s$. 

\begin{figure}[h]
	\centering
	\includegraphics[scale=0.9]{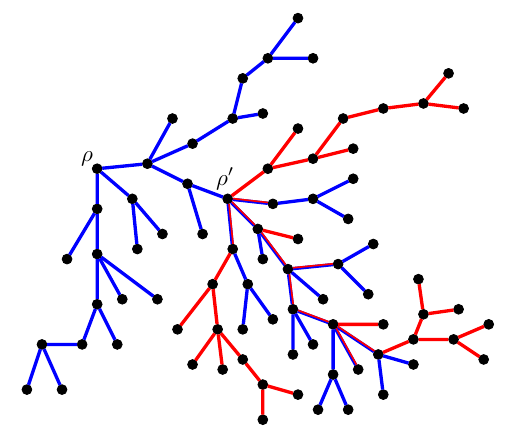}
	\caption{Random trees $T$ (blue) and $T'$ (red) from model $\dPlsd_d$ with $\delta=3$.}
	\label{fig:NTMA:image_GW}
\end{figure}

We now move to the analysis of matching rates for these models, which as explained before are crucial quantities which can help discriminate between the independent and the correlated setting.

\subsection{Matching rate of independent and correlated Galton-Watson trees}

\begin{proposition}
	\label{NTMA:prop:exponent_correlated}
	Let $\lambda>1$ and $s \in [0,1]$ such that $\lambda s >1$. Then, letting $\gamma(\lambda,s):=\gamma(\{\dPls_d\}_{d \geq 0})$, we have:
	\begin{equation*}
	\gamma(\lambda,s) \geq \lambda s.
	\end{equation*}
\end{proposition}

\begin{proof}
	Let $\tau^\star$ be the intersection tree between $T$ and $T'$. Branching process theory implies that $(\lambda s)^{-d}\big|\cL_d(\tau^\star)\big|$ converges almost surely to a random variable $Z$ as $d\to\infty$, such that $\dP\left(Z>0\right)=1-\pext$, with $\pext$ the extinction probability of the branching tree $\tau^\star$. Since $\pext<1$ when $\lambda s>1$, and for every small enough $\eps>0$, $$\lim_{d\to\infty}\dPls_d\left(\cW_d(T,T')\geq \left(\lambda s(1-\eps)\right)^d\right) \geq 1-\pext>0 \, ,$$ the result follows. 
\end{proof}

\begin{theorem}\label{NTMA:thm:lambda_close_to_1}
	Let $\gamma(\lambda):=\gamma(\{\dPl_d\}_{d \geq 0})$. There exists $\lambda_0>1$ such that for all $\lambda \in (1,\lambda_0]$, we have
	\begin{equation}
	\gamma(\lambda)<\lambda.
	\end{equation}
\end{theorem}
Evaluations of $\gamma(\lambda)$ by simulations, confirming and illustrating Theorem \ref{NTMA:thm:lambda_close_to_1}, are provided in Appendix \ref{app:tree}.

\textit{Outline of proof of Theorem \ref{NTMA:thm:lambda_close_to_1}.} The full proof of Theorem \ref{NTMA:thm:lambda_close_to_1} is detailed in Appendix \ref{appendix_proof_lambda_close_to_1}, but we here give the key steps. We introduce some notations. First, for a tree $t$ of depth at most $d$, let $r_d(t)$ denote the tree obtained by
iteratively pruning leaves of depth strictly less than $d$. When computing $\cW_d(t,t')$, the only informative subtrees are precisely in $r_d(t)$ and in $r_d(t')$, one of these being empty if $t$ or $t'$ doesn't survive up to depth $d$. In the rest of the chapter, we define $T_d$ the random variable $r_d(T)$ where $T$ is conditioned to survive up to depth $d$.

Consider $(T,T') \sim \dPl_{d}$. We let $\cE_d$ (respectively, $\cE'_d$) denote the event that tree $T$ (respectively, $T'$) becomes extinct before $d$ generations, i.e. $\cL_d(T)=\varnothing$ (respectively, $\cL_d(T')=\varnothing$). We let $p_{d}=\dP(\cE_d) = \dP(\cE'_d)$. It is well known that it satisfies the recursion
$$
p_{0}=0,\; p_{d}=e^{-\lambda(1-p_{d-1})}.
$$  We now state a lemma on the structure of $T_d$.
\begin{lemma}
	\label{NTMA:lemma:pruning_trees}
	For any $\lambda>1$, $T_d$ can be constructed by first sampling the number of children $D$ of the root $\rho(T)$ according to distribution
	\begin{equation*}
	\dP(D=k)=\one_{k>0}\frac{\dP(\Poi(\lambda(1-p_{d-1}))=k)}{\dP(\Poi(\lambda(1-p_{d-1}))>0)}=: q_{d,k},
	\end{equation*}
	and then attaching $D$ independent copies of $T_{d-1}$ to the $D$ children of $\rho(T)$.
\end{lemma}
\begin{proof}[Proof of Lemma \ref{NTMA:lemma:pruning_trees}]
	For a tree $t$, we identify $t$ to $(t_1,\ldots,t_k)$ the tuple of offsprings of its $k$ children. Write, defining $D$ the number of children of $\rho(T)$, fixing $k \geq 1$, $t_1,\ldots, t_k \in\cA_{d-1}$, and letting $S=i_1<\ldots<i_k$ run over all $k$ subsets of $[\ell]$:
	\begin{flalign*}
	\dP(T_d=(t_1,\ldots, t_k))&=\sum_{\ell \geq 0}\dP(T_d=(t_1,\ldots,t_k),D=\ell)\\
	&=\sum_{\ell \geq k}\sum_{S}\dP\left(D=\ell, r_d(T_{i_j})=t_j,j\in[k], r_d(T^v)=\varnothing, v\notin S \big |\overline{\cE}_d\right)\\
	&=\frac{1}{1- p_{d}}\sum_{\ell \geq k}\binom{\ell}{k}e^{-\lambda}\frac{\lambda^{\ell}}{\ell!}p_{d-1}^{\ell-k}\prod_{j=1}^k\dP(T_{d-1}=t_j) (1-p_{d-1})\\
	&=\frac{1}{1- p_{d}}\frac{(\lambda(1-p_{d-1}))^k}{k!}\prod_{j=1}^k\dP(T_{d-1}=t^j)\sum_{\ell\geq k}e^{-\lambda}\frac{(\lambda p_{d-1})^{\ell -k}}{(\ell-k)!}\\
	&=\frac{e^{-\lambda(1- p_{d-1})}}{1- p_{d}}\frac{(\lambda(1-p_{d-1}))^k}{k!}\prod_{j=1}^k\dP(T_{d-1}=t_j) \, .
	\end{flalign*}
	The conclusion follows by noting that $1-p_{d}=1-e^{-\lambda(1-p_{d-1})}$.
\end{proof} Assume $\eps = \lambda - 1$ to be small enough. Fix $r \in (0,1)$, let $\gamma=1+r\eps$. We first show using exponential moments that there exist $m,c>0$ and $d_0 >0$ such that for all $x>0$ $$\dP \left(\cW_{d_0}\left(T_{d_0},T'_{d_0}\right) \geq m x\right) \leq e^{-x+c}.$$
Then we define the random variables
$$X_d := \gamma^{-(d-d_0)} m^{-1} \cW_d\left(T_{d},T'_{d}\right).$$
Then, considering the number $D$ of children of the root in $T_{d}$ (resp. $D'$ in $T'_{d}$), using the previous lemma, one can establish, for all $x>0$, a recursive formula of the following form
\begin{equation*}
\dP\left( X_d \geq x \right) \leq \sum_{k,\ell \geq 1} q_{d,k} q_{d,\ell} \dP \left(\exists \mathfrak{m} \in \cM\left([k],[\ell]\right), \; \sum_{(i,u) \in \mathfrak{m}} X_{d-1,i,u} \geq \gamma x \right),
\end{equation*} where the $X_{d-1,i,u}$ are i.i.d. copies of $X_{d-1}$. The union bound yields
\begin{equation*}
\dP\left( X_d \geq x \right) \leq \sum_{k,\ell \geq 1} q_{d,k} q_{d,\ell} \min \left( 1, (k \vee \ell)^{\underline{k \wedge \ell}} \times \dP \left( \sum_{i=1}^{k \wedge \ell} X_{d-1,i,u} \geq \gamma x \right)\right),
\end{equation*} where $m^{\underline{p}}:= m(m-1)\ldots(m-p+1) = \frac{m!}{(m-p)!}$. This inequality enables, with a few more technical steps (see \ref{appendix_proof_lambda_close_to_1}), to propagate recursively the inequality
\begin{equation*}
\dP\left( X_d \geq x \right) \leq e^{-(x-c)_{+}}.
\end{equation*}

\subsection{Implications for a  hypothesis testing problem}
Let a pair of trees $(T,T')$ be distributed according to $\dPl_{d}$ under the null hypothesis $\cH_0$, and according to $\dPls_d$ under the alternative hypothesis $\cH_1$. They are thus independent under $\cH_0$, and correlated under $\cH_1$. Consider the following test:
$$
\hbox{Decide }\cH_0\hbox{ if }\cW_d(T,T') < \gamma^d,\; \cH_1\hbox{ otherwise.}
$$
Assume that $\gamma(\lambda)<\gamma<\lambda s$. Then in view of Remark~\ref{rem:NTMA:1.5} and Theorem~\ref{NTMA:thm:lambda_close_to_1} one has for some $c(\gamma)>1$:
$$
\dP\left(\hbox{decide }\cH_1\big|\cH_0\right)=O(e^{-c(\gamma)^d}),
$$
thus a super-exponential decay of the probability of false positive (first type error). Conversely, in view of Proposition \ref{NTMA:prop:exponent_correlated}, noting $\tau^\star$ the intersection tree under $\cH_1$, one has
$$
\dP\left(\hbox{decide }\cH_0\big|\cH_1, \hbox{non-extinction of } \tau^\star \right)=o_d(1).
$$
The false negative probability of this test thus also goes to zero, provided the intersection tree survives. As we will see in next section, this hypothesis testing problem on a pair of random trees is related to our original graph alignment problem much as the so-called tree reconstruction problem, reviewed in \cite{Mossel03}, is related to community detection in sparse random graphs (see e.g. \cite{Bordenave15}). This fundamental correspondence is studied in detail in Chapter \ref{chapter:MPAlign}.

\subsection{Matching rate of correlated shifted trees}
\begin{theorem}
	\label{NTMA:thm:lambda_close_to_1_delta}
 Let $\gamma(\lambda,s,\delta):=\gamma(\{\dPlsd_d\}_{d \geq 0})$. There exists $\lambda_0>1$ such that for all $\lambda \in (1,\lambda_0]$ we have
	\begin{equation}
	\sup_{\delta \geq 1} \gamma(\lambda,s,\delta) < \lambda.
	\end{equation}
\end{theorem}
Evaluations of $\gamma(\lambda,s,\delta)$ by simulations, confirming and illustrating Theorem \ref{NTMA:thm:lambda_close_to_1_delta}, are provided in Appendix \ref{app:tree}.

\textit{Outline of proof of Theorem \ref{NTMA:thm:lambda_close_to_1_delta}.} The full proof of Theorem \ref{NTMA:thm:lambda_close_to_1_delta} is detailed in Appendix \ref{appendix_proof_lambda_close_to_1_delta}, but we here give the key steps. The proof will again be by induction on $d$, the initial step being established with the same argument as in the proof of Theorem \ref{NTMA:thm:lambda_close_to_1}. 
The difference $\eps = \lambda -1$ is assumed to be small enough. We fix $r \in (0,1)$, and we let $\gamma=1+r\eps'$. We now work with the random variables
$$X'_d := \gamma^{-(d-d_0)} m^{-1} \cW_d\left(T_{d},T'_d\right),$$ conditionally on the event that the path from $\rho$ to $\rho'$ survives down to depth $d$ in $T$. Then, considering $D$ the number of children of $\rho$ in $T_{d}$, $D'$ the number of children of $\rho'$ in $T'_{d}$ that are in the intersection tree $T_d \cap T'_{d}$, and $D''$ the number of children of $\rho'$ in $T'_{d} \setminus T_d$, we establish for all $x>0$ a recursive formula of the following form
\begin{equation*}
\dP\left( X'_d \geq x \right) \leq \sum_{k,\ell \geq 1} \dP \left( D'+D''= k, D=\ell\right) \min \left( 1, (k \vee \ell)^{\underline{k \wedge \ell}} \, \dP \left( X'_{d-1} + \sum_{i=1}^{k \wedge \ell-1} X_{d-1,i,u} \geq \gamma x \right)\right),
\end{equation*}  where the $X_{d-1,i,u}$ are i.i.d. copies of $X_{d-1}$ as defined in the proof of Theorem \ref{NTMA:thm:lambda_close_to_1}. Again, with a few more technical steps (see \ref{appendix_proof_lambda_close_to_1_delta}), we are able to propagate recursively the inequality
\begin{equation*}
\dP\left( X'_d \geq x \right) \leq e^{-(x-c)_{+}}.
\end{equation*}

\section{Sparse graph alignment by matching trees}\label{NTMA:section:sparse_graph_alignment} 

We now describe our main algorithm and its theoretical guarantees. For simplicity we assume  that the underlying permutation $\pi^{\star}$ is the identity.

\subsection{Neighborhood Tree Matching Algorithm (\alg{NTMA}), main result}

The main intuition for the \alg{NTMA} algorithm is as follows. In order to distinguish matched pairs of nodes $(u,u')$, we consider their neighborhoods at a certain depth $d$, that are close to Galton-Watson trees. In the case where the two vertices are actual matches, the largest common subtree measured in terms of children at depth (exactly) $d$ is w.h.p. of size $\geq (\lambda s)^d $. However, when the two nodes $u$ and $u'$ are sufficiently distant in the union graph aligned with the ground truth, $G \cap G'$, previous study of matching rates shows that the growth rate of largest common subtree will be $< \lambda s$.  The natural idea is thus to apply the test comparing $\cW_d(\cB_{G}(u,d),\cB_{H}(u',d))$ to $\gamma^d$ for some well-chosen $\gamma$ to decide whether $u$ is matched to $u'$. 

However, as the reader may have noticed, testing $\cW_d(\cB_{G}(u,d),\cB_{H}(u',d)) > \gamma^d$ is not enough, because two-hop neighbors in $G \cap G'$ would dramatically increase the number of incorrectly matched pairs, making the performance collapse. To fix this, we use the \textit{dangling trees trick}: instead of just looking at their neighborhoods, we look for the downstream trees from two distinct neighbors $v \neq w$ of $u$, and $v' \neq w'$ of $u'$. The trick is now to compare both $\cW_{d-1}(v \leftarrow u,v' \leftarrow u')$  and $\cW_{d-1}(w \leftarrow u,w' \leftarrow u')$ to $\gamma^{d-1}$. This way, even if $u \neq u'$ and $u$ and $u'$ are close by, the pairs of rooted trees that can be considered will lead to one of the four cases considered and illustrated on Figure \ref{fig:NTMA:parrallel_construction_bis}, that are settled in the proof of Theorem \ref{NTMA:thm:no_mismatchs}. 

Our algorithm is as follows, where matching tree weights $\cW_{d-1}(j \leftarrow i, v \leftarrow u )$ are defined in \eqref{eq:NTMA:def:Wuuvv}:

\begin{algorithm}[H]
	\caption{\label{algo_theorique}Neighborhood Tree Matching Algorithm for sparse graph alignment}
	\SetAlgoLined
	
	\textbf{Input:} Two graphs $G$ and $H$ of size $n$, average degree $\lambda$, depth $d$, parameter $\gamma$.
	
	\textbf{Output:} A set of pairs $\cS \subset V(G) \times V(H)$.
	
	$\mathcal{S} \gets \varnothing$
	
	\For{$(u,u') \in V(G) \times V(H)$}{
		\If{$\cB_{G}(u,d)$ and $\cB_{H}(u',d)$ contain no cycle, and $\exists v \neq w \in \mathcal{N}_{G}(u), \exists v' \neq w' \in \mathcal{N}_{H}(u')$ such that $\cW_{d-1}(v \leftarrow u,v' \leftarrow u')> \gamma^{d-1} $ and $\cW_{d-1}(w \leftarrow u,w' \leftarrow u')> \gamma^{d-1} $}{	
			$\cS\gets \cS \cup \left\lbrace (u,u') \right\rbrace $
		}
	}
	\textbf{return} $\cS$
	
\end{algorithm}

\begin{remark}
	For $d = \lfloor c \log n \rfloor$, in view of Remark \ref{rem:complex1},  with high probability the complexity of \alg{NTMA} is $$O \left( \left|V(G)\right| \left|V(H)\right| (\log n)^2 n^{2c \log \lambda}  d_\mathrm{max}^2  \right) + O \left( \left|E(G)\right| \left|E(H)\right| (\log n)  d_\mathrm{max}^3  \right),$$ where $d_\mathrm{max}$ is the maximum degree in $G$ and $H$. In the context of Theorems \ref{NTMA:thm:lot_of_matchs} and \ref{NTMA:thm:no_mismatchs} the complexity is then 
	$O \left( (\log n)^4 n^{5/2}  \right)$.
\end{remark}
The two results to follow will readily imply Theorem \ref{NTMA:thm:theoreme_principal}.
\begin{theorem}
	\label{NTMA:thm:lot_of_matchs}
	Let $(G,H) \sim \G(n,\lambda/n,s)$ be $s-$correlated \ER graphs such that $\lambda s>1$. Let $d = \lfloor c \log n \rfloor$ with $c \log \left(\lambda\left(2-s\right)\right)<1/2$. Then for $\gamma\in(1,\lambda s)$, with high probability, if $\cS$ denotes the matching returned by \alg{NTMA},
	\begin{equation}
	\label{lot_of_matchs_eq}
	\frac{1}{n} \sum_{u \in [n]} \one_{\lbrace (u,u) \in \cS \rbrace}=\Omega(1).
	\end{equation} 
	In other words, a non vanishing fraction of nodes is correctly recovered by \alg{NTMA} (Algorithm \ref{algo_theorique}).
\end{theorem}

\begin{theorem}
	\label{NTMA:thm:no_mismatchs}
	Let $(G,H) \sim \G(n,\lambda/n,s)$ be two $s-$correlated \ER graphs. Assume that $\gamma_0(\lambda):= \max \left(\gamma(\lambda),\sup_{\delta \geq 1} \gamma(\lambda,s,\delta) \right)<\lambda s$, and that $d = \lfloor c \log n \rfloor$ with $c \log \lambda<1/4$. Then for $\gamma \in(\gamma_0(\lambda),\lambda s)$, with high probability,
	\begin{equation}
	\label{no_mismatchs_eq}
	\mathrm{err}(n):=\frac{1}{n}\sum_{u=1}^{n} \one_{\lbrace \exists u'  \neq u, \; (u,u') \in \cS \rbrace}=o(1) \, ,
	\end{equation} 
	i.e. only at most a vanishing fraction of nodes are incorrectly matched by \alg{NTMA} (Algorithm \ref{algo_theorique}).
\end{theorem}

\begin{remark}
	The set $\cS$ returned by the \alg{NTMA} is not necessarily a matching. Let $\cS'$ be obtained by removing all pairs $(i,u)$ of $\cS$ such that $i$ or $u$ appears at least twice. Theorems \ref{NTMA:thm:lot_of_matchs} and \ref{NTMA:thm:no_mismatchs} guarantee that $\cS'$ still contains a non-vanishing number of correct matches and a vanishing number of incorrect matches. Theorem \ref{NTMA:thm:theoreme_principal} easily follows. Simulations of \alg{NTMA--2}, a  simple variant of  of \alg{NTMA}, are reported  in Appendix \ref{appendix_simulations_simple_variant}. These confirm our theory, as the algorithm returns many good matches and few mismatches.  
\end{remark}

\subsection{Proof of Theorems \ref{NTMA:thm:lot_of_matchs} and \ref{NTMA:thm:no_mismatchs}}

We start by stating Lemmas,  adapted from \cite{Massoulie14} and \cite{Bordenave15} and proven in Appendix \ref{appendix_proof_lemmas_sec_2}, that are instrumental in the proofs of 
Theorems \ref{NTMA:thm:lot_of_matchs} and \ref{NTMA:thm:no_mismatchs}.

\begin{lemma}[Control of the sizes of the neighborhoods]
	\label{NTMA:lemma:control_S}
	Let $G \sim \G(n,\lambda/n)$, $d = \lfloor c \log n \rfloor$ with $c \log \lambda <1$. For all $\gamma>0$, there is a constant $C=C(\gamma)>0$ such that with probability $1-O\left(n^{-\gamma}\right)$, for all $u \in [n]$, $t \in [d]$:
	\begin{equation}
	\label{control_S_eq}
	\left| \cS_{G}(u,t) \right| \leq C (\log n) \lambda^t.
	\end{equation}
\end{lemma}

\begin{lemma}[Cycles in the neighborhoods in an $ER$ graph]
	\label{NTMA:lemma:cycles_ER}
	Let $G \sim \G(n,\lambda/n)$, $d = \lfloor c \log n \rfloor$ with $c \log \lambda <1/2$. There exists $\eps>0$ such that for any vertex $u \in [n]$, one has
	\begin{equation}
	\label{cycles_ER_eq}
	\dP\left(\cB_G(u,d) \mbox{ contains a cycle}\right) = O\left( n^{-\eps}\right).
	\end{equation}
\end{lemma}

\begin{lemma}[Two logarithmic neighborhoods are typically size-independent]
	\label{NTMA:lemma:indep_neighborhoods}
	Let $G \sim \G(n,\lambda/n)$ with $\lambda >1$, $d = \lfloor c \log n \rfloor$ with $c \log \lambda < 1/2 $. Then there exists $\eps>0$ such that for any fixed nodes $u \neq v$, the variation distance between the joint law of the neighborhoods $\mathcal{L} \left(\left(\cS_{G}(u,t),\cS_{G}(v,t)\right)_{t \leq d}\right)$ and the product law $\mathcal{L} \left(\left(\cS_{G}(u,t)\right)_{t \leq d}\right) \otimes \mathcal{L} \left(\left(\cS_{G}(v,t)\right)_{t \leq d}\right)$ tends to $0$ as $O\left(n^{-\eps}\right)$ for some $\eps>0$ when $n \to \infty$.
\end{lemma}

\begin{lemma}[Coupling the $\left|\cS_{G}\left(i,t\right)\right|$ with a Galton-Watson process]
	\label{lemma:NTMA:coupling_GW}
	Let $G \sim \G(n,\lambda/n)$, $d = \lfloor c \log n \rfloor$ with $c \log \lambda<1/2$. For a fixed $u \in [n]$, the variation distance between the law of $\left( \left|\cS_{G}(u,t)\right|\right)_{t \leq d}$ and the law of $\left( Z_t\right)_{t \leq d}$ where $(Z_t)_t$ is a Galton-Watson process of offspring distribution $\Poi(\lambda)$ tends to 0 as $O\left( n^{-\eps}\right)$ when $n \to \infty$.
\end{lemma}

\subsubsection{Proof of Theorems \ref{NTMA:thm:lot_of_matchs} and \ref{NTMA:thm:no_mismatchs}}
\begin{proof}[Proof of Theorem \ref{NTMA:thm:lot_of_matchs}]
	Define the joint graph $G_{\cup} = G \cup G'$. We recall that we assume that $\pi^\star=\id$, without loss of generality, hence $H=G'$. For $u \in [n]$, let $M_u$ denote the event that the algorithm matches $u$ in $G$ with $u$ in $H$, i.e. on which $\mathcal{B}_{G}(u,d)$ and $\mathcal{B}_{H}(u,d)$ contain no cycle, and $\exists v \neq w \in \mathcal{N}_{G}(u), \exists v' \neq w' \in \mathcal{N}_{H}(u)$ such that $\cW_{d-1}(v \leftarrow u,v' \leftarrow u)> \gamma^{d-1} $ and $\cW_{d-1}(w \leftarrow u,w' \leftarrow u)> \gamma^{d-1} $. Denote by ${C}_{\cup,u,d}$ the event that there is no cycle in $\mathcal{B}_{G_{\cup}}(u,d)$. 
	
	With the same arguments as in the proof of Lemma \ref{lemma:NTMA:coupling_GW}, the two neighborhoods  $\mathcal{B}_{G}(u,d)$ and $\mathcal{B}_{H}(u,d)$ can be coupled with trees distributed as $\dPls_{d}$ of Section \ref{NTMA:section:tree_matching}. However, we will instead consider the intersection graph $G_{\cap} = G \cap H$. Obviously,  $G_{\cap} \sim \G(n,\lambda s/n)$. By Lemma \ref{lemma:NTMA:coupling_GW}, the random variables $|\mathcal{S}_{G_{\cap}}(u,t)|$ can be coupled with a Galton-Watson process with offspring distribution $\Poi(\lambda s)$ up to depth $t=d$. Let $P_u$ denote the event that this coupling succeeds. Since $\lambda s>1$, there is a probability $2 \alpha >0$ that the first generation has at least two children whose offsprings survive up to depth $d-1$. Note $S$ this event. On event  $S$, the matching given by the identity on  the intersection tree implies the existence of two neighbors $ v \neq w \in \mathcal{N}_{G}(u)$ and $ v' \neq w' \in \mathcal{N}_{H}(u)$ such that with high probability $\cW_{d-1}(v \leftarrow u,v' \leftarrow u)> \gamma^{d-1}$ and $\cW_{d-1}(w \leftarrow u,w' \leftarrow u)> \gamma^{d-1}$, by standard martingale arguments, as in Proposition \ref{NTMA:prop:exponent_correlated}. This gives the lower bound for $\dP(M_u)$:
	$$
	\dP(M_u) \; \geq \dP\left({C}_{\cup,u,d} \cap P_u \cap S \right)\; \geq 2\alpha -o(1) > \alpha > 0.
	$$
	It is easy to see that $G_{\cup} \sim \G(n,\lambda(2-s)/n)$. For $u \neq v \in [n]$, define $I_{u,v}$ the event on which the two neighborhoods of $u$ and $v$ in $G_{\cup}$ coincide with their independent couplings up to depth $d$. By lemma \ref{NTMA:lemma:indep_neighborhoods}, $\dP(I_{u,v})=1-o(1)$. Then for $0<\eps<\alpha$ Markov's inequality yields
	
	\begin{flalign}
	\dP\left(\frac{1}{n} \sum_{u \in [n]} \one_{\lbrace (u,u) \in \cS \rbrace}<\alpha-\eps\right) & \leq \dP\left(\sum_{u \in [n]} \left(\dP(M_u)-\one_{M_u}\right)>\eps n\right)\\
	& \leq \frac{1}{n^2 \eps^2} \left(n \mathrm{Var}\left(\one_{M_1}\right)+ n(n-1)\mathrm{Cov}\left(\one_{M_1},\one_{M_2}\right) \right)\\
	& \leq \frac{\mathrm{Var}\left(\one_{M_1}\right)}{n \eps^2} + \frac{1-\dP\left(I_{1,2}\right)}{ \eps^2} \to 0. 
	\end{flalign}
	
\end{proof}

\begin{proof}[Proof of Theorem \ref{NTMA:thm:no_mismatchs}]
	
Define
	\begin{equation*}
	d_{\mathrm{max}}:=\max\left(\max_u \mathrm{deg}_{G}(u), \max_{u'} \mathrm{deg}_{H}(u')\right).
	\end{equation*} 
	We use the same notations as in the former proof: $G_{\cup} = G \cup H$ and $G_{\cap} = G \cap H$. Fix $u \in [n]$. In the rest of the proof we work conditionally to the event $C_{\cup,u,2d}$ that $\mathcal{B}_{G_{\cup}}(u,2d)$ has no cycle. Since $c \log \lambda<1/4$, $\dP\left( C_{\cup,u,2d} \right)=1-o(1)$ by Lemma \ref{NTMA:lemma:cycles_ER}. 
	
	Fix another vertex $u' \neq u$. The $d-$neighborhoods $\cB_{G}(u,d)$ and $\cB_{H}(u',d)$ have offspring distribution stochastically dominated by $\Bin(n,\lambda/n)$, which is also dominated by $\Poi(\lambda')$ as soon as $\lambda'=\lambda+O(1/n)$ (see e.g. \cite{Klencke09}). We can choose $\lambda'$ such that $\gamma > \gamma(\lambda',0)$ still holds: indeed, by a standard coupling argument, one can see that $\gamma : \lambda \mapsto \gamma(\lambda)$ is increasing. We now build two dominating (in the usual edge presence sense) tree-like $d-$neighborhoods of $i$ and $u$ with the following construction. 
	\begin{itemize}
		\item First, if the two neighborhoods do not intersect, we simply sample two independent trees from model $\dP^{(\lambda')}_d$ rooted in $u$ and in $u'$. 
		\item If the two neighborhoods intersect, condition to the event that $\alpha$ is the contact point in the path $\mathfrak{p}_{\cup}$ (unique by conditioning on $C_{\cup,u,2d}$) from $u$ to $u'$ in the joint graph. Then there is a path of edges of $G$ (say, blue) from $u$ to $\alpha$, then a path of edges of $H$ (say, red) from $\alpha$ to $u'$. Next, complete this construction: along $\mathfrak{p}_{\cup}$, propagate the blue path from $\alpha$ towards $u'$ with probability $s$ on each edge, stopping at the first time when one red edge is not selected. Do the symmetrical construction to propagate the red path from $\alpha$ towards $u$. Finally, to each double-colored vertex $w$, attach independent realizations of model $\dP^{(\lambda',s)}_{d(w)}$ of adapted depth, and to each single-colored vertex $z$, attach independent realizations of model $\dP^{(\lambda')}_{d(z)}$ of adapted depth.
	\end{itemize}
	Note that these constructions lead to at most one path $\mathfrak{p}_{\cup}$ between $u$ and $v'$ in $\cB_{G}(u,d) \cup \cB_{H}(u',d)$, so a fortiori in $\cB_{G}(u,d) \cap \cB_{H}(u',d)$. Denote by  $\mathfrak{p}_{\cap}$ this hypothetical path (cf. Figure \ref{fig:NTMA:parrallel_construction_bis}). We then distinguish between several cases.
	\begin{figure}[h]
		\centering
		\includegraphics[scale=0.83]{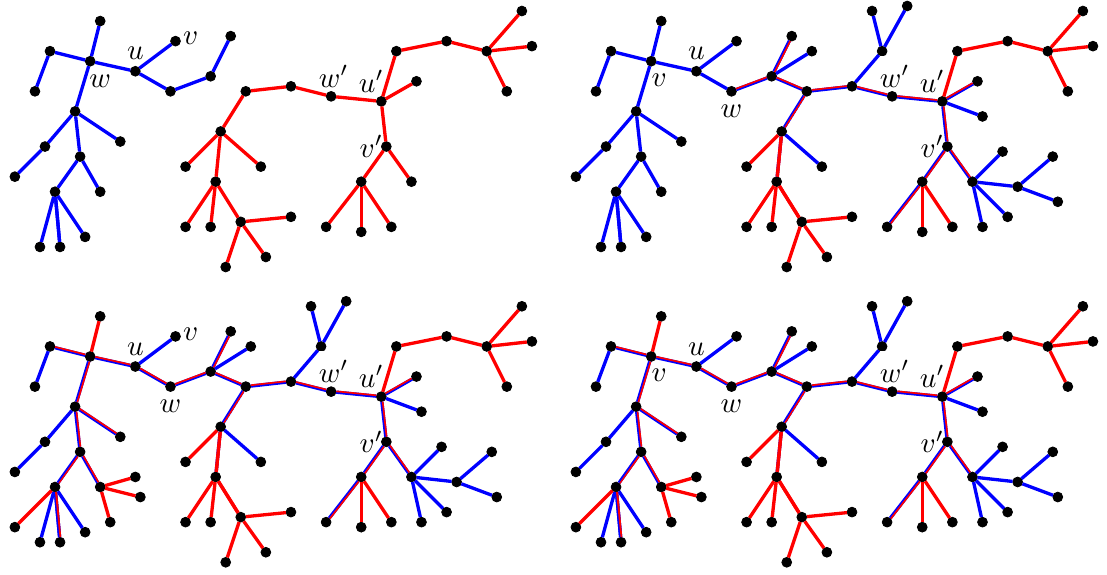}
		\caption{Possible realizations of $\cB_{G}(u,d)$ (blue) and $\cB_{H}(u',d)$ (red), with distinct cases $(i)$ (top left), $(ii)$ (top right), $(iii.a)$ (bottom left) and $(iii.b)$ (bottom right).}
		\label{fig:NTMA:parrallel_construction_bis} 
	\end{figure}

\proofstep{Case $(i)$: $\delta_{G_{\cup}}(u,u')>2d$} (Figure \ref{fig:NTMA:parrallel_construction_bis}, top left), i.e. $\cB_{G}(u,d)\cap\cB_{H}(u',d)=\varnothing$. The  construction gives a coupling with two independent trees from model $\dPl_{d}$. By assumption $\gamma(\lambda)<\lambda s$, the probability that there exist $v$ in $\mathcal{N}_{G}(u)$ and $v'$ in $\mathcal{N}_{H}(u')$ such that $\mathcal{W}_{d-1}(v \leftarrow u, v' \leftarrow u')>\gamma^{d-1}$ is upper bounded by $O\left(d_{\mathrm{max}}^2 \exp\left(-n^{\eps}\right)\right)$, following Remark \ref{rem:NTMA:1.5}. Hence $u$ is matched to $u'$ with at most this probability.\\
	
\proofstep{Case $(ii)$: $\delta_{G_{\cup}}(u,u') \leq 2d$ but $\mathfrak{p}_{\cap}$ does not exist} (see Figure \ref{fig:NTMA:parrallel_construction_bis}, top right). Take $v \neq w$ two neighbors of $u$ and $v' \neq w'$ two neighbors of $u'$. Then (at least) one of these vertices is not on $\mathfrak{p}_{\cup}$ (e.g. vertex $v$ on Figure \ref{fig:NTMA:parrallel_construction_bis}): the downstream tree from this vertex is independent from every other neighborhood in the other graph. They can be coupled with model $\dPl_{d}$, and the same bound as in case $(i)$ holds.\\
	
Now assume that $\mathfrak{p}_{\cap}$ exists, and let $v \neq w$ two neighbors of $u$ and $v \neq v'$ two neighbors of $u'$. 

\proofstep{Case $(iii.a)$:} at least one of the edges $\set{u,v},\set{u,w},\set{u',v'},\set{u',w'}$ is not in $G_{\cap}$ (e.g. edge $(u,v)$ on Figure \ref{fig:NTMA:parrallel_construction_bis}, bottom left): again, the same argument applies. \\

\proofstep{Case $(iii.b)$:}  Edges $\set{u,v},\set{u,w},\set{u',v'},\set{u',w'}$  are all in $G_{\cap}$ (see Figure \ref{fig:NTMA:parrallel_construction_bis}, bottom right). Then one pair of vertices (say $(w,w')$ as on Figure \ref{fig:NTMA:parrallel_construction_bis}) can be on $\mathfrak{p}_{\cap}$ and bring a high $\cW_{d-1}(w \leftarrow u, w' \leftarrow u')>\gamma^{d-1}$ matching weight, if their descendants  spread over a great part of the intersection. In that case, since $v$ and $v'$ can't be on $\mathfrak{p}_{\cap}$, the associated downstream trees are independent, and again $\cW_{d-1}(j \leftarrow i, v \leftarrow u)<\gamma^{d-1}$ with high probability. 

The remaining subcase to be considered is that of matches $(v,w')$ and $(w,v')$, with $w,w'$ on $\mathfrak{p}_{\cap}$. All trees involved are then correlated. However, the coupling construction induces a coupling of the two pairs of $(d-1)-$neighborhoods (from $(v,w')$ and from $(w,v')$, see Figure \ref{fig:NTMA:parrallel_construction_bis}) with two pairs of trees from model $GW(\lambda', s, \delta)$ where $\delta = \left|\mathfrak{p}_{\cap}\right|$. We assume in the Theorem that  $\gamma(\lambda,s,\delta)<\lambda s$ so by Theorem \ref{NTMA:thm:lambda_close_to_1_delta}, the probability that $\mathcal{W}_{d-1}(v \leftarrow u, w' \leftarrow u')>\gamma^{d-1}$ and $\mathcal{W}_{d-1}(w \leftarrow u, v' \leftarrow u')>\gamma^{d-1}$ is upper bounded by $O\left(\exp\left(-n^{\eps}\right)\right)$.\\
	
	Thus, for $u$ fixed, one has 
	$$
	\dP\left(\exists u'  \neq u, \; (u,u') \in \cS\right) \leq 1-\dP\left(C_{\cup,u,2d}\right)+ n \times \dP\left(C_{\cup,u,2d}\right) \times  d_{\mathrm{max}}^2 \times O\left(\exp\left(-n^{\eps}\right) \right) = o(1).
	$$
	The theorem then follows by appealing to Markov's inequality. \end{proof}

\begin{subappendices}

\section{Numerical experiments}
\addtocontents{toc}{\protect\setcounter{tocdepth}{0}}
\subsection{Simulations for tree matching}\label{app:tree}
We here present some simulations of matching rates $\gamma(\lambda)$ (Figure \ref{fig:NTMA:gamma_lambda}) and $\gamma(\lambda,s,\delta)$ for $s=1$ (Figure \ref{fig:NTMA:gamma_lambda_delta}) in order to illustrate Theorems \ref{NTMA:thm:lambda_close_to_1} and \ref{NTMA:thm:lambda_close_to_1_delta} and the final conjecture. For these simulations, error bars correspond to one standard deviation.

\begin{figure}[H]
	\centering
	\medskip
	\begin{subfigure}[t]{0.32\linewidth}
	\centering
	\includegraphics[scale=0.27]{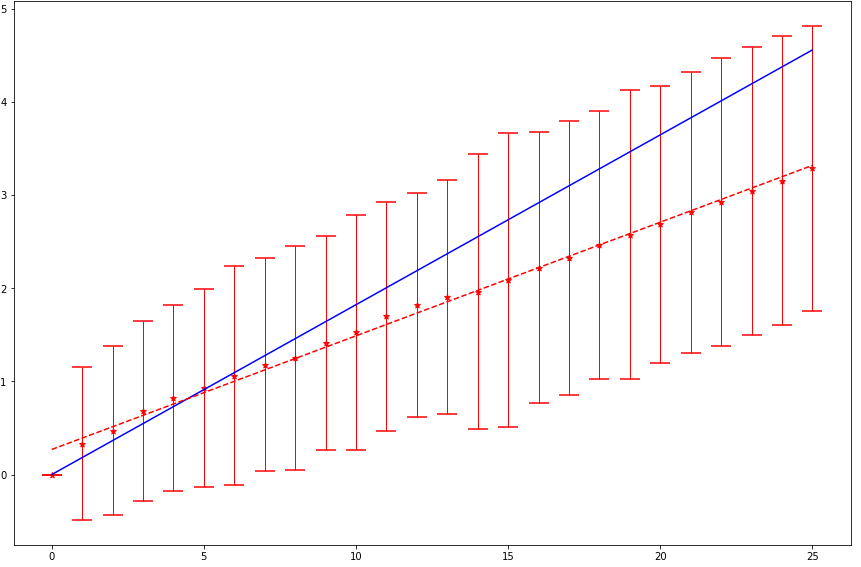}
	\caption{{\footnotesize $\lambda=1.2$, $\log \lambda \sim 0.18$. Red dashed slope $\sim 0.12$}}
	\end{subfigure}
	\hfill
	\begin{subfigure}[t]{0.32\linewidth}
	\centering
	\includegraphics[scale=0.27]{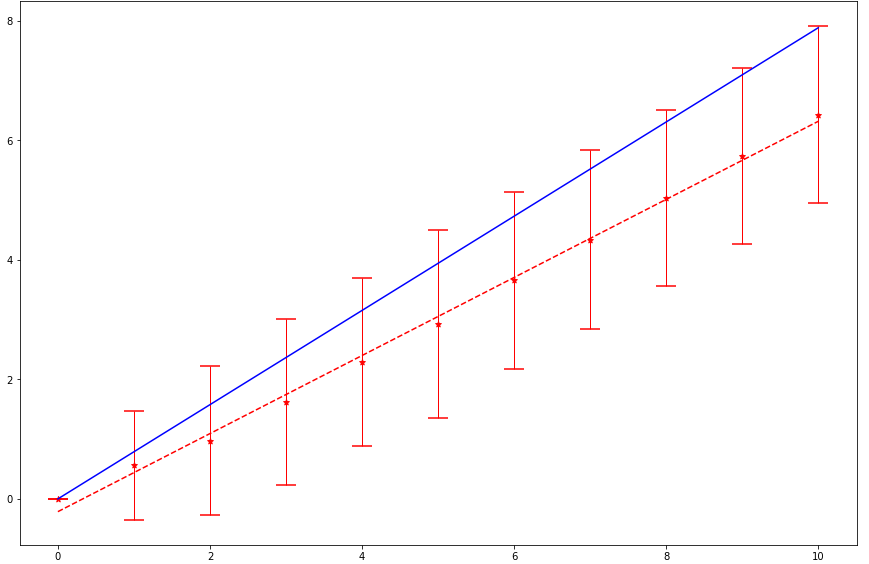}
	\caption{{\footnotesize $\lambda=2.2$, $\log \lambda \sim 0.79$. Red dashed slope $\sim 0.65$}}
	\end{subfigure}
	\hfill
	\begin{subfigure}[t]{0.32\linewidth}
	\centering
	\includegraphics[scale=0.27]{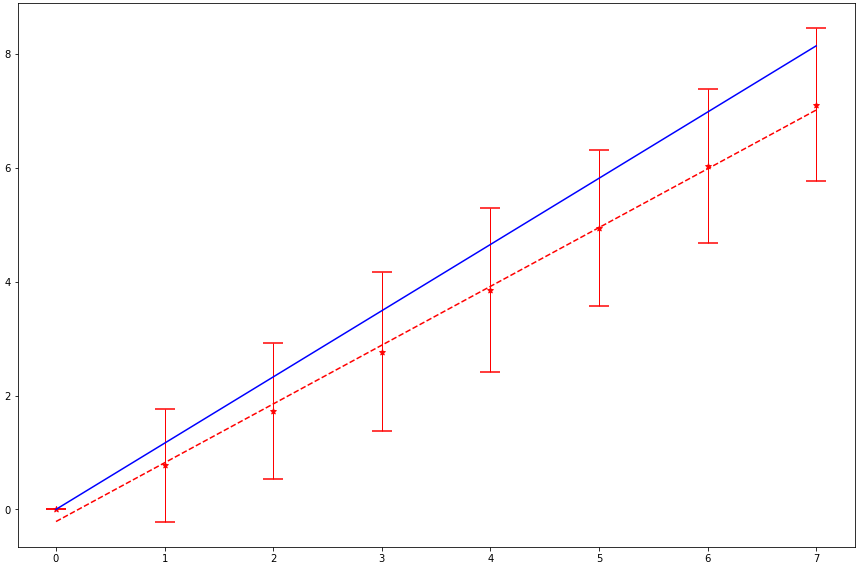}
	\caption{{\footnotesize $\lambda=3.2$, $\log \lambda \sim 1.16$. Red dashed slope $\sim 1.03$}}
	\end{subfigure}
	\caption{Comparison of $d \log \lambda$ (blue) and $\log \cW_d(T,T')$ (red) for $\cW_d(T,T') \sim \dPl_{d}$ conditioned to survive ($100$ iterations)}
	\label{fig:NTMA:gamma_lambda} 
\end{figure}

\begin{figure}[H]
	\begin{subfigure}[t]{0.32\linewidth}
		\centering
		\includegraphics[scale=0.52]{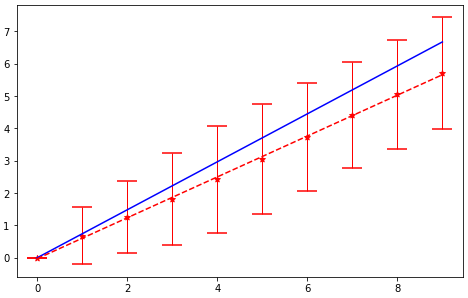}
		\caption{{\footnotesize $\lambda=2.1, \log \lambda \sim 0.74, \delta=1$. Red dashed slope $\sim 0.63$}}
	\end{subfigure}
	\hfill
	\begin{subfigure}[t]{0.32\linewidth}
		\centering
		\includegraphics[scale=0.52]{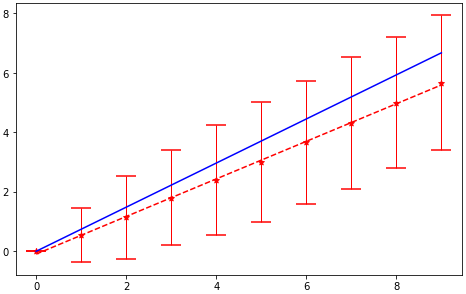}
		\caption{{\footnotesize $\lambda=2.1, \log \lambda \sim 0.74, \delta=2$. Red dashed slope $\sim 0.63$}}
	\end{subfigure}
	\hfill
	\begin{subfigure}[t]{0.32\linewidth}
		\centering
		\includegraphics[scale=0.52]{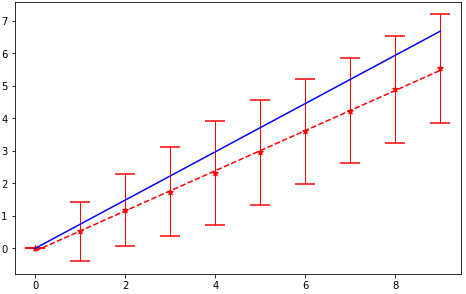}
		\caption{{\footnotesize $\lambda=2.1, \log \lambda \sim 0.74, \delta=5$. Red dashed slope $\sim 0.62$}}
	\end{subfigure}
	\caption{Comparison of $d \log \lambda$ (blue) and $\log \cW_d(T,T')$ (red) for $\cW_d(T,T') \sim \dPlsd_d$ with $s=1$, conditioned to survive ($50$ iterations)}
	\label{fig:NTMA:gamma_lambda_delta} 
\end{figure}

\subsection{\label{appendix_simulations_simple_variant} Simulations for a simple variant algorithm of \alg{NTMA}}
We here present some simulations of simple variant algorithm of \alg{NTMA}, \alg{NTMA--2}, which happens to be more efficient in practice. The Algorithm \alg{NTMA--2} is as follows.

\begin{algorithm}[H]
	\caption{\alg{NTMA--2}}
	\SetAlgoLined
	\textbf{Input:} Two graphs $G$ and $H$ of size $n$, average degree $\lambda$, depth $d$, parameter $\gamma$.
	
	\textbf{Output:} A set of pairs $\cS \subset V(G) \times V(H)$.
	
	$\cS \gets \varnothing$
	
	\For{$(u,u') \in V(G) \times V(H)$}{
		\If{$\cB_{G}(u,d)$ and $\cB_{H}(u',d)$ contain no cycle, and if, denoting $$\cW_{d}(u,v') := \one_{\mboxs{$\cB_{G}(u,d)$ and $\cB_{H}(v',d)$ contain no cycle}}\cW_{d}(\cB_{G}(u,d),\cB_{H}(v',d)) \, ,$$ one has $\cW_{d}(u,u')> \gamma^{d} $, $\cW_{d}(u,u')= \max_{v} \cW_{d}(v,u')$ and $\cW_{d}(u,u')= \max_{v'} \cW_{d}(u,v')$}{
			$\cS \gets \cS \cup \left\lbrace (u,v') \right\rbrace $}
	}
	\For{$(u,u') \neq (v,v') \in \cS$}{
		\If{$u=v$}{$\cS \gets \cS \setminus \left\lbrace (u,y), y \in V(H)\right\rbrace $}
		\If{$u'=v'$}{$\cS \gets \cS \setminus \left\lbrace (x,u'), x \in V(G)\right\rbrace $}
	}
	
	\textbf{return} $\cS$
\end{algorithm}

This algorithm only selects rows and columns weight maximums and match the corresponding pairs. The last part ensures that $\cS$ is a matching. For these simulations, error bars correspond to a confidence interval for the mean value of scores. In Figures \ref{fig:NTMA:NTMA2_21} and \ref{fig:NTMA:NTMA2_31} we compare the scores of \alg{NTMA--2} for $s=0.95$ with the isomorphism case $s=1.0$, for different values of $n$. We illustrate the fact that nearly no vertex is mismatched, whereas a non-negligible fraction of nodes is indeed recovered. In Figure \ref{fig:NTMA:sstar}, we compare the scores of \alg{NTMA--2} for fixed $n$ but varying $s$, illustrating the existence of a 'critical' parameter $s^*(\lambda)$.

\begin{figure}[h]
	\begin{subfigure}[t]{\linewidth}
		\centering
		\includegraphics[scale=0.4]{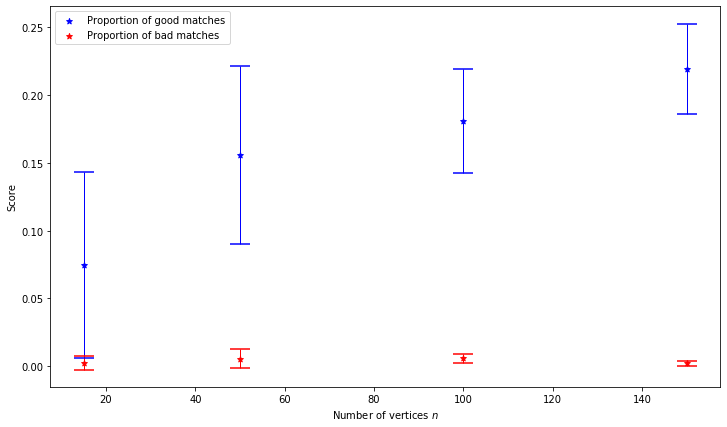}
		\caption{\footnotesize{$s=0.95$}}
	\end{subfigure}
	\begin{subfigure}[t]{1\linewidth}
		\centering
		\includegraphics[scale=0.4]{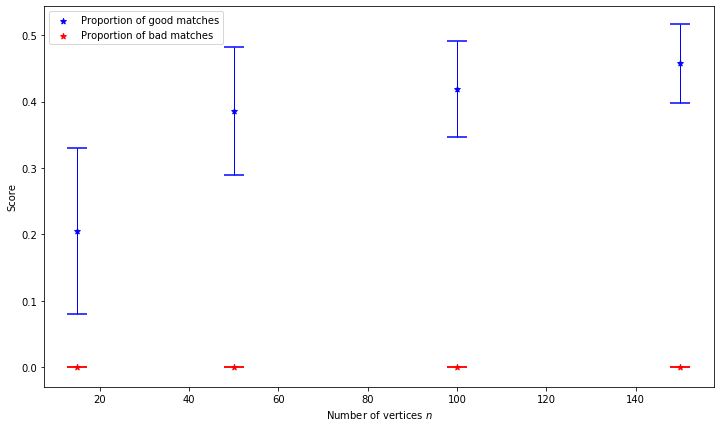}
		\caption{\footnotesize{Isomorphism case, $s=1.0$}}
	\end{subfigure}
	\caption{Mean score of \alg{NTMA--2} for $\lambda=2.1$, $d=5$ (25 iterations per value of $n$)}
	\label{fig:NTMA:NTMA2_21}
\end{figure}

\begin{figure}[h]
	\begin{subfigure}[t]{\linewidth}
		\centering
		\includegraphics[scale=0.5]{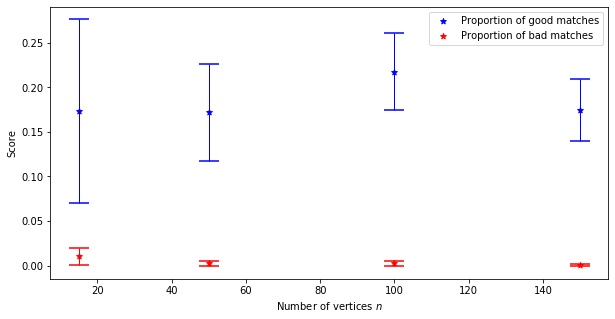}
		\caption{\footnotesize{$s=0.95$}}
	\end{subfigure}
	\hfill
	\begin{subfigure}[t]{\linewidth}
		\centering
		\includegraphics[scale=0.5]{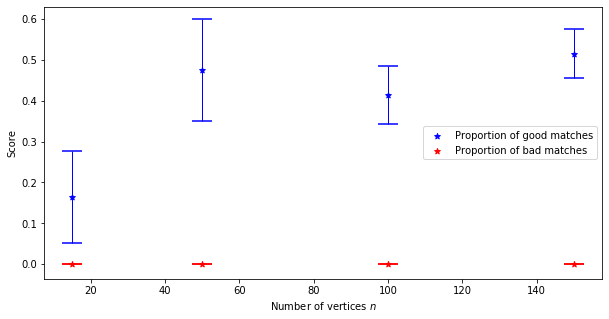}
		\caption{\footnotesize{Isomorphism case, $s=1.0$}}
	\end{subfigure}
	\caption{Mean score of \alg{NTMA--2} for $\lambda=3.1$, $d=4$ (25 iterations per value of $n$)}
	\label{fig:NTMA:NTMA2_31}
\end{figure}

\begin{figure}[h]
	\vspace{0.4cm}
	\begin{subfigure}[t]{\linewidth}
		\centering
		\includegraphics[scale=0.4]{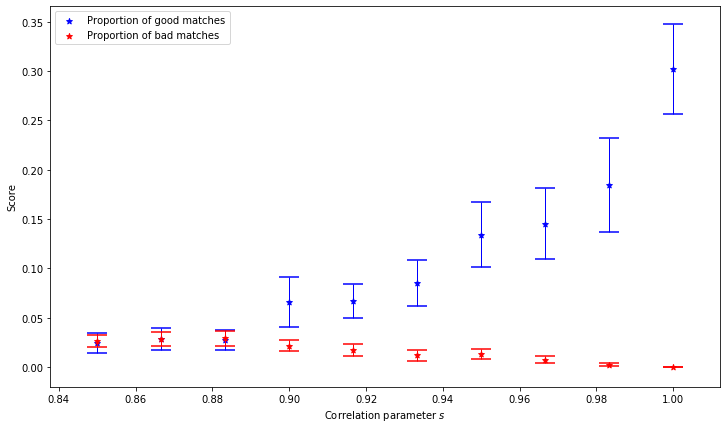} 
		\caption{\footnotesize{$s=0.95$}}
	\end{subfigure}
	\hfill
	\begin{subfigure}[t]{\linewidth}
		\centering
		\includegraphics[scale=0.4]{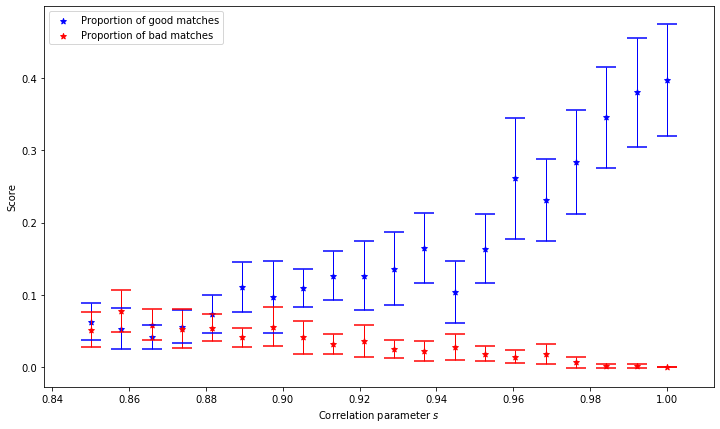}
		\caption{\footnotesize{Isomorphism case, $s=1.0$}}
	\end{subfigure}
	\caption{Mean score of \alg{NTMA--2} with different values of $s$ (25 iterations per value of $n$)}
	\label{fig:NTMA:sstar}
\end{figure}

\addtocontents{toc}{\protect\setcounter{tocdepth}{2}}
\section{Detailed proofs for Section \ref{NTMA:section:tree_matching}}
\addtocontents{toc}{\protect\setcounter{tocdepth}{0}}
\subsection{\label{appendix_proof_lambda_close_to_1}Proof of Theorem \ref{NTMA:thm:lambda_close_to_1}}

\begin{proof}[Proof of Theorem \ref{NTMA:thm:lambda_close_to_1}]
	We first state an easy corollary:
	\begin{corollary}
		\label{cor_exponential_moments}
		For any $d \geq 1$, the random variable $X=\left| \cL_d \left( T_d \right)\right|$ is such that $\dE\left[e^{\theta X} \right]< \infty$ for all $\theta>0$.
	\end{corollary}
	\begin{proof}
		This is easily seen by induction, based on the structure of $T_d$ given in Lemma \ref{NTMA:lemma:pruning_trees}.
	\end{proof}
	Recall that we let $\cE_d$ (respectively, $\cE'_d$) denote the event that tree $T$ (respectively, $T'$) becomes extinct before $d$ generations, i.e. $\cL_d(T)=\varnothing$ (respectively, $\cL_d(T')=\varnothing$). We let $p_{d}=\dP(\cE_d)$. It is well known that it satisfies the recursion
	$$p_{0}=0,\; p_{d}=e^{-\lambda(1-p_{d-1})},$$
	and converges monotonically to the smallest root in $[0,1]$ of $x=e^{-\lambda(1-x)}$. This root, that we denote $p_e$, is the probability of ultimate extinction. For small enough $\eps=\lambda-1$, it holds that 
	$$p_e=1- 2\eps +O(\eps^2),$$ as can be seen by analysis of the fixed point equation satisfied by $p_e$. Let then $d_0$ be such that for all $d\geq d_0$, $p_{d}=1-2\eps+O(\eps^2)$. Clearly, on the event $\cE_d\cup \cE'_d$, the set of matchings $M(T,T')$ is empty, so that $\cW_d(T,T')=0$. Recall that we define $T_d$ the random variable $r_d(T)$ where $T$ is conditioned to survive up to depth $d$.\\
	Now fix $r\in(0,1)$. We shall prove that for sufficiently small $\eps>0$, letting $\gamma=1+r \eps$, there exists some constants $c,m,d_0>0$ such that for all $x > 0$, all $d\geq d_0$, one has
	\begin{equation} 
	\label{exponential_control_eq}
	\dP\left(\cW_d(T_d,T'_d) \geq \gamma^{d-d_0} m x\right)\leq e^{-(x-c)_+}.
	\end{equation}
	We proceed by induction over $d-d_0$. To initialize the induction, notice that one obviously has $\cW_{d_0}(T_{d_0},T'_{d_0})\leq |\cL_{d_0}(T_{d_0})|=:X$. By Corollary \ref{cor_exponential_moments}, for all $m,x,\theta>0$, one has:
	\begin{flalign*}
	\dP\left(  \cW_{d_0}(T_{d_0},T'_{d_0})>m x \right)\leq \dP(X> mx) \leq \dE e^{\theta X} e^{-\theta m x}.
	\end{flalign*} Let now $\theta=1/m$. By taking $m$ sufficiently large, from dominated convergence we can make $\dE e^{(1/m)X}$ as close to $1$ as we like. Choose for instance $m$ such that $\dE e^{(1/m)X}\leq 2$.
	Then
	\begin{align*}
	\dP(\cW_{d_0}(T_{d_0},T'_{d_0})>m x) \leq 2 e^{-x} \leq e^{-x+c}.
	\end{align*} for any $c\geq \ln(2)$. Hence, for sufficiently large $m$, we can initialize the induction at $d=d_0$ with any $c \geq \ln(2)$.\\ 
	
	Recall we set $\gamma=1+r\eps$. Define the random variables
	$$X_d := \gamma^{-(d-d_0)} m^{-1} \cW_d\left(T_{d},T'_{d}\right).$$
	Let $D$ (resp. $D'$) denote the number of children of the root in $T_{d}$ (resp. $D'$ in $T'_{d}$). Given $D$ and $D'$, noting $T_d=\left(T_{d-1,1},\ldots,T_{d-1,D}\right)$ and $T'_d=\left(T'_{d-1,1},\ldots,T'_{d-1,D'}\right)$, we have that
	\begin{equation*}
	\cW_d(T_d,T'_d)=\sup_{\mathfrak{m} \in \cM([D],[D'])} \sum_{(i,u)\in \mathfrak{m}} \cW_{d-1}(T_{d-1,i},T'_{d-1,u}),  
	\end{equation*} where $\cM([D],[D'])$ denotes the set of all $(D \vee D')^{\underline{D\wedge D'}}$ maximal injective mappings between $\cE_0 \subseteq [D]$ and $[D']$. Let $$X_{d-1,i,u}:=\gamma^{-(d-1-d_0)} m^{-1}\cW_{d-1}(T_{d-1,i},T'_{d-1,u}).$$ 
	Note that conditional on $D$ and $D'$, for each matching $\mathfrak{m} \in \cM([D],[D'])$, the variables $\left( X_{d-1,i,u} \right)_{(i,u)\in \mathfrak{m}}$ are i.i.d. with the same distribution as $X_{d-1}$. The induction hypothesis states that each $X_{d-1,i,u}$ is less, for the strong stochastic ordering of comparison of cumulative distribution functions, than $c$ plus an exponential random variable with parameter 1. With an easy union bound, we can derive the following bounds:
	
	\begin{equation}
	\label{eq_key_bound}
	\dP\left(X_d>x\right) \leq \sum_{1 \leq k \leq \ell < \infty}\dP\left(D \wedge D'= k, D \vee D'= \ell\right) \min\left( 1, \ell^{\underline{k}} \, \dP \left( \cE_1 + \ldots + \cE_k > \gamma x - kc \right) \right),
	\end{equation}
	where $\cE_1, \ldots, \cE_k$ are independent exponential random variables of parameter $1$. Lemma \ref{NTMA:lemma:pruning_trees} states that
	\begin{equation*}
	\dP(D=k)=e^{-\lambda(1- p_{d-1})}\frac{\lambda^k(1-p_{d-1})^k}{k! \left(1-p_{d}\right)} =:  q_{d,k}.
	\end{equation*}We can increase $d_0$ such that for some constant $\kappa>0$, for all $d\geq d_0$:
	$$
	q_{d,1}\leq 1-\eps+\kappa\eps^2, \quad q_{d,k}\leq \frac{(3\eps)^{k-1}}{k!},\; k\geq 2.
	$$
	Note that for $x\leq c$, there is nothing to prove in (\ref{exponential_control_eq}), since a probability is always upper-bounded by $1$. We thus only need to consider the case $x>c$. We conclude the proof of this Theorem by appealing to the following technical Lemma, proved later on in Appendix \ref{NTMA:appendix:proof:lemma:q_control}:
	\begin{lemma}
		\label{NTMA:lemma:q_control}
		Let $\kappa,C>0$ and $r\in(0,1)$ be given constants. Then there exists $c>0$ large enough and $\eps_0>0$ such that, for all $\eps\in(0,\eps_0)$, letting  $\gamma=1+r\eps$, $q_1=1-\eps+\kappa \eps^2$, $q_k=(C \eps)^{k-1}/k!$ for $k \geq 2$, one has
		\begin{equation}
		\label{lemma_q_control_eq}
		\forall x>c,\quad \sum_{k,\ell \geq 1} q_k q_l \min\left(1, (k \vee \ell)^{\underline{k \wedge \ell}} \, \dP\left(\cE_1 + \ldots + \cE_{k \wedge \ell}>\gamma x-(k \wedge \ell)c\right) \right)\leq e^{-(x-c)},
		\end{equation}
		where the $\cE_i$ are independent exponential random variables of parameter $1$.
	\end{lemma}
	Its assumptions are indeed verified here with $C=3$, so \eqref{exponential_control_eq} can be propagated by using this Lemma in \eqref{eq_key_bound}, and the conclusion of Theorem \ref{NTMA:thm:lambda_close_to_1} follows.
	
\end{proof}

\subsection{\label{appendix_proof_lambda_close_to_1_delta}Proof of Theorem \ref{NTMA:thm:lambda_close_to_1_delta}}
\begin{proof}[Proof of Theorem \ref{NTMA:thm:lambda_close_to_1_delta}]
	We assume that $\lambda =1+\eps$. We fix $r \in (0,1)$, and we let $\gamma=1+r\eps$ for some fixed $r\in(0,1)$. We work with trees such that $(T,T') \sim \dPlsd_d$. If we assume that the path from $\rho$ to $\rho'$ does not survive down to depth $d$ in $T$, then this path is no more present in $T_d$, and the two trees $T_d$ and $T'_d$ can be coupled with two trees $\widetilde{T}_d$ and $\widetilde{T}'_d$ where $(\widetilde{T},\widetilde{T}') \sim \dPl_{d}$, and we are in the case of Theorem \ref{NTMA:thm:lambda_close_to_1}.
	
	In the following proof, we will thus condition to the event $S_{\rho,d}$ that the path from $\rho$ to $\rho'$ survives down to depth $d$ in $T$. Recall that the tree $T_d$ (resp. $T_d$) is obtained, conditionally on the fact that $T$ (resp. in $T'$) survives down to depth $d$, by suppressing nodes at depth greater than $d$ in $T$ (resp. in $T'$), and then pruning alternatively leaves of depth strictly less than $d$. As in the proof of Theorem \ref{NTMA:thm:lambda_close_to_1}, we shall establish that for sufficiently small $\eps>0$, there exist constants $c,m,d_0>0$ such that for all $x > 0$, all $d\geq d_0$, one has
	\begin{equation}
	\label{exponential_control_lambda_eq}
	\dP\left( \cW_d(T_{d},T'_d)\geq \gamma^{d-d_0} m x \big| S_{\rho,d} \right) \leq e^{-(x-c)^+}.
	\end{equation}
	Define the random variables
	$$X'_d := \gamma^{-(d-d_0)} m^{-1} \cW_d\left(T_{d+\delta},T'_d\right),$$ conditional on $S_{\rho,d}$. The proof will again be by induction on $d$, the initial step being established with the same argument as in the proof of Theorem \ref{NTMA:thm:lambda_close_to_1}. Note that this argument does not depend on $\delta$. 
	
	Denote by $D$ the number of children of $\rho$ in $T_{d}$, $D'$ the number of children of $\rho'$ in $T'_{d}$ that are in the intersection tree $T_d \cap T'_{d}$, and $D''$ the number of children of $\rho'$ in $T'_{d} \setminus T_d$. By branching property, note that these three variables are independent.
	
	Recall that $p_{d}$ denotes the probability that a Galton-Watson tree with offspring $\Poi(\lambda)$ becomes extinct before $d$ generations. Then, conditionally on $S_{\rho,d}$, the random variables $D, D'$ and $D''$ have the following distributions:
	\begin{flalign*}
	D \sim 1 + \Poi\left(\lambda\left(1-p_{d-1}\right)\right), \quad
	D' \sim \Poi\left(\lambda s \left(1-p_{d-1}\right)\right),\\
	D'' \sim \Poi\left(\lambda(1-s)\left(1-p_{d-1}\right)\right), \mbox{ conditionally on } D'+D'' >0.
	\end{flalign*} We show an illustration on Figure \ref{fig:NTMA:image_exemple_D_D'_D''}.
	\begin{figure}[H]
		\centering
		\includegraphics[scale=0.9]{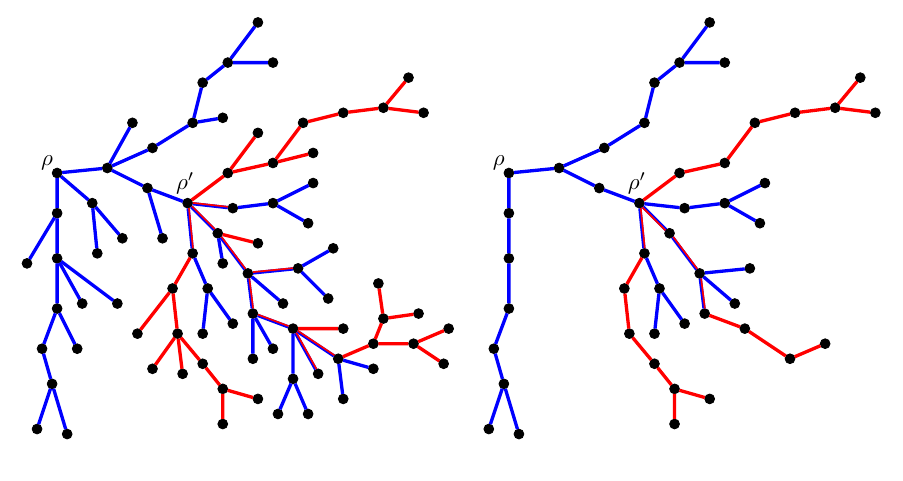}
		\caption{Random trees $T$ (blue) and $T'$ (red) from Figure \ref{fig:NTMA:image_GW} (left), and the results $T_d$  and $T'_d$ after applying $r_d$ (right). In this example, $\delta=3$ and $d=6$, $D=2$, $D' = 2$ and $D'' = 1$.}
		\label{fig:NTMA:image_exemple_D_D'_D''}
	\end{figure}
	We condition on the values $\ell,k',k''$ taken by $D,D',D''$. The number of maximal one-to-one mappings between the children of $\rho$ in $T_d$ and those of $\rho'$ in $T'_d$ is given by $[(k'+k'')\vee \ell]^{\underline{(k'+k'')\wedge \ell}}$, and each of them is of size $\ell\wedge(k'+k'')$. Note here again that for a fixed matching between the children of $\rho$ and $\rho'$, the weights of the matched subtrees are independent. We distinguish between several cases (to help understand these cases, the reader could keep Figure \ref{fig:NTMA:image_exemple_D_D'_D''} in mind):
	\begin{itemize}
		\item For a child $u$ of $\rho$ that is not on the path to $\rho'$, the corresponding subtrees are independent so that the corresponding weight is distributed as $\cW_{d-1}(\widetilde{T}_{d-1},\widetilde{T}'_{d-1})$ in the independent model $\dPl_{d}$.
		\item If the child of $\rho$ on the path to $\rho'$ is matched with a child of $\rho'$ that is not in the intersection tree, again the corresponding weight is similarly distributed.
		\item Finally, if the child $u$ of $\rho$ leading to $\rho'$ is matched to a child $u'$ of $\rho'$ in the intersection tree, setting the new root at $\widetilde{\rho}:=u$ in $T$ and at $\widetilde{\rho}':=u'$ in $T'$, the corresponding weight has the same distribution as $\cW_{d-1}(T_{d-1},T'_{d-1})$ in the model $\dPlsd_{d-1}$, still conditioned to $S_{ \widetilde{\rho},d-1}$. Indeed, there is a path from $\widetilde{\rho}$ to $\widetilde{\rho}'$, and the corresponding Poisson distributions are conserved. 
	\end{itemize} The induction hypothesis for case 3, together with Theorem \ref{NTMA:thm:lambda_close_to_1} for cases 1 and 2, therefore give us:
	
	\begin{equation*}
	\dP\left(X'_d\geq x\right) \leq \sum_{k,\ell} \dP\left(D'+D''=k,D=\ell\right)\min\left(1, (k \vee \ell)^{\underline{k \wedge \ell}} \, \dP\left(\cE_1 + \ldots + \cE_{k \wedge \ell}>\gamma x-(k \wedge \ell)c\right)\right),
	\end{equation*}where the $\cE_i$ are independent exponential random variables of parameter $1$. Assume, as in the proof of Theorem \ref{NTMA:thm:lambda_close_to_1}, that $d_0$ is chosen such that for all $d\geq d_0$, 
	$$
	p^{\lambda}_{d}=1-2\eps+O(\eps^2).
	$$
	With simple computations, we can then ensure that for some $\kappa>0$, noting $q_{d,\cdot}$ the distribution of $D$, one has
	$$
	q_{d,1} \leq 1-\eps+\kappa \eps^2,\quad q_{d,k} \leq \frac{(3\eps)^{k-1}}{(k-1)!} \leq \frac{(6\eps)^{k-1}}{k!},\; k \geq 2,
	$$ where we used $k \leq 2^{k-1}$ in the last step. By independence of $D'$ and $D''$, $D'+D''$ follows a $\Poi(\lambda (1-p_{d-1}))$ distribution, conditional on being positive. Noting $q'_{d,\cdot}$ this distribution, we have, as in the previous proof,
	$$
	q'_{d,1} \leq 1-\eps+\kappa \eps^2,\quad q'_{d,k} \leq \frac{(3\eps)^{k-1}}{k!},\; k \geq 2.
	$$
	We can then invoke Lemma \ref{NTMA:lemma:q_control} to conclude.
	Note that every control in the proof is made uniformly on $\delta \geq 1$.
\end{proof}

\subsection{Proof of Lemma \ref{NTMA:lemma:q_control}}\label{NTMA:appendix:proof:lemma:q_control}
\begin{proof}[Proof of Lemma \ref{NTMA:lemma:q_control}]. Let
	\begin{flalign*}
	S_1  := e^{x-c} q_1^2 e^{-(\gamma x -c)_+} +4 q_1 q_2 e^{-(\gamma x-c)_+}, \quad
	S_2 := 2 e^{x-c} q_1 \sum_{\ell\geq 3} q_\ell \min\left( 1, \ell e^{-(\gamma x -c)_+}\right),\\
	S_3 := 2 e^{x-c} \sum_{2\leq k\leq \ell} q_k q_{\ell} \min\left(1,\ell^{\underline{k}} \, \dP\left(\cE_1 + \ldots + \cE_k > \gamma x-kc\right)\right).
	\end{flalign*} Our goal is to show that for a suitable choice of $c$, for all $x>c$, $S_1+S_2+S_3\leq 1$. One has
	\begin{equation}
	\label{s1_eq}
	S_1 \leq e^{-r \eps x}\left((1-\eps+\kappa \eps^2)^2+2 C \eps \right) \leq e^{-r\eps x}(1+2 C \eps),
	\end{equation} and
	\begin{equation}
	\label{s2_eq}
	S_2\leq 2 e^{-r\eps  x}(1-\eps +\kappa \eps^2)\sum_{\ell \geq 3}\frac{(C\eps)^{\ell-1}}{(\ell-1)!} \leq 2 e^{-r\eps x}\left(e^{C\eps}-1-(C\eps) \right)\leq 2 e^{-r\eps x} C^2\eps^2.
	\end{equation} We let $k_0$ be such that $\gamma x \in[ k_0c,(k_0+1)c)$. We then upper-bound $S_3$ by $A+B$
	where
	\begin{flalign}
	A & = 2 e^{x-c} \sum_{k=2}^{k_0} q_k \sum_{\ell \geq k} q_{\ell} \frac{\ell!}{(\ell-k)!} \, \dP\left(\cE_1 + \ldots + \cE_k > \gamma x - kc\right), \\
	B & = 2 e^{x-c} \sum_{k \geq (k_0+1) \vee 2} \sum_{\ell\geq k} q_k q_{\ell}.
	\end{flalign} One readily has
	\begin{flalign}
	\label{s3_B_eq}
	B & \leq 2 e^{-r\eps  x} e^{\gamma x-c} \sum_{k\geq (k_0+1)\vee 2} \frac{(C\eps)^{k-1}}{k!} \sum_{\ell\geq k} \frac{(C\eps)^{\ell-1}}{\ell!}\\
	&\leq 2 e^{-r\eps x} e^{\gamma x -c}\sum_{k\geq (k_0+1)\vee 2}\frac{(C\eps)^{2(k-1)}}{k!}\\
	&\leq 2 e^{-r\eps x} e^{k_0 c} (C\eps)^{2\left((k_0+1)\vee 2 -1\right)}\\
	&\leq 2 e^{-r\eps x} \left( C\eps e^{c} \right)^2,
	\end{flalign} where in the last steps we assumed that $C\eps e^{c}<1$, so that $$e^{k_0 c} (C\eps)^{2\left((k_0+1)\vee 2 -1\right)} \leq (C\eps e^c)^{2\left((k_0+1)\vee 2 -1\right)} \leq \left(C\eps e^{c}\right)^2.$$
	
	Note that for $y\geq 0$, $\dP\left(\cE_1 + \ldots + \cE_k > y\right)=\dP(\Poi(y)<k)=e^{-y} \sum_{j=0}^{k-1}y^j/j!$. Write then
	\begin{flalign}
	\label{s3_A_eq}
	A & \leq 2 e^{x-c} \sum_{k=2}^{k_0} \frac{(C\eps)^{2(k-1)}}{k!} \sum_{j=0}^{k-1} e^{-\gamma x+k c} \frac{(\gamma x)^j}{j!} \nonumber \\
	&\leq 2 e^{-r\eps x} \sum_{k=2}^{k_0} \frac{(C^2 e^c)^{k}}{k!} \sum_{j=0}^{k-1}\frac{\left(\gamma x\eps^2\right)^j}{j!}\eps^{2(k-1-j)} \nonumber \\
	&\leq 2 e^{-r\eps x}\sum_{k=2}^{k_0}\frac{(C^2 e^c)^{k}}{k!}\left[\eps^{2(k-1)}+e^{\gamma x\eps^2}-1\right] \nonumber \\
	&\leq 2 e^{-r\eps x}\left[\eps^{-2}\left(e^{\eps^2 C^2e^c}-1-\eps^2 C^2 e^c \right)  +  e^{C^2 e^c}\left( e^{\gamma x\eps^2}-1  \right)\right]
	\end{flalign}
	Summing the upper bounds \eqref{s1_eq}-\eqref{s3_A_eq}, the desired property will then hold if for all $x>c$, one has:
	\begin{equation}
	\label{final_eq}
	e^{-r\eps x}\left[1+2C\eps+ 2C^2\eps^2+2 [C\eps e^{c}]^2+2\eps^{-2}\left(e^{\eps^2 C^2e^c}-1-\eps^2 C^2 e^c \right)\right] + 2 e^{-r\eps x} e^{C^2 e^c}\left( e^{\gamma x\eps^2}-1  \right)\leq 1.
	\end{equation}
	
	The first term is, for any fixed $c$, and for sufficiently small $\eps$, upper bounded by $$e^{-r\eps x}\left(1+(2C+1) \eps\right).$$
	
	We now distinguish three cases for $x$.
	
	\proofstep{Case 1:} $x\in [c,1/\sqrt{\eps}]$. The second term is then $O(\eps\sqrt{\eps})$. Provided $r c>2C+1$, since $e^{-r\eps x}\leq e^{-r\eps c}=1-r\eps c+O(\eps^2)$, the left-hand side of \eqref{final_eq} is then upper-bounded by $1-(rc-2C-1)\eps+O(\eps \sqrt{\eps})$, and is thus less than $1$.\\
	
	\proofstep{Case 2:} $x\in [1/\sqrt{\eps},1/\eps]$. Since $e^{-r\eps x}\leq e^{-r\sqrt{\eps}}=1-\Omega(\sqrt{\eps})$, and $e^{\gamma x\eps^2}-1\leq e^{\gamma \eps}-1=O(\eps)$, the left-hand side of \eqref{final_eq} is upper-bounded by $1-\Omega(\sqrt{\eps})$ and is thus less than 1.\\
	
	\proofstep{Case 3:} $x\geq 1/\eps$. 
	The first term is then bounded by $e^{-r}(1+ (2C+1)\eps)$, which is less than $1-\Omega(1)$ for $\eps$ small enough. Letting $y=\eps x$, the second term reads
		$$
		e^{-ry}[e^{\eps \gamma y}-1](2 e^{C^2 e^c}).
		$$
		For small $\eps$, this function is maximized for $y=1/r +O(\eps)$, at which point it evaluates to $O(\eps)$. Thus the left-hand side of \eqref{final_eq} is upper-bounded by $1-\Omega(1)$ in that range.\\

	We have thus shown that for any $r>0$, provided $c> (2C+1)/r$, then for all sufficiently small $\eps$, the desired property holds with $\gamma=1+r\eps$.
\end{proof} 

\addtocontents{toc}{\protect\setcounter{tocdepth}{2}}
\section{\label{appendix_proof_lemmas_sec_2}Detailed proofs for Section \ref{NTMA:section:sparse_graph_alignment}}
\addtocontents{toc}{\protect\setcounter{tocdepth}{0}}
The following  proofs are adapted from the previous work of \cite{Massoulie14} and \cite{Bordenave15}.
\subsection{Proof of Lemma \ref{NTMA:lemma:control_S}}
\begin{proof}[Proof of Lemma \ref{NTMA:lemma:control_S}]
	Fix $K>0$ to be specified later and $\gamma>0$. Fix $u \in [n]$, and define 
	\begin{equation*}
	T := \inf \left\lbrace t \leq d, \left|\cS_{G}(u,t)\right| \geq K \log n \right\rbrace.
	\end{equation*} If $T=\infty$, there is nothing to prove. Given $\left| \cS_{G}(u,T-1) \right|$, $$\left| \cS_{G}(u,T) \right| \sim \Bin(n-\left| \cS_{G}(u,0) \right|-\ldots-\left| \cS_{G}(u,T-1) \right|,1-(1-{\lambda}/{n})^{\left| \cS_{G}(u,T-1) \right|}).$$ Thus
	$$ \left| \cS_{G}(u,T) \right|  \overset{\mathrm{sto.}}{\leq} \Bin\left(n, \lambda K {(\log n)}/{n}\right).$$
	Using Bennett's inequality, for $K'>\lambda K$:
	\begin{equation*}
	\dP \left(\left| \cS_{G}(u,T) \right|  \geq K' \log n \right) \leq \exp\left[- \lambda K h\left(\frac{K'-\lambda K}{\lambda K}\right)\log n \right],
	\end{equation*} with $h(x)= (1+x) \log(1+x) -x$. This probability is $\leq n^{-2-\gamma}$ if $K'$ is large enough to verify $\lambda K h\left(\frac{K'-\lambda K}{\lambda K}\right) > \gamma +2$. With a simple use of the union bound, one gets that $\left|\cS_{G}(u,T) \right|  \in \left[K \log n, K' \log n\right]$ for all $u \in [n]$ with probability $1-O(n^{-1-\gamma})$.
	
	Fix $\eps > 0$ to be specified later. We then check by induction that with high probability, for all $T \leq t \leq d$,
	\begin{equation}
	\label{controle_produit_S}
	\left| \cS_{G}(u,t) \right|  \in \left[ K({\lambda}/{2})^{t-T}  \left(\log n\right) \prod_{s=T}^{t}\left(1-\eps \left({\lambda}/{2}\right)^{-({s-T})/{2}}\right) , K' \lambda^{t-T}  \left(\log n\right) \prod_{s=T}^{t}\left(1+\eps \lambda^{-({s-T})/{2}}\right)\right].
	\end{equation}The case $t=T$ is proved here above. We will next use the inequality
	\begin{equation}
	\label{ineq_lemma}
	\lambda x /(2n) \leq \lambda x /n - \lambda^2 x^2 /(2n^2) \leq 1-\left(1-\lambda/n\right)^x \leq \lambda x /n.
	\end{equation} that holds as soon as $\lambda x /n<1$.\\
	Assuming \eqref{controle_produit_S} holds up to $t$, inequality (\ref{ineq_lemma}) holds for $x = \left| \cS_{G}(u,t) \right| $ for $n$ large enough, since $\left| \cS_{G}(u,t) \right|  < n/\lambda$ for $c \log \lambda <1$. Thus for $n$ large enough $\dE\left| \cS_{G}(u,t+1) \right| $ lies in the interval 
	\begin{equation*}
	\bigg[K   \left({\lambda}/{2}\right)^{t-T}  \left(\log n\right) \underbrace{\prod_{s=T}^{t}\left(1-\eps \left({\lambda}/{2}\right)^{-({s-T})/{2}}\right)}_{=1-O(\eps)} , \lambda K' \lambda^{t-T}  \left(\log n\right) \prod_{s=T}^{t}\left(1+\eps \lambda^{-({s-T})/{2}}\right) \bigg] \, ,
	\end{equation*} with $\hat{\eps} > 0$ to be specified later, Bennett's inequality writes 
	\begin{multline*}
	\dP \left(\big|\left| \cS_{G}(u,t+1) \right| - \dE\left| \cS_{G}(u,t+1) \right|\big|\geq \hat{\eps} \dE\left| \cS_{G}(u,t+1) \right|\right)\\
	 \leq 2 \exp\left[-\left(\lambda/{2}\right)^{t-T+1}  \log n \left(1-O\left(\eps\right) \right) h\left(\hat{\eps}\right) \right],
	\end{multline*}
	which is $\leq n^{-2-\gamma}$ if $K \left({\lambda}/{2}\right)^{t+1-T} h(\hat{\eps})>2+\gamma$. Since for $u \to 0$, $h(u)=u^2 / 2 + o(u^2)$, it suffices to take $\hat{\eps}=\eps \left({\lambda}/{2}\right)^{-({t+1-T})/{2}}$ with $\eps$ small enough and $K$ large enough such that $K \eps >2+\gamma$. Thus \eqref{controle_produit_S} holds for $t+1$ with probability $1-O(n^{-2-\gamma})$. \\
	
	All this ensures that the desired inequality \eqref{control_S_eq} holds for all $u \in [n]$, $t \in [d]$ with probability $1-O(n^{-\gamma})$.
\end{proof}

\subsection{Proof of Lemma \ref{NTMA:lemma:cycles_ER}}
\begin{proof}[Proof of Lemma \ref{NTMA:lemma:cycles_ER}]
	Fix $u \in [n]$. Define
	$$k^* := \inf \lbrace t \leq d, \; \cB_G(u,t) \mbox{ contains a cycle}\rbrace .$$
	Note that $k^* \geq 2$, and that if $k^*=\infty$ then $\cB_G(u,d)$ does not contain any cycle. Now assume that $k^*<\infty$. For any $k \geq 2$, $k^*=k$ if and only if there are two vertices of $\cS_G(u,k-1)$ that are connected, or if there is a vertex of $\cS_G(u,k)$ connected to two vertices of $\cS_G(u,k-1)$. On the event 
	$$\mathcal{A}:=\underset{t \leq d}{\bigcap} \left\lbrace \left|\cS_G(u,t)\right|< C (\log n) \lambda^t \right\rbrace, $$
	this happens with probability at most $$ \left|\cS_G(u,k-1) \right|^2 \times \frac{\lambda}{n}  + \left|\cS_G(u,k) \right| \times \left|\cS_G(u,k-1) \right|^2 \times \frac{\lambda^2}{n^2} \leq C^2 \frac{(\log n)^2 \lambda^{2k}}{n} + C^3 \frac{(\log n)^3 \lambda^{3k}}{n^2}.$$
	
	Taking $\eps>0$ such that $c \log \lambda \leq 1/2-\eps$, choosing $C$ such that $\dP\left(\cA \right)=1-O\left( n^{-2 \eps} \right)$ with Lemma \ref{NTMA:lemma:control_S}, the probability that $\cB_G(u,d)$ contains a cycle is less than
	\begin{flalign*}
	\dP\left( k^* < \infty \right) & \leq  \dP\left( \bar{\cA} \right) +  \sum_{k = 2}^{d} \dP\left( k^* =k \, | \,  \cA \right)\\
	& \leq  O\left( n^{-2 \eps} \right) + O\left( \frac{(\log n)^2 \lambda^{2d}}{n} \right) + O\left( \frac{(\log n)^3 \lambda^{3d}}{n^2} \right) \\
	& \leq  O\left( n^{-2 \eps} \right)+O\left( (\log n)^2 n^{-2 \eps} \right)+O\left( (\log n)^3 n^{-3 \eps} \right) \leq O(n^{-\eps}).
	\end{flalign*}
\end{proof}

\subsection{Proof of Lemma \ref{NTMA:lemma:indep_neighborhoods}}
\begin{proof}[Proof of Lemma \ref{NTMA:lemma:indep_neighborhoods}]
	For fixed $u \neq v \in [n]$, let $\left(\tilde{\cS}(u,t)\right)_{t \leq d}$ and $\left(\tilde{\cS}(v,t)\right)_{t \leq d}$ denote two independent realizations of the neighborhoods (i.e. with independent underlying Bernoulli variables). We then construct recursively a coupling $\left(\cS(u,t),\cS(v,t)\right)_{t \leq k} $:
	\begin{itemize}
		\item For $k=1$, take $\cS(u,t)$ to be a set of vertices uniformly chosen among sets of $[n]$ of size $\left| \tilde{\cS}(u,0)\right|$. Independently, take $\cS(v,t)$ to be a set of vertices uniformly chosen among sets of $[n]$ of size $\left| \tilde{\cS}(v,0)\right|$.
		
		\item Now if $k>1$, construct $\cS(u,k)$ as follows: select a subset of $[n] \setminus \left(\underset{s \leq k-1}{\bigcup} \cS(u,s)\right)$ of size $ \left|\tilde{\cS}(u,k)\right| $ uniformly at random. Then we construct independently $\cS(v,k)$ taking a uniform subset of $[n] \setminus \left(\underset{s \leq k-1}{\bigcup} \cS(v,s)\right)$ of size $ \left|\tilde{\cS}(v,k)\right| $.
	\end{itemize}
	
	This coupling is well defined, and coincides with the independent setting up to step $k$ as long as the sets $\underset{s \leq k}{\bigcup} \cS(u,s)$ and $\underset{s \leq k}{\bigcup} \cS(v,s)$ do not intersect. On the event $$\mathcal{A}:=\underset{t \leq d}{\bigcap} \left\lbrace \left|\cS(u,t)\right|,\left|\cS(v,t)\right| < C (\log n) \lambda^t \right\rbrace, $$ one has
	\begin{flalign*}
	\dE\left[\left|\underset{k \leq d}{\bigcup} \cS(u,s) \cap \underset{k \leq d}{\bigcup} \cS(v,s)\right|\right] & \leq 
	\dE\left[\sum_{k=1}^{d} \Bin\left(C (\log n) \lambda^k, \frac{\sum_{t=1}^{k} C (\log n) \lambda^t}{n-\sum_{t=1}^{k} C (\log n) \lambda^t} \right) \right] \\
	& \leq C^2 (\log n)^2 \left(\frac{\lambda}{\lambda-1}\right) \sum_{k=1}^{d} \frac{\lambda^{2k} }{n-\frac{\lambda}{\lambda-1} C (\log n)\lambda^k}\\
	& \leq O \left((\log n)^2 \lambda^{2d}/n\right)
	\end{flalign*}if $(\log n) \lambda^d = o(n)$, which is the case if $c \log \lambda <1$. The expectation is upper-bounded by $O \left((\log n)^2 \lambda^{2d}/n\right) = O\left((\log n)^2 n^{-2\eps}\right)$ if $c \log \lambda \leq 1/2 -  \eps$.\\
	
	With Lemma \ref{NTMA:lemma:control_S}, choosing $C$ such that $\dP\left(\cA \right)=1-O\left( n^{-2 \eps} \right)$, we get
	\begin{multline*}
	\DTV\left(\mathcal{L} \left(\left(\cS_{G}(u,t),\cS_{G}(v,t)\right)_{t \leq d}\right),\mathcal{L} \left(\left(\cS_{G}(u,t)\right)_{t \leq d}\right) \otimes \mathcal{L} \left(\left(\cS_{G}(v,t)\right)_{t \leq d}\right)\right) \\ \leq O((\log n)^2 n^{-2 \eps}) + \dP\left(\bar{\mathcal{A}}\right) \leq O(n^{- \eps}).
	\end{multline*}
\end{proof}

\subsection{Proof of Lemma \ref{lemma:NTMA:coupling_GW}}
\begin{proof}
We work here conditionally on
$$\mathcal{A}:=\underset{t \leq d}{\bigcap} \left\lbrace \left|\cS_G(u,t)\right|< C (\log n) \lambda^t \right\rbrace.$$
Let's define a Galton-Watson process as follows: set $Z_0=1$, and for $t > 0$, $\cL\left( Z_{t} | \cG_{t-1} \right)=\Poi \left( \lambda Z_{t-1}\right)$, where $\cG_t=\sigma\left(Z_s,s \leq t \right)$. 
Fix $t>0$. Conditionally on $\cF_{t-1}:=\sigma\left(\left|\cS_G(u,s)\right|, s \leq t-1 \right)$, define a random variable $W_t$ with distribution $\Poi \left( \lambda \left|\cS_G(u,t-1)\right|\right)$. Note that
$$\cL\left( \left|\cS_G(u,t)\right| \big| \cF_{u-1} \right) = \Bin ( n-\left|\cS_G(u,0)\right|-\ldots-\left|\cS_G(u,t-1)\right|, \; 1-\left(1-{\lambda}/{n} \right)^{\left|\cS_G(u,t-1)\right|} ).$$
The Stein-Chen method (see e.g. \cite{Barbour05}) enables to bound  $\DTV \left(\Bin(n,\lambda/n),\Poi(\lambda)\right)$ by $\min(1,\lambda^{-1})\lambda^2/n \leq \lambda/n$. We also use the classical bound $\DTV \left(\Poi(\lambda),\Poi(\lambda')\right)\leq \left| \lambda-\lambda' \right|$ together with inequality \eqref{ineq_lemma} (which holds for $n$ large enough since $c \log \lambda <1)$ to obtain that conditionally on $\cF_{t-1}$:
\begin{multline*}
\DTV\left(\left|\cS_G(u,t)\right|,W_t\right) \leq n^{-1}  \left(n-\left|\cS_G(u,0)\right|-\ldots-\left|\cS_G(u,t-1)\right|\right) \frac{\lambda \left|\cS_G(u,t-1) \right|}{n} \\
+ \left| \left(n-\left|\cS_G(u,0)\right|-\ldots-\left|\cS_G(u,t-1)\right|\right) \left( 1-\left(1-{\lambda}/{n} \right)^{\left|\cS_G(u,t-1)\right|} \right) - \lambda \left|\cS_G(u,t-1) \right| \right| \\
 \leq \frac{\lambda \left|\cS_G(u,t-1) \right|}{n} + \lambda \left|\cS_G(u,t-1) \right| - \left(n-\left|\cS_G(u,0)\right|-\ldots-\left|\cS_G(u,t-1)\right|\right) \frac{\lambda \left|\cS_G(u,t-1) \right|}{n}  \\ + \frac{\lambda^2 \left|\cS_G(u,t-1) \right|^2}{2n} \, .
\end{multline*}
Now, for $\eps>0$ such that $c \log \lambda \leq 1/2 - \eps$, on the event $\cA$, all variables $\left|\cS_G(u,s) \right|$ are bounded by $C (\log n) n^{1/2 - \eps}$. This leads to
\begin{flalign*}
\DTV\left(\left|\cS_G(u,t)\right|,W_t\right) & \leq O\left( (\log n)n^{-1/2-\eps} \right) + O\left( (\log n)^3 n^{-2\eps} \right) + O\left( (\log n)^2 n^{-2\eps} \right) \\& = O\left( (\log n)^3 n^{-2\eps} \right).
\end{flalign*}
This proves by induction that the total variation distance between $\left( \left|\cS_G(u,t)\right| \right)_{t \leq d}$ and $\left( Z_t \right)_{t \leq d}$ is bounded by $O\left( (\log n)^4 n^{-2\eps} \right) = O\left(n^{-\eps}\right)$, taking $C$ large enough in Lemma \ref{NTMA:lemma:control_S} so that $\dP\left(\cA \right) \geq 1 - O\left(n^{-2\eps}\right)$.
\end{proof}

\addtocontents{toc}{\protect\setcounter{tocdepth}{2}}
\end{subappendices}

\chapter{Detecting correlation in trees}\label{chapter:MPAlign}
Following the way paved in Chapter \ref{chapter:NTMA}, motivated by alignment of correlated sparse random graphs, we are now studying the hypothesis testing problem of deciding whether or not two random trees are correlated more in detail. We obtain conditions under which this testing is impossible or feasible.

We propose \alg{MPAlign}, a message-passing algorithm for graph alignment inspired by the tree correlation detection problem. We prove \alg{MPAlign} to succeed in polynomial time at partial alignment whenever tree detection is feasible. As a result, our analysis of correlation detection in trees reveals new ranges of parameters for which partial alignment of sparse random graphs is feasible in polynomial time.

We then conjecture that graph alignment is not feasible in polynomial time when the associated tree detection problem is impossible. If true, this conjecture together with our sufficient conditions on tree detection impossibility would imply the existence of a hard phase for graph alignment, i.e. a parameter range where alignment cannot be done in polynomial time even though it is known to be feasible in non-polynomial time.\\


This chapter is based on the paper \textit{Correlation detection in trees for partial graph alignment} \cite{GMLTrees2021journal} (submitted), a joint work with M. Lelarge and L. Massoulié. A short version of this work, \textit{Correlation Detection in Trees for Planted Graph Alignment}, \cite{GMLTrees2021ITCS} is published at \emph{ITCS 2021}.

\section{Introduction}\label{MPAlign:section:intro}

We refer to Section \ref{intro:section:ga} for a presentation of the graph alignment, so as not to repeat ourselves. 

As done in Chapter \ref{chapter:NTMA}, we do not recall here the definition of the correlated \ER model, already introduced in the introduction (see \eqref{eq:CER_model}), and specified in Section \ref{impossibility:section:introduction} of Chapter \ref{chapter:impossibility} in the sparse case. We only recall that the parameters of $\G(n,\lambda/n,s)$ are the number of nodes $n$, the mean degree $\lambda >0$ and the correlation parameter $s \in [0,1]$. The vertices of the second graph $G'$ are relabeled with a uniform independent permutation $\pi^{\star} \in \cS_n$, and we observe $G$ and $H := G'^{\pi^{\star}}$.

\begin{figure}[h]
     \centering
     \begin{subfigure}[b]{\textwidth}
         \centering
         \includegraphics[scale=0.8]{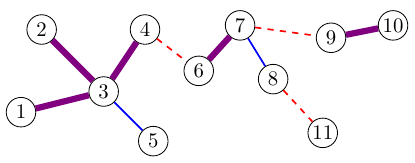}
         \caption{Union graph $(G,G')$}
         \label{fig:G_union}
     \end{subfigure}
     \hfill
     \begin{subfigure}[b]{\textwidth}
         \vspace{0.3cm}
         \centering
         \includegraphics[scale=0.8]{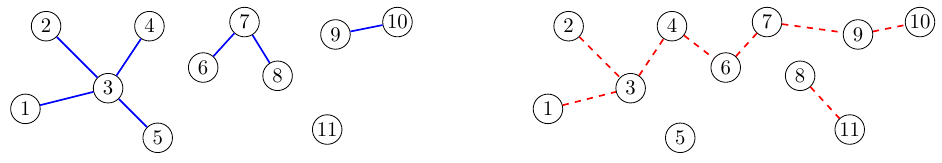}
         \caption{Graphs $G,G'$ in separate views}
         \label{fig:GGtilde}
     \end{subfigure}
     \hfill
     \begin{subfigure}[b]{\textwidth}
         \centering
         \vspace{0.1cm}
         \includegraphics[scale=0.8]{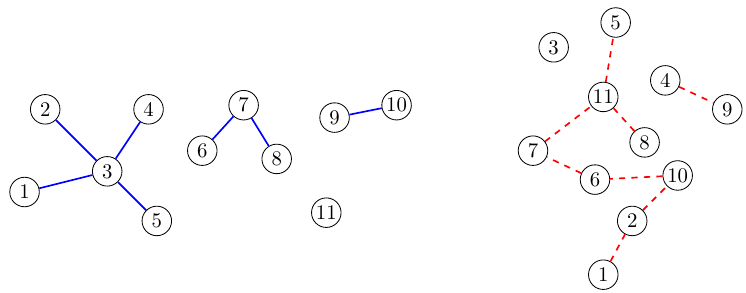}
         \caption{Graphs ${G}, H$}
         \label{fig:GH}
     \end{subfigure}
     
    \caption{A sample from model $\G(n,\lambda/n,s)$ with $n=11$, $\lambda=1.9$, $s=0.7$ (for the sake of readability, the two-colored edges are drawn thick and purple).}
    \label{fig:unplanted_pb}
\end{figure}

The previous model is used to study planted graph alignment -- the mean-case version of graph alignment -- consisting in finding an estimator $\hat{\pi}$ of the planted solution $\pi^{\star}$ upon observing $G$ and $H$. As stated earlier, for any subset $\cC \subset [n]$, the performance of any one-to-one estimator $\hat{\pi}: \cC \to [n]$ is now assessed through $\ov(\pi^{\star},\hat{\pi})$, its {\em overlap} with the unknown permutation $\pi^{\star}$, defined as
\begin{equation}\label{eq:MPAlign:def_overlap}
    \ov(\pi^{\star},\hat{\pi}):=\frac{1}{n}\sum_{u\in \cC}\one_{\hat{\pi}(u)=\pi^{\star}(u)}.
\end{equation} Note that the estimator $\hat{\pi}$ may not be in $\cS_n$, and only consist in a partial matching. The \emph{error fraction} of $\hat{\pi}$ with the unknown permutation $\pi^{\star}$ is defined as 
\begin{equation}\label{eq:MPAlign:def_error}
    \err(\pi^{\star},\hat{\pi}):=\frac{1}{n}\sum_{u\in \cC}\one_{\hat{\pi}(u) \neq \pi^{\star}(u)} = \frac{\card{\cC}}{n} - \ov(\pi^{\star},\hat{\pi}).
\end{equation} 

We recall that sequence of injective estimators $\{\hat{\pi}_n\}_n$ -- omitting the dependence in $n$ -- is said to achieve
\begin{itemize}
    \item \emph{Exact recovery} if $\; \dP(\hat{\pi}=\pi^{\star}) \underset{n \to \infty}{\longrightarrow} 1$,
    \item \emph{Almost exact recovery} if $\; \dP(\ov(\pi^{\star},\hat{\pi})= 1-o(1)) \underset{n \to \infty}{\longrightarrow} 1$,
    \item \emph{Partial recovery} if there exists some $\eps>0$ such that $\; \dP(\ov(\pi^{\star},\hat{\pi}) > \eps) \underset{n \to \infty}{\longrightarrow} 1$,
    \item \emph{One-sided partial recovery} if it achieves partial recovery and $ \dP(\err(\pi^{\star},\hat{\pi}) = o(1)) \underset{n \to \infty}{\longrightarrow} 1$.
\end{itemize}

\begin{remark}\label{remark:MPAlign:one_sided_partial}
One-sided partial recovery is by definition at least as hard as partial recovery. 
As already stated in the introduction, from an application standpoint it is more appealing than partial recovery: indeed, it may be of little use to know one has a permutation with 30\% of correctly matched nodes if one does not have a clue about which pairs are correctly matched. Our proposed algorithm will achieve one-sided partial recovery under suitable conditions.
\end{remark} 

\subsubsection{Phase diagram} 
In the studied sparse regime where the graphs have constant mean degree $\lambda$, it is known \cite{Cullina2017,Cullina18,ganassali2021impossibility} that the presence of $\Omega(n)$ isolated vertices in the underlying intersection graph of $G$ and $H$ makes exact and almost exact recovery impossible.
The main questions consist then in determining the phase diagram of the model $\G(n,\lambda/n,s)$ for partial alignment (or recovery), which we here recall the definition. We are interested in the range of parameters $(\lambda,s)$ for which, in the large $n$ limit:
\begin{itemize}
\item Any sequence of estimators fails to achieve partial recovery for any $\eps>0$. We refer to the corresponding range as the \emph{impossible phase};
\item There is a sequence of estimators $\hat{\pi}$ achieving partial recovery (not necessarily one-sided) with some $\eps>0$, which we refer to as the \emph{IT-feasible phase};
\item There is a sequence of estimators $\hat{\pi}$ \emph{that can be computed in polynomial-time} achieving  partial recovery with some $\eps>0$ (and sometimes even more, achieving also one-sided partial recovery): the \emph{easy phase}.
\end{itemize}

An interesting perspective on this problem is provided by research on community detection, or graph clustering, for random graphs drawn according to the stochastic block model. In that setup, above the so-called Kesten-Stigum threshold, polynomial-time algorithms for clustering are known \cite{BLM18,SpectralRedemption13, Mossel2016}, and the consensus among researchers in the field is that no polynomial-time algorithms exist below that threshold. Yet, there is a range of parameters with non-empty interior below the Kesten-Stigum threshold for which exponential-time algorithms are known to succeed at clustering \cite{Banks2016InformationtheoreticTF}. In other words, for graph clustering, it is believed that there is a non-empty \emph{hard phase}, consisting of the set difference between the IT-feasible phase and the polynomial-time feasible phase.

The picture available to date\footnote{at the time of this contribution.} for partial graph alignment is as follows. Work presented in Chapter \ref{chapter:impossibility} \cite{ganassali2021impossibility} shows that the IT-impossible phase includes the range of parameters $\{(\lambda,s): \lambda s\leq 1\}$, and Wu et al. \cite{Wu2021SettlingTS} have established that the IT-feasible phase includes the range of parameters $\{(\lambda,s): \lambda s> 4\}$ (condition $\lambda s>C$ for some large $C$ had previously been established in \cite{Hall20}). For the easy phase, we established in Chapter \ref{chapter:NTMA} \cite{Ganassali20a}  that it includes the range of parameters $\{(\lambda,s):\lambda\in[1,\lambda_0],s\in[s(\lambda),1]\}$ for some parameter $\lambda_0>1$ and some function $s(\lambda):(1,\lambda_0]\to[0,1]$. The \alg{NTMA} algorithm proposed in Chapter \ref{chapter:NTMA} based on tree matching weights achieves in this regime one-sided partial recovery. Figure \ref{fig:phase_diagram} depicts a phase diagram describing these prior results together with the new results of this chapter.

\begin{figure}[h]
    \centering 
    \hspace{-1.2cm}
    \includegraphics[scale=0.7]{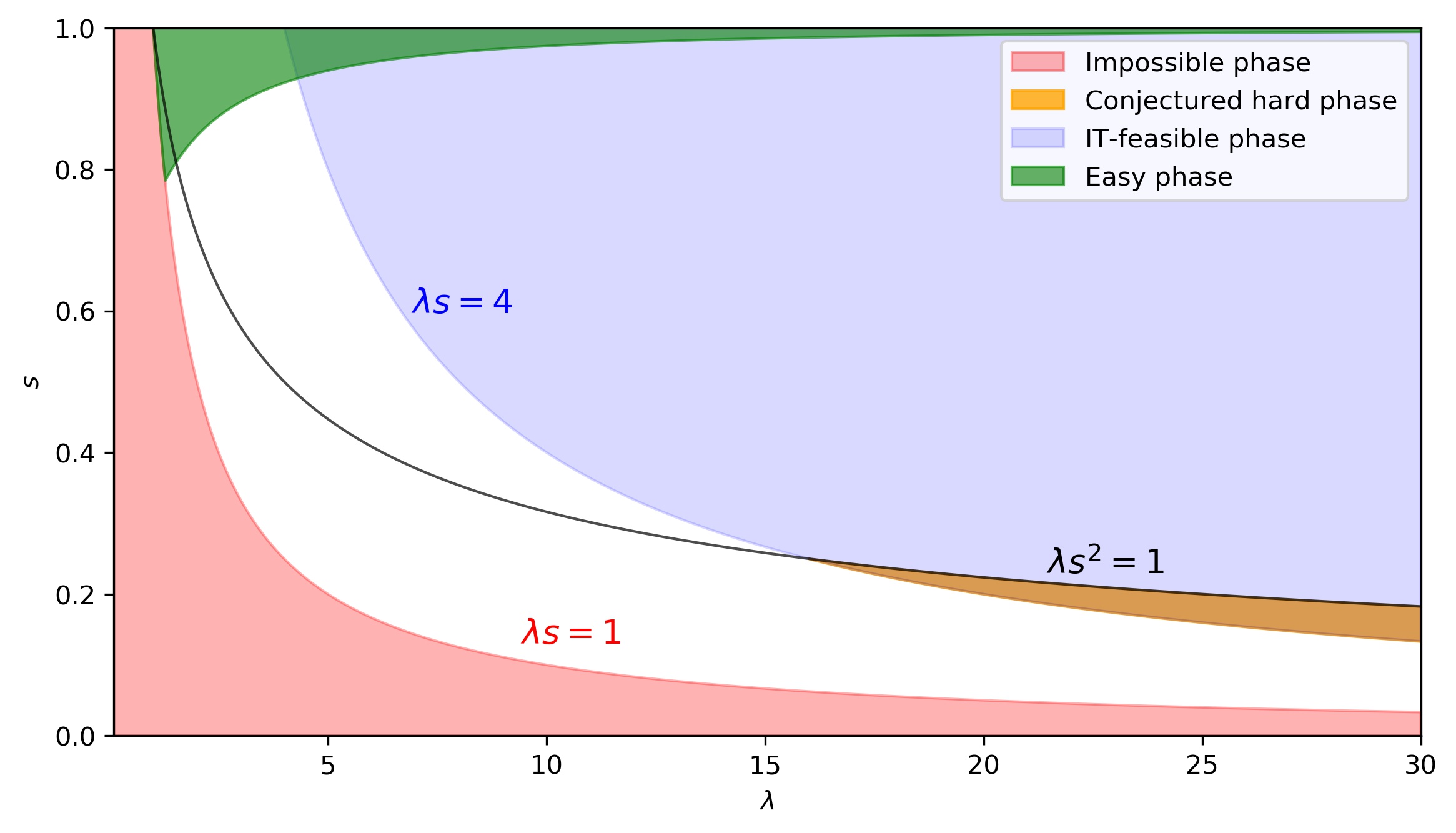}
    \caption{Diagram of the $(\lambda,s)$ regions where partial recovery is known\protect\footnotemark\, to be IT-impossible (\cite{ganassali2021impossibility}), IT-feasible (\cite{Wu2021SettlingTS}), or easy (\cite{Ganassali20a} and this chapter). In the orange region, though partial graph alignment is IT-feasible, one-sided detectability is impossible in the tree correlation detection problem, and partial graph alignment is conjectured to be hard (this chapter).}
    \label{fig:phase_diagram}
\end{figure}
\footnotetext{at the time of this contribution.}

\subsubsection{Problem description and main contributions} 

This partial picture leaves open the question of whether, similarly to the case of graph clustering, graph alignment features a hard phase or not. The contribution of the present work can be summarized in three points:
\begin{itemize}
    \item[$(1)$] We investigate a fundamental statistical problem, which to the best of our knowledge had not been previously studied: hypothesis testing for correlation detection in trees. We study the regimes in which the optimal test on trees succeeds or fails in the setting when the trees are correlated Galton-Watson trees (see Theorem \ref{MPAlign:theorem:main_result_TREES});
    
    \item[$(2)$] For this detection problem on trees, the computation of the likelihood ratio can be made recursively on the depth, which yields an optimal message-passing algorithm for this task running in polynomial-time in the number of nodes;
    
    \item[$(3)$] We remark that the previous detection problem on trees arises naturally from a local point of view in the related problem of one-sided partial recovery for graph alignment. In light of the previous analysis we then draw conclusions for our initial problem on graphs and doing so we precise the phase diagram shown in Figure \ref{fig:phase_diagram}, extending the regime for which one-sided partial alignment is provably feasible in polynomial time, and exhibiting the presence of a conjectured hard phase (see Theorem \ref{MPAlign:theorem:main_result_GRAPHS}).
\end{itemize}

Our approach to point $(3)$ follows the way paved in Chapter \ref{chapter:NTMA}. It essentially relies on an algorithm which lets $\hat{\pi}(u)=u'$ for $u$ such that the local structure of graph $G$ in the neighborhood of node $u$ is 'close' to the local structure of graph $H$ in the neighborhood of node $u'$. As exploited in Chapter \ref{chapter:NTMA}, the neighborhoods to distance $d$ of two nodes $u,u'$ in $G$ and $H$, provided that $u'=\pi^{\star}(u)$, are asymptotically distributed as correlated Galton-Watson branching trees (denoted $\dPls_d$ hereafter). On the other hand, for pairs of nodes $(u,u')$ taken at random in $[n]$, the joint neighborhoods of nodes $u$ and $u'$ in $G$ and $H$ respectively, up to depth $d$, are asymptotically distributed as a pair of independent Galton-Watson branching trees (distribution denoted $\dPl_d$).

Thus a fundamental step in our approach is to determine the efficiency of tests for deciding whether a pair of branching trees is drawn from either a product distribution, or a correlated distribution. \cite{Ganassali20a} relied on tests based on a so-called \emph{tree matching weight} to measure the similarity between two trees. In the present work we are instead interested in studying the existence of \emph{one-sided tests}, which are tests asymptotically guarantying a vanishing type I error and a non vanishing power. According to the Neyman-Pearson Lemma, optimal one-sided tests are based on the likelihood ratio $L_d$ of the distributions under the distinct hypotheses $\dPls_d$ and $\dPl_d$ (trees correlated or not)\footnote{This guarantees that whenever the test based on tree matching weight in Chapter \ref{chapter:NTMA} \cite{Ganassali20a} succeeds, the optimal test studied in this chapter also succeeds. On this point, Theorem \ref{MPAlign:theorem:suff_cond_KL} (see Section \ref{MPAlign:section:KL}) extends the sufficient conditions established in Chapter \ref{chapter:NTMA} for partial alignment (for small $\lambda$ and $s$ close to $1$).}. The mathematical formalization of point $(1)$ here above is the following
\begin{theorem}[Correlation detection in trees]\label{MPAlign:theorem:main_result_TREES}
Let 
\begin{equation*}
	\KL_d := \KL (\dPls_d\Vert\dPl_d)= \dE_{1,d} \left[ \log(L_d) \right].
\end{equation*}
Then the following propositions are equivalent:
\begin{itemize}
    \item[$(i)$] There exists a one-sided test for deciding $\dPl_d$ versus $\dPls_d$,
    \item[$(ii)$] $\underset{d \to \infty}{\lim} \KL_d = + \infty$ and $\lambda s >1$,
    \item[$(iii)$] There exists $(a_d)_d$ such that $a_d \to \infty$, $\dPl_d(L_d>a_d) \to 0$ and $\liminf_d \dPls_d(L_d>a_d) > 0$.
    \item[$(iv)$] The martingale $(L_d)_d$ (w.r.t. $\dPl_{\infty}$) is not uniformly integrable. 
    \item[$(v)$] $\lambda s>1$ and $\dPls_{\infty}\left(\liminf_{d\to\infty}(\lambda s)^{-d} \log L_d >0\right)\geq 1-\pext(\lambda s) $, where $\pext(\lambda s)$ is the probability that a Galton-Watson tree with  offspring distribution  $\Poi(\lambda s)$ gets extinct. 
\end{itemize}
\end{theorem}

\begin{remark}\label{MPAlign:remark:theorem1}
This Theorem gives general necessary and sufficient conditions for the existence of a one-sided test in the tree correlation detection problem. Several more explicit conditions in terms of $\lambda$ and $s$ will be obtained throughout the chapter which guarantee that the equivalent conditions of Theorem \ref{MPAlign:theorem:main_result_TREES} either fail or hold.
Condition $(v)$ will be used in the design of the algorithm in Section \ref{MPAlign:section:graph_matching}, choosing an appropriate threshold that will guarantee for the method to output both a substantial part of the underlying permutation and a vanishing number of mismatches. 
\end{remark}

The link between the problem on trees and sparse graph alignment is given in the following

\begin{theorem}[Consequences for one-sided partial graph alignment]\label{MPAlign:theorem:main_result_GRAPHS}
For given $(\lambda,s)$, if one-sided correlation detection is feasible, i.e. any of the conditions in Theorem \ref{MPAlign:theorem:main_result_TREES} holds, then one-sided partial alignment in the correlated \ER model $\G(n,\lambda/n,s)$ is achieved in polynomial time by our algorithm \alg{MPAlign} (Algorithm \ref{MPAlign:algo_GA}  in Section \ref{MPAlign:section:graph_matching}).
\end{theorem}

\begin{conj*}\label{conjecture:hard_phase}
We conjecture that if one-sided correlation detection in trees fails, i.e. none of the equivalent conditions in Theorem \ref{MPAlign:theorem:main_result_TREES} holds, then no polynomial-time algorithm achieves partial recovery. In view of Theorem \ref{MPAlign:theorem:suff_hard_phase} of Section \ref{MPAlign:section:hard_phase}, which guarantees existence of a non-empty parameter region where one-sided tree detection fails while partial graph alignment can be done in non-polynomial time, our conjecture would imply the \emph{hard phase} to be non-empty.
\end{conj*}

\subsubsection{Paper organization}
The outline of the chapter is as follows. We recall some notations and model of random trees and the in Section \ref{MPAlign:section:notations}. The derivation of the likelihood ratio between the relevant distributions is done in Section \ref{MPAlign:section:LR}, where points $(iii)$ and $(iv)$ of Theorem \ref{MPAlign:theorem:main_result_TREES} are proved (see \ref{proof_i_iii_iv_TH1}).  
In Section \ref{MPAlign:section:KL}, points $(ii)$ and $(v)$ of Theorem \ref{MPAlign:theorem:main_result_TREES} are proved (see Section \ref{proof_i_ii_TH1}) and a first sufficient condition for one-sided tree detectability (Theorem \ref{MPAlign:theorem:suff_cond_KL}) is obtained by analyzing Kullback-Leibler divergences: this condition is of the same kind as the one following from \cite{Ganassali20a} in Chapter \ref{chapter:NTMA}, however with a more direct derivation as well as a more explicit condition. Using a different approach, a second sufficient condition -- that of Theorem \ref{MPAlign:theorem:suff_cond_auto} -- is established in Section \ref{MPAlign:section:autos_GW} by analyzing the number of automorphisms of Galton-Watson trees. 

Next, we prove in Section \ref{MPAlign:section:hard_phase} another condition (see Theorem \ref{MPAlign:theorem:suff_hard_phase}) for the failure of one-sided detectability, hence showing that the conjectured hard phase is non-empty. The precise message-passing method for aligning graphs is introduced in Section \ref{MPAlign:section:graph_matching}, and guarantees on its output are established as well as the proof of Theorem \ref{MPAlign:theorem:main_result_GRAPHS}. 

Appendix \ref{MPAlign:appendix:numerical} is dedicated to numerical experiments as well as the description of the algorithm used in practice (\alg{MPAlign2}). Some additional proofs are deferred to Appendix \ref{MPAlign:appendix:additional_proofs}.

\section{Notations and problem statement}\label{MPAlign:section:notations}

\subsection{Notations}
In this first part we briefly introduce -- or recall -- some basic definitions that are used throughout the chapter.

\textit{Finite sets, permutations.}
For all $n>0$, we define $[n] := \left\lbrace 1, 2, \ldots, n \right\rbrace$. For any finite set $\mathcal{X}$, we denote by $\card{\mathcal{X}}$ its cardinal. $\cS_{\mathcal{X}}$ is the set of permutations on $\mathcal{X}$. We also denote $\cS_{k} = \cS_{[k]}$ for brevity, and we will often identify $\cS_{k}$ to $\cS_{\mathcal{X}}$ whenever $\card{\mathcal{X}}=k$. For any $0 \leq k \leq \ell$, we will write $\cS(k,\ell)$ (resp. $\cS(A,B)$) for the set of injective mappings from $[k]$ to $[\ell]$ (resp. between finite sets $A$ and $B$). By convention, $\card{\cS(0,\ell)}=1$.\\

\textit{Graphs.}
In a graph $G=(V,E)$ -- with node set $V$ and edge set $E$ -- we denote by $d_G(u)$ the degree of node $u$ in $G$ and $\cN_{G,d}(u)$ (resp. $\cS_{G,d}(u)$) the set of vertices at distance $\leq d$ (resp. exactly $d$) from node $u$ in $G$, $\cS_{G,d}(i)$. The \emph{neighborhood} of a node $u \in V$ is $\cN_G(u):= \cN_{G,1}(u)$, i.e. the set of all vertices that are connected to $u$ by an edge in $G$.\\

\textit{Labeled rooted trees.}
A \emph{labeled rooted tree} $t=(V,E)$ is an undirected graph with node set $V$ and edge set $E$ with no cycle. The \emph{root} of $t$ is a given distinguished node $\rho \in V$, and the \emph{depth} of a node is defined as its distance to the root $\rho$. The depth of tree $t$ is given as the maximum depth of all nodes in $t$. Each node $u$ at depth $d \geq 1$ has a unique \emph{parent} in $t$, which can be defined as the unique node at depth $d-1$ on the path from $u$ to the root $\rho$. Similarly, the \emph{children} of a node $u$ of depth $d$ are all the neighbors of $u$ at depth $d+1$. 

For any $u \in V$, we denote by $t_u$ the subtree of $t$ rooted at node $u$, and $c_t(u)$ the number of children of $u$ in $t$ -- or simply $c(u)$ where there is no ambiguity. Finally we define $\cV_d(t)$ (resp. $\cL_d(t)$) to be the set of nodes of $t$ at depth less than or equal to $d$ (resp. exactly $d$).\\

\textit{Canonical labeling.} A \emph{labeled rooted tree} can be canonically labeled by ordering nodes' children, giving the following labels. First, the label of the root node is set to the empty list $\varnothing$. Then, recursively, the label of a node $u$ is a list $\lbrace m,k \rbrace$ where $m$ is the label of its parent node, and $k$ is the rank of $u$ among the children of its parent.

We denote by $\cY_d$ the collection of such canonically labeled rooted  trees of depth no larger than $d$. Obviously, $\cY_0$ contains a single element, namely the rooted tree with only one node -- its root. Each tree $t$ in $\cY_d$ can be represented with a unique \emph{ordered list} $(t_1,\ldots,t_{c(\rho)})$ where each $t_u$ is the subtree of $t$ rooted at the $u-$th child of the root, and thus belongs to $\cY_{d-1}$. When $c(\rho)=0$, the previous ordered list is empty.\\


\textit{Tree subsampling.} For $s \in (0,1)$, a \emph{$s-$subsampling} of a tree $t$ is obtained by conserving every edge independently with probability $s$, and outputting the connected component of the root (which is still a tree). The nodes in the resulting tree inherit a canonical labeling from their order in the original tree.\\

\textit{Relabelings of trees.} 
A \emph{relabeling} $r(t)$ of a tree $t \in \cY_d$ is recursively identified as a permutation $\sigma\in\cS_{c(\rho)}$ of the children of the root node, together with relabelings $r_u(t_u)$ of its subtrees, resulting in tree $$r(t)=\left(r_{\sigma(1)}(t_{\sigma(1)}),\ldots,r_{\sigma(c(\rho))}(t_{\sigma(c(\rho))})\right).$$
A \emph{random uniform relabeling} $r(t)$ of a (un-)labeled tree $t$ of depth at most $d$ is defined as follows. Associate independently to each node $i$ of $t$ a permutation $\sigma_i$ of its children, uniformly distributed in $\cS_{c(i)}$. The relabeling is then defined by induction on the depth of nodes: the new label $r(\rho)$ of the root is $\varnothing$, and recursively, if the label of $u$ is $\lbrace m,k \rbrace$ and $v$ is the parent of $u$, we assign to $u$ the new label $$r(u) := \lbrace r(v),\sigma_v(k) \rbrace.$$
An important and easily verified property is that, for a given labeled tree $t \in \cY_d$, $r(t)$ is indeed uniformly distributed on the set of all possible relabelings of $t$. \\

\begin{figure}[h]
     \centering
     \begin{subfigure}[c]{0.39\textwidth}
         \centering
         \includegraphics[scale=0.8]{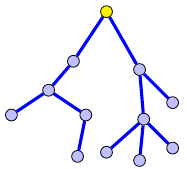}
         \caption{unlabeled rooted tree $t$}
         \label{fig:t_unlabeled}
     \end{subfigure}
 	\hfill
     \begin{subfigure}[c]{0.6\textwidth}
         \centering
         \includegraphics[scale=0.8]{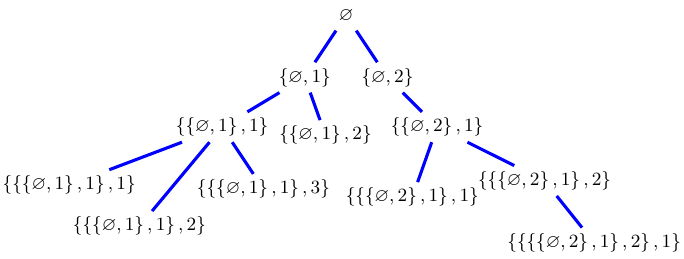}
         \caption{a random uniform relabeling of $t$}
         \label{fig:t_labeled}
     \end{subfigure}
     
    \caption{A rooted tree $t \in \cY_d$ with $n=4$ (the root is highlighted in yellow).}
    \label{fig:example_trees}
\end{figure}

\textit{Automorphisms of labeled trees.} 
Some of the relabelings of a labeled tree $t$ may be \emph{indistinguishable} from $t$, that is, equal to $t$ as labeled trees. These relabelings are called \emph{automorphisms of $t$}, and their set is denoted by $\Aut(t)$.\\

\textit{Injective mappings between labeled trees.} 
For two labeled trees $\tau,t\in\cY_d$, the set of \emph{injective mappings} from $\tau$ to $t$, denoted $\cS(\tau,t)$, is the set of injective mappings from the labels of vertices of $\tau$ to the labels of vertices of $t$ that preserve the rooted tree structure, in the sense that any $ \sigma  \in \cS(\tau,t)$ must verify 
\begin{equation*}
     \sigma (\varnothing)=\varnothing \quad \mbox{ and } \quad \sigma ( \left\lbrace p,k \right\rbrace)= \left\lbrace \sigma (p), j \right\rbrace \mbox{ for some $j$}.
\end{equation*} Note that $\cS(\tau,t)$ is not empty if and only if $\tau$ is, up to some relabeling, a subtree of $t$.\\

\textit{Probability.} 
For the sake of readability, we will denote by $\pi_\mu$ the Poisson distribution of parameter $\mu$, namely for all $k \geq 0$,
$
    \pi_\mu(k) := e^{-\mu} \frac{\mu^k}{k!}.
$

\subsection{Models of random trees, hypothesis testing}\label{MPAlign:subsection:model_random_trees}
We recall hereafter two models of random trees of Chapter \ref{chapter:NTMA}.

\subsubsection{Independent model $\dPl_d$} Under the independent model $\dPl_{d}$, $t$ and $t'$ are two independent $\GWl_{d}$, where $\lambda>0$ is the mean number of children in the graph.

\subsubsection{Tree augmentation} For $\lambda >0$ and $s \in [0,1]$, a (random) \emph{$(\lambda,s)-$augmentation} of a given tree $\tau =(V,E)$, denoted by $\Augls_d(\tau)$, is defined as follows. First, to each node $u$ in $V$ of depth $<d$, we attach a number $Z^{+}_u$ of additional children, where the $Z^{+}_u$ are i.i.d. of distribution $\Poi(\lambda (1-s))$. Let $V^+$ be the set of these additional children. To each $v \in V^+$ at depth $d_v$, we attach another random tree of distribution $\GWl_{d-d_v}$, independently of everything else.

\subsubsection{Correlated model $\dPls_d$} The correlated model $\dPls_{d}$ is built as follows: starting from an \emph{intersection tree} $\tau^\star \sim \GWls_d$, and $T$ and $T'$ are obtained as two independent $(\lambda,s)-$augmentations of $\tau^\star$. We denote $(T,T') \sim \dPls_d$.

In all these models, the labels of the trees $T$ and $T'$ are then uniformly resampled at random by the procedure described above. It can easily be verified that $T$ and $T'$ are marginally both  $\GWl_d$ under $\dPl_{d}$ and $\dPls_{d}$, namely. The parameters are $\lambda$, the mean number of children of a node, and the correlation $s$.

\begin{figure}[h]
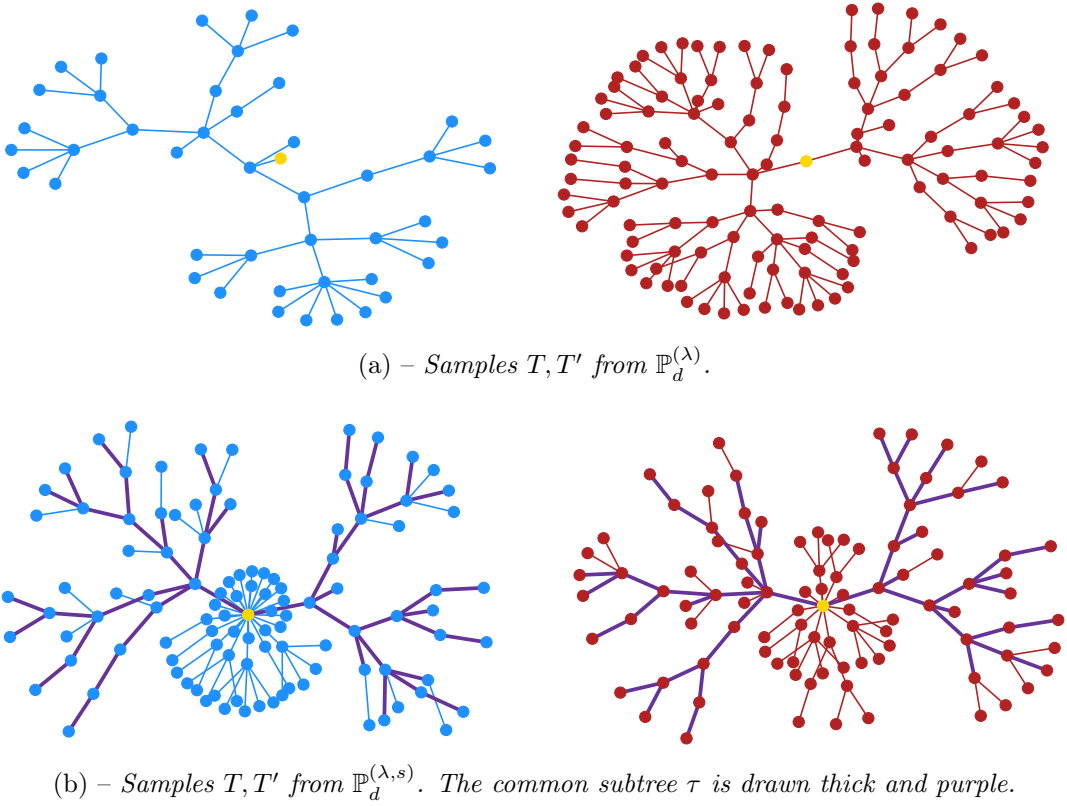

     \centering
     \begin{subfigure}[b]{\textwidth}
         \centering
         \includegraphics[scale=0.4]{chapters/6_MPAlign/pictures/P0blue.pdf}
         \hspace{0.25cm}
         \includegraphics[scale=0.4]{chapters/6_MPAlign/pictures/P0red.pdf}
         \caption{Samples $T,T'$ from $\dPl_{d}$.}
         \label{fig:P0}
     \end{subfigure}
     \begin{subfigure}[b]{\textwidth}
        \vspace{0.3cm}
         \centering
         \includegraphics[scale=0.4]{chapters/6_MPAlign/pictures/P1blue.pdf}
         \hspace{0.25cm}
         \includegraphics[scale=0.4]{chapters/6_MPAlign/pictures/P1red.pdf}
         \caption{Samples $T,T'$ from $\dPls_{d}$. The common subtree $\tau$ is drawn thick and purple.}
         \label{fig:P1}
     \end{subfigure}
     
    \caption{Samples from models $\dPl_{d}$ and $\dPls_{d}$, with $\lambda = 1.8$, $s=0.8$, and $d=5$. The root node is highlighted in yellow. Labels are not shown.}
\end{figure}

\subsubsection{Hypothesis testing, one-sided test} As mentioned earlier, we observe \emph{finite} trees in practice. A property that we will use implicitly in the sequel is that for $T,T'\sim \dPl_{d}$ (resp. $\sim \dPls_d$)  and $d'<d$, then $p_{d'}(T),p_{d'}(T')\sim \dPl_{d'}$ (resp. $\sim \dPls_{d'}$).

The hypothesis testing considered in this study can be formalized as follows: given the observation of a pair of trees $(t,t')$ in $\cY_d \times \cY_d$, we want to test
\begin{equation}
    \cH_0 = \mbox{"$t,t'$ are realizations under $\dPl_{d}$"} \quad \mbox{versus} \quad \cH_1 = \mbox{"$t,t'$ are realizations under $\dPls_{d}$"}.
\end{equation} 
More specifically, we are interested in being able to ensure the existence of a (asymptotic) \textit{one-sided test}, that is a test $\cT_n: \cY_d \times \cY_d \to \left\lbrace 0,1 \right\rbrace$ such that $\cT_n$ chooses hypothesis $\cH_0$ under $\dPl_{d}$ with probability $1-o(1)$, and chooses $\cH_1$ with some positive probability uniformly bounded away from 0 under $\dPls_{d}$, guaranteeing a vanishing type I error and a non vanishing power.
\begin{remark}\label{remark:one_sided_tests}
We here motivate one-sided tests once again. In statistical detection problems, the commonly considered tasks are that of  
\begin{itemize}
    \item \emph{strong detection}, i.e. tests $\cT_n$ that verify
    \begin{equation*}
        \underset{n \to \infty}{\lim} \left[\dPl_{d}\left( \cT_n(T,T') = 1 \right) + \dPls_{d}\left( \cT_n(T,T') = 0 \right)\right] = 0,
    \end{equation*}
    \item \emph{weak detection}, i.e. tests $\cT_n$ that verify
    \begin{equation*}
        \underset{n \to \infty}{\lim} \left[\dPl_{d}\left( \cT_n(T,T') = 1 \right) + \dPls_{d}\left( \cT_n(T,T') = 0 \right)\right] <1.
    \end{equation*}
\end{itemize} In other words, strong detection corresponds to discriminate  w.h.p. exactly the hypotheses, whereas weak detection corresponds to strictly outperforming random guess. We recall hereafter why neither strong detection nor weak detection are relevant for our problem. 

First, because of the event  that the intersection tree does not survive, which is of positive probability  under $\dPls_{d}$: we always have $\dPls_d(t,t') \geq C \cdot \dPl_d(t,t')$, with
\begin{equation*}
    C:= \frac{\pi_{\lambda s}(0)\pi_{\lambda (1-s)}(c)\pi_{\lambda (1-s)}(c')}{\pi_{\lambda }(c)\pi_{\lambda }(c')},
\end{equation*} where $c$ (resp. $c'$) is the degree of the root in $t$ (resp $t'$). This implies that $\dPl_d$ is always absolutely continuous w.r.t. $\dPls_d$, hence strong detection can never be achieved.

Second, weak detection is always achievable as soon as $s>0$: with the same notations as here above, the distribution of $c-c'$ is always centered but has different variance under $\dPl_d$ and under $\dPls_d$, hence these two distributions can be weakly distinguished, without any further assumption than $s>0$. Since we know by \cite{ganassali2021impossibility} that partial graph alignment is not feasible for $\lambda s \leq 1$, we conclude that weak detection in tree detection is not a relevant task either for graph alignment.
\end{remark}

\subsection{Warm-up discussion: the isomorphic case ($s=1$)}\label{subsection:warmup}

In this section, we discuss the graph alignment problem in the case where $s=1$ in the correlated \ER model \eqref{eq:CER_model}, namely when the graphs $G$ and $H$ are isomorphic, $\pi^{\star}$ being  one of the graph isomorphisms between $G$ and $H$. We then ask the question: \emph{what is the best fraction of nodes that can be recovered with high probability?} 

The answer to the above question comes with the following easy remark: the joint distribution of $(G,H)$ is invariant by any relabeling of $G$ according to some $\sigma \in \Aut(G)$, where $\Aut(G)$ denotes the automorphism group of $G$. The set of nodes that can be aligned w.h.p. is hence
\begin{equation}
    \label{eq:def_I(G)}
    \cI(G) := \left\lbrace u \in V(G), \; \forall \, \sigma \in \Aut(G), \sigma(u)=u \right\rbrace.
\end{equation}
In other words, $\cI(G)$ is the set of vertices of $G$ invariant under any automorphism.

Let us denote $\cC_1(G)$ the largest connected component of $G$ (the \emph{giant component}), and $\overline{\cC_1(G)}$ the subgraph made of all the smaller components. It is clear that
\begin{equation*}
    \Aut(G) = \Aut(\cC_1(G)) \times \Aut(\overline{\cC_1(G)}).
\end{equation*}
Recent work \cite{ganassali2021impossibility} shows that $\cI(G) \, \cap \, \overline{\cC_1(G)}$ contains at most a vanishing fraction of the points: it is not hard to see indeed that smaller components mainly consist in isolated trees, which are proved to have many copies in the graph when $n$ gets large, yielding some automorphisms that swap almost all vertices in $\overline{\cC_1(G)}$.
Hence, for our purpose, the main part of $\cI(G)$ comes from the study of $\Aut(\cC_1(G))$ and $\cI(\cC_1(G))$. 

When $G \sim \G(n,q)$, these sets have been thoroughly studied by Łuczak in \cite{luczak1988}. Vertices of the giant component that are not invariant under automorphism are mainly (i.e. up to $o(n)$ errors) vertices that do not belong to the \emph{2-core}\footnote{The \emph{2-core} of a graph is defined as the maximal subgraph of minimal degree at least $2$.} of $G$, denoted by $\cC^{(2)}(G)$. 

\begin{figure}[h]
    \centering
    \includegraphics[scale=0.55]{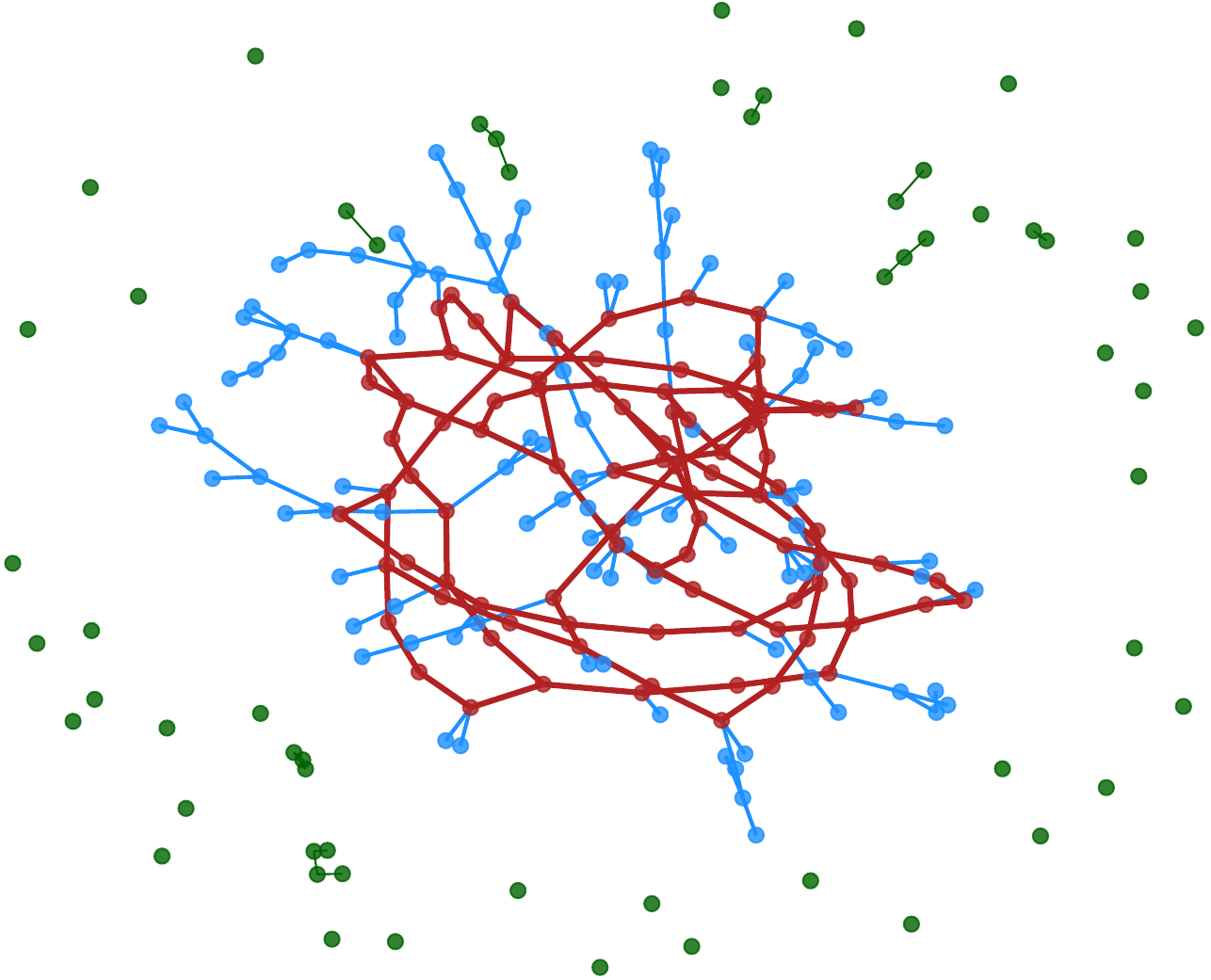}
    \caption{Sample $G$ from model $\G(n,\lambda/n)$, with $\lambda = 2$ and $n=250$. Vertices of $\overline{\cC_1(G)}$ (resp. of $\cC_1(G) \setminus \cC^{(2)}(G)$, $\cC^{(2)}(G)$) are drawn in green (resp. blue, red).}
    \label{fig:example_two_core}
\end{figure}

Simple structures appearing in $\cC_1(G) \setminus \cI(G)$ are leaves (degree one nodes) $v,w$ with common a neighbor $u$ in $\cC_1(G)$.
\cite{luczak1988} upper-bounds the size of $\cC_1(G) \setminus \cI(G)$ by the number of (generalizations) of such structures, thus obtaining the following

\begin{theorem}[\cite{luczak1988}, Theorems 3 and 4] \label{th:luczak}
Let $G \sim \G(n,q)$ with $q = \lambda/n$. Let $(K_n)_n$ be a sequence such that $K_n \to \infty$. There exists $\lambda_0 >0$ such that if $\lambda > \lambda_0$, then with high probability,
\begin{equation}\label{eq:th:luczak}
    \card{\cC^{(2)}(G)} - \card{\cI(\cC^{(2)}(G))} \leq K_n, \quad \mbox{and} \quad \card{\cC_1(G)} - \card{\cI(\cC_1(G))} \leq \lambda(\lambda + 5)e^{-2\lambda} n.
\end{equation}
\end{theorem}

Equation \eqref{eq:th:luczak} of Theorem \ref{th:luczak} states that for $\lambda$ large enough, almost all vertices of the 2-core of $G$ are invariant, whereas at most a fraction $\lambda(\lambda + 5)e^{-2\lambda}$ of the nodes are in the giant component and not in $\cI(G)$. In this case, with high probability, any isomorphism $\hat{\pi}$ between $G$ and $H$ will achieve partial recovery and will satisfy
$$ \ov \left(\hat{\pi},\pi^{\star} \right) \geq 1 - \pext(\lambda) - \lambda(\lambda + 5)e^{-2\lambda}, $$
where $\pext(\lambda)$ is defined as the probability that a Galton-Watson
tree of offspring $\Poi(\lambda)$ survives. 

However, finding efficiently such an isomorphism $\hat{\pi}$ is known to be challenging in the general case (see e.g. \cite{Arvind2002}): hence, whether there exists a polynomial-time algorithm achieving this optimal bound remains an open question\footnote{We can however cite a famous result of Bollob\'as (\cite{Bollobas2001}, Theorem 9.9) showing that in the dense case $ n p \geq \Theta(\log n)$, the vertices of every $G \sim \G(n,p)$ graph are uniquely determined by their distance sequences, and the automorphism group of $G$ is w.h.p. trivial.}.

\section{Derivation of the likelihood ratio} \label{MPAlign:section:LR}
For $t,t'\in \cY_d$, we introduce the likelihood ratio 
\begin{equation}
    \label{eq:def_LR}
    L_d(t,t'):= \frac{\dPls_{d}(t,t')}{\dPl_{d}(t,t')}.
\end{equation}

\subsection{Recursive computation}\label{MPAlign:subsection:recursion_L}
In this section, our aim is to obtain a recursive representation of the likelihood ratio $L_d$. First note that for two trees $t=(t_1,\ldots,t_c)$, $t'=(t'_1,\ldots,t'_{c'})$ both in $\cY_d$, we have
\begin{equation} \label{eq:MPAlign:P0_GW}
\dPl_{d}(t,t')=\GWl_{d}(t) \times \GWl_{d}(t'),
\end{equation} and that conditioned to $c$, $\GWl_{d}(t)$ satisfies the recursion
\begin{equation}\label{eq:MPAlign:recursion_GW}
\GWl_{d}(t) =\pi_{\lambda}(c)\prod_{u\in[c]} \GWl_{d-1}(t_u).
\end{equation} In the construction of $t,t'$ under $\cH_1$, partitioning on the permutations $\sigma \in\cS_c, \sigma' \in\cS_{c'}$ used to shuffle the children of the root nodes of $t$, $t'$, as well as on the number $k$ of children of the root in $\tau^{\star}$, we have the following 
\begin{multline*}
\dPls_{d}(t,t') = \sum_{k=0}^{c \wedge c'} \pi_{\lambda s}(k) \pi_{\lambda (1-s)}(c-k)\pi_{\lambda (1-s)}(c'-k) \\
\times \sum_{\sigma \in \cS_c,\sigma' \in \cS_{c'}}\frac{1}{c! \times c'!} \left(\prod_{u=1}^k\dP_{1,n-1}(t_{\sigma(u)},t'_{\sigma'(u)})\right) \\\times \left(\prod_{u=k+1}^d \GWl_{d-1}(t_{\sigma(u)})\right) 
\times \left(\prod_{i=k+1}^{d'} \GWl_{d-1}(t'_{\sigma'(u)})\right).
\end{multline*}

This together with Equations \eqref{eq:MPAlign:P0_GW}, \eqref{eq:MPAlign:recursion_GW} readily implies the following recursive formula for the likelihood ratio $L_d$:
\begin{equation}\label{eq:LR_rec_1}
L_d(t,t')=\sum_{k=0}^{c \wedge c'} \frac{\pi_{\lambda s}(k)\pi_{\lambda(1-s)}(c-k)\pi_{\lambda (1-s)}(c'-k)} {\pi_\lambda(c)\pi_\lambda(c') \times c! \times c'!}\sum_{\sigma\in\cS_c,\sigma'\in\cS_{c'}}\prod_{u=1}^k L_{d-1}(t_{\sigma(u)},t'_{\sigma'(u)}).
\end{equation}
In this expression, by convention the empty product equals 1. We will use in the sequel the following shorthand notation
\begin{flalign*}
\psi(k,c,c') & := \frac{\pi_{\lambda s}(k)\pi_{\lambda(1-s)}(c-k)\pi_{\lambda (1-s)}(c'-k)} {\pi_\lambda(c)\pi_\lambda(c')} \times \frac{(c-k)! \times (c-k')! }{c! \times c'!} \label{eq:def_psi} \\
& = e^{\lambda s} \times \frac{s^k (1-s)^{c+c'-2k}}{\lambda^{k} k!},
\end{flalign*} which enables an alternative, more compact recursive expression:
\begin{equation}\label{eq:MPAlign:LR_rec_2}
L_d(t,t')=\sum_{k=0}^{c \wedge c'}\psi(k,c,c')\sum_{\substack{\sigma \in \cS(k,c) \\ \sigma' \in \cS(k,c')}}\prod_{u=1}^k L_{d-1}(t_{\sigma(u)},t'_{\sigma'(u)}),
\end{equation} where we recall that $\cS(k,\ell)$ denotes the set of injective mappings from $[k]$ to $[\ell]$ and that by convention $\card{\cS(0,\ell)}=1$.

\begin{remark}\label{remark:util_rec_algo}
The above expression \eqref{eq:MPAlign:LR_rec_2} will be useful for efficient computations of the likelihood ratio in Algorithm \ref{MPAlign:algo_GA} in Section \ref{MPAlign:section:graph_matching}, through message-passing.
\end{remark}

\subsection{Explicit computation}
We now use the recursive expression \eqref{eq:MPAlign:LR_rec_2} to prove by induction on $d$ the following explicit formula for $L_d$.
\begin{lemma}\label{MPAlign:lemma:LR_developed}
With the previous notations, we have
\begin{equation}\label{eq:MPAlign:lemma:LR_developed}
L_d(t,t')=\sum_{\tau \in \cY_d} \sum_{\substack{\sigma \in \cS(\tau,t) \\ \sigma' \in \cS(\tau,t')} }\prod_{u \in \cV_{d-1}(\tau)}\psi\left(c_\tau(u),c_t(\sigma(u)),c_{t'}(\sigma'(u))\right).
\end{equation}
\end{lemma}

\begin{proof}[Proof of Lemma \ref{MPAlign:lemma:LR_developed}]
We prove this result by recursion on $d$. An empty product being set to $1$, there is nothing to prove in the case $d=0$. Let us first establish formula \eqref{eq:MPAlign:lemma:LR_developed} for $d=1$. In that case, the depth 1 trees $t$, $t'$ are identified by the degrees $c$, $c'$ of their root node. Since $\cY_0$ is a singleton, $L_0$ is identically 1, and from \eqref{eq:LR_rec_1} we have that
\begin{equation}\label{eq:expr_L1_rec}
    L_1(t,t')=\sum_{k=0}^{c \wedge c'}\frac{\pi_{\lambda s}(k)\pi_{\lambda(1-s)}(c-k)\pi_{\lambda (1-s)}(c'-k)}{\pi_\lambda(c)\pi_\lambda(c')}.
\end{equation} On the other hand, in evaluating expression \eqref{eq:MPAlign:lemma:LR_developed}, we only need consider trees $\tau$ in $\cY_1$ with root degree $k\leq c\wedge c'$, since for larger $k$, one of the two sets $\cS(\tau,t)$ or $\cS(\tau,t')$ is empty. For such $k$, we have $|\cS(\tau,t)|=c!/(c-k)!$. The right-hand term in \eqref{eq:MPAlign:lemma:LR_developed} thus writes
\begin{equation*}
    \sum_{k=0}^{c \wedge c'} \frac{c! \times c'!}{(c-k)! \times (c'-k)!} \psi(k,c,c'),
\end{equation*} which gives precisely \eqref{eq:expr_L1_rec}. 

Assume that \eqref{eq:MPAlign:lemma:LR_developed} has been established up to some $n-1\geq 1$. Expressing $L_{d}$ in terms of $L_{d-1}$ based on \eqref{eq:LR_rec_1}, and replacing in there the expression of $L_{d-1}$ by \eqref{eq:MPAlign:lemma:LR_developed} , we get
\begin{multline*}
   L_{d}(t,t') = \sum_{k=0}^{c \wedge c'} \frac{\psi(k,c,c')}{(c-k)!(c'-k)!} \\ \times
   \sum_{\sigma\in\cS_c,\sigma'\in\cS_{c'}} \prod_{u=1}^k \left[ \sum_{\tau_u \in\cY_{d-1}} \sum_{\substack{\sigma_u\in \cS(\tau_u,t_{\sigma(u)}) \\ \sigma'_u \in \cS(\tau_u,t'_{\sigma(u)})} }\prod_{v \in \cV_{d-1}(\tau_u)}\psi\left(c_{\tau_u}(v),c_{t_{\sigma(u)}}(\sigma_u(v)),c_{t'_{\sigma'(u)}}(\sigma'_u(v))\right)\right]. 
\end{multline*}

Note that the product term in the above expression depends on the permutations $\sigma$, $\sigma'$ only through their restriction to $[k]$: for given such restrictions there are $(c-k)! \times (c'-k)!$ corresponding pairs of permutations $\sigma,\sigma'$.

Moreover, there is a bijective mapping between an integer $k\in\left\lbrace 0,\ldots,c\wedge c'\right\rbrace$, pairs of injections $\sigma:[k]\to [c]$, $\sigma':[k]\to [c']$, $k$ trees $\tau_1, \ldots ,\tau_k\in\cY_{d-1}$, injections $\sigma_u \in \cS(\tau_u,t_{\sigma(u)})$ and $\sigma'_u \in \cS(\tau_u,t'_{\sigma'(u)})$ for all $u \in [k]$ and a tree $\tau \in \cY_d$ together with a pair of injections $\sigma,\sigma'\in \cS(\tau,t)\times \cS(\tau,t')$. This establishes formula \eqref{eq:MPAlign:lemma:LR_developed} at step $d$.

\end{proof} 

\subsection{Martingale properties and the objective of one-sided test} 
In this part, we assume that we observe $T,T'$ drawn under one of the two models $\dPl_{\infty}$ or $\dPls_{\infty}$. For $d \geq 0$, let $\cF_d:=\sigma(p_d(T),p_d(T'))$ be the sigma-field of the two trees $T,T'$ observed down to depth $d$. We then have
\begin{lemma}\label{MPAlign:lemma:LR_martingale}
The sequence $\left\lbrace L_d := L_d(p_d(T),p_d(T'))\right\rbrace_{d \geq 0}$ is a $\cF_d$-martingale under $\dPl_{\infty}$.
\end{lemma}

The above martingale property follows from general considerations of likelihood ratios. It is however informative to derive it by calculus, which we now do.
\begin{proof}[Proof of Lemma \ref{MPAlign:lemma:LR_martingale}]
There are several ways to see that $\left\lbrace L_d \right\rbrace_{d \geq 0}$ is a $\cF_d$-martingale under $\dPl_{\infty}$, depending on the formula used to write $L_{d+1}$ in terms of $L_d$. We here choose to use the developed expression \eqref{eq:MPAlign:lemma:LR_developed}, enabling simple computations:
\begin{multline*}
    L_{d+1} = \sum_{\tau \in \cY_{d+1}} \sum_{\substack{\sigma \in \cS(\tau,T) \\ {\sigma}' \in \cS({\tau},T')} }\prod_{u \in \cV_{d}({\tau})}\psi\left(c_{{\tau}}(u),c_T({\sigma}(u)),c_{T'}({\sigma}'(u))\right) \\
     = \sum_{\chi \in \cY_{d}} \sum_{\substack{\sigma \in \cS({\chi},p_d(T)) \\ \sigma' \in \cS(\chi,p_d(T'))} }\prod_{i\in \cV_{d-1}(\chi)}\psi\left(c_{\chi}(u),c_{p_d(T)}(\sigma(u)),c_{p_d(T')}(\sigma'(u))\right) \\
     \times \prod_{u \in \cL_d(\chi)} \sum_{k = 0}^{c_{T}(\sigma(u)) \wedge c_{T'}(\sigma'(u))} \frac{c_{T}(\sigma(u))! c_{T'}(\sigma'(u))!}{(c_{T}(\sigma(u))-k)! (c_{T'}(\sigma'(u))-k)!} \psi(k,c_{T}(\sigma(u)),c_{T'}(\sigma'(u))).
\end{multline*}
The last product is independent from $\cF_d$. Moreover, under $\dPl_{\infty}$, all terms in the last product are independent, the $c_{T}(u)$ and $c_{T'}(u)$ being independent $\Poi(\lambda)$ random variables. Since for any independent $\Poi(\lambda)$ random variables $c,c'$, one has
\begin{equation*}
    \dE\left[ \sum_{k=0}^{c \wedge c'}\frac{\pi_{\lambda s}(k)\pi_{\lambda(1-s)}(c-k)\pi_{\lambda (1-s)}(c'-k)}{\pi_\lambda(c)\pi_\lambda(c')} \right]=1,
\end{equation*} taking the expectation conditionally to $\cF_d$ entails the desired martingale property.
\end{proof}

We now consider the martingale almost sure limit $L_{\infty}$, and define $\ell := \dEl_{\infty}\left[ L_{\infty} \right]$. Using the recursive formula \eqref{eq:LR_rec_1} and conditioning on the root degrees $c$ and $c'$, it follows that  $\ell$ verifies the following fixed point equation
\begin{equation}\label{eq:fixed_point_l}
\ell=\sum_{k\geq 0}\pi_{\lambda s}(k) \ell^k.
\end{equation}
This is also (!) the fixed point equation for the extinction probability $\pext(\lambda s)$ of a Galton-Watson branching process with offspring distribution $\Poi(\lambda s)$.  For $\lambda s \leq 1$, the only solution of \eqref{eq:fixed_point_l} is $\ell=1$. For $\lambda s>1$, the equation also admits a non-trivial solution $\pext(\lambda s) \in (0,1)$.

Our goal is to find conditions on $(\lambda,s)$ for which the martingale $\left\lbrace L_d \right \rbrace_{d \geq 0}$ is not uniformly integrable and \emph{looses mass} at infinity, i.e. the conditions for which the martingale limit $L_{\infty}$ has expectation $\dEl_{\infty} \left[L_{\infty} \right]<1$. By the previous calculation we know that if this holds, then necessarily  $\dEl_{\infty} \left[L_{\infty} \right]=\pext(\lambda s)<1$. Simulations of $L_d$ displayed on Figure \ref{fig:simulation_LL} seem to indicate that its transition to non-uniform integrability does not coincide with the condition $\lambda s>1$. We shall obtain a theoretical confirmation of this fact with Theorem \ref{MPAlign:theorem:suff_hard_phase}.

Our interest in conditions for non-uniform integrability stem from the following simple Lemma:

\begin{lemma}
\label{lemma:one_sided_test}
Assume that $\dEl_{\infty} \left[L_{\infty} \right]<1$. Then there exists a one-sided test.
\end{lemma}

\begin{proof}[Proof of Lemma \ref{lemma:one_sided_test}]
Let us take $a>0$ a continuity point of the law of $L_\infty$ under $\dPl_{\infty}$. We have
\begin{equation}\label{eq:suff_test1}
\lim_{d\to\infty}\dPl_{\infty}(L_d>a)=\dPl_{\infty}(L_\infty>a).
\end{equation} Moreover,
\begin{flalign*}
1 = \dEl_{\infty} \left[L_d\right] & = \dEl_{\infty} \left[L_d \one_{L_d>a}\right] + \dEl_{\infty} \left[L_d \one_{L_d \leq a}\right]\\
& = \dPls_{\infty}(L_d>a) + \dEl_{\infty}  \left[L_d \one_{L_d \leq a}\right].
\end{flalign*} 
The last equation implies, under the assumption $\dEl_{\infty} \left[L_{\infty} \right]<1$ (that is $\dEl_{\infty} \left[L_{\infty} \right] = \pext(\lambda s)$), that
\begin{equation}\label{eq:suff_test2}
\liminf_{d \to\infty}\dPls_{\infty}(L_d>a)\geq 1-\dEl_{\infty} \left[L_\infty \right]=1-\pext(\lambda s)>0.
\end{equation}
In view of \eqref{eq:suff_test1} and \eqref{eq:suff_test2}, we can thus choose $a_d \to\infty$ such that:
\begin{equation*}
\lim_{d\to\infty} \dPl_{\infty}(L_d>a_d)=0 \quad \hbox{and} \quad \liminf_{d\to\infty}\dPls_{\infty}(L_d > a_d)\geq 1-\pext(\lambda s)>0.
\end{equation*}
\end{proof}

\subsubsection{Proof of $(i) \iff (iii) \iff (iv)$ in Theorem \ref{MPAlign:theorem:main_result_TREES}}
\label{proof_i_iii_iv_TH1}
\begin{proof}
The previous proof shows first that $(i) \iff (iii)$ in Theorem \ref{MPAlign:theorem:main_result_TREES} (applying  Neyman-Pearson's Lemma and a diagonal extraction procedure) as well as $(iii) \iff (iv)$, since condition
\begin{equation}\label{eq:write_cond_NON_UI}
\exists \, \eps>0, \; \forall a>0, \; \liminf_{d\to\infty}\dPls_{\infty}(L_d>a)\geq \eps>0
\end{equation} is exactly the condition of non-uniform integrability of the martingale $(L_d)_d$ with respect to $\dPl_{\infty}$.
\end{proof}

\subsection{A Markov transition kernel on trees}

In this section, we introduce a Markov transition semi-group on trees that arises naturally in our study. Indeed, the joint distribution of the pair of trees $(T,T')$ under $\dPls_{d}$ will be, up to relabeling, interpreted as the joint distribution of $(X_0,X_{r})$, where $X_0$ is the initial state of this Markov process, distributed according to its stationary distribution $\GWl_{d}$, and $X_r$ is its state at time $r$. The time parameter $r$ is in one-to-one correspondence with the correlation parameter $s$ of our model, through the relation 
\begin{equation*}
    r=-\log(s).
\end{equation*}

For $n>0$, we define $M_{d}$ the linear operator indexed on trees of $\cY_{d}$, defined as follows:
\begin{equation}\label{eq:def_M}
	M_{d}(t,t') := \frac{\dPls_{d}(t,t')}{\dPl_{d}(t)}.
\end{equation}
$M_{d}$ is identified to the \emph{transition kernel} of the Markov chain with transitions denoted $t \underset{\lambda,s}{\longrightarrow} t'$ where $t'$ is obtained from $t$ following the following three-step procedure:
\begin{itemize}
\item[$1.$] Extracting $\tau$, a $s-$subsampling of $t$;
\item[$2.$] Draw $\tau^+$, an augmentation $\Augls_d$ of $\tau$;
\item[$3.$] Take $t'$ to be a uniform relabeling of $\tau^+$.
\end{itemize}
We next denote $M_{d}(s)$ this transition kernel to emphasize its dependence on $s$.

\begin{figure}[h]
    \centering
    \includegraphics[scale=0.42]{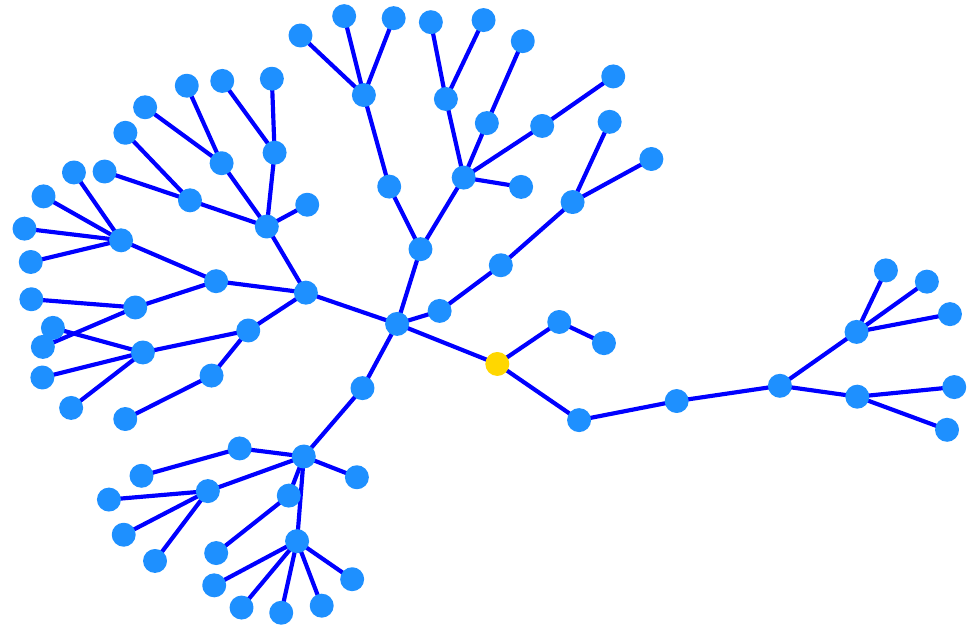}
    \hspace{0.06cm}
    \includegraphics[scale=0.42]{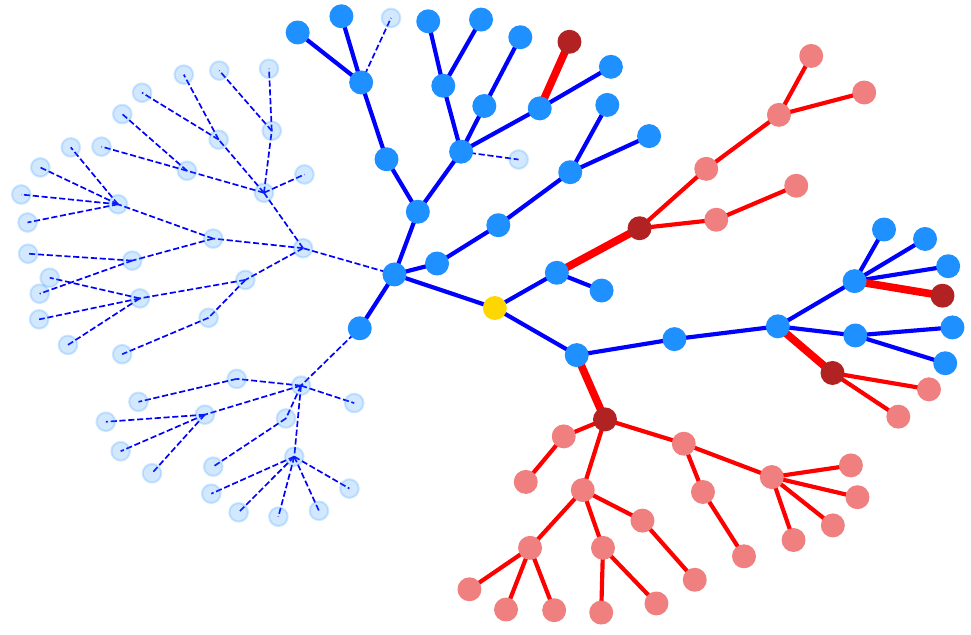}
    \caption{Example of a transition described hereabove, with $\lambda = 1.85$, $s=0.85$, at depth $d=5$. The original tree $t$ is drawn on the left. On the right, $t'$ is obtained as follows: first extracting a $s-$subsampling $\tau$ of $t$ (dashed blue edges are deleted), and drawing a $(\lambda,s)-$augmentation of $\tau$ -- first attaching new children to all vertices of $\tau$ (dark red nodes with thick edges), and attaching new Galton-Watson trees to these new children (light red nodes with standard edges).  Labels are not shown.}
    \label{fig:transition_markov}
\end{figure}

\begin{figure}[h]
     \begin{subfigure}[b]{\textwidth}
         \centering
         \includegraphics[width=0.6\textwidth]{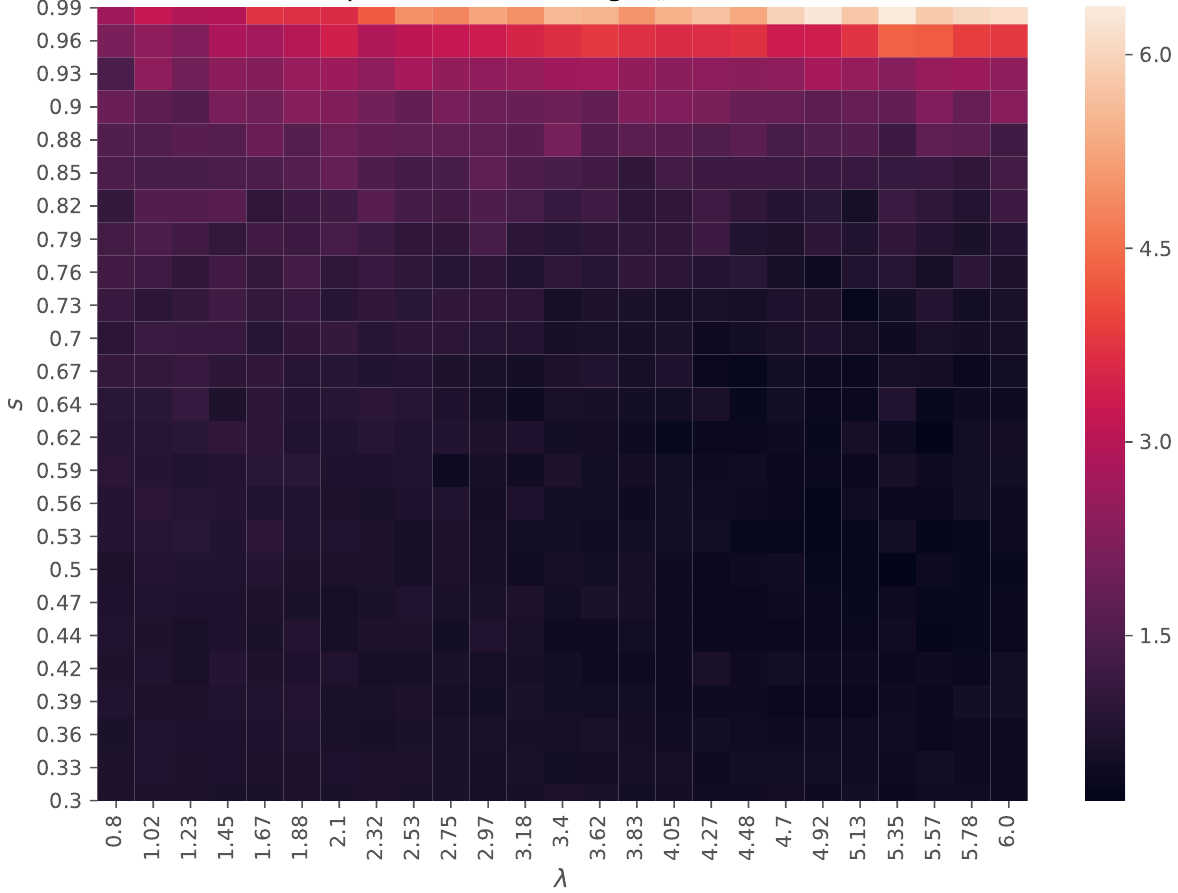}
         \caption{Empirical mean of $\log L_d$ for $d=2$. \scriptsize {$75$ simulations per value of $(\lambda,s)$.}}
         \label{fig:meanlogL_d2}
     \end{subfigure}
     \hfill
    
     \begin{subfigure}[b]{\textwidth}
        \vspace{0.25cm}
         \centering
         \includegraphics[width=0.6\textwidth]{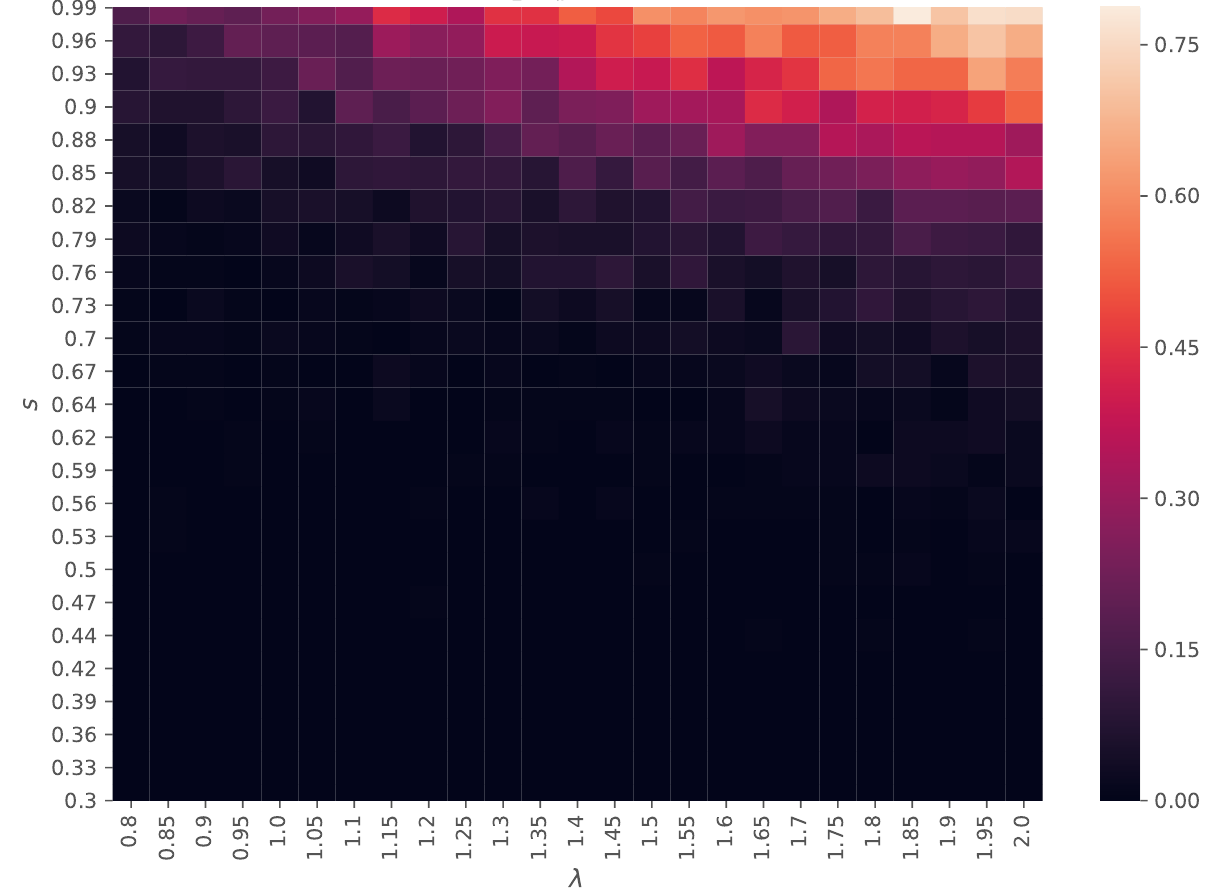}
         \caption{Estimate of $\dPls_d(L_d>\beta)$ for $d=3$, $\beta=10^2$. \scriptsize {$75$ simulations per value of $(\lambda,s)$.}}
         \label{fig:L_d3}
     \end{subfigure}
     \hfill
     
     \begin{subfigure}[b]{\textwidth}
        \vspace{0.25cm}
         \centering
         \includegraphics[width=0.6\textwidth]{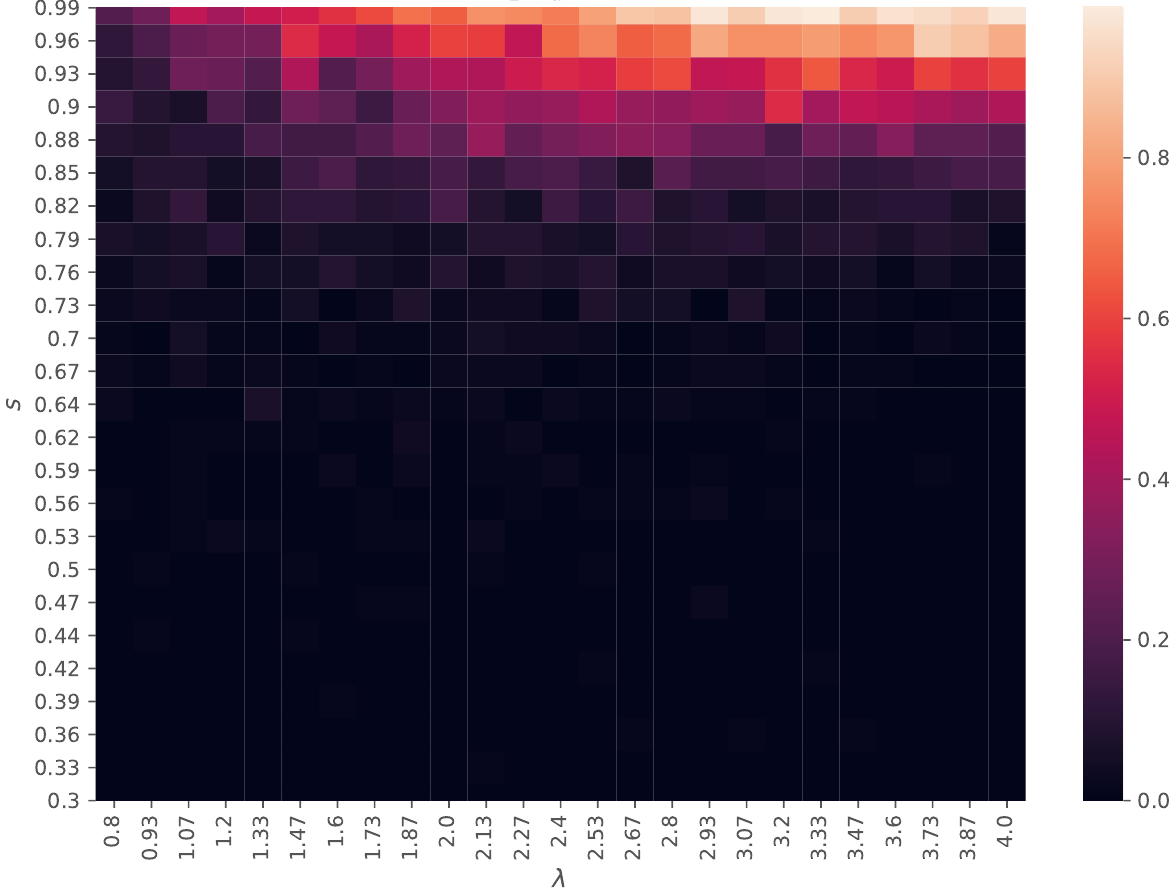}
         \caption{Estimate of $\dPls_d(L_d>\beta)$ for $d=5$, $\beta=5.10^3$. \scriptsize {$150$ simulations per value of $(\lambda,s)$.}}
         \label{fig:L_d5}
     \end{subfigure}
    \caption{Simulations of $L_d$ under model $\dPls_{s}$.}
    \label{fig:simulation_LL}
\end{figure}

A remarkable property of this kernel is the following semi-group structure:
\begin{proposition}[Consistency of kernels $M_d(s)$]\label{prop:consistency}
Let $\lambda >0$ and $s, s' \in [0,1]$. Then, for all $n \geq 1$,
\begin{equation}\label{eq:prop:consistency}
	M_d(s) M_d(s') = M_d(s') M_d(s) = M_d(ss').
\end{equation}
\end{proposition}
\begin{proof}
The proof consists in verifying that applying transitions $M_d(s)$ and $M_d(s')$ successively is equivalent in distribution to applying transition $M_d(s s')$. Let us first show that the unlabeled structures of the trees are equivalent in distribution. For $t \in \cY_d$, let us sample a sequence $t \underset{s}{\longrightarrow} \tilde{t} \underset{s'}{\longrightarrow} t'$ as follows. 

For $t \in \cY_d$, let us apply a first transition $t \underset{s}{\longrightarrow} \tilde{t}$: we extract $\tilde{\tau}$, a $s-$subsampling of $t$. To each vertex $u$ of $\tilde{\tau}$ we attach an independent number $\Poi(\lambda (1-s))$ of new children. The set of these new vertices is denoted $\tilde{V}_+$. Then, to each vertex $u \in \tilde{V}_+$ we attach an independent tree $\tilde{t}_u$ with distribution $\GW_\lambda$. We just sampled the unlabeled version of $\widetilde{t}$.

Let us now apply the second transition $\tilde{t} \underset{s'}{\longrightarrow} t'$. We sample $t$ as follows:
\begin{itemize}
    \item[$1.$] First, we sample all vertices of $\tilde{\tau}$ in $\tilde{t}$, keeping them independently with probability $s'$. The obtained subtree is denoted by $\tau$;
    \item[$2.$] To any vertex $u$ of $\tau$, we keep each previous child vertex in $\tilde{V}_+$ independently with probability $s'$, the set of children that are kept is denoted by $V^1_+$;
    \item[$3.$] To any vertex $u$ of $\tau$, we attach an independent number $\Poi(\lambda (1-s'))$ of new children. The set of these new vertices are referred to as $V^2_+$.
    \item[$4.$] To any vertex $u \in V^1_+$, we sample a transition $\tilde{t}_u \underset{s'}{\longrightarrow} t_u$, and attach $t_u$ to node $u$.
    \item[$5.$] To each vertex $v \in V^2_+$ we attach an independent tree $t_v$ with distribution $\GW_\lambda$.
\end{itemize}

Eventually we performed the following process: from the initial tree $t$, we extracted $\tau$ as a $ss'-$subsampling of $t$, and we attached to each vertex of $\tau$ some new children: the sum of two independent $\Poi(\lambda(1-s)s')$ (for children in $V^1_+$) and $\Poi(\lambda(1-s'))$ (for children in $V^2_+$), hence again of Poisson distribution with parameter $\lambda(1-s)s' + \lambda(1-s') = \lambda(1-ss')$. By steps 4. and 5., the trees attached to every vertex in $V_+ := V^1_+ \cup V^2_+$ are i.i.d. with distribution $\GW_\lambda$, independent of $t$. Hence, the unlabeled version of $t'$ can also be obtained from $t$ with the transition $t \underset{ss'}{\longrightarrow} t'$.

Finally, the definition of the tree subsampling ensures that the composition of the two relabelings in the two steps gives indeed a uniform relabeling of $t$, which ends the proof.

\end{proof}

\section{Conditions based on Kullback-Leibler divergences}\label{MPAlign:section:KL}
In the sequel we shall denote
\begin{equation}\label{eq:def_KL}
\KL_d:=\KL(\dPls_d\Vert\dPl_d)= \dEls_{d}  \left[ \log(L_d) \right].
\end{equation}
Note that by convexity of $\phi:x\to x\log(x)$, the martingale property of likelihood ratios $L_d$ under $\dPl_d$ and Jensen's inequality, the sequence $\KL_d$ is increasing with $d$ and therefore admits a limit $\KL_{\infty}$ as $d\to\infty$. 

\subsection{Phase transition for $\KL_\infty$}
Let us start with a simple proposition.

\begin{proposition}\label{MPAlign:proposition:KL_vs_Ent}
One has $\KL_d \leq \Ent (\GW^{(\lambda s)}_{d})$.
\end{proposition}
\begin{proof}
Consider the Markov transition kernel $K_d$ from $\cY_d^2$ to $\cY_d^2$ such that $K_d((\tau,\tau'),(t,t'))$ is the probability that independent $(\lambda,s)-$augmentations and relabelings of $(\tau,\tau')$ to depth $d$ produce the two trees $(t,t')$. 

Thus $\dPls_d$ is the law obtained by applying kernel $K_d$ to the distribution of $(\tau,\tau)$, where $\tau\sim \GW^{(\lambda s)}_{d}$ whereas $\dPl_d$ is the law obtained by applying kernel $K_d$ to the distribution of two independent $\GW^{(\lambda s)}_{d}$ trees $(\tau,\tau')$. Standard monotonicity properties of Kullback-Leibler divergence then guarantee that $\KL_d$ is upper-bounded by $\KL(\cL(\tau,\tau)\Vert\cL(\tau,\tau'))$. This divergence reads
\begin{equation*}
    \sum_{\tau\in\cY_d}\GW^{(\lambda s)}_{d}(\tau)\log\left( \frac{\GW^{(\lambda s)}_{d}(\tau)}{\GW^{(\lambda s)}_{d}(\tau)^2}\right)=\Ent(\GW^{(\lambda s)}_{d}).
\end{equation*}
\end{proof}
This readily implies the following
\begin{corollary}\label{MPAlign:cor:impossibility_KL}
Assume $\lambda s<1$. Then 
\begin{equation}\label{MPAlign:eq:cor:impossibility_KL}
    \KL_\infty=\lim_{d\to\infty} \KL_d \leq \frac{1}{1-\lambda s}\Ent(\pi_{\lambda s})<+\infty.
\end{equation}
\end{corollary}
\begin{proof}
Entropy $\Ent(\GW^{(\lambda s)}_{d})$ can be evaluated by the conditional entropy formula as
\begin{equation*}
    \Ent(\GW^{(\lambda s)}_{d})=\Ent(\GW^{(\lambda s)}_{d-1})+(\lambda s)^{d-1}\Ent(\pi_{\lambda s}).
\end{equation*}
The result follows from Proposition \ref{MPAlign:proposition:KL_vs_Ent}.
\end{proof}

We then have the following result:

\begin{proposition}\label{prop:suff_cond_KL}
Existence of one-sided tests holds if $\lambda s>1$ and $\KL_\infty = +\infty$, whereas it fails if $\KL_\infty <+\infty$. 
\end{proposition}

\begin{proof}
Assume existence of one-sided tests. As previously mentioned, equivalently there exists $\eps>0$ such that
\begin{equation*}
    \forall a>0,\; \liminf_{d\to\infty}\dPls_{d}(L_d>a)\geq \eps.
\end{equation*}
Fix $a>0$, and define for $d\in\dN$, $C_d := \{x\in\cY_d^2: L_d(x)>a\}$.
Write then, noting $\phi(u):=u\log(u)$:
\begin{flalign*}
\KL_d &\geq \dPls_d(C_d)\log(a)+\sum_{x\in\overline{C_d}}\dPls_d(x)\log \frac{\dPls_d(x)}{\dPl_d(x)}\\
&\geq \dPls_d(C_d)\log(a) +\dPl_d(\overline{C_d})\sum_{x\in\overline{C_d}}\frac{\dPl_d(x)}{\dPl_d(\overline{C_d})}\phi(L_{d}(x))\\
&\overset{(a)}{\geq} \dPls_d(C_d)\log(a) + \dPl_d(\overline{C_d}) \phi\left(\sum_{x\in \overline{C_d}}\frac{\dPl_d(x)}{\dPl_d(\overline{C_d})}L_d(x)\right)\\
&=\dPls_d(C_d)\log(a)+\dPls_d(\overline{C_d})\log\left(\dPls_d(\overline{C_d})\right) - \dPls_d(\overline{C_d}) \log({\dPl_d(\overline{C_d})})\\
&\geq \dPls_d(C_d)\log(a) +\inf_{u\in[0,1]}\phi(u) \geq \dPls_d(C_d)\log(a) -e^{-1} \, .
\end{flalign*}
We used convexity of $\phi$ in $(a)$. It thus follows from characterization of one-sided testability that for all $a>0$, 
\begin{equation*}
    \KL_\infty\geq \eps \log(a)-e^{-1},
\end{equation*} and thus $\KL_\infty=+\infty$.

Conversely, assume  $\lambda s>1$ and $\KL_\infty=+\infty$. Let under $\dPls_{\infty}$ define $$w:=\lim_{d\to\infty} \card{\cL_d(\tau^{\star})} (\lambda s)^{-d} \, .$$
On the event that $\tau^{\star}$ survives, which has strictly positive probability for $\lambda s>1$, it holds that $w>0$. In addition, we let $\pi^{\star}$, $(\pi')^{\star}$ denote the injections from $\tau^{\star}$ to $T$ and $T'$ respectively that result from uniform shuffling of the augmentations of $\tau^{\star}$.

Let $d,m$ be two integers. One then has the lower bound:
\begin{flalign*}
L_{d+m}(T,T')& \geq \prod_{u \in \cV_{d-1}(\tau^{\star})}\psi(c_{\tau^{\star}}(u),c_{T}(\pi^{\star}(u)),c_{T'}((\pi')^{\star}(u))\prod_{u\in \cL_d(\tau^{\star})}L_m(T_{\pi^{\star}(u)},T'_{(\pi')^{\star}(u)})
\\
&\geq \prod_{u\in \cV_{d-1}(\tau^{\star})}\psi(c_{\tau^{\star}}(u),c_{T}(\pi^{\star}(u)),c_{T'}((\pi')^{\star}(u))) e^{\card{\cL_d(\tau^{\star})}[\dEls_{m} \left[\log L_m\right]-o(1)]}.
\end{flalign*}
For $d$ large, by the law of large numbers, the first product is with high probability lower-bounded by $e^{C w (\lambda s)^d}$ for some fixed constant $C$. Choosing $m$ of order 1 but sufficiently large, since by assumption $\lim_{m\to\infty}\dEls_{m} \left[\log(L_m)\right]=+\infty$, we can ensure that the second factor is larger than $e^{C' w (\lambda s)^d}$ for some arbitrary $C'$. Taking $C'$ large enough ensures that, on the event that $\tau^{\star}$ survives, $\lim_{d\to\infty}L_d=+\infty$ almost surely. This readily implies one-sided testability.
\end{proof}

\subsubsection{Proof of $(i) \iff (ii) \iff (v)$ in Theorem \ref{MPAlign:theorem:main_result_TREES}}
\label{proof_i_ii_TH1}
\begin{proof}
Proposition  \ref{prop:suff_cond_KL} gives the implication $(ii)\implies (i)$. Its proof further gives $(ii)\implies (v)$. The converse $(v)\implies (ii)$ is obvious.  The second statement in Proposition  \ref{prop:suff_cond_KL} gives $(i)\implies \KL_\infty=+\infty$. To obtain that $(i)\implies (ii)$ and conclude, it thus only remains to show that $(i)\implies \lambda s>1$. 
 
As will be shown in Section \ref{MPAlign:section:graph_matching}, one-sided testability implies (polynomial-time) feasibility of partial graph alignment. However, \cite{ganassali2021impossibility} established that partial alignment is not feasible when $\lambda s \leq 1$ (see Chapter \ref{chapter:impossibility}). This establishes  $(i)\implies \lambda s>1$ as required.
\end{proof}

\subsection{Applications}
To apply condition $(ii)$ of Theorem \ref{MPAlign:theorem:main_result_TREES}, let us first establish the following
\begin{lemma}\label{lemma:geo_rec_KL}
For all $d \geq 1$, one has
\begin{equation}\label{eq:lemma:geo_rec_KL}
\KL_{d+1}\geq \lambda s \KL_d +\lambda s \left(\log(s/\lambda)  +1\right) +2\lambda (1-s)\log(1-s).
\end{equation}
\end{lemma}
\begin{proof}
Let $c$ denote under $\dPls_{d}$ the degree of $\tau^{\star}$'s root, and $c+\Delta$ (respectively $c+\Delta')$ the degree of the root nodes in $T$ and $T'$. By the recursive formula for $L_d$, considering only the term for $k=c$ in the first summation as well as the injections $\sigma:[c]\to[c+\Delta]$, $\sigma':[c]\to[c+\Delta']$ that correctly match the $c$ children of $\tau^{\star}$'s root in $T$ and $T'$, of which there are exactly $c!$ pairs, one has:
\begin{flalign*}
L_d(T,T')&\geq \psi(c,c+\Delta,c+\Delta') \times c!\times \prod_{u=1}^{c} L_{d-1}(T_u,T'_u) \\
&\geq e^{\lambda s}\frac{s^{c}(1-s)^{\Delta+\Delta'}}{\lambda ^{c}} \times \prod_{u=1}^{c} L_{d-1}(T_u,T'_u) .
\end{flalign*}
Taking logarithms and then expectations, since $\dEls_{d} \left[c\right]=\lambda s$ and $\dEls_{d} \left[ \Delta \right]=\dEls_{d} \left[\Delta'\right]=\lambda(1-s)$, the result follows. 
\end{proof}

We then have
\begin{corollary}\label{cor:suff_cond_KL}
Assume that $\lambda s>1$ and 
\begin{equation}\label{eq:cor:suff_cond_KL}
\KL_1 > \frac{1}{\lambda s-1} \left[\lambda s(\log(\lambda/s)-1)-2\lambda(1-s)\log(1-s)\right].
\end{equation}
Then $\KL_\infty=+\infty$.
\end{corollary}
\begin{proof}
This follows from \eqref{eq:lemma:geo_rec_KL}: indeed together with \eqref{eq:cor:suff_cond_KL} it implies that for all $d \geq 1$, 
$$\KL_{d+1}-\KL_1 \geq \lambda s(\KL_{d}-\KL_1),$$
hence $\KL_{d}$ diverges geometrically to infinity, provided we have $\KL_2>\KL_1$. The latter property is established by writing 
$$
\KL_2=\dEl_{\infty} \left[\phi(L_2) \right]=\dEl_{\infty}\left[ \dEl_{\infty}\left[ \phi(L_2)|\cF_1\right]\right]
$$ 
where $\phi(x)=x\log(x)$ is strictly convex. Jensen's inequality thus guarantees $\KL_2 \geq \KL_1=\dEl_{\infty} \left[\phi(L_1) \right]$, with equality only if almost surely, $L_2=L_1$. However this almost sure equality does not hold, hence the result. 
\end{proof}
These results have the following consequence:
\begin{theorem}\label{MPAlign:theorem:suff_cond_KL}
Assume that $\lambda \in(1,e)$. Let 
\begin{equation}
s^*(\lambda):=\sup\{s\in[0,1]: s(\log(\lambda/s)-1)-2(1-s)\log(1-s)\geq 0\}.
\end{equation}
Then $s^*(\lambda)<1$, and under the conditions
\begin{equation} \label{condition:lambda_s_KL}
    \lambda \in(1,e) \mbox{ and } s\in(s^*(\lambda),1],
\end{equation} one-sided detectability holds. 
\end{theorem}
\begin{proof}
The fact that $s^*(\lambda)<1$ follows by continuity, since for $s=1$ the function 
$$
s\to s(\log(\lambda/s)-1)-2(1-s)\log(1-s)
$$
evaluates to $\log(\lambda)-1$, which is negative by the assumption $\lambda <e$. By definition, for $s\in (s^*(\lambda),1]$, the right-hand side of \eqref{eq:cor:suff_cond_KL} is less than or equal to zero. Since $\KL_1>0$, the result is a consequence of Corollary \ref{cor:suff_cond_KL} and Proposition \ref{prop:suff_cond_KL}.  
\end{proof}
\begin{remark}
A result similar to that of Theorem \ref{MPAlign:theorem:suff_cond_KL} follows from \cite{Ganassali20a}. The present derivation is however more direct, and allows for more explicit upper bound $\lambda_0=e$ on the range of values of $\lambda$ considered, as well as characterization of the function $s^*(\lambda)$ involved.
\end{remark}

Condition \eqref{eq:cor:suff_cond_KL} of Corollary \ref{cor:suff_cond_KL} can also be used to identify conditions on $s$ for one-sided testability for large values of $\lambda$, based on corresponding evaluations of $\KL_1$. However, the resulting conditions do not appear as sharp as those obtained by the analysis of automorphisms of $\tau^{\star}$, that is the object of the next Section.

\section{Number of automorphisms of Galton-Watson trees}\label{MPAlign:section:autos_GW}
In this Section, we show how counting automorphisms of Galton-Watson trees gives a sufficient condition for the existence of one-sided tests in the tree correlation detection problem, and provide evaluations of this number of automorphisms. 

\subsection{A lower bound on the likelihood ratio}
Under $\dPls_{\infty}$, recall that $\tau^{\star}$ is the \emph{true} intersection tree used to perform correlated construction of $T$ and $T'$, and  $\pi^{\star}$, $(\pi')^{\star}$ denote the injections from $\tau^{\star}$ to $T$ and $T'$ respectively that result from uniform shuffling of the augmentations of $\tau^{\star}$. Without loss of generality, we assume in this section that $\pi^{\star}$ and $(\pi')^{\star}$ are the identity map. 
We denote, for each $u\in \cV_{d-1}(\tau^{\star})$:
\begin{equation}
c_u:=c_{\tau^{\star}}(u),\; \Delta_u:=c_{T}(u)-c_{\tau^{\star}}(u),\; \Delta'_u:=c_{T'}(u)-c_{\tau^{\star}}(u).
\end{equation} We now prove the following
\begin{lemma}\label{lemma:lower_bound_LR}
Under $\dPls_{d}$ we have the lower bound:
\begin{equation}\label{eq:lemma:lower_bound_LR}
L_d = L_d(T,T')\geq |\Aut(\tau^{\star})| \prod_{u\in \cV_{d-1}(\tau^{\star})}\frac{s^{ c_i}(1-s)^{\Delta_u+\Delta'_u}}{e^{-\lambda s}\lambda^{c_u}}\prod_{u\in \cL_{d-1}(\tau^{\star})}\binom{c_u+\Delta_u}{c_u}\binom{c_u+\Delta'_u}{c_u},
\end{equation}
where we recall that $\Aut(\tau^{\star})$ denotes the set of tree automorphisms of $\tau^{\star}$.
\end{lemma}

\begin{proof}
In view of the developed expression \eqref{eq:MPAlign:lemma:LR_developed}, we can lower-bound $L_d(T,T')$ by writing
\begin{equation}\label{eq:lower_bound_LR_1}
    L_d(T,T')\geq \sum_{\substack{\tau\in \cY_d \\ \tau \equiv \tau^{\star}}}\sum_{ \substack{\sigma \in \cS(\tau,T) \\ \sigma' \in \cS(\tau,T)}} \prod_{u\in \cV_{d-1}(\tau)}\psi\left(c_\tau(u),c_{T}(\sigma(u)),c_{T'}(\sigma'(u))\right),
\end{equation} where $\equiv$ is used to denote equality up to some relabeling. Let us compute the right hand term in \eqref{eq:lower_bound_LR_1}. Note that any tree $\tau\in \cY_d$ such that $\tau\equiv \tau^{\star}$ can be determined by a collection 
\begin{equation*}
    \xi(\tau) := \left \lbrace \xi_u(\tau) \in \cS_{c_{\tau^{\star}}(u)}, u \in \cV_{d-1}(\tau^{\star})\right \rbrace,
\end{equation*} giving the reordering of the children of each node of $\tau^{\star}$ at depth $d-1$. Moreover, the number of such permutations that produce this particular tree $\tau$ is precisely given by $\card{\Aut(\tau^{\star})}$. Thus the number of trees in the summation \eqref{eq:lower_bound_LR_1} is precisely
\begin{equation}\label{eq:toto}
\card{\{\tau\in \cY_d: \tau \equiv \tau^{\star}\}}= \frac{\prod_{u\in \cV_{d-1}(\tau^{\star})}c_{\tau^{\star}}(u)!}{\card{\Aut(\tau^{\star})}} \, .
\end{equation} Note that for any tree $\tau \equiv \tau^{\star}$, we can construct
\begin{equation}
\card{\Aut(\tau^{\star})}^2 \times 
\prod_{u\in\cL_{d-1}(\tau^{\star})}\binom{c_u+\Delta_u}{c_u}\binom{c_u+\Delta'_u}{c_u}
\end{equation} pairs of injections $(\sigma,\sigma')\in \cS(\tau,T)\times \cS(\tau,T')$.  Indeed the factor $\binom{c_u+\Delta_u}{c_u}$ (respectively,  $\binom{c_u+\Delta'_u}{c_u}$) denotes the number of subsets of the $c_u+\Delta_u$ children of $u$ in $t$ (respectively, of the $c_u+\Delta'_u$ children of $u$ in $t'$) that we can associate as children of $u$ in the injection $\sigma$ (respectively, $\sigma'$), the order in which they are considered being determined by the permutation $\xi_u$ in $\xi$. We thus have the following lower bound, for any tree $\tau\equiv \tau^{\star}$:

\begin{flalign}\label{eq:tototo}
\sum_{ \substack{\sigma \in \cS(\tau,T) \\ \sigma' \in \cS(\tau,T)}} \prod_{u\in \cV_{d-1}(\tau)} & \psi\left(c_\tau(u),c_{T}(\sigma(u)),c_{T'}(\sigma'(u))\right)  \nonumber \\
& \geq \card{\Aut(\tau^{\star})}^2 \prod_{u \in \cV_{d-1}(\tau^{\star})}\frac{s^{c_u}(1-s)^{\Delta_u+\Delta'_u}}{c_u! e^{-\lambda s} \lambda^{c_u}} \prod_{u \in \cL_{d-1}(\tau^{\star})}\binom{c_u+\Delta_u}{c_u}\binom{c_u+\Delta'_u}{c_u}.&& 
\end{flalign}

Combined, \eqref{eq:toto} and \eqref{eq:tototo} imply \eqref{eq:lemma:lower_bound_LR}.
\end{proof}

We now turn to lower-bounding  the number $|\Aut(\tau^{\star})|$ of automorphisms for $\tau^{\star}\sim \GW^{(\lambda s)}_d$:
\begin{proposition}\label{proposition:auto_GW}
Let $r$ be a sufficiently large constant (in particular, $r>1$). For $\tau^{\star}\sim \GW^{(r)}_d$, let us denote by $w$ the almost sure limit:
\begin{equation}\label{eq:mgle_limit}
w := \lim_{d\to\infty} \frac{1}{r^d}\card{\cL_d(\tau^{\star})}.
\end{equation}
We place ourselves on the event on which $\tau^{\star}$ survives, which occurs with probability $1-\pext(r)>0$, and on which $w>0$. We let
\begin{equation}\label{eq:def_K}
K:=\frac{w r^{d}}{r -1}\cdot
\end{equation}
We then have with high probability the lower bound
\begin{equation}\label{eq:prop_lower_bound_aut}
\log\left(\frac{\card{\Aut(\tau^{\star})}}{\prod_{u \in \cV_{d-1}(\tau^{\star})}e^{-r}r^{c_{\tau^{\star}}(u)}}\right) \geq K(1-o_{\dP}(1))\left[\frac{\log^{3/2} r}{3 \sqrt{r}}+O_r\left(\frac{\log^{5/4} r}{\sqrt{r}}\right)\right].
\end{equation}
\end{proposition}

Proposition \ref{proposition:auto_GW}, proved in Appendix \ref{appendix:proof:proposition:auto_GW}, could be of independent interest. We believe that a little more work could easily show that inequality \eqref{eq:prop_lower_bound_aut} is exponentially tight, i.e. gives the right exponential order for the estimation of the number of automorphism of a Galton-Watson tree. 
We next show that Lemma \ref{lemma:lower_bound_LR} together with Proposition \ref{proposition:auto_GW} yield a sufficient condition for the existence of  one-sided test.

\subsection{A sufficient condition for one-sided tests}
We are now in a position to prove the following

\begin{theorem}\label{MPAlign:theorem:suff_cond_auto}
There exists a constant $r_0$ such that if 
\begin{equation}\label{eq:theorem:suff_cond_auto}
\lambda s > r_0 \quad \mbox{and} \quad 1-s \leq \frac{1}{(3+\eta)}\sqrt{\frac{\log(\lambda s)}{\lambda ^3 s}},
\end{equation} for some $\eta>0$, then one-sided detectability of tree correlation holds.
\end{theorem}

\begin{proof} 
The proof consists in showing that in this regime, $L_d$ goes to $+\infty$ with positive probability under $\dPls_{\infty}$. Throughout, $X_\mu$ will denote a Poisson random variable with parameter $\mu$. In the lower bound \eqref{eq:lemma:lower_bound_LR} of Lemma \ref{lemma:lower_bound_LR}, consider the factor
 
\begin{equation*}
    \prod_{u\in \cV_{d-1}(\tau^{\star})}\frac{s^{c_u}(1-s)^{\Delta_u+\Delta'_u}}{e^{-\lambda s}\lambda ^{c_u}}\prod_{u\in \cL_{d-1}(\tau^{\star})}\binom{c_u+\Delta_u}{c_u}\binom{c_u+\Delta'_u}{c_u}.
\end{equation*}
Placing ourselves on the event on which $\tau^{\star}$ survives, reusing the notations $w$ and $K$ defined in equations \eqref{eq:mgle_limit} and \eqref{eq:def_K}, another appeal to the law of large numbers gives the following equivalents:
\begin{flalign}\label{eq:equiv_A1}
    A & := \log\left(\prod_{u \in \cV_{d-1}} \frac{s^{c_u} (1-s)^{\Delta_u+\Delta'_u}}{e^{-\lambda s}\lambda ^{c_u}} \right) \sim K \left( \lambda s (\log(s/\lambda)+1) + 2\lambda (1-s) \log (1-s)\right)
\end{flalign} and 
\begin{flalign}\label{eq:equiv_B1}
     B & := \log\left(\prod_{u\in\cL_{d-1}(\tau^{\star})}\binom{c_u+\Delta_u}{c_u}\binom{c_u+\Delta'_u}{c_u}\right) \nonumber \\
     & \sim w (\lambda s)^{n-1} \left( 2\dE\left[\log(X_{\lambda}!)\right]-2\dE\left[\log(X_{\lambda (1-s)}!)\right]-2\dE\left[\log(X_{\lambda s}!)\right]\right).
\end{flalign}
Let us introduce the notations $r := \lambda s$, $\alpha := \lambda (1-s)$, such that $\lambda = \alpha + r$ and $s = \frac{r}{\alpha + r}$. We will identify equivalents of exponents of interest as  $\alpha \to 0$ and $r \to \infty$. In this regime, \eqref{eq:equiv_A1} becomes
\begin{flalign*}
    A  & \sim K \left(-2 r\log(1+\alpha/r)+2\alpha\log\left(\frac{\alpha/r}{1+\alpha/r}\right)-r\log r+r\right)\\
    & \sim K \left(-r \log r + r - 2 \alpha \log r + 2\alpha\log\alpha  + O(\alpha) \right)
\end{flalign*} 

We have the classical estimate for large $\mu$:
\begin{equation}\label{eq:moment_factoriel}
\dE \left[\log(X_\mu!) \right]=\mu\log(\mu)-\mu+\frac{1}{2}\log(2\pi e \mu)+O\left(\frac{1}{\mu}\right),
\end{equation}
Using \eqref{eq:moment_factoriel} and noting that in this regime, $\dE\left[ \log (X_\alpha !) \right] = O(\alpha^2)$, \eqref{eq:equiv_B1} becomes
\begin{flalign*}
B & \sim 2 w r^{n-1}\left((r+\alpha)\log(r+\alpha)-r-\alpha+\frac{1}{2}\log(2\pi e(r+\alpha))-r\log(r)+r-\frac{1}{2}\log(2\pi e r)- O(\alpha^2)\right)\\
& \sim 2 w r^{n-1}\left(r\log(1+\alpha/r)+ \alpha\log(1+\alpha/r) + \alpha \log r - \alpha +\frac{1}{2}\log(1+\alpha/r) +O(\alpha)  \right)\\
& \sim 2 w r^{n-1}\left( \alpha\log(r)+O(\alpha)\right).\end{flalign*}

Combined, these approximations give:
\begin{flalign}\label{eq:equiv_A+B}
A+B & \sim K\left(\left(1-\frac{1}{r}\right)\times 2 \alpha\log(r) -r\log r + r - 2\alpha\log(r)+2\alpha\log(\alpha)+O(\alpha)\right) \nonumber\\
&\sim K\left(-r\log r + r + 2\alpha\log(\alpha)+O(\alpha)\right).
\end{flalign}
Combining \eqref{eq:equiv_A+B} with the results of Proposition \ref{proposition:auto_GW} entails

\begin{flalign*}
\log L_d & \geq K \left[r\log r - r +  \frac{\log^{3/2}(r)}{3\sqrt{r}}+O\left(\frac{\log^{5/4} r}{\sqrt{r}}\right) \right] + K \left[-r\log r + r + 2\alpha\log(\alpha)+O(\alpha)\right]\\
& = K \left[2\alpha \log \alpha +  \frac{\log^{3/2}(r)}{3\sqrt{r}}+O\left(\frac{\log^{5/4} r}{\sqrt{r}}\right)  + O(\alpha) \right].
\end{flalign*} 
Then, under assumption \eqref{eq:theorem:suff_cond_auto}, we have $\alpha\leq \frac{1}{3+\eta}\sqrt{\log(r)/r}$ so that, for sufficiently large $r$,
\begin{equation*}
    2\alpha \log \alpha +  \frac{\log^{3/2}(r)}{3\sqrt{r}} > \Omega\left(\frac{\log^{3/2}(r)}{\sqrt{r}}\right).
\end{equation*}
It follows that on the event on which $\tau^{\star}$ survives, which happens with probability $1-\pext(\lambda s)>0$, under condition \eqref{eq:theorem:suff_cond_auto}, $L_d$ goes to $+\infty$ with $d$. Thus one-sided detectability holds. 
\end{proof}

\section{Impossibility of correlation detection: conjectured hard phase for partial graph alignment}
\label{MPAlign:section:hard_phase}
In the present section we establish that, for $\lambda s^2<1$ and sufficiently large $\lambda$, $\KL_\infty<+\infty$ and hence, by Theorem \ref{MPAlign:theorem:main_result_TREES}, one-sided testability fails for our tree correlation problem. Since there exists a range of parameters $(\lambda,s)$ for which partial alignment can be information-theoretically achieved while $\lambda s^2<1$ (it suffices to have $4<\lambda s< s^{-1}$ in view of \cite{Wu2021SettlingTS}) we therefore conclude that the conjectured hard phase for partial graph alignment (see the conjecture at the end of Section \ref{MPAlign:section:intro}) is non empty. 

\subsection{Mutual information formulation}
Note that the Kullback-Leibler divergence $\KL_d$ also coincides with the mutual information between $T_d := p_d(T)$ and $T'_d := p_d(T')$ under $\dPls_{\infty}$. To emphasize this interpretation we rewrite 
\begin{equation*}
    \KL_d = I(T_d;T'_d).
\end{equation*}
Note that under $\dPls_{\infty}$, conditionally on $\tau^{\star}_d := p_d(\tau^{\star})$, $T_d$ and $T'_d$ are mutually independent, a property that we will depict with the dependence diagram 
\begin{equation*}
    T_d \mbox{  ----  } \tau^{\star}_d \mbox{  ----  } T'_d.
\end{equation*} By the data processing inequality, we thus have
\begin{equation*}
    \KL_d=I(T_d;T'_d) \leq  I(\tau^{\star}_d;T_d).
\end{equation*} To establish that $\KL_\infty<\infty$, it therefore suffices to prove that $I(\tau^{\star}_d;T_d)$ is bounded, uniformly in $d$. Write then
\begin{flalign*}
I(\tau^{\star}_d,t_d) &= \dEls_{d} \ln\left(\frac{\dPls_{d}(\tau^{\star}_d,T_d)}{\dPls_{d}(\tau^{\star}_d)\dPls_{d}(T_d)}\right) \\
& \leq \dEls_{d}\left[ \frac{\dPls_{d}(\tau^{\star}_d,T_d)}{\dPls_{d}(\tau^{\star}_d)\dPls_{d}(T_d)}-1\right] \leq \dEls_{d}\left[\frac{\dPls_{d}(\tau^{\star}_d,T_d)}{\dPls_{d}(\tau^{\star}_d)\dPls_{d}(T_d)}\right].
\end{flalign*} We have established the bound
\begin{equation}\label{eq:upper_bound_mutual_info}
I(\tau^{\star}_d,T_d) \leq V_d: = \dEls_{d} \left[\frac{\dPls_{d}(\tau^{\star}_d|T_d)}{\dPls_{d}(\tau^{\star}_d)}\right] \, .
\end{equation}

\subsection{Bounding the mutual information}
Let us denote by $c$ the degree of the root node in $\tau^{\star}$ and $c+\Delta$ the degree of the root node in $T$. Let us further write 
\begin{equation*}
    \tau^{\star}= (\tau^{\star}_1,\ldots,\tau^{\star}_{c}), \quad T=r(A(\tau^{\star}_1),\ldots,A(\tau^{\star}_{c}), \theta_1,\ldots,\theta_{\Delta})=(T_1,\ldots,T_{c+\Delta}),
\end{equation*}
where the $A(\tau^{\star}_u)$ are $(\lambda,s)-$augmentations, the $\theta_u$ are $\GWl_{d-1}$ trees, and $r$ is a uniform relabeling. Observe that
\begin{flalign*}
    \dPls_d(\tau^{\star}|T) & = \frac{\GW^{(\lambda s)}_{d}(\tau^{\star})}{\GWl_{d}(T)} e^{-\lambda(1-s)} \frac{(\lambda (1-s))^{\Delta}}{\Delta !} \\
    & \quad \quad \quad \quad \quad \times \sum_{\sigma \in \cS(c,c+\Delta)} \frac{\Delta!}{(c + \Delta)!} \prod_{u \in [c]} \dPls_{d-1}(T_{\sigma(u)} | \tau^{\star}_u) \prod_{u = c+1}^{c+\Delta} \GWl_{d-1}(T_{\sigma(u)})   \\
    & = \frac{ e^{-\lambda s} (\lambda s)^{c}/c! }{ e^{-\lambda} \lambda^{c+\Delta}/(c+\Delta)! } e^{-\lambda(1-s)} \frac{(\lambda (1-s))^{\Delta}}{\Delta !} \sum_{\sigma \in \cS(c,c+\Delta)} \frac{\Delta!}{(c + \Delta)!} \prod_{u \in [c]} \dPls_{d-1}(\tau^{\star}_u |T_{\sigma(u)}) \\
    & = \frac{s^{c} (1-s)^{\Delta}}{c !} \sum_{\sigma \in \cS(c,c+\Delta)} \prod_{u \in [c]} \dPls_{d-1}(\tau^{\star}_u |T_{\sigma(u)}),
\end{flalign*}
so that
\begin{equation*}
\frac{\dPls_d(\tau^{\star}|T)}{\dPls_d(\tau^{\star})}=\frac{s^{c} (1-s)^{\Delta}}{c ! \pi_{\lambda s}(c)}
\sum_{\sigma \in \cS(c,c+\Delta)} \prod_{u \in [c]} \frac{\dPls_{d-1}(\tau^{\star}_u|T_{\sigma(u)})}{\dPls_{d-1}(\tau^{\star}_u)}.
\end{equation*}
Taking expectation entails the following formula for $V_d$ defined in equation \eqref{eq:upper_bound_mutual_info}:
\begin{equation}\label{eq_V_d_recursive}
V_d=\sum_{c \geq 0}\sum_{\Delta \geq 0} \pi_{\lambda(1-s)}(\Delta)\frac{s^c (1-s)^\Delta}{c!}\sum_{\sigma\in \cS(c,c+\Delta)}\dEls_{d-1}\left[\prod_{i=1}^c \frac{\dPls_{d-1}(\tau^{\star}_i|T_{\sigma(i)})}{\dPls_{d-1}(\tau^{\star}_i)} \bigg| c  , \Delta \right].
\end{equation}
To evaluate the previous expression, we need to introduce the following notion of cycles.

\subsubsection{Open paths, closed cycles} For two integers $c,\Delta \geq 0$ and an injective mapping $\sigma \in \cS(c,c+\Delta)$, a sequence $(i_1,\ldots,i_\ell)$ of elements of $[c]$ is
\begin{itemize}
    \item an \emph{open path of $\sigma$} if
    \begin{equation*}
        i_1 \notin \sigma([c]), \quad \forall k=1,\ldots, \ell-1, \; \sigma(i_{k})=i_{k+1}, \quad \mbox{ and } \sigma(i_\ell)\notin [c].
    \end{equation*}
    
    \item a \emph{closed cycle of $\sigma$} if
    \begin{equation*}
        \forall k=1,\ldots, \ell-1, \; \sigma(i_{k})=i_{k+1} \quad \mbox{ and }\sigma(i_\ell)=i_1.
    \end{equation*}
\end{itemize}

It is an easy fact to check that each injective mapping $\sigma \in \cS(c,c+\Delta)$ can be factorized in disjoint open paths and closed cycles. Since each term $i$ in the product in \eqref{eq_V_d_recursive} only depends on the other terms $j$ in its own open path (resp. closed cycle), the expectation term in \eqref{eq_V_d_recursive} factorizes according to the path/cycle decomposition of $\sigma$. 

\begin{figure}[H]
	\centering
		\begin{tikzpicture}[scale=0.7,line width=0.4mm]
		\node[draw,circle,thick,scale=0.8] (A1) at (0,8.5) {$1$};
		\node[draw,circle,thick,scale=0.8] (B1) at (0,7.5) {$2$};
		\node[draw,circle,thick,scale=0.8] (C1) at (0,6.5) {$3$};
		\node[draw,circle,thick,scale=0.8] (D1) at (0,5.5) {$4$};
		\node[draw,circle,thick,scale=0.8] (E1) at (0,4.5) {$5$};
		\node[draw,circle,thick,scale=0.8] (F1) at (0,3.5) {$6$};

		\node[draw,circle,thick,scale=0.8] (A2) at (4,10) {$1$};
		\node[draw,circle,thick,scale=0.8] (B2) at (4,9) {$2$};
		\node[draw,circle,thick,scale=0.8] (C2) at (4,8) {$3$};
		\node[draw,circle,thick,scale=0.8] (D2) at (4,7) {$4$};
		\node[draw,circle,thick,scale=0.8] (E2) at (4,6) {$5$};
		\node[draw,circle,thick,scale=0.8] (F2) at (4,5) {$6$};
		\node[draw,circle,thick,scale=0.8] (G2) at (4,4) {$7$};
		\node[draw,circle,thick,scale=0.8] (H2) at (4,3) {$8$};
		\node[draw,circle,thick,scale=0.8] (I2) at (4,2) {$9$};
		
		\draw[blue,dashed,line width=0.8pt] (A1) to[bend right=0] (A2);
		\draw[blue,dashed,line width=0.8pt] (B1) to[bend right=0] (B2);
		\draw[blue,dashed,line width=0.8pt] (C1) to[bend right=0] (C2);
		\draw[blue,dashed,line width=0.8pt] (D1) to[bend right=0] (D2);
		\draw[blue,dashed,line width=0.8pt] (E1) to[bend right=0] (E2);
		\draw[blue,dashed,line width=0.8pt] (F1) to[bend right=0] (F2);

		\draw[red,line width=2pt,->] (A1) to[bend right=0] (D2);
		\draw[red,line width=2pt,->] (B1) to[bend right=0] (F2);
		\draw[red,line width=2pt,->] (C1) to[bend right=0] (B2);
		\draw[red,line width=2pt,->] (D1) to[bend right=0] (E2);
		\draw[red,line width=2pt,->] (E1) to[bend right=0] (A2);
		\draw[red,line width=2pt,->] (F1) to[bend right=0] (H2);
		\end{tikzpicture}
	\caption{Representation of $\sigma \in \cS(c,c+\Delta)$ with $c=6, \Delta=3$, and $\sigma(1)=4$, $\sigma(2)=6$, $\sigma(3)=2$, $\sigma(4)=5$, $\sigma(5)=1$, $\sigma(6) = 8$. In this example, $(1,4,5)$ (resp. $(3,2,6)$) is an open path (resp. closed path) of $\sigma$.}
	\label{fig:three graphs}
\end{figure}
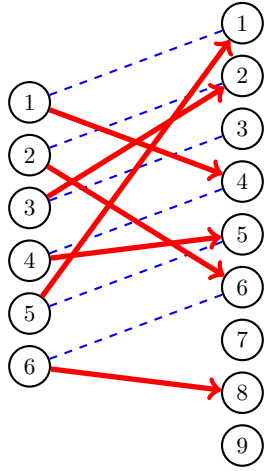
 
First consider an open path $O_\ell$ of $\sigma$ of length $\ell$, assumed without loss of generality to be given by $(1,\ldots,\ell)$, so that $\sigma(1)=2,\ldots,\sigma(\ell-1)=\ell$, and $\sigma(\ell)=c+1$. The expectation of the corresponding factor reads:
\begin{equation*} \label{eq:expectation_Ol}
    \dEls_{d-1} \left[ \prod_{i \in O_\ell} \frac{\dPls_{d-1}(\tau^{\star}_i|T_{\sigma(i)})}{\dPls_{d-1}(\tau^{\star}_i)} \right] =  \dEls_{d-1} \left[ \prod_{k=1}^{\ell-1}\frac{\dPls_{d-1}(\tau^{\star}_k|A(\tau^{\star}_{k+1}))}{\dPls_{d-1}(\tau^{\star}_k)}\times \frac{\dPls_{d-1}(\tau^{\star}_\ell|\theta_1)}{\dPls_{d-1}(\tau^{\star}_\ell)} \right].
\end{equation*}
Now integrated over $\theta_1$, $\dPls_{d-1}(\tau^{\star}_\ell|\theta_1)$ evaluates to $\dPls_{d-1}(\tau^{\star}_\ell)$ and the last factor disappears. Integrating then successively with respect to $A(\tau^{\star}_k)$, $k=\ell,\ell-1,\ldots,2$, we obtain that the factors corresponding to open cycles evaluate to $1$.

Consider next a closed cycle $C_\ell$ of $\sigma$ of length $\ell$. Assuming without loss of generality that $T_i=A(\tau^{\star}_i)$, the expectation reads
\begin{equation*} 
    \dEls_{d-1} \left[ \prod_{i \in C_\ell} \frac{\dPls_{d-1}(\tau^{\star}_i|T_{\sigma(i)})}{\dPls_{d-1}(\tau^{\star}_i)} \right] =  \dEls_{d-1} \left[ \prod_{k=1}^{\ell}\frac{\dPls_{d-1}(\tau^{\star}_k|A(\tau^{\star}_{(k+1) \Mod \ell}))}{\dPls_{d-1}(\tau^{\star}_k)}\right].
\end{equation*}
This reads, using for $t,\tau\in \cY_{d-1}$ the notations $p_{d-1}(t):=\GWl_{d-1}(t)$, $q_{d-1}(\tau|t):=\dPls_{d-1}(\tau|t)$, $r_{d-1}(\tau):=\GW^{(\lambda s)}_{d-1}(\tau)$:
\begin{equation}\label{eq:expectation_Cl}
    \sum_{\tau_1,t_1, \ldots, \tau_\ell, t_\ell \in \cY_{d-1}}\prod_{i\in [\ell]}p_{d-1}(t_i)q_{d-1}(\tau_i|t_i) \times \frac{q_{d-1}(\tau_i|t_{(i+1) \Mod \ell})}{r_{d-1}(\tau_i)}.
\end{equation}
Introduce the operator $\Psi_{d-1}$, indexed by trees in $\cY_{d-1}$: 
\begin{equation}\label{eq:def_op_M}
    \Psi_{d-1}(\tau_1,\tau_2):=\sum_{t\in \cY_{d-1}}p_{d-1}(t)\frac{q_{d-1}(\tau_1|t)q_{d-1}(\tau_2|t)}{\sqrt{r_{d-1}(\tau_2)r_{d-1}(\tau_2)}}.
\end{equation}
$M$ is symmetric and semi-definite positive, hence the operator is diagonalizable and its spectrum lies in $\dR_+$. Note that the expectation in \eqref{eq:expectation_Cl} coincides with the trace of matrix $\Psi_{d-1}^\ell$. It follows that \footnote{To make this argument fully rigorous we can consider truncated summations so that we are dealing with finite dimensional matrices, for which the trace inequality to follow clearly holds, and then use monotone convergence to obtain the desired inequality as written.}
\begin{equation*} 
    \dEls_{d-1} \left[ \prod_{i \in C_\ell} \frac{\dPls_{d-1}(\tau^{\star}_i|T_{\sigma(i)})}{\dPls_{d-1}(\tau^{\star}_i)} \right] =  
    \Tr(\Psi_{d-1}^\ell)\leq \Tr(\Psi_{d-1})^\ell= V_{d-1}.
\end{equation*}

We now have the ingredients in place to prove the following
\begin{lemma}\label{lemma:bounding_v_d}
The quantity $V_d$ verifies
\begin{equation}\label{eq:lemma:bounding_v_d}
V_d\leq f(V_{d-1}),
\end{equation}
where 
\begin{equation}\label{eq:f_bound_v_d}
f(x)=\frac{1}{1- s x}\exp\left(\frac{\kappa(1-s)(x-1)}{1-sx} \right)
\end{equation} with $\kappa := \lambda s^2$.
\end{lemma}
\begin{proof}
For given $c,\Delta \geq 0$ and an injection $\sigma \in \cS(c,c+\Delta)$, let $F(\sigma)$ denote the number of elements $i\in [c]$ that belong to  closed cycles of $\sigma$. From the previous evaluations \eqref{eq_V_d_recursive} -- \eqref{eq:def_op_M} we already have obtained the bound
\begin{equation*}
    V_d\leq \sum_{c,\Delta \geq 0} \pi_{\lambda(1-s)}(\Delta)\frac{s^c (1-s)^\Delta}{c!}\sum_{\sigma\in \cS(c,c+\Delta)} V_{d-1}^{F(\sigma)}.
\end{equation*}
To upper-bound this quantity, we use the facts that $V_{d-1} \geq 1$ and that $F(\sigma) \leq \card{[c] \cap \sigma([c])}$. Then, for any $0 \leq k \leq c$, there are $\binom{c}{k}\binom{\Delta}{c-k}$ ways to chose the set $\sigma([c])$ such that $\card{[c] \cap \sigma([c])} = k$, and $c!$ distinct injections $\sigma$ with the same set $\sigma([c])$. Hence $V_d \leq f(V_{d-1})$ with
\begin{flalign*}
f(x) & := \sum_{c, \Delta \geq 0} e^{-\lambda(1-s)}\frac{s^c (\lambda(1-s)^2)^\Delta}{\Delta!}\sum_{k=0}^{c}\binom{c}{k}\binom{\Delta}{c-k} x^k\\
& = e^{-\lambda(1-s)} \sum_{k, \Delta \geq 0} x^k \frac{(\lambda(1-s)^2)^\Delta}{\Delta!}\sum_{c \geq k}  \binom{c}{k}\binom{\Delta}{c-k} s^c \\
& = e^{-\lambda(1-s)} \sum_{k, \Delta \geq 0} x^k \frac{(\lambda(1-s)^2)^\Delta}{\Delta!} s^k \sum_{c = 0}^{\Delta} \binom{c+k}{k}\binom{\Delta}{c} s^c \\
& = e^{-\lambda(1-s)} \sum_{k, c \geq 0} \frac{1}{c!} \binom{c+k}{k} (sx)^k s^c (\lambda(1-s)^2)^{c} \sum_{\Delta \geq c} \frac{(\lambda(1-s)^2)^{\Delta-c}}{(\Delta-c)!}\\
& = e^{-\lambda s (1-s)} \sum_{c \geq 0} \frac{(\lambda s (1-s)^2)^{c} }{c!} \sum_{k \geq 0} \binom{c+k}{k} (sx)^k\\
& = e^{-\lambda s (1-s)} \frac{1}{1-sx} \sum_{c \geq 0} \frac{1}{c!}\left(\frac{\lambda s (1-s)^2}{1-sx}\right)^{c} = \frac{1}{1-sx} \exp\left( \frac{\lambda s^2(1-s)(x-1)}{1-sx}\right).
\end{flalign*}
\end{proof}

We are now in a position to prove the following 
\begin{theorem}\label{MPAlign:theorem:suff_hard_phase}
Assume $\kappa=\lambda s^2$ is fixed such that $\kappa<1$. Then for $\lambda$ sufficiently large, it holds that
\begin{equation}
\limsup_{d\to\infty}V_d<+\infty,
\end{equation}
so that one-sided testability fails. 
\end{theorem}
\begin{proof}
Let $\kappa<1$ be fixed, together with $\eps \in (0,4\kappa)$ such that $\kappa+\eps<1$. Let $\gamma>0$ be an arbitrary constant chosen such that 
\begin{equation*}
    \gamma > \frac{1}{1-\kappa-\eps}.
\end{equation*}
We shall consider $s>0$ sufficiently small, or equivalently $\lambda$ large enough, in particular such that $\gamma s<1$. Let  $y\in[0,\gamma s]$. Note that
\begin{equation*}
    \exp\left(\kappa\frac{y(1-s)}{1-s(y+1)}\right) \leq \exp\left(\kappa y /(1-2s)\right).
\end{equation*} Then, assuming $\frac{1}{1-2s} \leq 1+\eps/(4 \kappa)$ as well as $2 e^2 \kappa^2 \gamma s/\eps \leq 1$, we get

\begin{equation}\label{eq:proof_th4_1}
\exp\left(\kappa y /(1-2s)\right) \leq\exp\left(\kappa y + \eps y/4 \right) \leq 1+ (\kappa+\eps/2)y.
\end{equation}
Note also that, $1/(1-t) \leq 1+t+3t^2$ for $t \in (0,2/3)$. Assuming $s(y+1)\leq 2s < 2/3$, and using $y \leq \gamma s \leq 1$, we get
\begin{equation}\label{eq:proof_th4_2}
    \frac{1}{1-s(y+1)} \leq 1+s(y+1)+3[s(y+1)]^2 \leq 1+s + C s^2,
\end{equation} where $C := \gamma + 12$. Together, these last two bounds \eqref{eq:proof_th4_1} and \eqref{eq:proof_th4_2} entail, for any $y \in [0,\gamma s]$:

\begin{equation*}
    f(1+y)-1 \leq (1+(\kappa +\eps/2)y)(1+s+ C s^2)-1 \leq s+ C s^2 +(\kappa+\eps)y,
\end{equation*}
where we assumed $s$ sufficiently small that $(\kappa+\eps/2)(1+s+C s^2) \leq \kappa+\eps$. 
Note now that, provided
\begin{equation*}
    1+(\gamma+12)s + (\kappa+\eps)\gamma \leq \gamma,
\end{equation*}
it holds that 
\begin{equation}\label{eq:bound_iter_f}
f(1+y)-1\in[0,\gamma s].
\end{equation}
Note that this condition can be enforced, for any choice of $\gamma$ such that $\gamma > \frac{1}{1-\kappa-\eps}$ taking $s$ sufficiently small.

By induction on $d$, monotonicity of $f$ (which is easily obtained from the series expansion of $f$), and the initialization $V_0=1$, it follows from \eqref{eq:bound_iter_f} that for sufficiently small $s$ one has:
\begin{equation*}
    V_d -1 \leq (s + C s^2)\sum_{i=0}^{d-1}(\kappa+\eps)^i.
\end{equation*}
Since the right-hand side is uniformly bounded in $d$, the result follows. 
\end{proof}

\section{Consequences for polynomial time partial graph alignment}\label{MPAlign:section:graph_matching}
We now apply the previous results of Sections \ref{MPAlign:section:LR} -- \ref{MPAlign:section:autos_GW} to one-sided partial graph alignment. We will now describe our polynomial-time algorithm and its theoretical guarantees when one-sided detectablity holds in Theorem \ref{MPAlign:theorem:main_result_TREES}  -- in particular under condition \eqref{condition:lambda_s_KL} of Theorem \ref{MPAlign:theorem:suff_cond_KL} or condition \eqref{eq:theorem:suff_cond_auto} of Theorem \ref{MPAlign:theorem:suff_cond_auto}.

\subsection{Intuition, algorithm description}\label{subsection:intuition_algorithm_desription}
In all this part we assume that $(\lambda,s)$ satisfy one of the conditions in Theorem \ref{MPAlign:theorem:main_result_TREES}.

\subsubsection{Extending the tree correlation detection problem} 
Let $(G, H)$ be a pairs of relabeled $\G(n,\lambda/n,s)$ graphs, with underlying alignment $\pi^{\star}$. As done in Chapter \ref{chapter:NTMA}, in order to distinguish matched pairs of nodes $(u,u')$, we consider their neighborhoods $\cN_{d,G}(u)$ and $\cN_{d,H}(u')$ at a given depth $d$: these neighborhoods are close to Galton-Watson trees. In the case where the two vertices are actual matches, i.e. $u' = \pi^{\star}(u)$, we are exactly in the setting of our tree correlation detection problem under $\dPls_{d}$: Point $(v)$ of in Theorem \ref{MPAlign:theorem:main_result_TREES} shows that there exists a threshold $\beta_d$ such that with probability at least $1-\pext(\lambda s)>0$,
$$L_d(u,u') := L_d \left(\cN_{d,G}(u),\cN_{d,H}(u')\right) > \beta_d,$$ when $d \to \infty$. 
Point $(v)$ of Theorem \ref{MPAlign:theorem:main_result_TREES} shows that this threshold $\beta_d$ can be e.g. taken to be $\exp(n^\gamma)$ for some $\gamma \in (0,c \log (\lambda s))$.

At the same time, when nodes $u'$ and $\pi^{\star}(u)$ are distinct and sufficiently far away, we can argue that we are also -- with high probability -- in the setting of the tree correlation detection problem under $\dPl_d$: since $\dEl_{d}\left[L_d\right]=1$, Markov's inequality shows that with high probability when $d \to \infty$,
$$L_d(u,u') \leq \beta_d.$$

\subsubsection{Computation of the likelihood ratios}
As mentioned in Remark \ref{remark:util_rec_algo}, Formula \eqref{eq:MPAlign:LR_rec_2} enables to compute such likelihood ratios efficiently on a graph, giving the exact expression for a \emph{message passing} procedure, assuming that all neighborhoods are locally tree-like at depth $d$. Let us first define \emph{oriented likelihood ratios}: for any $u,v \in V(G)$ and $u',v' \in V(H)$, we write $L_d(u \leftarrow v,u' \leftarrow v')$ for the likelihood ratio at depth $d$ of two trees, the first one (resp. second one) being rooted at $u$ in $G$ (resp. $u'$ in $H$) where the edge $\set{u,v}$ (resp. $\set{u',v'}$), if initially present, has been deleted. In view of \eqref{eq:MPAlign:LR_rec_2} these oriented likelihood ratios satisfy the following recursion:

\begin{multline}\label{eq:rec_oriented_LR}
    L_d(u \leftarrow v,u' \leftarrow v') = \\ \sum_{k=0}^{d_u \wedge d'_{u'} - 1} \psi \left(k, d_u - 1,d'_{u'} - 1\right) \sum_{\substack{\sigma \in \cS\left([k],\cN_{G}(u) \setminus \set{v}\right) \\ \sigma' \in \cS\left([k],\cN_{H}(u') \setminus \set{v'} \right)}} \prod_{\ell=1}^{k} L_{d-1}(\sigma(\ell) \leftarrow u,\sigma'(\ell) \leftarrow u') \, ,
\end{multline} where $d_u := d_{G}(u)$ and $d'_{u'} := d_{H}(u')$. The likelihood ratio at depth $d$ between $u$ and $u'$ is then obtained by computing
\begin{equation}\label{eq:total_oriented_LR}
    L_d(u,u') = \sum_{k=0}^{d_u \wedge d'_{u'}} \psi \left(k, d_u,d'_{u'}\right) \sum_{\substack{\sigma \in \cS\left([k],\cN_{G}(u) \right) \\ \sigma' \in \cS\left([k],\cN_{H}(u') \right)}} \prod_{\ell=1}^{k} L_{d-1}(\sigma(\ell) \leftarrow u,\sigma'(\ell) \leftarrow u') \, .
\end{equation}

A natural idea is then to compute for each pair $(u,u')$ the likelihood ratio $L_d(u,u')$ with $d$ large enough (typically scaled in $\Theta(\log n)$ where $n$ is the number of vertices in $G$ and $H$) and to compare it to $\beta_d$ to decide whether $u$ in $G$ is matched to $u'$ in $H$. 

\subsubsection{A refined \emph{dangling trees trick}} 
However, as previously noted in Chapter \ref{chapter:NTMA}, without additional constraint, this strategy produces many falsely positive matches, tending e.g. to match $u$ with $u'$ if there exists $v$ such that $\set{u,v}$ is an edge of $G$ and $\set{u',\pi^{\star}(v)}$ is an edge of $H$, making the errors  increase and the performance collapse. 

To fix this issue, we use the \emph{dangling trees trick}, already introduced in \cite{Ganassali20a}, and improved here by considering three rather than two dangling trees: instead of just looking at their neighborhoods, we look for the downstream trees from distinct neighbors of $u$ in $G$ and of $u'$ in $H$. The trick is now to match $u$ with $u'$ if and only if there exists three distinct neighbors $v,w,x$ of $u$ in $G$ (resp. $v',w',x'$ of $u'$ in $H$) such that all three of the likelihood ratios $L_{d-1}(v \leftarrow u, v' \leftarrow u')$, $L_{d-1}(w \leftarrow u, w' \leftarrow u')$  and $L_{d-1}(x \leftarrow u, x' \leftarrow u')$ are larger than $\beta$. The proof of Theorem \ref{MPAlign:theorem:no_mismatchs} explains how this trick avoids false positives and why three dangling trees is a good choice.

\subsubsection{Algorithm description}
Our algorithm is as follows:

\begin{algorithm}[H]
\caption{\label{MPAlign:algo_GA} \alg{MPAlign}: Message-passing algorithm for sparse graph alignment}
\SetAlgoLined

\textbf{Input:} Two graphs $G$ and $H$ of size $n$, average degree $\lambda$, depth $d$, threshold parameter $\beta$

\textbf{Output:} A set of pairs $\cM \subset V(G) \times V(H)$.

$\cM \gets \varnothing$

Compute $L_d(u \leftarrow v,u' \leftarrow v')$ for all $\set{u,v} \in E$ and $\set{u',v'} \in E'$ with \eqref{eq:rec_oriented_LR}

\For{$(u,u') \in V(G) \times V(H)$}{
	\If{$\cN_{G}(u,d)$ and $\cN_{H}(u',d)$ contain no cycle, and $\exists \set{v,w,x}  \subset \cN_{G}(u), \exists \set{v',w',x'}  \subset \cN_{G}(u')$ such that $L_{d-1}(v \leftarrow u, v' \leftarrow u') > \beta$, $L_{d-1}(w \leftarrow u, w' \leftarrow u') > \beta$  and $L_{d-1}(x \leftarrow u, x' \leftarrow u')> \beta$}
	{	
	$\cM \gets \cM \cup \left\lbrace (u,u') \right\rbrace $
	}
}
\textbf{return} $\cM$
\end{algorithm}

\begin{remark}
To update the matrix of all likelihood ratios with \eqref{eq:rec_oriented_LR}, we update a matrix of size $O(n^2)$, each entry of which can be computed in time $O\left((d_{\max}!)^2\right)$ -- where $d_{\max}$ is the maximum degree in $G$ and $H$. Under the correlated \ER model,  $d_{\max} = O\left( \frac{\log n}{\log \log n} \right)$ \cite{Bollobas2001}, so that $d_{\max}!$ is polynomial in $n$.  Each iteration is thus polynomial in $n$ and since $d$ is taken order $\log(n)$, \alg{MPAlign} (Algorithm \ref{MPAlign:algo_GA}) runs in polynomial time.
\end{remark}

We now state two results, of the same flavour as Theorems \ref{NTMA:thm:lot_of_matchs} and \ref{NTMA:thm:no_mismatchs} in Chapter \ref{chapter:NTMA} for \alg{NTMA}, which will readily imply Theorem \ref{MPAlign:theorem:main_result_GRAPHS}.
\begin{theorem} \label{MPAlign:theorem:good_matches}
Let $(G,H)$ be drawn under the planted model with correlated $\G(n,\lambda/n,s)$ graphs such that any of the equivalent conditions of Theorem 1 holds.
Let $d = \lfloor c \log n \rfloor$ with $c \log \left(\lambda\left(2-s\right)\right)<1/2$. Let $\cM$ be the output of Alg. \ref{MPAlign:algo_GA}, taking $\beta  = \exp(n^\gamma)$ for some $\gamma \in (0,c \log (\lambda s))$. Then with high probability
\begin{equation}
\label{MPAlign:eq:theorem:good_matches}
\frac{1}{n} \sum_{u=1}^{n} \one_{\lbrace (u,\pi^{\star}(u)) \in\cM \rbrace} \geq \Omega(1).
\end{equation} 
In other words, a non vanishing fraction of nodes is correctly recovered by Algorithm \ref{MPAlign:algo_GA}.
\end{theorem}

\begin{theorem}\label{MPAlign:theorem:no_mismatchs}
Let $(G,G') \sim \G(n,\lambda/n,s)$ be two $s-$correlated \ER graphs. Assume that $d = \lfloor c \log n \rfloor$ with $c \log \lambda<1/4$. Let $\cM$ be the output of Alg. \ref{MPAlign:algo_GA}, taking $\beta  = \exp(n^\gamma)$ for some $\gamma \in (0,c \log (\lambda s))$. Then with high probability
\begin{equation}
\label{MPAlign:eq:theorem:no_mismatchs}
\mathrm{err}(n):=\frac{1}{n}\sum_{u=1}^{n} \one_{\lbrace \exists u'  \neq \pi^{\star}(u), \; (u,u') \in\cM \rbrace}=o(1),
\end{equation} 
i.e. only a vanishing fraction of nodes are incorrectly matched by Algorithm \ref{MPAlign:algo_GA}.
\end{theorem}

\begin{remark}
The set $\cM$ returned by Algorithm \ref{MPAlign:algo_GA} is not necessarily an injective mapping. Let $\cM'$ be obtained by removing all pairs $(u,u')$ of $\cM$ such that $i$ or $u$ appears at least twice. Theorems \ref{MPAlign:theorem:good_matches} and \ref{MPAlign:theorem:no_mismatchs} guarantee that $\cM'$ still contains a non-vanishing number of correct matches and a vanishing number of incorrect matches, hence one-sided partial alignment holds. Theorem \ref{MPAlign:theorem:main_result_GRAPHS} easily follows, the proposed local algorithm achieving one-sided partial graph alignment.

A slight adaptation of \alg{MPAlign} (Alg. \ref{MPAlign:algo_GA}), \alg{MPAlign2} (Alg. \ref{MPAlign:algo_MPAlign2}), can be found in Appendix \ref{MPAlign:appendix:numerical}, where some results are also reported. These confirm our theory, as the algorithm returns many good matches and few mismatches. A similar algorithm has been recently studied in \cite{piccioli2021aligning}.
\end{remark}

\subsection{Proof strategy}\label{subsection:proof_strategy_GA}
We start by recalling Lemmas that precise the link between sparse graph alignment and correlation detection in trees, as explained in Section \ref{subsection:intuition_algorithm_desription}. These Lemmas are directly taken from Chapter \ref{chapter:NTMA} (to which we refer for the proofs, see Lemmas \ref{NTMA:lemma:control_S}, \ref{NTMA:lemma:cycles_ER}, \ref{NTMA:lemma:indep_neighborhoods} and \ref{lemma:NTMA:coupling_GW}) and are instrumental in the proofs of Theorems \ref{MPAlign:theorem:good_matches} and \ref{MPAlign:theorem:no_mismatchs}.

\begin{lemma}[Control of the sizes of the neighborhoods]
	\label{MPAlign:lemma:control_S}
	Let $G \sim \G(n,\lambda/n)$, $d = \lfloor c \log n \rfloor$ with $c \log \lambda <1$. For all $\gamma>0$, there is a constant $C=C(\gamma)>0$ such that with probability $1-O\left(n^{-\gamma}\right)$, for all $u \in [n]$, $t \in [d]$:
	\begin{equation}
	\label{eq:lemma:control_S}
	\left| \cS_{G}(u,t) \right| \leq C (\log n) \lambda^t.
	\end{equation}
\end{lemma}

\begin{lemma}[Cycles in the neighborhoods in an \ER graph]
\label{MPAlign:lemma:cycles_ER}
Let $G \sim \G(n,\lambda/n)$, $d = \lfloor c \log n \rfloor$ with $c \log \lambda <1/2$. Then there exists $\eps>0$ such that for any vertex $u \in [n]$, one has
\begin{equation}
\label{eq:lemma:cycles_ER}
\mathbb{P}\left(\cN_{G,d}(u) \mbox{ contains a cycle}\right) = O\left( n^{-\eps}\right).
\end{equation}
\end{lemma}

\begin{lemma}[Two neighborhoods are typically independent]
\label{MPAlign:lemma:indep_deighborhoods}
Let $G \sim \G(n,\lambda/n)$ with $\lambda >1$, $d = \lfloor c \log n \rfloor$ with $c \log \lambda < 1/2 $. Then there exists $\eps>0$ such that for any fixed nodes $u \neq v$ of $G$, the total variation distance between the joint distribution of the neighborhoods $\mathcal{L} \left(\left(\cS_{G}(u,t),\cS_{G}(v,t)\right)_{t \leq d}\right)$ and the product distribution $\mathcal{L} \left(\left(\cS_{G}(u,t)\right)_{t \leq d}\right) \otimes \mathcal{L} \left(\left(\cS_{G}(v,t)\right)_{t \leq d}\right)$ tends to $0$ as $O\left(n^{-\eps}\right)$ when $n \to \infty$.
\end{lemma}

\begin{lemma}[Coupling neighborhoods with Galton-Watson trees]
\label{MPAlign:lemma:coupling_GW} We have the following couplings:
\begin{itemize}
    \item[$(i)$] Let $G \sim \G(n,\lambda/n)$, $d = \lfloor c \log n \rfloor$ with $c \log \lambda<1/2$. Then there exists $\eps>0$ such that for any fixed node $u$ of $G$, the variation distance between the distribution of $\cN_{G,d}(u)$ and the distribution $\GWl_{d}$ tends to 0 as $O\left( n^{-\eps}\right)$ when $n \to \infty$.
    \item[$(ii)$] For $(G,H)$ two correlated $\G(n,\lambda/n,s)$ graph with planted alignment $\pi^{\star}$, $d = \lfloor c \log n \rfloor$ with $c \log (\lambda s)<1/2$ and $c \log (\lambda (1-s))<1/2$, there exists $\eps>0$ such that for any fixed node $u$ of $G$, the variation distance between the distribution of $(\cN_{G,d}(u),\cN_{H,d}(\pi^{\star}(u)))$ and the distribution $\dPls_{d}$ (as defined in Section \ref{MPAlign:subsection:model_random_trees}) tends to 0 as $O\left( n^{-\eps}\right)$ when $n \to \infty$.
\end{itemize}
\end{lemma}

\subsubsection{Proof of Theorems \ref{MPAlign:theorem:good_matches} and \ref{MPAlign:theorem:no_mismatchs}}
\begin{proof}[Proof of Theorem \ref{MPAlign:theorem:good_matches}] 
First, since $c \log \left(\lambda\left(2-s\right)\right)<1/2$, we also have $c \log \left(\lambda\left(1-s\right)\right)<1/2$ and $c \log \left(\lambda s\right)<1/2$. For $i \in [n]$, point $(ii)$ of Lemma \ref{MPAlign:lemma:coupling_GW} implies that the two neighborhoods $\cN_{G,d}(u)$ and $\cN_{H,d}(\pi^{\star}(u))$ can be coupled with trees drawn under $\dPls_{d}$ as defined in Section \ref{MPAlign:subsection:model_random_trees} with probability $\geq 1-O(n^{- \eps})$. 

Under this coupling, there is a probability $\alpha_3 >0$ that the root in the intersection tree has at least three children, and since we work under the conditions of Theorem \ref{MPAlign:theorem:main_result_TREES} point $(v)$ implies that the three likelihood ratios are greater than $\beta$ with positive probability $(1-\pext(\lambda s))^3>0$. Hence, the probability of $M_u := \left\lbrace (u,\pi^{\star}(u)) \in\cM \right\rbrace$ is at least $(1-o(1)) \alpha_3 (1-\pext(\lambda s))^3  =: \alpha >0$. 

Let $G_{\cup}$ be the true union graph, that is $G_{\cup} := G^{\pi^{\star}} \cup H$ where 
$G^{\pi^{\star}}$ is the relabeling of $G$ according to permutation $\pi^{\star}$. We have $G_{\cup} \sim \G(n,\lambda(2-s)/n)$. For $u \neq v \in [n]$, define $I_{u,v}$ the event on which the two neighborhoods of $u$ and $v$ in $G_{\cup}$ coincide with their independent couplings up to depth $d$. Since $c \log \left(\lambda\left(2-s\right)\right)<1/2$, by Lemma \ref{MPAlign:lemma:indep_deighborhoods}, $\mathbb{P}(I_{u,v})=1-o(1)$. Then for $0<\eps<\alpha$, Markov's inequality yields

\begin{flalign*}
\mathbb{P}\left(\frac{1}{n} \sum_{u=1}^{n} \one_{\lbrace (u,\pi^{\star}(u)) \in\cM \rbrace}<\alpha-\eps\right) & \leq \mathbb{P}\left(\sum_{u=1}^{n} \left(\mathbb{P}(M_u)-\one_{M_u}\right)>\eps n\right)\\
& \leq \frac{1}{n^2 \eps^2} \left(n \mathrm{Var}\left(\one_{M_1}\right)+ n(n-1)\mathrm{Cov}\left(\one_{M_1},\one_{M_2}\right) \right)\\
& \leq \frac{\mathrm{Var}\left(\one_{M_1}\right)}{n \eps^2} + \frac{1-\mathbb{P}\left(I_{1,2}\right)}{ \eps^2} \to 0,
\end{flalign*}
which ends the proof.
\end{proof}

\begin{remark}
Note that in view of the proof here above, the recovered fraction $\Omega(1)$ guaranteed by in Theorem \ref{MPAlign:theorem:good_matches} can be taken as close as wanted to
 $$ \alpha(\lambda s) := (1-\pext(\lambda s))^3 \left(1- \pi_{\lambda s}(0) - \pi_{\lambda s}(1) - \pi_{\lambda s}(2) \right). $$
This fraction is a priori not optimal, and can be interestingly compared with recent results in \cite{ganassali2021impossibility} (Chapter \ref{chapter:impossibility}) showing that no more than a fraction $1-\pext(\lambda s)$ of the nodes can be recovered. 
\end{remark}

\begin{proof}[Proof of Theorem \ref{MPAlign:theorem:no_mismatchs}]
First, we condition on the event $\cA$ that all $d-$neighborhoods in $G$ and $H$ are of size at most $C( \log n)\lambda^d$, which happens with probability $1-o(1)$ by Lemma \ref{MPAlign:lemma:control_S}. Note that by assumption this uniform upper bound is $O((\log n)n^{1/4})$.

In order to control the probability that $u$ is matched with some 'wrong' $u' \neq \pi^\star(u)$ by our algorithm, we follow the same first steps as in the proof of Theorem \ref{NTMA:thm:no_mismatchs} of Chapter \ref{chapter:NTMA}: we will first fix $u$ in $G$ and work on the event $\cE_u$ where $\cN_{G_{\cup},2d}(u)$ has no cycle. Since $c\log(\lambda) <1/4$, this event happens with probability $1-o(1)$ by Lemma \ref{MPAlign:lemma:cycles_ER}. 

Consider then $u'$ in $H$ such that $u' \neq \pi^{\star}(u)$. If $u$ and $u'$ are matched by \alg{MPAlign}, then necessarily $\cN_{G}(u,d)$ and $\cN_{H}(u',d)$ contain no cycle: the $d-$neighborhoods are thus tree-like. For any choice of distinct neighbors $v,w,x$ of $u$ in $G$ (resp. $v',w',x'$ of $u'$ in $H$), we define the corresponding pairs of trees of the form $(T_\ell,T'_\ell)$, where $T_\ell$ (resp. $T'_\ell$) is the tree of depth $d-1$ rooted at $\ell \in \set{v,w,x}$ in $G$ (resp. $\ell \in \set{v',w',x'}$ in $H$) after deletion of edge $\set{u,\ell}$ in $G$ (resp. $\set{u',\ell}$ in $H$). A moment of thought shows that, no matter the choice of $v,w,x$ and $v',w',x'$, on event $\cE_i$, one of these three pairs $(T_\ell,T'_\ell)$ must be made of \emph{two disjoint trees}. 

We now focus on a pair $(T,T')$ of such disjoint trees: these trees of depth $d-1$ can be built recursively by sampling a binomial number of children for each vertex. Since we condition on the fact that the trees are not intersecting, if at some point $k$ vertices have been uncovered, then the number of children to be drawn is exactly of distribution $\Bin\left(n-k , \lambda/n \right)$. With this exact construction, we denote by $\widetilde{\dP}_d$ the distribution of the pair $(T,T')$. Define 

\begin{equation}\label{eq:M_def}
    M_{d-1} := \frac{\widetilde{\dP}_{d-1}(T,T')}{\dPl_{d-1}(T,T')}.
\end{equation}
We have that
\begin{flalign*}
\widetilde{\dP}_{d-1}(L_{d-1}(T,T')> \beta \cap \cA ) & = \dEl_{d-1}\left[M_{d-1} \times \one_{\cA} \times \one_{L_{d-1}(T,T')> \beta} \right] \\
& \leq \dEl_{d-1}[M^2_{d-1} \one_{\cA} ]^{1/2} \beta^{-1/2} \, ,
\end{flalign*} by a successive use of Cauchy-Schwarz and Markov's inequalities, using that $\dEl_{d-1}\left[L_{d-1}(T,T')\right]=1$. We now state the following Lemma, proved in Appendix \ref{MPAlign:appendix:proof_lemma:control_M2}:

\begin{lemma}\label{MPAlign:lemma:control_M2}
With the previous notations, we have
\begin{equation}\label{eq:lemma:control_M2}
    \dEl_{d-1}\left[ M^2_{d-1} \one_{\cA} \right] = O(1).
\end{equation}
\end{lemma}

Together with the previous Lemma, noting that with high probability the maximum degree in $G$ and $H$ is less than $\log n$, union bound gives
\begin{flalign*}
\mathbb{P}\left(\cA \cap \left\{\exists u'  \neq \pi^{\star}(u), \; (u,u') \in\cM \right\}\right) & \leq \dP(\bar{\cE_i})+ o(1) + n \times \log^6 n \times \beta^{-1/2} \\
& = O\left((\log^6 n) \times n \times \exp(-n^{\gamma/2})\right) = o(1).
\end{flalign*}
The proof follows by appealing to Markov's inequality.
\end{proof}

\newpage

\begin{subappendices}
\section{Numerical experiments for \alg{MPAlign2}}\label{MPAlign:appendix:numerical}

In this section, we give some details on a practical implementation of our algorithm. We start by introducing some notations.
Given an edge $\set{u,v}$ of a graph, we denote by $u\to v$ and $v\to u$ the associated directed edges. Now given two graphs $G=(V,E)$ and $H=(V',E')$, we define the matrix $(m^t_{u \to v, u'\to v'})_{ \set{u,v} \in E , \set{u',v'} \in E'}\in \dR_+^{2|E|\times 2|E'|}$ recursively in $t$, as follows:
\begin{equation}
\label{eq:message}
    m^{t+1}_{u \to v, u' \to v'} =\sum_{k=0}^{d_u\wedge d_{u'} -1}\widetilde{\psi}(k,d_u-1,d_{u'}-1)\sum_{\substack{\{\ell_1,\dots \ell_k\}\in \partial u \backslash v \\ \{w_1,\dots w_k\}\in \partial u' \backslash v'}}\sum_{\sigma\in \cS_k}\prod_{a=1}^k m^t_{\ell_a\to u, w_{\sigma(a)}\to u'} \, ,
\end{equation}
where $d_u := d_{G}(u)$, $d'_{u'} := d_{H}(u')$, $\widetilde{\psi}(k,d_1,d_2) = k!\psi(k,d_1,d_2)$, and $\partial u \backslash v$ (resp. $\partial u'  \backslash v'$) is a shorthand notation for $\cN_G(u)\setminus \{v\}$ (resp. $\cN_{H}(u')\setminus \{v'\}$) and by convention $m^0_{u \to v, u' \to v'}=1$.

Denoting $\partial u := \cN_G(u)$ (resp. $\partial u' := \cN_{H}(u')$), for $t\in \dN$ we define the matrix $(m^t_{u,u'}) \in \dR_+^{V\times V'}$ as follows:
\begin{equation}
\label{eq:aggregation}
    m^t_{u,u'} = \sum_{k=0}^{d_u\wedge d_{u'}}\widetilde{\psi}(k,d_u,d_{u'})\sum_{\substack{\{\ell_1,\dots \ell_k\}\in \partial u \\ \{w_1,\dots w_k\}\in \partial u'}}\sum_{\sigma\in \cS_k}\prod_{a=1}^k m^t_{\ell_a\to u, w_{\sigma(a)}\to u'}.
\end{equation}

It is easy to see that if the graphs $G$ and $H$ are tree-like up to depth $t$, then $m^t_{u,u'}$ is exactly the likelihood ratio $L_t(s_u,s'_{u'})$ where $s_u$ (resp. $s'_{u'}$) is the tree neighborhood of $u$ in $G$ (resp. of $u'$ in $H$).

In experiments, we run our algorithms on correlated \ER model with possible cycles, so that the matrix $m^t_{u,u'}$ is interpreted as an approximation of the true likelihood ratio. From such an approximation, we compute two mappings $\pi^t:V\to V'$ as $$\pi^t(u) = \argmax(m^t_{u,\cdot})$$ and $\sigma^t:V'\to V$ as $$\sigma^t(u') = \argmax(m^t_{\cdot, u'})$$ which are candidates for matching vertices from $G$ to $H$ or from $H$ to $G$.
If $t$ is small, then the approximation $m^t_{u,u'}$ will not be accurate as it does not incorporate sufficient information (only at depth $t$ in both graphs). When $t$ is large, cycles will appear in both graphs so that the recursion is not anymore valid. In order to choose an appropriate number of iterations $t$, we adopt the following simple strategy: we compute all the matrices $m^t_{u,u'}$ for all values of $t$ less than a parameter $d$; then from these matrices, we compute the corresponding mappings $\pi^t$ and $\sigma^t$ as described above; we then compute:
\begin{flalign}
\label{eq:edges}
e(t) &:= \text{match-edges}(G,H,\pi^t,\sigma^t) \nonumber \\
 &:= \frac{1}{|E|}\sum_{\set{u,v}\in E} \one_{(\pi^t(u),\pi^t(v))\in E'} + \frac{1}{|E'|}\sum_{\set{u',v'}\in E'} \one_{(\sigma^t(u'),\sigma^t(v'))\in E} \, .
\end{flalign}
Finally, we choose $$t^* =\argmax(e(t)) \, .$$
Note that, we are considering sparse \ER graphs which are typically not connected (the diameter is infinite). We know from \cite{ganassali2021impossibility} (Chapter \ref{chapter:impossibility}) that only the giant components of $G$ and $H$ can possibly be aligned. Hence as a first pre-processing step, we remove all the small connected components from $G$ and $H$ and keep only the largest one. As a result, our algorithm takes as input 2 connected graphs of (possibly) different sizes. The pseudo-code for our algorithm is given below:

\begin{algorithm}[H]
\caption{\label{MPAlign:algo_MPAlign2} \alg{MPAlign2}}
\SetAlgoLined

\textbf{Input:} Two connected graphs $G=(V,E)$ and $H=(V',E')$, parameter $d$ and parameters of the correlated \ER model $\lambda$ (average degree) and $s$

\For{$t \in \{1,\dots, d\}$}{
	compute $m^t_{u\to v, u'\to v'}$ thanks to \eqref{eq:message}
	
	compute $m^t_{u,u'}$ thanks to \eqref{eq:aggregation}

	compute $\pi^t:V\to V'$ as $\pi^t(u) = \argmax(m^t_{u,\cdot})$
	
	compute $\sigma^t:V'\to V$ as $\sigma^t(u') = \argmax(m^t_{\cdot, u'})$
	
	compute $e(t) = \text{match-edges}(G,H,\pi^t,\sigma^t)$ thanks to \eqref{eq:edges}
	
	}
$t^* = \argmax(e(t))$

\textbf{Return} $\pi^{t^*}$, $\sigma^{t^*}$, $m^{t^*}$

\end{algorithm}

Figure \ref{MPAlign:fig:overlap} shows some empirical results for graphs of size $200$ for values $\lambda = 2;2,5;3$ where the overlap is the mean of the overlaps given by $\pi^{t^*}$ and $\sigma^{t^*}$. The maximum number of iterations is fixed to $d=15$. For more numerical experiments on this algorithm, see \cite{piccioli2021aligning}.

\begin{figure}[H]

\includegraphics[width=14cm]{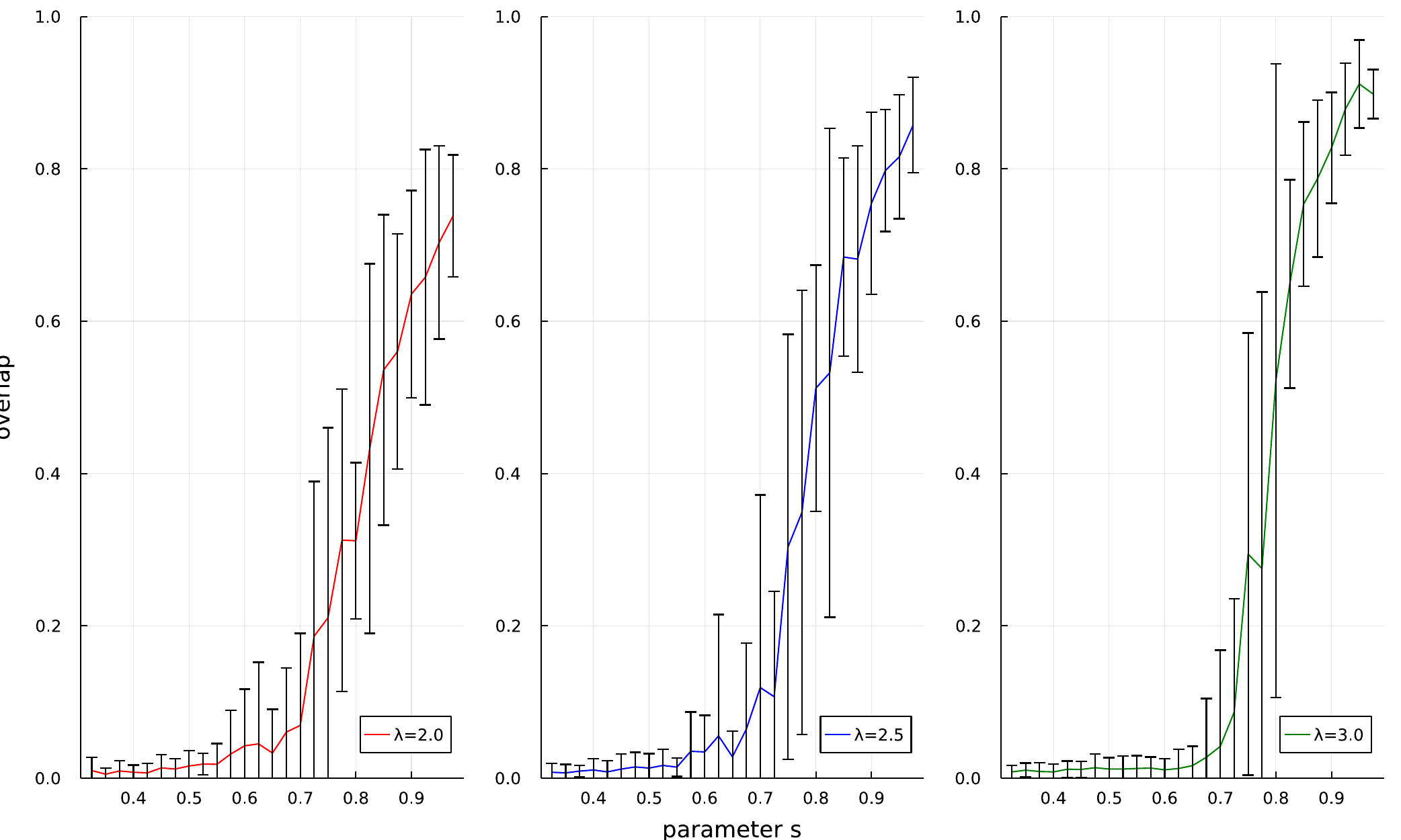}
\centering
\caption{\label{MPAlign:fig:overlap} Overlap as a function of the parameter $s$ for graphs with (initial) size $n=200$ for various values of $\lambda$ (parameter $d=15$). Each point is the average of $10$ simulations. }
\end{figure}

 This choice of $d=15$ is validated by the results presented in Figure \ref{fig:iter}. We plot for each simulation the mean overlap of $\pi^t$ and $\sigma^t$ as a function of $t\leq 15$. We see that for low values of $s$ (on the left $s=0.4$), the overlap behaves randomly. In this scenario, increasing the value of $d$ will probably not help as cycles will deteriorate the performance of the algorithm. For high value of $s$ (on the right $s=0.9$), we see that the overlap starts by increasing and then decreases abruptly to zero, this is due to numerical issues: some messages in $m^t$ are too large for our implementation of the algorithm to be able to deal with them. Finally for values of $s$, where signal is detected (in the middle $s=0.675$), we see that when the signal is detected, the overlap start by increasing until reaching a maximum and then decreases before numerical instability. We also note that our choice of $t^*$ thanks to the number of matched edges can be fairly sub-optimal. We believe that a better understanding of the performance of our algorithm for finite $n$ is an interesting open problem. Indeed, we refer to \cite{piccioli2021aligning} which provides more detailed experimental results on a similar algorithm.

\begin{figure}[H]

\includegraphics[width=14cm]{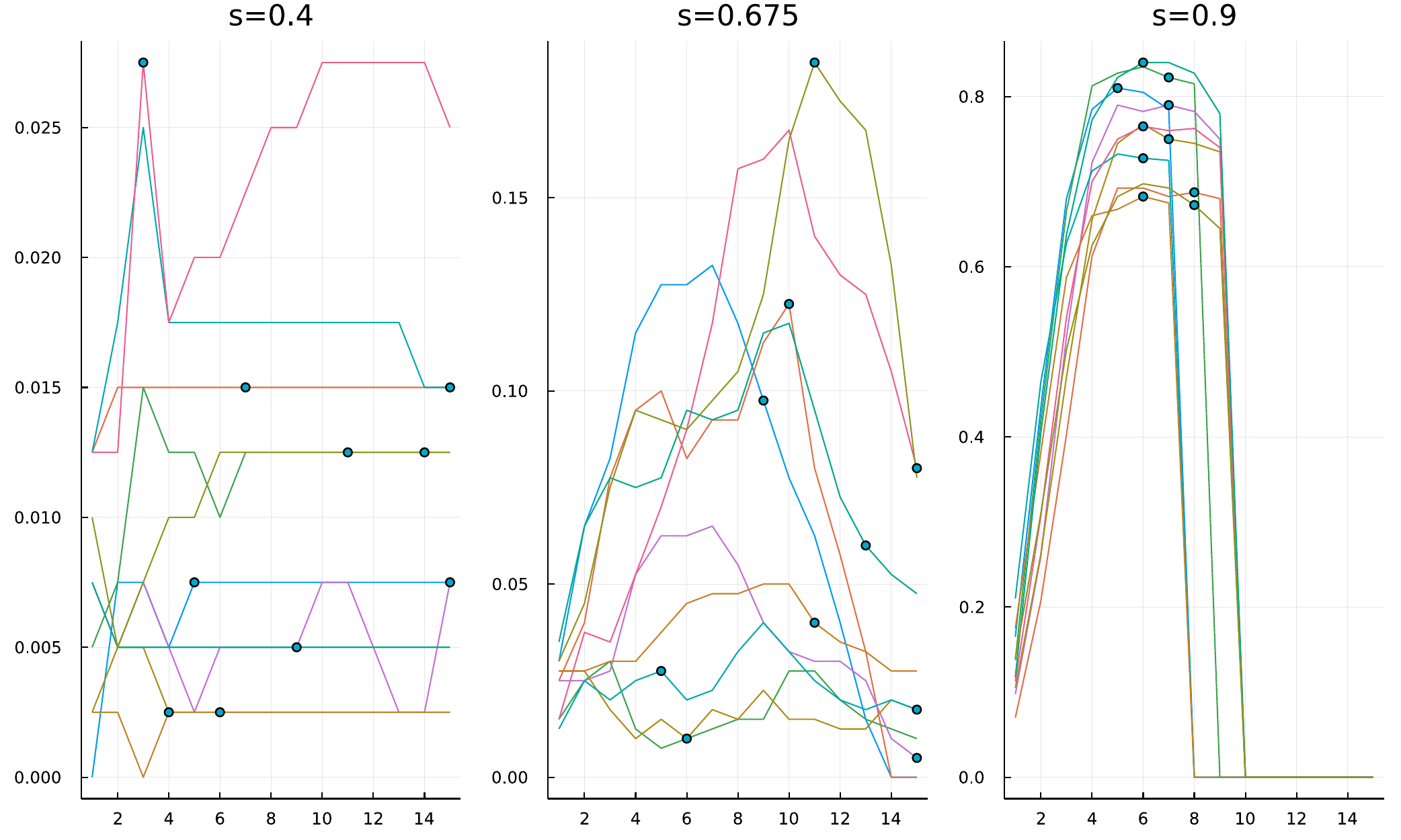}
\centering
\caption{\label{fig:iter} Overlap as a function the number of iterations $t$ for graphs with (initial) size $n=200$ for $\lambda=2.5$ (parameter $d=15$) and various values of $s$. The dotted point on each curve corresponds to $t^*$. Note that the $y$-axis of each plot have different scale. When overlap reaches zero, our algorithm hits infinity.
}
\end{figure}

\section{Additional proofs}\label{MPAlign:appendix:additional_proofs}
\addtocontents{toc}{\protect\setcounter{tocdepth}{0}}
\subsection{Proof of Proposition \ref{proposition:auto_GW}}
\label{appendix:proof:proposition:auto_GW}
\begin{proof}
Throughout the proof, let $X_\mu$ denote a Poisson random variable with parameter $\mu$. 
A node $u \in\cL_{n-2}(\tau^{\star})$ has, independently for each $k\in \dN$, a number $N_k \sim \Poi(r\pi_r(k))$ children who themselves have $k$ children. To each such node, we can associate 
\begin{equation*}
    \prod_{k\in \dN}N_k!
\end{equation*} permutations of its children that will preserve the labeled tree. Likewise, for each node $u \in \cL_{d-1}(\tau^{\star})$, there are $c_u!$ permutations of its children that don't modify the tree, where $c_u := c_{\tau^{\star}}(u)$. Thus by the strong law of large numbers, we have:
\begin{equation}\label{eq:equivalent_logaut1}
\log \card{\Aut(\tau^{\star})} \geq (1+o_{\dP}(1))\left[w r^{n-1}\dE\left[\log(X_r!)\right]+ w r^{n-2}\sum_{k\in \dN}\dE\left[\log(X_{r\pi_r(k)}!)\right]\right].
\end{equation}
Recall the classical estimate for large $\mu$:
\begin{equation}\label{eq:moment_factoriel2}
\dE \log(X_\mu!)=\mu\log(\mu)-\mu+\frac{1}{2}\log(2\pi e \mu)+O\left(\frac{1}{\mu}\right),
\end{equation}
and Stirling's formula gives
\begin{equation}\label{eq:stirling}
    \log(k!)= k \log k - k + \frac{1}{2}\log(2 \pi k) + O\left( \frac{1}{k} \right).
\end{equation}
We now give some estimates of the distribution $\pi_r(k)$ in the following Lemma, which proof is deferred to Appendix \ref{appendix:proof_lemma_equivalent_pi}.
\begin{lemma}\label{lemma_equivalent_pi}
Let $\eps(r)$ be such that $\eps(r) \to 0$ and $\eps(r)\log r \to +\infty$ when $r \to +\infty$. Let 
\begin{equation*}
    I_{r,\eps}:=\left[r-(1-\eps(r))\sqrt{r\log r},r+(1-\eps(r))\sqrt{r\log r}\right].
\end{equation*} Then 
\begin{itemize}
    \item[$(i)$] we have
    \begin{equation}\label{eq:lemma_concentration_poisson}
        \dP\left( X_r \notin I_{r,\eps} \right) = O\left(r^{-1/2} e^{\eps(r)\log r}\right).
    \end{equation}
    \item[$(ii)$] for all $k \in I_{r,\eps}$, letting $x_k=\frac{k-r}{\sqrt{r}}$, we have
\begin{equation}\label{eq:lemma_equivalent_pi}
\pi_r(k)=\frac{1}{\sqrt{2\pi r}}e^{-{x_k^2}/{2}}\left[1+\frac{x_k^3}{6\sqrt{r}}-\frac{x_k}{2\sqrt{r}}+O\left(\frac{x_k^6}{r}\right)\right].
\end{equation}
  \item[$(iii)$] Note that \eqref{eq:lemma_equivalent_pi} implies that for each $k\in I_{r,\eps}$, it holds that $r\pi_r(k) = \Omega\left(e^{\eps(r) \log r (1-o(1))}\right) $, thus diverges to $+\infty$.
\end{itemize}
\end{lemma}

Consider the function $\eps(r) := \frac{\log\log r}{4 \log r}$, which satisfies the assumptions of Lemma \ref{lemma_equivalent_pi}. Using expansion \eqref{eq:lemma_equivalent_pi} together with \eqref{eq:moment_factoriel2} gives:
\begin{flalign}\label{eq:logiso_2}
\sum_{k \in I_{r,\eps}} \dE & \left[\log(X_{r\pi_r(k)}!)\right] = \sum_{k\in I_{r,\eps}} r\pi_r(k)\log(r \pi_r(k))-r\pi_r(k)+\frac{1}{2}\log(2\pi e r\pi_r(k))+O\left(\frac{1}{r\pi_r(k)}\right) \nonumber\\
& =\sum_{k\in I_{r,\eps}}r\pi_r(k)\left[\frac{1}{2}\log(r)-\frac{1}{2}\log(2\pi)-\frac{x_k^2}{2}+\frac{x_k^3}{6\sqrt{r}}-\frac{x_k}{2\sqrt{r}}+O\left(\frac{\log^2 r}{r}\right)-1\right] \nonumber\\
 & \quad \quad +\sum_{k\in I_{r,\eps}}\frac{1}{2} \left[\log( 2\pi e)+\frac{1}{2}\log(r)-\frac{1}{2}\log(2\pi)-\frac{x_k^2}{2}+O\left(\frac{\log^{3/2}(r)}{\sqrt{r}}\right)\right]+ O\left(\sqrt{r \log r} \right) \nonumber\\
 & \overset{(a)}{=} \frac{1}{2}r\log(r)-\left(\frac{1}{2}\log(2\pi)+\frac{1}{2}+1\right)r +O(\sqrt{r}\log^{5/4} r ) \nonumber\\
 & \quad \quad + O(\sqrt{r\log r}) + \frac{1}{2}\left( 1-\eps(r)\right)\sqrt{r}\log^{3/2}(r)-\frac{1}{4}\sum_{k\in I_{r,\eps}} x_k^2 \nonumber\\
 & \overset{(b)}{=} \frac{1}{2}r\log(r)-\left(\frac{1}{2}\log(2\pi)+\frac{3}{2}\right)r + \frac{1}{3} \sqrt{r}\log^{3/2}(r) +O(\sqrt{r}\log^{5/4} r).
\end{flalign}
Let us give hereafter all the required details for the above computation.
\begin{itemize}
    \item At step $(a)$, we first used point $(i)$ of Lemma \ref{lemma_equivalent_pi}, which gives that $$r \log r \times \dP\left( X_r \notin I_{r,\eps} \right) = O\left(\sqrt{r} \log^{1/4} r\right) = O\left(\sqrt{r} \log^{5/4} r\right).$$ For the sum of the $x_k^2$, we remark that
    \begin{flalign*}
        \sum_{k\in I_{r,\eps}} r \pi_r(k) \frac{x_k^2}{2} &= \frac{r}{2} \left(1- \dE\left[\left(\frac{X_r-r}{\sqrt{r}}\right)^2 \one_{X_r \notin I_{r,\eps}}\right]\right),
    \end{flalign*} and that the expectation in the right-hand term can be written as follows
    \begin{flalign*}
    \dE&  \left[\left(\frac{X_r-r}{\sqrt{r}}\right)^2 \one_{\left|\frac{X_r-r}{\sqrt{r}}\right| \geq 2 \sqrt{\log r}}\right] + \dE\left[\left(\frac{X_r-r}{\sqrt{r}}\right)^2 \one_{ (1-\eps(r)) \sqrt{\log r} \leq \left|\frac{X_r-r}{\sqrt{r}}\right| \leq 2 \sqrt{\log r}}\right]\\
    & \leq \dE\left[\left(\frac{X_r-r}{\sqrt{r}}\right)^4 \right]^{1/2} \dP\left( \left|\frac{X_r-r}{\sqrt{r}}\right| \geq 2 \sqrt{\log r} \right)^{1/2} + 4\log r \times  \dP\left( X_r \notin I_{r,\eps} \right)\\
    & \leq O\left( r^{-1/2} \right) + O\left( r^{-1/2}  \log^{5/4} r \right).
    \end{flalign*} Hence, $\sum_{k\in I_{r,\eps}} r \pi_r(k) \frac{x_k^2}{2} = \frac{r}{2} - O\left(\sqrt{r} \log^{5/4} r \right)$. Finally, using the fact that $\dE\left[\left(\frac{X_r-r}{\sqrt{r}}\right)^3 \right]$ and $\dE\left[\frac{X_r-r}{\sqrt{r}} \right]$ are $O(1)$, the sums of the $x_k^3$ and $x_k$ easily incorporate into the $O\left(\sqrt{r} \log^{5/4} r \right)$ term.
    
    \item At step $(b)$, we first used the fact that $\eps(r)\sqrt{r}\log^{3/2} = O\left(\sqrt{r} \log^{5/4} r \right)$. The only term requiring more computations is
    \begin{flalign*}
        \sum_{k\in I_{r,\eps}} x_k^2 & = \sum_{k\in I_{r,\eps}} \left(\frac{k-r}{\sqrt{r}}\right)^2 = 2 \times \sum_{\ell = 0}^{(1-\eps(r))\sqrt{r\log r}} \frac{\ell^2}{r} = \frac{2}{3} \sqrt{r} \log^{3/2} r + O\left(\sqrt{r} \log^{5/4} r \right).
    \end{flalign*}
\end{itemize}

Copying \eqref{eq:logiso_2} together with \eqref{eq:moment_factoriel2} in \eqref{eq:equivalent_logaut1} yields:
\begin{flalign*}
\log(\card{\Aut(\tau^{\star})}) & \geq (1+o_{\dP}(1)) w r^{n-1} \left[r\log(r)-r+\frac{1}{2}\log(2\pi e r)+ O\left( \frac{1}{r}\right) \right] \\
& \quad \quad + (1+o_{\dP}(1)) w r^{n-1} \left[\frac{1}{2}\log(r)-\frac{1}{2}\log(2\pi)-\frac{3}{2}+\frac{\log^{3/2} r}{3\sqrt{r}}+O\left(\frac{\log^{5/4} r}{\sqrt{r}}\right)\right]\\
& = (1+o_{\dP}(1)) w r^{n-1}\left[r\log(r)-r+\log(r)-1+\frac{\log^{3/2}(r)}{3\sqrt{r}}+O\left(\frac{\log^{5/4} r}{\sqrt{r}}\right)\right].
\end{flalign*}
Another appeal to the strong law of large numbers entails that 

\begin{flalign*}
    \log\left(\prod_{u \in \cV_{d-1}(\tau^{\star})}e^{-r} r^{c_{\tau^{\star}}(u)} \right) & = (1+o_{\dP}(1))\left|\cV_{d-1}(\tau^{\star}) \right| \dE\left[-r + c_{\tau^{\star}}(\rho(\tau^{\star})) \log r \right]\\
    & = (1+o_{\dP}(1))K\left( -r+r\log(r)\right).
\end{flalign*}

Combined, these last two evaluations yield a lower bound of $\log\left(\frac{\card{\Aut(\tau^{\star})}}{\prod_{u \in \cV_{d-1}(\tau^{\star})}e^{-r} r^{c_{\tau^{\star}}(u)}}\right) $ under the event on which $\tau^{\star}$ survives, of the form
\begin{flalign*}
& (1-o_{\dP}(1))K \left[-r\log(r)+r+\left(1-\frac{1}{r}\right)\left(r\log(r)-r+\log(r)-1+\frac{\log^{3/2}(r)}{3\sqrt{r}}+O\left(\frac{\log^{5/4} r}{\sqrt{r}}\right)\right)\right]\\
& = (1-o_{\dP}(1)) K \left[ \frac{\log^{3/2}(r)}{3\sqrt{r}}+O\left(\frac{\log^{5/4} r}{\sqrt{r}}\right) \right] \, .
\end{flalign*}
\end{proof} 

\subsection{Proof of Lemma \ref{lemma_equivalent_pi}}
\label{appendix:proof_lemma_equivalent_pi}
\begin{proof}
\emph{$(i)$} The result follows directly from the classical Poisson concentration inequality 
\begin{equation*}
    \dP\left( \left|X_r - r \right| \geq x \right) \leq 2 \exp\left( - \frac{x^2}{2(r+x)}\right),
\end{equation*} noting that for $x=(1-\eps(r))\sqrt{r\log r}$,
$\frac{x^2}{2(r+x)} \geq \frac{1}{2} \log r - \eps \log r - o(1).$

\emph{$(ii)$}
When $k$ runs over $I_{r,\eps}$, $x_k$ runs over $\left[ -(1-\eps(r)) \sqrt{\log r}, (1-\eps(r)) \sqrt{\log r}\right]$. Using Stirling's formula \eqref{eq:stirling}, we get
\begin{flalign*}
\log\pi_r(k) & = \log\pi_r(r+x_k \sqrt{r}) = -r + k \log r - \log(k!) \\
& = -r + (r+x_k \sqrt{r}) \log r - (r+x_k \sqrt{r}) \log(r+x_k \sqrt{r}) + r+x_k \sqrt{r} - \frac{1}{2}\log(2 \pi (r+x_k \sqrt{r})) + O\left( \frac{1}{r} \right)\\
& = -r + r\log r + x_k \sqrt{r} \log r - (r+x_k \sqrt{r}) \left[ \log r + \frac{x_k}{r^{1/2}} - \frac{x_k^2}{2r} + \frac{x_k^3}{3r^{3/2}} + O\left( \frac{x_k^4}{ r^{2}}\right)  \right] \\
& \quad \quad \quad \quad + r + x_k \sqrt{r} - \frac{1}{2}\log(2 \pi) - \frac{1}{2}\log(r) - \frac{1}{2} \frac{x_k}{r^{1/2}} + O\left( \frac{x_k^2}{ r}\right)\\
& = -r - x_k \sqrt{r} -x_k^2 + \frac{x_k^2}{2} + \frac{x_k^3}{2 \sqrt{r}} - \frac{x_k^3}{3 \sqrt{r}} + O\left( \frac{x_k^4}{ r}\right) \\
& \quad \quad \quad \quad + r + x_k \sqrt{r} - \frac{1}{2}\log(2 \pi r) - \frac{1}{2} \frac{x_k}{r^{1/2}} + O\left( \frac{x_k^2}{ r}\right)\\
& = - \frac{x_k^2}{2} - \frac{1}{2}\log(2 \pi r) + \frac{x_k^3}{6 \sqrt{r}} - \frac{x_k}{2 \sqrt{r}} + O\left( \frac{x_k^4}{ r}\right).
\end{flalign*} Taking the exponential gives
\begin{flalign*}
\pi_r(k) & = \frac{1}{\sqrt{2\pi r}} e^{-{x_k^2}/{2}} \left[ 1+\frac{x_k^3}{6\sqrt{r}}-\frac{x_k}{2\sqrt{r}}+O\left(\frac{x_k^6}{r}\right)\right].
\end{flalign*}

$(iii)$ follows directly from $(ii)$.
\end{proof}

\subsection{Proof of Lemma \ref{MPAlign:lemma:control_M2}}\label{MPAlign:appendix:proof_lemma:control_M2}
\begin{proof}
We condition on $P$ be the number of recursive steps in the previous construction, which is $O((\log n)n^{1/4})$ under $\cA$. For each $s \in [P]$, we denote by $c_s$ the number of newly sampled children, and $V_s := \sum_{s'=0}^{s-1} c_{s'}$ the number of uncovered vertices before step $s$ (we set $V_0 : =0$). With these notations, it is easily seen than $M_{d-1}$ can be factorized as follows:
\begin{flalign*}
M_{d-1} & = \prod_{s \in [P]} \frac{\dP\left(\Bin\left(n-2-V_s , \lambda/n \right) = c_s \right)}{\pi_{\lambda}(c_s)} \leq \prod_{s \in [P]}\exp\left(  \frac{\lambda}{n}(V_s+2+c_s)\right) \\
& = \exp\left(\frac{2\lambda P}{n} + \frac{\lambda}{n} \sum_{s \in [P]} (P-s)c_s \right).
\end{flalign*}
Under $\dPl_d$, the variables $c_s$ are independent $\Poi(\lambda)$ variables, hence
\begin{flalign*}
 \dEl_{d}\left[ M^2_{d-1} \one_{\cA} \right] & \leq \exp\left(\frac{4\lambda P}{n} + \lambda \sum_{s \in [P]} \left( e^{{2 \lambda}(P-s)/{n}}-1\right) \right) \one_{P = O((\log n)n^{1/4})}\\
 & \leq \exp\left( C' {P^2}/{n} + o({P^2}/{n})\right)\one_{P = O((\log n)n^{1/4})} = O(1).
\end{flalign*}
\end{proof}
\addtocontents{toc}{\protect\setcounter{tocdepth}{2}}
\end{subappendices}

\chapter{\textsc{Addendum: new results for correlation detection in trees}}\label{chapter:addendum}
\renewcommand{\thesection}{\the\value{chapter}.\the\value{section}}
This addendum, which concludes the manuscript, presents new results for correlation in trees from a recent joint work with L. Massoulié and G. Semerjian (paper in preparation). These results are significantly improving on previous work and give a general understanding of the fundamental limits of the problem, as well as some interesting perspectives discussed afterwards in the conclusion.

We do not redefine here the problem of correlation detection in trees, since we already thouroughly did in Chapters \ref{chapter:NTMA} and \ref{chapter:MPAlign}, but we rather introduce some auxiliary defintions that proved useful for the analysis made his the sequel. For related work, we refer to Section \ref{intro:subsection:short_survey} of the introduction.

We however straightaway mention that the results presented in this last part are very much related to the recent study of Mao, Wu, Xu and Yu \cite{Mao21_counting} who studied the correlation detection problem in \ER graphs, and proposed an algorithm based on counting (signed) trees, which is provably distinguish graph correlation efficiently as soon as $s > \sqrt{\alpha}$, where $\alpha$ is the Otter's constant defined below in Proposition \ref{addendum:prop:otter}. The results presented here are different for several reasons: first, we study the problem on trees and consider an optimal test, which thus also meets the informational bounds. Moreover, we give a two-side bound, also showing that $s < \sqrt{\alpha}$ implies impossibility of one-sided detection. 

We believe that this study paves the way for many other works in this field, generalizing to other graph models, analyzing the computational hardness of the problem with different tools, and designing more efficient algorithms for tree correlation detection or graph alignment -- for more insights and details on these research directions, we refer to the conclusion.

\section{Main results}

We start by stating some familiar definitions, in the general context of \emph{unlabeled trees}, that are part of the rationale for results to follow.

\subsection{Definitions: trees}

\begin{definition}[Finite rooted unlabeled trees]\label{addendum:def:rooted_unlabeled_trees}
	We recursively define the set $\cX_d$ of \emph{finite rooted unlabeled trees of depth at most $d \geq 0$}. 
	
	For $d = 0$, $\cX_d$ contains the trivial tree reduced to its root node, denoted by $\bigcdot$. 
	
	For $d \geq 1$, having defined $\cX_0, \ldots, \cX_{d-1}$, we define $\cX_d$ as follows: a finite rooted unlabeled tree $t \in \cX_d$ consists in an integer sequence $\{N_{\tau} \}_{\tau \in \cX_{d-1}}$ with finite support, that is such that 
	$$ \sharp \left\{\tau \in \cX_{d-1}, N_{\tau} \neq 0 \right\} < \infty,$$
	where $N_{\tau}$ is the number of children of the root in $t$ which subtrees are copies of $\tau$.
\end{definition} 
Throughout all the chapter, we will only work with finite trees (with finite degrees), hence the adjective 'finite' will be omitted as a shortcut.  

\begin{remark}
	With the previous definition, equality between two rooted unlabeled trees $t := \{N_{\tau} \}_{\tau \in \cX_{d-1}}$ and $t' := \{N'_{\tau} \}_{\tau \in \cX_{d-1}}$ is defined as $N_{\tau} = N'_{\tau} $ for all $\tau \in \cX_{d-1}$.
\end{remark}

\begin{remark}
	Denoting one-to-one correspondence by $\simeq$, we remark that $\cX_1 \simeq \dN$ (the set of non-negative integers), and that more generally for each $d \geq 1$, 
	$$ \cX_d \; \simeq \; \bigcup_{\ell \geq 1} \; \bigcup_{\substack{\tau_1, \ldots, \tau_\ell \in \cX_{d-1} \\ i \neq j \implies \tau_i \neq  \tau_j}} \dN^{\ell}.$$ 
	Hence, $\cX_d$ is countably infinite, by recursion, for all $d \geq 1$.
\end{remark}

\begin{definition}[Size of a rooted unlabeled tree]
	The \emph{size}, or \emph{number of nodes}, of a tree $t \in \cX_{d}$ is denoted by $\card{t}$ and defined recursively as follows. First, if $d=0$, we set $\card{\bigcdot}=1$. Then, for $d \geq 1$, writing $t = \{N_{\tau} \}_{\tau \in \cX_{d-1}}$, one has
	\begin{equation*}
	\card{t} = 1 + \sum_{\tau \in \cX_{d-1}} N_{\tau} \cdot \card{\tau} \,.
	\end{equation*}
\end{definition}

\begin{definition}[Depth of a rooted unlabeled tree]
	The \emph{depth} of a tree $t \in \cX_{d}$ is denoted by $\depth{t}$ and defined recursively as follows. First, if $d=0$, we set $\depth{\bigcdot}=0$. Then, for $d \geq 1$, writing $t = \{N_{\tau} \}_{\tau \in \cX_{d-1}}$, one has
	\begin{equation*}
	\depth{t} = 1 + \max\left\{ \one_{N_\tau \geq 1} \cdot \depth{\tau}, \; \tau \in \cX_{d-1} \right\} \,.
	\end{equation*}
\end{definition}

\begin{definition}[Child of a rooted unlabeled tree]
	A rooted unlabeled tree $s \in \cX_{d-1}$ is said to be a \emph{child} of a tree $t = \{N_{\tau} \}_{\tau \in \cX_{d-1}}$ if $N_s \geq 1$. Moreover, if $N_s \geq 1$, $t$ is said to have $N_s$ \emph{children of type $s$}. Note that since the tree $t$ is finite, it has a finite number of children, given by $\sum_{\tau} N_{\tau}$.
\end{definition}

\begin{definition}[Subtree of a rooted unlabeled tree]
	Let us define recursively the notion of \emph{subtree}. First, the only subtree of $\bigcdot$ is $\bigcdot$ itself. Then for $d' \leq d$, an element $s \in \cX_{d'}$ is a subtree of $t  \in \cX_d$ if either $s=t$ or if $s$ is a subtree of some child of $t$.
\end{definition}

In the sequel, we are interested in the cardinality of the set of unlabeled trees with given size. 

\begin{definition}[Trees with given size and depth]\label{def:An_Adn}
	For $n \geq 1$, let us define
	\begin{equation}\label{eq:def:An}
	A_n := \card{ \left\{ t \in \bigcup_{d \geq 0} \cX_d, \; \card{t}=n \right\} } \, ,
	\end{equation} that is $A_n$ is the number of (distinct) unlabeled rooted trees of size $n$. For $d \geq 0$, we furthermore define
	\begin{equation}\label{eq:def:Adn}
	A_{d,n} := \card{ \left\{ t \in \cX_d, \; \card{t}=n \right\} } \, ,
	\end{equation} that is $A_{d,n}$ is the number of (distinct) unlabeled rooted trees of size $n$ and depth at most $d$. 
\end{definition}

We now state a celebrated result by Otter \cite{Otter48}, together with a proposition that will be useful in the sequel.

\begin{proposition}[Asymptotic number of unlabeled trees, \cite{Otter48}]\label{addendum:prop:otter}
	One has
	\begin{equation}\label{eq:theorem:otter}
	A_n \underset{n \to \infty}{\sim} \frac{C}{n^{3/2}} \left( \frac{1}{\alpha} \right)^n,
	\end{equation} for some $C>0$, where $\alpha \in (0,1)$ is the Otter constant, numerically $\alpha = 0.3383219...$
\end{proposition}

\begin{proposition}[Control of the generating function of the $A_{d,n}$]\label{addendum:prop:A_dn}
	For all $d \geq 0$, let
	\begin{equation}\label{eq:prop:A_dn:phid}
	\Phi_d(x) := \sum_{n \geq 1} A_{d,n} x^{n-1} \, .
	\end{equation} We have, for all $x>0$, 
	\begin{equation}\label{eq:prop:A_dn:cv_phi}
	\Phi_d(x) \underset{d \to \infty}{\longrightarrow} \Phi(x) \, ,
	\end{equation} where 
	\begin{equation}\label{eq:prop:A_dn:phi}
	\Phi(x) := \sum_{n \geq 1} A_{n} x^{n-1} \, .
	\end{equation} Moreover, for all $d \geq 0$ and $t \in [0,1)$, there exists $A=A(d,t)$ such that 
	\begin{equation}\label{eq:prop:A_dn}
	\forall x \in [0,t], \, \left| \Phi_d(x) \right| \leq A \, .
	\end{equation}
\end{proposition}

\begin{proof}[Proof of Proposition \ref{addendum:prop:A_dn}]
	Convergence \eqref{eq:prop:A_dn:cv_phi} follows from $A_{d,n} \underset{d \to \infty}{\longrightarrow} A_n$ and monotone convergence theorem. 
	
	We will now establish a recursion property on the $\Phi_d$, adapting the proof of Otter \cite{Otter48} of the case $d=\infty$ to general depth $d$. We can decompose a tree $t$ of depth at most $d+1$ with $n$ vertices according to its subtrees. For $i\geq 1$, we denote by $\mu_i \geq 0$ the number of subtrees of $t$ with $i$ nodes. The $\mu_i$ subtrees are then distributed in the $A_{d,i}$ categories, the only thing that distinguish them being the number in each $A_{d,i}$. For a given $i$, there is $\binom{A_{d,i}+\mu_i-1}{\mu_i}$ such choices. This gives the following recursion formula:
	\begin{equation*}
	A_{d+1,n} = \sum_{\set{\mu_i}_{i \geq 1}} \prod_{i \geq 1} \binom{A_{d,i}+\mu_i-1}{\mu_i} \one_{n = 1+ \sum_{i\geq 1} i \mu_i} \, .
	\end{equation*} Plugging the last equation in the definition of $\Phi_{d+1}$ gives
	\begin{flalign*}
	\Phi_{d+1}(x) & = \sum_{n \geq 1} x^{n-1} \sum_{\set{\mu_i}_{i \geq 1}} \prod_{i \geq 1} \binom{A_{d,i}+\mu_i-1}{\mu_i} \one_{n = 1+ \sum_{i\geq 1} i \mu_i}  \\
	& = \sum_{\set{\mu_i}_{i \geq 1}} \prod_{i \geq 1} \left[ \binom{A_{d,i}+\mu_i-1}{\mu_i} x^{i \mu_i} \right] 
	= \prod_{i \geq 1} \sum_{\mu \geq 0}  \left[ \binom{A_{d,i}+\mu-1}{\mu} x^{i \mu} \right]\\
	& = \prod_{i \geq 1} \frac{1}{(1-x^i)^{A_{d,i}}} = \exp\left(- \sum_{i \geq 1} A_{d,i}\log(1-x^i) \right)
	= \exp\left(\sum_{i,j \geq 1} A_{d,i}\frac{x^{i j}}{j} \right) \\
	& = \exp\left(\sum_{j \geq 1}\frac{x^j}{j} \sum_{i \geq 1}  A_{d,i} (x^j){i-1} \right) = \exp\left(\sum_{j \geq 1} \frac{x^j}{j} \Phi_d(x^{j}) \right) \, .
	\end{flalign*}
	Equation \eqref{eq:prop:A_dn} is then very easy to propagate by recursion with this last formula. For $d=0$, $A_{d,n} = \one_{n=1}$, hence $\Phi_0(x)=1$ and \eqref{eq:prop:A_dn} holds with $A=1$. Assume that \eqref{eq:prop:A_dn} holds at depth $d$ for all $t \in [0,1)$ with constant $A(d,t)$. Then, by the previous computation, for all $t \in [0,1)$ and $x \in [0,t]$, since $x^j \in [0,t]$ for all $j \geq 1$ we have
	\begin{flalign*}
	\Phi_{d+1}(x) & = \exp\left(\sum_{j \geq 1} \frac{x^j}{j} \Phi_d(x^{j}) \right) \leq \exp\left(\sum_{j \geq 1} \frac{x^j}{j} A(d,t) \right) \\
	& = \exp\left(- A(d,t) \log(1-x)  \right) = \left(\frac{1}{1-x}\right)^{A(d,t)} \leq \left(\frac{1}{1-t}\right)^{A(d,t)} =: A(d+1,t)\, .
	\end{flalign*} Thus, \eqref{eq:prop:A_dn} holds at depth $d+1$.
\end{proof}

\subsection{Definitions: random models of trees}

We now define the models and random trees considered in the study. They are the same as models of Chapter \ref{chapter:MPAlign}, but we restate them here with an equivalent\footnote{Note that with this ‘unlabeled' view, there is no need to define a tree augmentation for the correlated model $\dPls_d$ -- as in previous Chapter.} definition that fits the description $t = \{N_{\tau} \}_{\tau \in \cX_{d-1}}$ of elements of $\cX_d$.

\begin{definition}[Galton-Watson trees with Poisson offspring]\label{def:GW_trees}
	Let $\mu>0$. For $d=0$, $\GWmu_d$ is the Dirac mass at the trivial tree $\bigcdot \in \cX_0$. For $d \geq 1$, a tree $t = \{N_{\tau} \}_{\tau \in \cX_{d-1}} \sim \GWmu_d$ is sampled as follows: for all $\tau \in \cX_{d-1}$, $N_{\tau} \sim \Poi(\mu \GWmu_{d-1}(\tau))$ independently from everything else. Note that since the Poisson variables are independent, we have
	\begin{equation*}
	\sum_{\tau \in \cX_{d-1}} N_{\tau} \sim \Poi\left( \mu \sum_{\tau \in \cX_{d-1}} \GWmu_{d-1}(\tau) \right) = \Poi(\mu)
	\end{equation*}
	which is a.s. finite. Hence, we have $t \in \cX_d$ a.s.
	
	Throughout all the study, we will only work with Galton-Watson trees with Poisson offspring, which we will simply refer to as Galton-Watson trees, as a notation shortcut.
\end{definition}

\begin{definition}[Null model $\dPl_d$]\label{def:null_model}
	The null distribution $\dPl_d$ on $\cX_d \times \cX_d$ of parameter $\lambda>0$ is simply defined as the product $\GWl_{d} \otimes \GWl_{d}$: under the null model, the two trees are independent Galton-Watson trees with offspring $\Poi(\lambda)$.
\end{definition}

\begin{definition}[Correlated model $\dPls_d$]\label{def:correlated_model}
	The correlated model $\dPls_d$ on $\cX_d \times \cX_d$ with parameters $\lambda>0$ and $s \in [0,1)$ first verifies that $\dPls_{0}$ is the same as $\dPl_0$. 
	
	For $d \geq 1$, a pair of trees $(t,t') = (\{N_{\tau} \}_{\tau \in \cX_{d-1}}, \{N'_{\tau} \}_{\tau \in \cX_{d-1}}) \sim \dPls_d$ is sampled as follows. 
	\begin{equation}\label{eq:def_model_1}
	N_{\tau} := M_\tau + \sum_{\tau' \in \cX_{d_1}} N_{\tau, \tau'} \quad \mbox{and} \quad
	N'_{\tau} := M'_{\tau} + \sum_{\tau \in \cX_{d-1}} N_{\tau, \tau'} \, ,
	\end{equation} with $M_{\tau}$ and $M'_{\tau}$ i.i.d. $\Poi(\lambda (1-s) \GWmu_{d-1}(\tau))$ and $N_{\tau, \tau'}$ i.i.d. $\Poi(\lambda s \dPls_{d-1}(\tau,\tau'))$ variables.  Note that for all $d$, $\dPls_d = \dPl_d$ if $s=0$.
	
\end{definition}

\subsection{Main new results}
We recall from Chapter \ref{chapter:MPAlign} that the likelihood ratio $L_d$ is defined as
\begin{equation*}
	L_d(t,t') := \frac{\dPls(t,t')}{\dPl(t,t')} \, ,
\end{equation*} as well as the following result
\begin{theorem*}[Chapter \ref{chapter:MPAlign}, Theorem \ref{MPAlign:theorem:main_result_TREES}]
	Let 
	\begin{equation*}
	\KL_d := \KL (\dPls_d\Vert\dPl_d)= \dE_{1,d} \left[ \log(L_d) \right].
	\end{equation*}
	Then there exists a one-sided test for testing $\dPl_d$ versus $\dPls_d$ if and only if $\underset{d \to \infty}{\lim} \KL_d = + \infty$ and $\lambda s >1$.
\end{theorem*}

We now state the main results of this addendum. Let $\alpha$ be the Otter constant introduced in Proposition \ref{addendum:prop:otter}.

\begin{theorem}[Negative result]\label{addendum:thm:negative_result}
	If $s \leq \sqrt{\alpha}$, then for all $\lambda>0$, $\limsup_{d} \KL(\dPls_d \, || \, \dPl_d) < \infty$. Hence, one-sided detection is impossible.
\end{theorem}

\begin{theorem}[Positive result]\label{addendum:thm:positive_result}
	If $s > \sqrt{\alpha}$, then there exists $\lambda(s)>0$ such that for all $\lambda \geq \lambda(s)$, $\KL(\dPls_d \, || \, \dPl_d) \underset{d \to \infty}{\longrightarrow} +\infty$. One-sided detection is thus feasible for $\lambda$ large enough, and an optimal one-sided test is the likelihood-ratio test $\cT_{d} := \one_{L_d(t,t')>\beta_d}$ for some appropriate $\beta_d$ (see Theorem \ref{MPAlign:theorem:main_result_TREES}).
\end{theorem}

\begin{remark}
	These new results establish an almost sharp result for correlation detection in trees. The regime $s \leq \sqrt{\alpha}$ always lies within the impossible phase, and $s > \sqrt{\alpha}$ is in the easy phase for high mean degree $\lambda$.
	
	In view of Theorem \ref{MPAlign:theorem:main_result_GRAPHS} proved in Chapter \ref{chapter:MPAlign}, these results also extend the knowledge on the analysis of \alg{MPAlign}, and doing so on the phase diagram for partial graph alignment, for which a state-of-the-art version is given in Figure \ref{fig:addendum:sota_phase_diagram} below. 
\end{remark}

\begin{figure}[h]
	\centering
	\includegraphics[scale=0.65]{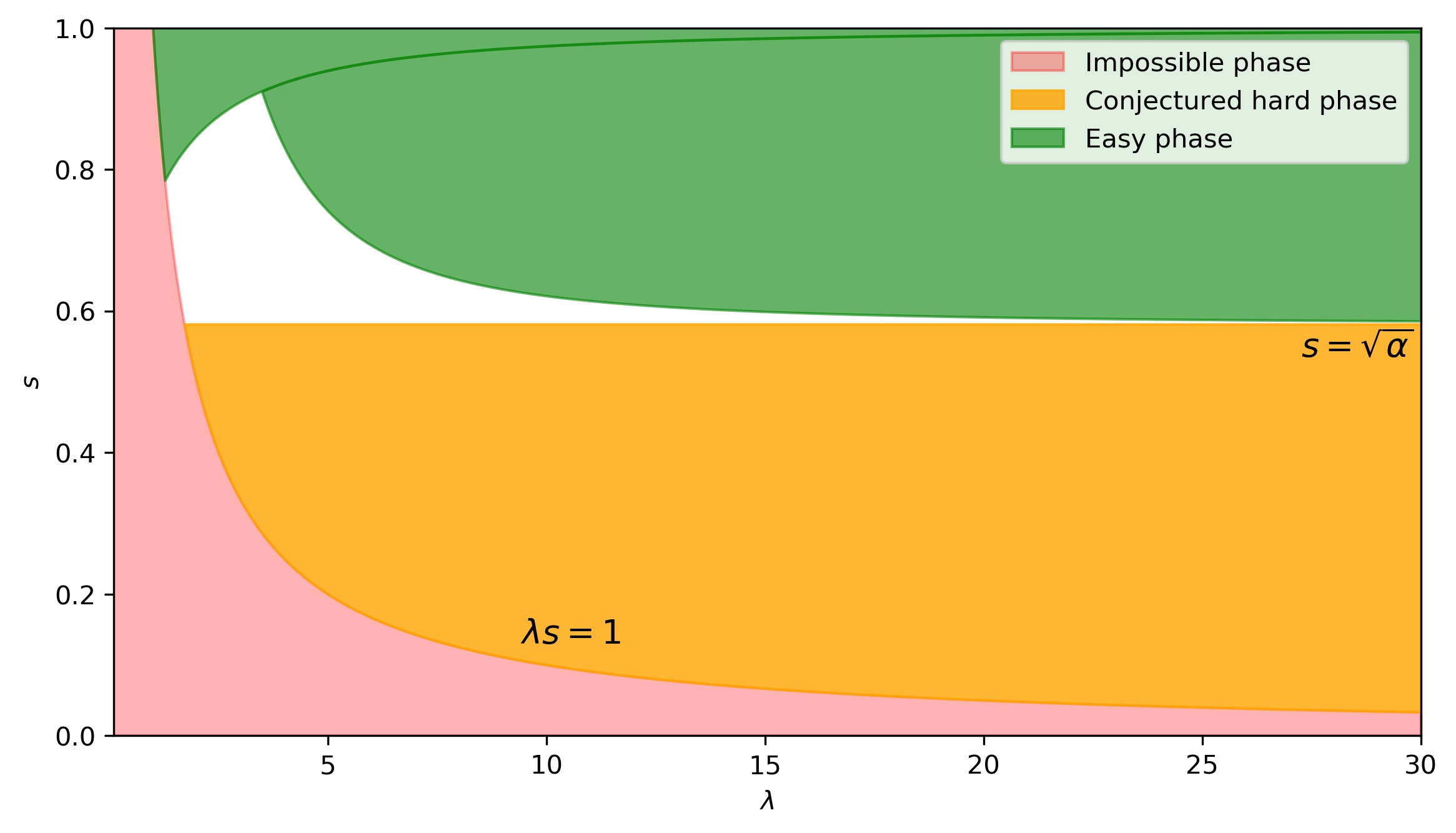}
	\caption{State-of-the art phase diagram for partial graph alignment.}
	\label{fig:addendum:sota_phase_diagram}
\end{figure}

Instrumental to the proofs of these results is the diagonalization of the likelihood ratio $L_d(t,t')$ in an orthogonal basis of eigenvectors (or eigenfunctions). A very noticeable result given in Theorem \ref{addendum:theorem:eigendecomposition} in the following section is that these eigenvectors (resp. eigenvalues) only depend on $\lambda$ (resp. on $s$) (!) 

\section{The impossible phase for $s \leq \sqrt{\alpha}$}

\subsection{Eigendecomposition of the likelihood ratio}

Let us start by stating the master result of this section, which will render the analysis easier and bring important corollaries for our analysis. 

\begin{theorem}[Eigendecomposition of the likelihood ratio]\label{addendum:theorem:eigendecomposition}
	For all $\lambda>0, d \geq 0$, there exists a collection $\{\fl_{d,\alpha} \}_{\alpha \in \cX_d}$ with $\fl_{d,\alpha} : \cX_d \to \dR$, such that for all $s \in [0,1)$, 
	
	\begin{equation}\label{eq:theorem:eigendecomposition}
	\forall t,t' \in \cX_d, \; L_d(t,t') = \sum_{\alpha \in \cX_d} s^{|\alpha| -1} \fl_{d,\alpha}(t) \fl_{d,\alpha}(t'), 
	\end{equation}
	Moreover, the $\fl_{d,\alpha}$ are independent of $s$ and verify the following properties:
	\begin{itemize}
		\item Value at the trivial tree: 
		\begin{equation}\label{eq:theorem:eigendecomposition_trivial_tree}
		\forall t \in \cX_d, \; \fl_{d, \bigcdot}(t) = 1 \, ,
		\end{equation}
		\item Orthogonality: 
		\begin{equation}\label{eq:theorem:eigendecomposition_orthogonality}
		\forall \alpha, \alpha' \in \cX_d, \; \sum_{t \in \cX_d} \GWl_d(t) \fl_{d, \alpha}(t) \fl_{d, \alpha'}(t) = \one_{\alpha = \alpha'} \, .
		\end{equation}
		\begin{equation}\label{eq:theorem:eigendecomposition_orthogonality_bis}
		\forall t,t' \in \cX_d, \; \GWl_d(t) \sum_{\alpha \in \cX_d}  \fl_{d, \alpha}(t) \fl_{d, \alpha'}(t') = \one_{t = t'} \, .
		\end{equation}
		\item Limit of higher-order mixed moments: more generally, for $n \geq 2$, $d \geq 1$ and $\beta^{(1)} = \{\beta^{(1)}_\alpha\}_{\alpha \in \cX_{d-1}}$, $\ldots, \beta^{(n)} = \{\beta^{(n)}_\alpha\}_{\alpha \in \cX_{d-1}} \in \cX_{d}$, one has
		\begin{multline}\label{eq:theorem:eigendecomposition_mixed_products}
		\sum_{t \in \cX_{d}} \GWl_{d}(t) \fl_{d, \beta^{(1)} }(t) \cdots \fl_{d, \beta^{(n)} }(t) \underset{\lambda \to \infty}{\longrightarrow} \prod_{\alpha \in \cX_{d-1}}\sqrt{\prod_{i=1}^n \beta^{(i)}_\alpha!} \left[ x^{\beta^{(1)}_\alpha}_1 \cdots \, x^{\beta^{(n)}_\alpha}_n \right] e^{\sum_{1 \leq i <j \leq n}x_i x_j} \, .
		\end{multline}
	\end{itemize}
\end{theorem}

\begin{remark}
	Note that in the above properties, \eqref{eq:theorem:eigendecomposition_orthogonality} implies the following first moment condition: 
	\begin{equation*}
	\forall \alpha \in \cX_d, \; \sum_{t \in \cX_d} \GWl_d(t) \fl_{d, \alpha}(t) = \one_{\alpha = \bigcdot}.
	\end{equation*}
\end{remark}

\begin{remark}
As remarked earlier, the eigenvectors $\fl_{d, \alpha}$ (resp. the eigenvalues) only depend on $\lambda$ and are independent of $s$ (resp. on $s$, independent of $\lambda$).
\end{remark}

\begin{remark}
	The $\fl_{d, \alpha}(\ell)$ for $d=1$, hence indexed by $\alpha, \ell \in \cX_1 \simeq \dN$, are given by 
	\begin{equation*}
	\fl_{1, \alpha}(\ell) := \sqrt{\alpha!} [x^\alpha] e^{-x \sqrt{\lambda}} \left( 1+\frac{x}{\sqrt{\lambda}}\right)^\ell \, ,
	\end{equation*} see equation \eqref{eq:fl1} in the proof. These functions are known as Charlier polynomials, which are orthogonal for the Poisson distribution. Theorem \ref{addendum:theorem:eigendecomposition} provides an extension of these polynomials, on trees of depth $d \geq 2$, that are orthogonal for the $\GWl_d$ distribution, consistent with $\GWl_1 \eqd \Poi(\lambda)$.

	Also, note that equation \eqref{eq:theorem:eigendecomposition} exhibits a duality between trees $t,t'$ in $\cX_d$ and the trees $\alpha \in \cX_d$. This duality turns out to be very helpful for analysis, as shown below, e.g. giving a nice space in which one can prove weak convergence results -- see Section \ref{subsection:gaussian_approx}.
\end{remark}

\begin{proof}[Proof of Theorem \ref{addendum:theorem:eigendecomposition}]
	
	We will prove the decomposition \eqref{eq:theorem:eigendecomposition} as well as the properties \eqref{eq:theorem:eigendecomposition_trivial_tree}, \eqref{eq:theorem:eigendecomposition_orthogonality}, \eqref{eq:theorem:eigendecomposition_mixed_products} by induction on $d$.\\
	
	\proofstep{Step 1: initialization at $d=1$.} We identify $\cX_1$ to $\dN$, and a tree $t$ of depth $d=1$ to its number of children $\ell \in \dN$: in this case, $|t|=\ell+1$. Denote by $\hatdPls_1$ the characteristic function defined on $[0,2 \pi]^2$ by 
	$\hatdPls_1(k,k') := \dE\left[ e^{ik\ell + ik'\ell'}\right]$ where $(\ell,\ell') \sim \dPls_1$. We have that
	\begin{flalign*}
	\hatdPls_1(k,k') & = \exp\left[ \lambda(1-s) (e^{ik} + e^{ik'} - 2) + \lambda s (e^{i(k+k')}-1) \right]\\
	& = e^{\lambda(e^{ik}-1)} e^{\lambda(e^{ik'}-1)} \exp\left[ \lambda s (e^{ik}-1) (e^{ik'}-1) \right]\\
	& = \sum_{\alpha \geq 0} s^{\alpha} \frac{\lambda^\alpha}{\alpha !} (e^{ik}-1)^\alpha e^{\lambda(e^{ik}-1)} (e^{ik'}-1)^\alpha e^{\lambda(e^{ik'}-1)} = \sum_{\alpha \geq 0} s^{\alpha} \hatgl_{1,\alpha}(k) \hatgl_{1,\alpha}(k') \, , 
	\end{flalign*} with 
	\begin{flalign*}
	\hatgl_{1,\alpha}(k) & := e^{- \lambda} \sqrt{\frac{\lambda^\alpha}{\alpha!}} e^{\lambda e^{ik}} (e^{ik}-1)^\alpha
	= e^{- \lambda} \sqrt{\alpha!} e^{\lambda e^{ik}} [x^\alpha] e^{x\sqrt{\lambda}(e^{ik}-1)}
	= e^{- \lambda} \sqrt{\alpha!}  [x^\alpha] e^{-x\sqrt{\lambda}} e^{(\lambda + x\sqrt{\lambda})e^{ik}} \, .
	\end{flalign*} 
	
	We have an easy upper bound of the form $\left| \hatgl_{1,\alpha}(k) \right| \leq \frac{C^\alpha}{\sqrt{\alpha!}}$, independently of $k$, which established normal convergence of the series $\hatdPls_1(k,k')$ in the above. Hence, inverting the Fourier transform, we get
	\begin{flalign*}
	\dPls_1(\ell,\ell') & = \int_{[0,2 \pi]^2} \frac{dk dk'}{(2 \pi)^2} e^{-ik\ell - i k' \ell'} \hatdPls_1(k,k') = \sum_{\alpha \geq 0} s^{|\alpha|-1} \gl_{1,\alpha}(\ell) \gl_{1,\alpha}(\ell') \, ,
	\end{flalign*} with 
	\begin{flalign*}
	\gl_{1,\alpha}(\ell) & := \int_{[0,2 \pi]} \frac{dk}{2 \pi} e^{-ik\ell} \hatgl_{1,\alpha}(k) \\
	& = e^{- \lambda} \sqrt{\alpha!}  [x^\alpha] e^{-x\sqrt{\lambda}} \int_{[0,2 \pi]} \frac{dk}{2 \pi} e^{-ik\ell}   e^{(\lambda + x\sqrt{\lambda})e^{ik}} \\
	& = e^{- \lambda} \sqrt{\alpha!}  [x^\alpha] e^{-x\sqrt{\lambda}} \frac{(\lambda + x\sqrt{\lambda})^\ell}{\ell!} \, .
	\end{flalign*} 
	We hence have that $L_{1}$ satisfies \eqref{eq:theorem:eigendecomposition} with $\fl_{1, \alpha}$ given by 
	\begin{flalign}\label{eq:fl1}
	\fl_{1, \alpha}(\ell) & = \frac{e^\lambda \ell!}{\lambda^\ell} \gl_{1,\alpha}(\ell) = \sqrt{\alpha!} [x^\alpha] e^{-x \sqrt{\lambda}} \left( 1+\frac{x}{\sqrt{\lambda}}\right)^\ell \, .
	\end{flalign} 
	Taking $\alpha = 0$ in \eqref{eq:fl1} gives $\fl_{1,\bigcdot} = 1$ and proves condition \eqref{eq:theorem:eigendecomposition_trivial_tree} at $d=1$. For orthogonality \eqref{eq:theorem:eigendecomposition_orthogonality}, note that for all $\alpha, \alpha' \in \dN$,
	\begin{flalign*}
	& \sum_{\ell \geq 0} \GWl_1(\ell) \fl_{1, \alpha}(\ell) \fl_{1, \alpha'}(\ell) = \sqrt{\alpha! \alpha'!} \sum_{\ell \geq 0} e^{- \lambda} \frac{\lambda^\ell}{\ell!}  [x^\alpha y^{\alpha'}] e^{-x \sqrt{\lambda} -y \sqrt{\lambda}} \left[\left( 1+\frac{x}{\sqrt{\lambda}}\right)\left( 1+\frac{y}{\sqrt{\lambda}}\right) \right]^\ell\\
	& \quad \quad \quad \quad = \sqrt{\alpha! \alpha'!} [x^\alpha y^{\alpha'}] e^{-x \sqrt{\lambda} -y \sqrt{\lambda}}  \exp \left[ - \lambda + \lambda\left( 1+\frac{x}{\sqrt{\lambda}}\right)\left( 1+\frac{y}{\sqrt{\lambda}}\right) \right] \\
	& \quad \quad \quad \quad = \sqrt{\alpha! \alpha'!} [x^\alpha y^{\alpha'}] e^{xy} = \one_{\alpha=\alpha'}\, ,
	\end{flalign*} which establishes \eqref{eq:theorem:eigendecomposition_orthogonality} for $d=1$. Previous computations are made rigorous by noticing that the series in \eqref{eq:fl1} has infinite radius of convergence, and appealing to Fubini's theorem. With the same arguments, writing
	\begin{flalign}\label{eq:fl1_bis}
	\fl_{1, \alpha}(\ell) & = \sqrt{\alpha!} \ell! [x^\alpha y^\ell] e^{-x \sqrt{\lambda} + y(1+x/\sqrt{\lambda})} \nonumber \\
	& = \sqrt{\alpha!} \ell! [x^\alpha y^\ell] e^{y + x (y/\sqrt{\lambda} -\sqrt{\lambda})} \nonumber \\
	& = \frac{\ell!}{\sqrt{\alpha!}}  [x^\ell] e^x \left(\frac{x}{\sqrt{\lambda}} -\sqrt{\lambda}\right)^\alpha  \, .
	\end{flalign} 
	Taking $\alpha = 0$ in \eqref{eq:fl1} gives $\fl_{1,\bigcdot} = 1$ and proves condition \eqref{eq:theorem:eigendecomposition_trivial_tree} at $d=1$. For orthogonality \eqref{eq:theorem:eigendecomposition_orthogonality}, note that for all $\alpha, \alpha' \in \dN$,
	\begin{flalign*}
	\sum_{\alpha \geq 0} \fl_{1, \alpha}(\ell) \fl_{1, \alpha}(\ell') & = \ell! (\ell')! [x^\ell y^{\ell'}] e^{x+y} \sum_{\alpha \geq 0} \frac{1}{\alpha!} \left(\frac{x}{\sqrt{\lambda}} -\sqrt{\lambda}\right)^\alpha \left(\frac{y}{\sqrt{\lambda}} -\sqrt{\lambda}\right)^\alpha\\
	& = \ell! (\ell')! [x^\ell y^{\ell'}] e^{xy/\lambda+ \lambda} = \one_{\ell=\ell'} \ell! e^{\lambda} \lambda^{-\ell} = \frac{\one_{\ell=\ell'}}{\GWl_{1}(\ell)} \, ,
	\end{flalign*} which establishes \eqref{eq:theorem:eigendecomposition_orthogonality_bis} for $d=1$.
	
	More generally, for $n \geq 2$, 
	\begin{flalign}\label{eq:fl2}
	& \sum_{\ell \geq 0} \GWl_{1}(\ell) \fl_{1, \alpha_1 }(\ell) \cdots \fl_{1, \alpha_n }(\ell) = \sqrt{\prod_{i=1}^{n} \alpha_i!} \, \sum_{\ell \geq 0} e^{- \lambda} \frac{\lambda^\ell}{\ell!}  [x_1^{\alpha_1} \, \cdots \, x_n^{\alpha_n}] e^{-\sqrt{\lambda} \sum_{i=1}^{n}x_i}  \prod_{i=1}^{n} \left(1+\frac{x_i}{\sqrt{\lambda}}\right)^\ell \nonumber \\
	& \quad \quad \quad \quad  = \sqrt{\prod_{i=1}^{n} \alpha_i!} \, [x_1^{\alpha_1} \, \cdots \, x_n^{\alpha_n}] \exp \left[- \lambda -\sqrt{\lambda} \sum_{i=1}^{n}x_i  + \lambda \prod_{i=1}^{n} \left(1+\frac{x_i}{\sqrt{\lambda}}\right) \right] \nonumber \\
	& \quad \quad \quad \quad = \sqrt{\prod_{i=1}^{n} \alpha_i!} \, [x_1^{\alpha_1} \, \cdots \, x_n^{\alpha_n}] \exp \left[ \sum_{1 \leq i < j \leq n} x_i x_j + \eps_{\lambda}(x_1,\ldots,x_n) \right] \, ,
	\end{flalign}with
	$$
	\eps_{\lambda}(x_1,\ldots,x_n) := \sum_{p=3}^{n} \lambda^{1-p/2} \sum_{1 \leq i_1 < \ldots < i_p \leq n} x_{i_1}^{\alpha_{i_1}} \cdots x_{i_p}^{\alpha_{i_p}}. 
	$$
	The terms corresponding to $[x_1^{\alpha_1} \, \cdots \, x_n^{\alpha_n}]$ in the expansion of $\exp \left[ \sum_{1 \leq i < j \leq n} x_i x_j + \eps_{\lambda}(x_1,\ldots,x_n) \right]$ 
	to which $\eps_{\lambda}(x_1,\ldots,x_n) $ contributes are in finite number (independently of $\lambda$) and are all of order $O(\lambda^{-1/2})$. Hence, taking $\lambda \to \infty$, property \eqref{eq:theorem:eigendecomposition_mixed_products} is proved for $d=1$ in \eqref{eq:fl2}.
	\\
	
	\proofstep{Step 2: recursion at $d+1$.} Let us take a pair of random trees in $\cX_{d+1}$ sampled from the correlated model given in Definition \ref{def:correlated_model}, with $N,N' \in \dR^{\cX_d}$ their corresponding vector representations. Given $k,y \in \dR^{\cX_d}$ we shall write $k \cdot y := \sum_{\alpha \in \cX_d} k_{\alpha} y_{\alpha}$. The characteristic function of $\dPls_{d+1}$ is defined as $\hatdPls_{d+1}(k,k') := \dE\left[ e^{ik \cdot N + ik' \cdot N'}\right]$ and writes
	\begin{flalign}\label{eq:recursive_eigen_step1}
	& \hatdPls_{d+1}(k,k') = \exp \left[ \lambda(1-s) \sum_{t \in \cX_d} \GWl_{d}(t)(e^{i k_t} + e^{i k'_t} - 2) + \lambda s \sum_{t,t' \in \cX_d} \dPls_{d}(t,t') (e^{i k_t + i k'_{t'}} -1) \right] \nonumber \\
	& \quad \quad = e^{\lambda \sum_{t \in \cX_d} \GWl_{d}(t)(e^{i k_t} -1) + \lambda \sum_{t \in \cX_d} \GWl_{d}(t)(e^{i k'_t} -1)} \exp \left[ \lambda s \sum_{t,t' \in \cX_d} \dPls_{d}(t,t') (e^{i k_t} -1)(e^{i k'_{t'}} -1) \right] \nonumber \\
	& \quad \quad = e^{\lambda \sum_{t \in \cX_d} \GWl_{d}(t)(e^{i k_t} -1) + \lambda \sum_{t \in \cX_d} \GWl_{d}(t)(e^{i k'_t} -1)} \nonumber \\
	& \quad \quad \quad \quad \quad  \times  \sum_{n \geq 0}s^n \frac{\lambda^n}{n!} \underbrace{\left( \sum_{t,t' \in \cX_d} \dPls_{d}(t,t') (e^{i k_t} -1)(e^{i k'_{t'}} -1)\right)^n}_{(i)}
	\, .
	\end{flalign} Let us use the decomposition \eqref{eq:theorem:eigendecomposition} at step $d$ in $(i)$. Denoting $\gl_{d,\alpha}(t) := \fl_{d,\alpha}(t) \GWl_d(t)$, this gives 
	\begin{flalign*}
	(i) & =\left(\sum_{\alpha \in \cX_d} s^{|\alpha| -1} \left[ \sum_{t  \in \cX_d}\gl_{d,\alpha}(t)(e^{i k_t} -1) \right] \left[ \sum_{t  \in \cX_d}\gl_{d,\alpha}(t) (e^{i k'_{t}} -1)\right]\right)^n
	\, \\
	& = \sum_{\beta = (\beta_\alpha)_{\alpha \in \cX_d}} {n!} s^{- n + \sum_{\alpha \in \cX_d}\beta_{\alpha}|\alpha| } \\
	& \quad \quad \quad\quad \quad \times \prod_{\alpha \in \cX_d} \frac{1}{\beta_{\alpha}!} \left[ \sum_{t  \in \cX_d}\gl_{d,\alpha}(t)(e^{i k_t} -1) \right]^{\beta_\alpha} \left[ \sum_{t  \in \cX_d}\gl_{d,\alpha}(t) (e^{i k'_{t}} -1)\right]^{\beta_\alpha} \one_{\sum_{\alpha \in \cX_d}\beta_{\alpha} = n} \,.
	\end{flalign*} Summing $(i)$ for $n \geq 0$ gives an overall sum over all $\beta = (\beta_\alpha)_{\alpha \in \cX_d}$ that is over all $\cX_{d+1}$. Moreover, for $\beta = (\beta_\alpha)_{\alpha \in \cX_d} \in \cX_{d+1}$, one always has $$|\beta| = 1+ \sum_{\alpha \in \cX_d}\beta_{\alpha}|\alpha| \,. $$ 
	Hence, equation \eqref{eq:recursive_eigen_step1} becomes
	\begin{flalign*}
	\hatdPls_{d+1}(k,k') & = \sum_{\substack{\beta \in \cX_{d+1} \\ \beta = (\beta_\alpha)_{\alpha \in \cX_d}}} s^{|\beta|-1} \prod_{\alpha \in \cX_d} \frac{1}{\beta_{\alpha}!} \left( \lambda \left[ \sum_{t  \in \cX_d}\gl_{d,\alpha}(t)(e^{i k_t} -1) \right] \left[ \sum_{t  \in \cX_d}\gl_{d,\alpha}(t) (e^{i k'_{t}} -1)\right]\right)^{\beta_\alpha} \nonumber \\
	& \quad \quad \quad \quad \quad \times  e^{\lambda \sum_{t \in \cX_d} \GWl_{d}(t)(e^{i k_t} -1) + \lambda \sum_{t \in \cX_d} \GWl_{d}(t)(e^{i k'_t} -1)} \nonumber \\
	& = \sum_{\substack{\beta \in \cX_{d+1} \\ \beta = (\beta_\alpha)_{\alpha \in \cX_d}}} s^{|\beta|-1} \hatgl_{d+1,\beta}(k) \hatgl_{d+1,\beta}(k')\,, 
	\end{flalign*} with 
	\begin{flalign*}
	\hatgl_{d+1,\beta}(k) & := e^{\lambda \sum_{t \in \cX_d} \GWl_{d}(t)(e^{i k_t} -1)}  \prod_{\alpha \in \cX_d}  \frac{1}{\sqrt{\beta_{\alpha}!}}  \left[ \sqrt{\lambda} \sum_{t  \in \cX_d}\gl_{d,\alpha}(t)(e^{i k_t} -1)\right]^{\beta_\alpha} \nonumber \\
	& = e^{- \lambda} \sqrt{\prod_\alpha \beta_\alpha!} \, [x^\beta] e^{ \lambda \sum_{t \in \cX_d} \GWl_{d}(t) e^{i k_t} + \sum_{\alpha \in \cX_d} x_{\alpha} \sqrt{\lambda} \sum_{t  \in \cX_d}\gl_{d,\alpha}(t)(e^{i k_t} -1)} \nonumber \\
	& = e^{- \lambda} \sqrt{\prod_\alpha \beta_\alpha!} \, [x^\beta] e^{ -  \sqrt{\lambda} \sum_{\alpha,t \in \cX_d} x_{\alpha} \gl_{d,\alpha}(t) + \sum_{t} e^{i k_t} \left[\lambda \GWl_{d}(t) + \sum_{\alpha} x_\alpha \sqrt{\lambda} \gl_{d,\alpha}(t) \right]} \, ,
	\end{flalign*} where $x = \set{x_\alpha}_{\alpha \in \cX_d}$ is a family of formal variables and $x^{\beta}$ denotes $\prod_\alpha x_{\alpha}^{\beta_\alpha}$ when $\beta = (\beta_\alpha)_{\alpha \in \cX_d}$. Note that since the trees are finite, only a finite number of coordinates $\beta_\alpha$ are non zero, which makes the infinite product problem disappear. The same arguments of normal convergence as in the case $d=1$ apply to justify the integral/sum permutations.
	
	As done in Step 1, we can invert the Fourier transform by integrating over every $k_t$, which gives  
	
	\begin{flalign*}
	\gl_{d+1,\beta}(N) = e^{- \lambda} \sqrt{\prod_\alpha \beta_\alpha!} \, [x^\beta] e^{ -  \sqrt{\lambda} \sum_{\alpha,t \in \cX_d} x_{\alpha} \gl_{d,\alpha}(t)} \prod_{t \in \cX_d} \frac{\left[\lambda \GWl_{d}(t) + \sum_{\alpha} x_\alpha \sqrt{\lambda} \gl_{d,\alpha}(t) \right]^{N_t}}{N_t!} \, .
	\end{flalign*} It is now established that $L_{d+1}(N,N')$ satisfies the decomposition \eqref{eq:theorem:eigendecomposition} with $\fl_{d+1, \beta}$ given by the following recursion
	\begin{equation}\label{eq:recursive_expression_f}
	\fl_{d+1,\beta}(N) := \sqrt{\prod_\alpha \beta_\alpha!} \, [x^\beta] e^{ -  \sqrt{\lambda} \sum_{\alpha,t \in \cX_d} x_{\alpha} \gl_{d,\alpha}(t)} \prod_{t \in \cX_d} \left( 1 + \sum_{\alpha \in \cX_d} \frac{x_\alpha}{\sqrt{\lambda}}  \fl_{d,\alpha}(t) \right)^{N_t} \, .   
	\end{equation}
	Taking $\beta = \bigcdot$ in \eqref{eq:recursive_expression_f}, that is $\beta_\alpha = 0$ for all $\alpha$, gives $\fl_{d+1,\bigcdot} = 1$ and proves condition \eqref{eq:theorem:eigendecomposition_trivial_tree} at $d+1$. 
	
	\proofstep{Step 2.1: recursion for \eqref{eq:theorem:eigendecomposition_orthogonality} at $d+1$.} For any $\beta = \{\beta_\alpha\}_{\alpha \in \cX_{d}}$,$\beta' = \{\beta'_\alpha\}_{\alpha \in \cX_{d}} \in \cX_{d+1}$, recursion \eqref{eq:recursive_expression_f} gives
	\begin{multline*}
	\sum_{N \in \cX_{d+1}} \GWl_{d+1}(N) \fl_{d+1, \beta }(N)  \fl_{d+1, \beta' }(N) = \sqrt{\prod_\alpha \beta_\alpha!} \sqrt{\prod_\alpha \beta'_\alpha!} \\ \times \left[ x^\beta y^{\beta'} \right]e^ { - \lambda - \sqrt{\lambda} \sum_{\alpha,t  \in \cX_d} (x_{\alpha}+y_{\alpha}) \gl_{d,\alpha}(t) + \lambda \sum_{t \in \cX_d} \GWl_d(t) \left(1 + \sum_{\alpha \in \cX_d} \frac{x_{\alpha}}{\sqrt{\lambda}} \fl_{d,\alpha}(t) \right) \left(1 + \sum_{\alpha \in \cX_d} \frac{y_{\alpha}}{\sqrt{\lambda}} \fl_{d,\alpha}(t) \right) } \, .
	\end{multline*} The expression in the exponential in the above factorizes in several terms, that of order $\lambda$ being $-1+1=0$. The term in $\sqrt{\lambda}$ is
	$$ - \sum_{\alpha,t  \in \cX_d} (x_{\alpha}+y_{\alpha})\gl_{d,\alpha}(t) + \sum_{\alpha,t  \in \cX_d} \GWl_d(t) x_{\alpha} \fl_{d,\alpha}(t) + \sum_{\alpha,t  \in \cX_d} \GWl_d(t) y_{\alpha} \fl_{d,\alpha}(t) = 0 \, ,$$ since $\gl_{d,\alpha}(t)=\GWl_d(t) \fl_{d,\alpha}(t)$ by definition. The only remaining term is constant and evaluates to 
	$$ \sum_{\alpha,\alpha'} x_{\alpha} y _{\alpha'} \sum_{t \in \cX_d} \GWl_d(t) \fl_{d,\alpha}(t) \fl_{d,\alpha'}(t) = \sum_{\alpha} x_{\alpha} y _{\alpha} \, ,$$ using the orthogonality property \eqref{eq:theorem:eigendecomposition_orthogonality} at step $d$. Hence,
	\begin{flalign*}
	\sum_{N \in \cX_{d+1}} \GWl_{d+1}(N) \fl_{d+1, \beta }(N)  \fl_{d+1, \beta' }(N) & = \sqrt{\prod_\alpha \beta_\alpha!} \sqrt{\prod_\alpha \beta'_\alpha!} \left[ x^\beta y^{\beta'} \right] e^ {\sum_{\alpha} x_{\alpha} y _{\alpha}} 
	 = \one_{\beta=\beta'} \, ,
	\end{flalign*} which proves \eqref{eq:theorem:eigendecomposition_orthogonality} at $d+1$. Previous computations are made rigorous since the trees are finite, and by noticing that the series in \eqref{eq:recursive_expression_f} has infinite radius of convergence, and appealing to Fubini's theorem. We use the same arguments to make computations rigorous in the rest of the proof.
	
	\proofstep{Step 2.2: recursion for \eqref{eq:theorem:eigendecomposition_orthogonality_bis} at $d+1$.} We are going to transform equation \eqref{eq:recursive_expression_f} as for the step $d=1$, as follows:
	\begin{flalign}\label{eq:recursive_expression_f_bis}
	\fl_{d+1,\beta}(N) & = \sqrt{\prod_\alpha \beta_\alpha!} \prod_t N_t! \, [x^\beta y^N] e^{ -  \sqrt{\lambda} \sum_{\alpha,t \in \cX_d} x_{\alpha} \gl_{d,\alpha}(t) + \sum_{t \in \cX_d} y_t + \sum_{\alpha,t \in \cX_d}  \frac{x_\alpha y_\alpha}{\sqrt{\lambda}}  \fl_{d,\alpha}(t) } \nonumber \\
	& =  \frac{\prod_t N_t! }{\sqrt{\prod_\alpha \beta_\alpha!}} \, [x^N] e^{\sum_t x_t} \prod_{\alpha} \left( \sum_t \fl_{d,\alpha}(t) \left(\frac{x_t}{\sqrt{\lambda}} - \sqrt{\lambda} \GWl_d(t) \right)^{\beta_\alpha} \right) \, .
	\end{flalign} Using \eqref{eq:recursive_expression_f_bis} gives that for all $N,N' \in \cX_{d+1}$,
	\begin{flalign*}
		& \sum_{\beta \in \cX_{d+1}} \fl_{d+1, \beta }(N)\fl_{d+1, \beta }(N') = \prod_{t} N_t! N'_t! \, [x^N y^{N'}] e^{\sum_t (x_t + y_t)}\\
		& \quad \quad  \quad \quad \quad \quad \times e^{\sum_{\alpha,t,t'} \fl_{d,\alpha}(t) \fl_{d,\alpha}(t') \left(\frac{x_t}{\sqrt{\lambda}} - \sqrt{\lambda} \GWl_d(t) \right) \left(\frac{y_{t'}}{\sqrt{\lambda}} - \sqrt{\lambda} \GWl_d(t') \right)} \\
		& \quad \quad = \prod_{t} N_t! N'_t! \, [x^N y^{N'}] e^{\sum_t (x_t + y_t) + \sum_{t} \frac{1}{\GWl_d(t)}  \left(\frac{x_t}{\sqrt{\lambda}} - \sqrt{\lambda} \GWl_d(t) \right)\left(\frac{y_t}{\sqrt{\lambda}} - \sqrt{\lambda} \GWl_d(t) \right)}  \, ,
	\end{flalign*} where we used \eqref{eq:theorem:eigendecomposition_orthogonality_bis} at step $d$ in the last step. This simplifies to 
	\begin{flalign*}
	\sum_{\beta \in \cX_{d+1}} \fl_{d+1, \beta}(N)\fl_{d+1, \beta }(N') & = \prod_{t} N_t! N'_t! \, [x^N y^{N'}] \, e^{\lambda \GWl_{d}(t) + \sum_t \frac{x_t y_t}{\lambda \GWl_d(t)}} \\
	& = \prod_t \one_{N=N'} e^{\lambda \GWl_{d}(t)} (\lambda \GWl_{d}(t))^{-N_t} N_t! = \frac{\one_{N=N'}}{\GWl_{d+1}(N)} \, ,
	\end{flalign*} which proves \eqref{eq:theorem:eigendecomposition_orthogonality_bis} at step $d+1$.
	
	\proofstep{Step 2.3: recursion for \eqref{eq:theorem:eigendecomposition_mixed_products} at $d+1$.} 
	Let us now prove property \eqref{eq:theorem:eigendecomposition_mixed_products}. For any $\beta^{(1)} = \{\beta^{(1)}_\alpha\}_{\alpha \in \cX_{d}}$, $\ldots, \beta^{(n)} = \{\beta^{(n)}_\alpha\}_{\alpha \in \cX_{d}} \in \cX_{d+1}$, recursion \eqref{eq:recursive_expression_f} gives
	\begin{multline*}
	\sum_{N \in \cX_{d+1}} \GWl_{d+1}(N) \fl_{d+1, \beta^{(1)} }(N) \cdots \fl_{d+1, \beta^{(n)} }(N) = \sqrt{\prod_{i=1}^{n} \prod_\alpha \beta_\alpha!} \, \left[ \prod_{i=1}^{n} (x^{(i)})^{\beta^{(i)}} \right] \\ \times \exp \left[ - \lambda - \sqrt{\lambda} \sum_{\alpha,t  \in \cX_d} \sum_{i=1}^{n} x_{\alpha}^{(i)} \gl_{d,\alpha}(t) + \lambda \sum_{t \in \cX_d} \GWl_d(t) \prod_{i=1}^{n} \left(1 + \sum_{\alpha \in \cX_d} \frac{x_{\alpha}^{(i)}}{\sqrt{\lambda}} \fl_{d,\alpha}(t) \right) \right] \, .
	\end{multline*}As in Step 2.1, when expanding the product in the exponential, the zero and first order terms simplify, which yields
	\begin{multline}\label{eq:recursive_mult_step1}
	\sum_{N \in \cX_{d+1}} \GWl_{d+1}(N) \fl_{d+1, \beta^{(1)} }(N) \cdots \fl_{d+1, \beta^{(n)} }(N) = \\ \sqrt{\prod_{i=1}^{n} \prod_\alpha \beta_\alpha!} \, \left[ \prod_{i=1}^{n} (x^{(i)})^{\beta^{(i)}} \right] 
	e^{\sum_{1 \leq i < j \leq n} \sum_{\alpha,\alpha' \in \cX_d} x_{\alpha}^{(i)} x_{\alpha'}^{(j)} \sum_{t \in \cX_d} \GWl_d(t) \fl_{d,\alpha}(t) \fl_{d,\alpha'}(t) + \eps_{\lambda}(x^{(1)}, \ldots, x^{(n)}) } \, ,
	\end{multline} with 
	$$
	\eps_{\lambda}(x^{(1)}, \ldots, x^{(n)}) := \sum_{p=3}^{n} \lambda^{1-p/2} \sum_{1 \leq i_1 < \ldots < i_p \leq n} \sum_{\alpha_1, \ldots, \alpha_p \in \cX_d} x^{i_1}_{\alpha_{i_1}} \cdots x^{i_p}_{\alpha_{i_p}} \sum_{t \in \cX_d} \GWl_d(t) \fl_{d,\alpha_1}(t) \cdots \fl_{d,\alpha_p}(t) \, . 
	$$
	Using orthogonality \eqref{eq:theorem:eigendecomposition_orthogonality} at step $d$, \eqref{eq:recursive_mult_step1} writes
	\begin{multline}\label{eq:recursive_mult_step2}
	\sum_{N \in \cX_{d+1}} \GWl_{d+1}(N) \fl_{d+1, \beta^{(1)} }(N) \cdots \fl_{d+1, \beta^{(n)} }(N) =  \\ \sqrt{\prod_{i=1}^{n} \prod_\alpha \beta_\alpha!} \, \left[ \prod_{i=1}^{n} (x^{(i)})^{\beta^{(i)}} \right] 
	e^{\sum_{1 \leq i < j \leq n} \sum_{\alpha \in \cX_d} x_{\alpha}^{(i)} x_{\alpha}^{(j)}} 
	\times \exp\left[\eps_{\lambda}(x^{(1)}, \ldots, x^{(n)}) \right] \, ,
	\end{multline} 
	Using property \eqref{eq:theorem:eigendecomposition_mixed_products} at step $d$, $\sum_{t \in \cX_d} \GWl_d(t) \fl_{d,\alpha_1}(t) \cdots \fl_{d,\alpha_p}(t)$ 
	has a finite limit when $\lambda \to \infty$. Hence, as in Step $1$, the terms corresponding to $\left[ \prod_{i=1}^{n} (x^{(i)})^{\beta^{(i)}} \right]$ in \eqref{eq:recursive_mult_step2} 
	to which $\eps_{\lambda}(x^{(1)}, \ldots, x^{(n)})$ contributes are in finite number (independent of $\lambda$) and are all of order $O(\lambda^{-1/2})$. 
	
	Taking $\lambda \to \infty$ thus establishes property \eqref{eq:theorem:eigendecomposition_mixed_products} for $d+1$ and completes the proof of Theorem \ref{addendum:theorem:eigendecomposition}.
\end{proof}

\subsection{Computation of cyclic moments, proof of Theorem \ref{addendum:thm:negative_result}}

Theorem \ref{addendum:theorem:eigendecomposition} hereabove has a very natural corollary that enables to compute the cyclic moments of the likelihood ratio.

\begin{corollary}[Cyclic moments]\label{corr:cyclic_moments}
	The \emph{$n-$th cyclic moment} of $L_d$ is defined as follows
	\begin{equation*}\label{eq:def:cyclic_moments}
	\Cls_{d,m} := \dEl_d\left[L_d(T_1,T_2)\cdots L_d(T_{m-1},T_m) L_d(T_m,T_1)\right],
	\end{equation*} where $T_1, \ldots, T_m$ are i.i.d. $\GWl_d$ in the above expectation. One has that
	\begin{equation}\label{eq:corr:cyclic_moments}
	\Cls_{d,m} = \sum_{\alpha \in \cX_d} (s^m)^{\card{\alpha}-1} = \sum_{n \geq 1} A_{d,n}(s^m)^{n-1} = \Phi_d(s^m),
	\end{equation} where $A_{d,n}$, as defined in \eqref{eq:def:Adn}, denotes the number of unlabeled trees with $n$ vertices of depth at most $d$, and $\Phi_d$ is the generating function defined in Proposition \ref{addendum:prop:A_dn}. Note that in particular, the $\Cls_{d,m}$ do not depend on $\lambda$ (!) and by Proposition \ref{addendum:prop:A_dn} they are upper bounded for each $d$ and $s \in [0,1)$ by some constant $A=A(d,s)$. We thus denote $\Cs_{d,m} := \Cls_{d,m}$ in the sequel.
\end{corollary}

\begin{proof}[Proof of Corollary \ref{corr:cyclic_moments}]
	By Theorem \ref{addendum:theorem:eigendecomposition} we have 
	$$ L_d(t,t') = \sum_{\alpha \in \cX_d} s^{\card{\alpha} -1} \fl_{d,\alpha}(t) \fl_{d,\alpha}(t') \, ,  $$
	hence, setting $\alpha_{m+1} = \alpha_1$,
	\begin{flalign*}
	\Cls_{d,m} & = \dEl_d\left[L_d(T_1,T_2)\cdots L_d(T_{m-1},T_m) L_d(T_m,T_1)\right] \\
	& = \sum_{\alpha_1, \ldots, \alpha_m \in \cX_d} s^{\sum_{i=1}^m (|\alpha_i| -1)} \dEl_d \left[\prod_{i=1}^{m} \fl_{d,\alpha_i}(T_i) \fl_{d,\alpha_{i+1}}(T_i)\right] \\
	& = \sum_{\alpha_1, \ldots, \alpha_m \in \cX_d} s^{\sum_{i=1}^m (|\alpha_i| -1)} \prod_{i=1}^{m} \dEl_d \left[\fl_{d,\alpha_i}(T) \fl_{d,\alpha_{i+1}}(T)\right] \\
	& = \sum_{\alpha_1, \ldots, \alpha_m \in \cX_d} s^{\sum_{i=1}^m (|\alpha_i| -1)} \one_{\alpha_1 = \, \ldots \, = \alpha_m}\\
	& = \sum_{\alpha \in \cX_d} (s^m)^{|\alpha|-1} \, .
	\end{flalign*}
	All steps in the above computations are legitimate by Fubini's theorem, since $$\dEl_d \left[\left|\fl_{d,\alpha}(T) \fl_{d,\alpha'}(T)\right|\right] \leq \dEl_d \left[(\fl_{d,\alpha}(T))^2\right]^{1/2} \dEl_d \left[ (\fl_{d,\alpha'}(T))^2\right]^{1/2} = 1 \, ,$$ by property \eqref{eq:theorem:eigendecomposition_orthogonality_bis} of Theorem \ref{addendum:theorem:eigendecomposition}.
\end{proof}

We are now ready to give a proof of Theorem \ref{addendum:thm:negative_result}.

\begin{proof}[Proof of Theorem \ref{addendum:thm:negative_result}]
	According to Corollary \ref{corr:cyclic_moments}, one has 
	\begin{equation}\label{eq:proof:thm:negative_result}
	\dEl_d\left[ L_d (T,T')^2\right] = \Cs_{d,2} = \sum_{n \geq 1} A_{d,n} s^{2(n-1)}.
	\end{equation}
	Moreover, since $A_{d,n} \leq A_n$ (by Definition \ref{def:An_Adn}) and $A_n \underset{n \to \infty}{\sim} \frac{C}{n^{3/2}} \left( \frac{1}{\alpha} \right)^n$ by Proposition \ref{addendum:prop:otter}, the assumption $s \leq \sqrt{\alpha}$ ensures that $\dEl_d\left[ L_d (T,T')^2\right] = \sum_{n \geq 1} A_{d,n} s^{2(n-1)} \leq \sum_{n \geq 1} A_{n} s^{2(n-1)} < \infty$, uniformly in $d$.
	
	Then, applying Jensen's inequality yields
	\begin{equation*}
	\KL(\dPls_d \, || \, \dPl_d) = \dEls_d\left[ \log L_d (T,T')\right] \leq \log \dEls_d\left[ L_d (T,T')\right] = \log \dEl_d\left[ L^2_d (T,T')\right] < \infty,
	\end{equation*} uniformly in $d$, and concludes the proof.
\end{proof}

For the positive result, we need to study the weak convergence of the likelihood ratio when $\lambda \to \infty$, which is the scope of the next Section, concluded by the proof of Theorem \ref{addendum:thm:positive_result}.

\section{The high-degree regime: positive result when $s > \sqrt{\alpha}$ in the gaussian approximation}

In view of definition \ref{def:correlated_model}, we recall that a pair of correlated trees $(t,t')$ of depth at most $d+1$ sampled from $\dPls_{d+1}$ are of the form $t = \set{N_\tau}_{\tau \in \cX_{d}}$ and $t' = \set{N'_\tau}_{\tau \in \cX_{d}}$ with 
\begin{equation}\label{eq:NN'M}
N_{\tau} := M_\tau + \sum_{\tau' \in \cX_{d}} N_{\tau, \tau'} \quad \mbox{and} \quad
N'_{\tau} := M'_{\tau} + \sum_{\tau \in \cX_{d}} N_{\tau, \tau'} \,.
\end{equation} with
\begin{equation*}
M_{\tau}, M'_{\tau} \overset{\mathrm{i.i.d.}}{\sim}\Poi(\lambda (1-s) \GWl_{d}(\tau)) \quad \mbox{and} \quad N_{\tau, \tau'} \sim \Poi(\lambda s \dPls_{d}(\tau,\tau')) \, .
\end{equation*}

\subsection{Gaussian approximation in the high-degree regime}\label{subsection:gaussian_approx}

Let us define $y = (y_\alpha)_{\alpha \in \cX_d}$  and $y' = (y'_{\alpha})_{\alpha \in \cX_d}$ as follows:
\begin{flalign}
y_\alpha := \frac{1}{\sqrt{\lambda}} \sum_{\tau \in \cX_{d}}\fl_{d,\alpha}(\tau) (N_\tau - \lambda \GWl_{d}(\tau)) \label{eq:y_alpha} \\
y'_\alpha := \frac{1}{\sqrt{\lambda}} \sum_{\tau \in \cX_{d}}\fl_{d,\alpha}(\tau) (N'_\tau - \lambda \GWl_{d}(\tau)) \label{eq:y'_alpha} 
\end{flalign}
where the $\fl_{d,\alpha}$ are defined in Theorem \ref{addendum:theorem:eigendecomposition}.
In other words, $y$ (resp. $y'$) is a centered version of $N$ (resp. $N'$), projected onto the basis of eigenvectors.

Let $(z,z') = ((z_\alpha)_{\alpha \in \cX_d},(z'_{\alpha'})_{\alpha' \in \cX_d}))$ be an (infinite-dimensional) centered Gaussian vector defined by its covariance matrix: 
\begin{equation}\label{eq:covariance_z}
\forall \alpha, \alpha' \in \cX_d, \quad \dE[z_{\alpha}z_{\alpha'}] = \dE[z'_{\alpha}z'_{\alpha'}] = \one_{\alpha = \alpha'}, \quad \dE[z_{\alpha}z'_{\alpha'}] = s^{|\alpha|}\one_{\alpha = \alpha'}.
\end{equation}

Let us denote by $\dpls_{d+1}$ the joint distribution of $(y,y')$, and $\gwl_{d+1}$ the marginal distribution of $y$ (or $y'$). Since the transformations $N \to y$ in \eqref{eq:y_alpha} and $N' \to y'$ in \eqref{eq:y'_alpha} are bijective in view of the orthogonality property  \eqref{eq:theorem:eigendecomposition_orthogonality_bis} in Theorem \ref{addendum:theorem:eigendecomposition}, one has

\begin{equation}\label{eq:equality_of_KL_centered}
\KL(\dPls_{d+1} \| \GWl_{d+1} \otimes \GWl_{d+1}) = \KL(\dpls_{d+1} \| \gwl_{d+1} \otimes \gwl_{d+1})
\end{equation}

\begin{lemma}\label{lemma:cv_weak_dual}
	When $\lambda \to \infty$, we have the following convergence in distribution:
	\begin{equation}\label{eq:lemma:cv_weak_dual}
	(y,y') \overset{\mathrm{(d)}}{\longrightarrow} (z,z').
	\end{equation}
\end{lemma}

\begin{proof}[Proof of Lemma \ref{lemma:cv_weak_dual}]
	With the canonical product sigma-field, convergence in distribution of $(y,y')$ amounts to convergence of all finite-dimensional distributions. Let us denote by $(k,k')$ a pair of real vectors in $\dR^{\cX_d \times \cX_d}$ such that only a finite number of entries are non-zero. We shall write $k \cdot y := \sum_{\alpha \in \cX_d} k_{\alpha} y_{\alpha}$. We will also define the following characteristic functions:
	\begin{equation}
	\hatdpls(k,k') := \dE\left[ e^{ik \cdot y + i k' \cdot y'} \right] \quad \mbox{ and } \quad
	\hatrs(k,k') := \dE\left[ e^{ik \cdot z + i k' \cdot z'} \right] \, .
	\end{equation}
	Proving Lemma \ref{lemma:cv_weak_dual} thus amounts to showing the simple convergence $\hatdpls(k,k') \to \hatrs(k,k')$ when $\lambda \to \infty$. Since the (gaussian) limit distribution is entirely determined by its moments, it suffices to show the convergence of the cumulants \cite{Janson00}. The covariance structure of $(z,z')$ given in \eqref{eq:covariance_z} immediately yields
	\begin{equation}\label{eq:cumulant_zz'}
	\hatrs(k,k') = \exp \left[-\frac{1}{2} \sum_{\alpha \in \cX_d} ((k_\alpha)^2 + (k'_\alpha)^2 + 2 s^{|\alpha|}k_\alpha k'_\alpha) \right] \, .
	\end{equation}
	
	Then, in view of \eqref{eq:NN'M}, \eqref{eq:y_alpha} and \eqref{eq:y'_alpha}, writing $\fl_{d}(\tau) := (\fl_{d,\alpha}(\tau))_{\alpha \in \cX_d}$, one has
	\begin{multline*}
	e^{ik \cdot y + i k' \cdot y'} = \exp\left[ - \sqrt{\lambda} \sum_{\tau \in \cX_d} \GWl_d(\tau) (i k \cdot \fl_{d}(\tau) + ik'  \cdot \fl_{d}(\tau)) \right] \\
	\times \prod_{\tau, \tau' \in \cX_d} \left( \exp \left[\frac{1}{\sqrt{\lambda}}  (ik \cdot \fl_{d}(\tau) + i k' \cdot \fl_{d}(\tau')) \right]\right)^{N_{\tau,\tau'}} \times \prod_{\tau \in \cX_d} \left( \exp \left[\frac{1}{\sqrt{\lambda}}  ik \cdot \fl_{d}(\tau) \right]\right)^{M_{\tau}} \\
	\times \prod_{\tau \in \cX_d} \left( \exp \left[\frac{1}{\sqrt{\lambda}}  ik' \cdot \fl_{d}(\tau) \right]\right)^{M'_{\tau}}.
	\end{multline*} Variables $N_{\tau,\tau'}, M_\tau, M'_\tau$ being independent Poisson variables, taking the expectation gives
	\begin{multline*}
	\hatdpls(k,k')  = \exp\left[ - \sqrt{\lambda} \sum_{\tau \in \cX_d} \GWl_d(\tau) (i k \cdot \fl_{d}(\tau) + ik'  \cdot \fl_{d}(\tau)) \right] \\
	\times \exp\left[ \lambda(1-s) \sum_{\tau \in \cX_d} \GWl_d(\tau) \left( e^{\frac{1}{\sqrt{\lambda}} i k \cdot \fl_{d}(\tau)} + e^{\frac{1}{\sqrt{\lambda}} i k' \cdot \fl_{d}(\tau)} - 2 \right) \right]\\
	\times \exp\left[ \lambda s \sum_{\tau, \tau' \in \cX_d} \dPls_{d}(\tau,\tau') \left( e^{\frac{1}{\sqrt{\lambda}} (i k \cdot \fl_{d}(\tau)+ i k' \cdot \fl_{d}(\tau'))} - 1 \right) \right] \, .
	\end{multline*}
	The cumulants are obtained by expanding the logarithm of the last expression in power series in $k,k'$. Using that $\sum_{\tau' \in \cX_d} \dPls_{d}(\tau,\tau') = \GWl_d(\tau)$, the first-order (linear) terms compensate to $0$, which is coherent with the fact that $\dE[y_\alpha]=\dE[y'_\alpha]=0$. The second-order terms in $\log \hatdpls(k,k')$ evaluate to 
	\begin{flalign*}
	& - \lambda (1-s) \sum_{\tau \in \cX_d} \GWl(\tau) \frac{1}{2 \lambda} \sum_{\alpha,\alpha' \in \cX_d} \fl_{d,\alpha}(\tau) \fl_{d,\alpha'}(\tau) \left( k_\alpha k_{\alpha'} + k'_\alpha k'_{\alpha'} \right)\\
	& - \lambda s \sum_{\tau, \tau' \in \cX_d} \dPls_{d}(\tau,\tau') \\
	& \quad \quad \quad \quad \times \frac{1}{2 \lambda} \sum_{\alpha,\alpha' \in \cX_d} \left( \fl_{d,\alpha}(\tau) \fl_{d,\alpha'}(\tau) k_\alpha k_{\alpha'} + \fl_{d,\alpha}(\tau') \fl_{d,\alpha'}(\tau') k'_\alpha k'_{\alpha'} + 2 \fl_{d,\alpha}(\tau) \fl_{d,\alpha'}(\tau') k_\alpha k'_{\alpha'} \right)
	\, .
	\end{flalign*} Using the orthogonality property of the eigenvectors in Theorem \ref{addendum:theorem:eigendecomposition}, namely that  $$\forall \alpha, \alpha' \in \cX_d, \quad \sum_{\tau \in \cX_d} \GWl_d(\tau) \fl_{d, \alpha}(\tau) \fl_{d, \alpha'}(\tau) = \one_{\alpha = \alpha'} \, ,$$ the previous equation simplifies to
	\begin{equation*}
	- \frac{1}{2} \sum_{\alpha \in \cX_d} \left( (k_\alpha)^2 + (k'_\alpha)^2 \right) 
	- s \sum_{\tau, \tau' \in \cX_d} \dPls_{d}(\tau,\tau') \sum_{\alpha,\alpha' \in \cX_d} \fl_{d,\alpha}(\tau) \fl_{d,\alpha'}(\tau') k_\alpha k'_{\alpha'} 
	\, ,
	\end{equation*} which in turn writes, using $\dPls_{d}(\tau,\tau') = \GWmu_{d}(\tau) \GWmu_{d}(\tau') \sum_{\alpha \in \cX_d} s^{|\alpha| -1} \fl_{d,\alpha}(\tau) \fl_{d,\alpha}(\tau')$
	\begin{flalign*}
	& - \frac{1}{2} \sum_{\alpha \in \cX_d} \left( (k_\alpha)^2 + (k'_\alpha)^2 \right) - s  \sum_{\alpha,\alpha',\alpha'' \in \cX_d} s^{|\alpha''| -1} k_\alpha k'_{\alpha'} \\
	& \quad \quad \quad \quad \quad \quad  \times \left( \sum_{\tau \in \cX_d} \GWmu_{d}(\tau) \fl_{d,\alpha}(\tau) \fl_{d,\alpha''}(\tau)  \right) \left( \sum_{\tau' \in \cX_d} \GWmu_{d}(\tau') \fl_{d,\alpha'}(\tau') \fl_{d,\alpha''}(\tau')  \right)\\
	& = - \frac{1}{2} \sum_{\alpha \in \cX_d}  \left( (k_\alpha)^2 + (k'_\alpha)^2 + 2 s^{|\alpha|} k_\alpha k'_{\alpha'} \right) 
	\, ,
	\end{flalign*} which is exactly the second cumulant of $(z,z')$ in \eqref{eq:cumulant_zz'}. The remaining step is to show that the higher order cumulants tend to $0$ when $\lambda$ gets large. Note that the cumulants of order $\geq 3$ have a factor $1/\sqrt{\lambda}$, but the implicit dependence of the $\fl_{d,\alpha}$ needs to be controlled. The previous computations show that the cumulants depend on terms of the form 
	$$ \sum_{t \in \cX_d} \GWl_d(t) \fl_{d,\alpha_1}(t) \cdots \fl_{d,\alpha_p}(t), $$ which are proved to remain finite when $\lambda \to \infty$ by property \eqref{eq:theorem:eigendecomposition_mixed_products} of Theorem \ref{addendum:theorem:eigendecomposition}. This shows that all the cumulants of order $\geq 3$ tend to $0$ and hence establishes the desired convergence in distribution.
\end{proof}

\subsection{Weak limit of likelihood-ratio, limit of the Kullback-Leibler divergence}
In view of the previous weak convergence established in Lemma \ref{lemma:cv_weak_dual}, we will now prove the following result, which compares the $\KL-$divergence with finite $\lambda$ to the $\KL-$divergence between the limiting Gaussian distributions of Lemma \ref{lemma:cv_weak_dual}. 
\begin{proposition}\label{addendum:prop:gaussian_KL_limit}
	Denoting
	\begin{equation}\label{eq:theorem:def_KL}
	\KLls_d := \KL(\dPls_d || \dPl_d) \, ,
	\end{equation} one has the following:
	\begin{equation}\label{eq:theorem:gaussian_KL_limit} 
	\forall d \geq 1, \; \liminf_{\lambda \to\infty} \KLls_d \geq \KLs_d := - \frac{1}{2}\sum_{\alpha \in \cX_{d-1}} \log(1-s^{2 |\alpha|}) = \frac{1}{2} \log\Cs_{d,2} \, .
	\end{equation}  We recall that
	\begin{equation}\label{eq:rappel:KLs}
	\Cs_{d,2} = \dEls_d[L_d] = \sum_{n \geq 1} A_{d,n}(s^2)^{n-1}  \, .
	\end{equation}
\end{proposition}

\begin{proof}[Proof of Theorem \ref{addendum:prop:gaussian_KL_limit}]
	Fix $d \geq 1$. In \eqref{eq:equality_of_KL_centered}, we established that $\KLls_d$ is also the $\KL$-divergence $\KL(\dpls_{d} \| \gwl_{d} \otimes \gwl_{d})$ where $\dpls_{d}$ is the distribution of $(y,y')$ defined in Section \ref{subsection:gaussian_approx}. Moreover, Lemma \ref{lemma:cv_weak_dual} establishes that $(y,y')$ converges in distribution a centered gaussian vector $(z,z')$ defined by its covariance matrix:
	\begin{equation}\label{eq:covariance_z_bis}
	\forall \alpha, \alpha' \in \cX_{d-1}, \quad \dE[z_{\alpha}z_{\alpha'}] = \dE[z'_{\alpha}z'_{\alpha'}] = \one_{\alpha = \alpha'}, \quad \dE[z_{\alpha}z'_{\alpha'}] = s^{|\alpha|}\one_{\alpha = \alpha'}.
	\end{equation}
	If we denote by $p^{(s)}_1$ the joint distribution of the gaussian vector $(z,z')$ and $p^{(s)}_0$ the product of the marginals, the KL-divergence $\KL(p^{(s)}_1 || p^{(s)}_0)$ is easily given by $-\frac{1}{2}\log\det\Sigma$, where $\Sigma$ is the covariance matrix of $(z,z')$, which is similar to a matrix with diagonal blocks of the form $\begin{pmatrix} 1 & s^{|\alpha|} \\ s^{|\alpha|} & 1\end{pmatrix}$ for all $\alpha \in \cX_{d-1}$, which gives
	\begin{flalign*}
	\KL(p^{(s)}_1 || p^{(s)}_0) = - \frac{1}{2} \log \prod_{\alpha \in \cX_{d-1}}  \frac{1}{1-s^{2|\alpha|}} \, .
	\end{flalign*} The last term is indeed $\KLs_d$ as defined in \eqref{eq:theorem:gaussian_KL_limit}, since
	\begin{flalign*}
	\dEls_d[L_d]  = \sum_{\beta \in\cX_{d}} s^{2 |\beta|-1}  = \prod_{\alpha \in \cX_{d-1}} \sum_{\beta_\alpha \geq 0} s^{2 {\beta_\alpha} |\alpha|} = \prod_{\alpha \in \cX_{d-1}}  \frac{1}{1-s^{2|\alpha|}} \, .
	\end{flalign*}
	
	The roof is concluded by appealing to the lower semi-continuity property of the $\KL$-divergence (see e.g. \cite{PolyanskiyLecturenotes}, Theorem 3.6), namely that 
	\begin{equation*}
	\liminf_{\lambda \to\infty} \KLls_d \geq \KLs_d \, .
	\end{equation*}
\end{proof}

\subsection{Propagating bounds on the $\KL-$divergence, proof of Theorem \ref{addendum:thm:positive_result}}
The goal of this section is to use the result of Proposition \ref{addendum:prop:gaussian_KL_limit} and the fact that in view of \eqref{eq:rappel:KLs} for $s>\sqrt{\alpha}$ (where $\alpha$ is Otter's constant), $\KLs_d \to +\infty$ with $d$, in order to obtain Theorem \ref{addendum:thm:positive_result}, that is that for fixed $s>\sqrt{\alpha}$, there exists $\lambda=\lambda(s)$ such that:
\begin{equation*}
\lambda\geq \lambda(s)\Rightarrow \lim_{d\to\infty} \KLls_d=+\infty \, .
\end{equation*}

The following Lemma shows that if $s > \sqrt{\alpha}$, for any small (resp. any large but bounded) probability that we fix, there exists a depth $d_0$ and an event $S$ that has this small (resp. large) probability under $\dPl_{d_0}$ (resp. $\dPls_{d_0}$). The proof is deferred to Appendix \ref{addendum:appendix:proof:lemma:initialize_lower_bound_on_KL}.

\begin{lemma}\label{lemma:initialize_lower_bound_on_KL}
	Assume that $s> \sqrt{\alpha}$. Then for any $c\in (0,1)$  such that $\eps \in (0,1)$
	$$
	c< \frac{1}{15} \, ,
	$$
	and  any $\eps\in (0,1)$, there exists $\lambda_1=\lambda_1(s,c,\eps)>0$ and $d_0 = d_0(s,c,\eps)\in\dN$ such that, for all $\lambda\geq \lambda_1$, there exists an event $S =S(s,c,\eps) \subset \cX_{d_0}^2$ for which the following inequalities hold:
	\begin{equation*}
	\dPls_{d_0}( S)\geq c \quad \mbox{and} \quad \dPl_{d_0}(S)\leq \eps \, .
	\end{equation*}
\end{lemma}

Now that we know that this event $S$ exists at a certain initial depth $d_0$, we want to propagate the bounds for arbitrary depth $d \geq d_0$. This is the object of the following Proposition, proved in Appendix \ref{addendum:appendix:proof:prop:bounding_LR_by_induction}.

\begin{proposition}\label{prop:bounding_LR_by_induction}
	For any fixed  $c\in(0,1)$ there exist constants $\eps=\eps(s,c) \in (0,1)$ and $\lambda_0 = \lambda_0(s,c)>0$ such that the following holds. For any $\lambda\geq \lambda_0$, any $d\in\dN$, if there exists an event $S\subset \cX_d^2$ such that
	\begin{equation*}
	\dPl_d(S)\leq \eps \quad \mbox{and} \quad \dPls_d(S)\geq c \, ,
	\end{equation*} then there exists an event $S'\subset \cX_{d+1}^2$ such that 
	\begin{equation*}
	\dPl_{d+1}(S')\leq \frac{1}{2}\dPl_{d}(S) \leq \frac{\eps}{2}  \quad \mbox{and} \quad \dPls_{d+1}(S')\geq c \, .
	\end{equation*}
	
	In fact, using the usual notations $t=\{N_\tau\}_{\tau\in\cX_d}$, $t'=\{N'_\tau\}_{\tau\in\cX_d}$ for elements of $\cX_{d+1}$, and denoting, for all $\tau\in \cX_d$
	\begin{equation*}
	\widetilde{N}_\tau:=N_\tau -\lambda \GWl_d(\tau) \quad \mbox{and} \quad \widetilde{N}'_\tau= N'_\tau-\lambda \GWl_d(\tau)\, ,
	\end{equation*}
	the event $S'$ in the above is defined from $S$ in the following way 
	\begin{equation*}
	S'=\set{Z_S \geq \sigma} \, ,
	\end{equation*} where
	\begin{equation}\label{eq:def_Z}
	Z_S :=\sum_{(\tau,\tau')\in S}\widetilde{N}_\tau \widetilde{N}'_{\tau'} \, ,
	\end{equation} and for some suitable threshold $\sigma = \sigma(S)$.
\end{proposition}

Together, Lemma \ref{lemma:initialize_lower_bound_on_KL} and Proposition \ref{prop:bounding_LR_by_induction} yield the proof of Theorem \ref{addendum:thm:positive_result}.

\subsubsection{Proof of Theorem \ref{addendum:thm:positive_result}}
\begin{proof}[Proof of Theorem \ref{addendum:thm:positive_result}]
	Assume that $s > \sqrt{\alpha}$.
	Choose $c\in(0,1)$, $c\leq 1/[16\log(16)+15]$ and let $\eps=\eps(s,c)$, $\lambda_0=\lambda_0(s,c)$ be the corresponding quantities from Proposition \ref{prop:bounding_LR_by_induction}. Now that $c,\eps$ are fixed, we appeal to  Lemma \ref{lemma:initialize_lower_bound_on_KL} to obtain some $\lambda_1=\lambda_1(s,c,\eps)$ and $d_0 = d_0(s,c,\eps) \in\dN$ such that, taking $\lambda\geq \lambda_0 \vee \lambda_1$, there exists some event $S_d\subset\cX_d^2$ such that
	\begin{equation*}
	\dPl_d(S)\leq \eps \quad \mbox{and} \quad \dPls_d(S)\geq c \, .
	\end{equation*} Proposition \ref{prop:bounding_LR_by_induction} then ensures the existence of a sequence of events $S_d\subset \cX_d^2$, $d>d_0$ such that
	\begin{equation*}
	\dPl_d(S_d)\leq 2^{-(d-d_0)}\eps \quad \mbox{and} \quad \dPls_d(S_d) \geq c.
	\end{equation*}
	It follows that, for all $d > d_0$,
	\begin{flalign*}
	\KLls_d &\geq \dPls_d(S_d)\log\left(\frac{\dPls_d(S_d)}{\dPl_d(S_d)}\right)+(1-\dPls_d(S_d))\log\left(\frac{1-\dPls_d(S_d)}{1-\dPl_d(S_d)}\right)
	\\
	& \geq c \log( 2^{d-d_0}/\eps)-h(\dPls_d(S_d)) - (1-\dPls_d(S_d))\log((1-\dPl_d(S_d))) \\
	& \geq c \log( 2^{d-d_0}/\eps)-h(\dPls_d(S_d)) \, ,
	\end{flalign*} where for $x \in [0,1]$, $h$ is defined by $h(x) :=  -x \log(x) - (1-x) \log(1-x)$. Function $h$ is maximal at $x=1/2$ and $h(1/2)=\log(2)$, which gives the final bound $\KLls_d \geq c \log(2) (d-d_0) - c\log(\eps) - \log(2)$. It readily follows that $\lim_{d\to\infty}\KL_{\lambda,d}=+\infty$.
\end{proof}

\begin{subappendices}
\addtocontents{toc}{\protect\setcounter{tocdepth}{0}}
\section{Postponed proofs}
\subsection{Proof of Lemma \ref{lemma:initialize_lower_bound_on_KL}}\label{addendum:appendix:proof:lemma:initialize_lower_bound_on_KL}
\begin{proof}[Proof of Lemma \ref{lemma:initialize_lower_bound_on_KL}]
	Fix $c< 1/15$ and $\eps \in (0,1)$. Since $s> \sqrt{\alpha}$, we have that $\KLs_{d} \to \infty$ when $d \to \infty$, in view of \eqref{eq:rappel:KLs}, the fact that $A_{d,n} \to A_n$ when $d \to \infty$, and Otter's formula \ref{eq:theorem:otter}. For arbitrarily large $C = C(c,\eps)$ to be specified later, we can thus choose $d_0 = d_0(s,c,\eps)$ such that $\KLs_{d_0} = \frac{1}{2}\log(\Cs_{d_0,2}) \geq C$. 
	
	In turn, in view of \eqref{eq:theorem:gaussian_KL_limit}, we can choose $\lambda_1 = \lambda_1(s,c,\eps)$ such that
	$$
	\lambda\geq \lambda_1\Rightarrow \KLls_{d_0}\geq \frac{1}{2}\KLs_{d_0} \geq C/2 \, .
	$$
	
	Write then
	\begin{equation*}
	\frac{1}{4}\log(\Cs_{d_0,2})\leq \KLls_{d_0}\leq \int_1^{\infty}\log(x)\dPls_{d_0}(L_{d_0}\in dx) =  \int_{1}^{\infty}\frac{1}{u}\dPls_{d_0}(L_{d_0}\geq u)du.
	\end{equation*}
	Also, since $\Cs_{d_0,2}= \dEls_{d_0}[L_{d_0}]$,
	\begin{equation*}
	\Cs_{d_0,2}=\int_0^{\infty} x\dPls_{d_0}(L_{d_0}\in dx)\geq \int_1^{\infty} \dPls_{d_0}(L_{d_0}\geq u)du.
	\end{equation*}
	Now, for any $A,B>1$, $A<B$ we then have
	\begin{flalign*}
	\frac{1}{4}\log(\Cs_{d_0,2})&\leq \int_1^{A}\frac{1}{u}du+\dPls_{d_0}(L_{d_0}\geq A)\int_A^B\frac{du}{u}+\int_B^\infty\frac{1}{B}\dPls_{d_0}(L_{d_0}\geq u)du\\
	&\leq \log(A)+\dPls_{d_0}(L_{d_0}\geq A)\log(B/A)+\frac{1}{B}\Cs_{d_0,2} \, .
	\end{flalign*} This yields
	\begin{equation*}
	\dPls_{d_0}(L_{d_0}\geq A)\geq \frac{\frac{1}{4}\log(\Cs_{d_0,2})-\log(A) -\frac{1}{B}\Cs_{d_0,2}}{\log(B/A)} \, .
	\end{equation*}
	Choose $A=(\Cs_{d_0,2})^{1/16}$, $B=16 \Cs_{d_0,2}/\log(\Cs_{d_0,2})$. This yields
	\begin{equation*}
	\dPls_{d_0}(L_{d_0}\geq A)\geq \frac{1}{8}\frac{\log(\Cs_{d_0,2})}{\log(16)+(1-1/16)\log(\Cs_{d_0,2})-\log(\log(\Cs_{d_0,2}))} \, .
	\end{equation*}
	Recall that $d_0$ is taken such that $\frac{1}{2}\log(\Cs_{d_0,2})\geq C$, where $C$ is some constant as large as we want. The right-hand side being equivalent to $c_\infty := \frac{2}{15}$ when $C \to \infty$, and $c_\infty>c$ by definition of $c$. By Markov's inequality we also have $\dPl_{d_0}(L_{d_0}\geq A) \leq A^{-1} \leq \exp(-C/8)$, we can thus choose $C=C(\eps,c)$ such that 
	\begin{equation*}
	\dPls_{d_0}( \set{L_{d_0}\geq A})\geq c \quad \mbox{and} \quad \dPl_{d_0}(\set{L_{d_0}\geq A})\leq \eps \, .
	\end{equation*}
	The claimed result is proved with $S=S(s,c,\eps)=\set{L_{d_0}\geq A}$.
\end{proof}

\subsection{Proof of Proposition \ref{prop:bounding_LR_by_induction}} \label{addendum:appendix:proof:prop:bounding_LR_by_induction}
The proof of Proposition \ref{prop:bounding_LR_by_induction} relies on the following lemma. 
\begin{lemma}\label{lemma_bounds_Z}
	Assume $\lambda\geq 1$. The random variable $Z := Z_S$ defined in \eqref{eq:def_Z} verifies the following:
	\begin{itemize}
		\item[$(i)$] $\dEl_{d+1}[Z] = 0 \, .$
		\item[$(ii)$] $\dEls_{d+1}[Z] = \lambda s \dPls_{d}(S)\, .$
		\item[$(iii)$] $\dEl_{d+1}[Z^4] \leq 36 \lambda^4\dPl_d(S)^2+ 13\lambda^3 \dPl_d(S)\, .$
		\item[$(iv)$] $\Varls_{d+1}[Z] \leq \dEls_{d+1}[Z]+\lambda^2(1+s^2)\dPl_d(S)\, .$
	\end{itemize}
\end{lemma}
\begin{proof}[Proof of Lemma \ref{lemma_bounds_Z}]
	Recall the definition of $Z=Z_S$:
	\begin{equation*}
	Z_S :=\sum_{(\tau,\tau')\in S}\widetilde{N}_\tau \widetilde{N}'_{\tau'} \, ,
	\end{equation*} where
	\begin{equation*}
	\widetilde{N}_\tau=N_\tau -\lambda \GWl_d(\tau) \quad \mbox{and} \quad \widetilde{N}'_\tau= N'_\tau-\lambda \GWl_d(\tau)\, .
	\end{equation*}
	
	\proofstep{Point $(i)$} is immediate because under $\dPl_{d+1}$, for each pair $(\tau,\tau')\in\cX_d^2$,  the random variables $\widetilde{N}_\tau$, $\widetilde{N}'_{\tau'}$ are independent and zero mean.\\
	
	\proofstep{Point $(ii)$.} Recall that under $\dPls_{d+1}$, $N$ and $N'$ are sampled as follows:
	\begin{equation*}
	N_{\tau} = \Delta_\tau + \sum_{\theta' \in \cX_{d}} M_{\tau, \theta'} \quad \mbox{and} \quad
	N'_{\tau} = \Delta'_{\tau} + \sum_{\theta \in \cX_{d}} M_{\theta, \tau'} \, ,
	\end{equation*} with $\Delta_{\tau}$ and $\Delta'_{\tau}$ i.i.d. $\Poi(\lambda (1-s) \GWl_{d}(\tau))$ and $M_{\theta, \theta'}$ i.i.d. $\Poi(\lambda s \dPls_{d}(\theta, \theta'))$ variables. Introduce the notations:
	\begin{equation*}
	\widetilde{\Delta}_\tau:=\Delta_\tau-\lambda (1-s) \GWl_d(\tau), \quad  \widetilde{\Delta}'_{\tau'}:=\Delta'_{\tau'}-\lambda (1-s) \GWl_d(\tau'), \quad \widetilde{M}_{\theta,\theta'}=M_{\theta,\theta'}-\lambda s \dPls_d(\theta,\theta') \, .
	\end{equation*}
	Since the marginals of $\dPls_d$ are given by $\GWl_d$, it holds that 
	\begin{equation*}
	\widetilde{N}_\tau=\widetilde{\Delta}_\tau+\sum_{\theta' \in \cX_{d}}\widetilde{M}_{\tau,\theta'} \quad \mbox{and} \quad \widetilde{N}'_{\tau'}=\widetilde{\Delta}'_{\tau'}+\sum_{\theta \in \cX_{d}} \widetilde{M}_{\theta,\tau'} \, ,
	\end{equation*} which shows that 
	\begin{equation*}
	\dEls_{d+1}[ \widetilde{N}_\tau \widetilde{N}'_{\tau'}]=\Varls_{d+1}( M_{\tau,\tau'})=\lambda s \dPls_d(\tau,\tau') \,.
	\end{equation*} Point $(ii)$  follows.\\
	
	\proofstep{Point $(iii)$.} Write
	\begin{equation*}
	\dEl_{d+1}[Z^4] = \sum_{\substack{(\tau_1,\tau'_1) \in S, (\tau_2,\tau'_2) \in S \\ (\tau_3,\tau'_3) \in S (\tau_4,\tau'_4) \in S}} \dEl_{d+1}\left[ \prod_{i =1}^{4} \widetilde{N}_{\tau_i}\right] \dEl_{d+1}\left[ \prod_{i =1}^{4} \widetilde{N}'_{\tau'_i}\right] \, .
	\end{equation*} The only non-zero terms in the above summation are such that:
	\begin{equation*}
	|\set{\tau_i , i\in[4]}| \in \set{1,2} \quad \mbox{and} \quad |\set{\tau'_i , i\in[4]}| \in \set{1,2} \, .
	\end{equation*} We let $\sigma(u,v)$ denote the summation of terms with $|\set{\tau_i,i\in[4]}|=u$, $|\set{\tau'_i,i\in[4]}|=v$, for $u,v\in \set{1,2}$.
	
	We have $\sigma(1,1)=\sum_{(\tau,\tau')\in S}\dEl_{d+1}[\widetilde{N}^4_{\tau}] \dEl_{d+1} [\widetilde{N}'^4_{\tau'}]$. We use the following 
	\begin{lemma}\label{lemma:4th_moment_centered_poisson}
		If $X \sim \Poi(\mu)$ then $\dE[(X-\mu)^4]= 3 \mu^2 +\mu.$
	\end{lemma}
	Lemma \ref{lemma:4th_moment_centered_poisson} implies, using the fact that $\GWl_d(\tau),\GWl_d(\tau')\leq 1$, that
	\begin{flalign*}
	\sigma(1,1) & = \sum_{(\tau,\tau')\in S}\left[3\lambda^2 \GWl_d(\tau)^2+\lambda \GWl_d(\tau)\right]\left[3\lambda^2 \GWl_d(\tau')^2+\lambda \GWl_d(\tau')\right]\\
	&\leq 9\lambda^4\sum_{(\tau,\tau')\in S} \GWl_d(\tau)^2 \GWl_d(\tau')^2+[6\lambda^3+\lambda^2]\dPl_{d}(S) 
	\leq 9\lambda^4 \dPl_{d}(S)^2+ 7\lambda^3 \dPl_d(S) \,
	\end{flalign*} since $\lambda \geq 1$ which we shall assume, and using the easy bound $\sum_i x^2_i \leq \left( \sum_i x_i \right)^2$ for positive $x_i$.
	
	The term $\sigma(1,2)$ verifies
	\begin{flalign*}
	\sigma(1,2) & \leq  3\sum_{\tau} \dEl_{d+1}[\widetilde{N}^4_\tau]\sum_{\substack{\tau':(\tau,\tau')\in S \\ \theta':(\tau,\theta')\in S}}\dEl_{d+1}[\widetilde{N}'^2_{\tau'}]\dEl_{d+1}[\widetilde{N}'^2_{\theta'}]\\
	&= 3\sum_\tau \left[3\lambda^2 \GWl_d(\tau)^2+\lambda \GWl_d(\tau)\right] \sum_{\substack{\tau':(\tau,\tau')\in S \\ \theta':(\tau,\theta')\in S}}\lambda^2 \GWl_d(\tau') \GWl_d(\theta')\\
	&\leq   9 \lambda^4\sum_\tau \GWl_d(\tau)^2\sum_{\substack{\tau':(\tau,\tau')\in S \\ \theta':(\tau,\theta')\in S}}\GWl_d(\tau')\GWl_d(\theta')+3\lambda^3\sum_{(\tau,\tau')\in S}\GWl_d(\tau)\GWl_d(\tau') \, ,
	\end{flalign*} where we used the fact that $\sum_{\theta':(\tau,\theta')\in S}\GWl_d(\theta')\leq 1$. Note now that 
	\begin{equation*}
	\sum_{\substack{\tau':(\tau,\tau')\in S \\ \theta':(\tau,\theta')\in S}}\GWl_d(\tau')\GWl_d(\theta')\leq \dPl_d(S)^2
	\end{equation*}to conclude that $\sigma(1,2)\leq 9\lambda^4 \dPl_d(S)^2+3\lambda^3 \dPl_d(S)$, and the same bound also holds for $\sigma(2,1)$.\\
	
	Finally, $\sigma(2,2)$ can be bounded as follows. Having fixed $\tau_1$, there must be one index $j\in\{2,3,4\}$ such that $\tau_j=\tau_1$. Consider thus that $j=3$ and $\tau_4=\tau_2$. By symmetry, when accounting only for this case, we just need to multiply our evaluation by 3. This leads to the following bound:
	\begin{flalign*}
	\sigma(2,2)&\leq 3\sum_{\tau_1,\tau_2}\lambda^2\GWl_d(\tau_1)\GWl_d(\tau_2)\sum_{\tau'_i,i\in[4]}\one_{(\tau_1,\tau'_1)\in S,(\tau_2,\tau'_2)\in S,(\tau_1,\tau'_3)\in S,(\tau_2,\tau'_4)\in S}\one_{|\{\tau'_i\}_i|=2}\dE_0\left[\prod_{i=1}^{4} \widetilde{N}'_{\tau'_i}\right]\\
	&\leq 3\sum_{\tau_1,\tau_2}\lambda^2\GWl_d(\tau_1)\GWl_d(\tau_2)\sum_{\tau'_1,\tau'_2}3\lambda^2\GWl_d(\tau'_1)\GWl_d(\tau'_2)\one_{(\tau_1,\tau'_1)\in S,(\tau_2,\tau'_2)\in S} \, .
	\end{flalign*} Indeed, there are three possibilities for the choice of index $j'$ such that $\tau'_j=\tau'_1$, and for each such choice the contribution is upper bounded by the same term. This yields $\sigma(2,2)\leq 9\lambda^4 \dPl_d(S)^2$.
	
	Summing our bounds on $\sigma(u,v)$ for $u,v\in\set{1,2}$ yields $(iii)$. \\
	
	\proofstep{Point $(iv)$.} Write $\dEls_{d+1}(Z^2)$ is the form
	\begin{multline*}
	\dEls_{d+1} \sum_{(\tau,\tau')\in S}\sum_{(\theta,\theta')\in S}\left[\widetilde{\Delta}_\tau+\sum_{u'}\widetilde{M}_{\tau,u'}\right]\left[\widetilde{\Delta}'_{\tau'}+\sum_u \widetilde{M}_{u,\tau'}\right]\left[\widetilde{\Delta}_\theta+\sum_{v'}\widetilde{M}_{\theta,v'}\right]\left[\widetilde{\Delta}'_{\theta'}+\sum_v \widetilde{M}_{v,\theta'}\right] \, .
	\end{multline*} When expanding the product of brackets, the only terms that will yield a non-zero expectation must have the following sequence of degrees in variables $(\widetilde{\Delta},\widetilde{\Delta}',\widetilde{M})$: $(2,2,0)$, $(2,0,2)$, $(0,2,2)$, or $(0,0,4)$. Denote $\sigma(u,v,w)$ the summation of terms corresponding to exponents $(u,v,w)$. We have:
	\begin{flalign*}
	\sigma(2,2,0) =\sum_{(\tau,\tau')\in S}\dEls_d [\widetilde{\Delta}^2_\tau\widetilde{\Delta}'^2_{\tau'}] = \lambda^2(1-s)^2\dPl_d(S) \, .
	\end{flalign*} We next have
	\begin{flalign*}
	\sigma(2,0,2) \sum_{(\tau,\tau')\in S}\dEls_d \left[\widetilde{\Delta}^2_\tau \sum_u\widetilde{M}^2_{u,\tau'}\right] 
	= \lambda^2 s (1-s)\dPl_d(S)\, ,
	\end{flalign*} and the same expression holds for $\sigma(0,2,2)$. We finally evaluate $\sigma(0,0,4)$. It reads
	\begin{flalign*}
	\sigma(0,0,4)=\sum_{\substack{(\tau,\tau')\in S \\ (\theta,\theta')\in S}}\sum_{u,u',v,v'}\dEls_{d+1} \left[ \widetilde{M}_{\tau,u'}\widetilde{M}_{u,\tau'}\widetilde{M}_{\theta, v'}\widetilde{M}_{v,\theta'} \right] \, .
	\end{flalign*} The non-zero terms in this expectation must comprise either the same term at the power 4, or two distinct terms each at power 2. This yields 4 contributions, that we denote by $A,B,C,D$, which satisfy
	\begin{flalign*}
	A & = \sum_{(\tau,\tau')\in S}\dEls_{d+1} [\widetilde{M}_{\tau,\tau'}^4] =  \sum_{(\tau,\tau')\in S}\left[ 3\lambda^2 s^2 \dPls_d(\tau,\tau')^2+ \lambda s \dPls_d(\tau,\tau')\right] \, , \\
	B & = \sum_{(\tau,\tau')\in S}\dEls_{d+1}[\widetilde{M}^2_{\tau,\tau'}]\sum_{\substack{(\theta,\theta')\in S \\ (\theta,\theta') \neq (\tau,\tau')}}\dEls_{d+1} [\widetilde{M}^2_{\theta,\theta'}] = \lambda^2 s^2 \dPls_d(S)^2-\lambda^2 s^2 \sum_{(\tau,\tau')\in S}\dPls_d(\tau,\tau')^2 \, ,
	\end{flalign*} and
	\begin{flalign*}
	C & = \sum_{(\tau,\tau')\in S} \left[\sum_{u'}\dEls_{d+1} [\widetilde{M}^2_{\tau,u'}]\right] \left[\sum_{u: (u,\tau')\ne (\tau,u')} \dEls_{d+1}[\widetilde{M}^2_{u,\tau'}] \right] \\
	& =
	\lambda^2 s^2 \dPl_d(S) -\lambda^2 s^2 \sum_{(\tau,\tau')\in S}\dPls_d(\tau,\tau')^2 \, , \\
	D & =  \sum_{(\tau,\tau')\in S} \sum_{\substack{(\theta,\theta')\in S \\ (\tau,\theta') \neq (\theta,\tau')}} \dEls_{d+1} [\widetilde{M}^2_{\tau,\theta'}] \dEls_{d+1} [\widetilde{M}^2_{\theta,\tau'}] \\
	& \leq 
	\sum_{(\tau,\tau')\in S}\sum_{\theta,\theta'}\dEls_{d+1} [\widetilde{M}^2_{\tau,\theta'}] \dEls_{d+1} [\widetilde{M}^2_{\theta,\tau'}] -\lambda^2 s^2\sum_{(\tau,\tau')\in S}\dPls_d(\tau,\tau')^2
	\\
	& =
	\lambda^2 s^2 \dPl_d(S)-\lambda^2 s^2\sum_{(\tau,\tau')\in S}\dPls_d(\tau,\tau')^2.
	\end{flalign*}
	
	Summing the expressions of $\sigma(2,2,0)$, $\sigma(2,0,2)$, $\sigma(0,2,2)$, $A$, $B$, $C$ and the upper bound of $D$ we obtain
	\begin{flalign*}
	\dEls_{d+1}[Z^2] & \leq \dEls_{d+1}[Z]^2+\dEls_{d+1}[Z]+\lambda^2(1+s^2)\dPl_d(S) \, ,
	\end{flalign*} and upper bound $(iv)$ follows. 
\end{proof}

With the help of Lemma \ref{lemma_bounds_Z}, we are now ready to turn to the proof of Proposition \ref{prop:bounding_LR_by_induction}.
\begin{proof}[Proof of Proposition \ref{prop:bounding_LR_by_induction}] 
	Assuming that $S\subset \cX_d^2$ is such that
	\begin{equation*}
	\dPl_d(S)\leq \eps \quad \mbox{and} \quad \dPls_d(S)\geq c \, ,
	\end{equation*} Our goal is to choose a threshold $\sigma$ such that
	\begin{equation*}
	\dPl_{d+1}(Z\geq \sigma)\leq \frac{1}{2}\dPl_d(S) \leq \frac{\eps}{2}  \quad \mbox{and} \quad \dPls_{d+1}(Z\geq \sigma)\geq c \, .
	\end{equation*}
	
	\proofstep{First point.} Using point $(iii)$ of Lemma \ref{lemma_bounds_Z}, and Markov's inequality we have
	\begin{equation*}
	\dPl_{d+1}(Z\geq \sigma) \leq \frac{1}{\sigma^4}\dEl_{d+1} [Z^4]\leq \frac{1}{\sigma^4} (36 \lambda^4\dPl_{d}(S)^2+15\lambda^3 \dPl_{d}(S)) \, .
	\end{equation*} It thus suffices to choose $\sigma^4=\max\left(144\lambda^4 \dPl_d(S),60 \lambda^3\right)$ to ensure the first property, that is guarantying that $\dPl_{d+1}(Z\geq \sigma)\leq \frac{1}{2}\dPl_d(S)$. We can a fortiori take $\sigma=\max(4 \lambda \dPl_d(S)^{1/4}, 3 \lambda ^{3/4})$.\\
	
	\proofstep{Second point.} By point $(ii)$ of Lemma \ref{lemma_bounds_Z}, since $\dEls_{d+1}[Z]=\lambda s \dPls_d(S)\geq \lambda s c$  we shall have $\dEls_{d+1}[Z]\geq 2\sigma$ provided
	$$
	8 \dPl_d(S)^{1/4}\leq s c \hbox{ and } 6 \lambda ^{-1/4} \leq s c \, 
	$$ or equivalently
	\begin{equation}\label{eq:tmp_epsilon_0_lambda_0}
	\dPl_d(S)\leq \left(\frac{sc}{8}\right)^4\hbox{ and }\lambda \geq \left(\frac{6}{sc}\right)^4.
	\end{equation}
	This provides the conditions on $\lambda_0$ and $\eps$ required in the statement of the proposition, but let us assume that \eqref{eq:tmp_epsilon_0_lambda_0} is satisfied for now. Using $\sigma\leq \dEls_{d+1}[Z]/2$, Chebyshev's inequality as well as the bound $(iv)$ of Lemma \ref{lemma_bounds_Z}:
	\begin{flalign*}
	\dPls_{d+1}(Z\leq \sigma)& \leq \dPls_{d+1}\left( |Z-\dEls_{d+1}[Z]|\geq \frac{1}{2}\dEls_{d+1}[Z]\right)
	\leq 4 \frac{\Varls_{d+1}(Z)}{\dEls_{d+1}[Z]^2} \\
	&\leq 4\frac{ \lambda s \dPls_{d}(S)+2 \lambda^2 \dPl_d(S)}{\lambda^2 s^2 \dPls_{d}(S)^2}
	\leq \frac{4}{\lambda s c}+\frac{8 \dPl_d(S)}{s^2 c^2} \, .
	\end{flalign*}
	In order to ensure that $\dPls_{d+1}(Z < \sigma)\leq 1-c$, it thus suffices to require
	\begin{equation*}
	\frac{4}{\lambda s c}+\frac{8\dPl_d(S)}{s^2 c^2}\leq 1-c \, .
	\end{equation*}
	We can for instance require
	\begin{equation*}
	\lambda\geq \frac{8}{s c (1-c)}, \quad \dPl_d(S)\leq \frac{(1-c)s^2 c^2}{16}\, .
	\end{equation*}
	Combining this requirement with \eqref{eq:tmp_epsilon_0_lambda_0} we have the announced property by requiring
	$$
	\lambda\geq \lambda_0(s,c):=\max\left(\frac{8}{s c (1-c)},\left(\frac{6}{sc}\right)^4\right),\; \dPl_d(S)\leq \eps(s,c):=\min\left( \left(\frac{sc}{8}\right)^4, \frac{(1-c)s^2 c^2}{16}\right).
	$$
\end{proof}

\begin{proof}[Proof of Lemma \ref{lemma:4th_moment_centered_poisson}] Let $X$ be a Poisson random variable with parameter $\mu$. Write
	\begin{flalign*}
	\dE[X^2]& =\mu^2+\mu,
	\\
	\dE[X^3] &=\dE\left[X(X-1)(X-2)+X[X^2-(X-1)(X-2)]\right]
	\\
	&=\mu^3 +\dE\left[ X[3 X -2] \right]\\
	&=\mu^3 +3(\mu^2+\mu)-2\mu\\
	&=\mu^3+3\mu^2+\mu,\\
	\dE[X^4]&=\dE\left[X(X-1)(X-2)(X-3)+ X[X^3-(X-1)(X-2)(X-3)]\right]
	\\
	&=\mu^4+\dE\left[ X[6 X^2-11 X +6]\right]\\
	&=\mu^4+6[\mu^3+3\mu^2+\mu]-11(\mu^2+\mu)+6\mu\\
	&=\mu^4+6 \mu^3+7\mu^2+\mu.
	\end{flalign*} Write next
	\begin{flalign*}
	\dE[(X-\mu)^4] &= \dE\left[X^4-\binom{4}{1}X^3 \mu +\binom{4}{2}X^2 \mu^2 -\binom{4}{3} X \mu^3+\mu^4   \right]\\
	&=[\mu^4+ 6 \mu^3 +7\mu^2+\mu] -4[\mu^4 + 3 \mu^3 +\mu^2]+6[\mu^4+\mu^3]-4\mu^4+\mu^4\\
	&= 3\mu^2+\mu,
	\end{flalign*}
	as announced.
\end{proof}

\addtocontents{toc}{\protect\setcounter{tocdepth}{2}}
\end{subappendices}

\newpage
\chapter*{Conclusion and research directions}\label{conclusion}
\addcontentsline{toc}{chapter}{Conclusion and research directions}
We have described through this dissertation several contributions to graph alignment and to the tree correlation detection problem. We studied the Gaussian and the \ER models, both from the information-theoretic side (Chapters \ref{chapter:gaussian_alignment_IT}, \ref{chapter:impossibility}, \ref{chapter:MPAlign}, \ref{chapter:addendum}) and from the computational side (Chapters \ref{chapter:EIG1}, \ref{chapter:NTMA}, \ref{chapter:MPAlign}, \ref{chapter:addendum}). We proposed several methods and algorithms, sometimes spectral (Chapter \ref{chapter:EIG1}), based on message-passing using tree similarity (Chapter \ref{chapter:NTMA}) or computing likelihood ratios for detecting local correlation (Chapter \ref{chapter:MPAlign}). 

This field of research is young, and it is certain that many work still remains to be done in order to understand this problem in more generalized settings. We hereafter briefly mention some open questions and research directions that we believe are of particular interest.

\subsubsection*{Typical values of QAP and matching weights in the null model}
Recall that the non-planted version of graph alignment of two graphs with adjacency matrices $A$ and $B$ consists in solving the quadratic assignment problem \eqref{eq:QAP}. A question of interest is the value of the objective $$\max_\Pi \langle A, \Pi B \Pi^T\rangle$$ in the large size limit in the null model, e.g. when $A,B$ are independent \ER graphs. Some upper bounds are  obtained in the literature \cite{Wu20} -- to study the detection problem -- but to the best of our knowledge no exact equivalent is known.
	
In Chapter \ref{chapter:NTMA}, a similarity score between trees $t$ are $t'$ is studied: the tests are based on the matching weight, defined as the largest number of leaves at depth $d$ of a common subtree of $t$ and $t'$. Here again, under the null model, where $t$ and $t'$ are e.g. independent Galton-Watson trees, understanding more thoroughly the typical matching weight of $t$ and $t'$ is still open.
	
\subsubsection*{Optimal fraction for partial recovery}
One could be very interested in the optimal overlap -- or, the largest subset $\cC^*$ -- that one can hope to align in the sparse regime. It is shown in Chapter \ref{chapter:impossibility} that -- up to some vanishing fraction of the nodes --  $\cC^*$ is contained in the giant component $\cC_1$ of the intersection graph. In Section \ref{subsection:warmup} we dealt with the exact isomorphism case $s=1$, for which $\cC^*$ is almost -- i.e, up to some vanishing fraction -- the set of all points invariant by any automorphism. We conjecture that this observation could be generalized to the non-isomorphic case $s <1$, namely that $\cC^*$ is almost the set $\cI$ of invariant nodes \emph{in the intersection graph}. 

\subsubsection*{Generalization to other locally tree-like models}
Detection of correlation in trees, introduced and studied in this manuscript, is a fundamental statistical task of intrinsic interest besides its original motivation from  graph alignment. While in this manuscript we focused on \ER graphs and hence Poisson branching trees, more general locally tree-like graphs could be considered, such as the configuration model, giving rise to correlation detection problems on more general branching trees, for which an extension of the \alg{MPAlign} method could very well be obtained.

\subsubsection*{More efficient algorithms}
Efficient methods proposed in the literature have up to now rather high time complexity -- at least $O(n^3)$ most of the time. We are in a position to ask whether some other methods could perform with a better scaling to large graphs. For graph alignment and related inference problems on graphs, graph neural networks seem to bring relevant architectures and obtain competitive results with lower time complexity (see Azizian and Lelarge \cite{Azizian2021}); giving strong theoretical guarantees however still remains thorny and may be the object of future research in this field.

Another class of algorithms that may shed a new light on the problem are the spectral methods on non-backtracking matrices, following the way paved by community detection literature (see e.g. \cite{BLM18, Moore17}). In our context, there is a chance that these non-local methods may exploit more information than local neighborhoods, and may still be able to perform partial alignment even below the threshold $s < \sqrt{\alpha}$ for the correlation detection problem, which would re-localize the conjectured hard phase. 

\subsubsection*{Computational hardness}

Other active branches of research are seeking for insights on computational hardness for inference problems (see \cite{Bubeck18} for a reduction-based approach). Giving more quantitative results on hardness of graph alignment is still open: several ideas are worth being investigated. The low degree method \cite{Kunisky19} suggests that projecting the likelihood ratio on the space of low-degree polynomials gives strong insights on the poly-time feasibility of a detection problem. Let us mention another concept originally introduced in spin-glass theory, the overlap gap property, which is postulated to reveal algorithmic hardness in combinatorial optimization and inference in planted models, and has recently been exhibited for the planted clique problem \cite{Gamarnik2019TheLO}.

\subsubsection*{Extensions to other settings}

The study of graph alignment for \ER graphs is fundamental and exhibits interesting phenomena. However, \ER is far from being realistic for real-life graphs, that contain more geometry and are scale-free. Studying graph alignment in preferential attachment models -- for instance the Barabási–Albert model -- seems a natural direction for future research.

Also, a recent paper by Wang, Wu, Xu, Yolou establishes interesting results for alignment of geometric graphs \cite{Wang22}, and \cite{Racz21} studies the correlated stochastic block model: results from both community detection and graph alignment are merged together and enable to recover the communities upon observing multiple correlated SBMs, even in regimes where one observation would not suffice. These works can foreshadow similar interesting extensions, enhancing any inference problem on graphs with graph alignment -- e.g, planted clique with additional information coming from several correlated observations.

We close these research directions by mentioning a locally tree-like model in which graph alignment appears very challenging: the regular model. In particular, any method based on exploiting the locally tree-like structure -- if no other information such as labels on nodes is known -- will fail. So, we may ask the question: \emph{what are the information-theoretic and computational limits for regular graph alignment?}

\newpage
\addcontentsline{toc}{chapter}{Bibliography}
\bibliographystyle{alpha}
\bibliography{biblio}

\newcommand{\etalchar}[1]{$^{#1}$}
\begin{thebibliography}{FMWX19b}

\bibitem[AB14]{Allez14bis}
Romain Allez and Jean-Philippe Bouchaud.
\newblock Eigenvector dynamics under free addition.
\newblock {\em Random Matrices: Theory and Applications}, 03(03):1450010, Jul
  2014.

\bibitem[ABB14]{Allez14}
Romain {Allez}, Jo{\"e}l {Bun}, and Jean-Philippe {Bouchaud}.
\newblock {The eigenvectors of Gaussian matrices with an external source}.
\newblock {\em arXiv e-prints}, page arXiv:1412.7108, Dec 2014.

\bibitem[Abb18]{Abbe18}
Emmanuel Abbe.
\newblock Community detection and stochastic block models: Recent developments.
\newblock {\em Journal of Machine Learning Research}, 18(177):1--86, 2018.

\bibitem[ABGL02]{Anstreicher2002}
Kurt Anstreicher, Nathan Brixius, Jean-Pierre Goux, and Jeff Linderoth.
\newblock Solving large quadratic assignment problems on computational grids.
\newblock {\em Mathematical Programming}, 91(3):563--588, Feb 2002.

\bibitem[ABT22]{Araya22}
Ernesto Araya, Guillaume Braun, and Hemant Tyagi.
\newblock Seeded graph matching for the correlated wigner model via the
  projected power method, 2022.

\bibitem[AGZ09]{Anderson09}
Greg~W. Anderson, Alice Guionnet, and Ofer Zeitouni.
\newblock {\em An Introduction to Random Matrices}.
\newblock Cambridge Studies in Advanced Mathematics. Cambridge University
  Press, 2009.

\bibitem[AK02]{Arvind2002}
V.~Arvind and P.P. Kurur.
\newblock Graph isomorphism is in spp.
\newblock In {\em The 43rd Annual IEEE Symposium on Foundations of Computer
  Science, 2002. Proceedings.}, pages 743--750, 2002.

\bibitem[AL21]{Azizian2021}
Waiss Azizian and Marc Lelarge.
\newblock Expressive power of invariant and equivariant graph neural networks.
\newblock In {\em International Conference on Learning Representations}, 2021.

\bibitem[AS16]{ProbaMethod}
Noga Alon and Joel~H. Spencer.
\newblock {\em The Probabilistic Method}.
\newblock Wiley Publishing, 4th edition, 2016.

\bibitem[BBH18]{Bubeck18}
Matthew Brennan, Guy Bresler, and Wasim Huleihel.
\newblock Reducibility and computational lower bounds for problems with planted
  sparse structure.
\newblock In Sébastien Bubeck, Vianney Perchet, and Philippe Rigollet,
  editors, {\em Proceedings of the 31st Conference On Learning Theory},
  volume~75 of {\em Proceedings of Machine Learning Research}, pages 48--166.
  PMLR, 06--09 Jul 2018.

\bibitem[BBM05]{Berg05}
A.~C. {Berg}, T.~L. {Berg}, and J.~{Malik}.
\newblock Shape matching and object recognition using low distortion
  correspondences.
\newblock In {\em 2005 IEEE Computer Society Conference on Computer Vision and
  Pattern Recognition (CVPR'05)}, volume~1, pages 26--33 vol. 1, 2005.

\bibitem[BC05]{Barbour05}
A.~D. {Barbour} and Louis H.~Y. {Chen}.
\newblock {\em An Introduction to Stein's Method}.
\newblock co-published with Singapore University, 2005.

\bibitem[BCL{\etalchar{+}}19]{Barak2019}
Boaz Barak, Chi-Ning Chou, Zhixian Lei, Tselil Schramm, and Yueqi Sheng.
\newblock (nearly) efficient algorithms for the graph matching problem on
  correlated random graphs.
\newblock In H.~Wallach, H.~Larochelle, A.~Beygelzimer, F.~d\textquotesingle
  Alch\'{e}-Buc, E.~Fox, and R.~Garnett, editors, {\em Advances in Neural
  Information Processing Systems}, volume~32. Curran Associates, Inc., 2019.

\bibitem[BDT{\etalchar{+}}20]{Bagaria20}
Vivek Bagaria, Jian Ding, David Tse, Yihong Wu, and Jiaming Xu.
\newblock {Hidden Hamiltonian Cycle Recovery via Linear Programming}.
\newblock {\em Operations Research}, 68(1):53--70, January 2020.

\bibitem[BHK{\etalchar{+}}16]{Barak16}
Boaz Barak, Samuel~B. Hopkins, Jonathan Kelner, Pravesh~K. Kothari, Ankur
  Moitra, and Aaron Potechin.
\newblock A nearly tight sum-of-squares lower bound for the planted clique
  problem, 2016.

\bibitem[BLM18]{BLM18}
Charles Bordenave, Marc Lelarge, and Laurent Massoulié.
\newblock {Nonbacktracking spectrum of random graphs: Community detection and
  nonregular Ramanujan graphs}.
\newblock {\em The Annals of Probability}, 46(1):1 -- 71, 2018.

\bibitem[BMNN16]{Banks2016InformationtheoreticTF}
Jessica~E. Banks, C.~Moore, Joe Neeman, and Praneeth Netrapalli.
\newblock Information-theoretic thresholds for community detection in sparse
  networks.
\newblock {\em ArXiv}, abs/1607.01760, 2016.

\bibitem[BMV{\etalchar{+}}17]{Banks17}
Jess Banks, Cristopher Moore, Roman Vershynin, Nicolas Verzelen, and Jiaming
  Xu.
\newblock Information-theoretic bounds and phase transitions in clustering,
  sparse pca, and submatrix localization.
\newblock In {\em 2017 IEEE International Symposium on Information Theory
  (ISIT)}, pages 1137--1141, 2017.

\bibitem[Bol01]{Bollobas2001}
Béla Bollobás.
\newblock {\em Random Graphs}.
\newblock Cambridge Studies in Advanced Mathematics. Cambridge University
  Press, 2 edition, 2001.

\bibitem[BS11]{Benjamini2011}
Itai Benjamini and Oded Schramm.
\newblock {\em Recurrence of Distributional Limits of Finite Planar Graphs},
  pages 533--545.
\newblock Springer New York, New York, NY, 2011.

\bibitem[CFVS04]{CFVS04}
Donatello Conte, Pasquale Foggia, Mario Vento, and Carlo Sansone.
\newblock {Thirty Years Of Graph Matching In Pattern Recognition}.
\newblock {\em {International Journal of Pattern Recognition and Artificial
  Intelligence}}, 18(3):265--298, 2004.

\bibitem[Cha14]{Chatterjee14}
Sourav Chatterjee.
\newblock {\em Superconcentration and related topics}.
\newblock Springer, 2014.

\bibitem[CK17]{Cullina2017}
Daniel Cullina and Negar Kiyavash.
\newblock Exact alignment recovery for correlated {E}rd{\H{o}}s-{R}{\'{e}}nyi
  graphs, 2017.

\bibitem[CKMP18]{Cullina18}
Daniel Cullina, Negar Kiyavash, Prateek Mittal, and H.~Vincent Poor.
\newblock Partial recovery of {E}rd{\H{o}}s-{R}{\'{e}}nyi graph alignment via
  k-core alignment.
\newblock {\em CoRR}, abs/1809.03553, 2018.

\bibitem[CMK18]{Cullina18data}
Daniel Cullina, P.~Mittal, and N.~Kiyavash.
\newblock Fundamental limits of database alignment.
\newblock {\em 2018 IEEE International Symposium on Information Theory (ISIT)},
  pages 651--655, 2018.

\bibitem[CTYZ17]{Muhao17}
Muhao Chen, Yingtao Tian, Mohan Yang, and Carlo Zaniolo.
\newblock Multilingual knowledge graph embeddings for cross-lingual knowledge
  alignment.
\newblock In {\em Proceedings of the Twenty-Sixth International Joint
  Conference on Artificial Intelligence, {IJCAI-17}}, pages 1511--1517, 2017.

\bibitem[DD22]{Ding22}
Jian Ding and Hang Du.
\newblock Matching recovery threshold for correlated random graphs, 2022.

\bibitem[DF16]{Roee16}
Roee David and Uriel Feige.
\newblock On the effect of randomness on planted 3-coloring models.
\newblock In {\em Proceedings of the Forty-Eighth Annual ACM Symposium on
  Theory of Computing}, STOC '16, page 77–90, New York, NY, USA, 2016.
  Association for Computing Machinery.

\bibitem[DKMZ11]{Decelle11}
Aurelien Decelle, Florent Krzakala, Cristopher Moore, and Lenka Zdeborov\'a.
\newblock Asymptotic analysis of the stochastic block model for modular
  networks and its algorithmic applications.
\newblock {\em Phys. Rev. E}, 84:066106, Dec 2011.

\bibitem[DM13]{Deshpande13}
Yash Deshpande and Andrea Montanari.
\newblock Finding hidden cliques of size $\sqrt{N/e}$ in nearly linear time,
  2013.

\bibitem[DML17]{dym2017ds++}
Nadav Dym, Haggai Maron, and Yaron Lipman.
\newblock Ds++: A flexible, scalable and provably tight relaxation for matching
  problems.
\newblock {\em arXiv preprint arXiv:1705.06148}, 2017.

\bibitem[DMWX21]{Ding18}
Jian Ding, Zongming Ma, Yihong Wu, and Jiaming Xu.
\newblock Efficient random graph matching via degree profiles.
\newblock {\em Probability Theory and Related Fields}, 179(1):29--115, Feb
  2021.

\bibitem[Dwo08]{Dwork08}
Cynthia Dwork.
\newblock Differential privacy: A survey of results.
\newblock In {\em Theory and Applications of Models of Computation—TAMC},
  volume 4978 of {\em Lecture Notes in Computer Science}, pages 1--19. Springer
  Verlag, April 2008.

\bibitem[DWXY21]{Ding21_matching}
Jian Ding, Yihong Wu, Jiaming Xu, and Dana Yang.
\newblock The planted matching problem: Sharp threshold and infinite-order
  phase transition, 2021.

\bibitem[ECK19]{Dai19}
Osman {Emre Dai}, Daniel {Cullina}, and Negar {Kiyavash}.
\newblock {Database Alignment with Gaussian Features}.
\newblock {\em arXiv e-prints}, page arXiv:1903.01422, March 2019.

\bibitem[EKHK15]{elkebir2015}
Mohammed El-Kebir, Jaap Heringa, and Gunnar Klau.
\newblock Natalie 2.0: Sparse global network alignment as a special case of
  quadratic assignment.
\newblock {\em Algorithms}, 8(4), December 2015.

\bibitem[ER59]{Erdos59}
P.~Erd\"{o}s and A.~R\'{e}nyi.
\newblock On random graphs i.
\newblock {\em Publicationes Mathematicae Debrecen}, 6:290, 1959.

\bibitem[EYY10]{Erdos10}
Laszlo {Erdos}, Horng-Tzer {Yau}, and Jun {Yin}.
\newblock {Rigidity of Eigenvalues of Generalized Wigner Matrices}.
\newblock {\em arXiv e-prints}, page arXiv:1007.4652, Jul 2010.

\bibitem[FAP{\etalchar{+}}19]{Fishkind19}
Donniell~E. Fishkind, Sancar Adali, Heather~G. Patsolic, Lingyao Meng, Digvijay
  Singh, Vince Lyzinski, and Carey~E. Priebe.
\newblock Seeded graph matching.
\newblock {\em Pattern Recognition}, 87:203--215, 2019.

\bibitem[FCC{\etalchar{+}}21]{Frigo2021}
Matteo Frigo, Emilio Cruciani, David Coudert, Rachid Deriche, Samuel
  Deslauriers-Gauthier, and Emanuele Natale.
\newblock {Network alignment and similarity reveal atlas-based topological
  differences in structural connectomes}.
\newblock {\em {Network Neuroscience}}, May 2021.

\bibitem[FMWX19a]{Fan2019Wigner}
Zhou Fan, Cheng Mao, Yihong Wu, and Jiaming Xu.
\newblock Spectral graph matching and regularized quadratic relaxations {I}:
  The gaussian model, 2019.

\bibitem[FMWX19b]{fan2019ERC}
Zhou Fan, Cheng Mao, Yihong Wu, and Jiaming Xu.
\newblock Spectral graph matching and regularized quadratic relaxations {II}:
  {E}rd{\H{o}}s-{R}{\'{e}}nyi graphs and universality, 2019.

\bibitem[For93]{Forrester93}
P.J. Forrester.
\newblock The spectrum edge of random matrix ensembles.
\newblock {\em Nuclear Physics B}, 402(3):709 -- 728, 1993.

\bibitem[FQM{\etalchar{+}}16]{Feizi16}
Soheil Feizi, Gerald Quon, Mariana~Recamonde Mendoza, Muriel M{\'{e}}dard,
  Manolis Kellis, and Ali Jadbabaie.
\newblock Spectral alignment of networks.
\newblock {\em CoRR}, abs/1602.04181, 2016.

\bibitem[FR10]{Feige10}
Uriel Feige and Dorit Ron.
\newblock {Finding hidden cliques in linear time}.
\newblock In Drmota, Michael, Gittenberger, and Bernhard, editors, {\em {21st
  International Meeting on Probabilistic, Combinatorial, and Asymptotic Methods
  in the Analysis of Algorithms (AofA'10)}}, volume DMTCS Proceedings vol. AM,
  21st International Meeting on Probabilistic, Combinatorial, and Asymptotic
  Methods in the Analysis of Algorithms (AofA'10) of {\em DMTCS Proceedings},
  pages 189--204, Vienna, Austria, 2010. {Discrete Mathematics and Theoretical
  Computer Science}.

\bibitem[Gan22]{Ganassali21MSML}
Luca Ganassali.
\newblock Sharp threshold for alignment of graph databases with gaussian
  weights.
\newblock In Joan Bruna, Jan Hesthaven, and Lenka Zdeborova, editors, {\em
  Proceedings of the 2nd Mathematical and Scientific Machine Learning
  Conference}, volume 145 of {\em Proceedings of Machine Learning Research},
  pages 314--335. PMLR, 16--19 Aug 2022.

\bibitem[GJS74]{Garey74}
M.~R. Garey, D.~S. Johnson, and L.~Stockmeyer.
\newblock Some simplified np-complete problems.
\newblock In {\em Proceedings of the Sixth Annual ACM Symposium on Theory of
  Computing}, STOC '74, page 47–63, New York, NY, USA, 1974. Association for
  Computing Machinery.

\bibitem[GLM22]{GLM19}
Luca Ganassali, Marc Lelarge, and Laurent Massoulié.
\newblock Spectral alignment of correlated gaussian matrices.
\newblock {\em Advances in Applied Probability}, 54(1):279–310, 2022.

\bibitem[GM20]{Ganassali20a}
Luca Ganassali and Laurent Massouli\'e.
\newblock From tree matching to sparse graph alignment.
\newblock volume 125 of {\em Proceedings of Machine Learning Research}, pages
  1633--1665. PMLR, 09--12 Jul 2020.

\bibitem[GML21a]{GMLTrees2021journal}
Luca Ganassali, Laurent Massoulié, and Marc Lelarge.
\newblock Correlation detection in trees for planted graph alignment, 2021.

\bibitem[GML21b]{ganassali2021impossibility}
Luca Ganassali, Laurent Massoulié, and Marc Lelarge.
\newblock Impossibility of partial recovery in the graph alignment problem.
\newblock volume 134 of {\em Proceedings of Machine Learning Research}, pages
  2080--2102. PMLR, 15--19 Aug 2021.

\bibitem[GML22]{GMLTrees2021ITCS}
Luca Ganassali, Laurent Massouli\'{e}, and Marc Lelarge.
\newblock {Correlation Detection in Trees for Planted Graph Alignment}.
\newblock In Mark Braverman, editor, {\em 13th Innovations in Theoretical
  Computer Science Conference (ITCS 2022)}, volume 215 of {\em Leibniz
  International Proceedings in Informatics (LIPIcs)}, pages 74:1--74:8,
  Dagstuhl, Germany, 2022. Schloss Dagstuhl -- Leibniz-Zentrum f{\"u}r
  Informatik.

\bibitem[GZ19]{Gamarnik2019TheLO}
David Gamarnik and Ilias Zadik.
\newblock The landscape of the planted clique problem: Dense subgraphs and the
  overlap gap property.
\newblock {\em ArXiv}, abs/1904.07174, 2019.

\bibitem[HLL83]{Holland1983}
Paul Holland, Kathryn~B. Laskey, and Samuel Leinhardt.
\newblock Stochastic blockmodels: First steps.
\newblock {\em Social Networks}, 5:109--137, 1983.

\bibitem[HM20]{Hall20}
Georgina Hall and Laurent Massoulié.
\newblock Partial recovery in the graph alignment problem, 2020.

\bibitem[HNM05]{Haghighi05}
Aria~D. Haghighi, Andrew~Y. Ng, and Christopher~D. Manning.
\newblock Robust textual inference via graph matching.
\newblock In {\em Proceedings of the Conference on Human Language Technology
  and Empirical Methods in Natural Language Processing}, HLT '05, pages
  387--394, Stroudsburg, PA, USA, 2005. Association for Computational
  Linguistics.

\bibitem[Hof16]{hofstad2016}
Remco van~der Hofstad.
\newblock {\em Random Graphs and Complex Networks}.
\newblock Cambridge Series in Statistical and Probabilistic Mathematics.
  Cambridge University Press, 2016.

\bibitem[HW71]{HansonWright1971}
D.~L. Hanson and F.~T. Wright.
\newblock A bound on tail probabilities for quadratic forms in independent
  random variables.
\newblock {\em Ann. Math. Statist.}, 42(3):1079--1083, 06 1971.

\bibitem[Hå99]{Hastad99}
Johan Håstad.
\newblock {Clique is hard to approximate within $n^{1-\epsilon}$}.
\newblock {\em Acta Mathematica}, 182(1):105 -- 142, 1999.

\bibitem[IMB]{IMBD07}
The internet movie database.
\newblock \url{http:// www.imdb.com/}.
\newblock 2007.

\bibitem[Jer92]{Jerrum92}
Mark Jerrum.
\newblock Large cliques elude the metropolis process.
\newblock {\em Random Structures \& Algorithms}, 3(4):347--359, 1992.

\bibitem[JLR00]{Janson00}
Svante Janson, Tomasz Luczak, and Andrzej Rucinski.
\newblock {\em Random graphs}.
\newblock Wiley-Interscience series in discrete mathematics and optimization.
  Wiley, 2000.

\bibitem[JLTV12]{Janson12}
Svante Janson, Tomasz Luczak, Tatyana Turova, and Thomas Vallier.
\newblock {Bootstrap percolation on the random graph $G_{n,p}$}.
\newblock {\em The Annals of Applied Probability}, 22(5):1989 -- 2047, 2012.

\bibitem[Kar72]{Karp72}
Richard~M. Karp.
\newblock {\em Reducibility among Combinatorial Problems}, pages 85--103.
\newblock Springer US, Boston, MA, 1972.

\bibitem[KG16]{Kazemi16}
Ehsan Kazemi and Matthias Grossglauser.
\newblock On the structure and efficient computation of isorank node
  similarities.
\newblock {\em ArXiv}, abs/1602.00668, 2016.

\bibitem[KHG15]{kasemi15}
Ehsan Kazemi, S.~Hamed Hassani, and Matthias Grossglauser.
\newblock Growing a graph matching from a handful of seeds.
\newblock {\em Proc. VLDB Endow.}, 8(10):1010–1021, jun 2015.

\bibitem[KM09]{Klencke09}
Achim Klenke and Lutz Mattner.
\newblock Stochastic ordering of classical discrete distributions, 2009.

\bibitem[KMM{\etalchar{+}}13]{SpectralRedemption13}
Florent Krzakala, Cristopher Moore, Elchanan Mossel, Joe Neeman, Allan Sly,
  Lenka Zdeborov{\'a}, and Pan Zhang.
\newblock Spectral redemption in clustering sparse networks.
\newblock {\em Proceedings of the National Academy of Sciences},
  110(52):20935--20940, 2013.

\bibitem[KSMG13]{Kollias2013}
Giorgos Kollias, Madan Sathe, Shahin Mohammadi, and Ananth Grama.
\newblock A fast approach to global alignment of protein-protein interaction
  networks.
\newblock {\em BMC Research Notes}, 6(1):35, Jan 2013.

\bibitem[Kuh55]{Kuhn55}
H.~W. Kuhn.
\newblock The hungarian method for the assignment problem.
\newblock {\em Naval Research Logistics Quarterly}, 2(1‐2):83--97, 1955.

\bibitem[KWB19]{Kunisky19}
Dmitriy Kunisky, Alexander~S. Wein, and Afonso~S. Bandeira.
\newblock Notes on computational hardness of hypothesis testing: Predictions
  using the low-degree likelihood ratio, 2019.

\bibitem[LFP14a]{lyzinski14_seeded}
Vince Lyzinski, Donniell~E. Fishkind, and Carey~E. Priebe.
\newblock Seeded graph matching for correlated erdos-renyi graphs.
\newblock {\em Journal of Machine Learning Research}, 15(108):3693--3720, 2014.

\bibitem[LFP14b]{Lyzinski14}
Vince Lyzinski, Donniell~E. Fishkind, and Carey~E. Priebe.
\newblock Seeded graph matching for correlated erdos-renyi graphs.
\newblock {\em Journal of Machine Learning Research}, 15:3693--3720, 2014.

\bibitem[LLB{\etalchar{+}}09]{Liao2009}
Chung-Shou Liao, Kanghao Lu, Michael Baym, Rohit Singh, and Bonnie Berger.
\newblock Isorankn: spectral methods for global alignment of multiple protein
  networks.
\newblock {\em Bioinformatics (Oxford, England)}, 25(12):i253--i258, Jun 2009.
\newblock 19477996[pmid].

\bibitem[LP01]{lin2001dirt}
Dekang Lin and Patrick Pantel.
\newblock {DIRT}: {D}iscovery of {I}nference {R}ules from {T}ext.
\newblock In {\em Proceedings of the Seventh ACM SIGKDD International
  Conference on Knowledge Discovery and Data Mining (KDD'01)}, pages 323--328,
  New York, NY, USA, 2001. ACM Press.

\bibitem[LS18]{lubars18}
Joseph Lubars and R.~Srikant.
\newblock Correcting the output of approximate graph matching algorithms.
\newblock In {\em IEEE INFOCOM 2018 - IEEE Conference on Computer
  Communications}, pages 1745--1753, 2018.

\bibitem[Luc88]{luczak1988}
Tomasz Luczak.
\newblock The automorphism group of random graphs with a given number of edges.
\newblock {\em Mathematical Proceedings of the Cambridge Philosophical
  Society}, 104(3):441–449, 1988.

\bibitem[Mad16]{madi2016}
Kamel Madi.
\newblock {\em {Inexact graph matching : application to 2D and 3D Pattern
  Recognition}}.
\newblock Theses, {Universit{\'e} de Lyon}, December 2016.

\bibitem[Mas14]{Massoulie14}
Laurent Massouli{\'e}.
\newblock {Community detection thresholds and the weak Ramanujan property}.
\newblock In {\em {STOC 2014: 46th Annual Symposium on the Theory of
  Computing}}, pages 1--10, New York, United States, June 2014.

\bibitem[Mat72]{matula1972}
David~W. Matula.
\newblock The employee party problem.
\newblock In {\em Notices of the American Mathematical Society}, volume~19,
  page A382–A382. 1972.

\bibitem[MMS14]{Makarychev14}
Konstantin Makarychev, Rajsekar Manokaran, and Maxim Sviridenko.
\newblock Maximum quadratic assignment problem: Reduction from maximum label
  cover and lp-based approximation algorithm.
\newblock {\em CoRR}, abs/1403.7721, 2014.

\bibitem[MMX21]{Moharrami21}
Mehrdad Moharrami, Cristopher Moore, and Jiaming Xu.
\newblock {The planted matching problem: Phase transitions and exact results}.
\newblock {\em The Annals of Applied Probability}, 31(6):2663 -- 2720, 2021.

\bibitem[MNS15]{Mossel15}
Elchanan Mossel, Joe Neeman, and Allan Sly.
\newblock Reconstruction and estimation in the planted partition model.
\newblock {\em Probability Theory and Related Fields}, 162(3):431--461, Aug
  2015.

\bibitem[MNS16]{Mossel2016}
Elchanan Mossel, Joe Neeman, and Allan Sly.
\newblock {Belief propagation, robust reconstruction and optimal recovery of
  block models}.
\newblock {\em The Annals of Applied Probability}, 26(4):2211 -- 2256, 2016.

\bibitem[MNS18]{Mossel2018}
Elchanan Mossel, Joe Neeman, and Allan Sly.
\newblock A proof of the block model threshold conjecture.
\newblock {\em Combinatorica}, 38(3):665--708, Jun 2018.

\bibitem[{Moo}17]{Moore17}
Cristopher {Moore}.
\newblock {The Computer Science and Physics of Community Detection: Landscapes,
  Phase Transitions, and Hardness}.
\newblock {\em arXiv e-prints}, page arXiv:1702.00467, February 2017.

\bibitem[MP03]{Mossel03}
Elchanan Mossel and Yuval Peres.
\newblock Information flow on trees.
\newblock {\em The Annals of Applied Probability}, 13(3):817--844, 2003.

\bibitem[MRT21a]{Mao21_constant_corr}
Cheng Mao, Mark Rudelson, and Konstantin Tikhomirov.
\newblock Exact matching of random graphs with constant correlation, 2021.

\bibitem[MRT21b]{Mao21}
Cheng Mao, Mark Rudelson, and Konstantin Tikhomirov.
\newblock Random graph matching with improved noise robustness.
\newblock In Mikhail Belkin and Samory Kpotufe, editors, {\em Conference on
  Learning Theory, {COLT} 2021, 15-19 August 2021, Boulder, Colorado, {USA}},
  volume 134 of {\em Proceedings of Machine Learning Research}, pages
  3296--3329. {PMLR}, 2021.

\bibitem[MST19]{stephan19}
Laurent Massouli\'{e}, Ludovic Stephan, and Don Towsley.
\newblock Planting trees in graphs, and finding them back.
\newblock In Alina Beygelzimer and Daniel Hsu, editors, {\em Proceedings of the
  Thirty-Second Conference on Learning Theory}, volume~99 of {\em Proceedings
  of Machine Learning Research}, pages 2341--2371. PMLR, 25--28 Jun 2019.

\bibitem[MWXY21]{Mao21_counting}
Cheng Mao, Yihong Wu, Jiaming Xu, and Sophie~H. Yu.
\newblock Testing network correlation efficiently via counting trees, 2021.

\bibitem[MX19]{Mossel19}
Elchanan Mossel and Jiaming Xu.
\newblock Seeded graph matching via large neighborhood statistics.
\newblock In {\em Proceedings of the Thirtieth Annual ACM-SIAM Symposium on
  Discrete Algorithms}, SODA '19, page 1005–1014, USA, 2019. Society for
  Industrial and Applied Mathematics.

\bibitem[NET]{NET06}
Netflix prize.
\newblock \url{http:// www.netflixprize.com/}.
\newblock 2006.

\bibitem[NS08]{Narayanan08}
A.~{Narayanan} and V.~{Shmatikov}.
\newblock Robust de-anonymization of large sparse datasets.
\newblock In {\em 2008 IEEE Symposium on Security and Privacy (sp 2008)}, pages
  111--125, May 2008.

\bibitem[NS09]{Narayanan09}
A.~{Narayanan} and V.~{Shmatikov}.
\newblock De-anonymizing social networks.
\newblock In {\em 2009 30th IEEE Symposium on Security and Privacy}, pages
  173--187, May 2009.

\bibitem[OB20]{osman20}
Ahmed~Hamza Osman and Omar~Mohammed Barukub.
\newblock Graph-based text representation and matching: A review of the state
  of the art and future challenges.
\newblock {\em IEEE Access}, 8:87562--87583, 2020.

\bibitem[{O'R}10]{ORourke10}
Sean {O'Rourke}.
\newblock {Gaussian Fluctuations of Eigenvalues in Wigner Random Matrices}.
\newblock {\em Journal of Statistical Physics}, 138(6):1045--1066, Mar 2010.

\bibitem[OSA16]{Olivetti16}
Emanuele Olivetti, Nusrat Sharmin, and Paolo Avesani.
\newblock Alignment of tractograms as graph matching.
\newblock {\em Frontiers in Neuroscience}, 10, 2016.

\bibitem[Ott48]{Otter48}
Richard Otter.
\newblock The number of trees.
\newblock {\em Annals of Mathematics}, 49(3):583--599, 1948.

\bibitem[OVW16]{ORourke16}
Sean {O'Rourke}, Van {Vu}, and Ke~{Wang}.
\newblock {Eigenvectors of random matrices: A survey}.
\newblock {\em arXiv e-prints}, page arXiv:1601.03678, Jan 2016.

\bibitem[PG11]{Pedarsani11}
Pedram Pedarsani and Matthias Grossglauser.
\newblock On the privacy of anonymized networks.
\newblock In {\em Proceedings of the 17th ACM SIGKDD International Conference
  on Knowledge Discovery and Data Mining}, KDD '11, page 1235–1243, New York,
  NY, USA, 2011. Association for Computing Machinery.

\bibitem[PRW94]{Pardalos94}
Panos Pardalos, Franz Rendl, and Henry Wolkowicz.
\newblock {\em The Quadratic Assignment Problem: A Survey and Recent
  Developments}, pages 1--42.
\newblock 08 1994.

\bibitem[PSSZ21]{piccioli2021aligning}
Giovanni Piccioli, Guilhem Semerjian, Gabriele Sicuro, and Lenka Zdeborová.
\newblock Aligning random graphs with a sub-tree similarity message-passing
  algorithm, 2021.

\bibitem[PW17]{PolyanskiyLecturenotes}
Yury Polyanskiy and Yihong Wu.
\newblock Lecture notes on information theory, 2012-2017.

\bibitem[PWC16]{Pananjady16}
Ashwin {Pananjady}, Martin~J. {Wainwright}, and Thomas~A. {Courtade}.
\newblock {Linear Regression with an Unknown Permutation: Statistical and
  Computational Limits}.
\newblock {\em arXiv e-prints}, page arXiv:1608.02902, August 2016.

\bibitem[QCM{\etalchar{+}}09]{OureshiPPI09}
Amir Qureshi, Vineet Chaoji, Dony Maiguel, Hafeez Faridi, Constantinos Barth,
  Saeed Salem, Mudita Singhal, Darren Stoub, Bryan Krastins, Mitsunori Ogihara,
  Mohammed Zaki, and Vineet Gupta.
\newblock Proteomic and phospho-proteomic profile of human platelets in basal,
  resting state: Insights into integrin signaling.
\newblock {\em PloS one}, 4:e7627, 10 2009.

\bibitem[RS21]{Racz21}
Miklos Racz and Anirudh Sridhar.
\newblock Correlated stochastic block models: Exact graph matching with
  applications to recovering communities.
\newblock In M.~Ranzato, A.~Beygelzimer, Y.~Dauphin, P.S. Liang, and J.~Wortman
  Vaughan, editors, {\em Advances in Neural Information Processing Systems},
  volume~34, pages 22259--22273. Curran Associates, Inc., 2021.

\bibitem[SGE17]{Shirani2017SeededGM}
Farhad Shirani, Siddharth Garg, and Elza Erkip.
\newblock Seeded graph matching: Efficient algorithms and theoretical
  guarantees.
\newblock {\em 2017 51st Asilomar Conference on Signals, Systems, and
  Computers}, pages 253--257, 2017.

\bibitem[SSZ20]{Semerjian20}
Guilhem Semerjian, Gabriele Sicuro, and Lenka Zdeborov\'a.
\newblock Recovery thresholds in the sparse planted matching problem.
\newblock {\em Phys. Rev. E}, 102:022304, Aug 2020.

\bibitem[SXB07]{Singh07}
Rohit Singh, Jinbo Xu, and Bonnie Berger.
\newblock Pairwise global alignment of protein interaction networks by matching
  neighborhood topology.
\newblock In Terry Speed and Haiyan Huang, editors, {\em Research in
  Computational Molecular Biology}, pages 16--31, Berlin, Heidelberg, 2007.
  Springer Berlin Heidelberg.

\bibitem[SXB08]{Singh08}
Rohit Singh, Jinbo Xu, and Bonnie Berger.
\newblock Global alignment of multiple protein interaction networks with
  application to functional orthology detection.
\newblock {\em Proceedings of the National Academy of Sciences},
  105(35):12763--12768, 2008.

\bibitem[TW98]{Tracy98}
Craig~A. {Tracy} and Harold {Widom}.
\newblock {Correlation Functions, Cluster Functions, and Spacing Distributions
  for Random Matrices}.
\newblock {\em Journal of Statistical Physics}, 92(5-6):809--835, Sep 1998.

\bibitem[Ume88]{Umeyama88}
S.~Umeyama.
\newblock An eigendecomposition approach to weighted graph matching problems.
\newblock {\em IEEE Transactions on Pattern Analysis and Machine Intelligence},
  10(5):695--703, 1988.

\bibitem[VCL{\etalchar{+}}11]{Vogelstein2011}
Joshua~T. Vogelstein, John~M. Conroy, Vince Lyzinski, Louis~J. Podrazik,
  Steven~G. Kratzer, Eric~T. Harley, Donniell~E. Fishkind, R.~Jacob Vogelstein,
  and Carey~E. Priebe.
\newblock Fast approximate quadratic programming for large (brain) graph
  matching, 2011.

\bibitem[WWXY22]{Wang22}
Haoyu Wang, Yihong Wu, Jiaming Xu, and Israel Yolou.
\newblock Random graph matching in geometric models: the case of complete
  graphs, 2022.

\bibitem[WX19]{WuXuLecturenotes}
Yihong Wu and Jiaming Xu.
\newblock Statistical inference on graphs (lecture notes), August 2019.

\bibitem[WXY20]{Wu20}
Yihong {Wu}, Jiaming {Xu}, and Sophie~H. {Yu}.
\newblock {Testing correlation of unlabeled random graphs}.
\newblock {\em arXiv e-prints}, page arXiv:2008.10097, August 2020.

\bibitem[WXY21]{Wu2021SettlingTS}
Yihong Wu, Jiaming Xu, and Sophie~H. Yu.
\newblock Settling the sharp reconstruction thresholds of random graph
  matching.
\newblock {\em ArXiv}, abs/2102.00082, 2021.

\bibitem[YG13]{Yartseva13}
Lyudmila Yartseva and Matthias Grossglauser.
\newblock On the performance of percolation graph matching.
\newblock In {\em Proceedings of the First ACM Conference on Online Social
  Networks}, COSN '13, page 119–130, New York, NY, USA, 2013. Association for
  Computing Machinery.

\bibitem[YXL21]{Yu21}
Liren Yu, Jiaming Xu, and Xiaojun Lin.
\newblock Graph matching with partially-correct seeds.
\newblock {\em Journal of Machine Learning Research}, 22(280):1--54, 2021.

\bibitem[ZBV09]{Bach09}
M.~{Zaslavskiy}, F.~{Bach}, and J.~{Vert}.
\newblock A path following algorithm for the graph matching problem.
\newblock {\em IEEE Transactions on Pattern Analysis and Machine Intelligence},
  31(12):2227--2242, 2009.

\end{thebibliography}

\end{document}